\documentclass[12pt, a4paper,twoside]{report}
\usepackage{amsmath,amssymb,amsthm}
\usepackage{graphicx}
\usepackage[left=3cm, right=2cm, top=3cm, bottom=2.5cm]{geometry}
\usepackage[english]{babel}
\usepackage{simplewick}
\usepackage{hyperref}
\usepackage{tikz}
\usepackage{svg}
\usepackage{dsfont}
\newcommand{\sref}[2]{\hyperref[#2]{#1 \ref*{#2}}}
\usepackage{chngcntr}
\usepackage{fancyhdr}
\newcommand{\J}{\mathbb{J}}
\newcommand{\Tr}{\,\mathrm{Tr}}
\newcommand{\B}{\mathcal{B}}
\newcommand{\F}{\mathcal{F}}
\newcommand{\E}{\mathcal{E}}
\newcommand{\U}{\mathcal{U}}
\newcommand{\Z}{\mathcal{Z}}
\newcommand{\I}{\mathrm{i}}
\newcommand{\N}{\mathbb{N}}
\newcommand{\R}{\mathbb{R}}

\newcommand{\mN}{\mathcal{N}}
\newcommand{\D}{\mathfrak{D}}
\newcommand{\C}{\mathbb{C}}

\newcommand{\Hlog}{\mathrm{Hlog}}
\newcommand{\narrowunderline}[1]{\mkern1mu\underline{\mkern-1mu#1\mkern-2mu}\mkern2mu }
\newcommand{\G}{\mathcal{G}}
\DeclareMathOperator{\Res}{Res}
\newcommand{\RNum}[1]{\uppercase\expandafter{\romannumeral #1\relax}}
\allowdisplaybreaks[4]

\def\smalllozenge{\mbox{\scriptsize$\lozenge$}}
\def\smallblacklozenge{\mbox{\scriptsize$\blacklozenge$}}

\newtheoremstyle{mydefinition}{3pt}{3pt}{\sf}{}{\bf}{.}{.5em}{}

\theoremstyle{mydefinition}
\newtheorem{dfnt}{Definition}
\numberwithin{dfnt}{chapter}
\newtheorem{lemma}{Lemma}
\numberwithin{lemma}{chapter}
\newtheorem{thrm}{Theorem}
\numberwithin{thrm}{chapter}
\newtheorem{prps}{Proposition}
\numberwithin{prps}{chapter}
\newtheorem{cor}{Corollary}
\numberwithin{cor}{chapter}
\newtheorem{assumption}{Assumption}
\numberwithin{assumption}{chapter}
\newtheorem{Alemma}{Lemma}
\numberwithin{Alemma}{chapter}

\newtheorem{rmk}{Remark}
\numberwithin{rmk}{chapter}
\newtheorem{exm}{Example}
\numberwithin{exm}{chapter}
\numberwithin{equation}{section}

\begin{document}
{\pagestyle{empty}
  \addtocontents{toc}{\protect\thispagestyle{empty}} 
\begin{center}
{\LARGE Mathematik\\

\vspace{10.ex}
}

\vspace{5.ex}
{\huge
\textbf{Matrix Field Theory}}\\

\vspace{50.ex}
{\large Inaugural-Dissertation\\
zur Erlangung des Doktorgrades\\
der Naturwissenschaften im Fachbereich\\
Mathematik und Informatik\\
der Mathematisch-Naturwissenschaftlichen Fakult\"at\\
der Westf\"alischen Wilhelms-Universit\"at M\"unster\\

\vspace{10.ex}
vorgelegt von\\
Alexander Hock\\
aus Dschangi-Dscher/Kirgisistan\\
\vspace{2.ex}

-2020-\\
}
\end{center}

\newpage

\vspace*{100.ex}
 
\noindent
Dekan: \qquad\qquad\qquad\qquad\qquad\qquad Prof. Dr. Xiaoyi Jiang\\
Erster Gutachter:\qquad\qquad\qquad\qquad\!\!\!\!  Prof. Dr. Raimar Wulkenhaar\\
Zweiter Gutachter:\qquad\qquad\qquad\quad\,\!\! Prof. Dr. Gernot M\"unster\\
Tag der m\"undlichen Pr\"ufung:\\
Tag der Promotion:
\newpage
{\huge
\textbf{Abstract}}
\vspace{5.ex}
\\
This thesis studies matrix field theories, which are a special type of matrix models. 
First, the different types of applications are pointed out, from (noncommutative) 
quantum field theory over 2-dimensional quantum gravity up to algebraic geometry with 
explicit computation of intersection numbers on the moduli space of complex curves.

The Kontsevich model, which has proved the Witten conjecture, 
is the simplest example of a matrix field theory. Generalisations of this model will be studied, where 
different potentials and the spectral dimension are introduced. 
Because they are 
naturally embedded into a Riemann surface,
the correlation functions are graded by the genus and the number of boundary components. The renormalisation procedure 
of quantum field theory leads to finite UV-limit.

We provide a method to determine closed Schwinger-Dyson equations with the usage of Ward-Takahashi
identities in the continuum limit. The cubic (Kontsevich model) and the quartic (Grosse-Wulkenhaar model)
potentials are studied separately.

The cubic model is solved completely for any spectral dimension $<8$, i.e. all 
correlation functions are derived explicitly. 
Inspired by topological recursion, we propose
creation and annihilation operators by differential 
and residue operators.
The exact results are confirmed by perturbative computations with Feynman graphs renormalised by 
Zimmermann's forest formula. The number and the amplitudes of the graphs grow factorially, which is known 
as renormalon problem. However, these series are convergent since the exact results are provided.
A further differential operator is derived to determine all free energies. Additionally, by the theorem
of Kontsevich, the intersection numbers of the moduli space of complex 
curves $\overline{\mathcal{M}}_{g,b}$ are found.

For the quartic model, the 2-point function is
derived for any spectral dimension $<6$ explicitly. 
The first step is to derive an angle function which is, after analytic continuation, interpreted 
as an effective measure. On the 4-dimensional noncommutative Moyal space, the effective 
measure is given by a hypergeometric function. Its asymptotic behaviour changes the 
spectral dimension effectively to $4-2\frac{\arcsin(\lambda\pi)}{\pi}$ for $|\lambda|<\frac{1}{\pi}$.
This dimension drop prevents the quantum field theoretical 4-dimensional $\Phi^4$-model on the Moyal space from the triviality problem.
After combinatorial analysis, an explicit (not recursive) formula for any planar $N$-point function 
is provided.

The evident  difference between the cubic and the quartic model is of algebraic-geometric nature.
Computing correlation functions via topological recursion needs the spectral curve as initial 
data. This algebraic curve has for the cubic model only one branch point which coincides with the 
pole of the stable correlation functions. However,
the quartic model has a spectral curve which admits 
infinitely many branch points in the continuum
limit.

\newpage
{\huge
\textbf{Zusammenfassung}}
\vspace{5.ex}\\
Diese Dissertation besch\"aftigt sich mit Matrix-Feldtheorien, einer speziellen Form der Matrixmodelle.
Zun\"achst werden unterschiedliche Anwendungsm\"oglichkeiten hervorgehoben, die
von (nichtkommutativer) Quantenfeldtheorie \"uber 2-dimensionale Quantengravitation bis hin zur 
algebraischen Geometrie mit expliziter Berechnung von Schnittzahlen auf dem
Modulraum komplexer Kurven reichen.

Das Kontsevich-Modell ist ein paragdigmatisches Beispiel einer 
Matrix Feldtheorie, mit der Wittens Vermutung bewiesen wurde.
Es werden Verallgemeinerungen dieses Mo\-dells betrachtet, die durch die Wahl eines
anderen Potentials und durch Einf\"uhrung der spektralen Dimension.
Die Korrelationsfunktionen werden durch das Geschlecht und die Zahl der Randkomponenten unterschieden, da 
diese eine nat\"urliche Darstellung auf Riemannschen Fl\"achen besitzen.
Um dem UV Limes Bedeutung zu verleihen,
werden Renormierungsmethoden aus der Quantenfeldtheorie verwendet.

Wir zeigen, wie geschlossene Schwinger-Dyson-Gleichungen f\"ur
die Korrelationsfunktionen mit Hilfe von Ward-Takahashi-Identit\"aten im Kontiuumslimes
bestimmt werden k\"onnen. Die spezielle Wahl eines kubischen (Kontsevich-Modell) und eines
quartischen (Grosse-Wulkenhaar-Modell) Potentials wird separat betrachtet und untersucht. 

Das kubische Modell wird vollst\"andig f\"ur eine spektrale Dimension $<8$ gel\"ost, d.h. 
es werden alle Korrelationsfunktionen explizit berechnet. 
Die Erzeuger- und Vernichterope\-ratoren werden als Differential- und Residuumsoperator angegeben, 
wobei die Konstruktion durch topologische Rekursion inspiriert wurde.
Die Resultate werden
durch st\"orungstheoretische Rechnungen best\"atigt, in denen in Feynman-Graphen entwickelt wird, die
durch Zimmermanns Waldformel renormiert werden. Die Anzahl der Graphen und die Amplitude der Graphen steigen 
mit $\mathcal{O}(n!)$, welches als Renormalon-Problem bezeichnet wird; dennoch konvergieren diese Reihen, da 
wir sie konkret angeben. 
Es wird ein Differentialoperator angegeben, der die freien Energien berechnet und somit nach 
dem Theorem von Kontsevich die Schnittzahlen auf dem Modulraum der 
komplexen Kurven $\overline{\mathcal{M}}_{g,b}$ erzeugt.

F\"ur das quartische Modell geben wir die 2-Punkt-Funktion f\"ur die spektrale Dimension $<6$ explizit an. 
Hierzu wird zun\"achst eine Winkelfunktion berechnet, die nach komplexer Fortsetzung als 
effektives Ma\ss \,interpretiert werden kann. Auf dem 4-dimensionalen 
nichtkommutativen Moyalraum ist das effektive Ma\ss\, durch
eine hypergeometrische Funktion gegeben, die die spektrale Dimension effektiv 
zu $4-2\frac{\arcsin(\lambda\pi)}{\pi}$ f\"ur $|\lambda|<\frac{1}{\pi}$ \"andert. Durch diese 
effektive \"Anderung wird das Trivialit\"atsproblem des quantenfeldtheoretischen $\Phi^4$-Modells auf dem 
4-dimensionalen Moyalraum verhindert.
Schlie\ss lich zeigen wir nach kombinatorischer Analyse, wie jede planare $N$-Punkt-Funktion 
im quartischen Model explizit (nicht rekursiv) berechnet werden kann.

Die starken Unterschiede zwischen dem kubischen und quartischen Modell sind
algebraisch-geometrischer Natur. Die Bestimmung der Korrelationsfunktionen durch topologische Rekursion 
bedient sich einer algebraischen Kurve, die im kubischen Modell nur einen Verzweigungspunkt hat,
der mit der Polstelle der stabilen Korrelationsfunk\-tionen \"ubereinstimmt. Das quartische Modell hingegen 
hat im Kontinuumslimes unendlich viele Verzweigungspunkte, die zu unendlich vielen Zweigen f\"uhren.


\newpage
{\huge
\textbf{Danksagung}}
\vspace{5.ex}\\
Zu Beginn m\"ochte ich mich bei meinem Betreuer Raimar Wulkenhaar bedanken, 
der sich jedes Mal die Zeit f\"ur Diskussion mit mir genommen hat. 
Seine vorausgehenden Arbeiten und seine Unterst\"utzung waren ma\ss geblich 
an der Anfertigung dieser Dissertation beteiligt. 

Desweiteren will ich mich bei Harald Grosse, Akifumi Sako, J\"org Sch\"urmann und Roland Speicher 
bedanken, die mir durch die unterschiedlichsten Unterhaltungen verschiedene
neue Blickwinkel auf meine Arbeit er\"offnet haben. Ihre Hinweise haben daf\"ur gesorgt, dass
meine Arbeit an der Dissertation kontinuierlich voranging.

Die angenehme Arbeitsatomsph\"are mit Alexander Stottmeister, Carlos I. P\'erez-S\'anchez, Jins de Jong,
Johannes Branahl, Johannes 
Th\"urigen und Romain Pascalie war durch intensive Diskussionen
erf\"ullt. 

Finanziell wurde ich w\"ahrend meiner Promotion unterst\"utzt durch das Mathematische Institut M\"unster
und die DFG mit dem 
Cluster of Excellence ''Mathematics M\"unster'', dem RTG 2149 und dem SFB 878. 

Ein gro\ss er Dank geht auch an meine Eltern und meine Freunde, 
die mich schon immer auf meinem Weg in der Physik und der Mathematik unterst\"utzt haben. 
Doch am meisten bedanke ich mich bei meinen beiden S\"ohnen, Philipp und Maxim, sowie meiner
geliebten Frau Caroline.

\newpage
\vspace*{30.ex}

\begin{center}
{\Huge \textit{
to my wife Caroline\\
  }}
\end{center}

{
  \pagestyle{empty}
  \addtocontents{toc}{\protect\thispagestyle{empty}} 
  \tableofcontents
  \clearpage
}

\pagestyle{fancy}
\chapter{Introduction}
\setcounter{page}{1} 
This thesis studies a special type of matrix models, namely matrix field theory models. 
These models have implications in modern areas of mathematics and mathematical physics which seem to be
different but connected via matrix field theory. The various types of implications and applications 
of matrix field theory to quantum field theory, 
quantum field theory on noncommutative geometry, 2D quantum gravity and 
algebraic geometry will be introduced, respectively.


\section{Quantum Field Theory}
Nature is, on fundamental level, governed by four different interactions of two separated theories in physics.
Elementary particles are described in the Standard Model with three interactions, 
the weak interaction, the strong interaction and
the electromagnetic interaction. Quantum field theory (QFT) describes the dynamics of these elementary particles 
by fundamental principles. The fourth interaction is gravity and described by the theory of general relativity.

General relativity is, from a mathematical point of view, rigorously understood.
The achievement of Einstein was to recognise that the 4-dimensional spacetime is curved by energy densities, 
and the motion occurs along geodesics.
This theory is
confirmed experimentally with astonishing precision, 
e.g. recently
by the measurement of
gravitational waves \cite{PhysRevLett.116.061102}.

Also the predictions of the Standard Model are verified day-by-day in huge particle colliders. The
theoretical prediction, for instance, for the anomalous magnetic moment 
of an electron agrees with the experimental data up to eleven decimal digits \cite{Odom:2006zz}.
However, the mathematical construction of QFT is, independent of the particle content, hard to formulate rigorously. 

Wightman formulated these fundamental principles for a QFT on Minkowski space 
with natural axioms for operator-valued tempered distributions, smeared
over the support of a test function,
on a separable Hilbert space \cite{Wightman:1956zz,Streater:1989vi}.
The first application was to show that the 4-dimensional free scalar field satisfies these axioms, which indeed holds. 
Furthermore, Wightman's powerful reconstruction theorem implies that if the full 
set of correlation functions is known, then 
under certain conditions the Hilbert space and the entire quantum field theory can be
reconstructed. Unfortunately, the axiomatic formulation of Wightman has one problem: no interacting QFT model
satisfying these axioms could be constructed in 
4D, yet. An equivalent formulation to Wightman's axioms on the Euclidean space, instead of Minkowski space,
was found by Osterwalder and Schrader \cite{Osterwalder:1973dx,Osterwalder:1974tc}. 

A different approach to QFT makes use of the path integral formalism. 
The idea behind is that a particle propagates between two points not along 
the path with extremal/minimal action, but along any path weighted by some probability. 
For a QFT, the particle is described by a field, a scalar field $\phi$ 
can be for instance a Schwartz function $\phi\in\mathcal{S}(\R^D)$. Therefore, the path integral translates 
into a sum (or even an integral) over all field configurations of the field 
content of the model \cite{Popov:1984mx}. This expression is on Minkowski space not well-defined and has, even 
on Euclidean space, a lot of technical issues. Nevertheless,
the path (or better: functional) integral formalism can be used to 
approximate correlation functions around the free theory, which is called pertubative expansion. To
make the pertubative expansion well-defined, certain parameters of the model need to be adjusted (renormalised) 
appropriately. These approximated and renormalised correlation function can then be compared via LSZ reduction 
formula \cite{Lehmann1955} to the experiment. This comparison of theory and experiment fits remarkably well.

Up to now, it is not clear whether the approximation of a correlation function 
by perturbation theory converges in any sense. The number of terms for the perturbative expansion
grows factorially from order to order. Furthermore, the values of the different terms themselves increase after 
renormalisation (renormalon problem) such that even Borel
summablility seems to be a hopeless concept \cite{PhysRev.85.631}.

It will be proved in this thesis that \textit{matrix field theory} provides non-trivial examples
for models which have the same issues as QFT models, but the pertubative expansion is indeed convergent,
in fact we will determine the function it converges to. 
We will define the dimension of the matrix field theory model in the natural sense given by Weyl's law \cite{Weyl1911}.
The entire machinery of renormalisation will be necessary, as in QFT, to generate finite results for the 
perturbative expansion. We will see for selected examples that the number and the value of the 
terms grow for the perturbative expansion factorially,
just like in QFT. The exact results of the correlation function will be computed 
directly and coincide with perturbative expansion after applying Zimmermann's forest formula for renormalising 
all divergences and subdivergences.

From these examples, the following question arises:
What are the mathematical conditions that the perturbative expansion (in the sense of QFT) converges?

\section{Quantum Field Theory on Noncommutative Geometry}
As mentioned before, QFT is described on a flat spacetime (Minkowski or Euclidean space). Since
the theory of general relativity implies a curved spacetime, a natural question is whether both theories 
can be combined, which is first of all not the case.
For instance, Heisenberg's uncertainty relation of quantum mechanics
implies for a spherical symmetric black whole 
(solution of Einstein's field equation in general relativity) an uncertainty of the Schwarzschild radius. Applying this
to a quantum field yields that the support of the quantum field cannot be localised better than the Planck scale
$l_P= \sqrt{\frac{G\hbar}{c^3}}$, where $G$ is Newton's constant, $\hbar$ Planck's constant and $c$ 
the speed of light \cite{Misner1973}.

Noncommutative geometry can avoid this gravitational collapse caused by localising events 
with extreme precision \cite{Doplicher:1994tu}. The coordinate uncertainties have to satisfy certain inequalities
which are induced by noncommutative coordinate operators $\hat{x}^\mu$ satisfying 
$[\hat{x}^\mu,\hat{x}^\nu]=\mathrm{i} \hat{\Theta}^{\mu\nu}$, where $\hat{\Theta}^{\mu\nu}$ are the components of a
2-form with the properties $\langle \hat{\Theta},\hat{\Theta}\rangle=0$ and $\langle \hat{\Theta},* \hat{\Theta}\rangle
=8l^4_P $ in 4D.

This suggests that, if QFT and gravity (in the classical sense of general relativity) 
are combined, spacetime itself 
should be noncommutative. First examples of scalar QFTs on noncommutative spaces face in the perturbative expansion
the problem of mixing ultraviolet and infrared divergences \cite{Minwalla:1999px}. This mixing problem was solved 
by adding a harmonic oscillator term depending on $\hat{\Theta}$ to the action \cite{Grosse:2004yu}. 
The most natural example of
a scalar QFT is the quartic interacting model, the Grosse-Wulkenhaar model, which was proved to be renormalisable to all order in
perturbation theory, a necessary condition for a QFT \cite{Grosse:2004yu}. 

The representation of a scalar QFT model on a noncommutative space (especially on the Moyal space)
is approximated in momentum space by large matrices \cite{GraciaBondia:1987kw}. At the self-dual point 
\cite{Langmann:2002cc}, this type of model becomes a \textit{matrix field theory} model with a special choice 
for the external matrix (or better the Laplacian) defining the dynamics. The QFT model itself is reconstructed in the 
limit of infinitely large matrices.

\section{2D Quantum Gravity}
Quantum gravity designs a different approach
to combine QFT and gravity. Spacetime, 
and therefore gravity itself, is quantised in the sense of a quantum field. 

Remarkable 
results were achieved for quantum gravity in 2 dimensions, since orientable 
manifolds of dimension 2 are Riemann surfaces which are simpler than higher dimensional 
manifolds. 
The quantisation of gravity implies (in the sense of the path integral formalism) an average of special 
weights (corresponding to the physical theory) over all
geometries of Riemann surfaces. 

One way of doing so was by discretising the Riemann surfaces into polygons which are glued together.
The dual picture of a discretisation of a Riemann surface is a ribbon graph such that a sum over 
discretised Riemann surfaces can be performed as a sum over the dual ribbon graphs \cite{DiFrancesco:1993cyw}. 
In analogy to 
the perturbative expansion of QFT, ribbon graphs are generated by the Hermitian 1-matrix models. To end up in finite volumes
for the Riemann surfaces in the continuum limit, the size of the polygons has to tend to zero, whereby the number of the 
polygons tends to infinite (double-scaling limit). 
Conjecturally, matrix models should 
provide 2-dimensional quantum gravity in this double-scaling 
limit, which was for a long time not understood rigorously. 

A second approach to 2D quantum gravity was formulated by Polyakov \cite{Polyakov:1981rd} under the name 
of Liouville quantum gravity. His idea was to sum over all metrics on a surface instead of summing over all surfaces. 
In 2 dimensions, any metric can be transformed in a conformal form, i.e. it is after transformation diagonal
and
characterised by a scalar, the Liouville
field which can be coupled to gravity. The Jacobian to achieve the conformal form of the 
metric is called Liouville action which is by itself conformally 
invariant. This conformal invariance gives strong conditions on the correlation functions given by representations
of the Virasoro algebra (due to the conformal group). 
Finite representations of the conformal group are
classified by Kac's table
into $(p,q)$-minimal models, which implies that the partition function of a conformal field theory coupled to 
gravity is a $\tau$-function of KdV hierarchy (nonlinear partial differential
equation of Painlev\'e type) \cite{DiFrancesco:639405}.

Heuristic asymptotics yield the guess that the partition function of matrix models is in the 
double-scaling limit a $\tau$-function of a $(p,q)$-minimal model. In other words,
the partition function of a matrix model satisfies a partial differential equation in the double-scaling limit. 
This conjecture was 
later proved rigorously (see e.g. \cite[Ch. 5]{Eynard:2016yaa}). Consequently,
2D quantum gravity was proved to be approximated by a particular discretisation of the underlying space. 

The interest in matrix models increased due to the relation to Liouville quantum gravity. 
Further examples of matrix models were investigated. 
The Kontsevich model \cite{Kontsevich:1992ti} had even higher impact
which is the 
first non-trivial example for a \textit{matrix field theory}, where the attention of this thesis lies on. 
The ribbon graph expansion consists of weighted graphs with only trivalent vertices. Unexpectedly, 
the Kontsevich model was proved to be in the limit of infinite matrix size equivalent 
to the Hermitian 1-matrix model by a certain choice of the parameters, the 
so-called Miwa-transformation (or Kontsevich times) \cite{Ambjorn:1993sj}.
Hence, the Kontsevich model, as first non-trivial example for a matrix field theory, agrees with the 
$\tau$-function of KdV hierarchy and is therefore also an counterintuitive approximation for 2D quantum gravity.

\section{Algebraic Geometry}
A third approach to 2D quantum gravity goes back to concepts of algebraic geometry. 
This approach (so-called topological gravity) tries to take the sum over all Riemann surfaces up to holomorphic reparametrisations. 
The set of Riemann surfaces for given topology modulo 
holomorphic reparametrisation is called moduli space which is a finite dimensional complex variety. 
For the interest of quantum gravity, an integral over the moduli space (or better its compactification) 
should be performed. 
A volume form on the moduli space is constructed from wedging the Chern classes of the line bundles which are naturally 
constructed by the cotangent spaces at the marked points of the Riemann surface. If these forms are of top dimension, 
then the integral over the compactified moduli space provides a nonvanishing rational number, which is 
called the intersection number.
These numbers are topological invariants characterising the corresponding moduli space.

The original motivation of integrating over the moduli spaces coming from 2D quantum gravity 
inspired Witten to his famous conjecture \cite{Witten:1990hr} that the 
generating function of the intersection numbers of stable
Riemann surfaces ($=$ stable complex curves) is a $\tau$-function of KdV hierarchy.
Liouville quantum gravity is related to the 
KdV hierarchy. Otherwise stated, the approach of Liouville quantum gravity 
and the approach of topological gravity are equivalent. 

This conjecture was proved by Kontsevich \cite{Kontsevich:1992ti} by relating the 
generating function for a special choice of the formal parameters (Kontsevich times) to the 
weighted ribbon graphs generated by the Kontsevich model. As mentioned before, the Kontsevich model is the easiest 
example for a \textit{matrix field theory} and satisfies 
via the connection to Hermitian 1-matrix models the PDE's of the KdV hierarchy.

Intensive studies on matrix models have shown that also the correlation functions (and not only the 
partition function) of the Hermitian 1-matrix model and the 
Kontsevich model 
are related in some sense. 
The correlation functions obey the same type of recursive relations, the so-called topological recursion. 
The beauty of topological recursion is that for a given initial data (the spectral curve) 
topological recursion universally produces symmetric meromorphic functions \cite{Eynard:2007kz}. 
These are, in the case of matrix models, the correlation functions of the corresponding model. 

Topological recursion provides a modern formulation of the 
equivalence between algebraic geometric numbers and geometric models. 
In the last few years, special choices of 
the spectral curve have produced via topological recursion numbers of algebraic geometric
significance, e.g. 
Hurwitz numbers \cite{Bouchard:2007hi}, Weil-Petersson volumes of moduli spaces \cite{Mirzakhani:2006fta},
Gromov-Witten invariants \cite{Dunin-Barkowski:2013wca} and Jones polynomials of knot theory \cite{Borot:2012cw}.

Since the simplest matrix field theoretical model, the Kontsevich model, is known to obey topological
recursion \cite{Eynard:2007kz}, a natural question is whether also other matrix field theory models obey 
topological recursion (or any generalisation of it) and what their equivalent algebraic geometric meanings are.

Take the example of the Hermitian 2-matrix model, it 
fulfils a generalised form of topological recursion \cite{Eynard:2007gw}, 
where the algebraic geometric meaning is 
still open.
We will give hints that this model is possibly
related to the Grosse-Wulkenhaar model.

\section{Outline of the Thesis}
The thesis starts in  \sref{Ch.}{ch:matrix} with an introduction to matrix field theory in general. 
The basic definitions are given for the action of a matrix field theory, the partition function and the expectation
values. To get an intuition for these models, \sref{Sec.}{Sec:Pert}
is included which explains the perturbative expansion in detail. 
The general setting of obtaining equations and identities between expectation values (Schwinger-Dyson equation and 
Ward-Takahashi identity) is described in 
\sref{Sec.}{Sec:SDE}.
In \sref{Sec.}{Sec:LargeLimit}, a scaling limit is performed which provides matrix field theory models 
of spectral dimension greater than 0 in the sense of QFT. 
For this limit, renormalisation (\sref{Sec.}{Sec:Renorm}) is necessary which is a technique developed 
by physicists. The perturbative expansion needs for a renormalised matrix field theory a careful 
treatment by Zimmermann's forest formula (\sref{Sec.}{Sec.Zimmer}) to avoid all divergences in the scaling limit.
The chapter is finished by \sref{Sec.}{Sec.Moyal} which shows the explicit construction of 
QFT on the noncommutative Moyal space from a matrix field theory model. 

\sref{Ch.}{chap:cubic} is dedicated to the simplest matrix field theory model with cubic interaction, the
Kontsevich model. This model is solved completely in \sref{Sec.}{Sec:CubicSolution} which means that an 
algorithm is given to compute exactly any correlation function for any spectral dimension $\D<8$. 
The Kontsevich model is for higher spectral dimension $\D\geq 8$ nonrenormalisable.
The main theorems for the algorithm are \sref{Theorem}{finaltheorem} and \sref{Theorem}{thm:G-residue}.
The free energies (and therefore the intersection numbers on the moduli space of stable complex curves) 
are determined in \sref{Sec.}{Sec:CubicFreeEnergy} 
via a Laplacian.

The case of quartic interaction (known as Grosse-Wulkenhaar model) 
is developed in \sref{Ch.}{ch:quartic}. The total set
of Schwinger-Dyson equations is derived in \sref{Sec.}{Sec.quartSD}. The initial step in computing all 
correlation function starts for the quartic model with the 2-point correlation function 
described in \sref{Sec.}{Sec.quartSolution}. The exact solution of this function is given in 
\sref{Theorem}{prop:HT} for spectral dimension $\D<6$, where the two important special cases of finite matrices and 
on the 4-dimensional Moyal space are explained in \sref{Sec.}{sec.fm} and \sref{Sec.}{Sec.4dSol}, respectively.
We give in \sref{Sec.}{Sec.quartHO} an outline for the correlation function with higher topology. In the planar case 
with one boundary (of arbitrary length), the entire combinatorial 
structure is analyised in \sref{Sec.}{Sec.quartRec}.

To make the thesis fluently readable, a lot of technical details are outsourced to the appendix.  
Basic properties of the Moyal space and the description of Schwinger functions on it are 
found in \sref{App.}{App:Moyal} and \sref{App.}{App:Schwinger}, respectively.
The proof of \sref{Theorem}{finaltheorem} is split in several lemmata in \sref{App.}{appendixC}.
An important cross-check for the validity of the results is derived in 
\sref{App.}{App:Pert} by perturbative calculations with Feynman graphs and Zimmermann's
forest formula. Additionally, the perturbative analysis of the quartic model on the 4-dimensional Moyal space
is discussed in much more detail in \sref{App.}{App:Solv}. Examples for the combinatorial 
constructions used in \sref{Sec.}{Sec.quartRec}
are given in \sref{App.}{App:Expl}.
The last appendix \sref{App.}{App:3C} provides a multi-matrix field theory model which interestingly shares properties of
both models,
the cubic model of \sref{Ch.}{chap:cubic} and the quartic model of \sref{Ch.}{ch:quartic}.

\chapter{Matrix Field Theory}\label{ch:matrix}

Let $\hat{\Phi}=\hat{\Phi}^\dagger$ be a self-adjoint compact linear operator on an infinite-dimensional 
Hilbert space $\mathcal{H}$. The operator $\hat{\Phi}$ will be called the \textit{field}.
Let $E$ be an unbounded self-adjoint positive operator on $\mathcal{H}$
with compact resolvent $(E-z\mathds{1})^{-1}$.
The following class of so-called \textit{action functionals} will be considered
\begin{align}\label{eq:OpAction}
 S[\hat{\Phi}]=V \Tr\left(E \hat{\Phi}^2+\mathcal{V}(\hat{\Phi})\right),
\end{align}
defined on some subspace of compact operators depending on $E$.
The trace and the products are defined for operators on Hilbert space $\mathcal{H}$. 
The \textit{potential} $\mathcal{V}$ is of the form
$\mathcal{V}(\phi)=\sum_{k=3}^d \frac{\lambda_k}{k}\hat{\Phi}^k$ with \textit{coupling constants} $\lambda_k\in\R$, 
and the parameter $V\in \R$ will be specified later.

The action implies that $E$ is interpreted as a Laplacian, where a canonical dimension 
is induced according to Weyl's theorem \cite{Weyl1911} by the asymptotic behaviour of
spectrum of $E$:
\begin{dfnt}\label{Def:Spec}
 The operator $E$ encodes the \textit{dimension} $D:=\lfloor \D\rfloor$ via the
 \textit{spectral dimension}
 \begin{align*}
  \D:=\inf\{ p\in\R_+\,:\;\frac{1}{V}\Tr[(1+E)^{-p/2}]<\infty\}.
 \end{align*}
\end{dfnt}
\noindent
Equivalently, the spectral dimension can be expressed by the spectral measure $ d\varrho(t)$ which is uniquely 
determined by $E$ by the spectral theorem of unbounded self-adjoint operators.
The spectral dimension can therefore be computed additionally by
\begin{align}\label{eq:SpectralDim}
 \D=\inf\{ p\in\R_+\,:\;\int \frac{\, d\varrho(t)}{(1+t)^{p/2}}<\infty\},
\end{align}
where it is not necessarily integer-valued. We will give an explicit example (see \sref{Sec.}{Sec.4dSol}) 
where the 
Laplacian can effectively be changed $E\to E_\lambda$, or $\varrho\to\varrho_\lambda$, through 
the potential $\mathcal{V}$ and 
therefore 
the spectral dimension $\D\to\D_\lambda$ as well.

The field operator $\hat{\Phi}$ can be approximated
by a matrix $\Phi$ of finite rank $(\mathcal{N}+1)$, where the taken topology
depends on the specific problem. The spectral theorem of finite-dimensional spaces implies
the existence of a $\star$-homomorphism to Hermitian 
matrices $\Phi\mapsto (\Phi_{nm})$ in which $E$ is projected by the projection $P$ 
to a Hermitian $(\mN+1)\times (\mN+1)$-matrix. 
Let $(E_n)_{n=0}^\mN$ be the eigenvalues of the projection $PEP$.
We can choose without loss of generality that $E=(E_n\delta_{n,m})$ is diagonal
with ordered eigenvalues $E_n\leq E_{n+1}$, since the later defined partition function \eqref{eq:Part1} is invariant 
under a global unitary transformation $\Phi\mapsto U^\dagger \Phi U$ which can diagonalise $E$. 
The action \eqref{eq:OpAction} is under the $\star$-homomorphism after symmetrisation of the kinetic term of the form
\begin{align}\label{eq:Action}
 S[\Phi]=&V\left(\sum_{n,m=0}^\mathcal{N}
 \frac{H_{nm}}{2}\Phi_{nm}
 \Phi_{mn}+\sum_{k=3}^d\frac{\lambda_k}{k}
 \sum_{n_1,..,n_k=0}^\mathcal{N}
 \Phi_{n_1n_2}..\Phi_{n_kn_1}\right),\\
 H_{nm}:=&E_{n}+E_{m}.\nonumber
\end{align}
Notice that $H_{nm}$ is not a $(\mN+1)\times (\mN+1)$-matrix. It can be understood as the $(\mN+1)^2\times (\mN+1)^2$-matrix
$(E\otimes \mathbb{I}+\mathbb{I}\otimes E)_{nm}=H_{nm}$, where its inverse takes the r\^ole of a free propagator.

Field theories in general have got a dynamical 
construction such that a propagating field carries an energy or momentum dependence.
Matrix models are known to 
provide exact results which are established by the existence of a high symmetry. 
A \textit{matrix field theory} combines both approaches, 
where
the dynamics is considered by the external matrix $E$.
 For constant $E$, the usual Hermitian 1-matrix model is recovered.

\begin{rmk}
Staying in the subspace of 
finite matrices would define only $D=0$ matrix field theory models, since
\begin{align*}
 \frac{1}{V}\Tr((1+E)^{-p/2})=\frac{1}{V}\sum_{n=0}^\mN\frac{1}{(1+E_n)^{p/2}}=\int
 \frac{dt\,\frac{1}{V}\sum_{n=0}^\mN \delta(t-E_n)}{(1+t)^{p/2}}
 =\int \frac{\varrho(t)\,dt}{(1+t)^{p/2}}
\end{align*}
is finite for any $p>0$. The $\mN\to\infty$ limit is of greater interest which is conveniently
combined with a $V\to \infty$ limit to achieve $\D>0$.
\end{rmk}
\noindent
The action \eqref{eq:Action} gives rise to a well-defined definition of the
\textit{partition function} $\Z[J]$ depending on the Hermitian
$(\mN+1)\times(\mN+1)$-matrix $(J_{nm})$ (called the \textit{source}) by
\begin{align}      \label{eq:Part1}
 \Z[J]=\int_{\text{formal}} D\Phi \exp\left(-S[\Phi]+V \Tr(J \Phi) \right).
\end{align}
The subscript \textit{formal} means that $e^{-V\mathrm{Tr}(\mathcal{V}(\phi))}$ 
is expanded as a formal series and the order of the
integral and the series
is exchanged. We will skip this subscript from now on.
The integration is over all Hermitian
$(\mN+1)\times(\mN+1)$-matrices $(\Phi_{nm})$ with Lebesgue measure 
$D\Phi :=\prod_{n<m}d\Phi^{I}_{nm}\prod_{n\leq m}d\Phi_{nm}^R$.
Each variable is separated in the real and imaginary part $\Phi_{nm}=\Phi^R_{nm}+\I \Phi_{nm}^I$ with
$\Phi^R_{nm}=\Phi^R_{mn}$ and $\Phi_{nm}^I=-\Phi_{mn}^I$ such that the partition function is defined over 
a $
 \frac{(\mN+2)(\mN+1)}{2}+\frac{(\mN+1)\mN}{2}=(\mN+1)^2
$-dimensional space.
A partial derivative with respect to the source $J_{nm}$ produces a factor $\Phi_{mn}$ in the integrand. 
This means we have a correspondence $\Phi_{nm}\leftrightarrow
\frac{1}{V}\frac{\partial}{\partial J_{mn}}$, or more explicitly 
\begin{align*}
 &\frac{1}{V}\frac{\partial}{\partial J_{mn}}\int D\Phi f(\Phi)
 \exp\left(-S[\Phi]+V \Tr(J \Phi) \right)\\=&
 \int D\Phi\frac{f(\Phi)}{V}\frac{\partial}{\partial J_{mn}}
 \exp\left(-S[\Phi]+V \Tr(J \Phi) \right)\\=&
 \int D\Phi\,f(\Phi) \Phi_{nm}\exp\left(-S[\Phi]+V \Tr(J \Phi) \right).
\end{align*}
Two partial derivatives commute $ \frac{\partial^2}{\partial J_{nm}\partial J_{lk}}=
\frac{\partial^2}{\partial J_{lk}\partial J_{nm}}$. We employ  
this correspondence to rewrite the interaction term (or the potential) 
$S_{int}[\Phi]:=V\Tr(\mathcal{V}(\Phi))$ as
\begin{align*}
 S_{int}[\Phi]\leftrightarrow S_{int}\left[\frac{1}{V}\frac{\partial}{\partial J}\right].
\end{align*}
We will further combine 
the kinetic and the source term to
\begin{align*}
&
\left(
 \frac{H_{nm}}{2}\Phi_{nm}
 \Phi_{mn}-J_{nm}\Phi_{mn}\right)\\
 =&
 \frac{H_{nm}}{2} 
 \left(\Phi_{nm}-\frac{J_{nm}}{H_{nm}}\right)
 \left(\Phi_{mn}-\frac{J_{mn}}{H_{nm}}\right)-
 \frac{J_{nm}J_{mn}}{2H_{nm}}
\end{align*}
for any $n,m$.
Transforming the variables $\Phi_{nm}\mapsto \Phi'_{nm}=
\Phi_{nm}+\frac{J_{nm}}{H_{nm}}$ with obviously invariant measure 
$D\Phi=D\Phi'$
leads to a very useful form of the partition function
\begin{align}     \label{eq:Part2}
 \Z[J]=& K \exp\left(-S_{int}\left[\frac{1}{V}\frac{\partial}{\partial J}\right]\right)
 \Z_{free}[J],\\\nonumber
 \text{with}\qquad \Z_{free}[J]:=&\exp\left(V\sum_{n,m=0}^\mathcal{N}\frac{J_{nm}J_{mn}}{2H_{nm}}\right),
\end{align}
and constant $K:=\int D\Phi \exp \left(-V\sum_{n,m=0}^\mathcal{N}
 \frac{H_{nm}}{2}\Phi_{nm}\Phi_{mn}\right)=\prod_{n,m=0}^\mN \sqrt{\frac{2 \pi}{VH_{nm}}}$.
 
The partition function gives rise to a definition of expectation values.
We are mainly interested in the connected expectation values of a theory
\begin{align}\label{eq:Expectation}
 \langle \Phi_{p_1q_1}\Phi_{p_2q_2}..\Phi_{p_Nq_N}\rangle_c,
\end{align}
which is the connected part of the
full expectation value defined by
\begin{align}\label{eq:FullExpectation}
 \langle \Phi_{p_1q_1}\Phi_{p_2q_2}..\Phi_{p_Nq_N}\rangle :=\frac{\int D\Phi
 \, \Phi_{p_1q_1}\Phi_{p_2q_2}..\Phi_{p_Nq_N} e^{-S[\Phi]}}{\int D\Phi
 \,  e^{-S[\Phi]}}.
\end{align}
The full expectation value is given in terms of the connected ones by
\begin{align*}
 \langle \Phi_{p_1q_1}\Phi_{p_2q_2}..\Phi_{p_Nq_N}\rangle=\sum_{\text{Partitions}}
 \langle \Phi_{p_{i^1_1}q_{i^1_1}}..\Phi_{p_{i^1_{k^1}}q_{i^1_{k^1}}}\rangle_c
 ..\langle \Phi_{p_{i^j_1}q_{i^j_1}}..\Phi_{p_{i^j_{k^j}}q_{{i^j_{k^j}}}}\rangle_c,
\end{align*}
where the sum over partitions is understood as a sum over all possible 
decompositions.
An equivalent definition of the connected expectation value is obtained by the correspondence between $\Phi$ and the 
derivatives wrt to $J$
\begin{align*}
 \frac{1}{V^N}\frac{\partial^N}{\partial J_{p_1q_1}\partial J_{p_2q_2}..\partial J_{p_Nq_N}}\log \frac{\Z[J]}{\Z[0]}\bigg\vert_{J=0}=
 \langle \Phi_{p_1q_1}\Phi_{p_2q_2}..\Phi_{p_Nq_N}\rangle_c.
\end{align*}
 The numbers $p_i,q_j\in\{0,..,\mN\}$ give different types of
restrictions to the expectation values. To get an understanding which
$p_i,q_j$'s produce a non vanishing 
expectation value, and how an expectation value might look
one would first look at the perturbative expansion.

\section{Perturbation Theory}\label{Sec:Pert}
Perturbation theory is a mathematical method to approximate a result or a solution which can
possibly not expressed exactly. Quantum field theory has a perturbative expansion 
which is graphically described by Feynman graphs and its corresponding Feynman rules. The approximation via the perturbative
expansion fits tremendously well with experimental data. An example how far perturbation theory can run
is shown in the electron $g-2$ anomaly calculation. The complete $4^{\text{th}}$ order computation was recently 
finished \cite{Laporta:2017okg} in a long-term project. Incredible
891 4-loop QED diagrams contributed to the calculation. 

Nonperturbative results are quite rare in QFT. One possibility are
\textit{constructive 
QFT} models in lower dimensions, which were successful in the past for $D<4$  \cite{Rivasseau:1991ub}.
Another example is numerical lattice calculations 
on computer clusters \cite{Montvay:1994cy}, which gave in the last years great insights to nonperturbative QFT. 
However, constructing or determining exact results in 4 dimensions was not yet accomplished.

One natural question is whether the perturbative expansion is mathematically rigorous, and whether we can
extract nonperturbative information from it.
Take the following example as an analogue to the quartic interaction
\begin{align}  \label{eq:1Pert}
  \int_{-\infty}^\infty dxe^{-ax^2-\frac{\lambda}{4} x^4},
\end{align}
which is finite for any $\lambda>0$. A closed result of \eqref{eq:1Pert} exists
 in terms of the modified Bessel function $K_\alpha$ 
\begin{align*}
 e^{\frac{a^2}{2\lambda}}\sqrt{\frac{a}{\lambda}}
K_{\frac{1}{4}}\left(
 \frac{a^2}{2\lambda}\right)
\end{align*}
for $a\geq 0$.
The result is holomorphic in $\lambda$ in a certain domain, where $\lambda=0$ 
lies on the boundary of the holomorphicity domain. 
For the perturbative approach we expand $e^{-\frac{\lambda}{4}x^4}=
\sum_{n=0}^\infty \frac{(-\frac{\lambda}{4})^n}{n!}x^{4n}$ at $\lambda=0$ and naively 
exchange the order of the series and the integral
\begin{align*}
   \sum_{n=0}^\infty \frac{(-\frac{\lambda}{4})^n}{n!}  
   \int_{-\infty}^\infty dx\,x^{4n}  e^{-ax^2}=    \frac{1}{\sqrt{a}}
   \sum_{n=0}^\infty \left(-\frac{\lambda}{4a^2}\right)^n\frac{\Gamma(\frac{4n+1}{2})}{n!}.
\end{align*}
By ratio test, the series has a vanishing convergence radius in $\lambda$. 
Thus, the naive expansion does not reconstruct the exact result without further effort.

Borel summability addresses exactly this kind of problem. 
Let $f(z)$ be a holomorphic function with formal power series $\sum_{k=0}^\infty a_k z^k$ about $z=0$.
Define the \textit{Borel transform} by $B(t)=\sum_{k=0}^\infty \frac{a_k}{k!} t^k$ with nonvanishing
convergence radius. Furthermore, suppose $B(t)$ is 
well-defined for a neighbourhood of $t\geq 0$, then the integral
\begin{align*}
 \frac{1}{z}\int_0^\infty dt e^{-tz}B(t)
\end{align*} 
converges to $f(z)$. 

The assumption that the Borel transform can be continued analytically to the positive real line 
does not always hold and has therefore to be shown problem-specifically.

The hope is that certain QFT models are Borel summable. However,
the perturbative expansion does not indicate Borel summability because the number of graphs (introduced later)
growths at least with $\mathcal{O}(n!)$ and the amplitude of a graph can grow with 
$\mathcal{O}(n!)$ due to the renormalon problem as
renormalisation artifact. Both properties appear also in the later considered matrix field theory
models.

Nevertheless, an elegant way to classify the appearing integrals in the perturbative expansion 
was discovered by Richard Feynman: The usage of
\textit{Feynman graphs} and their associated \textit{Feynman rules}.

\subsection{Ribbon Graphs}
A \textit{ribbon graph} (not necessarily planar) consists of edges, vertices 
and faces. The end of an edge is either open or a vertex of degree $d\geq3$.
If it is open the edge is called open.
A ribbon graph is called connected if any edges are connected. Ribbon graphs without 
open edges are called vacuum ribbon graphs.
Each edge has two faces which 
can possibly be the same.
A face is open if it is attached to an open edge, otherwise it is called 
closed.
We identify two ribbon graphs if the edge-vertex connectivity and the vertex orientation is the same.

A ribbon graph can naturally be drawn on a Riemann surface, which is in particular a 2-dimensional orientable manifold. 
A Riemann surface is characterised by the genus $g$, i.e. the number of handles, and the number $b$ of
boundary components (cycles). 
Cutting a Riemann surface in a finite number of polygons (triangles, quadrangles,..)
 is dual to the picture of ribbon graphs. 
 A face of a polygon corresponds to a vertex of
 the ribbon graph, an edge of a polygon corresponds to an edge of
 the ribbon graph, and a vertex of the polygon corresponds to a face of
 the ribbon graph. For example, a triangulation corresponds
 to a ribbon graph with vertices only of degree 3. \vspace*{5ex}

\def\svgwidth{0.9\textwidth}
\begingroup%
  \makeatletter%
  \providecommand\color[2][]{%
    \errmessage{(Inkscape) Color is used for the text in Inkscape, but the package 'color.sty' is not loaded}%
    \renewcommand\color[2][]{}%
  }%
  \providecommand\transparent[1]{%
    \errmessage{(Inkscape) Transparency is used (non-zero) for the text in Inkscape, but the package 'transparent.sty' is not loaded}%
    \renewcommand\transparent[1]{}%
  }%
  \providecommand\rotatebox[2]{#2}%
  \ifx\svgwidth\undefined%
    \setlength{\unitlength}{454.97946349bp}%
    \ifx\svgscale\undefined%
      \relax%
    \else%
      \setlength{\unitlength}{\unitlength * \real{\svgscale}}%
    \fi%
  \else%
    \setlength{\unitlength}{\svgwidth}%
  \fi%
  \global\let\svgwidth\undefined%
  \global\let\svgscale\undefined%
  \makeatother%
  \begin{picture}(1,0.17759604)%
    \put(0,0){\includegraphics[width=\unitlength]{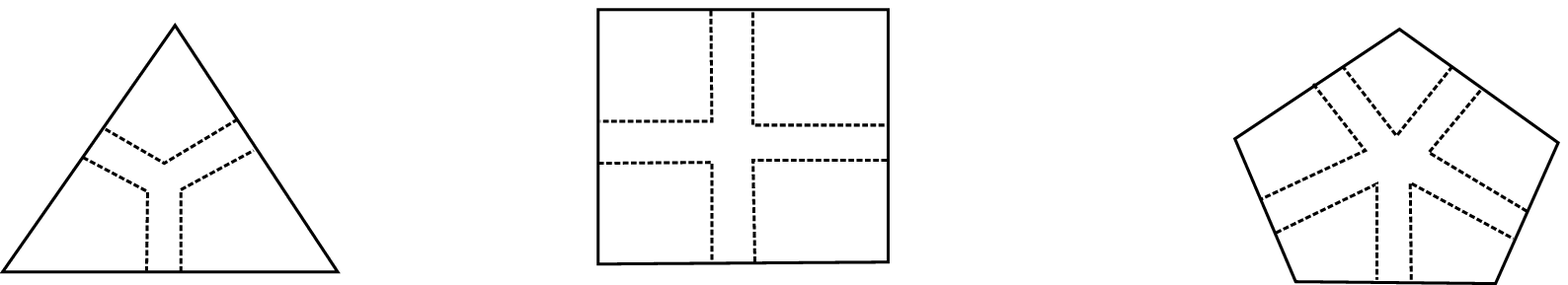}}%
  \end{picture}%
\endgroup%

\vspace*{2ex}\\
A topological invariant of a Riemann surface is the Euler characteristic
\begin{align}
 \chi=2-2g-b.
\end{align}
An open edge of a graph ends in a boundary component.
Let $e$ be the number of edges, $f$ the number of faces and $v$ the number of vertices, then the Euler characteristic is
also given by (Euler's formula)
\begin{align}
 \chi=v+f-e.
\end{align}
This formula originally holds for cutting a Riemann surface into polygons and counting the 
vertices, faces and edges after gluing. Due to the duality, this formula holds 
also for ribbon graphs since vertices and faces are exchanged 
 $v \leftrightarrow f$ under which the formula is invariant.

We will consider only the embedding of connected ribbon graphs with distinguished faces!
Two different faces of the ribbon graph should also be distinguished
after the embedding.
By this assumption, it follows that
an open face is attached to only one boundary component.
\begin{figure}[h]
    \centering
    \def\svgwidth{0.6\textwidth}
    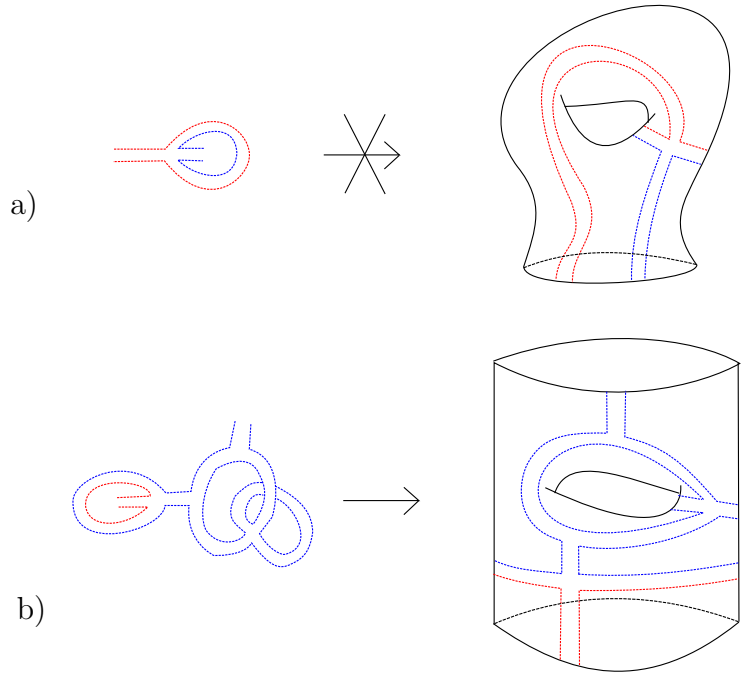
    \caption{a) The ribbon graph has two different faces (red and blue). It cannot be embedded 
    into a Riemann surface of $g=1 $ and 
    $b=1$ since the embedding would identify both faces. This ribbon graph
    is embedded into a connected Riemann surface with $g=0$ and $b=2$.   
   b) The ribbon graph has two different faces (red and blue) which are clearly distinguished after the embedding. The 
   Euler characteristic is $\chi=-2$ where
   two 3-valent vertices and two 4-valent vertices $(v=4)$, two 
    faces (red and blue) $(f=2)$ and 9 edges $(e=9)$ are embedded into a Riemann surface with$g=1$ and $b=2$.}
    \label{Fig:Embedding}
\end{figure}
If two open faces
are attached to the same open edge, they correspond to the same boundary component. If two open faces $F,F'$ are 
attached to the open edge $E$, and $F',F''$ open are attached to the open edge $E'$, 
then $F$ and $F''$ correspond to the
same boundary component. 
We define for that the infinite set
of ribbon graphs with distinguished open faces
which can be embedded into a Riemann surface

\begin{dfnt} \label{Def:Ribbon}
Let $p_1^\beta,..,p_{N_\beta}^\beta$ be for all $\beta\in\{1,..,b\}$ pairwise different.
 Let $N_{\beta}$ be the number of open faces of the $\beta^{\text{th}}$ boundary component of a connected ribbon
 graph. Label the open faces of the $\beta^{\text{th}}$ boundary component
 by the $p_i^{\beta}$'s, where two adjacent open faces are $p_i^{\beta}$ and $p_{i+1}^{\beta}$ 
 in positive orientation, $N_\beta+1\equiv1$.
Then we define the set of connected ribbon graphs with labelled open faces which are naturally embedded 
in a Riemann surface with
$b$ boundary components and genus $g$ by 
$\mathfrak{G}^{(g,b)}_{p_1^1,p_2^1,..,p_{N_b}^b}$.
\end{dfnt}

\subsection{Perturbative Expansion}
Let $\mN+1$ be the size of the matrix. Consider 
for simplicity an interaction (potential) by $\mathcal{V}(\Phi)=\frac{\lambda}{d}\Phi^d$ of degree $d$ with $d>2$. 
The perturbation theory follows from the naive expansion of $e^{-V \Tr(\mathcal{V}(\Phi))}$. 
By definition \eqref{eq:FullExpectation}, 
 the numerator of the full expectation value is therefore given by
\begin{align}\label{eq:Pert}
 \sum_{k=0}^\infty \frac{(-\frac{\lambda V}{d})^k}{k!}\int 
 D\Phi \,\Phi_{p_1q_1}..\Phi_{p_Nq_N}\left(\sum_{n_1,..,n_p=0}^\mN \Phi_{n_1n_2}..\Phi_{n_dn_1}\right)^k
 e^{-V\sum_{n,m}
 \frac{H_{nm}}{2}\Phi_{nm}
 \Phi_{mn}}.
\end{align}
The integral factorises into independent integrals of $\Phi_{ij}^R$ and $\Phi_{ij}^I$. Due to the 
Gaussian integral $\int dx f(x)e^{-x^2}=0$ for $f(x)$ odd and $\int dx\,dy (x+\mathrm{i}y)^2 
e^{-x^2-y^2}=0$, the only nonvanishing contributions for
\eqref{eq:Pert} occur if the factor in front of 
the exponential consists of complex conjugated pairs of
$\Phi_{i,j}$'s. Considering again the denominator of the full expectation value 
 \eqref{eq:FullExpectation}, all nonvanishing contributions for any $k$ in \eqref{eq:Pert} 
are after factorisation in $\Phi_{ij}^R$ and $\Phi_{ij}^I$ of the form 
\begin{align}\label{n1}
 \frac{\int_{-\infty}^\infty d\Phi^R_{pq}d\Phi^I_{pq}\,(\Phi_{pq}\Phi_{qp})^ne^{-V\left(
 H_{pq}\Phi_{pq}
 \Phi_{qp}\right)}}{\int_{-\infty}^\infty d\Phi^R_{pq}d\Phi^I_{pq}e^{-V\left(
 H_{pq}\Phi_{pq}
 \Phi_{qp}\right)}}=\frac{n!}{(VH_{pq})^n}.
\end{align}
We will call from now on $\frac{1}{H_{pq}}=\frac{1}{E_p+E_q}$ the \textit{free propagator}. 

The factorisation of \eqref{eq:Pert} into a free propagator with the pair $\Phi_{ij}$ 
and $\bar{\Phi}_{ij}=\Phi_{ji}$ gives restrictions on the $n_i$ as well as on the $p_j$ and $q_j$.
Integrating out a pair of complex conjugated $\Phi$'s is called \textit{Wick contraction}.
Since the interaction term has a cyclic order due to the trace, the $\Phi_{p_iq_i}$ are also forced to 
be of cyclic orders considering that nonvanishing integrals are necessarily of the form \eqref{n1}.
Therefore, a nonvanishing expectation value of $b$ cycles, each of \textit{length} $N_\beta$ 
with $\beta\in\{1,..,b\}$, has necessarily
the form
\begin{align*}
 \langle \Phi_{p_1^1p_2^1}\Phi_{p_2^1p_3^1}..\Phi_{p_{N_1}^1p_1^1}
 \Phi_{p_1^2p_2^2}..\Phi_{p_{N_2}^2p_1^2}..\Phi_{p_1^bp_2^b}..\Phi_{p_{N_b}^bp_1^b}\rangle
\end{align*}
and can be completely separated 
 into ''Wick-contractible'' integrals order by order in $\lambda$.

\begin{exm}\label{Ex1}
 Assume a quartic interaction $d=4$ and the expectation value $\langle \Phi_{pq}\Phi_{qp}\rangle$ at order
$\lambda^1$ with $q\neq p$. The following four contractions are possible 
\begin{align*}
\contraction{\int D\Phi \,}{\Phi}{{}_{pq}}{\Phi}
\contraction{\int D\Phi \,\Phi_{pq}\Phi_{qp} \!\!\!\!\!\!\sum_{n_1,n_2,n_3,n_4}\!\!\!\!}{\Phi}{{}_{n_1n_2}}{\Phi}
\contraction{\int D\Phi \,\Phi_{pq}\Phi_{qp} \!\!\!\!\!\!\sum_{n_1,n_2,n_3,n_4}\!\!\!\!\Phi_{n_1n_2}\Phi_{n_2n_3}
}{\Phi}{{}_{n_3n_4}}{\Phi}
 &\int D\Phi \,\Phi_{pq}\Phi_{qp} \!\!\!\!\!\!\sum_{n_1,n_2,n_3,n_4}\!\!\!\!
 \Phi_{n_1n_2}\Phi_{n_2n_3}\Phi_{n_3n_4}\Phi_{n_4n_1}e^{-V\sum_{n,m}
 \frac{H_{nm}}{2}\Phi_{nm}
 \Phi_{mn}}\sim \frac{1}{H_{pq}}\sum_{n_1,n_2,n_4}\!\!\!{\frac{1}{H_{n_1n_2}H_{n_1n_4}}}\\
 \contraction{\int D\Phi \,}{\Phi}{{}_{pq}}{\Phi}
\contraction[1.5ex]{\int D\Phi \,\Phi_{pq}\Phi_{qp} \sum_{n_1,n_2,n_3,n_4}}{\Phi}{{}_{n_1n_2}
\Phi_{n_2n_3}}{\Phi}
\contraction{\int D\Phi \,\Phi_{pq}\Phi_{qp} \sum_{n_1,n_2,n_3,n_4}\Phi_{n_1n_2}
}{\Phi}{{}_{n_2n_3}\Phi_{n_3n_4}}{\Phi}
 &\int D\Phi \,\Phi_{pq}\Phi_{qp} \sum_{n_1,n_2,n_3,n_4}
 \Phi_{n_1n_2}\Phi_{n_2n_3}\Phi_{n_3n_4}\Phi_{n_4n_1}e^{-V\sum_{n,m}
 \frac{H_{nm}}{2}\Phi_{nm}
 \Phi_{mn}}\sim \frac{1}{H_{pq}}\sum_{n_1}{\frac{1}{H_{n_1n_1}^2}}\\
 \contraction{\int D\Phi \,}{\Phi}{_{pq}\Phi_{qp} \sum_{n_1,n_2,n_3,n_4}}{\Phi}
 \contraction[1.5ex]{\int D\Phi \,\Phi_{pq}}{\Phi}{{}_{qp} \sum_{n_1,n_2,n_3,n_4}\Phi_{n_1n_2}}{\Phi}
 \contraction{\int D\Phi \,\Phi_{pq}\Phi_{qp} \sum_{n_1,n_2,n_3,n_4}
 \Phi_{n_1n_2}\Phi_{n_2n_3}}{\Phi}{{}_{n_3n_4}}{\Phi}
  &\int D\Phi \,\Phi_{pq}\Phi_{qp} \sum_{n_1,n_2,n_3,n_4}
 \Phi_{n_1n_2}\Phi_{n_2n_3}\Phi_{n_3n_4}\Phi_{n_4n_1}e^{-V\sum_{n,m}
 \frac{H_{nm}}{2}\Phi_{nm}
 \Phi_{mn}}\sim \frac{1}{H_{pq}^2}\sum_{n_4}{\frac{1}{H_{qn_4}}}\\
 \contraction[1.5ex]{\int D\Phi \,}{\Phi}{{}_{pq}\Phi_{qp} \sum_{n_1,n_2,n_3,n_4}
 \Phi_{n_1n_2}}{\Phi}
 \contraction{\int D\Phi \,\Phi_{pq}}{\Phi}{{}_{qp} \sum_{n_1,n_2,n_3,n_4}}{\Phi}
 \contraction{\int D\Phi \,\Phi_{pq}\Phi_{qp} \sum_{n_1,n_2,n_3,n_4}
 \Phi_{n_1n_2}\Phi_{n_2n_3}}{\Phi}{{}_{n_3n_4}}{\Phi}
  &\int D\Phi \,\Phi_{pq}\Phi_{qp} \sum_{n_1,n_2,n_3,n_4}
 \Phi_{n_1n_2}\Phi_{n_2n_3}\Phi_{n_3n_4}\Phi_{n_4n_1}e^{-V\sum_{n,m}
 \frac{H_{nm}}{2}\Phi_{nm}
 \Phi_{mn}}\sim \frac{1}{H_{pq}^2}\sum_{n_4}{\frac{1}{H_{pn_4}}},
\end{align*}
where the linked $\Phi$'s are \textit{Wick-contracted}. It means for example that in the fourth line $n_1=n_3=p$ and $n_2=q$. 
Wick's Theorem
says that all
possible contractions give a contribution to the expectation value of a certain order. 
Note that not all contractions give nonvanishing results
\begin{align*}
 \contraction{\int D\Phi \,}{\Phi}{_{pq}\Phi_{qp} \sum_{n_1,n_2,n_3,n_4}}{\Phi}
 \contraction[1.5ex]{\int D\Phi \,\Phi_{pq}}{\Phi}{{}_{qp} \sum_{n_1,n_2,n_3,n_4}\Phi_{n_1n_2}\Phi_{n_2n_3}}{\Phi}
  \int D\Phi \,\Phi_{pq}\Phi_{qp} \sum_{n_1,n_2,n_3,n_4}
 \Phi_{n_1n_2}\Phi_{n_2n_3}\Phi_{n_3n_4}\Phi_{n_4n_1}e^{-V\sum_{n,m}
 \frac{H_{nm}}{2}\Phi_{nm}
 \Phi_{mn}},
\end{align*}
where $n_1=n_4=q$ and  $n_2=n_3=p$. The last contraction is for $q \neq p$ not possible such that this integral 
vanishes. 
\end{exm}

\begin{exm}\label{Ex2}
 Assume a quartic interaction $d=4$ and the expectation value $\langle \Phi_{pp}\Phi_{qq}\rangle$ at order
$\lambda^1$ with $q\neq p$. We have one possible contraction 
\begin{align*}
 \contraction{\int D\Phi \,}{\Phi}{_{pp}\Phi_{qq} \sum_{n_1,n_2,n_3,n_4}}{\Phi}
 \contraction[1.5ex]{\int D\Phi \,\Phi_{pp}}{\Phi}{{}_{qq} \sum_{n_1,n_2,n_3,n_4}\Phi_{n_1n_2}\Phi_{n_2n_3}}{\Phi}
 \contraction{\int D\Phi \,\Phi_{pp}\Phi_{qq} \sum_{n_1,n_2,n_3,n_4}
 \Phi_{n_1n_2}}{\Phi}{{}_{n_2n_3}\Phi_{n_3n_4}}{\Phi}
  \int D\Phi \,\Phi_{pp}\Phi_{qq} \sum_{n_1,n_2,n_3,n_4}
 \Phi_{n_1n_2}\Phi_{n_2n_3}\Phi_{n_3n_4}\Phi_{n_4n_1}e^{-V\sum_{n,m}
 \frac{H_{nm}}{2}\Phi_{nm}
 \Phi_{mn}}\sim \frac{1}{H_{pq}H_{pp}H_{qq}}.
\end{align*}
\end{exm}
\noindent
Both examples together show that if two or more $p_i^j$'s coincide, we will have a degenerate case:
\begin{align*}
 \langle \Phi_{pp}\Phi_{pp}\rangle=\langle \Phi_{pq}\Phi_{qp}\rangle\vert_{q=p}
 +\langle \Phi_{pp}\Phi_{qq}\rangle\vert_{q=p}.
\end{align*}
Assuming pairwise different $p_i^j$'s avoids this problem.

\subsection{Feynman Rules}\label{Sec.Feynman}
Associate to each $\Phi_{nm}$ of the integrand a half edge (ribbon) either open or connected. 
The left face is labelled by the first
index and the right face with the second one.
Associate a Wick contraction to a full edge (ribbon) by connecting two half edges (ribbons) 
with coinciding faces.
Associate to any interaction term
a vertex of degree $d$ with half edges labelled by $ \Phi_{n_1n_2}..\Phi_{n_dn_1}$:
\vspace*{5ex}

\def\svgwidth{0.9\textwidth}
\begingroup%
  \makeatletter%
  \providecommand\color[2][]{%
    \errmessage{(Inkscape) Color is used for the text in Inkscape, but the package 'color.sty' is not loaded}%
    \renewcommand\color[2][]{}%
  }%
  \providecommand\transparent[1]{%
    \errmessage{(Inkscape) Transparency is used (non-zero) for the text in Inkscape, but the package 'transparent.sty' is not loaded}%
    \renewcommand\transparent[1]{}%
  }%
  \providecommand\rotatebox[2]{#2}%
  \ifx\svgwidth\undefined%
    \setlength{\unitlength}{400.16534424bp}%
    \ifx\svgscale\undefined%
      \relax%
    \else%
      \setlength{\unitlength}{\unitlength * \real{\svgscale}}%
    \fi%
  \else%
    \setlength{\unitlength}{\svgwidth}%
  \fi%
  \global\let\svgwidth\undefined%
  \global\let\svgscale\undefined%
  \makeatother%
  \begin{picture}(1,0.39079514)%
    \put(-0.00248921,0.29459304){\color[rgb]{0,0,0}\makebox(0,0)[lb]{\smash{$\Phi_{nm}$}}}%
    \put(0.20100169,0.29518455){\color[rgb]{0,0,0}\makebox(0,0)[lb]{\smash{$\Phi_{ij}$}}}%
    \put(0,0){\includegraphics[width=\unitlength,page=1]{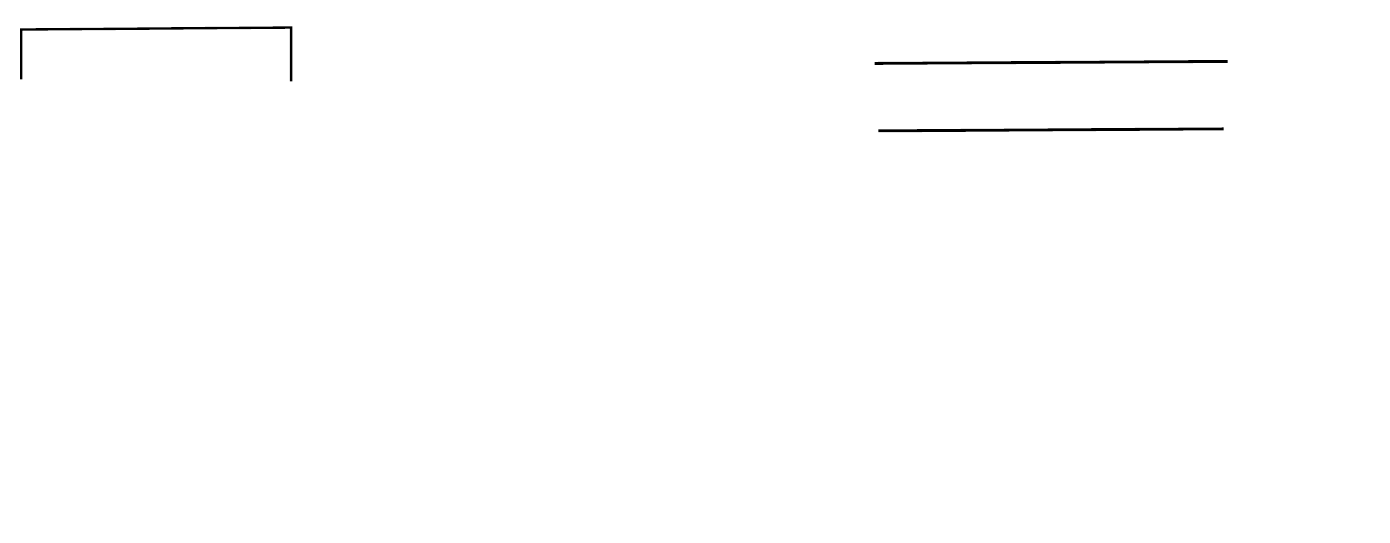}}%
    \put(0.62318799,0.31446056){\color[rgb]{0,0,0}\makebox(0,0)[lb]{\smash{$\Phi_{nm}$}}}%
    \put(0.83891912,0.31446056){\color[rgb]{0,0,0}\makebox(0,0)[lb]{\smash{$\Phi_{ij}$}}}%
    \put(0.62939605,0.36801159){\color[rgb]{0,0,0}\makebox(0,0)[lb]{\smash{$n$}}}%
    \put(0.62792376,0.26172076){\color[rgb]{0,0,0}\makebox(0,0)[lb]{\smash{$m$}}}%
    \put(0.83355306,0.36167944){\color[rgb]{0,0,0}\makebox(0,0)[lb]{\smash{$j$}}}%
    \put(0.83926505,0.2588648){\color[rgb]{0,0,0}\makebox(0,0)[lb]{\smash{$i$}}}%
    \put(0.66652353,0.26091299){\color[rgb]{0,0,0}\makebox(0,0)[lb]{\smash{}}}%
    \put(0.04705956,0.11219999){\color[rgb]{0,0,0}\makebox(0,0)[lb]{\smash{$\Phi_{n_1n_2}\Phi_{n_2n_3}..\Phi_{n_dn_1}$}}}%
    \put(0,0){\includegraphics[width=\unitlength,page=2]{Feynman.pdf}}%
    \put(0.8134587,0.10915599){\color[rgb]{0,0,0}\makebox(0,0)[lb]{\smash{.}}}%
    \put(0.85485785,0.03039668){\color[rgb]{0,0,0}\makebox(0,0)[lb]{\smash{}}}%
    \put(0.77609853,0.04049403){\color[rgb]{0,0,0}\makebox(0,0)[lb]{\smash{.}}}%
    \put(0.82658535,0.05664977){\color[rgb]{0,0,0}\makebox(0,0)[lb]{\smash{.}}}%
    \put(0.83264364,0.07684449){\color[rgb]{0,0,0}\makebox(0,0)[lb]{\smash{.}}}%
    \put(0.80639058,0.04352324){\color[rgb]{0,0,0}\makebox(0,0)[lb]{\smash{.}}}%
    \put(0.74730971,0.11473138){\color[rgb]{0,0,0}\makebox(0,0)[lb]{\smash{$\Phi_{n_2n_3}$}}}%
    \put(0.68724537,0.10836847){\color[rgb]{0,0,0}\makebox(0,0)[lb]{\smash{$\Phi_{n_1n_2}$}}}%
    \put(0.70389114,0.06627635){\color[rgb]{0,0,0}\makebox(0,0)[lb]{\smash{$\Phi_{n_dn_1}$}}}%
    \put(0.72427077,0.00707227){\color[rgb]{0,0,0}\makebox(0,0)[lb]{\smash{$n_d$}}}%
    \put(0.6374335,0.08583162){\color[rgb]{0,0,0}\makebox(0,0)[lb]{\smash{$n_1$}}}%
    \put(0.69296896,0.17569806){\color[rgb]{0,0,0}\makebox(0,0)[lb]{\smash{$n_2$}}}%
    \put(0.78788407,0.16459096){\color[rgb]{0,0,0}\makebox(0,0)[lb]{\smash{$n_3$}}}%
    \put(0,0){\includegraphics[width=\unitlength,page=3]{Feynman.pdf}}%
  \end{picture}%
\endgroup%

\vspace*{5ex}\\
\noindent
\sref{Example}{Ex1} is associated to the following four graphs\vspace*{5ex}

\def\svgwidth{0.9\textwidth}
\begingroup%
  \makeatletter%
  \providecommand\color[2][]{%
    \errmessage{(Inkscape) Color is used for the text in Inkscape, but the package 'color.sty' is not loaded}%
    \renewcommand\color[2][]{}%
  }%
  \providecommand\transparent[1]{%
    \errmessage{(Inkscape) Transparency is used (non-zero) for the text in Inkscape, but the package 'transparent.sty' is not loaded}%
    \renewcommand\transparent[1]{}%
  }%
  \providecommand\rotatebox[2]{#2}%
  \ifx\svgwidth\undefined%
    \setlength{\unitlength}{497.77229761bp}%
    \ifx\svgscale\undefined%
      \relax%
    \else%
      \setlength{\unitlength}{\unitlength * \real{\svgscale}}%
    \fi%
  \else%
    \setlength{\unitlength}{\svgwidth}%
  \fi%
  \global\let\svgwidth\undefined%
  \global\let\svgscale\undefined%
  \makeatother%
  \begin{picture}(1,0.21394005)%
    \put(0,0){\includegraphics[width=\unitlength]{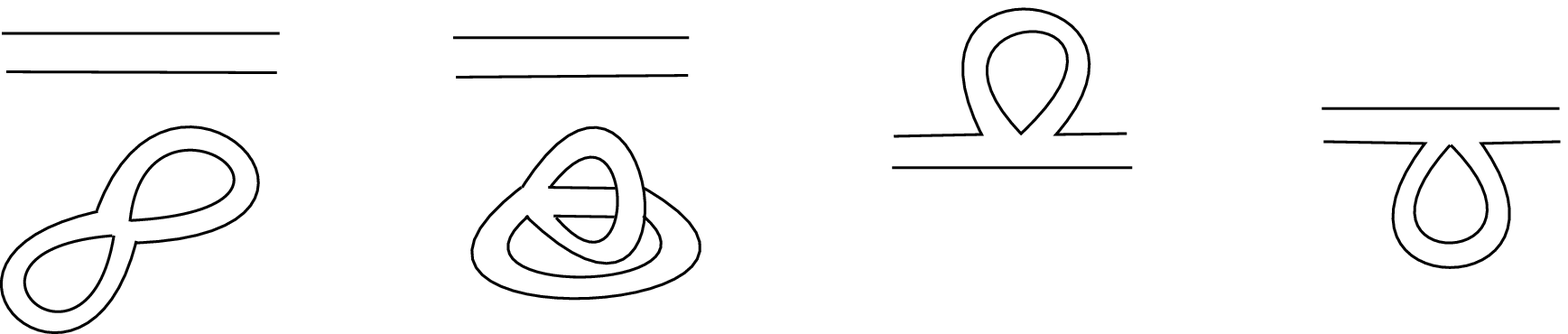}}%
    \put(0.06438471,0.20478206){\color[rgb]{0,0,0}\makebox(0,0)[lb]{\smash{$q$}}}%
    \put(0.06438471,0.14910544){\color[rgb]{0,0,0}\makebox(0,0)[lb]{\smash{$p$}}}%
    \put(0.02071747,0.0881872){\color[rgb]{0,0,0}\makebox(0,0)[lb]{\smash{$n_1$}}}%
    \put(0.097096,0.09363866){\color[rgb]{0,0,0}\makebox(0,0)[lb]{\smash{$n_4$}}}%
    \put(0.02419979,0.03509209){\color[rgb]{0,0,0}\makebox(0,0)[lb]{\smash{$n_2$}}}%
    \put(0.35728513,0.20264419){\color[rgb]{0,0,0}\makebox(0,0)[lb]{\smash{$q$}}}%
    \put(0.35097128,0.14868951){\color[rgb]{0,0,0}\makebox(0,0)[lb]{\smash{$p$}}}%
    \put(0.41565818,0.09708257){\color[rgb]{0,0,0}\makebox(0,0)[lb]{\smash{$n_1$}}}%
    \put(0.58917538,0.14294965){\color[rgb]{0,0,0}\makebox(0,0)[lb]{\smash{$q$}}}%
    \put(0.90027586,0.16131721){\color[rgb]{0,0,0}\makebox(0,0)[lb]{\smash{$q$}}}%
    \put(0.6241886,0.08555106){\color[rgb]{0,0,0}\makebox(0,0)[lb]{\smash{$p$}}}%
    \put(0.86181868,0.10219666){\color[rgb]{0,0,0}\makebox(0,0)[lb]{\smash{$p$}}}%
    \put(0.64238268,0.16768285){\color[rgb]{0,0,0}\makebox(0,0)[lb]{\smash{$n_4$}}}%
    \put(0.91158205,0.08560287){\color[rgb]{0,0,0}\makebox(0,0)[lb]{\smash{$n_4$}}}%
  \end{picture}%
\endgroup%

\vspace*{5ex}\\
The first two graphs are disconnected graphs and include vacuum graphs, respectively. The second 
vacuum graph is of genus $g=1$.\\
\sref{Example}{Ex2} is associated to the following graph\vspace*{5ex}

\hspace*{20ex}
\def\svgwidth{0.3\textwidth}
\begingroup%
  \makeatletter%
  \providecommand\color[2][]{%
    \errmessage{(Inkscape) Color is used for the text in Inkscape, but the package 'color.sty' is not loaded}%
    \renewcommand\color[2][]{}%
  }%
  \providecommand\transparent[1]{%
    \errmessage{(Inkscape) Transparency is used (non-zero) for the text in Inkscape, but the package 'transparent.sty' is not loaded}%
    \renewcommand\transparent[1]{}%
  }%
  \providecommand\rotatebox[2]{#2}%
  \ifx\svgwidth\undefined%
    \setlength{\unitlength}{104.76948476bp}%
    \ifx\svgscale\undefined%
      \relax%
    \else%
      \setlength{\unitlength}{\unitlength * \real{\svgscale}}%
    \fi%
  \else%
    \setlength{\unitlength}{\svgwidth}%
  \fi%
  \global\let\svgwidth\undefined%
  \global\let\svgscale\undefined%
  \makeatother%
  \begin{picture}(1,0.60698325)%
    \put(0,0){\includegraphics[width=\unitlength]{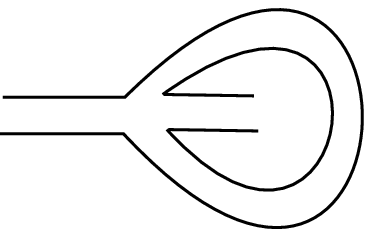}}%
    \put(0.08870338,0.44305489){\color[rgb]{0,0,0}\makebox(0,0)[lb]{\smash{$q$}}}%
    \put(0.70577009,0.41605825){\color[rgb]{0,0,0}\makebox(0,0)[lb]{\smash{$p$}}}%
  \end{picture}%
\endgroup%

\vspace*{5ex}\\
The aforementioned consideration gives a map $\mathfrak{h}$ called \textit{Feynman rules} from the set 
$\mathfrak{G}^{(g,b)}_{p_1^1,p_2^1,..,p_{N_b}^b}$ (set of \textit{Feynman graphs}) 
to all terms appearing 
in the perturbative expansion of the connected expectation value 
$\langle \Phi_{p^1_1p^1_2}\Phi_{p^1_2p_3^1}..
\Phi_{p^1_{N_1}p^1_1}..\Phi_{p^b_{N_b}p^b_1}\rangle_c$ with pairwise different $p_i^j$'s by:
\vspace*{2ex}\\
\fbox{\parbox{\textwidth}{
 \begin{itemize}
 \item An edge labelled by $p,q$ corresponds to the factor $\frac{1}{VH_{pq}}$
 \item A vertex of degree $d$ corresponds to the factor $-\lambda_d V$
 \item Take a sum $\sum_{n_i}$ over all closed faces $n_i$
\end{itemize}}}
\vspace*{2ex}\\
An example appearing in the expectation value $\langle \Phi_{p^1_1p^1_2}\Phi_{p^1_2p^1_1}\Phi_{p^2p^2}\rangle_c$ 
of genus $g=1$ and two boundary components with three open faces and one closed 
face, therefore an element of $\mathfrak{G}^{(1,2)}_{p_1^1,p_2^1,p^2}$, is
\vspace*{3ex}\\

\def\svgwidth{2\textwidth}
\begingroup%
  \makeatletter%
  \providecommand\color[2][]{%
    \errmessage{(Inkscape) Color is used for the text in Inkscape, but the package 'color.sty' is not loaded}%
    \renewcommand\color[2][]{}%
  }%
  \providecommand\transparent[1]{%
    \errmessage{(Inkscape) Transparency is used (non-zero) for the text in Inkscape, but the package 'transparent.sty' is not loaded}%
    \renewcommand\transparent[1]{}%
  }%
  \providecommand\rotatebox[2]{#2}%
  \ifx\svgwidth\undefined%
    \setlength{\unitlength}{1296.23504639bp}%
    \ifx\svgscale\undefined%
      \relax%
    \else%
      \setlength{\unitlength}{\unitlength * \real{\svgscale}}%
    \fi%
  \else%
    \setlength{\unitlength}{\svgwidth}%
  \fi%
  \global\let\svgwidth\undefined%
  \global\let\svgscale\undefined%
  \makeatother%
  \begin{picture}(1,0.13340803)%
    \put(0,0){\includegraphics[width=\unitlength,page=1]{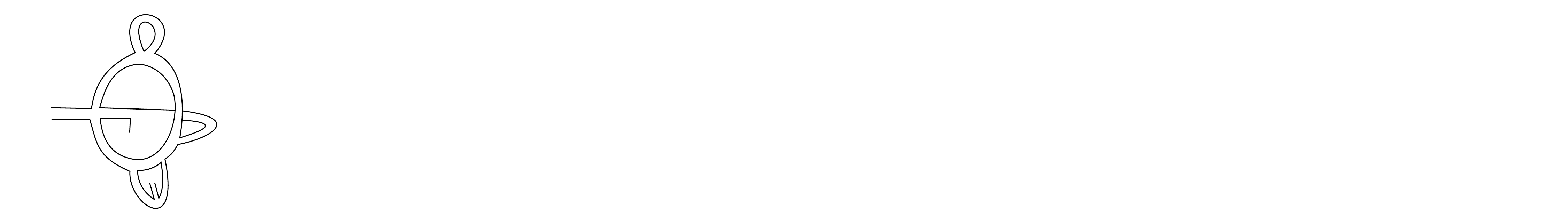}}%
    \put(0.03692109,0.0719002){\color[rgb]{0,0,0}\makebox(0,0)[lb]{\smash{${\scriptstyle p^1_1}$}}}%
    \put(0.09238732,0.01801649){\color[rgb]{0,0,0}\makebox(0,0)[lb]{\smash{$ {\scriptstyle p^2}$}}}%
    \put(0.09124358,0.11039801){\color[rgb]{0,0,0}\makebox(0,0)[lb]{\smash{$ {\scriptstyle n}$}}}%
    \put(-0.00076845,0.06254325){\color[rgb]{0,0,0}\makebox(0,0)[lb]{\smash{$\mathfrak{h}$}}}%
    \put(0,0){\includegraphics[width=\unitlength,page=2]{Graph.pdf}}%
    \put(0.17813209,0.06224498){\color[rgb]{0,0,0}\makebox(0,0)[lb]{\smash{${\displaystyle=\frac{-V^7\lambda_3^5 \lambda_4^2}{V^{13} H_{p_1^1p_2^1}^5H_{p_2^1p^2}^3H_{p^2p^2}H_{p_2^1p_2^1}^3}\sum_{n}\frac{1}{H_{p_1^1n}}}$}}}%
    \put(0.27154153,0.05303486){\color[rgb]{0,0,0}\makebox(0,0)[lb]{\smash{}}}%
    \put(0,0){\includegraphics[width=\unitlength,page=3]{Graph.pdf}}%
    \put(0.03532021,0.04797252){\color[rgb]{0,0,0}\makebox(0,0)[lb]{\smash{${\scriptstyle p^1_2}$}}}%
  \end{picture}%
\endgroup%

\vspace*{2ex}\\
Summing over all possible contractions (Wick's Theorem) corresponds to summing over all Feynman graphs. We 
conclude:
\begin{prps}
 The formal expansion of a connected expectation value 
with pairwise different $p_i^j$'s is 
\begin{align*}
 \langle \Phi_{p^1_1p^1_2}\Phi_{p^1_2p^1_3}..
\Phi_{p^1_{N_1}p^1_1}..\Phi_{p^b_{N_b}p^b_1}\rangle_c=\sum_{g=0}^\infty 
\sum_{\mathfrak{g}\in \mathfrak{G}^{(g,b)}_{p_1^1,p_2^1,..,p_{N_b}^b}}
\mathfrak{h}(\mathfrak{g}).
\end{align*}
\end{prps}\noindent

The full expectation value is given by the same formula, where the sum $\sum_\mathfrak{g}$ is taken over all 
not necessarily connected ribbon graphs, since the map $\mathfrak{h}$ has the property
\begin{align*}
 \mathfrak{h}(\mathfrak{g}\cup \mathfrak{g}')=\mathfrak{h}(\mathfrak{g}) \mathfrak{h}(\mathfrak{g}').
\end{align*}
Due to this property, all vacuum graphs are removed from the perturbative expansion of the
expectation values \eqref{eq:FullExpectation}. The vacuum graphs 
factor out in the numerator in the formal expansion, whereas the denominator produces only vacuum graphs.

\begin{rmk}\label{rmk:hermitian}
 The Hermitian 1-matrix model $(E=const.)$ has a much bigger set of graphs in its perturbative expansion.
 The expectation values 
 are defined differently by $\langle \Tr \Phi^{k_1}\Tr \Phi^{k_2}..\Tr \Phi^{k_n}\rangle$,
 where the trace gives also contributions if two or more $p_i^j$'s are equal.
 The degenerate case is therefore automatically included.
 This means from the graphical point of view that
 an open face can correspond to different boundary components. Furthermore, two faces can have the same boundary component
 even if they are 
 not connected along the edges.\\
 The expectation value $\langle \Tr \Phi^{2}\rangle=\langle \sum_{n,m} \Phi_{nm}\Phi_{mn}\rangle$ 
 has for genus 1 for all $m=n$ a contribution from the graph (different in comparison to \sref{Fig.}{Fig:Embedding} a) 
 coming from the degenerated quadrangulation.
 \vspace*{3ex}\\
 \hspace*{3ex}
 \def\svgwidth{0.9\textwidth}
\begingroup%
  \makeatletter%
  \providecommand\color[2][]{%
    \errmessage{(Inkscape) Color is used for the text in Inkscape, but the package 'color.sty' is not loaded}%
    \renewcommand\color[2][]{}%
  }%
  \providecommand\transparent[1]{%
    \errmessage{(Inkscape) Transparency is used (non-zero) for the text in Inkscape, but the package 'transparent.sty' is not loaded}%
    \renewcommand\transparent[1]{}%
  }%
  \providecommand\rotatebox[2]{#2}%
  \ifx\svgwidth\undefined%
    \setlength{\unitlength}{508.23357206bp}%
    \ifx\svgscale\undefined%
      \relax%
    \else%
      \setlength{\unitlength}{\unitlength * \real{\svgscale}}%
    \fi%
  \else%
    \setlength{\unitlength}{\svgwidth}%
  \fi%
  \global\let\svgwidth\undefined%
  \global\let\svgscale\undefined%
  \makeatother%
  \begin{picture}(1,0.34315266)%
    \put(0,0){\includegraphics[width=\unitlength,page=1]{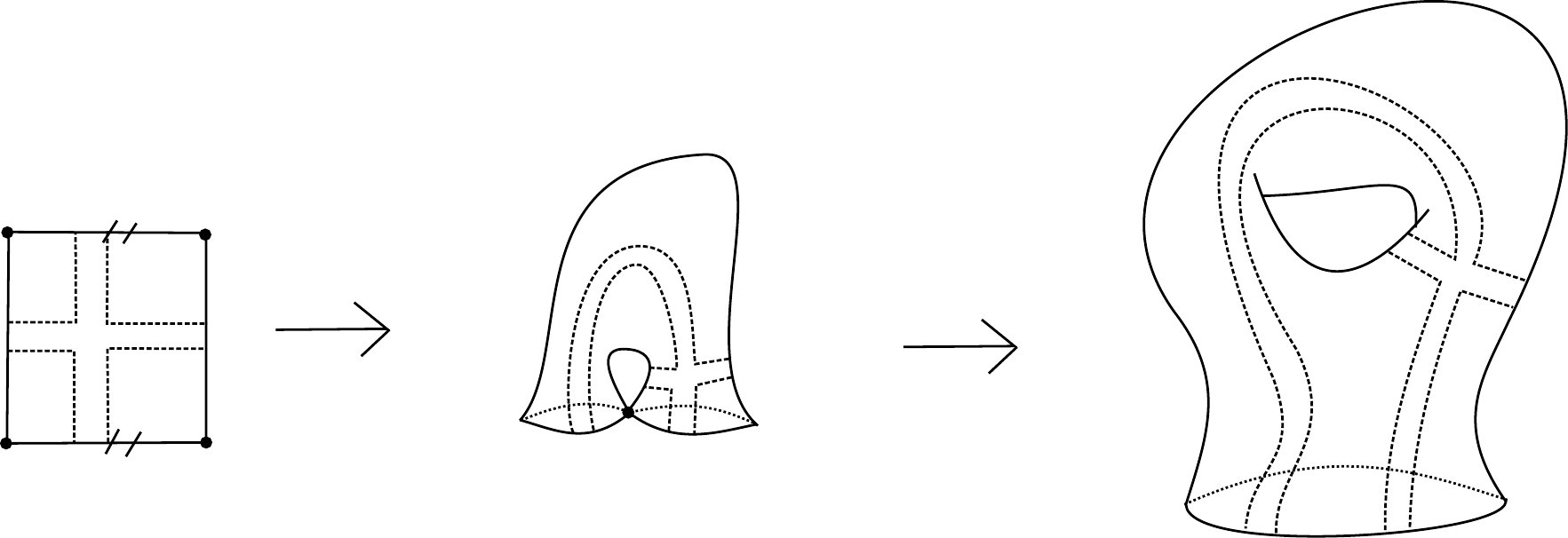}}%
  \end{picture}%
\endgroup%

\end{rmk}

\section{Schwinger-Dyson Equation and Ward-Takahashi Identity}\label{Sec:SDE}
To have a convergent $\mN\to\infty$ limit
an appropriate 
scaling in the parameter $V$ is necessary.
We assume that the sum $\sum^\mN$ and the parameter $V$ are of the same order
\begin{align*}
 \sum^\mN\sim V. 
\end{align*}
Let the interaction be $\mathcal{V}(\Phi)=\sum_{k=3}^d \frac{\lambda_k}{k}\Phi^k$. We conclude from the perturbative
expansion of $\langle \Phi_{p^1_1p^1_2}\Phi_{p^1_2q_1}..
\Phi_{p^1_{N_1}p^1_1}..\Phi_{p^b_{N_b}p^b_1}\rangle_c$ that a graph of order $k_i$ in $\lambda_i$ has the factor
\begin{align}\label{eq:VAsymptot}
 \frac{V^{k_3+k_4+..+k_d}}{V^{\frac{N+3k_3 +4k_4 +..+dk_d }{2}}},
\end{align}
where $N=N_1+N_2+..+N_b$.
On the other hand, Euler's formula gives 
\begin{align*}
 2-2g-b=\underbrace{k_3+..+k_d}_{v}+\underbrace{N+\Sigma}_{f}-\underbrace{\frac{N+k_33+..k_d d}{2}}_{e}\\
 \Rightarrow\quad \Sigma=2-2g-b-(k_3+..+k_d+N)+
 \frac{N+3k_3+..+dk_d }{2},
\end{align*}
where $\Sigma$ is the number of closed faces. Since the sum is taken over all closed faces
and any sum will be of order $V$, 
we have to multiply 
\eqref{eq:VAsymptot} by $V^\Sigma$ and achieve the scaling factor
\begin{align*}
 V^{2-2g-b-N}
\end{align*}
of $\langle \Phi_{p^1_1p^1_2}\Phi_{p^1_2q_1}..
\Phi_{p^1_{N_1}p^1_1}..\Phi_{p^b_{N_b}p^b_1}\rangle_c$
being independent of the perturbative expansion.
This gives rise to the definition of the \textit{correlation function}
which is finite in the limit $(\lim \frac{1}{V}\sum^\mN\to \text{finite})$
\begin{align}\label{eq:DefCorr}
 G_{|p_1^1..p_{N_1}^1|..|p^b_{1}..p^b_{N_b}|}:=V^{b+N_1+..+N_b-2}
 \langle \Phi_{p_1^1p_2^1}\Phi_{p_2^1p_3^1}..\Phi_{p_{N_1}^1p_1^1}..\Phi_{p_1^bp_2^b}..\Phi_{p_{N_b}^bp_1^b}\rangle_c,
\end{align}
where $p_i^j$'s are pairwise different.
We will call $G_{|p_1^1..p_{N_1}^1|..|p^b_{1}..p^b_{N_b}|}$ a $(N_1+..+N_b)$-point function which have
a formal genus $g$-expansion 
\begin{align}\label{eq:DefCorrg}
 G_{|p_1^1..p_{N_1}^1|..|p^b_{1}..p^b_{N_b}|}=:\sum_{g=0}^\infty V^{-2g}
 G^{(g)}_{|p_1^1..p_{N_1}^1|..|p^b_{1}..p^b_{N_b}|}
\end{align}
defining the $(N_1+..+N_b)$-point function of genus $g$.
The $b$ different boundary components are separated by vertical lines in the index of $G$.
The definition \eqref{eq:DefCorr} has the benefit that if two or more $p_i^j$'s coincide,
the degeneracy is separated by the correlation functions. Take the earlier discussed example
\begin{align*}
 G_{|pp|}=&V\langle \Phi_{pq}\Phi_{qp}\rangle_c\vert_{q=p}\\
 G_{|p|p|}=&V^2\langle \Phi_{pp}\Phi_{qq}\rangle_c\vert_{q=p}\\
 \langle \Phi_{pp}^2\rangle_c=&\frac{1}{V}G_{|pp|}+\frac{1}{V^2}G_{|p|p|}.
\end{align*}
A more convenient form of the definition uses
the correspondence between $\Phi$ and $J$-derivatives which shows that
the partition function generates all
correlation functions in powers of $J$
\begin{align}\label{eq:DefCorrSour}
 \log \frac{\Z[J]}{\Z[0]}=\sum_{b=1}^\infty \sum_{N_1,..,N_b=1}^\infty \sum_{p_1^1,..,p_{N_b}^b=0}^\mN
 \sum_{g=0}^\infty V^{2-b-2g}\frac{G^{(g)}_{|p_1^1..p_{N_1}^1|..|p^b_{1}..p^b_{N_b}|}}{b!}
 \prod_{\beta=1}^b\frac{\J_{p_1^\beta..p_{N_\beta}^\beta}}{N_\beta}.
\end{align}
The shorthand notation
$\J_{p_1^\beta..p_{N_\beta}^\beta}:=\prod_{i=1}^{N_\beta} J_{p_i^\beta p_{i+1}^\beta}$ with $N_{\beta}+1\equiv1$ is used. 
The symmetry factor $\frac{1}{N_\beta}$ arises since a correlation function is by definition \eqref{eq:DefCorr}
invariant under a 
cyclic permutation
in each boundary component labelled by $\beta$. The symmetry factor $\frac{1}{b!}$ arises 
since a correlation function is by definition \eqref{eq:DefCorr}
symmetric under changing two boundary components $\beta\leftrightarrow \beta'$ and therefore under changing any boundary
component.

Assuming pairwise different $p_i^j$'s a correlation function is directly extracted by partial derivatives from
equation \eqref{eq:DefCorrSour} at $J=0$:
\begin{align*}
 V^{b-2}\frac{\partial^{N_1+..+N_b}}{\partial J_{p_1^1 p_{2}^1}.. \partial J_{p_{N_b}^bp_{1}^b} }
 \log \frac{\Z[J]}{\Z[0]}\bigg\vert_{J=0}=
 G_{|p_1^1..p_{N_1}^1|..|p^b_{1}..p^b_{N_b}|}=\sum_{g=0}^\infty V^{-2g} 
 G^{(g)}_{|p_1^1..p_{N_1}^1|..|p^b_{1}..p^b_{N_b}|}.
\end{align*}
However, the interesting operation appears if the $J$-derivatives act on $\frac{\Z[J]}{\Z[0]}$
or have coinciding $p_i^j$'s. 
The first non-trivial example follows easily from \eqref{eq:DefCorrSour}:
\begin{align*}
 \frac{\partial^2}{\partial J_{pp}\partial J_{pp}}\frac{\Z[J]}{\Z[0]}\bigg\vert_{J=0}
 =&VG_{|pp|}+G_{|p|p|}+V^2G_{|p|}G_{|p|}\\
 =&\sum_{g=0}^\infty V^{-2g}\left(
 VG^{(g)}_{|pp|}+G^{(g)}_{|p|p|}+V^2\sum_{h+h'=g}G^{(h)}_{|p|}G^{(h')}_{|p|}\right)\\
 =&V^2\langle \Phi_{pp}\Phi_{pp}\rangle_c +V^2\langle \Phi_{pp}\rangle^2_c=V^2\langle \Phi_{pp}\Phi_{pp}\rangle.
\end{align*}
Considering an interaction $S_{int}[\Phi]$, the aim is to determine all correlation functions.
The first step in doing so is to determine equations between correlation functions 
which are called \textit{Schwinger-Dyson equations} (SDEs), (or loop equations). 
These equations are mainly achieved by the following trick 
\begin{lemma}\label{Lemma:Trick}
 Let $f(x)$ be a polynomial in $x$, $g(x)$ smooth and $\partial_x=\frac{\partial}{\partial x}$. 
 Define the operator $\exp\left(f\left(\partial_x\right)\right)
 :=\sum_{k=0}^\infty \frac{f(\partial_x)^k}{k!}$
 then we have
 \begin{align*}
  e^{f(\partial_x)}(x\cdot g(x))=x e^{f(\partial_x)}g(x) +f'(\partial_x) e^{f(\partial_x)}g(x).
 \end{align*}
\begin{proof}
 Expanding $e^{f(\partial_x)}$ by the definition and using the Leibniz rule to have $[f(\partial_x)^k, x] =
 k f'(\partial_x) f(\partial_x)^{k-1}$ 
 gives the rhs after shifting the index $k\to k+1$.
\end{proof}
\end{lemma}\noindent
Applying \sref{Lemma}{Lemma:Trick} with $f=S_{int}$
creates the derivative of $S_{int}$ by
\begin{align}\label{eq:SintAbl}
 \left(\frac{\partial S_{int}}{\partial \Phi_{pq}}\right)\left[\frac{1}{V}\frac{\partial}{\partial J}\right] =
 \sum_{k=3}^{d} \frac{\lambda_k}{V^{k-2}}\sum_{n_1,..,n_{k-2}=0}^\mN 
 \frac{\partial^{k-1}}{\partial J_{p n_{1}}\partial J_{n_{1}n_{2}}..\partial J_{n_{k-2} q}}.
\end{align}
The SDE is established for pairwise different $p_i^j$'s by
 \begin{align}                    \nonumber
        &G_{|p_1^1..p_{N_1}^1|..|p^b_{1}..p^b_{N_b}|}    
        =     V^{b-2}\partial \J_{p_1^1..p_{N_b}^b}
                \log \frac{\Z[J]}{\Z[0]}\bigg\vert_{J=0}\\     \nonumber
        =& V^{b-1}K\partial \J_{p_2^1..p_{N_b}^b}\frac{1}{\Z[J]} e^{-S_{int}[\frac{1}{V}\frac{\partial}{\partial J}]}
        \frac{J_{p_2^1p_1^1}}{H_{p_1^1p_2^1}}  \Z_{free}[J]\vert_{J=0} \\      \label{eq:SDeqGeneral}
        =&\frac{\delta_{b,1}\delta_{N_1,2}}{H_{p_1^1p_2^1}}
        -   \frac{V^{b-1}}{H_{p_1^1p_2^1}}\partial \J_{p_2^1..p_{N_b}^b}\frac{1}{\Z[J]}
        \left(\frac{\partial S_{int}}{\partial \Phi_{p_1^1p_2^1}}\right)
        \left[\frac{1}{V}\frac{\partial}{\partial J}\right]   \Z[J]\vert_{J=0}
 \end{align}   
 with the shorthand notation $\partial \J_{p_1^1..p_{N_b}^b}:=
\frac{\partial^{N_1+..+N_b}}{\partial J_{p_1^1 p_2^1}..
\partial J_{p_{N_1}^1 p_{N_1+1}^1}..\partial J_{p_1^b p_2^b}..\partial J_{p_{N_b}^b p_{N_b+1}^b}}$  and
$N_{\beta}+1\equiv1$.
  The second line is achieved by deriving $Z_{free}$ of \eqref{eq:Part2} with respect to $J_{p_1^1p_2^1}$
  and third line by \sref{Lemma}{Lemma:Trick} as well as considering $J=0$. Recall that 
  $\frac{\partial S_{int}}{\partial \Phi_{pq}}$ is of degree $d-1$ such that the SDE
  \eqref{eq:SDeqGeneral}
  has for instance on the rhs $((N_1+d-2)+N_2+..+N_b)$-point functions with $d-1$-summations, or 
  different types of nonlinear terms if one or more summation indices coincide with other indices (degenerate case).
 The SDEs form a tower of equations. Writing precisely the first equations shows 
 that the 1-point function depends on the 
 2-point, the 2-point on the 3-point and so on. The tower of equations does not decouple and is therefore
 inefficient to determine correlation functions explicitly. 
 
 Notice that equation \eqref{eq:SDeqGeneral} has a certain base point $J_{p_1^1p_2^1}$. The two variables 
 $p_1^1$ and $p_2^1$ play a special r\^ole. The resulting correlation function, however, obeys
 the cyclic symmetry
 in each boundary and an additional symmetry between the boundaries. In other words,
 a highly symmetric function fulfils an equation which is non-symmetric 
 in one of its variables. 
 
 Nevertheless, a decoupling of the tower is possible and achieved in the large $V$-limit and by the 
 \textit{Ward-Takahashi identity} arising from a 
 symmetry transformation of the partition function $\Z[J]$, first derived in \cite{Disertori:2006nq}. 
 Transform the integration variables $(\Phi_{nm})$ of the 
 partition function \eqref{eq:Part1} by a unitary transformation $\Phi\mapsto \Phi'=U\Phi U^\dagger$ with $U\in
 \mathrm{U}(\mN+1)$. The partition function is invariant under this transformation since 
 the property of Hermitian matrices is preserved.
 From the invariant measure $\mathrm{det}\frac{\partial U\Phi U^\dagger}{\partial \Phi}
 =(\mathrm{det} U\,\mathrm{det}U^\dagger)^{\mN+1}=1$, it follows
 \begin{align}
  0=\int D\Phi \,\exp(-S[\Phi]+V\Tr(J\Phi))-\exp(-S[\Phi']+V\Tr(J\Phi')).
 \end{align}
Expanding the unitary transformation about the identity $U=\mathbb{I}+\mathrm{i}\varepsilon A+\mathcal{O}(\varepsilon^2)$
gives the transformed matrix $\Phi'=\Phi +\mathrm{i}\varepsilon (A\Phi-\Phi A)+\mathcal{O}(\varepsilon^2)$  
and finally, at order $\varepsilon^1$,
\begin{align}\label{eq:WardHer}
 0= \int D\Phi\, (E\Phi^2-\Phi^2 E-J\Phi+\Phi J)_{pq}\exp(-S[\Phi]+V\Tr(J\Phi)).
\end{align}
Note that the interaction is invariant under unitary transformation $S_{int}[\Phi]=S_{int}[\Phi']$, 
whereas the kinetic term and the source 
term are not. Since \eqref{eq:WardHer} vanishes for any interaction $S_{int}$, the factor in front of the exponential 
has to vanish. We obtain after applying the correspondence between $\Phi$ and $J$-derivatives 
\begin{prps}(\cite{Disertori:2006nq})\label{Prop:WardId}
 The partition function $\Z[J]$ obeys the Ward-Takahashi identity
 \begin{align*}
  0=\sum_{n=0}^\mN \left(\frac{E_p-E_q}{V}\frac{\partial^2}{\partial J_{qn}\partial J_{np}}+
  J_{nq}\frac{\partial}{\partial J_{np}}-J_{pn}\frac{\partial}{\partial J_{qn}}\right)\Z[J].
 \end{align*}
\end{prps}\noindent
The Ward-Takahashi identity reduces a second-order derivative together with a sum over the intermediate index 
to a first-order derivative with an additional factor in $J$. Precisely this reduction can decouple the 
tower of SDEs in the limit $\mN,V\to\infty$. We emphasise that the decoupling is
possible for an arbitrary interaction term $S_{int}$. The Ward-Takahashi identity has to be applied therefore
possibly several times.

We can further generalise \sref{Proposition}{Prop:WardId} 
by a direct calculation instead of unitary transformation
(observed in \cite{Hock:2018wup})
\begin{prps}\label{Prop:GenWardId}
 The partition function $\Z[J]$ obeys the generalised Ward-Takahashi identity
 \begin{align*}
  &\left\{\left(\frac{\partial S_{int}}{\partial \Phi_{nq}}\right)\left[\frac{1}{V}\frac{\partial}{\partial J}\right]
  \frac{\partial}{\partial J_{np}}-
  \left(\frac{\partial S_{int}}{\partial \Phi_{pn}}\right)\left[\frac{1}{V}\frac{\partial}{\partial J}\right]
  \frac{\partial}{\partial J_{qn}}\right\}\Z[J]
  \\
  =&\left(\frac{E_p-E_q}{V}\frac{\partial^2}{\partial J_{qn}\partial J_{np}}+
  J_{nq}\frac{\partial}{\partial J_{np}}-J_{pn}\frac{\partial}{\partial J_{qn}}\right)\Z[J].
 \end{align*}
 \begin{proof}
  Direct computation leads to
  \begin{align*}
			&\frac{E_{p}-E_{q}}{V}
			\frac{\partial^2}{\partial J_{qn}\partial J_{np}}\mathcal{Z}[J]\\
			=&\frac{1}{V}\frac{\partial^2}
			{\partial J_{qn}\partial J_{np}}
			\left((E_{p}+E_n)-(E_n+E_{q})\right)\mathcal{Z}[J]\\
			=&K\Bigg\{\frac{\partial}{\partial J_{qn}}
			\exp\left(-S_{int}\left[\frac{1}{V}\frac{\partial}{\partial J}\right]\right)J_{pn}-
			\frac{\partial}{\partial J_{np}}\exp\left(-S_{int}
			\left[\frac{1}{V}\frac{\partial}{\partial J}\right]\right)
			J_{nq}\Bigg\}\mathcal{Z}_{free}[J]\\
			=&\left(J_{pn}\frac{\partial}
			{\partial J_{qn}}-J_{nq}\frac{\partial}
			{\partial J_{np}}\mathcal{Z}[J]
			-\left(\frac{\partial S_{int}}{\partial \Phi_{pn}}\right)
			\left[\frac{1}{V}\frac{\partial}{\partial J}\right]
			\frac{\partial}{\partial J_{qn}}+
			\left(\frac{\partial S_{int}}{\partial \Phi_{nq}}\right)
			\left[\frac{1}{V}\frac{\partial}{\partial J}\right]
			\frac{\partial}{\partial J_{np}}\right)\Z[J],
		\end{align*}
where we used the partition function of the form of \eqref{eq:Part2}. The third line is achieved by acting with 
$\frac{\partial}{\partial J_{np}}$ on $\Z_{free}$ for the first term and with
$\frac{\partial}{\partial J_{qn}}$ for the second term. For the last line, \sref{Lemma}{Lemma:Trick} was applied.
 \end{proof}
\end{prps}\noindent
The Ward-Takahashi identity of \sref{Proposition}{Prop:WardId} is recovered 
by taking the sum of the generalised identity of \sref{Proposition}{Prop:GenWardId} over $n$.

Both propositions have a problem, the prefactor $(E_{p}-E_{q})$ cannot be divided out without any more ado.
The decomposition of $\Z$ into correlation functions is degenerate for coinciding indices. 

Let us assume that two eigenvalues
are distinct $E_p\neq E_q$ for $p\neq q$, which seems to be a strong assumption. However, all correlation
functions depend on the $E_p$'s continuously according to the perturbative expansion. 
A sum over all eigenvalues 
can therefore be partitioned into 
a sum over distinct eigenvalues associated with a discrete measure governed their multiplicities. Including 
an appropriate measure therefore covers the assumption of distinct eigenvalues (see \sref{Sec.}{Sec:LargeLimit}). 

The distinct eigenvalues $E_p$ with $E_p<E_{p+1}$ can be understood as a discretisation of a continuously differentiable 
function. Furthermore, correlation functions are at least perturbatively rational functions of $E_p$ such that they become 
differentiable with respect to the index through the continuous extension of the eigenvalues. Under this consideration,
the following
was proved
\begin{thrm}(\cite[Theorem 2.3]{Grosse:2012uv})\label{Thm:Raimar}
 For injective $n\mapsto E_n$, the partition function $\Z[J]$ with action $S[\Phi]=V \Tr(E\Phi^2+ \mathcal{V}(\Phi))$, where
 $E_{nm}=E_n\delta_{nm}$ is diagonal and $\mathcal{V}(\Phi)$ a polynomial, satisfies the Ward-Takahashi identity
 \begin{align*}
  \sum_{n=0}^\mN \frac{\partial^2\Z[J]}{\partial J_{qn}\partial J_{np}}=&
  \delta_{p,q}( W^1_p[J]+W^2_p[J])\Z[J]+\frac{V}{E_p-E_q}\sum_{n=0}^\mN
  \left(J_{pn}\frac{\partial}{\partial J_{qn}}-J_{nq}\frac{\partial}{\partial J_{np}}\right)\Z[J],\\
  W^2_p[J]:=&\sum_{b=1}^\infty \sum_{N_1,..,N_b=1}^\infty \sum_{p_1^1,..,p_{N_b}^b=0}^\mN
 \sum_{g=0}^\infty \frac{V^{2-b-2g}}{b!}\prod_{\beta=1}^b\frac{\J_{p_1^\beta..p_{N_\beta}^\beta}}{N_\beta}\times\\
 &\bigg(\frac{1}{V^2} G^{(g)}_{|p|p|p_1^1..p_{N_1}^1|..|p^b_{1}..p^b_{N_b}|}+\frac{1}{V}\sum_{n=0}^\mN G^{(g)}_{|pn|
 p_1^1..p_{N_1}^1|..|p^b_{1}..p^b_{N_b}|}\\
 &+\sum_{M=3}^\infty\sum_{n,q_1,..,q_{M-3}=0}^\mN G^{(g)}_{|npnq_1..q_{M-3}|
 p_1^1..p_{N_1}^1|..|p^b_{1}..p^b_{N_b}|}J_{nq_1}J_{q_1q_2}..J_{q_{M-3}n}\bigg),\\
 W^1_p[J]:=&\sum_{b,b'=1}^\infty \sum_{N_1,..,N_b,M_1,..,M_{b'}=1}^\infty \sum_{p_1^1,..,p_{N_b}^b,q_1^1,..
 q_{M_{b'}}^{b'}=0}^\mN
 \sum_{g,g'=0}^\infty \frac{V^{4-b-b'-2(g+g')}}{b!b'!}\times\\
 &\prod_{\beta=1}^b\frac{\J_{p_1^\beta..p_{N_\beta}^\beta}}{N_\beta}
 \prod_{\beta'=1}^{b'}\frac{\J_{q_1^{\beta'}..q_{M_{\beta'}}^{\beta'}}}{M_\beta'}
 \frac{1}{V}G^{(g)}_{|p|p_1^1..p_{N_1}^1|..|p^b_{1}..p^b_{N_b}|}
 \frac{1}{V}G^{(g')}_{|p|q_1^1..q_{M_1}^1|..|q^{b'}_{1}..q^{b'}_{M_{b'}}|}.
 \end{align*}
\end{thrm}\noindent
The idea of the proof consists of writing
\begin{align*}
 \frac{\partial^2\Z[J]}{\partial J_{qn}\partial J_{np}}=\Z[J]
 \left(\frac{\partial^2\log \Z[J]}{\partial J_{qn}\partial J_{np}}\right)+\Z[J]
 \left(\frac{\partial\log \Z[J]}{\partial J_{qn}}\right)\left(\frac{\partial\log\Z[J]}{\partial J_{np}}\right).
\end{align*}
Then each part is separated in 
\begin{align*}
 \frac{\partial^2\log \Z[J]}{\partial J_{qn}\partial J_{np}}=&\delta_{pq} W^2_p[J]+W^{2,reg}_{pq}[J],\\
  \left(\frac{\partial\log \Z[J]}{\partial J_{qn}}\right)\left(\frac{\partial\log\Z[J]}{\partial J_{np}}\right)
 =&\delta_{pq} W^1_p[J]+W^{1,reg}_{pq}[J],
\end{align*}
where $\delta_{p,q} W^{i}_p$ contains the degenerate terms coming from $q=p$. The $W^{i,reg}$ is the remaining 
part regular
in the limit $q\to p$. The function $\log\Z[J]$ generates the correlation functions 
through \eqref{eq:Part2} such that the first term in $W^2_p$ is generated by $n=p$. The second
 in $W^2_p$ is the only 2-cycle 
and the third consists
of all possible 
cycles of higher length. The quadratic term $W^1_p$ is generated if $q=n=p$. \\
The regular terms obey by definition for $q\neq p$
\begin{align*}
 \sum_{n=0}^\mN (W^{2,reg}_{pq}[J]+W^{1,reg}_{pq}[J])\Z[J]=
 \frac{V}{E_p-E_q}\sum_{n=0}^\mN
  \left(J_{pn}\frac{\partial}{\partial J_{qn}}-J_{nq}\frac{\partial}{\partial J_{np}}\right)\Z[J]
\end{align*}
with a unique limit $q\to p$. This limit is valid by extending $E_n$ to a differentiable function 
such that all correlation functions are differentiable in its indices and L'H\^{o}pital's rule is applicable.

The SDE \eqref{eq:SDeqGeneral} consists of $k-1$ sums with $k\in\{3,..,d\}$. Applying \sref{Theorem}{Thm:Raimar}
$\lfloor\frac{k-1}{2}\rfloor$ times 
for each term in $\left(\frac{\partial S_{int}}{\partial \Phi_{p_1^1p_2^1}}\right)
        \left[\frac{1}{V}\frac{\partial}{\partial J}\right]$ leads from a 
        $((N_1+k-2)+N_2+..+N_b)$-point function to a $(N_1+N_2+..+N_b)$-point function plus 
        additional terms of different topologies.
        In the limit $\mN,V \to \infty$ and an additional topological expansion $G=\sum_gV^{-2g}G^{(g)}$, all 
        equations decouple in a way that a correlation function of Euler characteristic $\chi$
        obeys a linear equation, where the inhomogeneity
        does only depend on
        correlation functions of Euler characteristic $\chi'> \chi$.

\section{Large $\mN,V$-Limit}\label{Sec:LargeLimit}
As mentioned before, we assume that the sum and the parameter $V$ are of the same order
\begin{align*}
 \sum^{ \mN}\sim V.
\end{align*}
To be precise, let $(e_k)_{k=0}^{\mN'}$, $\mN'\leq\mN$ 
be the ordered distinct eigenvalues
of the projection of $E$. Let $r_k$ be the multiplicity of $e_k$ of $E$. This means 
\begin{align*}
 E_0=E_1=..=E_{r_0-1}=& \,e_0\\
 E_{r_0}=..=E_{r_0+r_1-1}=&\, e_1\\
  \vdots&
\end{align*}
Let $f$ be a function depending on the eigenvalues $E_n$ and not on the multiplicities. We obtain for the sum
\begin{align}\label{eq:SumFinite}
 \frac{1}{V}\sum_{n=0}^\mN f(E_n)=\frac{1}{V}\sum_{k=0}^{\mN'}r_k f(e_k)=&
 \int_0^\infty dt \,\varrho(t)\, f(t)\\
 \text{where}\qquad \varrho(t)=&\frac{1}{V}\sum_{k=0}^{\mN'}r_k\,\delta(t-e_k).
\end{align}
The distinct eigenvalues $e_k$ are extended to a continuous differentiable function, where $e_k$ are discrete
points. It is more convenient to perform the limit of $\mN'$ which indicates the limit of $\mN$.
In the limit $\mN'\to\infty$ we assume the asymptotic behaviour
\begin{align*}
 \lim_{k\to\infty}r_k\sim& k^{\delta-1}\\
 \lim_{k\to\infty}e_k\sim &k.
\end{align*}
The limit $V\to\infty$ is adjusted such that the summation index $k$ converges to a continuous variable $x$
depending on the asymptotic behavior of $r_k$, namely
 \begin{align}\label{eq:Limit}
  \lim_{\mN',V\to\infty}\frac{k}{V^{\frac{1}{\delta}}}\to x. 
 \end{align}
 From physical motivation, we want to 
denote the smallest eigenvalue by $e_0=\frac{\mu^2}{2}$, where $\mu$ is called the \textit{mass}.
Take the monotonic, continuously differentiable functions $e(x)$ with $e(0)=0$ and $r(x)$ from 
the multiplicities by
\begin{align}\label{eq:econt}
 e\left(\frac{k}{V^{\frac{1}{\delta}}}\right):=&e_k-\frac{\mu^2}{2},
 &&\lim_{\mN',V\to\infty}e\left(\frac{k}{V^{\frac{1}{\delta}}}\right)=e(x)\stackrel{x\to\infty}{\sim} x\\\label{eq:rcont}
 r\left(\frac{k}{V^{\frac{1}{\delta}}}\right):=&\frac{r_k}{V^{1-\frac{1}{\delta}}},&& \lim_{\mN',V\to\infty}
 r\left(\frac{k}{V^{\frac{1}{\delta}}}\right)=r(x)\stackrel{x\to\infty}{\sim} x^{\delta-1},
\end{align}
which are unique in the upper limit.
The measure $\varrho(t)$ converges to
\begin{align*}
 \varrho(t)&=\lim_{\mN',V\to\infty}
 \frac{1}{V}\sum_{k=0}^{\mN'}r_k\,\delta(t-e_k)=
 \lim_{\mN',V\to\infty}
 \frac{1}{V^{\frac{1}{\delta}}}\sum_{k=0}^{\mN'}\frac{k^{\delta-1}}{V^{\frac{\delta-1}{\delta }}}
 \,\delta\left(t-\frac{\mu^2}{2}-e\left(\frac{k}{V^{\frac{1}{\delta}}}\right)\right)\\
 &=\int_{-\epsilon}^\infty dx\,r(x)\delta\left(t-\frac{\mu^2}{2}-e\left(x\right)\right)=\frac{r(e^{-1}(t-\frac{\mu^2}{2}))}
 {e'(e^{-1}(t-\frac{\mu^2}{2}))},
\end{align*}
which is equivalently to 
\begin{align*}
 \varrho\left(\frac{\mu^2}{2}+e(x)\right)dx=r(x) dx.
\end{align*}
Finally, we find that the sum \eqref{eq:SumFinite} converges in the limit \eqref{eq:Limit} to
\begin{align*}
 \lim_{\mN,V}\frac{1}{V}\sum_{n=0}^\mN f(E_n)=\int_0^\infty dx \,r(x)\, f\left(\frac{\mu^2}{2}+e(x)\right),
\end{align*}
where the upper limit has its origin in the unbounded property of $E$. The last step is now to determine the 
spectral dimension $\D$ of \sref{Definition}{Def:Spec}. Inserting the spectral measure $\varrho(t)$ into \eqref{eq:SpectralDim} gives
\begin{align*}
 \int_0^\infty dt \,\frac{\varrho(t)}{(1+t)^{p/2}}=
 \int_0^\infty dx\, \frac{r(x)}{(1+\frac{\mu^2}{2}+e(x))^{p/2}}.
\end{align*}
The asymptotic $r(x)\sim x^{\delta-1}$ and $e(x)\sim x$ provides the spectral dimension 
(\sref{Definition}{Def:Spec}) if the 
integrand behaves asymptotically with
$\frac{1}{x}$ such that 
\begin{align*}
 -1=\delta-1-\frac{\D}{2}\quad \Rightarrow \quad \D=2\delta.
\end{align*}
\begin{rmk}\label{Rmk:LargeN}
 The assumption that $e(x)$ behaves asymptotically linear
 is easily substantiated. 
 Assume first the asymptotics $e(x)\sim x^k$ and 
 $r(x)\sim x^{\delta-1}$ for large $x$ and transform the variable $x$ for the integral by $y=x^k$ and  
 $k y^{\frac{k-1}{k}}\,dx =dy$. The asymptotic in $y$ is $e(x(y))\sim y$ and $r(x(y))\sim y^{\frac{\delta-1}{k}}$. Defining then
  $\tilde{ e}(y):=e(x(y))$ and $ \tilde{r}(y):=\frac{ r(x(y))}{k y^{\frac{k-1}{k}}}$ with $\tilde{r}(y)\sim 
 y^{\frac{\delta}{k}-1}$ satisfies the initial assumption by an adjusted measure $\tilde{r}$ with
 $\tilde{\delta}=\frac{\delta}{k}$. Note the modification of the spectral dimension to
 $\D=2\frac{\delta}{k}$.
\end{rmk}
\noindent
The next step is to understand the correlation function of genus $g$ \eqref{eq:DefCorrg} as a 
discretisation of a continuously differentiable function by 
\begin{align}\label{eq:Gcont}
 G^g(x_1^1,x_2^1,..,x_{N_1}^1|x_1^2,..,x_{N_2}^2|..|x_1^b,..,x_{N_b}^b)=
 G^{(g)}_{|p_1^1..p_{N_1}^1|..|p^b_{1}..p^b_{N_b}|}\vert_{p_i^j=x_i^jV^{2/\D}}.
\end{align}
We construct, in the same manner as the functions $e(x)$ and $r(x)$, the continuous function $G^g$ which
is unique in the upper limit. 

For later purpose, we stronger assume that $G^g$ is H\"older-continuous, i.e. $\exists C>0$ such that 
\begin{align*}
 |G^g(..,x,..)-G^g(..,y,..)|<C |x-y|^\alpha
\end{align*}
for all $ x,y\in U\subset\R_+$ in every component of $G^g$ with $0<\alpha\leq 1$.

\section{Renormalisation}\label{Sec:Renorm}
A sum converges in perturbative expansion
in the large $\mN,V$-limit to an integral over all closed face variables.
The measure $\varrho(t)$ depends asymptotically on the spectral dimension
which naturally leads in dimensions higher than zero to infinities related to the 
upper bound of the integral. For this purpose, let us modify the large $\mN,V$-limit with an additional 
definition of the \textit{cut-off} $\Lambda^2$ as the ratio in the limit
\begin{align}\label{eq:limitLambda}
 \lim_{\mN',V\to\infty}\frac{\mN'}{V^{2/\D}}=\Lambda^2.
\end{align}
Notice that $\mN'+1$ was the number of distinct eigenvalues $e_k$.
Sending $\Lambda^2$ to infinity (also called UV-limit) produces divergences where the degree of the 
divergence depends on the spectral dimension.

\begin{exm}
 Take the third contraction of \sref{Example}{Ex1} appearing in $G^{(0)}_{|pq|}=V\langle \Phi_{pq}\Phi_{qp}\rangle$
 at order $\lambda^1$. 
 Let $e_k$ be the distinct eigenvalues of multiplicity $r_k=k$ which indicates $\D=D=4$. 
 Then, the sum of the example converges in the limit discussed 
 in \sref{Sec.}{Sec:LargeLimit} with the ratio 
 \eqref{eq:limitLambda} to
 \begin{align*}
  \lim_{\mN',V\to\infty} \frac{1}{(e_q+e_p)^2}\frac{1}{V}\sum_{k=0}^{\mN'}\frac{k}{e_k+e_q}=\frac{1}{(\mu^2+e(x)+e(y))^2}
  \int_0^{\Lambda^2}dt\frac{t}{\mu^2+e(t)+e(x)},
 \end{align*}
where $\frac{p}{\sqrt{V}}=x$, $\frac{q}{\sqrt{V}}=y$,
$e(\frac{k}{\sqrt{V}})=e_k-\frac{\mu^2}{2}$, $e_0=\frac{\mu^2}{2}$ and $e(t)\stackrel{t\to\infty}{\sim} t$.
The integral has therefore a linear divergence in $\Lambda^2$ for $\Lambda^2\to \infty$.
\end{exm}\noindent

A systematic procedure, how the UV divergences can be compensated for QFTs, is called \textit{renormalisation} and is
described in its full generality by the BPHZ theorem \cite{bogoliubow1957,
Hepp:1966eg,Zimmermann:1969jj}. Counterterms 
cancel all divergences and subdivergences by the  Bogoliubov-Parasiuk R-operation 
such that the results are
uniquely determined at any order through additional boundary conditions. 
A compact description of this procedure was given by Zimmermann and is 
also known as \textit{Zimmermann's forest formula}. 

Adjusting the action 
with additional $\Lambda^2$-depending constants, called  \textit{renormalisation constants},
produces exactly these counterterms coming from the R-operation by choosing the same boundary conditions.
The renormalisation constants are usually divergent in the $\Lambda^2\to \infty$ limit, typical examples are
\begin{align*}
 &\text{mass renormalisation:}\quad &&\mu^2\to\mu^2_{bare}(\Lambda^2)\\
 &\text{field renormalisation:}\quad &&\Phi \to Z^{1/2}(\Lambda^2)\Phi\\
 &\text{coupling constant renormalisation:}\quad &&\lambda_i\to\lambda_{i,bare}(\Lambda^2).
\end{align*}
However, any finite number of further renormalisation constants is permitted. If a model is UV-finite (renormalised) after
a finite number of renormalisation constants, it is called \textit{renormalisable}, otherwise 
\textit{nonrenormalisable}. How many renormalisation constants are necessary depends on the type of the model, the 
interactions and on the order 
of the divergences related to the dimension. 

Let us recall Zimmermann's forest formula for ribbon graphs. We admit that the complexity of the forest 
formula will get much easier for ribbon
graphs and additionally on a Euclidean space. Faces of ribbon graphs are globally labelled 
whereas loops in ordinary QFT need locally assigned momenta. 
We follow the description of the forest 
formula for ribbon graphs given in \cite{Grosse:2016qmk} with a slight generalisation 
for the dimension and interaction.

\subsection{Zimmermann's Forest Formula}\label{Sec.Zimmer}
Let $\Gamma$ be a connected ribbon graph. Let $\B_\Gamma$ be the set of open faces 
and $\F_\Gamma$ the set of closed faces of $\Gamma$. A \textit{ribbon subgraph} $\gamma\subset\Gamma$ 
consists of a subset of closed faces $\F_\gamma\subset \F_\Gamma$ together with all 
bordering edges of $\F_\gamma$ and all vertices at the end of the edges such that $\gamma$ is connected and after removing 
any vertex the subgraph is still connected.
Each ribbon subgraph uniquely defines the set of adjacent faces of $\gamma$ by $\E_\gamma\subset\B_\Gamma\cup
\F_\Gamma\setminus\F_\gamma$, i.e. any element of $\E_\gamma$ is bordering an edge of $\gamma$.
Let the \textit{extended subgraph} $\bar{\gamma}$ of $\gamma$ be $\gamma$ together with all half-edges
within $\Gamma$ which are connected to the vertices of $\gamma$, and its attached faces (which are now understood as 
open faces) of 
the half-edges which can possibly be outside $\E_\gamma$.

Let $x_i$ with $i\in\{1,..,k\}$ be the labellings of the faces in $\E_\gamma$ and 
$f(\gamma)=\{x_1,..,x_{k}\}$ the set of the face variables $x_i$. Further, let $r_\gamma(f(y),y_1,..,y_m)$ be 
a rational function, where $y_i$ with $i\in \{1,..,m\}$ are the labellings of the faces of ${\gamma}$. 
Let $(T^\omega_{f(\gamma)}r_\gamma)(f(\gamma),y_1,..,y_m)$ be the 
$\omega^{\text{th}}$ order multivariate Taylor polynomial of $r_{\gamma}$ with respect 
to the variables $f(\gamma)$. The Taylor polynomial
is $T^\omega_{f(\gamma)}r_\gamma\equiv0$ for $\omega<0$.

A forest $\U_\Gamma$ in $\Gamma$ is a set of ribbon subgraphs $\{\gamma_1,..,\gamma_l\}$ such that any pair
of subgraphs $\gamma_i,\gamma_j$ obeys one of the three conditions
\begin{align*}
 \gamma_i\subset\gamma_j,\qquad 
 \gamma_j\subset\gamma_i,\qquad
 \gamma_i\cap\gamma_j=\emptyset,
\end{align*}
where the empty set means that $\gamma_i,\gamma_j$ are disjoint also for edges, but \textit{not} necessarily
for vertices. The empty 
forest $\U_\Gamma=\emptyset$ is included in the definition, whereas $\Gamma$ is not possible as forest since $\Gamma$
has open faces and any ribbon subgraph does not.

The three conditions for the elements of the forest equip the forest with a partial ordering 
by the following construction. 
Let a descendant $ \gamma_{i_j}\in U_{\Gamma}$ be the subset of $\gamma_i$, more precisely $\gamma_{i_j}\subset\gamma_i
\in U_{\Gamma}$, together with the 
conditions that any two different descendants
are disjoint $\gamma_{i_j}\cap \gamma_{i_l}=\emptyset$ and 
there exists no $\gamma'\in\U_\Gamma$ such that 
$\gamma_{i_j}\subset\gamma'\subset\gamma_i$.
For any $\gamma_i\in \U_\Gamma$ there exists a unique set $o(\gamma_i)=\{\gamma_{i_1},..,\gamma_{i_k}\}$
, $\gamma_{i_j}\in \U_\Gamma$ of descendants
$\gamma_{i_j}\subset\gamma_i$ such that for any $\gamma_j\in\U_\Gamma\setminus \{\gamma_i,o(\gamma_i)\}$ one of the three 
conditions hold
\begin{align*}
 \gamma_j\subset o(\gamma_i),\qquad \gamma_i\subset\gamma_j,\qquad \gamma_i\cap\gamma_j=\emptyset.
\end{align*}
The first condition means that $\gamma_j$ is a subset of one descendant of $\gamma_i$.

Now let $I_\Gamma$ be the integrand of $\mathfrak{h}(\Gamma) V^{2-2g-b-N}$ (see \sref{Sec.}{Sec.Feynman}) of the ribbon graph $\Gamma$ with genus 
$g$, $b$ boundary components and $N$ open faces in the large $\mN,V$-limit, discussed in \sref{Sec.}{Sec:LargeLimit}. 
Any forest $\U_\Gamma$ defines a unique partition due to the partial ordering of an integrand by
\begin{align*}
 I_\Gamma=I_{\Gamma\setminus\U_\Gamma}\prod_{\gamma\in\U_\Gamma} I_{\gamma\setminus o(\gamma)},\qquad 
 \prod_{\gamma\in \emptyset}I_\emptyset=1.
\end{align*}

We call $\omega(\gamma)$ the superficial degree of divergence of $\gamma$
which is defined
as the degree of the numerator subtracted by the degree of the denominator of $I_{{\gamma}}$. 
Let $g_\gamma$ be the genus, $b_\gamma$
the number of boundary components, $k_i$ the number
of $i$-valent vertices, $v_\gamma=k_3+..+k_d$ the number of all vertices,
$N_{\gamma}$ the number of 
open faces and $e_{\gamma}$ the number 
of edges of the extended subgraph $\bar{\gamma}$ of the subgraph $\gamma$. 
Let $\Sigma_\gamma$ be the number of (closed) faces and 
$e_c$ the number of edges of $\gamma$ ($e_c$ can also be understood as
number of edges attached to a closed face of $\bar{\gamma}$). 
The superficial degree of divergence is then
\begin{align*}
 \omega(\gamma)=\frac{\D}{2}\Sigma_\gamma-e_c.
\end{align*}
Euler's formula gives the number of closed faces $\Sigma_\gamma=2-2g_\gamma-b_\gamma+e_\gamma-N_{\gamma}-v_\gamma$. 
The 1PI construction of $\gamma$ as a subgraph leads to two further relations
\begin{align*}
 e_{\gamma}=&\frac{N_{\gamma}+3k_3+4k_4+..+dk_d}{2}\\
 e_{\gamma}=&3k_3+4k_4+..+dk_d-e_c,
\end{align*}
so we can conclude:
\begin{align}\label{eq:superficial}
 \omega(\gamma)=\frac{\D}{2}\left(2-2g_\gamma-b_\gamma-v_\gamma\right)+
 \left(\frac{\D}{2}-1\right)\frac{3k_3+4k_4+..+dk_d-N_{\gamma}}{2}.
\end{align}
Now we are ready to formulate the adapted BPHZ theorem for matrix field theory models in the large $\mN,V$-limit 
\begin{thrm}(\cite{Zimmermann:1969jj})\label{Thm:Zimmermann}
 A formal power series in the coupling constants $\lambda_i$ of appropriated chosen renormalisation 
 constants $(\mu^2_{bare},Z,..)$ with the same number of normalisation conditions 
  (with the reference point at zero momentum)
 results in the replacement of the integrand $I_\Gamma$ of any ribbon graph in the perturbation theory by 
  \begin{align*}
   I_\Gamma\mapsto \mathcal{R}(I_\Gamma):=\sum_{\U_\Gamma} I_{\Gamma\setminus \U_\Gamma}\prod_{\gamma\in\U_\Gamma}
   (-T^{\lfloor \omega(\gamma)\rfloor}_{f(\gamma)}I_{\gamma\setminus o(\gamma)}),
  \end{align*}
where the sum over all forests includes the empty forest $\emptyset$. The product $\prod_{\gamma\in\U_\Gamma}$ 
over all elements
of a forest takes the partial ordering into consideration 
\begin{align*}
 \prod_{\gamma\in\U_\Gamma}
   (-T^{\lfloor \omega(\gamma)\rfloor}_{f(\gamma)}I_{\gamma\setminus o(\gamma)})=...(-T^{\lfloor \omega(\gamma)\rfloor}_{f(\gamma)}I_{\gamma\setminus o(\gamma)})\prod_{\gamma'\in o(\gamma)}
(-T^{\lfloor \omega(\gamma')\rfloor}_{f(\gamma')}I_{\gamma'\setminus o(\gamma')})...
\end{align*}
\end{thrm}\noindent
The forest formula subtracts divergences in a minimal way with the reference point at zero momentum.
This means that only divergences and subdivergences are removed.
Finite graphs receive no subtraction. It is not clear from the 
beginning how the boundary conditions are fixed by the forest
formula, e.g. $G(0,0)=1$ or $G(0,0)=\frac{1}{\mu^2}$.

In the forthcoming work, we will discuss for cubic interaction the dimensions $D=2,4$ which are \textit{super-renormalisable}.
Because of this and the later chosen normalisation conditions by $G^{(0)}(0)=0$ in $D=2$ and
additionally $\partial_x G^{(0)}(x)\vert_{x=0}=
G^{(0)}(0,0)-1=0$ in $D=4$, the forest formula has to be adapted 
since these conditions additionally renormalise finite graphs. 
For $D=6$ is the model 
\textit{just-renormalisable} and we will choose the conditions in the same way as
the forest formula does.

Looking at the quartic interaction, the $D=2$ case is 
super-renormalisable and the $D=4$ case just-renormalisable.
We are in particular interested in
$D=4$ and will later chose the renormalisation
differently since the exact solution provides another natural choice for the boundary conditions.

All perturbative computations described by Zimmermann's forest formula are collected in \sref{App.}{App:Pert}.
 \begin{rmk}
  To avoid redundant factors of the renormalised mass $\mu$ (after renormalisation),
 we passed in the cubic model due to 
 $G^{(0)}(0,0)= \frac{1}{\mu^2}=1$ to mass-dimensionless quantities without mentioning.
 More concretely, this can either be achieved by multiplying each quantity $Q_i$ of mass-dimension $k_i$
 with $\mu^{-k_i}$ or, equivalently, to
 choose $\mu^2=1$ from the beginning. A detailed treatment of appropriate factors in 
 $\mu$ can be found in \cite{Grosse:2016qmk} for the cubic model.
 \end{rmk}

\section{QFT on the Moyal Space and Matrix Field Theory}\label{Sec.Moyal}
The underlying $D$-dimensional space of a QFT is the $\R^D$, we will 
restrict in this subsection to even $D\in\{2,4,6,..\}$ dimensions. 
In a rigorous formulation, a field operator of a QFT has to be smeared out by a test function $g\in\mathcal{S}(\R^D)$.
The space of complex-valued Schwartz functions consists of rapidly decaying functions
\begin{align*}
 \mathcal{S}(\R^D)=\{g\in C^\infty(\R^D)\,:\, \sup_{x\in \R^D}|x^\alpha D^\beta g(x)|<\infty 
 \quad \forall \alpha,\beta\in \N^D\},
\end{align*}
where $\alpha,\beta$ are multi-indices, i.e. $\alpha=(\alpha_1,..,\alpha_D)$ and $\beta=(\beta_1,..,\beta_D)$. 
The expectation value of these operator-valued distributions defines the correlation functions of a QFT model.
The underlying space $\R^D$ of an ordinary QFT is equipped with 
local multiplications which are commutative $f(x)g(x)=g(x)f(x)$ for
 $f,g\in \mathcal{S}(\R^D)$.

As an alternative, we deform the space \cite{rieffel1989}
by the algebra $\mathcal{A}_\star=(\mathcal{S}(\R^D),\star)$ to achieve the \textit{Moyal space}, 
where
the noncommutative Moyal $\star$-product is defined by \cite{GraciaBondia:1987kw}
\begin{align}\label{eq:StarProd}
 &(g\star h)(x)=\int \frac{d^D k}{(2\pi)^D}\int d^Dy\, g(x+\tfrac{1}{2}\Theta k)\,h(x+y)e^{\mathrm{i} k\cdot y},\\
 &\Theta=\mathds{1}_{D/2}\otimes \left( \begin{array}{rr}
0 & \theta   \\
 -\theta & 0 
\end{array}\right),\quad  \theta\in \R,\quad x\in\R^D ,\quad g,h\in\mathcal{S}(\R^D).\nonumber
\end{align}
The originally investigated Moyal space is 2-dimensional, whereas we take $\frac{D}{2}$ copies of the 2-dimensional 
version multiplied by a Cartesian product. The structure of the $D$-dimensional Moyal space is covered by
tensorial structure of $\Theta$.

The $\star$-product generalises the ordinary Euclidean QFT construction since for $\theta=0$ the $k$-integral
generates the delta distribution $\delta(y)$ which gives after $y$-integration a pointwise multiplication
$(g\star h)(x)\vert_{\theta=0}=g(x)h(x)$.

Due to the noncommutative nature and the $\frac{D}{2}$ copies, the Moyal product possesses a matrix base
\begin{align}\label{eq:MatrixBase}
 &b_{\narrowunderline{n}\narrowunderline{m}}(x)=f_{n_1m_2}
 (x_1,x_2)f_{n_2m_2}(x_3,x_4)..f_{n_{D/2}m_{D/2}}(x_{D-1}x_D),\\
 &\narrowunderline{n}=(n_1,..,n_{D/2})\in\N^{D/2},\quad 
 \narrowunderline{m}=(m_1,..,m_{D/2})\in\N^{D/2},\quad x=(x_1,..,x_D)\in\R^D\nonumber
\end{align}
which is separated into $\frac{D}{2}$ bases $f_{ij}(x,y)$ of the 2-dimensional Moyal spaces 
through the tensorial structure of $\Theta$.
We refer to \sref{App.}{App:Moyal} for more details. The base has the matrix multiplication 
property and a trace by
\begin{align}\label{eq:Trace}
 \int d^Dx \,b_{\narrowunderline{n}\narrowunderline{m}}(x)=&(2\pi\theta)^{D/2}
 \delta_{\narrowunderline{n},\narrowunderline{m}}\\\label{eq:MatrixProd}
 (b_{\narrowunderline{n}\narrowunderline{m}}\star 
 b_{\narrowunderline{k}\narrowunderline{l}})(x)=&\delta_{\narrowunderline{m},\narrowunderline{k}}
 b_{\narrowunderline{n}\narrowunderline{l}}(x),
\end{align}
where $\delta_{\narrowunderline{n},\narrowunderline{m}}=\delta_{n_1,m_1}..\delta_{n_{D/2},m_{D/2}}$.

The algebra $\mathcal{A}_\star$ is a pre-$C^\star$-algebra from which a Hilbert space $H=H_1 \otimes H_2$ 
can be constructed via GNS-construction. Restricting to $H_1$, any element of $\mathcal{A}_\star$ 
is compact on $H_1$ and even further a trace-class operator.

Under these considerations, an action of the form
\begin{align}\label{eq:BadAction}
 \int \frac{d^Dx}{(8\pi)^{D/2}} \left(\frac{1}{2}\Phi\star 
 (-\Delta +\mu^2)\Phi+\mathcal{V}^\star (\Phi)\right)(x)
\end{align}
has got the right properties of a field theory, where $\Phi\in \mathcal{S}(\R^D)$, 
$\Delta=\sum_{i=1}^D\frac{\partial^2}{\partial x_i^2}$ denotes the 
Laplacian and $\mathcal{V}^\star (\Phi)=\sum_{k=3}^d\frac{\lambda_k}{k}\Phi^{\star k}$ as a potential 
with $\Phi^{\star k}=
\Phi\star..\star\Phi$ exactly $k$ times.
Using \eqref{eq:BadAction} as the action and writing out the first orders of the perturbative 
expansion was shown to have mixing of UV and IR divergence \cite{Minwalla:1999px}. 
The origin of the mixing is the non-local definition of the product. 

The works \cite{Grosse:2003nw,Grosse:2004yu} give a way to handle the UV/IR-mixing problem 
by adding a harmonic oscillator
term to the action. The mixing problem is therefore a property of the Laplacian and not of chosen interaction. The idea was 
to add a term to the action such that Laplacian breaks in the matrix base \eqref{eq:MatrixBase} down to a 
matrix (only for $\Omega=1$) instead of a tensor of rank 4. Therefore, let the action be 
\begin{align}\label{eq:NCAction}
 S[\Phi]=\int \frac{d^Dx}{(8\pi)^{D/2}} \left(\frac{1}{2}\Phi\star 
 (-\Delta +\Omega^2\|2\Theta^{-1}\cdot x\|^2+\mu_0^2)\Phi+ \mathcal{V}^\star (\Phi)\right)(x),
\end{align}
where $\Omega\in\R$ regulates the harmonic oscillator independently of $\theta$.
The harmonic oscillator term breaks translational invariance, however 
invariance can be recovered in the $\theta\to\infty$ limit or $\Omega=0$.

The expansion of $\Phi(x)=\sum_{\narrowunderline{n},\narrowunderline{m}\in\N^{D/2}}
\Phi_{\narrowunderline{n}\narrowunderline{m}}
b_{\narrowunderline{n}\narrowunderline{m}}(x)$ and further properties of the Moyal $\star$-product listed in 
\sref{App.}{App:Moyal} give the action \eqref{eq:NCAction} after integration
\begin{align}
 S[\Phi]=\sum_{\narrowunderline{n},\narrowunderline{m},\narrowunderline{k},\narrowunderline{l}\in\N^{D/2}}
 \Phi_{\narrowunderline{n}\narrowunderline{m}}
 G_{\narrowunderline{n}\narrowunderline{m}
  ;\narrowunderline{k}\narrowunderline{l}}\Phi_{\narrowunderline{k}\narrowunderline{l}}+\sum_{i=3}^d\frac{\lambda_i}{i}
  \sum_{\narrowunderline{n_1},..,\narrowunderline{n_i}\in\N^{D/2}}
  \Phi_{\narrowunderline{n_1}\narrowunderline{n_2}}
  \Phi_{\narrowunderline{n_2}\narrowunderline{n_3}}..
  \Phi_{\narrowunderline{n_i}\narrowunderline{n_1}}
\end{align}
with
\begin{align}\label{eq:NCProp}
 G_{\narrowunderline{n}\narrowunderline{m}
  ;\narrowunderline{k}\narrowunderline{l}}=&\left(\frac{\theta}{4}\right)^{D/2}\bigg(\frac{\mu_0^2}{2}+
  \frac{1+\Omega^2}{\theta}\bigg(\frac{D}{2}+|\narrowunderline{n}|+|\narrowunderline{m}|\bigg)\bigg)
 \delta_{\narrowunderline{m},\narrowunderline{k}}\delta_{\narrowunderline{n},\narrowunderline{l}}\\\nonumber
 &- \left(\frac{\theta}{4}\right)^{D/2}\frac{1-\Omega^2}{\theta}\sum_{i=1}^{D/2}(\sqrt{n_im_i} \delta_{n_i-1,l_i}
 \delta_{m_i-1,k_i}+\sqrt{k_il_i} \delta_{n_i+1,l_i}
 \delta_{m_i+1,k_i})\check{\delta}^i_{\narrowunderline{m},\narrowunderline{k}}\check{\delta}^i_{\narrowunderline{n},\narrowunderline{l}},
\end{align}
where $|\narrowunderline{n}|=n_1+n_2+..+n_{D/2}$ and 
$\check{\delta}^i_{\narrowunderline{n},\narrowunderline{m}}$ is $\delta_{\narrowunderline{n},\narrowunderline{m}}$
with omitted $\delta_{n_i,m_i}$.

The Moyal $\star$-product is known to have a duality between 
position space and momentum space $x\leftrightarrow p$ 
\cite{Langmann:2002cc}. The action transforms under this duality by
\begin{align}\label{eq:SelfDual}
 S[\Phi;\mu_0,\lambda_i,\Omega]\mapsto \Omega^2 \,S\left[\Phi;\frac{\mu_0}{\Omega},\frac{\lambda_i}{\Omega^2},\frac{1}{\Omega}\right].
\end{align}
The special case $\Omega=1$ is called \textit{self-dual} since it 
leaves the action invariant under the position-momentum duality and breaks 
$G_{\narrowunderline{n}\narrowunderline{m}
  ;\narrowunderline{k}\narrowunderline{l}}$ down to a matrix because all next-to diagonal terms  are 
  cancelled (second line of \eqref{eq:NCProp}). 

We consider from now on the self-dual noncommutative QFT model $(\Omega=1)$. 
The noncommutative QFT action \eqref{eq:NCAction} is then
exactly a matrix field action \eqref{eq:Action} by truncating of the sum to 
$\narrowunderline{n}\in \N^{D/2}_{\mN'}$ with 
$\N^{D/2}_{\mN'}:=\{ \narrowunderline{n}\in \N^{D/2}\,:\,|\narrowunderline{n}|\leq \mN'\}$ and identifying
\begin{align*}
V=&\left(\frac{\theta}{4}\right)^{D/2}\\
 \mu^2=&\mu^2_0+\frac{D}{4 V^{2/D}}\\
 e_{|\narrowunderline{n}|}=&\frac{\mu^2}{2}+\frac{|\narrowunderline{n}|}{V^{2/D}}\\
 H_{|\narrowunderline{n}|,|\narrowunderline{m}|}=&2 G_{\narrowunderline{n}\narrowunderline{m}
 ;\narrowunderline{k}\narrowunderline{l}}\,
 \delta_{\narrowunderline{m},\narrowunderline{k}}\delta_{\narrowunderline{n},\narrowunderline{l}}
  \\
 r_{|\narrowunderline{n}|+1}=&\binom{|\narrowunderline{n}|+\frac{D}{2}-1}{\frac{D}{2}-1}.
\end{align*}
The model depends only on the norm $|\narrowunderline{n}|$, i.e.
the explicit dependence on $\narrowunderline{n}$ drops out.  
The multiplicity of $e_{|\narrowunderline{n}|}$ is
$r_{|\narrowunderline{n}|}$, which is the number of all $\narrowunderline{n}\in \N^{D/2}$ with norm 
$|\narrowunderline{n}|$ .
The correlation functions are naturally labelled by the tuples $\narrowunderline{p}_i^j\in\N^{D/2}_{\mN'}$
\begin{align*}
 G_{|\narrowunderline{p}_1^1
 ..\narrowunderline{p}_{N_1}^1|
 ..|\narrowunderline{p}_1^b
 ..\narrowunderline{p}_{N_b}^b|}, 
\end{align*}
where two correlation functions are the same if the norms of their indices are the same, i.e.
\begin{align*}
 G_{|\narrowunderline{p}_1^1
 ..\narrowunderline{p}_{N_1}^1|
 ..|\narrowunderline{p}_1^b
 ..\narrowunderline{p}_{N_b}^b|}=G_{|\narrowunderline{q}_1^1
 ..\narrowunderline{q}_{N_1}^1|
 ..|\narrowunderline{q}_1^b
 ..\narrowunderline{q}_{N_b}^b|},\quad \text{if} \quad |\narrowunderline{p}_i^j|=
 |\narrowunderline{q}_i^j|\quad \forall i,j.
\end{align*}

The deformation parameter $\theta$ of the Moyal space is directly related to the parameter $V$. 
The discussed limit $\mN,V\to \infty$ (with $\mN=\sum_{i=1}^{\mN'}
r_{i}$)
guarantees translational invariance of the action \eqref{eq:NCAction} and recovers 
the infinitely large base of the 
Moyal space. The monotonic, continuously differentiable
functions $e(x)$ (see \eqref{eq:econt}) and $r(x)$ (see \ref{eq:rcont}) converge to
\begin{align*}
 e(x)=&\,x\\
 r(x)=&\,\frac{x^{D/2-1}}{(\frac{D}{2}-1)!}.
\end{align*}
The model can now be treated as matrix field theory model. All correlation functions defined in \sref{Sec.}{Sec:SDE} 
have to be determined in the limit of \sref{Sec.}{Sec:LargeLimit} with the renormalisation from
\sref{Sec.}{Sec:Renorm}. 

However, we have to determine the correlation functions for the matrix field theory and 
insert them into the expansion in the base
$b_{\narrowunderline{n}\narrowunderline{m}}(x)$ to give statements for the expectation values 
on the Moyal space. 
Hence, we \textit{define} the \textit{connected Schwinger function} by
\begin{align}\label{eq:Schwinger}
 S_c(\xi_1,..,\xi_N):=&\lim_{\Lambda^2\to\infty} \lim_{\stackrel{V,\mN'\to\infty}{\frac{\mN'}{V^{2/D}}=\Lambda^2}}
 \sum_{N_1+N_2+..+N_b=N}\sum_{\narrowunderline{p}_1^1
 ,..,\narrowunderline{p}_{N_b}^b\in\N^{D/2}_{\mN'}} \frac{G_{|\narrowunderline{p}_1^1
 ..\narrowunderline{p}_{N_1}^1|
 ..|\narrowunderline{p}_1^b
 ..\narrowunderline{p}_{N_b}^b|}}{(8 \pi)^{D/2}\, b!}\\
 &\times\sum_{\sigma\in S_N}\prod_{\beta=1}^b\frac{b_{\narrowunderline{p}_1^\beta
 \narrowunderline{p}_2^\beta}(\xi_{\sigma(s_\beta+1)})..b_{\narrowunderline{p}_{N_\beta}^\beta
 \narrowunderline{p}_1^\beta}(\xi_{\sigma(s_\beta+N_\beta)})}{V N_\beta},\nonumber
\end{align}
where $s_\beta=
N_1+..+N_{\beta-1}$ and $S_N$ is the symmetric group, consisting of all permutation of a set with $N$ elements.
The Schwinger function is, by definition, symmetric $S_c(\xi_1,..,\xi_i,..,\xi_j,..,\xi_N)=
S_c(\xi_1,..,\xi_j,..,\xi_i,..,\xi_N)$ for all $i,j$.

The definition \eqref{eq:Schwinger} is motivated by the
partition function of the Moyal space in position space before applying the expansion. 
This partition function should be treated with caution and is formally
defined as
the functional integral
\begin{align*}
 \Z[J]=\int D[\Phi] \exp\left(-S[\Phi]+V \int d^Dx\,(J\star \Phi)(x)\right),
\end{align*}
where $J\in\mathcal{S}(\R^D)$ is the source and the measure $D[\Phi]$ should be understood formally by the function 
$\Phi\in\mathcal{S}(\R^D)$. The measure as well as the entire expression have no well-defined limit such that 
the functional integral has to be understood more or less symbolically as the limit of \eqref{eq:Part1}. 
The source and the field are expanded by $J(x)=\frac{1}{V (8\pi)^{D/2}}\sum_{\narrowunderline{n},\narrowunderline{m}\in\N^{D/2}}
J_{\narrowunderline{n}\narrowunderline{m}}
b_{\narrowunderline{n}\narrowunderline{m}}(x)$ and $\Phi(x)=\sum_{\narrowunderline{n},\narrowunderline{m}\in\N^{D/2}}
\Phi_{\narrowunderline{n}\narrowunderline{m}}
b_{\narrowunderline{n}\narrowunderline{m}}(x)$.
Staying at the formal level, the Schwinger function has originally the formal definition
\begin{align*}
 S_c(\xi_1,..,\xi_N)=\frac{1}{(8\pi)^{D/2} }\frac{\delta^N\frac{1}{V^2} \log \frac{\Z[J]}{\Z[0]}}
 {\delta J(\xi_1)..\delta J(\xi_N)}\bigg\vert_{J=0}.
\end{align*}
The expansion of $J(x)$ leads with $J_{\narrowunderline{n}\narrowunderline{m}}
=\int d^Dx\,
b_{\narrowunderline{n}\narrowunderline{m}}(x) J(x)$ to  $\frac{\delta J_{\narrowunderline{n}\narrowunderline{m}}}{\delta J(x)}=
b_{\narrowunderline{n}\narrowunderline{m}}(x)$. This recovers the definition \eqref{eq:Schwinger} from the chain rule
together with \eqref{eq:DefCorrSour}.

The invariance of the Schwinger function \eqref{eq:Schwinger} under the full Euclidean group is shown 
by another representation
\begin{align}\label{eq:SchwingerRep}
 S_c(\xi_1,..,\xi_N)=&\sum_{\stackrel{N_1+..+N_b=N}{N_\beta\, \text{even}}}
 \sum_{\sigma\in S_N}\left(\prod_{\beta=1}^b \frac{2^{DN_\beta/2}}{N_\beta}\int\frac{d^D p^{\beta}}{(2\pi)^{D/2}}
 e^{\mathrm{i} p^\beta\cdot (\xi_{\sigma (s_\beta+1)}-
 \xi_{\sigma (s_\beta+2)}+..-\xi_{\sigma (s_\beta+N_\beta)})}\right)\\\nonumber
 &\times\frac{1}{(8\pi)^{D/2}b!}G^0\bigg(\underbrace{\frac{\|p^1\|^2}{2},..,\frac{\|p^1\|^2}{2}}_{N_1}\bigg\vert...\bigg\vert
 \underbrace{\frac{\|p^b\|^2}{2},..,\frac{\|p^b\|^2}{2}}_{N_b}\bigg),
 \end{align}
 \begin{align*}
 \text{where}\quad &
 G^g(x_1^1,..,x_{N_1}^1|..|x_1^b,..,x_{N_b}^b)=\lim_{V,\mN',\Lambda^2}
 G^{(g)}_{|\narrowunderline{p}_1^1..\narrowunderline{p}_{N_1}^1|..|\narrowunderline{p}^b_{1}
 ..\narrowunderline{p}^b_{N_b}|}\big\vert_{|\narrowunderline{p}_i^j|=x_i^jV^{2/D}}.
\end{align*}
This representation can be found for $D=4$ in \cite{Grosse:2013iva}. The derivation 
of the general
case is outsourced to \sref{App.}{App:Schwinger}.

We have completed the procedure to determine Schwinger functions of a self-dual real scalar QFT model 
on the Moyal space by 
applying first the matrix base expansion. To solve then the corresponding matrix field theory model 
in the combined limit 
$V,\mN\to\infty$, we have to take the planar sector such that the Schwinger function is given by \eqref{eq:SchwingerRep}.
This procedure obviously holds only in the discussed $V,\mN$-limit. For any finite $V$, translational invariance is broken 
and for any finite $\mN$, the Moyal space cannot be recovered. Going away from the self-dual point $\Omega=1$ the limit of 
$V$ could be relaxed. The off-diagonal terms of the covariance
$G_{\narrowunderline{n}\narrowunderline{m}
  ;\narrowunderline{k}\narrowunderline{l}}$ for $\Omega\neq 1$
  survive, however it is possible to diagonalise $G_{\narrowunderline{n}\narrowunderline{m}
  ;\narrowunderline{k}\narrowunderline{l}}$ in $D=2$ \cite{Grosse:2003nw} and in $D=4$ 
  by Meixner polynomials \cite{Grosse:2004yu}. Nevertheless, the next step of deriving
  correlation functions in this situation 
  is an outstanding challenge.

\begin{rmk}
 From the representation \eqref{eq:SchwingerRep} together with the correlation functions $G^0$, the reflection 
positivity property of Schwinger functions can be checked. If it holds, for instance in $D=4$, 
the first exactly solvable non-trivial 
QFT model in 4 dimensions can be derived. All further Osterwalder-Schrader axioms \cite{Osterwalder:1973dx,
Osterwalder:1974tc} (except clustering) are fulfilled by the representation \eqref{eq:SchwingerRep}. 
\end{rmk}

\chapter[Cubic Interaction]{Cubic Interaction:\\
The Renormalised Kontsevich Model}\label{chap:cubic}
In this chapter, we consider a matrix field theory with cubic interaction 
$V(\Phi)=\frac{\lambda}{3}\Phi^3$. Historically, 
this model has great importance. It was designed by Maxim Kontsevich \cite{Kontsevich:1992ti}
to prove Edward Witten's conjecture \cite{Witten:1990hr} about the equivalence between two different 
2-dimensional quantum gravity approaches. In particular, Witten conjectured that the 
generating function of intersection numbers on the moduli space $\overline{\mathcal{M}}_{g,b}$ 
of stable complex curves of genus $g$ and $b$ distinct marked points satisfies the string equation and an entire hierarchy 
of KdV equations. These equations are nonlinear partial differential equations which are 
recursively constructed by Gelfand-Dikii polynomials \cite{gelfand_dikii_1977}.

It was later discovered that the 0-dimensional Kontsevich model (finite matrices) has an underlying structure known as
\textit{topological recursion}, which was
developed by Bertrand Eynard and Nicolas Orantin \cite{Eynard:2007kz}. 

In \sref{Sec.}{Sec:CubicSchwinger} we recall the derivation of the SDEs. 
\sref{Sec.}{Sec:CubicSolution} is split into several subsections which provide the 
continuum limit together with the renormalisation depending on the spectral dimension $\D$.
We will prove also in this section that the renormalised
Kontsevich model obeys topological recursion by inverting a linear integral operator 
with combinatorial methods using Bell polynomials. 
Furthermore, we will construct a boundary creation operator independent of the dimension. 
The free energy will then be given in \sref{Sec.}{Sec:CubicFreeEnergy}.
As by-product
by combining the boundary insertion operator with the inverse of an integral operator,
we will derive a second-order differential operator
to compute intersection numbers of $\psi$-classes. 
By doing so, we will prove that the 
stable partition function of the renormalised Kontsevich model
is annihilated by the generators of a \textit{deformed Virasoro algebra} 
due to the change of an implicitly defined constant $c$ depending on the dimension. 
In \sref{Sec.}{Sec:CubicOverRenom} we take up the 
question, whether \textit{over-subtraction} or rather \textit{over-renormalisation} will create problems 
for the cubic model. 

\section[Schwinger-Dyson Equations]{Schwinger-Dyson Equations
\footnote{The SDEs of this section were already derived in \cite{Grosse:2016qmk}}}
\label{Sec:CubicSchwinger}
All SDEs will be derived with the complete set of renormalisation constants for $\D<8$.
Due to the tadpole renormalisation, the action is equipped with an additional linear term such that
the renormalised Kontsevich model becomes
\begin{align}\label{eq:actioncubic}
  S[\Phi]&=\!V\bigg(\sum_{n,m=0}^\mN \!\!\!
Z\frac{H_{nm}}{2}
\Phi_{nm}\Phi_{mn}\!
 + \!\! \sum_{n=0}^\mN
(\kappa
{+}\nu E_{n}{+}\zeta E_{n}^2) 
\Phi_{nn}\!
+\frac{\lambda_{bare}Z^{3/2}}{3} \!\!\!\!\!
\sum_{n,m,k=0}^\mN
\!\!\!\!\!
\Phi_{nm}
 \Phi_{mk}
\Phi_{kn}\!\bigg),\\
H_{nm}&=E_n+E_m,\nonumber
\end{align}
where two eigenvalues $E_i,E_j$ are not necessarily different. 
The partition function is also slightly affected by the renormalisation 
constants and is similarly to \eqref{eq:Part2} given by
\begin{align}
 \mathcal{Z}[J]=&\int \mathcal{D} \Phi \exp\left(-S[\Phi]+V\mathrm{Tr}(J\Phi) \right)\nonumber\\\label{eq:partioncubic}
 =&K\,\exp\Big(-\frac{\lambda_{bare}Z^{3/2} }{3V^2}\sum_{n,m,k=0}^\mN
 \frac{\partial^3}{\partial J_{nm}\partial 
 J_{mk}\partial J_{kn}}
 \Big)\mathcal{Z}_{\mathrm{free}}[J],
\\\nonumber
\mathcal{Z}_{\mathrm{free}}[J]:=&
 \exp\!\bigg(\!V\!\!\!\sum_{n,m=0}^\mN \!\frac{(J_{nm}-(\kappa+\nu 
 E_{n}+\zeta E_{n}^2)\delta_{m,
 n})(J_{mn}-(\kappa+\nu 
 E_{n}+\zeta E_{n}^2)\delta_{m,
 n})}{2ZH_{nm}}\bigg),
\end{align}
and $K:=\int \mathcal{D}\Phi\exp\Big(-VZ\sum_{n,m=0}^\mN\frac{H_{nm}}{2}
\Phi_{nm}\Phi_{mn}\Big)=\prod_{n,m=0}^\mN\sqrt{\frac{2\pi}{VZH_{nm}}}$. Finally, we recall the 
impact of the renormalisation constants on the Ward-Takahashi identity of \sref{Theorem}{Thm:Raimar}
\begin{align}\nonumber
 \sum_{n=0}^\mN \frac{\partial^2\Z[J]}{\partial J_{qn}\partial J_{np}}=&
  \frac{V}{(E_p-E_q)Z}\sum_{n=0}^\mN
  \left(J_{pn}\frac{\partial}{\partial J_{qn}}-J_{nq}\frac{\partial}{\partial J_{np}}\right)\Z[J]
  -\frac{V}{Z}(\nu+\zeta H_{pq})\frac{\partial \Z[J]}{\partial J_{qp}}\\
  &+\delta_{E_p,E_q}( W^1_p[J]+W^2_p[J])\Z[J].\label{eq:Wardcubic}
\end{align}
Notice that \eqref{eq:Wardcubic} holds even for $E_p=E_q$ by regularity assumption
even if $p\mapsto E_p$ is not injective. The later computed correlation functions depend only on the 
distinct eigenvalues $e_k$ with $k\mapsto e_k$ injective. The sum can then 
be written as $\sum_{n=0}^\mN\to\sum_{k=0}^{\mN'}r_k$, where 
$\mN'$ are the numbers of distinct eigenvalues and $r_k$ its multiplicities.

The renormalisation constants have singular behaviour in the later taken limit depending
on the spectral dimension $\D$. Nevertheless, an appropriate choice of the constants leaves the correlation
functions finite after removing the cut-off. Notice that there is still a freedom in the choice 
of the renormalisation constants by boundary conditions which we will fix
for $\D<6$ differently than in the perturbative
expansion of Zimmermann's forest formula:
\begin{align}\label{eq:cond1}
 &\D\geq 2: &&G^{(0)}_{|0|}=0,\\\label{eq:cond2}
 &\D\geq 4: &&\frac{\partial}{\partial p}G^{(0)}_{|p|}\bigg\vert_{p=0}=0,\quad G^{(0)}_{|00|}=1,\\\label{eq:cond3}
 &\D\geq 6: &&\frac{\partial^2}{\partial p^2}G^{(0)}_{|p|}\bigg\vert_{p=0}=0,
 \quad \frac{\partial}{\partial p}G^{(0)}_{|pq|}\bigg\vert_{p=q=0}=
 \frac{\partial}{\partial q}G^{(0)}_{|pq|}\bigg\vert_{p=q=0}=0.
\end{align}
Let us first determine the 2-point function $G_{|pq|}$, where
$E_p\neq E_q$,  recursively from the 1-point function $G_{|p|}$ by
\begin{align*}
 G_{|pq|}=&\frac{1}{ZH_{pq}}-\frac{\lambda_{bare}Z^{1/2}}{V^2 H_{pq}}\frac{1}{\Z[0]}
 \sum_{n=0}^\mN \frac{\partial^3}{\partial J_{pq}\partial J_{qn}\partial J_{np}}
 \Z[J]\bigg\vert_{J=0}\\
 =&\frac{1}{ZH_{pq}}\bigg(1+\lambda_{bare}Z^{1/2}\frac{G_{|p|}-G_{|q|}}
  {E_p-E_n}+\lambda_{bare}Z^{1/2}(\nu+\zeta H_{pq})G_{|pq|}\bigg),
\end{align*}
where equation \eqref{eq:SDeqGeneral} together with \eqref{eq:partioncubic} 
and \eqref{eq:Wardcubic} is applied. This expression is now equivalent 
to 
\begin{align}\label{eq:2to1}
 G_{|pq|}=&\frac{1}
 {Z(1-Z^{-1/2}\lambda_{bare} \zeta)(E_p+E_q-\frac{Z^{-1/2} \lambda_{bare}\nu}{1-Z^{-1/2}\lambda_{bare} \zeta})}
 \\\nonumber
 &+\frac{\lambda_{bare} Z^{1/2}}{Z(1-Z^{-1/2}\lambda_{bare} \zeta)} 
 \frac{G_{|p|}-G_{|q|}}{(E_p-E_q)(E_p+E_q-\frac{Z^{-1/2} \lambda_{bare}\nu}{1-Z^{-1/2}\lambda_{bare} \zeta})}.
\end{align}
Notice that the limit $E_q\to E_p$ is uncritical since the 1-point function can be continued to 
differentiable function, where the derivative $\frac{\partial}{\partial E_p}G_{|p|}$ has a meaningful 
expression. Furthermore, the perturbative expansion of $G_{|pp|}$ is unique. 

The two conditions \eqref{eq:cond2} and \eqref{eq:cond3} together with the expression \eqref{eq:2to1} yield for $g=0$ after 
genus expansion the relations
\begin{align}\label{eq:renorm1}
 \mu^2_{bare}-Z^{1/2}\lambda_{bare} \nu=&1\\\label{eq:renorm2}
 Z(1-Z^{-1/2}\lambda_{bare}\zeta)=&1,
\end{align}
where $\frac{\mu^2_{bare}}{2}$ is the smallest eigenvalue of $E$. It is also more convenient to 
define the renormalised coupling constant $\lambda$ and the shifted eigenvalues $F_p$ 
(which are UV-finite in the later taken limit) by
\begin{align}\label{eq:renorm3}
 \lambda:=&Z^{1/2}\lambda_{bare}\\\label{eq:renorm4}
 F_p:=&E_p-\frac{\lambda\nu}{2}.
\end{align}
Inserting \eqref{eq:renorm1}, \eqref{eq:renorm2} as well as \eqref{eq:renorm3}, \eqref{eq:renorm4} into 
\eqref{eq:2to1} gives the compact formula
\begin{align}\label{eq:2to11}
 G_{|pq|}=\frac{1}{F_p+F_q}+\lambda \frac{G_{|p|}-G_{|q|}}{F^2_p-F^2_q}.
\end{align}
The setup is now established, and we are able to derive all SDEs starting with the 1-point function:
\begin{prps}\label{Prop:CubicOneP}
 The shifted 1-point function $W_{|p|}:=2\lambda G_{|p|}+2F_p$
 satisfies 
 \begin{align*}
  (W_{|p|})^2+2\lambda \nu W_{|p|}+\frac{2\lambda^2}{V}\sum_{n=0}^\mN
 \frac{W_{|p|}-W_{|n|}}{F^2_{p}-F^2_{n}}+
 \frac{4\lambda^2}{V^2}G_{|p|p|}=&\frac{4F_{p}^2}{Z}+C,
 \end{align*}
where $C:=-\frac{\lambda^2\nu^2(1+Z)+4 \kappa \lambda}{Z}$.
In particular, the genus expansion $W_{|p|}=:\sum_{g=0}^\infty V^{-2g} 
W^{(g)}_{|p|}=2F_p+2\lambda \sum_{g=0}^\infty V^{-2g} G^{(g)}_{|p|}$ gives a linear equation for 
$g>0$ and a nonlinear one for $g=0$:
\begin{align}\label{g11}
\sum_{h+h'=g} W^{(h)}_{|p|}W^{(h')}_{|p|}+2\lambda\nu W^{(g)}_{|p|}+\frac{2\lambda^2}{V}
\sum_{n=0}^\mN
\frac{W^{(g)}_{|p|}-W^{(g)}_{|n|}}{F^2_{p}-
F^2_{n}}+4\lambda^2G^{(g-1)}_{|p|p|}=
\delta_{0,g}\Big(\frac{4F_{p}^2}{Z}+C\Big).
\end{align}
\begin{proof}
 By definition, the 1-point function is given by
 \begin{align*}
  G_{|p|}=&\frac{1}{V}\frac{\partial}{\partial J_{pp}}\log\Z[J]\bigg\vert_{J=0}\\
  =&\frac{1}{ZH_{pp}}\bigg(-\kappa-\nu E_p-\zeta E_p^2-\frac{\lambda_{bare}Z^{3/2}}{V^2\Z[0]}\sum_{n=0}^\mN
  \frac{\partial^2}{\partial J_{pn}\partial J_{np}}\Z[J]\bigg)\bigg\vert_{J=0}\\
  =&\frac{1}{ZH_{pp}}\bigg(-\kappa-\nu E_p-\zeta E_p^2-\lambda_{bare}Z^{3/2}
  \big((G_{|p|})^2
  + \frac{ G_{|p|p|}}{V^2}
  +\frac{1}{V}\sum_{n=0}^\mN G_{|np|}\big)\bigg),
 \end{align*}
where the second line is achieved by using \eqref{eq:SDeqGeneral} and considering the additional 
terms via \eqref{eq:partioncubic}. The definition of the correlation function gives the last line,
where $n=p$ produces the quadratic term and the $(1+1)$-point function. 

Inserting the conditions \eqref{eq:renorm1} and \eqref{eq:renorm2} together with the 
definitions \eqref{eq:renorm3}, \eqref{eq:renorm4} and $W_{|p|}:=2\lambda G_{|p|}+2F_p$ 
leads to the first equation of the proposition since $G_{|pq|}$ (see \eqref{eq:2to11}) becomes
\begin{align}\label{eq:Cubic2point}
 G_{|pq|}=\frac{1}{2}\frac{W_{|p|}-W_{|q|}}{F^2_{p}-F^2_{q}}.
\end{align}
Expanding the shifted 1-point and the $(1+1)$-point function in the genus expansion, gives
at order $V^{-2g}$ the second equation of the proposition.
\end{proof}
\end{prps}
\begin{rmk}
 The shift from $G_{|p|}$ to $W_{|p|}$ is understood as transformation $\Phi\mapsto \Phi'=\Phi-\frac{E}{\lambda}$ 
 of the integration variable of
 the partition function, which cancels the quadratic term 
 of the action such that the action only consists of a linear and a cubic term in $\Phi'$, and a constant.
\end{rmk}
\noindent
The generalisation of the recursive equation \eqref{eq:2to11} to all correlation functions is given by:
\begin{prps}\label{Prop:CubisRecur}
For any $j\in\{1,..,b\}$ and $i\in\{1,..,N_j\}$,
 the $(N_1+N_2+..+N_b)$-point function is
 given recursively in terms of the $(N_1+..+N_{j-1}+(N_j-1)+N_{j+1}+..+N_b)$-point function
 \begin{align*}
 G_{|p_1^1..p_{N_1}^1|..|p_1^b..p_{N_b}^b|}=-\lambda
  \frac{G_{|p_1^1..|p_1^j..p_{i-1}^jp_{i+1}^j..p_{N_j}^j|..p_{N_b}^b|}-
  G_{|p_1^1..|p_1^j..p_{i}^jp_{i+2}^j..p_{N_j}^j|..p_{N_b}^b|}}
  {F^2_{p_i^j}-F^2_{p_{i+1}^j}}.
\end{align*}
 \begin{proof}
  Let $\frac{\partial^N}{\partial \J_{p_1..p_N}}:=
  \frac{\partial^N}{\partial J_{p_1p_2}..\partial J_{p_{N-1}p_N}\partial J_{p_Np_1}}$ and\\
  $\frac{\partial^N}{\partial \hat{\J^i}_{p_1..p_N}}:=
  \frac{\partial^N}{\partial J_{p_1p_2}..\partial J_{p_{i-1}p_i}
  \partial J_{p_{i+1}p_{i+2}}..\partial J_{p_{N-1}p_N}\partial J_{p_Np_1}}$ with omitted 
   $\frac{\partial}{\partial J_{p_ip_{i+1}}}$-derivative. Let the $E_{p_j^i}$'s be pairwise different.
  By definition, we have
  \begin{align*}
   &G_{|p_1^1p_2^1..p_{N_1}^1|..|p_1^b..p_{N_b}^b|}=V^{b-2}
   \frac{\partial^{N_1+..+N_b}}{\partial \J_{p^1_1..p^1_{N_1}}..\partial \J_{p^b_1..p^b_{N_b}}}
   \log\frac{\Z[J]}{\Z[0]}\bigg\vert_{J=0}\\
   &=-\frac{V^{b-3}\lambda_{bare}Z^{3/2}}{ZH_{p_i^jp_{i+1}^j}}
   \frac{\partial^{N_1+..+N_b-1}}
   {\partial \J_{p^1_1..p^1_{N_1}}..\partial \hat{\J}^i_{p^j_1..p^j_{N_j}}..\partial \J_{p^b_1..p^b_{N_b}}}
   \frac{1}{\Z[J]}\sum_{n=0}^\mN\frac{\partial^2}{\partial J_{p^j_in}\partial J_{np^j_{i+1}}}\Z[J]\bigg\vert_{J=0}\\
   &=-\lambda_{bare}Z^{-1/2}\frac{G_{|p_1^1..p_{N_1}^1|..|p_1^j..p_{i-1}^jp_{i+1}^j..p_{N_j}^j|..|p_1^b..p_{N_b}^b|}-
  G_{|p_1^1..p_{N_1}^1|..|p_1^j..p_{i}^jp_{i+2}^j..p_{N_j}^j|..|p_1^b..p_{N_b}^b|}}
  {E^2_{p_i^j}-E^2_{p_{i+1}^j}}\\
  & \qquad +\frac{\lambda_{bare}(\nu+\zeta H_{p_i^jp_{i+1}^j})}{Z^{1/2}H_{p_i^jp_{i+1}^j}}
  G_{|p_1^1..p_{N_1}^1|..|p_1^b..p_{N_b}^b|},
  \end{align*}
where
  the second line is achieved by applying \eqref{eq:SDeqGeneral} and in the last two lines by \eqref{eq:Wardcubic}. 
  Inserting the conditions \eqref{eq:renorm1}, \eqref{eq:renorm2} together with the 
definitions \eqref{eq:renorm3}, \eqref{eq:renorm4} leads to the result. 

Since the lhs is regular if two or more $E_{p_j^i}$'s coincide, the rhs has a well-defined limit which can be understood by 
continuing the correlation functions to differentiable functions.
 \end{proof}
\end{prps}\noindent
We emphasise that \sref{Proposition}{Prop:CubisRecur} implies recursively that any $(N_1+..+N_b)$-point function
depends \textit{linearly} on $(1+..+1)$-point functions with $b$ boundary components. 
The explicit result is formulated in:
\begin{prps}\label{Prop:CubisExpl}
 The explicit formula for the $(N_1+..+N_b)$-point function is
 \begin{align*}
  G_{|p_1^1..p_{N_1}^1|..|p_1^b..p_{N_b}^b|}=\lambda^{N-b}\sum_{k_1=1}^{N_1}..\sum_{k_b=1}^{N_b}
  G_{|p^1_{k_1}|p^2_{k_2}|..|p^b_{k_b}|}\prod_{\beta=1}^b\prod_{\substack{l_\beta=1\\ 
  l_{\beta}\neq k_{\beta}}}^{N_\beta} \frac{1}{F^2_{p_{k_{\beta}}^{\beta}}-F^2_{p_{l_{\beta}}^{\beta}}},
 \end{align*}
where $N:=N_1+..+N_b$, and for $b=1$
\begin{align*}
  G_{|p_1..p_{N}|}=\lambda^{N-1}\sum_{k=1}^{N}
  \frac{W_{|p_{k}|}}{2\lambda}\prod_{\substack{l=1\\ 
  l\neq k}}^{N} \frac{1}{F^2_{p_{k}}-F^2_{p_{l}}}.
 \end{align*}

\begin{proof}
 The proof is provided by induction. For $N_1=2$ and $N_\beta=1$ for all $\beta\in\{2,..,b\}$ it holds. 
 Assume it holds for the $(N_1+..+N_b)$-point function, then we have for the $((N_1+1)+N_2+..+N_b)$-point function 
 with \sref{Proposition}{Prop:CubisRecur}
 \begin{align*}
  &G_{|p_1^1..p_{N_1+1}^1|..|p_1^b..p_{N_b}^b|}=-\lambda 
  \frac{G_{|p_2^1..p_{N_1+1}^1|..|p_1^b..p_{N_b}^b|}-G_{|p_1^1p_3^1..p_{N_1+1}^1|..|p_1^b..p_{N_b}^b|}}
  {F^2_{p^1_1}-F^2_{p^1_{2}}}\\
  &=-\lambda^{N+1-b} \sum_{k_2=1}^{N_2}..\sum_{k_b=1}^{N_b}\frac{1}{F^2_{p^1_1}-F^2_{p^1_{2}}}\\
&\times  \left(\sum_{k_1=2}^{N_1+1} G_{|p^1_{k_1}|p^2_{k_2}|..|p^b_{k_b}|}
  \prod_{\substack{l_1=2\\ 
  l_1\neq k_1}}^{N_1+1} \frac{1}{F^2_{p_{k_{1}}^1}-F^2_{p_{l_{1}}^{1}}}
  -\sum_{\substack{k_1=1\\k_1\neq 2}}^{N_1+1} G_{|p^1_{k_1}|p^2_{k_2}|..|p^b_{k_b}|}
  \prod_{\substack{l_1=1\\ 
  l_1\neq k_1,\l_1\neq 2}}^{N_1+1} \frac{1}{F^2_{p_{k_{1}}^1}-F^2_{p_{l_{1}}^{1}}}
  \right)\\
  &\times \prod_{\beta=2}^b\prod_{\substack{l_\beta=1\\ 
  l_{\beta}\neq k_{\beta}}}^{N_\beta} \frac{1}{F^2_{p_{k_{\beta}}^{\beta}}-F^2_{p_{l_{\beta}}^{\beta}}}\\
  &=-\lambda^{N+1-b} \sum_{k_2=1}^{N_2}..\sum_{k_b=1}^{N_b}\bigg[G_{|p^1_{1}|p^2_{k_2}|..|p^b_{k_b}|}
  \prod_{l_1=2}^{N_1+1} \frac{1}{F^2_{p_{1}^1}-F^2_{p_{l_{1}}^{1}}} 
+G_{|p^1_{2}|p^2_{k_2}|..|p^b_{k_b}|}
  \prod_{\substack{l_1=1\\l_1\neq2}}^{N_1+1} \frac{1}{F^2_{p_{2}^1}-F^2_{p_{l_{1}}^{1}}}\\
  &+\sum_{k_1=3}^{N_1+1} G_{|p^1_{k_1}|p^2_{k_2}|..|p^b_{k_b}|}\frac{1}{F^2_{p^1_1}-F^2_{p^1_{2}}}\bigg\{
  \prod_{\substack{l_1=2\\ 
  l_1\neq k_1}}^{N_1+1} \frac{1}{F^2_{p_{k_{1}}^1}-F^2_{p_{l_{1}}^{1}}}
  -\prod_{\substack{l_1=1\\ 
  l_1\neq k_1,\l_1\neq 2}}^{N_1+1} \frac{1}{F^2_{p_{k_{1}}^1}-F^2_{p_{l_{1}}^{1}}}
  \bigg\}\bigg]\\
  &\times \prod_{\beta=2}^b\prod_{\substack{l_\beta=1\\ 
  l_{\beta}\neq k_{\beta}}}^{N_\beta} \frac{1}{F^2_{p_{k_{\beta}}^{\beta}}-F^2_{p_{l_{\beta}}^{\beta}}}.
 \end{align*}
The curly brackets factorise $\prod_{\substack{l_1=3\\ 
  l_1\neq k_1}}^{N_1+1} \frac{1}{F^2_{p_{k_{1}}^1}-F^2_{p_{l_{1}}^{1}}}$, what remains is
  \begin{align*}
   \frac{1}{F^2_{p^1_1}-F^2_{p^1_{2}}}\bigg(\frac{1}{F^2_{p_{k_{1}}^1}-F^2_{p_{2}^{1}}}-
   \frac{1}{F^2_{p_{k_{1}}^1}-F^2_{p_{1}^{1}}}\bigg)=-\frac{1}{F^2_{p_{k_{1}}^1}-F^2_{p_{2}^{1}}}
   \frac{1}{F^2_{p_{k_{1}}^1}-F^2_{p_{1}^{1}}}
  \end{align*}
for any $k_1>2$. We conclude for any $k_1>2$ the factors $\prod_{\substack{l_1=1\\ 
  l_1\neq k_1}}^{N_1+1} \frac{1}{F^2_{p_{k_{1}}^1}-F^2_{p_{l_{1}}^{1}}}$, which finishes the proof for $N_1\to N_1+1$.
  Since the correlation functions are symmetric between the boundary components, it is also proved for any $N_\beta
  \to N_\beta+1$.
  
  For $b=1$, the proof proceeds analogously. For $N=2$, the starting point is given by 
  \eqref{eq:Cubic2point}.
\end{proof}
\end{prps}

\begin{rmk}
 Notice that the explicit formula of \sref{Proposition}{Prop:CubisExpl} is fully symmetric within 
 any boundary component. In other words, a $(N_1+..+N_b)$-point function of the cubic model 
 is not only cyclic symmetric within a boundary (by definition), it is fully symmetric.
This is even surprising from the perspective of perturbative expansion.  For example:
 \begin{align*}
  G_{|p_1p_2p_3p_4|}=G_{|p_2p_1p_3p_4|}.
  \end{align*}
\end{rmk}
\noindent
After genus expansion, we observe that any $(N_1+..+N_b)$-point function of genus $g$ is expressed by 
the $(1+..+1)$-point function of genus $g$ with $b$ boundary components. The final SDE to study
is therefore the equation of the $(1+..+1)$-point function.

We impose the shorthand notation $G_{|I|}:=G_{|i^1|..|i^b|}$ for the set $I:=\{i^1,..,i^b\}$ with
$|I|=b$ and $i^j\in \{0,..,\mN\}$.
\begin{prps}\label{Prop:CubicNP}
 Let $J=\{p^2,..,p^b\}$, then the $(1+..+1)$-point function
 with $b$ boundary components satisfies the linear equation
 \begin{align*}
  &(W_{|p^1|}+\nu\lambda)G_{|p^1|J|}+\frac{\lambda^2}{V}\sum_{n=0}^\mN 
  \frac{G_{|p^1|J|}-G_{|n|J|}}{F^2_{p^1}-F^2_{n}}\\
  &\qquad =-\lambda \sum_{\beta=2}^bG_{|p^1p^\beta p^\beta|J\backslash\{\beta\}|}
  -\frac{\lambda}{V^2}G_{|p^1|p^1|J|}-\lambda \sum_{\substack{I\uplus I'=J\\
  0\neq |I|\neq b}}G_{|p^1|I|}G_{|p^1|I'|}.
 \end{align*}
In particular, the genus $g$ correlation function satisfies
\begin{align*}
 &\sum_{h+h'=g}(W^{(h)}_{|p^1|}+\delta_{h,0}\nu\lambda)G^{(h')}_{|p^1|J|}+\frac{\lambda^2}{V}\sum_{n=0}^\mN 
  \frac{G^{(g)}_{|p^1|J|}-G^{(g)}_{|n|J|}}{F^2_{p^1}-F^2_{n}}\\
  &\qquad =-\lambda \sum_{\beta=2}^bG^{(g)}_{|p^1p^\beta p^\beta|J\backslash\{p^\beta\}|}
  -\lambda G^{(g-1)}_{|p^1|p^1|J|}-\lambda \sum_{h+h'=g}\sum_{\substack{I\uplus I'=J\\
  0\neq |I|\neq b}}G^{(h)}_{|p^1|I|}G^{(h')}_{|p^1|I'|}.
\end{align*}
\begin{proof}
 Let the $E_{p^\beta}$'s be pairwise different. It follows from the definition
 \begin{align*}
  &G_{|p^1|J|}=V^{b-2}\frac{\partial^b}{\partial J_{p^1p^1}\partial J_{p^2p^2}..\partial J_{p^bp^b}}
  \log\frac{\Z[J]}{\Z[0]}\bigg\vert_{J=0}\\
  =&-\frac{\lambda_{bare}Z^{1/2}V^{b-3}}{H_{p^1p^1}}
  \frac{\partial^{b-1}}{\partial J_{p^2p^2}..\partial J_{p^bp^b}}\frac{1}{\Z[J]}
  \sum_{n=0}^\mN \frac{\partial^2}{\partial J_{p^1n}\partial J_{np^1}}\Z[J]\bigg\vert_{J=0}\\
  =&-\frac{\lambda_{bare}Z^{1/2}V^{b-3}}{H_{p^1p^1}}
  \frac{\partial^{b-1}}{\partial J_{p^2p^2}..\partial J_{p^bp^b}}
  \sum_{n=0}^\mN \bigg(\frac{\partial^2\log \Z[J]}{\partial J_{p^1n}\partial J_{np^1}}+
  \frac{\partial\log \Z[J]}{\partial J_{np^1}}
  \frac{\partial\log \Z[J]}{\partial J_{p^1n}}\bigg)\bigg\vert_{J=0}\\
  =&-\frac{\lambda_{bare}Z^{1/2}}{H_{p^1p^1}}\bigg(\frac{1}{V}\!\sum_{n=0}^\mN G_{|p^1n|J|}+
  \frac{1}{V^2}G_{|p^1|p^1|J|}+G_{|p^1p^\beta p^\beta |J\backslash \{p^\beta\}|}
   +\sum_{I\uplus I'=J}G_{|p^1|I|}G_{|p^1|I'|}\bigg),
 \end{align*}
where in the last line the $(2+1+..+1)$-point function appears if $n=p^1$, and the $(3+1+..+1)$-point function 
if $n=p^\beta$ with $\beta\in\{2,..,b\}$. Shifting the quadratic term with $|I|=0$ and $|I|=b$ to the lhs and using $H_{p^1p^1}+\lambda_{bare}Z^{1/2}
2 G_{|p^1|}=W_{|p^1|}+\lambda_{bare}Z^{1/2} \nu$ leads to the provided equation with $\lambda_{bare}Z^{1/2}=\lambda$.

If two or more $E_{p^\beta}$'s coincide, both sides of the equation are uncritical. 
\end{proof}
\end{prps}
\noindent
Looking at all SDE of \sref{Proposition}{Prop:CubicOneP} and
\sref{Proposition}{Prop:CubicNP}, we notice that any correlation function $f(p)$ of Euler 
characteristic $\chi=2-2g-b$
can be computed by inverting the linear equation
\begin{align}\label{eq:cubicSD1}
 (W^{(0)}_{|p|}+\nu\lambda)f(p)+\frac{\lambda^2}{V}\sum_{n=0}^\mN 
  \frac{f(p)-f(n)}{F^2_{p}-F^2_{n}}= g_{inh}(p)
\end{align}
where $g_{inh}(p)$ is a inhomogeneity depending on correlation functions of Euler characteristic $\chi'>\chi$. Only the 
1-point function of genus $g=0$ plays a special r\^ole and satisfies a nonlinear equation. All solutions are known in case 
of $\D=0$ and obey topological recursion. For higher dimensions, the limit $\mN,V\to\infty$ 
needs to be performed before the correlation functions are computed.

\section{Solution of the Schwinger-Dyson Equations}\label{Sec:CubicSolution}
The solution of all planar correlation function was already found in \cite{Grosse:2016qmk}.
We have generalised these results in our paper \cite{Grosse:2019nes} which will be presented 
in this section.
\subsection{Large $\mN,V$-Limit}
Following the limit discussed in \sref{Section}{Sec:LargeLimit}, the distinct eigenvalues $(e_0,e_1,..,e_{\mN'})$ 
with $\mN'\leq \mN$ are given by the shifted eigenvalues $F_p=E_p-\frac{\lambda\nu}{2}$ of multiplicities $r_k$
\begin{align*}
 F_0=F_1=..=F_{r_0-1}=& \,e_0=\frac{1}{2}\\
 F_{r_0}=..=F_{r_0+r_1-1}=&\, e_1\\
  \vdots&
\end{align*}
The asymptotic behaviour of the multiplicities defines the spectral dimension of the model by 
$\lim_{k\to\infty}r_k\sim k^{\frac{\D}{2}-1}$ since $e_k$ can be assumed asymptotically linear in $k$ 
(see \sref{Remark}{Rmk:LargeN}). The monotonic, continuously differentiable function $e(x)$ with $e(0)=0$ 
and $r(x)$ are defined by \eqref{eq:econt} and \eqref{eq:rcont} which are in the $\mN',V$-limit unique
\begin{align*}
 \lim_{\mN',V\to\infty}e_k-\frac{\mu^2}{2}=&
 \lim_{\mN',V\to\infty}e\left(\frac{k}{V^{\frac{2}{\D}}}\right)=e(x)\stackrel{x\to\infty}{\sim} x\\
  \lim_{\mN',V\to\infty}\frac{r_k}{V^{1-\frac{2}{\D}}}=&\lim_{\mN',V\to\infty}
 r\left(\frac{k}{V^{\frac{2}{\D}}}\right)=r(x)\stackrel{x\to\infty}{\sim} x^{\frac{\D}{2}-1}.
\end{align*}
The ratio $\frac{\mN'}{V^{2/\D}}=\Lambda^2$ will be fixed 
in such a way that the sum converges to the integral with cut-off $\Lambda^2$ 
\begin{align*}
 \lim \sum_{n=0}^\mN f\left(\frac{n}{V^{2/\D}}\right)
 =\lim \sum_{k=0}^{\mN'} r_k f\left(\frac{k}{V^{2/\D}}\right)=\int_0^{\Lambda^2}dx\,r(x)f(x).
\end{align*}
The SDE of the previous section becomes
together with the continuation of the correlation function (unique 
for $\mN',V\to\infty$) defined in \eqref{eq:Gcont}:
\begin{align}\label{eq:CubicCont1P}
 &\sum_{h+h'=g} W^{h}(x)W^{h'}(x)+2\lambda\nu W^{g}(x)+2\lambda^2
\int_0^{\Lambda^2}dt\,r(t)
\frac{W^{g}(x)-W^{g}(t)}{(\frac{1}{2}+e(x))^2-
(\frac{1}{2}+e(t))^2}\\\nonumber
&=-4\lambda^2G^{g-1}(x|x)+
\delta_{0,g}\Big(\frac{(1+2e(x))^2}{Z}+C\Big),
\end{align}
\begin{align}\label{eq:CubicContNP}
&\sum_{h+h'=g}(W^{h}(x^1)+\delta_{h,0}\nu\lambda)G^{h'}(x^1|J)+\lambda^2\int_0^{\Lambda^2}dt\,r(t)
  \frac{G^{g}(x^1|J)-G^{g}(t|J)}{(\frac{1}{2}+e(x^1))^2-
(\frac{1}{2}+e(t))^2}\\\nonumber
  & =-\lambda \sum_{\beta=2}^bG^{g}(x^1,x^\beta,x^\beta|J\backslash\{x^\beta\})
  -\lambda G^{g-1}(x^1|x^1|J)-\lambda \sum_{\substack{
  h+h'=g\\
  I\uplus I'=J\\
  0\neq |I|\neq b}}G^{h}(x^1|I)G^{h'}(x^1|I'),
\end{align}
where $W^g(x):=2\lambda G^{g}(x)+\delta_{g,0}(1+2e(x))$, $C=\lim_{\mN,V}\big(
-\frac{\lambda^2\nu^2(1+Z)+4 \kappa \lambda}{Z}\big)$ and $J=\{x^2,..,x^b\}$. We used the
notation $G^g(I):=
G^g(y^1|..|y^b)$ for the set of variables $I:=\{y^1,..,y^b\}$ with
$|I|=b$ and $y^j\in [0,\Lambda^2]$.

\subsection{The Planar 1-Point Function}\label{Sec:1Pg0}
Computing the solution of the 1-point function is the hardest part in solving the model since it obeys 
the nonlinear integral equation. The 1-point function 
was originally solved in the appendix of \cite{Makeenko:1991ec} by transforming the problem to a 
boundary value problem, also known as Riemann-Hilbert problem. This solution can be extended from $\D=0$ to $\D<2$. 
The solution for higher spectral dimensions was then 
generalised by a similar ansatz in \cite{Grosse:2016qmk} up to $\D<8$.

To stay self-contained, we recall these results.
A more convenient form of the nonlinear integral equation \eqref{eq:CubicCont1P} for
 $g=0$ is achieved by the variable transformation
\begin{align}\label{eq:DefW1}
 X:=(1+2e(x))^2,\qquad \tilde{W}(X(x)):=W^0(x)
\end{align}
such that the nonlinear integral equation for the 1-point function takes the form
\begin{align}\label{eq:nonlin1}
  &\tilde{W}(X)^2+2\lambda\nu \tilde{W}(X)+
\int_1^{(1+2e(\Lambda^2))^2}dY\,\varrho(Y)
\frac{\tilde{W}(X)-\tilde{W}(Y)}{X-Y}+\frac{X}{Z}
=C,\\
&\text{where}\qquad  \varrho(Y):=\frac{2\lambda^2 \cdot r\!\left(e^{-1}(\frac{\sqrt{Y}-1}{2})\right)}{\sqrt{Y} \cdot
e'\!\!\left(e^{-1}(\frac{\sqrt{Y}-1}{2})\right)}.\nonumber
\end{align}
Its solution is:
\begin{prps}(\cite{Grosse:2016qmk})\label{Prop:1Punkt}
 For any measure $\varrho(Y)$, the nonlinear integral equation \eqref{eq:nonlin1} is solved by
 \begin{align*}
  &\tilde{W}(X)=\frac{\sqrt{X+c}}{\sqrt{Z}}-\lambda\nu +\frac{1}{2}
  \int_1^{(1+2e(\Lambda^2))^2}dY\frac{\varrho(Y)}{(\sqrt{X+c}+\sqrt{Y+c})\sqrt{Y+c}},
 \end{align*}
 where $c, Z$ and $\nu$ are fixed by the renormalisation conditions.
\begin{proof}
 Inserting the solution into the integral equation 
 leads with $\frac{\sqrt{X+c}-\sqrt{Y+c}}{X-Y}=\frac{1}{\sqrt{X+c}+\sqrt{Y+c}}$ to
 \begin{align*}
  &\int_1^{(1+2e(\Lambda^2))^2}\!\!\!\!\!\!\!\!\!\!\!\!\!\!\! dY\,\varrho(Y)
\frac{\tilde{W}(X)-\tilde{W}(Y)}{X-Y}\\
=&\frac{1}{\sqrt{Z}}\int_1^{(1+2e(\Lambda^2))^2}\!\!\!\!\!\!\!\!\!\!\!\!
\frac{dY\,\varrho(Y)}{\sqrt{X+c}+\sqrt{Y+c}}\\
-&\frac{1}{2}
\int_1^{(1+2e(\Lambda^2))^2}\!\!\!\!\!\!\!\!\!\!\!\!\!\!\!\!\!\!\!
\frac{dT\,\varrho(T)}{(\sqrt{X+c}+\sqrt{T+c})\sqrt{T+c}}
\int_1^{(1+2e(\Lambda^2))^2}\!\!\!\!\!\!\!\!\!\!\!\!\!\!\!\!\!\!\!
\frac{dY\,\varrho(Y)}{(\sqrt{X+c}+\sqrt{Y+c})(\sqrt{Y+c}+\sqrt{T+c})}.
 \end{align*}
Using in the last line the symmetrisation $\int_I \int_I dt\,dy\, f(t,y)=\int_I \int_I dt\,dy\, f(y,t)
=\frac{1}{2}\int_I \int_I dt\,dy\, (f(t,y)+f(y,t))$ factorises
the integrals
\begin{align*}
&\int_1^{(1+2e(\Lambda^2))^2}\!\!\!\!\!\!\!\!\!\!\!\!\!\!\! dY\,\varrho(Y)
\frac{\tilde{W}(X)-\tilde{W}(Y)}{X-Y}\\
 =&\frac{1}{\sqrt{Z}}\int_1^{(1+2e(\Lambda^2))^2}
\frac{dY\,\varrho(Y)}{\sqrt{Y+c}}-\frac{\sqrt{X+c}}{\sqrt{Z}}\int_1^{(1+2e(\Lambda^2))^2}\!\!\!\!\!\!\!\!\!\!\!\!\!\!\!\!\!\!\!
\frac{dY\,\varrho(Y)}{\sqrt{Y+c}(\sqrt{Y+c}+\sqrt{X+c})}\\
&\qquad -\frac{1}{4}\left(\int_1^{(1+2e(\Lambda^2))^2}\!\!\!\!\!\!\!\!\!\!\!\!\!\!\!\!\!\!\!
\frac{dY\,\varrho(Y)}{\sqrt{Y+c}(\sqrt{Y+c}+\sqrt{X+c})}\right)^2\\
=&-\tilde{W}(X)^2-2\lambda\nu\tilde{W}(X)+\frac{X+c}{Z}-(\lambda \nu)^2+\frac{1}{\sqrt{Z}}\int_1^{(1+2e(\Lambda^2))^2}
\frac{dY\,\varrho(Y)}{\sqrt{Y+c}}.
\end{align*}
Comparing this with \eqref{eq:nonlin1} identifies the constant $C$ through $c$. We therefore fix the renormalisation
constants
by \eqref{eq:cond1}, \eqref{eq:cond2} and \eqref{eq:cond3} and the definition of $\tilde{W}(X)$ 
depending on the spectral dimension.
\end{proof}
\end{prps}\noindent
The proposition provides the asymptotic behaviour for the solution $\tilde{W}(X)=\sqrt{X}
+\mathcal{O}((\sqrt{X}-1)^{\D/2})$. The conditions \eqref{eq:cond1}, \eqref{eq:cond2} and \eqref{eq:cond3}
to fix the remaining renormalisation constants are translated to
\begin{align}\label{eq:W1cond}
 \underbrace{\tilde{W}(1)=1}_{\D\geq 2},\qquad 
 \underbrace{\frac{\partial}{\partial X}\tilde{W}(X)\bigg\vert_{X=1}=\frac{1}{2}}_{\D\geq 4},\qquad 
 \underbrace{\frac{\partial^2}{\partial X^2}\tilde{W}(X)\bigg\vert_{X=1}=-\frac{1}{4}}_{\D\geq 6}.
\end{align}
Remember that for $\D<2$ non of the conditions are necessary 
which implies the following implicit equation for $c$ 
(from the last line of the proof of \sref{Proposition}{Prop:1Punkt} with $\nu=Z-1=C=0$):
\begin{align*}
 -c=\int_1^{(1+2e(\Lambda^2))^2}dY\frac{\varrho(Y)}{\sqrt{Y+c}}.
\end{align*}

\subsubsection*{Finite Matrices}
To recover the solution for finite $(\mN'+1)\times( \mN'+1)$ matrices with distinct eigenvalues $\frac{1}{2}=
e_0<..<e_{\mN'}$ of the 
external matrix $E$ with multiplicities $r_0,..,r_{\mN'}$, the discrete Dirac measure can be used 
$\varrho(Y)=\frac{8\lambda^2}{V}\sum_{k=0}^{\mN'}r_k\delta(Y-4e_k^2)$.  This provides the solution
\begin{align*}
 &W^0(e_n)=\tilde{W}(4e_n^2)=\sqrt{4e_n^2+c}+\frac{4\lambda^2}{V}
  \sum_{k=0}^{\mN'} \frac{r_k}{(\sqrt{4e_n^2+c}+\sqrt{4e_k^2+c})\sqrt{4e_k^2+c}},\\
  &\text{with}\qquad -c=\frac{8\lambda^2}{V}\sum_{k=0}^{\mN'} \frac{r_k}{\sqrt{4e_k^2+c}}.
\end{align*}

\subsubsection*{$0\leq \D< 2$}
All renormalisation constants are independent of $\Lambda^2$ and can be set 
to $Z=1$, $\nu=0$, $C=0$ and $\mu_{bare}^2=1$. The solution
becomes
\begin{align*}
  &\tilde{W}(X)=\sqrt{X+c}+\frac{1}{2}
  \int_1^\infty \frac{dY \,\varrho(Y)}{(\sqrt{X+c}+\sqrt{Y+c})\sqrt{Y+c}},\\
  &\text{with}\qquad -c=\int_1^\infty \frac{dY\,\varrho(Y)}{\sqrt{Y+c}}.
 \end{align*}
The integral of the implicit equation of $c$ converges since $r\!\left(e^{-1}(\frac{\sqrt{Y}-1}{2})\right)\sim 
{\sqrt{Y}}^{\frac{\D}{2}-1}=Y^{\frac{\D}{4}-\frac{1}{2}}$ and therefore $\varrho(Y)\sim
\frac{1}{Y^{1-\frac{\D}{4}}}$. The spectral dimension $\D=2$ is the critical dimension where the implicit equation 
diverges logarithmically.

Notice that the constant $c$ depends on $r(x),e(x)$ and $\lambda$ with $c=0$ for $\lambda=0$. From the implicit 
function theorem we know that $c(\lambda)$ is unique and differentiable in an open neighbourhood about $\lambda=0$.

\subsubsection*{$2\leq \D< 4$}
The first condition of \eqref{eq:W1cond} implies the following
implicit equation for $c$:
\begin{align}\label{eq:cond1Wc}
 1=\frac{\sqrt{1+c}}{\sqrt{Z}}-\lambda\nu +\frac{1}{2}
  \int_1^{(1+2e(\Lambda^2))^2}\!\!\!\!\!\!\!\!\!\!\!\!\frac{dY\,\varrho(Y)}{(\sqrt{1+c}+\sqrt{Y+c})\sqrt{Y+c}}.
\end{align}
In this case the limit of the cut-off $\Lambda^2\to\infty$ is 
safe since $\varrho(Y)\sim
\frac{1}{Y^{1-\frac{\D}{4}}}$. The integral representation of $\tilde{W}(X)$ 
of \sref{Propsition}{Prop:1Punkt} converges for $\Lambda^2\to\infty$. For $\D=4$, the expression diverges
logarithmically.
The remaining 
renormalisation constants are $Z=1$ and $\nu=0$. 
Inserting the renormalisation constants yields
\begin{align*}
  &\tilde{W}(X)=\sqrt{X+c} +\frac{1}{2}
  \int_1^\infty \frac{dY\,\varrho(Y)}{(\sqrt{X+c}+\sqrt{Y+c})\sqrt{Y+c}},\\
 &1-\sqrt{1+c}=\frac{1}{2}
  \int_1^\infty \frac{dY\,\varrho(Y)}{(\sqrt{1+c}+\sqrt{Y+c})\sqrt{Y+c}}.
\end{align*}
Since the rhs of the implicit function is positive for real $\lambda$, the lhs indicates $c\in]-1,0]$. The argumentation 
extends for complex $\lambda$ to $c\in\C\backslash]-\infty,-1]$.

\subsubsection*{$4\leq \D<6$}
Subtracting the implicit function \eqref{eq:cond1Wc} 
from the solution of $\tilde{W}(X)$ of \sref{Proposition}{Prop:1Punkt} cancels $\lambda\nu$
and leads to
\begin{align}\label{eq:D41P}
 \frac{\tilde{W}(X)-1}{\sqrt{X+c}-\sqrt{1+c}}=\frac{1}{\sqrt{Z}}-\frac{1}{2}
 \int_1^{(1+2e(\Lambda^2))^2}\!\!\!\!\!\!\!\!\!\!\!\!\!\!\!\!\!\!\!
 \frac{ dY\,\varrho(Y)}{(\sqrt{X+c}+\sqrt{Y+c})(\sqrt{1+c}+\sqrt{Y+c})\sqrt{Y+c}}.
\end{align}
Since $\varrho(Y)\sim
Y^{\frac{\D}{4}-1}$, the limit $\Lambda^2\to\infty$ implies logarithmic divergence for $\D=6$. The additional 
condition $\frac{\partial}{\partial X}\tilde{W}(X)\bigg\vert_{X=1}=\frac{1}{2}$
provides
\begin{align}\label{eq:D4c}
 \sqrt{1+c}=\frac{1}{\sqrt{Z}}-\frac{1}{2}
 \int_1^{(1+2e(\Lambda^2))^2}\!\!\!\!\!\!\!\!\!\!\!\!\!\!\!\!\!\!\!
 \frac{ dY\,\varrho(Y)}{(\sqrt{1+c}+\sqrt{Y+c})^2\sqrt{Y+c}},
\end{align}
where $Z=1$ is safe. Sending $\Lambda^2\to\infty$ in \eqref{eq:D41P} and \eqref{eq:D4c} gives the renormalised 
solution for $4\leq \D<6$.

\subsubsection*{$6\leq \D<8$}
Since \eqref{eq:D41P} is for $\D\geq6$ not divergent for $Z=1$, we eliminate $Z$ through \eqref{eq:D4c} and get
\begin{align*}
 \tilde{W}(X)=&\sqrt{X+c}\sqrt{1+c}-c\\
 &+\frac{1}{2}
 \int_1^{(1+2e(\Lambda^2))^2}\!\!\!\!\!\!\!\!\!\!\!\!\!\!\!\!\!\!\!
 \frac{ dY\,\varrho(Y) (\sqrt{X+c}-\sqrt{1+c})^2}{(\sqrt{X+c}+\sqrt{Y+c})(\sqrt{1+c}+\sqrt{Y+c})^2\sqrt{Y+c}},
\end{align*}
where now $\Lambda\to\infty$ is finite for all $\D<8$ since $\varrho(Y)\sim
Y^{\frac{\D}{4}-1}$. The last condition $\frac{\partial^2}{\partial X^2}\tilde{W}(X)\bigg\vert_{X=1}=-\frac{1}{4}$ gives
\begin{align*}
 -c=\int_1^\infty \frac{dY\,\varrho(Y)}{(\sqrt{1+c}+\sqrt{Y+c})^3\sqrt{Y+c}}.
\end{align*}

Collecting all three cases gives the general result:
\begin{cor}\label{Coro:1P}
Let $D=2\lfloor \frac{\D}{2}\rfloor$.
 The UV-finite, shifted planar 1-point function \eqref{eq:DefW1} is for the cubic matrix field
 of spectral dimension $\D<8$ with the 
 renormalisation conditions \eqref{eq:W1cond} given by
 \begin{align*}
  \tilde{W}(X)=&\sqrt{X+c}\left(\sqrt{1+c}\right)^{\delta_{D,6}}+\delta_{D,4}(1-\sqrt{1+c})
  -\delta_{D,6} c\\
  &+\frac{1}{2}
 \int_1^{\infty}
 \frac{ dY\,\varrho(Y) (\sqrt{1+c}-\sqrt{X+c})^{\frac{D}{2}-1}}
 {(\sqrt{X+c}+\sqrt{Y+c})(\sqrt{1+c}+\sqrt{Y+c})^{\frac{D}{2}-1}\sqrt{Y+c}},
 \end{align*}
 \begin{align*}
 \text{where}\quad & (1-\sqrt{1+c})\left(\frac{1+\sqrt{1+c}}{2}\right)^{\!\delta_{D,6}+\delta_{D,0}}\!\!\!\!\!\!=
 \frac{1}{2}
 \int_1^{\infty}
 \frac{ dY\,\varrho(Y) }{(\sqrt{1+c}+\sqrt{Y+c})^{D/2}\sqrt{Y+c}}.
 \end{align*}
\end{cor}\noindent
In all cases, $c(\lambda)$ is a differentiable function in a small neighbourhood about $\lambda=0$. The expansion of 
$c$ in $\lambda$ is expressed by the Lagrange inverse theorem
\begin{align}\label{eq:cexpansion}
 c=\sum_{n=1}^\infty \frac{1}{n!}\frac{d^{n-1}}{dw^{n-1}}
\Bigg\vert_{w=0}\left(\frac{\frac{w}{2}\int_{1}^{\infty}dY
\frac{\varrho(Y)}{(\sqrt{1+w}+\sqrt{Y+w})^{D/2}\sqrt{Y+w}}}{
(1-\sqrt{1+w})\left(\frac{1+\sqrt{1+w}}{2}\right)^{\delta_{D,0}+\delta_{D,6}}}\right)^n,
\end{align}
where $D=2\lfloor \frac{\D}{2}\rfloor$.

\begin{rmk}(\cite{Grosse:2016qmk})\label{rmk:beta}
 The $\beta$-function $\beta_\lambda$ of the running coupling constant 
 is $\lambda_{bare}(\Lambda^2)=\frac{\lambda}{\sqrt{Z}}$. It is 
 in dimension $6\leq \D<8$ and for real $\lambda$ positive
 \begin{align*}
  \beta_\lambda=&\Lambda^2\frac{d\lambda_{bare}(\Lambda^2)}{d\Lambda^2}\\
  =&\frac{2\lambda^2\Lambda^6}{(\sqrt{1+c}+\sqrt{(1+2e(\Lambda^2))^2+c})^2\sqrt{(1+2e(\Lambda^2))^2+c}}>0.
 \end{align*}
This calculation is easily checked by \eqref{eq:D4c} and assuming $c$ is independent of $\Lambda^2$ which is achieved by
choosing slightly different renormalisation conditions which converges in the limit $\Lambda^2\to\infty$ to the previous 
one. 
The $\beta$-function is finite for finite $\Lambda^2$ and therefore has no Landau pole. 
\end{rmk}
\noindent
Recall that the 1-point function is in any dimension given from the shifted 1-point function by 
\begin{align}\label{eq:G1}
 G^0(x)=\frac{\tilde{W}((1+2e(x))^2)-(1+2e(x))}{2\lambda}.
\end{align}
Let's look at the examples of the $D$-dimensional Moyal space with $D\in\{2,4,6\}$.
\begin{exm}($D=\D=2$ Moyal space)\label{Ex:D2}
\\
 The Moyal space admits linear eigenvalues $e(x)=x$ and for $D=2$ with multiplicity one, i.e. $r(x)=1$,
 such that $\varrho(Y)=\frac{2\lambda^2}{\sqrt{Y}}$. Then \sref{Corollary}{Coro:1P} gives for the shifted 
 1-point function after integration and simplification
 \begin{align*}
  \tilde{W}(X)=&\sqrt{X+c}+\frac{2\lambda^2}{\sqrt{X}}\log\bigg(\frac{(\sqrt{X+c}+\sqrt{X})(\sqrt{X}+1)}
  {\sqrt{X}\sqrt{1+c}+\sqrt{X+c}}\bigg)\\
  1=&\sqrt{1+c}+2\lambda^2\log\bigg(1+\frac{1}
  {\sqrt{1+c}}\bigg).
 \end{align*}
 The convergence radius in $\lambda$ of this solution is induced by the domain where $c(\lambda)$ 
 is uniquely invertible with
 $|\lambda|<\lambda_c\approx 0.4907$.
Expanding $c$ in $\lambda$ by \eqref{eq:cexpansion} gives $c(\lambda)=-\lambda^24\log 2-\lambda^44(\log 2-(\log 2)^2)
-\lambda^6 2(2\log 2-(\log2)^2)+\mathcal{O}(\lambda^8)$.
Inserting in \eqref{eq:G1} gives the first orders of the 1-point function 
\begin{align*}
 G^0(x)=\lambda \frac{\log(1+x)}{1+2x}+\lambda^3 \frac{(2\log2)^2(1+x)x}{(1+2x)^3}+\mathcal{O}(\lambda^5)
\end{align*}
which is confirmed in \sref{App.}{App:PertCubic} by Feynman graph calculations.
\end{exm}
\begin{exm}($D=\D=4$ Moyal space)\label{Ex:D4}
\\
 Again, linear eigenvalues $e(x)=x$ and for $D=4$ with growing multiplicity of the form $r(x)=x$
 such that $\varrho(Y)=\frac{\lambda^2(\sqrt{Y}-1)}{\sqrt{Y}}$. Then \sref{Corollary}{Coro:1P} gives for the shifted 
 1-point function after integration and simplification
 \begin{align*}
  \tilde{W}(X)=&1+\sqrt{X+c}-\sqrt{1+c}-\lambda^2\bigg\{
  \log\bigg(\frac{\sqrt{1+c}+\sqrt{X+c}}
  {2(1+\sqrt{1+c})}\bigg)\\
 &+\frac{1}{\sqrt{X}}\log\bigg(\frac{(\sqrt{X}+\sqrt{X+c})(\sqrt{X}+1)}
  {\sqrt{X}\sqrt{1+c}+\sqrt{X+c}}\bigg)\bigg\}\\
  1=&\sqrt{1+c}+\lambda^2\bigg\{1-\sqrt{1+c}\log\bigg(1+\frac{1}
  {\sqrt{1+c}}\bigg)\bigg\}.
 \end{align*}
 The function $c(\lambda)$ is uniquely invertible for $|\lambda|<\lambda_c\approx 1.1203$.
Expanding $c$ in $\lambda$ by \eqref{eq:cexpansion} gives $c(\lambda)=-\lambda^22(1-\log 2)
+\lambda^4(2-5\log 2+3(\log 2)^2)
-\lambda^6 (\frac{7}{4}-
7\log 2+\frac{37}{4}(\log2)^2-4 (\log2)^3)+\mathcal{O}(\lambda^8)$.
Inserting in \eqref{eq:G1} gives the first orders of the 1-point function 
\begin{align*}
 G^0(x)=\lambda \frac{x-(1+x)\log(1+x)}{1+2x}-\lambda^3 \frac{(1-\log2)^2(4x+3)x^2}{(1+2x)^3}+\mathcal{O}(\lambda^5)
\end{align*}
which is confirmed in \sref{App.}{App:PertCubic} by Feynman graph calculations.
\end{exm}
\begin{exm}($D=\D=6$ Moyal space)\label{Ex:D6}
\\
 Again, linear eigenvalues $e(x)=x$ and for $D=6$ with growing multiplicity of the form $r(x)=\frac{x^2}{2}$
 such that $\varrho(Y)=\frac{\lambda^2(\sqrt{Y}-1)^2}{4\sqrt{Y}}$. Then \sref{Corollary}{Coro:1P} gives for the shifted 
 1-point function after integration and simplification
 \begin{align*}
  \tilde{W}(X)=&\sqrt{X+c}\sqrt{1+c}-c+\frac{\lambda^2}{2}\bigg\{ \sqrt{1+c}-\sqrt{X+c}+
  \log\bigg(\frac{\sqrt{X+c}+\sqrt{1+c}}{2(1+\sqrt{1+c})}\bigg)\\
 &+\frac{(1+X)}{2\sqrt{X}}\log\bigg(\frac{(\sqrt{X}+\sqrt{X+c})(1+\sqrt{X})}
  {\sqrt{X}\sqrt{1+c}+\sqrt{X+c}}\bigg)\bigg\}\\
  -4c=&\lambda^2\bigg\{1-2\sqrt{1+c}+2(1+c)\log\bigg(1+\frac{1}
  {\sqrt{1+c}}\bigg)\bigg\}.
 \end{align*}
 The function $c(\lambda)$ is uniquely invertible for $|\lambda|<\lambda_c\approx 2.3647$.
Expanding $c$ in $\lambda$ by \eqref{eq:cexpansion} gives $c(\lambda)=-\lambda^2\frac{2\log2-1}{4}
+\lambda^4\frac{8(\log2)^2-10\log2+3}{32}
-\lambda^6 \frac{128(\log2)^3-252( \log2)^2+164\log 2-35}{1024}+\mathcal{O}(\lambda^8)$.
Inserting in \eqref{eq:G1} gives the first orders of the 1-point function 
\begin{align*}
 G^0(x)=\lambda \frac{2(1+x)^2\log(1+x)-x(2+3x)}{4(1+2x)}
 +\lambda^3 \frac{x^3(2+3x)(2\log2-1)^2}{16(1+2x)^3}+\mathcal{O}(\lambda^5)
\end{align*}
which is confirmed in \sref{App.}{App:PertCubic} by Feynman graph calculations.
\end{exm}

\begin{rmk}(\cite{Grosse:2016qmk})\label{rmk:renorm}
The cubic model on the $D=6$ Moyal space admits the renormalon problem which provides
no problem for the exact formula.
Determining the amplitude of the Feynman graph below according to the Feynman rules together with Zimmermann's forest formula
gives
\begin{align*}
 &\frac{(-\lambda)^{4+2n}}{(1+2x)^4}\int_0^\infty \frac{y^2\,dy}{2}\frac{1}{(1+x+y)^{4+n}}
 \bigg(\frac{2(1+y)^2\log(1+y)-y(2+3y)}{4(1+2y)}\bigg)^n\\
 \sim &\,  \frac{(-\lambda)^{4+2n}}{ 4^n2(1+2x)^4}\int_R^\infty \frac{dy}{y^2}\log(y)^n
 \sim \, \frac{(-\lambda)^{4+2n}}{ 4^n2(1+2x)^4}\cdot n!.
\end{align*}
\hspace*{20ex}
\def\svgwidth{0.5\textwidth}
\begingroup%
  \makeatletter%
  \providecommand\color[2][]{%
    \errmessage{(Inkscape) Color is used for the text in Inkscape, but the package 'color.sty' is not loaded}%
    \renewcommand\color[2][]{}%
  }%
  \providecommand\transparent[1]{%
    \errmessage{(Inkscape) Transparency is used (non-zero) for the text in Inkscape, but the package 'transparent.sty' is not loaded}%
    \renewcommand\transparent[1]{}%
  }%
  \providecommand\rotatebox[2]{#2}%
  \ifx\svgwidth\undefined%
    \setlength{\unitlength}{128.70155205bp}%
    \ifx\svgscale\undefined%
      \relax%
    \else%
      \setlength{\unitlength}{\unitlength * \real{\svgscale}}%
    \fi%
  \else%
    \setlength{\unitlength}{\svgwidth}%
  \fi%
  \global\let\svgwidth\undefined%
  \global\let\svgscale\undefined%
  \makeatother%
  \begin{picture}(1,0.52591254)%
    \put(0,0){\includegraphics[width=\unitlength,page=1]{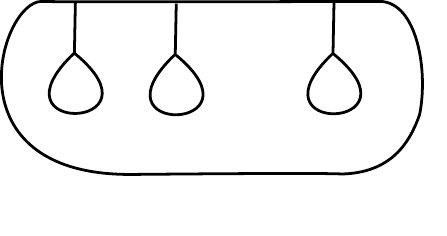}}%
    \put(0.51057869,0.42290544){\color[rgb]{0,0,0}\makebox(0,0)[lb]{\smash{...}}}%
    \put(0.14914983,0.30746659){\color[rgb]{0,0,0}\makebox(0,0)[lb]{\smash{$y_1$}}}%
    \put(0.38374869,0.31412124){\color[rgb]{0,0,0}\makebox(0,0)[lb]{\smash{$y_2$}}}%
    \put(0.72806818,0.30746113){\color[rgb]{0,0,0}\makebox(0,0)[lb]{\smash{$y_n$}}}%
    \put(0,0){\includegraphics[width=\unitlength,page=2]{Renormalon.pdf}}%
    \put(0.07001659,0.10322355){\color[rgb]{0,0,0}\makebox(0,0)[lb]{\smash{$x$}}}%
    \put(0.39285178,0.1831426){\color[rgb]{0,0,0}\makebox(0,0)[lb]{\smash{$y$}}}%
    \put(0.28600674,0.04802122){\color[rgb]{0,0,0}\makebox(0,0)[lb]{\smash{$x$}}}%
    \put(0.46138485,0.05024103){\color[rgb]{0,0,0}\makebox(0,0)[lb]{\smash{$x$}}}%
    \put(0.63010296,0.05468104){\color[rgb]{0,0,0}\makebox(0,0)[lb]{\smash{$x$}}}%
  \end{picture}%
\endgroup%
\\
\end{rmk}

\begin{rmk}
To show that the cubic matrix field theory model is in fact a QFT model, 
one has to check that the connected Schwinger functions are reflection positive.
For the connected 2-point Schwinger function
reflection positivity is by the work \cite{widder} and formula \eqref{eq:SchwingerRep} equivalent
to say that the diagonal 2-point function
$G^{(0)}(x,x)=2\tilde{W}((1+2x)^2)$
is a Stieltjes function. It was proved in \cite{Grosse:2016qmk} that the 2-point function on the
$D=2$ Moyal space is \textit{not} a Stieltjes 
function and therefore not reflection positive. However, the function $G^{(0)}(x,x)$ is  for $D=4$ and $D=6$ with $\lambda\in \R$
a Stieltjes function  \cite{Grosse:2016qmk}. The 2-point function is special since 
reflection positivity is equivalent to the Stieltjes property which does not hold in general.
Checking whether higher Schwinger functions are reflection positive is much harder and work in progress.
\end{rmk}

\subsection[Analytic Continuation]{Analytic Continuation}
The crucial step in deriving all solutions
is to continue analytically the correlation functions 
with the ''right'' variable transformation. 
The result of \sref{Proposition}{Prop:1Punkt} provides a natural choice by
\begin{align*}
 z(x)=:\sqrt{X+c}=\sqrt{(1+2e(x))^2+c}.
\end{align*}
Defining next the $(1+..+1)$-point function of 
genus $g$ with $b$ boundary components from definition \eqref{eq:Gcont} by
\begin{align*}
 \G_g\big(z_1(x^1),..,z_b(x^b)\big):&=G^g(x^1|..|x^b)\qquad \text{for}\qquad (g,b)\neq (0,1)
 \end{align*}
 and the planar 1-point function, which is already known by \sref{Proposition}{Prop:1Punkt}, by
 \begin{align}\nonumber
 \G_0(z(x)):&=\frac{ \tilde{W}(X(x))}{2\lambda}=\frac{ W^0(x)}{2\lambda}=G^0(x)
 +\frac{1+2e(x)}{2\lambda}\\\label{eq:Cubic1PComplex}
 &=\frac{1}{2\lambda}\bigg(\frac{z}{\sqrt{Z}}-\lambda\nu +\frac{1}{2}
  \int_{\sqrt{1+c}}^{\sqrt{(1+2e(\Lambda^2))^2+c}}\!\!\!\!\!\!\!\!dt\frac{\tilde{\varrho}(t)}{(z+t)t}\bigg),\\\nonumber
  &\text{with}\qquad 
  \tilde{\varrho}(t):=2t\,\varrho\!\left(t^2-c\right)=\frac{4\lambda^2 t 
  \cdot r\!\left(e^{-1}(\frac{\sqrt{t^2-c}-1}{2})\right)}{\sqrt{t^2-c} \cdot
e'\!\!\left(e^{-1}(\frac{\sqrt{t^2-c}-1}{2})\right)}.
\end{align}
With this definition, the 1-point function extends uniquely to a sectional holomorphic function
with branch cut along $]-\sqrt{(1+2e(\Lambda^2))^2+c},- \sqrt{1+c}]$.
Since we have $-1<c\leq 0$ for $|\lambda|<\lambda_c$ for all $\D<8$, the 1-point function is holomorphic in a small region around $z=0$.
Analogously, all SDEs are recursively complexified. The 
analyticity domains of the corresponding correlation functions will be discussed later.

For a more convenient reading, we define:
\begin{dfnt}\label{defint}
Let $\hat{K}_z$ be the integral operator,
\begin{align*}
\hat{K}_z f(z):=&2\lambda \G_0(z)f(z)
+\lambda\nu f(z)
+\frac{1}{2}\int_{\sqrt{1+c}}^{\sqrt{(1+2e(\Lambda^2))^2+c}}\!\!\!\!\!\!\!\!\!\!\!\! dt\,
\tilde{\varrho}(t)\frac{f(z)-f(t)}{z^2-t^2}\\
=&\lambda z f(z)\bigg(\G_0(z)-\G_0(-z) \bigg)-
\frac{1}{2}\int_{\sqrt{1+c}}^{\sqrt{(1+2e(\Lambda^2))^2+c}}\!\!\!\!\!\!\!\!\!\!\!\!\!\!\!\! dt\,
\tilde{\varrho}(t)\frac{f(t)}{z^2-t^2}.
\end{align*}
\end{dfnt}
\noindent Expressing \eqref{eq:CubicCont1P}
and  \eqref{eq:CubicContNP} for all $(g,b)\neq (0,1)$ 
with the variable $z$ through the integral operator $\hat{K}_z$ leads to
\begin{align}\label{eq:CubicSchwingerComplex}
 &\hat{K}_{z_1}\G_g\big(z_1,J\big)
 +\lambda
 \G_{g-1}(z_1,z_1,J)+\lambda \sum_{\substack{I\uplus I'=J\\
  h+h'=g}}^\prime \G_{h}(z_1,I)\G_{h'}(z_1,I')\\\nonumber
  &+(2\lambda)^3\sum_{\zeta\in J}
 \frac{\partial}{\zeta \partial \zeta }\frac{\G_g(z_1,J\backslash\{\zeta\})-\G_g(\zeta,
 {J\backslash\{\zeta\}})}{z_1^2-\zeta^2} =0,
\end{align}
where $J=\{z_2,..,z_b\}$ and $\G_g(I):=\G_g(\zeta_1,..,\zeta_{p})$ for the set $I=\{\zeta_1,..,\zeta_{p}\}$.
The sum $\sum^\prime$ excludes $(h,I)=(0,\emptyset)$ and $(h,I)=(g,J)$.
The second line is achieved from the $(3+1+..+1)$-point function through the 
recursive equation of \sref{Proposition}{Prop:CubisRecur} which produces a derivative.

We mention that the dependence on the spectral dimension contributes in $\hat{K}_{z}$ by the renormalisation 
constant $Z$ for $\D\geq6$ and in the constant $c$. The entire structure of the SDEs 
is independent of the spectral dimension.

For later purpose, we introduce the renormalised  \textit{Kontsevich times} which play a distinguished r\^ole:
\begin{dfnt}\label{KontTime}
Let the Kontsevich times $\{\varrho_k\}_{k\in\N}$ be defined by
\begin{align*}
\varrho_k:=\frac{\delta_{k,0}}{\sqrt{Z}}-
\frac{1}{2}\int_{\sqrt{1+c}}^{\sqrt{(1+2e(\Lambda^2))^2+c}} \frac{dt\,
\tilde{\varrho}(t)}{t^{3+2k}}
\end{align*}
with converging limit $\Lambda^2\to\infty$.
\end{dfnt}\noindent
The Kontsevich times are related to the integral operator by:
\begin{lemma}\label{lemmaopK}
The operator $\hat{K}_z$ defined in \sref{Definition}{defint} satisfies
\begin{align*}
\hat{K}_z\Big(\frac{1}{z}\Big) &=\frac{1}{\sqrt{Z}} , & 
\hat{K}_z\Big(\frac{1}{z^{3+2n}}\Big) &= 
\sum_{k=0}^{n}\frac{\varrho_k}{z^{2n+2-2k}}.
\end{align*}
\begin{proof}
 Let $a=\sqrt{1+c}$ and $b=\sqrt{(1+2e(\Lambda^2))^2+c}$. The definition of the linear operator implies
 \begin{align*}
  \hat{K}_z f(z)=&\frac{zf(z)}{\sqrt{Z}}+
\frac{1}{2}\int_{a}^{b} dt\,
\tilde{\varrho}(t)\frac{zf(z)-tf(t)}{t(z^2-t^2)},
 \end{align*}
where the case $f(z)=\frac{1}{z}$ is directly clear. The identity
\begin{align*}
 \frac{\frac{1}{z^k}-\frac{1}{t^k}}{(z-t)}=-\sum_{l=0}^{k-1}\frac{z^lt^{k-1-l}}{z^kt^k}
\end{align*}
gives for $f(z)=\frac{1}{z^{3+2n}}$
\begin{align*}
 \hat{K}_z \!\left(\frac{1}{z^{3+2n}}\right)=&\frac{1}{z^{2+2n}\sqrt{Z}}+
\frac{1}{2}\int_{a}^{b} dt\,
\tilde{\varrho}(t)\frac{\frac{1}{z^{2+2n}}-\frac{1}{t^{2+2n}}}{t(z+t)(z-t)}\\
=&\frac{1}{z^{2+2n}\sqrt{Z}}-\frac{1}{2}
\int_{a}^{b} dt\,
\tilde{\varrho}(t)\sum_{l=1}^{2+2n}\frac{z^l t^{2+2n-l}}{z^{3+2n}t^{3+2n}(z+t)}\\
=&\frac{1}{z^{2+2n}\sqrt{Z}}-\sum_{l=0}^{n}\frac{1}{z^{2+2n-2l}}\frac{1}{2}
\int_{a}^{b} \frac{dt\,
\tilde{\varrho}(t)}{t^{3+2l}},
\end{align*}
which is the definition of the Kontsevich times.
\end{proof}
\end{lemma}

\subsection{Solution for $\chi=2-2g-b\geq-1$}
As mentioned before, the 1-point function plays a special r\^ole
since it obeys a nonlinear equation. Also the $(1+1)$-point
function is different from the others. The correlation functions
are embedded into Riemann surfaces which are topologically clearly distinct
into the cases $\chi=1$, $\chi=0$ and $\chi<0$. The automorphism group for $\chi\geq0$ (unstable) Riemann
surfaces is infinite, whereas stable surfaces ($\chi<0$) have a finite group of automorphisms.
These distinctions are inherited to the pole structure of the correlation functions. 

Recall from \eqref{eq:CubicSchwingerComplex} that the $(1+1)$-point function obeys the integral equation
\begin{align}\label{eq:1+1complex}
 \hat{K}_{z_1}\G_0\big(z_1,z_2\big)
   =-(2\lambda)^3
 \frac{\partial}{z_2 \partial z_2 }\frac{\G_0(z_1)-\G_0(z_2)}{z_1^2-z_2^2}
\end{align}

\begin{prps}\label{Prop:Cubic1+1}
For $\chi=0$,
 the $(1+1)$-point function $(b=2)$ of genus $g=0$, which solves \eqref{eq:1+1complex}, is given by
 \begin{align*}
  \G_0(z_1,z_2)=\frac{(2\lambda)^2}{z_1z_2(z_1+z_2)^2}.
 \end{align*}
\begin{proof}
 Let $a=\sqrt{1+c}$ and $b=\sqrt{(1+2e(\Lambda^2))^2+c}$.
 The lhs of \eqref{eq:1+1complex} gives
 \begin{align*}
  \hat{K}_{z_1}\G_0\big(z_1,z_2\big)=&\frac{4\lambda^2}{z_2(z_1+z_2)^2\sqrt{Z}}+
\frac{4\lambda^2}{2}\int_{a}^{b} dt\,
\tilde{\varrho}(t)\frac{\frac{1}{z_2(z_1+z_2)^2}-\frac{1}{z_2(t+z_2)^2}}{t(z_1+t)(z_1-t)}\\
=&\frac{4\lambda^2}{z_2(z_1+z_2)^2\sqrt{Z}}-
\frac{4\lambda^2}{2z_2}\int_{a}^{b} dt\,
\tilde{\varrho}(t)\frac{(z_1+z_2)+(t+z_2)}{t(z_1+t)(z_1+z_2)^2(t+z_2)^2}\\
=&-\frac{4\lambda^2}{z_2}\frac{\partial}{\partial z_2}\bigg(\frac{1}{(z_1+z_2)\sqrt{Z}}
-\frac{1}{2}\int_{a}^{b} \frac{dt\,
\tilde{\varrho}(t)}{t(z_1+t)(z_1+z_2)(t+z_2)}\bigg)
 \end{align*}
which coincides with the rhs.
\end{proof}
\end{prps}

\begin{prps}\label{Prop:Cubic1+1+1}
For $\chi=2-2g-b=-1$, the 
 $(1+1+1)$-point function $(b=3)$ of genus $g=0$ and the 1-point function $(b=1)$ of genus $g=1$ are given by
 \begin{align*}
  \G_0(z_1,z_2,z_3)=-\frac{32\lambda^5}{\varrho_0\, z_1^3z_2^3z_3^3},\qquad 
  \G_1(z)=\frac{\lambda^3\varrho_1}{\varrho_0^2z^3}-\frac{\lambda^3}{\varrho_0\, z^5}.
 \end{align*}
\begin{proof}
 The $(1+1+1)$-point function obeys \eqref{eq:CubicSchwingerComplex}
 \begin{align*}
  &\hat{K}_{z_1}\G_0\big(z_1,z_2,z_3\big)
 +\lambda 2 \G_{0}(z_1,z_2)\G_{0}(z_1,z_3)\\
  &=-(2\lambda)^3\bigg(
 \frac{\partial}{z_2 \partial z_2 }\frac{\G_0(z_1,z_3)-\G_0(z_2,
 z_3)}{z_1^2-z_2^2} +\frac{\partial}{z_3 \partial z_3 }\frac{\G_0(z_1,z_2)-\G_0(z_3,
 z_2)}{z_1^2-z_3^2}\bigg).
 \end{align*}
The ansatz $\G_0\big(z_1,z_2,z_3\big)=\frac{\gamma}{z_1^3z_2^3z_3^3}$ leads with \sref{Lemma}{lemmaopK}
to
\begin{align*}
 \hat{K}_{z_1}\G_0\big(z_1,z_2,z_3\big)=\frac{\gamma \varrho_0}{z_1^2z_2^3z_3^3}.
\end{align*}
Inserting \sref{Proposition}{Prop:Cubic1+1} gives, after performing the derivatives and simplifying,
\begin{align*}
 &(2\lambda)^3\bigg(
 \frac{\partial}{z_2 \partial z_2 }\frac{\G_0(z_1,z_3)-\G_0(z_2,
 z_3)}{z_1^2-z_2^2} +\frac{\partial}{z_3 \partial z_3 }\frac{\G_0(z_1,z_2)-\G_0(z_3,
 z_2)}{z_1^2-z_3^2}\bigg)\\
 &+2\lambda  \G_{0}(z_1,z_2)\G_{0}(z_1,z_3)\\
 =&\frac{32\lambda^5}{z_1^2z_2^3z_3^3}
\end{align*}
which gives by comparing the coefficient $\gamma=-\frac{32\lambda^5}{\varrho_0}$.
\\
The $1$-point function of genus $g=1$ obeys \eqref{eq:CubicSchwingerComplex}
 \begin{align*}
  &\hat{K}_{z_1}\G_1\big(z\big)
 + \lambda \G_{0}(z,z)=0.
 \end{align*}
 The ansatz $\G_1\big(z\big)=\frac{\alpha}{z^3}+\frac{\beta}{z^5}$ gives with \sref{Lemma}{lemmaopK}
 $\hat{K}_{z_1}\G_1\big(z\big)=\frac{\alpha \varrho_0}{z^2}+\frac{\beta \varrho_0}{z^4}+\frac{\beta \varrho_1}{z^2}$.
 Comparing the coefficients with $\G_{0}(z,z)=\frac{\lambda^2}{z^4}$ yields $\beta=-\frac{\lambda^3}{\varrho_0}$
 and $\alpha=-\frac{\beta\varrho_1}{\varrho_0}=\frac{\lambda^3\varrho_1}{\varrho_0^2}$.
\end{proof}
\end{prps}\noindent
Notice that $\G_0(z_1,z_2)$ is a meromorphic function with poles at $z_1=0=z_2$ and at the diagonal $z_1=-z_2$, whereas
the 1-point function has even a branch cut
at $z\in [-\sqrt{(1+2e(\Lambda^2))^2+c},-\sqrt{1+c}]$. For $\chi=-1$ ,
the solutions of \sref{Proposition}{Prop:Cubic1+1+1} are 
meromorphic functions with a pole only at $z_i=0$. By the recursive hypothesis 
of \eqref{eq:CubicSchwingerComplex}, all further correlation functions with 
$\chi<-1$ may have poles at $z_i=0$, $z_i\pm z_j$ and possibly a 
branch cut at $z_i\in [-\sqrt{(1+2e(\Lambda^2))^2+c},-\sqrt{1+c}]$.
However, we will prove that any correlation function $\G_g(z_1,..,z_b)$ with $\chi<0$
is a meromorphic function with poles
of odd order
at $z_i=0$ for $i\in\{1,..,b\}$.

\subsection[Solution for $b>1$ via Boundary Creation Operator]
{Solution for $b>1$ via Boundary Creation Operator\footnote{Parts of this subsection are taken from our
paper \cite{Grosse:2019nes}}}
The goal is to construct an operator that increases the number of boundaries. This operator 
is already known for spectral dimensions $\D<2$ and used in \cite{Makeenko:1991ec,Witten:1991mn}.
Assume for the following consideration that $E_k$ are distinct eigenvalues of multiplicity one. Then,
the boundary
creation operator is defined for $\D=0$ by $T_n:=\frac{\lambda }{2E_n}\frac{\partial}{\partial E_n}$. 
The definition makes sense by the following formal considerations, where the renormalisation
constants are taken trivially $Z-1=\kappa=\nu=\zeta=0$, and the matrix size $\mN$ as well as $V$ are finite,
\begin{align*}
&T_n\log \bigg( \int D\Phi \,e^{-V\Tr\big(E\Phi^2+\frac{\lambda}{3}\Phi^3 \big)}\bigg)\\
=&-\frac{V\lambda}{2\Z[0]}\int D\Phi \,\frac{\sum_n\Phi_{nm}\Phi_{mn}}{E_n}
e^{-V\Tr\big(E\Phi^2+\frac{\lambda}{3}\Phi^3 \big)}\\
=&\frac{1}{2\Z[0]}\int D\Phi \,\frac{1}{E_n}\bigg(\frac{\partial}{\partial \Phi_{nn}}+2E_n \Phi_{nn}\bigg)
e^{-V\Tr\big(E\Phi^2+\frac{\lambda}{3}\Phi^3 \big)}\\
=&\frac{1}{\Z[0]}\int D\Phi \,\Phi_{nn}
e^{-V\Tr\big(E\Phi^2+\frac{\lambda}{3}\Phi^3 \big)}=G_n.
\end{align*}
We go back to the Dirac measure $r(x)=\frac{1}{V}\sum_{n=0}^{\mN'}r_n\delta(x-e_n)$, where $r_k$ are the multiplicities
of the distinct eigenvalues $e_k$ of $E$. Take also the differentiable
function $e(x)$ discussed in \sref{Sec.}{Sec:LargeLimit} into account, then the boundary creation operator is expressed formally
by the functional derivative
\begin{align}\label{eq:Tx}
 T_n\to T(x):= \frac{2\lambda}{1+2e(x)}\frac{\delta}{\delta (1+
 2e(x))}\frac{\delta}{\delta (r(x)dx)},
\end{align}
where the derivative with respect to the measure is formally defined by $\frac{\delta}{\delta r(x) dx}
\int dt\, r(t)\,f(t)=f(x)$.

The choice of the complex variable $z=\sqrt{(1+2e(x))^2+c}$ and the property of \sref{Lemma}{lemmaopK} implies a 
dependence of the solutions on the parameters $\varrho_l$ and $c$, where
\begin{align}\label{eq:cimplicit}
 c=-8\lambda^2\int \frac{dt\,r(t)}{\sqrt{(1+2e(t))^2+c}}, \qquad \text{for}\quad \D<2. 
\end{align}
The formally defined boundary creation operator $T(x)$ gives on $c$ with the Leibniz rule
\begin{align}\nonumber
 T(x)c=&
 \frac{8\lambda^3}{\sqrt{(1+2e(x))^2+c}^3}+
 4\lambda^2\int \frac{dt\,r(t)}{\sqrt{(1+2e(t))^2+c}^3}T(x)c\\
 \Leftrightarrow \quad T(x)c=&
 \frac{8\lambda^3}{\varrho_0 \sqrt{(1+2e(x))^2+c}^3}\label{eq:cderive}
\end{align}
with \sref{Definition}{KontTime} for $\varrho_0$,
and on $\varrho_l$
\begin{align*}
 T(x)\varrho_l=&\frac{8\lambda^3(3+2l)}{\sqrt{(1+2e(x))^2+c}^{5+2l}}+ (3+2l)
 \int \frac{4\lambda^2\,dt\,r(t)}{\sqrt{(1+2e(t))^2+c}^{5+2l}}T(x)c\\
 =&\frac{8\lambda^3(3+2l)}{\sqrt{(1+2e(x))^2+c}^{5+2l}}-\frac{8\lambda^3(3+2l)\varrho_{l+1}}{\varrho_0
 \sqrt{(1+2e(x))^2+c}^3},
\end{align*}
where \eqref{eq:cderive} was used.

To avoid the formally defined functional derivative, the chain rule is applied to achieve
partial derivatives with respect to $c$ and $\varrho_l$. The next step is to switch to the complex variable
$z=\sqrt{(1+2e(x))^2+c}$. We are then able to show that the boundary creation admits a rigorous and universal replacement
for any $\D<8$ by:
\begin{dfnt}\label{DefOp}
Let be $J=\{z_1,\dots,z_p\}$ with $|J|:=p$. 
Then, the \textit{boundary creation} is
\begin{align}
\hat{\mathrm{A}}^{\dag g}_{J,z}
:=
\sum_{l= 0}^{3g-3+|J|} \Big(-\frac{(3+2l) \varrho_{l+1}}{
\varrho_0 z^3}+\frac{3+2l}{z^{5+2l}}\Big)
\frac{\partial}{\partial \varrho_l}
+ \sum_{\zeta\in J}\frac{1}{\varrho_0 z^3 \zeta}
\frac{\partial}{\partial \zeta}.
	\end{align}
\end{dfnt}
\noindent We emphasise that the last variable $z$ in $\hat{\mathrm{A}}^{\dag
  g}_{J,z}$ plays a very different r\^ole compared to all the $z_i\in J$.

\begin{lemma}
The differential operators $\hat{\mathrm{A}}^{\dag g}_{J,z}$ are commutative,
\begin{align*}
\hat{\mathrm{A}}^{\dag g}_{J,z_p,z_q}
\hat{\mathrm{A}}^{\dag g}_{J,z_p}
=\hat{\mathrm{A}}^{\dag g}_{J,z_q,z_p}
\hat{\mathrm{A}}^{\dag g}_{J,z_q}.
\end{align*}
\end{lemma}
\begin{proof} Being a derivative, it is enough to verify 
$\hat{\mathrm{A}}^{\dag g}_{J,z_p,z_q}
\hat{\mathrm{A}}^{\dag g}_{J,z_p}(\varrho_k)
=\hat{\mathrm{A}}^{\dag g}_{J,z_q,z_p}
\hat{\mathrm{A}}^{\dag g}_{J,z_q}(\varrho_k)$ for any $k$ and 
$\hat{\mathrm{A}}^{\dag g}_{J,z_p,z_q}
\hat{\mathrm{A}}^{\dag g}_{J,z_p}(z_i)
=\hat{\mathrm{A}}^{\dag g}_{J,z_q,z_p}
\hat{\mathrm{A}}^{\dag g}_{J,z_q}(z_i)$ for any $z_i\in J$. 
This is guaranteed by 
\begin{align*}
\hat{\mathrm{A}}^{\dag g}_{J,z_p,z_q}
\hat{\mathrm{A}}^{\dag g}_{J,z_p}(\varrho_k)
&= 
\frac{(3+2k)(5+2k) \varrho_{k+2}}{
\varrho_0^2 z_q^3z_p^3}
-\frac{(3+2k)(5+2k)}{
\varrho_0 z_q^{7+2k} z_p^3}
- \frac{3(3+2k)\varrho_{k+1}\varrho_1}{\varrho_0^2 z_q^3z_p^3}
\\
&
+ \frac{3(3+2k)\varrho_{k+1}}{\varrho_0^2 z_p^5z_p^3}
+ \frac{3(3+2k) \varrho_{k+1}}{\varrho_0^2 z_q^3z_p^5}
-\frac{(3+2k)(5+2k)}{\varrho_0 z_q^3z_p^{7+2k}},
\\
\hat{\mathrm{A}}^{\dag g}_{J,z_p,z_q}
\hat{\mathrm{A}}^{\dag g}_{J,z_p}(z_i)
&= \frac{3\varrho_1}{\varrho_0^2 z_q^3 z_p^3 z_i}
-\frac{3}{\varrho_0^2 z_q^5 z_p^3 z_i}
-\frac{3}{\varrho_0^2 z_q^3z_p^5 z_i}
- \frac{1}{\varrho_0^2 z_q^3 z_p^3 z_i^3}.
\end{align*}
\end{proof}
\noindent
This shows that boundary components labelled by $z_i$ behave 
like bosonic particles at position $z_i$. The creation operator 
$(2\lambda)^3 \hat{\mathrm{A}}^{\dag g}_{J,z}$ adds to a $|J|$-particle state 
another particle at position $z$. The $|J|$-particle state is 
precisely given by $\G_g(J)$: 
\begin{thrm}\label{finaltheorem}
Assume that 
$\G_g(z)$ is, for $g\geq 1$, an odd function of $z\neq 0$ and a
rational function of 
$\varrho_0,\dots,\varrho_{3g-2}$ (true for $g=1$).
Then the  $(1+1+...+1)$-point function of genus $g\geq 1$ and $b$ 
boundary components of the renormalised $\Phi^3_\D$ matrix field theory
model for spectral dimension $\D<8$ has the solution
\begin{align}
\G_g(z_1,...,z_b)=(2\lambda)^{3b-3}\hat{\mathrm{A}}^{\dag g}_{z_1,\dots,z_b}
\big(\hat{\mathrm{A}}^{\dag g}_{z_1,\dots,z_{b-1}}
\big(\cdots \hat{\mathrm{A}}^{\dag g}_{z_1,z_2}
\G_g(z_1)...\big)\big),\qquad
z_i\neq 0,
\label{eqfinalthm}
\end{align}
where $\G_g(z_1)$ is the 1-point function of genus $g\geq 1$ and 
the boundary creation operator $\hat{\mathrm{A}}^{\dag g}_{J}$ is defined in \sref{Definition}{DefOp}.
For $g=0$ the boundary creation operators act on the $(1+1)$-point function 
\begin{align*}
\G_0(z_1,...,z_b)=(2\lambda)^{3b-6}  \hat{\mathrm{A}}^{\dag 0}_{z_1,\dots,z_B}
\big( \hat{\mathrm{A}}^{\dag 0}_{z_1,\dots,z_{B-1}}
\big(\cdots  \hat{\mathrm{A}}^{\dag 0}_{z_1,z_2,z_3} 
\G_0(z_1|z_2)...\big)\big).
\end{align*} 
\end{thrm}
\begin{proof}
We rely on several Lemmata proved in \sref{App.}{appendixC}.
Regarding (\ref{eqfinalthm}) as a \emph{definition}, we prove in
\sref{Lemma}{lemma9} an equivalent formula for the 
linear integral equation \eqref{eq:CubicSchwingerComplex}. This expression is satisfied because 
\sref{Lemma}{lemma6} and \sref{Lemma}{lemma8} add up to $0$. 
Consequently, the family of functions (\ref{eqfinalthm}) 
satisfies (\ref{eq:CubicSchwingerComplex}). This solution is unique because 
of uniqueness of the perturbative expansion.
\end{proof}

\begin{cor}\label{coro2}
Let $J=\{z_2,...,z_b\}$. Assume that $z\mapsto \G_g(z)$ is holomorphic 
in $\mathbb{C}\setminus \{0\}$ with 
$\G_g(z)=-\G_g(-z)$ for all 
$z\in \mathbb{C}\setminus \{0\}$ and $g\geq 1$. Then
all $\G_g(z_1,J)$ with $2-2g-b<0$ 
\begin{enumerate}\itemsep 0pt
\item are holomorphic in every $z_i \in \mathbb{C}\setminus \{0\}$

\item are odd functions in every $z_i$, i.e.\
$ \G_g(-z_1,J)=-\G_g(z_1,J)$ for all
$z_1,z_i\in\mathbb{C}\setminus\{0\}$.
\end{enumerate}
\end{cor}
\begin{proof}
The boundary creation operator $\hat{\mathrm{A}}^{\dag g}_{J,z}$ of 
\sref{Definition}{DefOp} preserves holomorphicity in 
$\mathbb{C}\setminus\{0\}$ and 
maps odd functions into odd functions. Thus, only the initial conditions 
need to be checked. They are fulfilled
for $\G_0(z_1,z_2,z_3)$ and $\G_1(z_1)$ according to \sref{Proposition}{Prop:Cubic1+1+1}; for 
$g\geq 2$ by assumption.
\end{proof}
\noindent
The assumption will be verified later in \sref{Proposition}{structure-Gz}.

\begin{cor}\label{coro1}
The boundary creation operator $\hat{\mathrm{A}}^{\dag g}_{J,z}$ acting on a 
$(N_1+...+N_b)$-point function of genus $g$ and complex variables $J=\{z_1^1,..,z_{N_b}^b\}$ gives the 
$(1+N_1+...+N_b)$-point function of genus $g$.
\end{cor}
\begin{proof}
This follows from the change to complex variables in the
equation of \sref{Proposition}{Prop:CubisExpl} and $\hat{\mathrm{A}}^{\dag g}_{z_J,z}
\bigg(\frac{1}{(z^{j_1}_{i_1})^2-(z^{j_2}_{i_2})^2}\bigg)=0$ for $j_1\neq j_2$ and $i_1\neq i_2$.
\end{proof}

\begin{rmk}
 The motivation of defining the boundary creation operator is due to the $\D<2$ case by the implicit 
 equation \eqref{eq:cimplicit} of $c$. For $\D\geq 2$, we have formally a different result for $T(x)c$
 since a different implicit equation holds (see \sref{Corollary}{Coro:1P}) due to the renormalisation constants.
 Nevertheless, the \sref{Definition}{DefOp} is universal for $\D< 8$.
\end{rmk}

\subsection[Solution for $b=1$ and $g>0$]{Solution for $b=1$ and $g>0$\footnote{Parts of this subsection are taken from our
paper \cite{Grosse:2019nes}}}
It remains to check that the 1-point function $\G_g(z)$ at genus 
$g\geq 1$ satisfies the assumptions of \sref{Theorem}{finaltheorem} and 
\sref{Corollary}{coro2}, namely: 
\begin{enumerate}\itemsep 0pt
\item  $\G_g(z)$ depends only on
the moments $\varrho_0,\dots,\varrho_{3g-2}$ of the measure, 

\item $z\mapsto \G_g(z)$ is holomorphic on $\mathbb{C}\setminus \{0\}$ and 
an odd function of $z$.
\end{enumerate}
We establish these properties by solving \eqref{eq:CubicSchwingerComplex} for $b=1$
via a formula for the inverse of $\hat{K}_z$. This formula is inspired by 
topological recursion.
\begin{dfnt}\label{def:Bell}
The {\sffamily Bell polynomials} are defined by
\begin{align*}
B_{n,k}(x_1,...,x_{n-k+1})=\sum \frac{n!}{j_1!j_2!...j_{n-k+1}!}\left(\frac{x_1}{1!}\right)^{j_1}\left(\frac{x_2}{2!}\right)^{j_2}...\left(\frac{x_{n-k+1}}{(n-k+1)!}\right)^{j_{n-k+1}}
	\end{align*}
	for $n\geq 1$, where the sum is over non-negative integers $j_1,...,j_{n-k+1}$ with $j_1+...+j_{n-k+1}=k$ and $1j_1+...+(n-k+1)j_{n-k+1}=n$. 
Moreover, one defines $B_{0,0}=1$ and $B_{n,0}=B_{0,k}=0$ for $n,k>0$.
\end{dfnt}
\noindent
An important application is Fa\`a di 
Bruno's formula, the $n$-th order chain rule:
\begin{align}\label{Faadi}
\frac{d^n}{dx^n}f(g(x))=\sum_{k=1}^{n}f^{(k)}(g(x)) \,
B_{n,k}(g'(x),g''(x),...,g^{(n-k+1)}(x)).
\end{align}
\begin{prps}\label{Prop5}
Let $f(z)=\sum_{k=0}^\infty \frac{a_{2k}}{z^{2k}}$ be an even Laurent 
series about $z=0$ bounded at $\infty$.
Then the inverse of the integral operator $\hat{K}_z$ of 
\sref{Definition}{defint} is given by the residue formula
\begin{align*}
&\Big[z^2\hat{K}_z\frac{1}{z}\Big]^{-1}f(z) 
=- \Res\displaylimits_{z'\to 0}\left[ K(z,z')\, f(z')dz'\right],\\
&\text{where}\qquad 
K(z,z'):=\frac{1}{\lambda(\G_0(z')-\G_0(-z'))(z'^2-z^2)}.
\end{align*}
\end{prps}
\begin{proof}
The formulae (\ref{eq:Cubic1PComplex}) give rise to the series expansion 
\begin{align}
\label{Gzminusz}
\lambda(\G_0(z')-\G_0(-z'))= 
\sum_{l=0}^\infty \varrho_l \cdot (z')^{2l+1},
\end{align}
where the Kontsevich times $\varrho_l$ are given by \sref{Definition}{KontTime}.
The series of its reciprocal is found using \eqref{Faadi}:
\begin{align}
&\frac{1}{\lambda(\G_0(z')-\G_0(-z'))} = \frac{1}{z' \varrho_0}
\sum_{m=0}^\infty \frac{(z')^{2m}}{m!} S_m,
\label{S-Bell}
\\
&S_m:= \frac{d^m}{d \tau^m}\Big|_{\tau=0} \Big(\sum_{l=0}^\infty 
\frac{\varrho_l}{\varrho_0} \tau^l\Big)^{-1} 
= \sum_{i=0}^m\frac{(-1)^i i!}{\varrho_0^{i}} 
B_{m,i}(1!\varrho_1,2! \varrho_2,...,(m-i+1)!\varrho_{m-i+1}).
\nonumber
\end{align}
Multiplication by the geometric series gives 
\begin{align}
K(z,z')=-\frac{1}{z^2 z' \varrho_0} 
\sum_{n,m=0}^\infty \frac{(z')^{2m+2n}}{m! z^{2n}} S_m.
\label{Kzz}
\end{align}
The residue of a monomial in 
$f(z')=\sum_{k=0}^\infty \frac{a_{2k}}{(z')^{2k}}$ is then
\begin{align}\label{Res}
\Res_{z'\to 0}\left[ K(z,z')\frac{dz'}{(z')^{2k}}\right]
=-\frac{1}{\varrho_0}\sum_{j=0}^{k} \frac{S_j}{j!z^{2k-2j+2}}.
\end{align}
In the next step we apply the operator $z^2\hat{K}\frac{1}{z}$ 
to \eqref{Res}, where \sref{Lemma}{lemmaopK} is used:
\begin{align}
&z^2 \hat{K}_z\Big(\frac{1}{z}
\frac{(-1)}{\varrho_0} 
\sum_{j=0}^{k} \frac{S_j}{j!z^{2k-2j+2}}\Big)
=-\frac{z^2}{\varrho_0} 
\sum_{j=0}^{k}\sum_{i=0}^{k-j}\frac{S_j\varrho_i}{j! z^{2k-2j-2i+2}}
\nonumber
\\
&=-\sum_{j=0}^{k}\frac{S_{k-j}}{(k-j)! z^{2j}} 
-\frac{1}{\varrho_0} 
\sum_{i=0}^{k-1} \sum_{j=i+1}^{k}
\frac{S_{k-j}\varrho_{j-i}}{(k-j)! z^{2i}}.
\label{Bell2}
\end{align}
The last sum over $j$ is treated as follows, where the 
Bell polynomials are inserted for $S_m$:
\begin{align*}
&\sum_{j=i+1}^{k}\frac{S_{k-j} \varrho_{j-i}}{(k-j)!}
=\sum_{j=1}^{k-i}\frac{S_{k-j-i}\varrho_{j}}{(k-j-i)! }
\\
&=\sum_{j=1}^{k-i}\sum_{l=0}^{k-j-i}\frac{(-1)^{l}l!}{(k-j-i)!
\varrho_0^{l}}\,\varrho_jB_{k-j-i,l}(1!\varrho_1,...,(k-j-i-l+1)!
\varrho_{k-j-i-l+1})
\\
&=\sum_{l=0}^{k-j-i}\frac{(-1)^{l}l!}{\varrho_0^{l}(k-i)!}
\sum_{j=1}^{k-i-l}\binom{k-i}{j}j!\,\varrho_j
B_{k-j-i,l}(1!\varrho_1,...,(k-j-i-l+1)!\varrho_{k-j-i-l+1})
\\
&=\sum_{l=0}^{k-i}\frac{(-1)^{l}(l+1)!}{\varrho_0^{l}(k-i)!} 
B_{k-i,l+1}(1!\varrho_1,...,(k-i-l)!\varrho_{k-i-l})
\\
&=-\varrho_0\frac{S_{k-i}}{(k-i)!}.
\end{align*}
We have used $B_{n,0}=0$ and $B_{0,n}=0$ for $n>0$ to eliminate some 
terms, changed the order of sums, and used the following identity 
for the Bell polynomials \cite[§ Lemma 5.9]{Grosse:2016pob}
\begin{align}
\sum_{j=1}^{n-k}\binom{n}{j}x_jB_{n-j,k}(x_1,...,x_{n-j-k+1})
=(k+1)B_{n,k+1}(x_1,...,x_{n-k}).
\label{xBell}
\end{align}
Inserting back, we find that \eqref{Bell2} reduces to the ($j=k$)-term
of the first sum in the last line of \eqref{Bell2}, i.e.\
\begin{align*}
&z^2 \hat{K}_z \Big(
\frac{1}{z} \Res\displaylimits_{z'\to 0}\Big[ 
K(z,z')\frac{dz'}{(z')^{2k}}\Big]\Big)= -\frac{1}{z^{2k}}.
\end{align*}
This finishes the proof.
\end{proof}
\begin{thrm} \label{thm:G-residue}
For any $g\geq 1$ and $z\in \mathbb{C}\setminus \{0\}$ 
one has 
\begin{align}
\G_g(z) = \frac {\lambda}{ z} 
\Res\displaylimits_{z'\to 0} 
\left[ K(z,z') \Big\{
\sum_{h=1}^{g-1}\G_h(z') 
\G_{g-h}(z') +\G_{g-1}(z',z')\Big\} 
(z')^2 dz'\right].
\label{eq-G1-final}
\end{align}
\end{thrm}
\begin{proof} 
The formula arises when applying \sref{Proposition}{Prop5} to
\eqref{eq:Cubic1PComplex} with $b=1$ and holds if 
the function in $\{~\}$ is an even Laurent polynomial in $z'$ 
bounded in $\infty$.
This is the case for $g=1$ where only $\G_0(z',z')
=\frac{\lambda^2}{(z')^4}$ contributes. Evaluation of the residue 
reconfirms the 1-point function of \sref{Proposition}{Prop:Cubic1+1+1}. We proceed by induction in 
$g\geq 2$, assuming that all 
$\G_h(z')$ with $1\leq h<g$ on the rhs of \eqref{eq-G1-final} 
are odd Laurent polynomials bounded in $\infty$; 
their product is even. The induction hypothesis 
also verifies the assumption of 
\sref{Theorem}{finaltheorem} so that 
$\G_{g-1}({-}z',{-}z')=-\G_g(z',{-}z')=\G_g(z',z')$ is even and, because of 
$\G_{g-1}(z',z'')=(2\lambda)^3 \hat{\mathrm{A}}^{\dag g}_{z'',z'} 
\G_{g-1}(z'')$, inductively a 
Laurent polynomial bounded in $\infty$. 
Thus, equation \eqref{eq-G1-final} holds for genus $g\geq 2$ and, 
as consequence of \eqref{Res}, $\G_g(z)$ is 
again an odd Laurent polynomial bounded in $\infty$. 
Equation~\eqref{eq-G1-final}  is thus proved 
for all $g\geq 1$, and the assumption of 
\sref{Theorem}{finaltheorem} is verified.
\end{proof}
\noindent
A more precise characterisation can be given. It relies on
\begin{dfnt}
A polynomial $P(x_1,x_2,\dots)$ is called \textit{$n$-weighted} if 
\begin{align*}\sum_{k=1}^\infty k x_k \frac{\partial}{\partial x_k} P(x_1,x_2,\dots)
=nP(x_1,x_2,\dots).
\end{align*}
\end{dfnt}
\noindent 
The Bell polynomials $B_{n,k}(x_1,\dots,x_{n-k+1})$ are $n$-weighted. 
The number of monomials in an $n$-weighted polynomial is $p(n)$, 
the  number of partitions of $n$.
The product of an $n$-weighted 
by an $m$-weighted polynomial is $(m+n)$-weighted.
\begin{prps}
\label{structure-Gz}
 For $g\geq 1$ one has
\[
\G_g(z)= (2 \lambda)^{4g-1} 
\sum_{k=0}^{3g-2} \frac{P_{3g-2-k}(\varrho)}{\varrho_0^{2g-1} z^{2k+3}},
\]
where $P_0\in \mathbb{Q}$ and the $P_j(\varrho)$ with $j\geq 1$ are 
$j$-weighted polynomials in 
$\{\frac{\varrho_1}{\varrho_0},\dots,\frac{\varrho_j}{\varrho_0}\}$
with rational coefficients. 
\end{prps}
\begin{proof}
The case $g=1$ is directly checked. We proceed by induction in $g$ for both 
terms in $\{~\}$ in \eqref{eq-G1-final}. 
The hypothesis gives
$\G_h(z')\G_{g-h}(z')= (2\lambda)^{4g-2} \sum_{k=0}^{3g-4} 
\frac{P_{3g-4-k}(\varrho)}{\varrho_0^{2g-2} (z')^{2k+6}}$. In the
second term in $\{~\}$, 
$\G_{g-1}(z,z)=(2\lambda)^3 \hat{\mathrm{A}}^{\dag\,g{-}1}_{z,z} 
\G_{g-1}(z)$, the three 
types of contributions in the boundary creation operator act as follows:
\begin{align*}
&(2\lambda)^3 
 \hat{\mathrm{A}}^{\dag\,g{-}1}_{z',z'} \G_{g-1}(z')\\
=& \frac{(2\lambda)^{4g-2}}{\varrho_0^{2g-2} }
\sum_{k=0}^{3g-5} 
\Big( 
\sum_{l=0}^{3g-5-k} \!\!\!
\Big(\frac{\varrho_{l+1}}{\varrho_0 (z')^3} {+}\frac{1}{(z')^{5+2l}} \Big) 
\frac{P_{3g-5-k-l}(\varrho)}{(z')^{2k+3}}
{+} \frac{1}{(z')^4} 
\frac{P_{3g-5-k}(\varrho)}{(z')^{2k+4}}\Big)
\\
=& \frac{(2\lambda)^{4g-2}}{\varrho_0^{2g-2} }
\sum_{k=0}^{3g-5} 
\Big( 
\frac{P_{3g-4-k}(\varrho)}{(z')^{2k+6}}
+\frac{P_{3g-5-k}(\varrho)}{(z')^{2k+8}}\Big),
\end{align*}
which has the same structure as $\G_h(z')\G_{g-h}(z')$.
Application of \eqref{Res} yields 
\begin{align*}
&\frac{\lambda}{z}\Res\displaylimits_{z'\to 0} 
\Big[K(z,z') dz' (z')^2 (2\lambda)^{4g-2}
\sum_{k=0}^{3g-4} 
\frac{P_{3g-4-k}(\varrho)}{\varrho_0^{2g-2} (z')^{2k+6}}\Big]\\
=&(2\lambda)^{4g-1}\sum_{k=0}^{3g-4} \sum_{j=0}^{k+2} 
\frac{P_{3g-4-k}(\varrho)S_j(\varrho)}{\varrho_0^{2g-1} z^{2k+7-2j}}
\\
=&(2\lambda)^{4g-1}\sum_{k=0}^{3g-2} 
\frac{P_{3g-2-k}(\varrho)}{\varrho_0^{2g-1} z^{2k+3}},
\end{align*}
because $S_j(\varrho)$ is also a $j$-weighted polynomial by 
\eqref{S-Bell}.
\end{proof}
\noindent
In particular, this proves the assumption of 
\sref{Theorem}{finaltheorem}, namely that 
$\G_g(z)$ depends only on $\{\varrho_0,\dots, \varrho_{3g-2}\}$.
To be precise, we reciprocally increase the genus in 
\sref{Theorem}{finaltheorem} and \sref{Proposition}{structure-Gz}.

\subsection{Link to Topological Recursion}\label{sec.linkcub}
A $(1+1+...+1)$-point
function of genus $g$ with $b$ boundary components fulfils a
universal structure called \textit{topological recursion}. To
introduce it, we have to define the following functions:
\begin{dfnt}\label{def1}
 The function $\omega_{g,b}$ is defined by
 \begin{align*}
  \omega_{g,b}(z_1,...,z_b):=&\left(\prod_{i=1}^b z_i\right) 
  \left(\G_g(z_1,...,z_b)
+16\lambda^2\frac{\delta_{g,0}\delta_{2,b}}{(z_1^2-z_2^2)^2}\right)dz_1\,..dz_n
 \end{align*}
 as symmetric differential forms on $\hat{\C}^n$ with $\hat{\C}=\C\cup \{\infty\}$,
and the {\sffamily spectral curve} $y(z(x))$ by $x(z)=z^2-c$ and
\begin{align}\label{spec:cubic}
 y(z)=&-\G_0(-z)= \frac{z}{2\lambda\sqrt{Z}}
+\frac{\nu}{2}-\frac{1}{4\lambda}\int_{\sqrt{1+c}}^{\sqrt{(1+2e(\Lambda^2))^2+c}}dt
\frac{\tilde{\varrho}(t)}{t(t-z)}.
\end{align}
\end{dfnt}
\noindent
It can be checked that with these definitions, up to
trivial redefinitions by powers of $2\lambda$, the theorems proved
in \cite{Eynard:2016yaa} apply for topological recursion. These determine all 
$\omega_{g,B}$ with $2-2g-B<0$ out of the initial data 
$y(z),x(z)$ and $\omega_{0,2}$:
\begin{thrm}[{\cite[Thm.\ 6.4.4]{Eynard:2016yaa}}]\label{theorem1}
 For $2-2g-(1+b)<0$ and $J=\{z_1,...,z_b\}$ the 
function $\omega_{g,b+1}(z_0,...,z_B)$ is given by topological recursion
 \begin{align*}
  &\omega_{g,b+1}(z_0,...,z_b)\\
  =&\Res\displaylimits_{z\to 0}
\bigg[K(z_0,z)\,dz\Big(\omega_{g-1,b+2}(z,-z,J)+
  \sum^{\quad \prime}_{\substack{h+h'=g\\ I\uplus I' =J}}\omega_{h,|I|+1}(z,I)
  \omega_{h',|I'|+1}(-z,{I'})\Big)\bigg],
 \end{align*}
 where $K(z_0,z)=\frac{\lambda}{(z^2-z_0^2)(y(z)-y(-z))}$ and 
 the sum $\sum^{\prime}$ excludes $(h,I)=(0,\emptyset)$ and $(h,I)=(g,J)$.
\end{thrm}
\noindent
This theorem motivated our ansatz 
for an inverse of $\hat{K}_z$ as the residue involving $K(z,z')$.
The case $J=\emptyset$ of \sref{Theorem}{theorem1} is essentially the same 
as \sref{Theorem}{thm:G-residue}.
For us there is no need to prove the 
general case because higher $\omega(J)$ can be obtained from 
\sref{Theorem}{finaltheorem}. 

Notice that our proof differs completely and is more combinatorial in comparison to 
the proof of \sref{Theorem}{theorem1}, since \sref{Theorem}{theorem1}
considers a deformation of a complex contour for which
$y(z)$ has to be a meromorphic function (which can be achieved 
for finite $\mN',V$) instead of a sectional holomorphic function.
The contour is then moved through all possible poles and picks up the residues. However, moving the contour 
through a cut (even of infinite length for $\Lambda^2\to\infty$) is much harder and cannot be done 
without further effort. 

In other words, we have shown that the proofs of \sref{Theorem}{finaltheorem} and \sref{Theorem}{thm:G-residue}
are also valid if the large $\mN,V$-limit is performed first, which implies that the topological recursion and 
the large  $\mN,V$-limit commutes in that special case.

\section[Free Energy]{Free Energy $F_g$\footnote{This section is taken from our paper
\cite{Grosse:2019nes}}}\label{Sec:CubicFreeEnergy}
The free energy $F_g$ generates the connected vacuum graphs of genus $g$. Since the partition function $\Z[0]$
generates (not necessarily
connected) graphs, the free energy is understood as the logarithm of the partition function or equivalently after genus expansion
\begin{align*}
 \Z[0]=\exp\left(\sum_{g=0}^\infty V^{2-2g} F_g\right).
\end{align*}
To derive the free energies, the inverse of the boundary creation operator has to be defined in such a way 
that it produces the 
free energy $F_g$  from the 1-point function $\G_g(z)$ of genus $g$ uniquely. 
\begin{dfnt}
We introduce the operators
\begin{align}
\hat{\mathrm{A}}^{\dag}_z &:=
\sum_{l= 0}^{\infty} \Big(-\frac{(3+2l) \varrho_{l+1}}{
\varrho_0 z^3}+\frac{3+2l}{z^{5+2l}}\Big)
\frac{\partial}{\partial \varrho_l}, &
\hat{\mathrm{N}} &= - \sum_{l=0}^\infty
\varrho_l\frac{\partial}{\partial \varrho_l} ,
\nonumber
\\
\hat{\mathrm{A}}_{\check{z}}  f(\bullet) &:=-
\sum_{l=0}^\infty \Res\displaylimits_{z\to 0}
\Big[\frac{z^{4+2l} \varrho_l}{3+2l} f(z) dz\Big].
\end{align}
We call $\hat{\mathrm{A}}_{\check{z}}$ a {\sffamily boundary
annihilation operator} acting on Laurent polynomials $f$. 
\end{dfnt}
\begin{prps}\label{prop:F}
There is a unique function $F_g$ of $\{\varrho_l\}$ satisfying 
$\G_g(z) = (2\lambda)^3 \hat{\mathrm{A}}^{\dag}_z F_g$, 
\begin{align*}
F_1=-\frac{1}{24} \log \varrho_0,\qquad 
F_g&= 
\frac{1}{(2g-2)(2\lambda)^3} 
\hat{\mathrm{A}}_{\check{z}}  \G_g(\bullet) 
\text{ for } g\geq 1.
\end{align*}
The $F_g$ have for $g>1$ a presentation as
\begin{equation}\label{Fg-homogeneous}
F_g= (2 \lambda)^{4g-4} \frac{P_{3g-3}(\varrho) }{\varrho_0^{2g-2}},
\end{equation}
where $P_{3g-3}(\varrho)$ is a $(3g-3)$-weighted polynomial in 
$\{\frac{\varrho_1}{\varrho_0},\dots,
\frac{\varrho_{3g-3}}{\varrho_0}\}$.
\end{prps}
\begin{proof} 
The case $g=1$ is checked by direct comparison with \eqref{Prop:Cubic1+1+1}. 
From \sref{Proposition}{structure-Gz}
we conclude 
\begin{align*}
\frac{1}{(2\lambda)^3} \hat{\mathrm{A}}_{\check{z}} \G_g(\bullet)
&=-
(2 \lambda)^{4g-4} \Res\displaylimits_{z\to 0}\Big[
\sum_{l=0}^\infty \frac{\varrho_lz^{4+2l}}{(3+2l)}
\sum_{k=0}^{3g-2} \frac{P_{3g-2-k}(\varrho)}{\varrho_0^{2g-1} z^{2k+3}}
dz\Big]
\\
&= (2 \lambda)^{4g-4} 
\sum_{k=1}^{3g-2} \frac{
\frac{\varrho_{k-1}}{\varrho_0}
\cdot P_{3g-2-k}(\varrho)}{\varrho_0^{2g-2}}
= (2 \lambda)^{4g-4} \frac{P_{3g-3}(\varrho) }{\varrho_0^{2g-2}},
\end{align*}
which confirms \eqref{Fg-homogeneous}. Observe that the 
total $\varrho$-counting operator $\hat{\mathrm{N}}$ applied to 
any polynomial in $\{\frac{\varrho_1}{\varrho_0},
\frac{\varrho_2}{\varrho_0},\dots\}$ is zero. Therefore, for $g>1$,
\[
\hat{\mathrm{N}} 
\Big(\frac{1}{(2\lambda)^3} \hat{\mathrm{A}}_{\check{z}} \G_g(\bullet)\Big)
= (2g-2) \cdot 
\Big(\frac{1}{(2\lambda)^3} \hat{\mathrm{A}}_{\check{z}} \G_g(\bullet)\Big).
\]
The boundary annihilation operator is designed to satisfy
$\hat{\mathrm{A}}_{\check{z}} \circ \hat{\mathrm{A}}^{\dag}_{\bullet} 
=\hat{\mathrm{N}}$.
Dividing the previous equation
by $(2g-2)$ and inserting the ansatz for $F_g$ given in the 
proposition, we have
\[
0= \hat{\mathrm{N}} F_g 
-\frac{1}{(2\lambda)^3} \hat{\mathrm{A}}_{\check{z}} \G_g(\bullet)
= \hat{\mathrm{A}}_{\check{z}} \Big(\hat{\mathrm{A}}^{\dag}_{\bullet} F_g
-\frac{1}{(2\lambda)^3} \G_g(\bullet)\Big).
\]
Since $f(z):=\hat{\mathrm{A}}^{\dag}_{z} F_g
-\frac{1}{(2\lambda)^3} \G_g(z)$ is by (\ref{Fg-homogeneous}) 
and Proposition~\ref{structure-Gz} a Laurent polynomial 
bounded at $\infty$, application of 
$\hat{\mathrm{A}}_{\check{z}}$ can only vanish if $f(z)\equiv 0$. 
This finishes the proof.
\end{proof}
\begin{rmk}
  \sref{Proposition}{prop:F} shows that the $F_g$ provide the most
  condensed way to describe the non-planar sector of the
  $\Phi^3$-matricial QFT model.  All information about the genus-$g$
  sector is encoded in the $p(3g-3)$ rational numbers which form the
  coefficients in the $(3g-3)$-weighted polynomial in
  $\{\frac{\varrho_1}{\varrho_0},
  \frac{\varrho_2}{\varrho_0},\dots\}$. From these polynomials we
  obtain the $(1+\dots+1)$-point function with $b$ boundary components
  via $\G_g(z)=(2\lambda)^3 \hat{\mathrm{A}}^{\dag}_z F_g$
  followed by \sref{Theorem}{finaltheorem}.
\end{rmk}

\begin{lemma}
For $(2g+b-2)>0$, the operator $\hat{\mathrm{N}}$ measures the 
Euler characteristics, 
\begin{align*} 
\hat{\mathrm{N}} \G_g(z_1,...,z_b) = (2g+b-2)\G_g(z_1,...,z_b) .
\end{align*}
\end{lemma}
\begin{proof}
Both cases with $(2g+b-2)=1$ are directly checked. The general case follows 
by induction from $[\hat{\mathrm{N}},\hat{\mathrm{A}}^{\dag g}_{J,z}]=\hat{\mathrm{A}}^{\dag g}_{J,z}$ 
in combination with \sref{Theorem}{finaltheorem} and 
$\hat{\mathrm{N}} F_g=(2g-2)F_g$ for $g\geq 2$.
\end{proof}
\begin{cor}
\begin{align*}
\hat{\mathrm{A}}_{\check{z}} \G_g(\bullet,z_2,...,z_b)
= (2\lambda)^{3} (2g+b-3)\G_g(z_2,...,z_b)
\end{align*}
whenever $(2g+b-3)>0$.
\end{cor}
\noindent
Hence, up to a rescaling, $\hat{\mathrm{A}}_{\check{z}}$
indeed removes the boundary component previously located at $z$.  We
also have $\hat{\mathrm{A}}_{\check{z}} F_g=0$ for all $g\geq
1$, so the $F_g$ play the r\^ole of a vacuum or the 
free energy.
Note that $\G_0(z)$ cannot be produced with the operator $\hat{\mathrm{A}}^{\dag }_{z}$ by whatever $F_0$. 

\begin{rmk}\label{Rmk:F0}
 In the case of $\D=0$ the free energy of genus $g=0$ exists and is given by (adapted from \cite{Makeenko:1991ec})
 \begin{align*}
  F_0=&\frac{1}{12\lambda^2}\int_{0}^{\Lambda^2}dt\,r(t)\sqrt{(1+2e(t))^2+c}^3-
  \frac{1}{12\lambda^2}\int_{0}^{\Lambda^2}dt\,r(t)(1+2e(t))^3\\
  &-\frac{ c}{8\lambda^2}\int_{0}^{\Lambda^2}dt\,r(t)\sqrt{(1+2e(t))^2+c}
  -\frac{c^3}{12\cdot 4^3\lambda^4}\\
  &-\frac{1}{2}\int_{0}^{\Lambda^2}dx\,r(x)\int_{0}^{\Lambda^2}dt\,r(t)\log(\sqrt{(1+2e(t))^2+c}+
  \sqrt{(1+2e(x))^2+c}).
 \end{align*}
Acting with the operator $T(x)$ (defined in \eqref{eq:Tx})
on $F_0$ gives precisely $G^0(x)=\frac{ \tilde{W}((1+2e(x))^2)-(1+2e(x))}{2\lambda}$. 
All parts proportional to $T(x)c$ vanish which makes $F_0$ stationary with respect to $c$.
By performing the computation $T(x)F_0$, the implicit function of $c$ 
\eqref{eq:cimplicit} for $\D<2$
is of tremendous importance. Note that the planar free energy diverges for $\Lambda^2\to\infty$
such that finite $\mN',V$ or finite $\Lambda^2$ are necessary which implies $\D=0$. 
Any renormalisation constant prevents to write down $F_0$ in a closed form even for finite $\Lambda^2$
since the implicit equation of $c$ differs.
\end{rmk}

\subsection{A Laplacian to Compute Intersection Numbers}\label{Sec:CubicLaplacian}
The naive picture of intersection numbers is that they are counting the number of intersections of curves, which should give 
positive integers. However, if one assumes complex curves up to some equivalence class this will change tremendously 
and can give rational intersection numbers instead of integers.

Let $\mathcal{M}_{g,b}$ be the \textit{moduli space} of equivalence classes of complex curves of genus $g$ with $b$ distinct 
marked points, modulo biholomorphic reparametrisation. For a negative Euler characteristic $\chi=2-2g-b$, the moduli space 
$\mathcal{M}_{g,b}$ is locally parameterised by $d_{g,b}:=(3g-3+b)$ complex parameters called \textit{moduli}. 
The fact that $\mathcal{M}_{g,b}$ has rational intersection numbers is due to the orbifold structure, which looks 
locally like the quotient space of $\mathbb{C}^n$ under a linear action of a finite group.
An orbifold is a generalisation of a manifold, where the finite group is trivial. 

The moduli space $\mathcal{M}_{g,b}$ is, in general, not compact and can be compactified by adding degenerate curves. 
The Deligne-Mumford compactification \cite{PMIHES_1969__36__75_0} provides an analytic structure and is the usual compactification used in the literature
denoted by $\overline{\mathcal{M}}_{g,b}$.

The $b$ distinct marked points lead naturally to a family $\{\mathcal{L}_1,..,\mathcal{L}_b\}$ of complex
line bundles over $\overline{\mathcal{M}}_{g,b}$. A classification of the line bundles is given by the first 
Chern class $c_1(\mathcal{L}_i)$ which is an element of the second cohomology group
$H^2(\overline{\mathcal{M}}_{g,b},\mathbb{Q})$. The first Chern class $c_1(\mathcal{L}_i)$ is independent of the 
choice of the connection on $\mathcal{L}_i$ and 
has a curvature form 
as representative. 

The \textit{intersection numbers} of $\overline{\mathcal{M}}_{g,b}$ are then defined by
\begin{align}\label{eq:Intersec}
 \langle\tau_{d_1}...\tau_{d_s}\rangle:=\int_{\overline{\mathcal{M}}_{g,b}}\bigwedge_{i=1}^s 
 \big(c_1(\mathcal{L}_i)\big)^{d_i},\qquad \sum_{i=1}^sd_i=d_{g,b}
\end{align}
where the 2-forms $c_1(\mathcal{L}_i)$ are multiplied with a commutative wedge product.
These numbers are topological invariants of $\overline{\mathcal{M}}_{g,b}$.

It was then conjectured by Witten \cite{Witten:1990hr} that the generating function of these 
intersection numbers satisfies a partial differential equation called the Korteweg-de Vries equation. 
Additionally, it was known that the integrability of matrix models is due to partial differential equation of the free energy
which was also of Korteweg-de Vries type. This observation was subsequently proved to be not a coincidence
\begin{thrm}(\cite{Kontsevich:1992ti})\label{Thm:Kontsevich}
 The generating function of the intersection numbers of $\psi$-classes on the moduli space $\overline{\mathcal{M}}_{g,b}$ 
 defined in \eqref{eq:Intersec} is equal to the free energy $F_g$ of the Kontsevich model
 \begin{align*}
  F_g=\sum_{(k)}\frac{\langle\tau_{2}^{k_2}\tau_3^{k_3}...\tau_{3g-2}^{k_{3g-2}}\rangle}{(-t_1)^{2g-2+\sum_ik_i}}
  \prod_{i=2}^{3g-2}\frac{t_i^{k_i}}{k_i!},\qquad \sum_{i\geq 2}(i-1)k_i=3g-3,
 \end{align*}
where $t_{i+1}:=-(2i+1)!!\varrho_{i}$ of \sref{Definition}{KontTime}.
\end{thrm}
\noindent
Since we have now a prescription to derive the free energy $F_g$ for any $g>0$, also the intersection numbers can be 
derived simultaneously due to \sref{Theorem}{Thm:Kontsevich}.

To achieve this, we express for $b=1$ in \eqref{eq:CubicSchwingerComplex}
$\G_g(z)=(2\lambda)^3 \hat{\mathrm{A}}^{\dag}_z F_g$ and 
$\G_g(z,z)=(2\lambda)^3
(\hat{\mathrm{A}}^{\dag}_z +\frac{1}{\varrho_0 z^4} 
\frac{\partial}{\partial z})(\G_g(z))$ and
multiply by $\frac{2 V^{4-2g}}{(2\lambda)^8}
\mathcal{Z}^{np}_V$. Summation over $g\geq 1$ gives
\begin{align}\label{ident-Z}
&0 =
\Big(\frac{2V^2}{(2\lambda)^4} \hat{K}_z \hat{\mathrm{A}}^\dag_z + 
\Big(\hat{\mathrm{A}}^\dag_z +\frac{1}{\varrho_0 z^4} 
\frac{\partial}{\partial z} \Big)
\hat{\mathrm{A}}^\dag_z +\frac{V^2}{4 (2\lambda)^4 z^4} \Big)
\mathcal{Z}_V^{np},\\
&\text{where}\qquad 
\mathcal{Z}_V^{np}
:= \exp\Big(\sum_{g=1}^\infty V^{2-2g} F_g\Big).\nonumber
\end{align}
We invert $\hat{K}_z$ via \sref{Proposition}{Prop5} and apply 
$\hat{\mathrm{A}}_{\check{z}}$ given by the residue
in \sref{Proposition}{prop:F}:
\begin{align*}
&\frac{2V^2}{(2\lambda)^4} 
\hat{\mathrm{N}} \mathcal{Z}_V^{np}
\\
&= {-} \!\sum_{\ell=0}^\infty 
\Res\displaylimits_{z\to 0}\!\Big[dz  \frac{z^{3+2\ell} \varrho_\ell}{(3+2\ell)}
\Res\displaylimits_{z'\to 0}\!
\Big[ dz'(z')^2 K(z,z')\Big(\Big(
\hat{\mathrm{A}}^\dag_{z'} {+}\frac{1}{\varrho_0 (z')^4} 
\frac{\partial}{\partial z'} \Big)
\hat{\mathrm{A}}^\dag_{z'} {+}\frac{V^2}{4 (2\lambda)^4 (z')^4} \Big)
\Big]\Big]\mathcal{Z}_V^{np}.
\nonumber
\end{align*}
We insert $K(z,z')$ from \sref{Proposition}{Prop5}, expand only 
the geometric series about $z'=0$ while keeping (\ref{Gzminusz}). Then 
the outer residue in $z$ is immediate
\begin{align}
\frac{2V^2}{(2\lambda)^4} 
\hat{\mathrm{N}} \mathcal{Z}_V^{np}
&=  
\Res\displaylimits_{z'\to 0}
\Big[ 
(z')^3\,dz' \frac{\sum_{\ell=0}^\infty \frac{(z')^{2\ell} \varrho_\ell}{
(3+2\ell)}}{\sum_{j=0}^\infty \varrho_j  (z')^{2j}}
\Big(\Big(
\hat{\mathrm{A}}^\dag_{z'} +\frac{1}{\varrho_0 (z')^4} 
\frac{\partial}{\partial z'} \Big)
\hat{\mathrm{A}}^\dag_{z'} +\frac{V^2}{4 (2\lambda)^4 (z')^4} \Big)
\Big]\Big]\mathcal{Z}_V^{np}.
\nonumber
\end{align}
We rename $z'$ to $z$ and introduce the function
\begin{align*}
\mathcal{R}(z)=\frac{\sum_{\ell=0}^\infty \frac{ \varrho_\ell z^{2\ell}}{
(3+2\ell)}}{\sum_{j=0}^\infty \varrho_j  z^{2j}}
=\sum_{m=0}^\infty \mathcal{R}_m(\varrho)\, z^{2m}.
\end{align*}
The denominator is given by \eqref{S-Bell}, without the 
$\frac{1}{z'\varrho_0}$ 
prefactor. It combines with the numerator to 
\begin{align}
\mathcal{R}_m(\varrho)&= 
\frac{S_m(\varrho)}{3 m!} - \sum_{k=1}^m 
\frac{\varrho_k}{(3+2k)\varrho_0} \frac{S_{m-k}(\varrho)}{(m-k)!}
=-\frac{2}{3} \sum_{k=1}^m 
\frac{k \varrho_k}{(3+2k)\varrho_0} \frac{S_{m-k}(\varrho)}{(m-k)!},
\label{Rmrho}
\end{align}
where we have used \eqref{xBell} for the first $S_m(\varrho)$ 
to achieve better control of signs.

The residue of $\frac{V^2}{4 (2\lambda)^4 z^4}$ is immediate and can 
be moved to the lhs:
\begin{align*}
\frac{2V^2}{(2\lambda)^4} 
\Big(\hat{\mathrm{N}}-\frac{1}{24}\Big) \mathcal{Z}_V^{np}
&=  \sum_{m=0}^\infty \mathcal{R}_m(\varrho)
\Res\displaylimits_{z\to 0}
\Big[ z^{3+2m}\,dz 
\Big(\Big(
\hat{\mathrm{A}}^\dag_{z} +\frac{1}{\varrho_0 z^4} 
\frac{\partial}{\partial z} \Big)
\hat{\mathrm{A}}^\dag_{z} 
\Big)\Big]\mathcal{Z}_V^{np}
\nonumber
\\
&=\Big[  
\sum_{k=0}^\infty 
\Big(-\frac{3(3+2k)\varrho_1 \varrho_{k+1}\mathcal{R}_1(\varrho)}{\varrho_0^3}
+\frac{3(3+2k)\varrho_{k+1} \mathcal{R}_2(\varrho)}{\varrho_0^2}
\Big)\frac{\partial}{\partial \varrho_k} 
\\
&+
\sum_{k,l=0}^\infty \frac{(3+2k)(3+2l)\mathcal{R}_1(\varrho)}{\varrho_0^2}
\varrho_{l+1}\frac{\partial}{\partial \varrho_l}
\varrho_{k+1}\frac{\partial}{\partial \varrho_k}
\\
&-\sum_{k,l=0}^\infty \frac{(3+2k)(3+2l)\mathcal{R}_{l+2}(\varrho)}{\varrho_0}
\Big(
\varrho_{k+1}\frac{\partial}{\partial \varrho_k}
\frac{\partial}{\partial \varrho_l}
+\frac{\partial}{\partial \varrho_l}
\varrho_{k+1}\frac{\partial}{\partial \varrho_k}
\Big)
\\
&+\sum_{k,l=0}^\infty (3+2k)(3+2l)\mathcal{R}_{k+l+3}(\varrho)
\frac{\partial}{\partial \varrho_k}
\frac{\partial}{\partial \varrho_l}
\\
&+\sum_{k=0}^\infty \frac{3(3{+}2k)\varrho_{k+1}\mathcal{R}_2(\varrho)}{\varrho_0^2}
\frac{\partial}{\partial \varrho_k}
-\sum_{k=0}^\infty \frac{(3{+}2k)(5{+}2k)\mathcal{R}_{k+3}(\varrho)}{\varrho_0}
\frac{\partial}{\partial \varrho_k}
\Big]\mathcal{Z}_V^{np}.
\end{align*}
Next we separate the $\varrho_0$-derivatives:
\begin{align*}
&\frac{2V^2}{(2\lambda)^4} 
\Big(\hat{\mathrm{N}}-\frac{1}{24}\Big) \mathcal{Z}_V^{np}
\nonumber
\\
&= 
\Big[ 
\Big(\frac{9 \mathcal{R}_1(\varrho) \varrho_1^2}{\varrho_0^2}
- \frac{18 \mathcal{R}_{2}(\varrho)\varrho_1}{\varrho_0}
+9\mathcal{R}_{3}(\varrho)\Big)
\frac{\partial^2}{\partial \varrho_0^2}
\\
& +\Big(
-\frac{9\varrho_1^2 \mathcal{R}_1(\varrho)}{\varrho_0^3}
+\frac{18\varrho_{1}\mathcal{R}_2(\varrho)}{\varrho_0^2}
+\frac{15 \mathcal{R}_1(\varrho) \varrho_2}{\varrho_0^2}
-\frac{30 \mathcal{R}_{3}(\varrho)}{\varrho_0}\Big)
\frac{\partial}{\partial \varrho_0}
\\
&
+\sum_{k=1}^\infty 6(3+2k)\Big(
\mathcal{R}_{k+3}(\varrho)
-\frac{\mathcal{R}_{2}(\varrho) \varrho_{k+1}}{\varrho_0}
-\frac{\mathcal{R}_{k+2}(\varrho)\varrho_{1}}{\varrho_0}
+\frac{\mathcal{R}_1(\varrho)\varrho_{k+1}\varrho_{1}}{\varrho_0^2}\Big)
\frac{\partial}{\partial \varrho_k}
\frac{\partial}{\partial \varrho_0}
\\
&
+
\sum_{k,l=1}^\infty (3+2k)(3+2l)\Big(
\frac{ \varrho_{l+1}\varrho_{k+1} \mathcal{R}_1(\varrho)}{\varrho_0^2}
+ \mathcal{R}_{k+l+3}(\varrho)
- \frac{2\varrho_{k+1}\mathcal{R}_{l+2}(\varrho)}{\varrho_0}
\Big)
\frac{\partial}{\partial \varrho_k}
\frac{\partial}{\partial \varrho_l}
\\
&
+\sum_{k=1}^\infty (3k+2)
\Big(-\frac{3\varrho_1 \varrho_{k+1}\mathcal{R}_1(\varrho)}{\varrho_0^3}
+\frac{6\varrho_{k+1} \mathcal{R}_2(\varrho)}{\varrho_0^2}
\\
&\qquad\qquad 
+ \frac{(5+2k)\varrho_{k+2} \mathcal{R}_1(\varrho)}{\varrho_0^2}
-\frac{2(5+2k)\mathcal{R}_{k+3}(\varrho)}{\varrho_0}
\Big)\frac{\partial}{\partial \varrho_k}
\Big]\mathcal{Z}_V^{np}.
\end{align*}
We isolate $F_1$, i.e.\ 
$\mathcal{Z}_V^{np}=\varrho_0^{-\frac{1}{24}} 
\mathcal{Z}_V^{stable}$, where 
$\mathcal{Z}_V^{stable}
=\exp\Big(\sum_{g=2}^\infty V^{2-2g} F_g\Big)$. We commute the factor 
$\varrho_0^{-\frac{1}{24}}$ in front of $[~~]$ and move it to the other side:
\begin{align*}
\frac{2V^2}{(2\lambda)^4} \hat{\mathrm{N}} \mathcal{Z}_V^{stable}
&= 
\Big[ 
\Big(\frac{49 \varrho_1^2 \mathcal{R}_1(\varrho)}{64\varrho_0^4}
- \frac{49 \varrho_1 \mathcal{R}_{2}(\varrho)}{32\varrho_0^3}
-\frac{5 \mathcal{R}_1(\varrho) \varrho_2}{8\varrho_0^3}
+\frac{105 \mathcal{R}_{3}(\varrho)}{64\varrho_0^2}\Big)
\\
&\!\!\!
+
\Big(\frac{9 \mathcal{R}_1(\varrho) \varrho_1^2}{\varrho_0^2}
- \frac{18 \mathcal{R}_{2}(\varrho)\varrho_1}{\varrho_0}
+9\mathcal{R}_{3}(\varrho)\Big)
\frac{\partial^2}{\partial \varrho_0^2}
\\
& \!\!\!
+\Big(
-\frac{39\varrho_1^2 \mathcal{R}_1(\varrho)}{4\varrho_0^3}
+\frac{39\varrho_{1}\mathcal{R}_2(\varrho)}{2\varrho_0^2}
+\frac{15 \mathcal{R}_1(\varrho) \varrho_2}{\varrho_0^2}
-\frac{123 \mathcal{R}_{3}(\varrho)}{4\varrho_0}\Big)
\frac{\partial}{\partial \varrho_0}
\\
&\!\!\!
+\sum_{k=1}^\infty 6(3{+}2k)\Big(
\mathcal{R}_{k+3}(\varrho)
-\frac{\mathcal{R}_{2}(\varrho) \varrho_{k+1}}{\varrho_0}
-\frac{\mathcal{R}_{k+2}(\varrho)\varrho_{1}}{\varrho_0}
+\frac{\mathcal{R}_1(\varrho)\varrho_{k+1}\varrho_{1}}{\varrho_0^2}\Big)
\frac{\partial}{\partial \varrho_k}
\frac{\partial}{\partial \varrho_0}
\\
&\!\!\!
+
\sum_{k,l=1}^\infty (3{+}2k)(3{+}2l)\Big(
\frac{ \varrho_{l+1}\varrho_{k+1} \mathcal{R}_1(\varrho)}{\varrho_0^2}
+ \mathcal{R}_{k+l+3}(\varrho)
- \frac{2\varrho_{k+1}\mathcal{R}_{l+2}(\varrho)}{\varrho_0}
\Big)
\frac{\partial}{\partial \varrho_k}
\frac{\partial}{\partial \varrho_l}
\\
&\!\!\!
+\sum_{k=1}^\infty (3{+}2k)\Big(
-\frac{\mathcal{R}_{k+3}(\varrho)}{4 \varrho_0}
+\frac{25\varrho_{k+1} \mathcal{R}_2(\varrho)}{4\varrho_0^2}
+\frac{\mathcal{R}_{k+2}(\varrho)\varrho_{1}}{4\varrho_0^2}
-\frac{13\varrho_1 \varrho_{k+1}\mathcal{R}_1(\varrho)}{4\varrho_0^3}
\\
&\qquad\qquad 
+ \frac{(5{+}2k)\varrho_{k+2} \mathcal{R}_1(\varrho)}{\varrho_0^2}
-\frac{2(5{+}2k)\mathcal{R}_{k+3}(\varrho)}{\varrho_0}
\Big)\frac{\partial}{\partial \varrho_k}
\Big]\mathcal{Z}_V^{stable}.
\end{align*}
Next observe
\begin{align*}
\hat{\mathrm{N}}\mathcal{Z}_V^{stable}
&
=\sum_{g=2}^\infty (V^{-2})^{g-1}(2g-2) \mathcal{Z}_g
=2V^{-2}\frac{d}{dV^{-2}} \sum_{g=2}^\infty (V^{-2})^{g-1}
\mathcal{Z}_g
=2V^{-2}\frac{d}{dV^{-2}} \mathcal{Z}_V^{stable}.
\end{align*}
Consequently, we obtain a parabolic differential equation in $V^{-2}$
which is easily solved. Inserting 
\[
\mathcal{R}_1(\varrho)=-\frac{2}{15} \frac{\varrho_1}{\varrho_0},\quad
\mathcal{R}_2(\varrho)=\frac{2}{15} \frac{\varrho_1^2}{\varrho_0^2}
-\frac{4}{21} \frac{\varrho_2}{\varrho_0},\quad
\mathcal{R}_3(\varrho)=
-\frac{2}{15} \frac{\varrho_1^3}{\varrho_0^3}
+\frac{34}{105} \frac{\varrho_1\varrho_2}{\varrho_0^2}
-\frac{2}{9} \frac{\varrho_3}{\varrho_0},
\]
we have:
\begin{thrm}
\label{thm:Deltarho}
When expressed in terms of the moments of the measure $\varrho$, 
the stable partition function is given by
\begin{align*}
 \mathcal{Z}_V^{stable}
&:=\exp\Big(\sum_{g=2}^\infty V^{2-2g} F_g(\varrho)\Big)
= \exp\Big( -\frac{(2\lambda)^4}{V^2} \Delta_{\varrho}+F_2(\rho)\Big) 1,\qquad
\end{align*}
where 
\begin{align}
F_2&= \frac{(2\lambda)^4}{V^2}
\Big(
-\frac{21 \varrho_1^3 }{160\varrho_0^5}
+\frac{29}{128} \frac{\varrho_1\varrho_2}{\varrho_0^4}
-\frac{35}{384} \frac{\varrho_3}{\varrho_0^3}
\Big),
\\
\Delta_\varrho &:=
-\Big(-\frac{6\varrho_1^3}{5\varrho_0^3} + \frac{111
  \varrho_1\varrho_2}{70 \varrho_0^2} -
\frac{\varrho_3}{2\varrho_0}\Big)
\frac{\partial^2}{\partial \varrho_0^2}
-\Big(
\frac{2\varrho_1^3}{\varrho_0^4} - \frac{1097 \varrho_1\varrho_2}{280
  \varrho_0^3} +\frac{41 \varrho_3}{24 \varrho_0^2}\Big)
\frac{\partial}{\partial \varrho_0}
\nonumber
\\
&
-\sum_{k=1}^\infty (3+2k)\Big(
\Big(-\frac{2\varrho_1^2}{5\varrho_0^3} 
+ \frac{2 \varrho_2}{7\varrho_0^2}\Big)\varrho_{k+1}
-\frac{3\mathcal{R}_{k+2}(\varrho)\varrho_{1}}{2\varrho_0}
+\frac{3\mathcal{R}_{k+3}(\varrho)}{2}
\Big)
\frac{\partial^2}{\partial \varrho_k \partial \varrho_0}
\nonumber
\\
&
+\sum_{k,l=1}^\infty (3{+}2k)(3{+}2l)\Big(
\frac{ \varrho_1 \varrho_{l+1}\varrho_{k+1}}{30\varrho_0^2}
+\frac{\varrho_{k+1}\mathcal{R}_{l+2}(\varrho)}{4\varrho_0}
+ \frac{\varrho_{l+1}\mathcal{R}_{k+2}(\varrho)}{4\varrho_0}
-\frac{\mathcal{R}_{k+l+3}(\varrho)}{4}
\Big)
\frac{\partial^2}{\partial \varrho_k\partial \varrho_l}
\nonumber
\\
&
-\sum_{k=1}^\infty (3+2k)\Big(
\Big(\frac{19 \varrho_1^2}{60 \varrho_0^4} 
- \frac{25 \varrho_2}{84 \varrho_0^3}\Big)\varrho_{k+1}
+\frac{\varrho_{1}\mathcal{R}_{k+2}(\varrho)}{16\varrho_0^2}
-\frac{\mathcal{R}_{k+3}(\varrho)}{16 \varrho_0}
\nonumber
\\
&\qquad\qquad 
- \frac{(5+2k)\varrho_1\varrho_{k+2}}{30\varrho_0^3}
-\frac{(5+2k)\mathcal{R}_{k+3}(\varrho)}{2\varrho_0}
\Big)\frac{\partial}{\partial \varrho_k}
\end{align}
and $\mathcal{R}_m(\varrho)$ given by \eqref{Rmrho}.
\end{thrm}
\noindent
The $F_g(\varrho)$ are recursively extracted from 
$\mathcal{Z}_g(\varrho):=\frac{1}{(g-1)!} (-\Delta_\varrho+F_2(\varrho))^{g-1} 1$ by
\begin{align*}
F_g(\varrho)&=\textstyle \mathcal{Z}_g(\varrho)
- \frac{1}{(g-1)!}\sum_{k=2}^{g-1} B_{g-1,k}\big(
\big\{h! F_{h+1}(\varrho)\}_{h=1}^{g-k}\big)\\
&= \textstyle \mathcal{Z}_g(\varrho)- \frac{1}{(g-1)!}
\sum_{k=2}^{g-1} (-1)^{k-1}(k-1)!B_{g-1,k}\big(
\big\{h! \mathcal{Z}_{h+1}(\varrho)\}_{h=1}^{g-k}\big).
\end{align*}

These free energies are listed in different conventions in the
literature. As mentioned before, for any spectral dimension $\D<8$  the 
planar sector is renormalised such that \sref{Theorem}{Thm:Kontsevich}
holds. Therefore, \sref{Theorem}{thm:Deltarho} produces the intersection 
numbers via the differential operator $\Delta_\varrho$. 

The formula can easily be implemented in computer
algebra and quickly computes the 
free energies $F_g(t)$ to moderately large $g$. For convenience we list
\begin{align*}
F_3&=
\frac{1225}{144}{\cdot}\frac{t_2^6}{6!(-t_1)^{10}} 
+ \frac{193}{288} {\cdot} \frac{t_2^4 t_3}{4! (-t_1)^9} 
+ \frac{205}{3456} {\cdot} \frac{t_2^2 t_3^2}{2!2! (-t_1)^8} 
+ \frac{53}{1152} {\cdot} \frac{t_2^3 t_4}{3! (-t_1)^8} \\
&+ \frac{583}{96768} {\cdot} \frac{t_3^3}{3! (-t_1)^7} 
+ \frac{1121}{241920} {\cdot} \frac{t_2 t_3 t_4}{(-t_1)^7} 
+ \frac{17}{5760} {\cdot} \frac{t_2^2 t_5}{2! (-t_1)^7}
+ \frac{607}{1451520} {\cdot} \frac{t_4^2}{2! (-t_1)^6} \\
&+ \frac{503}{1451520} {\cdot} \frac{t_3 t_5}{(-t_1)^6} 
+ \frac{77}{414720} {\cdot} \frac{t_2 t_6}{(-t_1)^6} 
+ \frac{1}{82944} {\cdot} \frac{t_7}{(-t_1)^5}
\end{align*}
with $t_{i+1}=-(2i+1)!!\varrho_i$,
which is already given in \cite[eq.~(5.30)]{Itzykson:1992ya}.

\sref{Theorem}{thm:Deltarho} seems to be closely related with 
$\exp(\sum_{g\geq 0} F_g)= \exp(\hat{W})1$ proved by Alexandrov
\cite{Alexandrov:2010bn}, where
$\hat{W}:=\frac{2}{3}\sum_{k=1}^\infty (k+\frac{1}{2}) t_k
\hat{L}_{k-1}$ involves the generators $\hat{L}_n$ of the Virasoro
algebra. Including  $V$ and moving
$\exp(V^2 F_0+F_1)$ to the other side, our $\Delta_\varrho$ is in principle
obtained via Baker-Campbell-Hausdorff formula from Alexandrov's
equation.

\subsection{Deformed Virasoro Algebra}\label{Sec:CubicVirasoro}
The Kontsevich model without any renormalisation constants has some global constraints $L_n\Z=0$, where
$L_n$ is a differential operator with $n\in\N$. These operators
form a Witt algebra
\begin{align*}
 [L_n,L_m]=(n-m)L_{n+m},
\end{align*}
an infinite-dimensional Lie algebra. The differential operator
can be found by reducing the partition function to an integral over eigenvalues $x_i$ of the 
Hermitian matrices. The partition function is then unchanged under the diffeomorphism
$x_i\mapsto x_i^{n+1}\frac{d}{dx_i}$ which gives 
exactly the generators $L_n$ \cite{Makeenko:1991ec}. 

Since the renormalisation constants change the partition function, and are in general 
divergent in the limit $\Lambda^2\to\infty$, the Virasoro constraints are affected. Furthermore, the 
differential operator depends explicitly on $c(\lambda)$ which obeys different implicit equations for different spectral dimensions.

To find the right Virasoro constraints for any $\D$
we return to \eqref{ident-Z}, but instead of applying the inverse of
$\hat{K}_z$ we directly take the residue
\begin{align*}
\tilde{L}_n:= \Res_{z\to 0} \Big[z^{3+2n}
\Big(\frac{2V^2}{(2\lambda)^4} \hat{K}_z \hat{\mathrm{A}}^\dag_z + 
(\hat{\mathrm{A}}^\dag_z )^2 +\frac{1}{\varrho_0 z^4} \frac{\partial 
\hat{\mathrm{A}}^\dag_z }{\partial z} 
+\frac{V^2}{4(2\lambda)^4 z^4} \Big)dz\Big].
\nonumber
\end{align*}
By construction, $\tilde{L}_n \mathcal{Z}^{np}_V=0$. Recall that in
the Kontsevich model $L_n$ annihilates the \emph{full}
partition function.
However, for $\D>0$ the free energy $F_0$ does not exist 
such that the constraints can annihilate only the non-planar part of the partition function.
As explained below, these $\tilde{L}_n$ do not satisfy
the commutation relations of the Virasoro algebra exactly. 
An explicit expression is obtained from 
\begin{align*}
\hat{K}_z \hat{\mathrm{A}}^{\dag}_z &=
\sum_{l= 0}^{\infty} \sum_{j=0}^{l} \frac{(3+2l)\varrho_{l-j} }{z^{4+2j}}
\frac{\partial}{\partial \varrho_l},
\\
\frac{1}{\varrho_0 z^4} \frac{\partial}{\partial z}
\hat{\mathrm{A}}^{\dag}_z
&=
\sum_{l= 0}^\infty \Big(\frac{3(3+2l) \varrho_{l+1}}{
\varrho_0^2 z^8}-\frac{(3+2l)(5+2l)}{\varrho_0 z^{10+2l}}\Big)
\frac{\partial}{\partial \varrho_l},
\\
\hat{\mathrm{A}}^{\dag}_z\hat{\mathrm{A}}^{\dag}_z 
&=\sum_{k= 0}^{\infty} \Big(\frac{(5+2k) \varrho_{k+2}}{
\varrho_0 z^3}-\frac{5+2k}{z^{7+2k}}\Big)\frac{(3+2k)}{
\varrho_0 z^3}
\frac{\partial}{\partial \varrho_k}
\\
&+
\sum_{k= 0}^{\infty} \Big({-}\frac{3\varrho_{1}}{
\varrho_0 z^3}+\frac{3}{z^{5}}\Big)
\Big(\frac{(3{+}2k) \varrho_{k+1}}{
\varrho_0^2 z^3}\Big)
\frac{\partial}{\partial \varrho_k}
+
\sum_{l,k= 0}^{\infty} \frac{(3{+}2l) (3{+}2k) \varrho_{k+1} \varrho_{l+1}}{
\varrho_0^2 z^6}
\frac{\partial^2}{\partial \varrho_l \partial \varrho_k}
\\
& -
\sum_{l,k= 0}^{\infty} \frac{2(3+2l)(3+2k) \varrho_{l+1}}{
\varrho_0 z^{8+2k}}
\frac{\partial^2}{\partial \varrho_l \partial \varrho_k}
+ 
\sum_{l,k= 0}^{\infty}
\frac{(3+2l)(3+2k)}{z^{10+2l+2k}}
\frac{\partial^2}{\partial \varrho_l \partial \varrho_k}.
\end{align*}
Evaluating the residues and defining $A=\frac{(2\lambda)^4}{4V^2}$ and rescaling
$\hat{L}_n:=A \tilde{L}_n$ gives 
\begin{align*}
\hat{L}_0&=\frac{1}{16}+ \frac{1}{2}\sum_{l=0}^\infty (3+2l)\varrho_l 
\frac{\partial} {\partial \varrho_l},
\\
\hat{L}_1&=\frac{1}{2}\sum_{l=0}^\infty 
(5+2l)\varrho_{l}\frac{\partial}{\partial \varrho_{l+1}}
+A\Big(\sum_{k=0}^\infty \frac{(3+2k)}{\varrho_0^2}
\varrho_{k+1}\frac{\partial}{\partial \varrho_k}
- \frac{3\varrho_1}{\varrho_0^2}\Big)
\sum_{l=0}^\infty (3+2l)\varrho_{l+1}
\frac{\partial}{\partial \varrho_l}
\nonumber
\\
& \hspace*{-0.5cm} \text{and for $n\geq 2$: } \nonumber
\\
\nonumber
\hat{L}_{n}&=\frac{1}{2}
\sum_{l=0}^\infty 
(3{+}2n{+}2l)\varrho_{l}\frac{\partial}{\partial \varrho_{n+l}}
+A
 \delta_{n,2} \sum_{l=0}^\infty \frac{6(3+2l) \varrho_{l+1}}{
\varrho_0^2}\frac{\partial}{\partial \varrho_{l}}
-A\frac{2(2n{-}3)(2n{-}1)}{\varrho_0}
\frac{\partial}{\partial \varrho_{n-3}}
\nonumber
\\
&+A 
\sum_{l= 0}^{n-3}(3+2l)(2n-2l-3)
\frac{\partial^2}{\partial \varrho_l \partial \varrho_{n-3-l}}
-A\sum_{l= 0}^{\infty} \frac{2(3+2l)(2n-1) \varrho_{l+1}}{
\varrho_0}
\frac{\partial^2}{\partial \varrho_{n-2} \partial \varrho_l}.
\nonumber
\end{align*}
To write it in a more compact way, it is convenient to introduce the
differential operator
\begin{align}\label{roh-1}
 \hat{D}:=
 \sum_{l=0}^\infty \frac{(3+2l)\varrho_{l+1}}{\varrho_0}
\frac{\partial}{\partial \varrho_l}\;.
\end{align}
Note that $\frac{\partial}{\partial \varrho_{l}}
\hat{D}\neq\hat{D}\frac{\partial}
{\partial \varrho_{l}}$. The result is:
\begin{lemma} \label{LemmaVira}
The nonplanar partition function $\mathcal{Z}_V^{np}
:= \exp\Big(\sum_{g=1}^\infty V^{2-2g} F_g\Big)$ satisfies 
the constraints $\hat{L}_n\mathcal{Z}^{np}_V=0$ for all $n\in
\mathbb{N}$, where 
\begin{align}\label{Virasoro}
\hat{L}_0&=\frac{1}{16}+ \frac{1}{2}\sum_{l=0}^\infty (3+2l)\varrho_l 
\frac{\partial} {\partial \varrho_l},
\\
\hat{L}_1&=\frac{1}{2}\sum_{l=0}^\infty 
(5+2l)\varrho_{l}\frac{\partial}{\partial \varrho_{l+1}}
+A\hat{D}^2
\nonumber
\\
& \hspace*{-0.5cm} \text{and for $n\geq 2$: } \nonumber
\\
\nonumber
\hat{L}_{n}&=\frac{1}{2}
\sum_{l=0}^\infty 
(3{+}2n{+}2l)\varrho_{l}\frac{\partial}{\partial \varrho_{n+l}}
+A 
\sum_{l= 0}^{n-3}(3+2l)(2n-2l-3)
\frac{\partial^2}{\partial \varrho_l 
\partial \varrho_{n-3-l}}
\nonumber
\\
&-2A(2n-1)\frac{\partial}
{\partial \varrho_{n-2}}\hat{D},
\nonumber
\end{align}
where $\hat{D}$ is the differential operator defined by \eqref{roh-1} 
and $A=\frac{(2\lambda)^4}{4V^2}$.
\end{lemma}
\noindent
With the commutation rules
\begin{align*}
\Big[\hat{D},\varrho_l\Big]
&= \frac{3+2l}{\varrho_0} \varrho_{l+1}\;, & l\geq 0\;,
\\
\Big[\hat{D},\frac{\partial}{\partial
  \varrho_l} \Big]
&= -\frac{1+2l}{\varrho_0} \frac{\partial}{\partial \varrho_{l-1}}
& l\geq 1\;, && 
\Big[\hat{D},\frac{\partial}{\partial
  \varrho_0} \Big]
&= \frac{1}{\varrho_0} \hat{D}\;
\nonumber,
\end{align*}
we end up after long but straightforward computation:
\begin{lemma}
 The generators $\hat{L}_n$ of \sref{Lemma}{LemmaVira} obey the commutation relation
 \begin{align*}
  [\hat{L}_0,\hat{L}_n]=-n\hat{L}_n
 \end{align*}
and for any $m,n\geq 1$,
\begin{align*}
[\hat{L}_m,\hat{L}_n]
&= (m-n) \hat{L}_{m+n}
-4A (m+1)B_{n-2}\hat{L}_{m-1}
+4A (n+1)B_{m-2}\hat{L}_{n-1}
\\
&- 4A \delta_{m,1} \frac{n(n+1)}{\varrho_0^2} \hat{L}_{n-2}
+ 4A \delta_{n,1} \frac{m(m+1)}{\varrho_0^2} \hat{L}_{m-2}
\nonumber
\end{align*}
where 
\[
B_m:= (2m+3)\frac{\partial}{\partial \varrho_{m}}
\frac{1}{\varrho_0} \quad\text{for }m\geq 0\;,\qquad
B_{-1}:= -\frac{1}{2}\Big\{\hat{D},
\frac{1}{\varrho_0} \Big\} \;.
\]
\end{lemma}
\begin{rmk}
  The differential operator $\hat{D}$ has its origin in the implicit
  definition of the constant $c$ (see \sref{Corollary}{Coro:1P}) and the dependence
  of $\varrho_l$ on $c$. Since the expression
 \begin{align*}
  \varrho_{-1}:=-8\lambda^2\int_0^{\Lambda^2} dx\;\frac{r(x)}{\sqrt{(1+2e(x))^2+c}}, 
\qquad  (\varrho_{-1}=c\;\text{for}\;\D<2)
 \end{align*}
diverges for any $\D\geq2$ in the limit $\Lambda^2\to\infty$, it was necessary
to construct the analogue of the 
derivative $\frac{\partial}{\partial c}$ 
through the differential  operator $\hat{D}$. 
Replacing the differential operator by 
$\hat{D}\mapsto \frac{\partial}{\partial \varrho_{-1}}$ and the generators by
 \begin{align*}
  \hat{L}_n \mapsto \hat{L}_n+\frac{1}{2}\varrho_{-1}\frac{\partial}{\partial \varrho_{-1}}=L_n,
 \end{align*}
 recovers the original undeformed Virasoro algebra. As explained
 above, $\varrho_{-1}$ and, consequently, the standard Virasoro
 generators do not exist in $\D\geq 2$.  The renormalisation 
 necessary for $\D\geq2$ alters the definition of $c$ and prevents the
 construction of $\hat{L}_{-1}$ and $F_0$ which in $\D=0$ depend on $\varrho_{-1}$ (see \sref{Remark}{Rmk:F0}).
 Higher topologies ($\chi\leq 0$) are not affected because any 
 explicit $c$-dependence drops out.
\end{rmk}

\section{Over-Renormalisation}\label{Sec:CubicOverRenom}
In this short section which has not yet appeared in the literature, we look upon the question whether over-renormalisation would produce any
problems for the cubic model. We will understand over-renormalisation by renormalising a model of dimension $D$ 
as a model of dimension $\tilde{D}>D$ which involves more renormalisation than necessary. 
Of course, more conditions will fix these additional renormalisation constants. However, on
the level of Feynman graphs over-renormalisation would creates a big number of additional 
but finite counterterms due to Zimmermann's
forest formula \cite{Zimmermann:1969jj}. 
A very natural question is whether this over-subtraction of the additional counterterms converges or produces new problems. 
For each Feynman diagram, it is clear that the diagram is still finite after over-subtraction, but whether a global problem 
appears or not, is unclear. This question can only be answered if a model is resummable after renormalisation. 
Therefore, we can tackle this question perfectly for matrix field theories.

We saw that any correlation function is built recursively by the 1-point 
function with \sref{Theorem}{finaltheorem} and \sref{Theorem}{thm:G-residue}.
It is therefore sufficient to show the convergence of the 1-point function.
The general solution
is given by \sref{Corollary}{Coro:1P}. Taking the renormalisation with respect to a 
higher dimensional model, the integrals are clearly convergent in \sref{Corollary}{Coro:1P}.
This directly implies that the large number of over-subtracted terms can not destroy resummablility. 

To be more precise, we selected two examples: 

\begin{exm}($D=2$ Moyal space with renormalisation of $D=4$)\label{Ex:D24}
\\
 Linear eigenvalues $e(x)=x$ and for $D=2$ with growing multiplicity of the form $r(x)=1$
 implies $\varrho(Y)=\frac{2\lambda^2}{\sqrt{Y}}$. Then \sref{Corollary}{Coro:1P} gives for $D=4$ for the shifted 
 1-point function after integration and simplification
 \begin{align*}
  \tilde{W}(X)=&1+\sqrt{X+c}-\sqrt{1+c}-2\lambda^2\bigg\{
  \log\bigg(1+\frac{1}
  {\sqrt{1+c}}\bigg)\\
 &-\frac{1}{\sqrt{X}}\log\bigg(\frac{(\sqrt{X}+\sqrt{X+c})(\sqrt{X}+1)}
  {\sqrt{X}\sqrt{1+c}+\sqrt{X+c}}\bigg)\bigg\}\\
  1=&\sqrt{1+c}+\lambda^2\bigg\{\frac{1-2\sqrt{1+c}}{\sqrt{1+c}}+2\sqrt{1+c}\log\bigg(1+\frac{1}
  {\sqrt{1+c}}\bigg)\bigg\}.
 \end{align*}
The function $c(\lambda)$ is uniquely invertible for $|\lambda|<\lambda_c\approx 0.6886$.
Expanding $c$ in $\lambda$ by \eqref{eq:cexpansion} gives $c(\lambda)=-\lambda^22(2\log 2-1)
+\lambda^4(5-16\log 2+12(\log 2)^2)
+\lambda^6 (\frac{21}{2}-
46\log 2+66(\log2)^2-32 (\log2)^3)+\mathcal{O}(\lambda^8)$.
Inserting in \eqref{eq:G1} gives the first orders of the 1-point function 
\begin{align*}
 G^0(x)=\lambda \frac{\log(1+x)-x}{1+2x}
 -\lambda^3 \frac{x^2(3+4x)(2\log2-1)^2}{(1+2x)^3}+\mathcal{O}(\lambda^5)
\end{align*}
which is confirmed in \sref{App.}{App:PertCubic} by Feynman graph calculations.
\end{exm}

\begin{exm}($D=2$ Moyal space with renormalisation of $D=6$)\label{Ex:D26}
\\
 Again, linear eigenvalues $e(x)=x$ and for $D=2$ with growing multiplicity of the form $r(x)=1$
 implies $\varrho(Y)=\frac{2\lambda^2}{\sqrt{Y}}$. Then \sref{Corollary}{Coro:1P} gives with $D=6$ for the shifted 
 1-point function after integration and simplification
 \begin{align*}
  \tilde{W}(X)=&\sqrt{X+c}\sqrt{1+c}-c+2\lambda^2\bigg\{ \sqrt{1+c}-\sqrt{X+c}-\frac{1}{2}+\frac{\sqrt{X+c}}{2\sqrt{1+c}}
\\&+ (\sqrt{1+c}(\sqrt{X+c}-\sqrt{1+c})-1) \log\bigg(1+\frac{1}{\sqrt{1+c}}\bigg)\\
 &+\frac{1}{\sqrt{X}}\log\bigg(\frac{(\sqrt{X}+\sqrt{X+c})(1+\sqrt{X})}
  {\sqrt{X}\sqrt{1+c}+\sqrt{X+c}}\bigg)\bigg\}\\
  -c=&\lambda^2\bigg\{\frac{1}{2(1+c)}+3-6\sqrt{1+c}+2(3c+2)\log\bigg(1+\frac{1}
  {\sqrt{1+c}}\bigg)\bigg\}.
 \end{align*}
 The function $c(\lambda)$ is uniquely invertible for $|\lambda|<\lambda_c\approx 0.8891$.
Expanding $c$ in $\lambda$ by \eqref{eq:cexpansion} gives $c(\lambda)=-\lambda^2(4\log2-\frac{5}{2})
+\lambda^4\frac{96(\log2)^2-132\log2+45}{4}
-\lambda^6 \frac{4608(\log2)^3-9600( \log2)^2+6672 \log 2-1545}{32}+\mathcal{O}(\lambda^8)$.
Inserting in \eqref{eq:G1} gives the first orders of the 1-point function 
\begin{align*}
 G^0(x)=\lambda \frac{2\log(1+x)-x(2-x)}{2(1+2x)}
 +\lambda^3 \frac{x^3(2+3x)(8\log2-5)^2}{4(1+2x)^3}+\mathcal{O}(\lambda^5)
\end{align*}
which is confirmed in \sref{App.}{App:PertCubic} by Feynman graph calculations.
\end{exm}
\noindent
Unexpectedly, the convergence radius on the $D=2$ Moyal space increases if the renormalisation corresponding to a higher
dimensional model is taken. This counterintuitive behaviour is even more fascinating if we notice 
that the coupling constant is additionally renormalised by a finite factor 
for the 6-dimensional renormalisation. This finite factor 
(infinite for the original $D=6$ Moyal space) shifts the convergence radius to a similar value in 
comparison to the $D=2$ and $D=4$ renormalisation.

\section{Summary}
The construction of the renormalised $\Phi^3$ matrix field theory model 
is complete. 
We established an algorithm to compute any correlation function 
$G^{(g)}(x_1^1,..,x^1_{N_1}|..|x_1^b,..,x^b_{N_b})$ for all genus $g$ 
with spectral dimension $\D<8$:

\begin{enumerate}\itemsep 0pt
\item Compute the free energy $F_g(\varrho)$ for $g\geq2$ via \sref{Theorem}{thm:Deltarho}.
It encodes the $p(3g-3)$ intersection numbers 
of $\psi$-classes on the moduli space of complex curves of genus $g$. 
Take $F_1=-\frac{1}{24} \log \varrho_0$ for $g=1$. 

\item Apply to $F_g(\varrho)$ according to 
\sref{Proposition}{prop:F} and \sref{Theorem}{finaltheorem} 
the boundary creation operators $
\hat{\mathrm{A}}^{\dag g}_{z_1,\dots,z_b}\circ 
\dots \hat{\mathrm{A}}^{\dag g}_{z_1,z_2} \circ 
\hat{\mathrm{A}}^{\dag g}_{z_1}$ defined in \sref{Definition}{def1}.
Multiply by $(2\lambda)^{4g+3b-4}$ to obtain
$\G_g(z_1,...,z_b)$ for $g>0$. Take \sref{Corollary}{Coro:1P} as well as 
the second part of
\sref{Theorem}{finaltheorem} for $g=0$.

\item Transform the variables by $z_\beta(x^\beta)= (
(1+2e(x^\beta))^2+c)^{1/2}$, where $c$ is given 
by the implicit equation of \sref{Corollary}{Coro:1P}, 
to obtain $G^{(g)}(x^1|...|x^b)=\G_g(z_1(x^1),...,z_b(x^b))$.

\item Pass to 
$G^{(g)}(x_1^1,..,x^1_{N_1}|..|x_1^b,..,x^b_{N_b})$ via difference quotients by
\sref{Proposition}{Prop:CubisRecur} which holds in the $\mN,V$-limit with $F_{p}\mapsto
\frac{1}{2}+e(x)$.

\end{enumerate}
\noindent
We remark that, in spite of the relation to the integrable Kontsevich
model, this $\Phi^3$-model provides a fascinating toy model for
a possible QFT in higher dimensions 
which shows many facets of renormalisation. Our
exact formulae can be expanded about $\lambda=0$ via \sref{Corollary}{Coro:1P} and
agree with the usual perturbative renormalisation which needs
Zimmermann's forest formula \cite{Zimmermann:1969jj}. Also note that 
at fixed genus $g$ one expects $\mathcal{O}(n!)$ graphs with $n$
vertices so that a convergent summation at fixed $g$ cannot be
expected a priori. To our knowledge, this is the first bosonic model which is just-renormalisable 
and still resummable for each genus.
All correlation functions have a finite radius of convergent
in $\lambda$. Taking the renormalisation procedure of a higher dimensional model 
for lower dimensions creates no problem, even though infinitely many finite counterterms are subtracted 
additionally.
Moreover, for $6\leq \D<8$ the $\beta$-function of the
coupling constant is positive for real $\lambda$ (see \sref{Remark}{rmk:beta}), which in this
particular case possesses not the slightest problem for summation. Furthermore, the renormalon problem appears 
(see \sref{Remark}{rmk:renorm}) also for $6\leq \D<8$.


\chapter[Quartic Interaction]{Quartic Interaction: Grosse-Wulkenhaar Model}\label{ch:quartic}
We consider in this chapter the matrix field theory model with quartic interaction.
This type of interaction was also studied 
in the past with great interest, however from a completely different 
perspective than the cubic model. The motivation goes back to QFT.

In QFT, a scalar model with a quartic self-interaction is 
the first natural choice for a rigorous understanding of QFT in general. 
The potential is, whitout renormalisation, bounded from below such that the model has good convergence property.
However, even in 2D the scalar $\Phi^4$ model was for a long time not understood rigorously. It was achieved by
constructive QFT in a series of papers by Glimm and Jaffe \cite{Glimm:1968kh,Glimm:1968kh1,Glimm:1968kh2,Glimm:1968kh3}. 

Later, interest in QFT models on deformed spacetime arose which was motivated by a fundamental minimal 
scale in Nature also called Planck scale. An example of a deformation is the Moyal space described in \sref{Sec.}{Sec.Moyal}.
It was recognised that a model living on the $D$-dimensional 
Moyal space forces an UV/IR mixing problem \cite{Minwalla:1999px}.
This problem could be resolved by introducing a harmonic oscillator term \cite{Grosse:2003nw,Grosse:2004yu}.
Assuming the self-dual point $\Omega=1$ (see \eqref{eq:SelfDual}) for the for the harmonic oscillator term, it was proved 
\cite{Disertori:2006nq}
that the $\beta$-function vanishes $(\beta=0)$ to all orders in perturbation theory. Furthermore, 
the model was proved to be renormalisable to all orders 
in perturbation theory in $D=4$ \cite{Grosse:2003aj}. 

With the techniques described in \sref{Sec.}{Sec:SDE}, which were first used in \cite{Disertori:2006nq}, 
a nonlinear integral equation for the 2-point function was derived in \cite{Grosse:2009pa} 
and proved to have a solution in \cite{Grosse:2012uv}.
We will derive this solution as special case in \sref{Sec.}{Sec.4dSol} since the 4D Moyal space corresponds to a special 
case of the matrix field theory model. 

We will start in \sref{Sec.}{Sec.quartSD} with computing the SDE in its full generality. 
The distinction between $(N_1+..+N_b)$-point function for even or odd $N_i$ will be discussed. This distinction does not 
exist in the cubic model. In \sref{Sec.}{Sec.quartSolution} the solution of the 2-point function will be computed, 
where $\D=0$ (finite matrices) and the $D=4$ Moyal space are discussed as particular examples. Applying the solution 
in \sref{Sec.}{Sec.quartHO} to
the SDE creates a special form of the equation compared to the cubic model.
The explicit structure of the recursive equation for the planar $N$-point function $(b=1)$ is analysed in 
\sref{Sec.}{Sec.quartRec} which is \textit{nonlinear} in comparison to \sref{Proposition}{Prop:CubisExpl} for the cubic model.

\section{Schwinger-Dyson Equations}\label{Sec.quartSD}
We derive the SDEs in the way described in \sref{Sec.}{Sec:SDE}. 
In contrast to the cubic model, the Ward-Takahashi identity has to be applied in its full 
generality with the degenerate terms $W^1[J],W^2[J]$ appearing in 
\sref{Theorem}{Thm:Raimar}. 

Recall the action with a quartic interaction
\begin{align}\label{eq:actionquart}
  S[\Phi]&=V\bigg(\sum_{n,m=0}^\mN \!\!\!
\frac{H_{nm}}{2}
\Phi_{nm}\Phi_{mn}\!
+\frac{\lambda}{4} \!\!\!\!\!
\sum_{n,m,k,l=0}^\mN
\!\!\!\!\!
\Phi_{nm}
 \Phi_{mk}
\Phi_{kl}\Phi_{ln}\!\bigg),\\
H_{nm}&=E_n+E_m,\nonumber
\end{align}
where two eigenvalues $E_i,E_j$ are not necessarily distinct. The partition function is due to \eqref{eq:Part2} given by
\begin{align}
 \mathcal{Z}[J]=&\int \mathcal{D} \Phi \exp\left(-S[\Phi]+V\mathrm{Tr}(J\Phi) \right)\nonumber\\\label{eq:partionquart}
 =&K\,\exp\Big(-\frac{\lambda }{4V^3}\sum_{n,m,k,l=0}^\mN
 \frac{\partial^4}{\partial J_{nm}\partial 
 J_{mk}\partial J_{kl}\partial J_{ln}}
 \Big)\mathcal{Z}_{\mathrm{free}}[J],
\\\nonumber
\mathcal{Z}_{\mathrm{free}}[J]:=&
 \exp\bigg(V\sum_{n,m=0}^\mN\frac{J_{nm}J_{mn}}{2H_{nm}}\bigg),
\end{align}
and $K:=\int \mathcal{D}\Phi\exp\Big(-VZ\sum_{n,m=0}^\mN\frac{H_{nm}}{2}
\Phi_{nm}\Phi_{mn}\Big)=\prod_{n,m=0}^\mN\sqrt{\frac{2}{VZH_{nm}}}$.

\begin{rmk}\label{Rmk:quartRenorm}
 The action and the partition function are not renormalised in \eqref{eq:actionquart} and \eqref{eq:partionquart}. 
 The mass renormalisation changes $\mu^2\mapsto\mu^2_{bare}$, the coupling constant renormalisation $\lambda\mapsto
 \lambda_{bare}$ and the field renormalisation $\Phi\mapsto\sqrt{Z}\Phi$. Inserting in the action
 gives the substitution $E_n\to Z E_n+const$ and $\lambda \to Z^2\lambda_{bare}$ for all SDEs. 
 We will see later that the coupling constant does not need renormalisation. Thus, we can take $\lambda=\lambda_{bare}$. 
 We will avoid the factors of $Z$ for this section. Renormalisation of the quartic model
 will be discussed in greater detail in \sref{Sec.}{Sec.quartSolution}.
\end{rmk}
\noindent
The quartic model has no 1-point function which becomes directly clear from perturbative expansion. 
Therefore, we start with the SDE of the 2-point function which was already computed before.
\begin{prps}(\cite[eq. 3.4]{Grosse:2012uv})\label{Prop:Quart2P}
 The 2-point function of the quartic matrix field theory model satisfies
 \begin{align*}
  G_{|pq|}&=
\frac{1}{E_p+E_q}-\frac{\lambda}{E_p+E_q} 
\bigg\{  G_{|pq|}\Big(\frac{1}{V^2}G_{|p|p|}+\frac{1}{V}\sum_{n=0}^\mN G_{|pn|}\Big)
+ \frac{1}{V^4}G_{|p|p|pq|}\\
&+\frac{1}{V^3}\sum_{n=0}^\mN G_{|pn|pq|}
+\frac{1}{V^2}( G_{|pppq|}+G_{|pqpq|})
+\frac{1}{V}\sum_{n=0}^\mN \frac{G_{|pq|}- G_{|nq|}}{E_n-E_p} 
+\frac{1}{V^2} \frac{G_{|p|q|}- G_{|q|q|}}{E_q-E_p} 
\bigg\}\;.
 \end{align*}
\begin{proof}
 By definition, the 2-point function for $p,q$ with $E_p\neq E_q$ is
 \begin{align*}
  G_{|pq|}=&\frac{1}{V}\frac{\partial^2}{\partial J_{pq}\partial J_{qp}}\log\Z[J]\bigg\vert_{J=0}\\
  =&\frac{1}{H_{pq}}-\frac{\lambda}{H_{pq}\Z[0]}\sum_{n,m=0}^\mN
  \frac{\partial^4\Z[J]}{\partial J_{qp}\partial J_{pn}\partial J_{nm}\partial J_{mq}}\bigg\vert_{J=0}\\
  =&\frac{1}{H_{pq}}-\frac{\lambda}{H_{pq}\Z[0]V^3}
  \frac{\partial^2 (W^1_p[J]+W^2_p[J])\Z[J]}{\partial J_{qp}\partial J_{pq}}\bigg\vert_{J=0}\\
  &\qquad -\frac{\lambda}{H_{pq}\Z[0]V^2}
  \sum_{n,m=0}^\mN\frac{1}{E_m-E_p}
  \frac{\partial^2}{\partial J_{mq}\partial J_{qp}}\bigg(J_{mn}\frac{\partial}{\partial J_{pn}}-
  J_{np}\frac{\partial}{\partial J_{nm}}\bigg)\Z[J]\bigg\vert_{J=0}\\
  =&
\frac{1}{H_{pq}}-\frac{\lambda}{H_{pq}} 
\bigg\{  G_{|pq|}\Big(\frac{1}{V^2}G_{|p|p|}+\frac{1}{V}\sum_{n=0}^\mN G_{|pn|}\Big)
+ \frac{1}{V^4}G_{|p|p|pq|}+\frac{1}{V^3}\sum_{n=0}^\mN G_{|pn|pq|}\\
&
+ \frac{1}{V^2}( G_{|pppq|}+G_{|pqpq|})
+\frac{1}{V}\sum_{n=0}^\mN \frac{G_{|pq|}- G_{|nq|}}{E_n-E_p} 
+\frac{1}{V^2} \frac{G_{|p|q|}- G_{|q|q|}}{E_q-E_p}\bigg\},
 \end{align*}
where we used in the second line of equation \eqref{eq:SDeqGeneral} and in the third and fourth line of
\sref{Theorem}{Thm:Raimar}. The last two lines are achieved by acting with the second order 
derivative $\frac{\partial^2}{\partial J_{qp}\partial J_{pq}}$ on $\Z[J]$ or 
$W^1_p[J]$, and fixing $n$ or $m$ by acting with $\frac{\partial^2}{\partial J_{mq}\partial J_{qp}}$ on $J_{mn}$
or $J_{np}$, respectively.

Since the lhs is regular if $E_p=E_q$, the rhs has a well-defined limit due to 
the continuation of the correlation functions to differentiable functions.
\end{proof}
\end{prps}
\noindent
Expanding the SDE of \sref{Proposition}{Prop:Quart2P} in the genus yields at order $V^{-2g}$ 
\begin{align}\label{eq:2Pphi4}
 G^{(g)}_{|pq|}&=
\frac{\delta_{0,g}}{E_p+E_q}-\frac{\lambda}{E_p+E_q} 
\bigg\{  \sum_{h=0}^{g-1} G^{(h)}_{|pq|}G^{(g-1-h)}_{|p|p|}+\sum_{h=0}^{g} G^{(h)}_{|pq|}
\frac{1}{V}\sum_{n=0}^\mN G^{(g-h)}_{|pn|}
+ G^{(g-2)}_{|p|p|pq|}\\
&+\frac{1}{V}\sum_{n=0}^\mN G^{(g-1)}_{|pn|pq|}
+ G^{(g-1)}_{|pppq|}+G^{(g-1)}_{|pqpq|}
+\frac{1}{V}\sum_{n=0}^\mN \frac{G^{(g)}_{|pq|}- G^{(g)}_{|nq|}}{E_n-E_p} 
+ \frac{G^{(g-1)}_{|p|q|}- G^{(g-1)}_{|q|q|}}{E_q-E_p} \bigg\}.\nonumber
\end{align}
Notice that the planar 2-point function satisfies a nonlinear equation, whereas the 2-point function of genus $g>0$
of Euler characteristic $\chi=1-2g$ satisfies a linear equation with a inhomogeneity depending on correlation functions of 
Euler characteristic $\chi'>\chi$. 

The 2-point function of genus $g$ depends on 4-point functions of genus $g-1$. In general, correlation functions with 
boundary lengths larger than 2 are recursively expressed by correlation functions only of lengths 1 or 2. This recursive 
behavior was found in \cite{Grosse:2012uv} for the $N$-point 
function and the $(N+M)$-point function of genus $g$, where $N+M$ is even. 
Let us define the shorthand notation $G^{(g)}_{|\mathcal{I}|}:=G^{(g)}_{|I^1|..|I^b|}:
=G^{(g)}_{|p_1^1..p^1_{N_1}|..|p_1^b..p^b_{N_b}|}$ for $\mathcal{I}=\{I^1,..,I^b\}$ and $I^\beta=\{p_1^\beta,..,p_{N_\beta}^\beta\}$
for $\beta\in\{1,..,b\}$ and the cardinality $|\mathcal{I}|=|I^1|+..+|I^b|$. The recursive equation, in 
full generality, reads for any number of boundary components: 
\begin{prps}\label{Prop:quartRec}
Let $\mathcal{J}=\{J^2,..,J^b\}$, $J^\beta=\{p^\beta_1,..,p^\beta_{N_\beta}\}$ and $\beta\in\{2,..,b\}$.
 Then, the $(N_1+..+N_b)$-function with $b$ boundary components satisfies for $N_1\geq 3$ and $N=\sum_{i=1}^bN_i$ 
 even the recursive
 equation
 \begin{align*}
  G_{|p_1^1 p^1_2..p^1_{N_1}|\mathcal{J}|}=&-\frac{\lambda}{E_{p^1_2}-E_{p^1_{N_1}}}\bigg\{\frac{1}{V^2}
  \sum_{k=2}^{N_1}\frac{G_{|p^1_2p^1_3..p^1_k|p^1_{k+1}..p^1_{N}p^1_1|\mathcal{J}|}
 -G_{|p^1_1p^1_2p^1_3..p^1_{k-1}|p^1_{k}..p^1_{N_1}|\mathcal{J}|}}{E_{p^1_k}-E_{p^1_1}}\\
 +&\sum_{\beta=2}^{b}\sum_{k=1}^{N_\beta}\frac{G_{|p_1^\beta p_2^\beta..p_k^\beta p^1_2p^1_3..p^1_{N_1}
 p^1_1 p_{k+1}^\beta..p^\beta_{N_\beta}|\mathcal{J}\backslash \{J^\beta\}|}-
 G_{|p_1^\beta p_2^\beta..p_{k-1}^\beta p^1_1 p^1_2..p^1_{N_1}
  p_{k}^\beta..p^\beta_{N_\beta}|\mathcal{J}\backslash \{J^\beta\}|}}{E_{p^\beta_k}-E_{p^1_1}}\\
  +&\sum_{k=2}^{N_1}
  \sum_{\substack{\mathcal{I}\uplus \mathcal{I}'=\mathcal{J}}}\frac{G_{|p^1_{k+1}..p^1_{N_1}p^1_1|\mathcal{I}|} 
  G_{|p^1_2..p^1_k|\mathcal{I}'|}-
 G_{|p^1_{k}..p^1_{N_1}|\mathcal{I}|} G_{|p^1_1 p^1_2..p^1_{k-1}|\mathcal{I}'|}}{E_{p^1_k}-E_{p^1_1}}\bigg\}
 \end{align*}
 with $p^i_{N_i+j}\equiv p^i_j$. The last row sums only over $N_1-k+1+|\mathcal{I}|$ even.
 \begin{proof}
  Assume $p^i_j$ such that all $E_{p^i_j}$ are pairwise different. Set $a=p^1_1$, $d=p^1_2$ and $c=p^1_{N_1}$ to have 
  clear distinction between these and the remaining $p^i_j$. Define all derivatives besides the distinct ones by
  \begin{align*}
   \hat{D}_{dc}=\frac{\partial^{N_1+..+N_b-2}}{\partial J_{dp^1_3}\partial J_{p^1_3p^1_4}..
   \partial J_{p^1_{N_1-1}c} \partial J_{p^2_1p^2_2}\partial J_{p^2_2p^2_3}..
   \partial J_{p^2_{N_2}p^2_1} ..\partial J_{p^b_1p^b_2}..
   \partial J_{p^b_{N_b}p^b_1}}.
  \end{align*}
Bringing the global denominator of the proposition to the other side yields by definition of the correlation function
\begin{align*}
 &(E_{d}-E_{c})G_{|p_1^1 p^1_2..p^1_{N_1}|\mathcal{J}|}\\
 =&(E_{d}-E_{c})V^{b-2} \hat{D}_{dc}\frac{\partial^2}{\partial J_{ca}\partial J_{ad}}
 \log \Z[J]\big\vert_{J=0}\\
 =&V^{b-2} \hat{D}_{dc}\frac{\partial^2}{\partial J_{ca}\partial J_{ad}}(H_{ad}-H_{ac})\log \Z[J]\big\vert_{J=0}\\
 =&K\hat{D}_{dc}V^{b-1}\left(\frac{\partial}{\partial J_{ca}} \frac{\exp(-V S_{int}(\frac{1}{V}\partial))J_{da}}
 {\Z[J]}-
 \frac{\partial}{\partial J_{ad}} \frac{\exp(-V S_{int}(\frac{1}{V}\partial))J_{ac}}{\Z[J]}\right)\Z_{free}[J]\big\vert_{J=0}\\
 =&-V^{b-4}\lambda \hat{D}_{dc}\sum_{n,m}\left(\frac{\partial}{\partial J_{ca}}
 \frac{\frac{\partial^3}{\partial J_{an}\partial J_{nm}\partial J_{md}}}{\Z[J]}-\frac{\partial}{\partial J_{ad}}
\frac{ \frac{\partial^3}{\partial J_{cm}\partial J_{mn}\partial J_{ma}}}{\Z[J]}\right)\Z[J]\big\vert_{J=0},
\end{align*}
where equation \eqref{eq:SDeqGeneral} was applied. 
For $E_m=E_a$, the bracket vanishes for regular and non-regular terms since $\frac{\partial}{\partial J_{ca}}$ and 
$\frac{\partial}{\partial J_{ad}}$ do not act on $\frac{1}{\Z[J]}$ because it gives 0 after taking $J=0$ (no cycle in $a$).
Therefore, we can assume $E_m\neq E_a$ and apply the Ward-Takahasi identity of \sref{Proposition}{Prop:WardId} to have
\begin{align*}
 =&-\lambda \hat{D}_{dc}V^{b-3}\sum_{\substack{n,m \\ m\neq a} }\bigg(\frac{\partial}{\partial J_{ca}}
 \frac{\frac{\partial}{\partial J_{md}}}{\Z[J](E_m-E_a)}
 \big(J_{mn}\frac{\partial}{\partial J_{an}}-
 J_{na}\frac{\partial}{\partial J_{nm}}\big)\\
 &\qquad +\frac{\partial}{\partial J_{ad}}\frac{
 \frac{\partial}{\partial J_{cm}}}{\Z[J](E_m-E_a)}\big(J_{an}\frac{\partial}{\partial J_{mn}}-
 J_{nm}\frac{\partial}{\partial J_{na}}\big)\bigg)\Z[J]\big\vert_{J=0}.
\end{align*}
In the first line, $\frac{\partial}{\partial J_{ca}}$ has to act for the second term on $J_{na}$, where $n=c$. 
In the second line, $\frac{\partial}{\partial J_{ad}}$ has to act for the first term on $J_{an}$, where $n=d$. 
Both resulting terms cancel for any $m$. If 
$\frac{\partial}{\partial J_{md}}$ acts on $J_{mn}$ in the first line, and 
$\frac{\partial}{\partial J_{cm}}$ on $J_{nm}$ in the second line, the resulting terms cancel as well. 
We end up with the surviving terms
\begin{align*}
 =-\lambda \hat{D}_{dc}V^{b-3}\sum_{\substack{n,m \\ m\neq a} }\frac{1}{E_m-E_a}
 \left(\frac{\partial}{\partial J_{ca}}J_{mn}\frac{
 \frac{\partial^2}{\partial J_{an}\partial J_{md}}}{\Z[J]}
 -\frac{\partial}{\partial J_{ad}}J_{nm}\frac{
 \frac{\partial^2}{\partial J_{cm}\partial J_{na}}}{\Z[J]}\right)\Z[J]\big\vert_{J=0}.
\end{align*}
 Applying next the identity $ \frac{\partial_x \partial_yf}{f}=\partial_x\partial_y\log(f)+\partial_x\log(f) \partial_y\log(f)
$ yield
\begin{align*}
 &=-\lambda \hat{D}_{dc}V^{b-3}\sum_{\substack{n,m \\ m\neq a} }\frac{1}{E_m-E_a}
 \bigg\{\frac{\partial}{\partial J_{ca}}J_{mn}
 \frac{\partial^2}{\partial J_{an}\partial J_{md}}\log\Z[J]-\frac{\partial}{\partial J_{ad}}J_{nm}
 \frac{\partial^2}{\partial J_{cm}\partial J_{na}}\log\Z[J]\\
 &\qquad\qquad\qquad\qquad\qquad\qquad +
 \frac{\partial}{\partial J_{ca}}J_{mn}
 \bigg(\frac{\partial}{\partial J_{an}}\log\Z[J]\bigg)\bigg(\frac{\partial}{\partial J_{md}}\log\Z[J]\bigg)\\
& \qquad\qquad\qquad\qquad\qquad\qquad-\frac{\partial}{\partial J_{ad}}J_{nm}
 \bigg(\frac{\partial}{\partial J_{cm}}\log\Z[J]\bigg)
 \bigg(\frac{\partial}{\partial J_{na}}\log\Z[J]\bigg)\bigg\}\bigg\vert_{J=0}.
\end{align*}
  The $n,m$ are fixed by a derivative acting on $J_{mn}$ (or $J_{nm}$). In the first line, there are two possibilities, 
  either 
  a derivative of the form $\frac{\partial}{\partial J_{p^1_k p^1_{k+1}}}$ fixes the $n,m$, which produces 
  separated cycles, or a derivative of the form $\frac{\partial}{\partial J_{p^\beta_k p^\beta_{k+1}}}$ with 
  $\beta>1$ the $n,m$, which merges the first cycle with the $\beta^{\text{th}}$-cycle. In the last two lines,
  it is only possible that 
  a derivative of the form $\frac{\partial}{\partial J_{p^1_k}\partial J_{p^1_{k+1}}}$ fixes the $n,m$, otherwise $J=0$ 
  vanishes. Acting with the remaining derivatives of $\hat{D}_{dc}$ on the 
  product of the logarithms by considering 
  the Leibniz rule leads to the assertion for pairwise different $E_{p^i_j}$. 
  
  The expression stays true for coinciding $E_{p^i_j}$ since the lhs is regular which induces a well-defined limit of the rhs
  by continuation to differentiable functions. From perturbative considerations it is clear that the number of external legs 
  is necessarily even, which corresponds to even $N_1-k+1+|\mathcal{I}|$.
  
 \end{proof}
\end{prps}
\noindent
After genus expansion, \sref{Proposition}{Prop:quartRec} gives at order $V^{-2g}$
\begin{align}\nonumber
  G^{(g)}_{|p_1^1 p^1_2..p^1_{N_1}|\mathcal{J}|}=&-\frac{\lambda}{E_{p^1_2}-E_{p^1_{N_1}}}\bigg\{
  \sum_{k=2}^{N_1}\frac{G^{(g-1)}_{|p^1_2p^1_3..p^1_k|p^1_{k+1}..p^1_{N_1}p^1_1|\mathcal{J}|}
 -G^{(g-1)}_{|p^1_1p^1_2p^1_3..p^1_{k-1}|p^1_{k}..p^1_{N_1}|\mathcal{J}|}}{E_{p^1_k}-E_{p^1_1}}\\\label{eq:quartRec}
 +&\sum_{\beta=2}^{b}\sum_{k=1}^{N_\beta}\frac{G^{(g)}_{|p_1^\beta p_2^\beta..p_k^\beta p^1_2p^1_3..p^1_{N_1}
 p^1_1 p_{k+1}^\beta..p^\beta_{N_\beta}|\mathcal{J}\backslash \{J^\beta\}|}-
 G^{(g)}_{|p_1^\beta p_2^\beta..p_{k-1}^\beta p^1_1 p^1_2..p^1_{N_1}
  p_{k}^\beta..p^\beta_{N_\beta}|\mathcal{J}\backslash \{J^\beta\}|}}{E_{p^\beta_k}-E_{p^1_1}}\\
  +&\sum_{h+h'=g}\sum_{k=2}^{N_1}
  \sum_{\substack{\mathcal{I}\uplus \mathcal{I}'=\mathcal{J}}}\frac{G^{(h)}_{|p^1_{k+1}..p^1_{N_1}p^1_1|\mathcal{I}|} 
  G^{(h')}_{|p^1_2..p^1_k|\mathcal{I}'|}-
 G^{(h)}_{|p^1_{k}..p^1_{N_1}|\mathcal{I}|} G^{(h')}_{|p^1_1 p^1_2..p^1_{k-1}|\mathcal{I}'|}}{E_{p^1_k}-E_{p^1_1}}\bigg\}.\nonumber
 \end{align}

\begin{exm}\label{Exm:quartRec}
 Take $\mathcal{J}=\emptyset$ and $g=0$. Then, \eqref{eq:quartRec} gives
 \begin{align*}
  G^{(0)}_{|p_1 p_2..p_{N}|}=&-\lambda
  \sum_{k=1}^{\frac{N-2}{2}}
  \frac{G^{(0)}_{|p_{2k+2}..p_{N}p_1|} 
  G^{(0)}_{|p_2..p_{2k+1}|}-
 G^{(0)}_{|p_{2k+1}..p_{N}|} G^{(0)}_{|p_1 p_2..p_{2k}|}}{(E_{p_{2k+1}}-E_{p_1})(E_{p_2}-E_{p_{N}})},
 \end{align*}
 where $p_i\equiv p^1_i$.
 This result coincides with \cite[Prop. 3.4]{Grosse:2012uv} since all correlation functions with boundary
 of odd length vanishes.
\end{exm}
\noindent 
The recursive equation of \sref{Proposition}{Prop:quartRec} is \textit{nonlinear}
on the rhs (strictly different to the cubic model). 
A correlation function with boundary lengths larger than 
3 is therefore built recursively by multiplications of correlation functions of shorter boundary lengths. 
This is a further big difference to the cubic model, since
the recursive equation of \sref{Proposition}{Prop:CubisRecur} is linear.

Notice also that the first boundary labelled by $p^1_1,..,p_{N_1}^1$ as 
well as the labellings $p^1_1,p^1_2$ and $p^1_{N_1}$ play 
a special r\^ole for the recursion. One could write the recursion for any other three different adjacent 
labellings $p^\beta_k,p^\beta_{k+1}$ and $p^\beta_{k-1}$, where 
$p^\beta_k$ is then called the \textit{base point}. 
Since all correlation functions have a unique solution, at least from the 
perturbative expansion in $\lambda$, the recursive equation needs to be independent of the order of the base points 
chosen for each recursion until only correlation functions of boundary length 1 and 2 survive.
The explicit structure is analysed and discussed extensively in \sref{Sec.}{Sec.quartRec} for \sref{Example}{Exm:quartRec}.

\begin{rmk}
 The recursive equation of 
 \sref{Proposition}{Prop:CubisRecur} is invariant under the renormalisation
 discussed in \sref{Remark}{Rmk:quartRenorm}, $\lambda\to Z^2\lambda_{bare}$ and $E_n\to ZE_n+const$. 
 Since this holds in particular for the planar 4-point function
 by setting $N=4$ in \sref{Example}{Exm:quartRec}, there
 is no need for additional coupling constant renormalisation $\lambda_{bare}=\lambda$.
\end{rmk}
\noindent
Next, we will consider an analogue to the boundary creation operator in the quartic interacting case. 
Assume for the following considerations that all $E_k$ are distinct with multiplicity one.
Assume further that the $E_k$ can be varied in a small disjoint neighbourhood $U_k$ such that we have
\begin{align}\nonumber
 &\frac{\partial}{\partial E_i} \int D\Phi \,e^{-V\Tr\big(E\Phi^2+\frac{\lambda}{4}\Phi^4 -J\Phi\big)}
 =-V\int D\Phi\,\sum_{n=0}^\mN \Phi_{in}\Phi_{ni} 
 \,e^{-V\Tr\big(E\Phi^2+\frac{\lambda}{4}\Phi^4 -J\Phi\big)}\\\label{eq:quartdef1}
 =&-\frac{1}{V}\sum_{n=0}^\mN\frac{\partial^2}{\partial J_{in}\partial J_{ni}}\int D\Phi\, 
 \,e^{-V\Tr\big(E\Phi^2+\frac{\lambda}{4}\Phi^4 -J\Phi\big)}.
\end{align}
Having multiplicities $r_k$
for the distinct eigenvalues $e_k=E_q=..=E_{q+r_k-1}$ changes \eqref{eq:quartdef1} to
\begin{align}\label{eq:quartdef2}
 \frac{V}{r_k}\frac{\partial}{\partial E_q} \int D\Phi \,e^{-V\Tr\big(E\Phi^2+\frac{\lambda}{4}\Phi^4 -J\Phi\big)}
 =-\sum_{n=0}^\mN\frac{\partial^2}{\partial J_{qn}\partial J_{nq}}\int D\Phi\, 
 \,e^{-V\Tr\big(E\Phi^2+\frac{\lambda}{4}\Phi^4 -J\Phi\big)}.
\end{align}
We conclude with this idea that the
derivative wrt $E_q$ on a correlation function gives:
\begin{prps}\label{Prop:quartder}
 Let $\mathcal{J}=\{J^1,..,J^b\}$, $J^\beta=\{p^\beta_1,..,p^\beta_{N_\beta}\}$ and $\beta\in\{1,..,b\}$.
 Let $e_k$ be the distinct eigenvalues of $E$ of multiplicity $r_k$ with $e_k=E_q=..=E_{q+r_k-1}$.
 Then, we have for $N=\sum_{i=1}^bN_i$ even, and $E_q\neq E_{p^i_j}$
 \begin{align*}
  -\frac{V}{r_k}\frac{\partial}{\partial E_q}G_{|\mathcal{J}|}=\frac{1}{V}\sum_{n=0}^\mN G_{|qn|\mathcal{J}|}
  +\frac{1}{V^2}
 G_{|q|q|\mathcal{J}|}+\sum_{\beta=1}^b\sum_{k=1}^{N_\beta}G_{|qp^\beta_{k}..p^\beta_{N_\beta+k}|\mathcal{J}
  \backslash \{J^\beta\}|}
  +\!\!\!\sum_{\mathcal{I}\uplus \mathcal{I}'=\mathcal{J}}G_{|q|\mathcal{I}|}G_{|q|\mathcal{I}'|}
 \end{align*}
 with $p^i_{N_i+j}\equiv p^i_j$. The last sum over $\mathcal{I}$ contains 
 only terms with odd $|\mathcal{I}|$.
 \begin{proof}
  Assume $p^i_j$ such that all $E_ {p^i_j}$ are pairwise distinct. Let 
  \begin{align*}
   \hat{D}=\frac{\partial^{N_1+..+N_b}}{\partial J_{p^1_1p^1_2}\partial J_{p^1_2p^1_3}..
   \partial J_{p^1_{N_1}p^1_1} \partial J_{p^2_1p^2_2}\partial J_{p^2_2p^2_3}..
   \partial J_{p^2_{N_2}p^2_1} ..\partial J_{p^b_1p^b_2}..
   \partial J_{p^b_{N_b}p^b_1}}.
  \end{align*}
  Due to \eqref{eq:quartdef2} the derivative wrt to $E_q$ commutes with the derivative wrt $ J_{nm}$. Using the definition 
  of the correlation function yields
  \begin{align*}
   -\frac{V}{r_k}\frac{\partial}{\partial E_q}G_{|\mathcal{J}|}=&
   -\frac{V^{b-1}}{r_k}\hat{D}\frac{\partial}{\partial E_q}\log \Z[J]\big\vert_{J=0}\\
   =&V^{b-2}\hat{D}\sum_{n=0}^\mN \frac{\frac{\partial^2}{\partial J_{qn}\partial J_{nq}}\Z[J]}{\Z[J]}\bigg\vert_{J=0}\\
   =&V^{b-2}\hat{D}\sum_{n=0}^\mN\bigg\{\bigg(\frac{\partial\log \Z[J]}{\partial J_{qn}}\bigg)
   \bigg(\frac{\partial\log \Z[J]}{\partial J_{nq}}\bigg)+\frac{\partial^2\log \Z[J]}{\partial J_{qn}\partial J_{nq}}
   \bigg\}\bigg\vert_{J=0}.
  \end{align*}
Since $E_q\neq E_ {p^i_j}$ the multiplicative term generates a cycle only in the case of $n=q$ 
and vanishes otherwise for $J=0$. 
The last term survives and creates one more boundary for $n=q$
and merges two boundaries for $n=p^i_j$, which gives the expression of the Proposition.
\\
For coinciding $E_{p^i_j}$ the rhs is regular which induces regularity on the lhs.
 \end{proof}
\end{prps}
\begin{cor}\label{Cor:quartF}
 Let $F:=\frac{1}{V^2}\log \Z[0]$. With the eigenvalue distribution assumed in \sref{Proposition}{Prop:quartder}, 
 we have in particular
 \begin{align*}
  -\frac{V}{r_k}\frac{\partial}{\partial E_q}F=\frac{1}{V}\sum_{n=0}^\mN G_{|qn|}+\frac{1}{V^2}G_{|q|q|}.
 \end{align*}
\end{cor}
\noindent
Notice that \sref{Proposition}{Prop:quartder} has a problem if $E_q=E_{p^i_j}$, therefore it is natural to define a 
derivative $\mathcal{T}_{p}$ which ignores the fact that $G$ depends (possibly) explicitly on $E_q$:
\begin{dfnt}\label{Def:Dp}
Let $\mathcal{J}=\{J^1,..,J^b\}$, $J^\beta=\{p^\beta_1,..,p^\beta_{N_\beta}\}$ and $\beta\in\{1,..,b\}$.
 The operator $\mathcal{T}_{q}$ is defined on the correlation function $G_{|\mathcal{J}|}$ by 
 \begin{align*}
  \mathcal{T}_q G_{|\mathcal{J}|}:=\frac{1}{V}\sum_{n=0}^\mN G_{|qn|\mathcal{J}|}
  +\frac{1}{V^2}
 G_{|q|q|\mathcal{J}|}+\sum_{\beta=1}^b\sum_{k=1}^{N_\beta}G_{|qp^\beta_{k}p^\beta_{k+1}..p^\beta_{N_\beta+k}|\mathcal{J}
  \backslash \{J^\beta\}|}+\sum_{\mathcal{I}\uplus \mathcal{I}'=\mathcal{J}}G_{|q|\mathcal{I}|}G_{|q|\mathcal{I}'|},
 \end{align*}
 where $q\in \mathcal{J}$ is possible.
\end{dfnt}
\noindent 
The operator $\mathcal{T}_{q}$ can be understood 
due to the considerations of \sref{Proposition}{Prop:quartder} 
as a derivative which does not act on any external energies.
From the perturbative point of view, this operator acts only on the closed faces
of the Feynman graphs, even if $E_q$ coincides 
with an external face. 

The SDE of \sref{Proposition}{Prop:Quart2P} can be rewritten with $\mathcal{T}_{q}$ to 
\begin{align*}
  G_{|pq|}&=
\frac{1}{H_{pq}}-\frac{\lambda}{H_{pq}} 
\bigg\{  G_{|pq|}\mathcal{T}_p F+\frac{1}{V^2}\mathcal{T}_p G_{|pq|}
+\frac{1}{V}\sum_{n=0}^\mN \frac{G_{|pq|}- G_{|nq|}}{E_n-E_p} 
+\frac{1}{V^2} \frac{G_{|p|q|}- G_{|q|q|}}{E_q-E_p} 
\bigg\}\;
 \end{align*}
with $F$ defined in \sref{Corollary}{Cor:quartF}.

\begin{lemma}\label{Lemma:Dp}
 Let $\mathcal{J}=\{J^1,..,J^b\}$, $J^\beta=\{p^\beta_1,..,p^\beta_{N_\beta}\}$ and $\beta\in\{1,..,b\}$. Let further be 
 \begin{align*}
  \hat{D}=\frac{\partial^{N_1+..+N_b}}{\partial J_{p^1_1p^1_2}\partial J_{p^1_2p^1_3}..\partial J_{p^1_{N_1}p^1_1}
  ..\partial J_{p^b_1p^b_2}..\partial J_{p^b_{N_b}p^b_1}}
 \end{align*}
 with $G_{|\mathcal{J}|}=V^{b-2}\hat{D}\log\Z[J]\vert_{J=0}$. Then,
the identity
\begin{align*}
 \hat{D}(W^1_q+W^2_q)=\mathcal{T}_q G_{|\mathcal{J}|}
\end{align*}
holds,
where $W^1_q$ and $W^2_q$ are defined in \sref{Theorem}{Thm:Raimar}.
\begin{proof}
 Take the definition of $W^1_q$ and $W^2_q$ of \sref{Theorem}{Thm:Raimar} and act with $\hat{D}$ on it, which is exactly the 
 \sref{Definition}{Def:Dp}
 of $\mathcal{T}_q$ on $G_{|\mathcal{J}|}$.
\end{proof}
\end{lemma}

\begin{rmk}
 In the limit $\mN,V\to\infty$, the operator $\mathcal{T}_q$ converges with $e(x)$ and $r(x)$ described 
in \sref{Sec.}{Sec:LargeLimit} formally to the functional derivative
\begin{align*}
 \mathcal{T}_q\to D(x):=-\frac{\delta}{\delta e(t)}\frac{\delta }{\delta (r(t)dt)}\bigg\vert_{t=x}.
\end{align*}
Setting $t=x$ after derivation, the operator $D(x)$ avoids the degenerate case, i.e. 
$D(x)$ ignores an explicit dependence on $e(x)$. 
\end{rmk}
\noindent
The next step is to derive the SDEs in its full generality. \sref{Proposition}{Prop:quartRec} shows that it is sufficient 
to assume boundaries of lengths 1 and 2. However, we have different SDEs if the base point is chosen
from a boundary of length 1 or length 2. For the base point coming from boundary length 1, the SDE reads:
\begin{prps}\label{Prop:quart11P}
 Let $\mathcal{J}=\{J^2,..,J^{2b},J^{2b+1},..,J^{2b+b'}\}$, $J^\beta=\{p^\beta\}$ for $\beta\in\{2,..,2b\}$ and
 $J^\beta=\{q^\beta_1,q^\beta_2\}$ for $\beta\in\{2b+1,..,2b+b'\}$. Then, the $(1+..+1+2+..+2)$-point function
 with $2b$ boundaries of length 1 and $b'$ boundaries of length 2 satisfies
 \begin{align*}
  &G_{|p|\mathcal{J}|}=-\frac{\lambda}{H_{pp}}\bigg\{\frac{1}{V}\sum_{m=0}^\mN \frac{G_{|p|\mathcal{J}|}
  -G_{|m|\mathcal{J}|}}{E_m-E_p}
  +\sum_{\beta=2}^{2b}\frac{G_{|pp^\beta|\mathcal{J}\backslash\{J^\beta\}|}
 -G_{|p^\beta p^\beta|\mathcal{J}\backslash\{J^\beta\}|}}{E_{p^\beta}-E_p}
 \\
 &+\sum_{i=1}^2\sum_{\beta=2b+1}^{b'+2b}\frac{G_{|pq^{\beta}_iq^{\beta}_{i+1}|\mathcal{J}\backslash\{J^{\beta}\}|}
 -G_{|q^{\beta}_iq^{\beta}_iq^{\beta}_{i+1}|\mathcal{J}\backslash\{ J^{\beta}\}|}}{E_{q^{\beta}_i}-E_p}
  +\frac{1}{V^2}\mathcal{T}_p G_{|p|\mathcal{J}|}+
  \sum_{\mathcal{I}\uplus \mathcal{I}'=\mathcal{J}} G_{|p|\mathcal{I}'|}\mathcal{T}_p G_{|\mathcal{I}|}\bigg\},
 \end{align*}
where $q^\beta_3\equiv q^\beta_1$, $G_{|\emptyset|}=F$ of \sref{Corollary}{Cor:quartF}, 
$\mathcal{T}_pG$ is defined in \sref{Definition}{Def:Dp}, and the sums over the sets are restricted to 
correlation functions where the boundary lengths sum to an even number.
 \begin{proof}
  Let us assume $p^i$, $q^i_j$ and $p$ such that $E_{p^i},E_{q^i_j}$ and $E_p$ are pairwise different. Define further 
  \begin{align*}
   \hat{D}=\frac{\partial^{2b+b'-1}}{\partial J_{p^2p^2}\partial J_{p^3p^3}..
   \partial J_{p^{2b}p^{2b}} \partial J_{q^{2b+1}_1q^{2b+1}_2}\partial J_{q^{2b+1}_2p^{2b+1}_1}..
   \partial J_{q^{2b+b'}_1q^{2b+b'}_2}\partial J_{q^{2b+b'}_2p^{2b+b'}_1}}.
  \end{align*}
  By definition, the correlation function can be expressed by
  \begin{align*}
   G_{|p|\mathcal{J}|}=&V^{2b+b'-2}\hat{D}\frac{\partial}{\partial J_{pp}}\log\Z[J]\big\vert_{J=0}\\
   =&-\frac{V^{2b+b'-4}\lambda \hat{D}}{H_{pp}}\sum_{n,m}\frac{\frac{\partial^3}{\partial J_{pn}\partial J_{nm}
   \partial J_{mp}}}{\Z[J]}\Z[J]\big\vert_{J=0}\\
   =&-\frac{V^{2b+b'-4}\lambda \hat{D}}{H_{pp}}\frac{1}{\Z[J]}\bigg\{
   \sum_{n,m}\frac{\partial}{
   \partial J_{mp}}\frac{V}{E_m-E_p}\big(J_{mn}\frac{\partial}{\partial J_{pn}}-
 J_{np}\frac{\partial}{\partial J_{nm}}\big)\Z[J]\\
 &\qquad\qquad \qquad\qquad\qquad+\frac{\partial}{\partial J_{pp}}(W^1_p[J]+W^2_p[J])\Z[J]\bigg\}\bigg\vert_{J=0},
  \end{align*}
where we used \eqref{eq:SDeqGeneral} and
\sref{Theorem}{Thm:Raimar}. The quotient difference term is treated as usual which gives 
the quotient difference terms of the proposition. The last line is rewritten to
\begin{align}
 \frac{1}{\Z[J]}\frac{\partial}{\partial J_{pp}}(W^1_p[J]+W^2_p[J])\Z[J]
 =
 \frac{\partial}{\partial J_{pp}}(W^1_p[J]+W^2_p[J])+(W^1_p[J]+W^2_p[J])
 \frac{\partial}{\partial J_{pp}}\log \Z[J].\label{eq:quarttemp}
 \end{align}
Apply \sref{Lemma}{Lemma:Dp} to the second and last term of \eqref{eq:quarttemp}.
The first term produces a multiplication of two correlation functions, where each has a boundary of length 1 labelled by
$p$. The third term of \eqref{eq:quarttemp} produces a cubic term since $W^1_p[J]$ is already quadratic. 
Collecting all terms finishes the proof, after considering regularity on both sides of the equation for 
coinciding eigenvalues $E_i$.
 \end{proof}
\end{prps}
\begin{exm}
 Applying the genus expansion to \sref{Proposition}{Prop:quart11P}, putting $b=1$, $b'=0$ with $\mathcal{J}=J^2=\{p^2\}$,
 we achieve the linear equation for the planar $(1+1)$-point function
 \begin{align*}
  &G^{(0)}_{|p|p^2|}=-\frac{\lambda}{H_{pp}}\bigg\{\frac{1}{V}\sum_{m=0}^\mN \frac{G^{(0)}_{|p|p^2|}
  -G^{(0)}_{|m|p^2|}}{E_m-E_p}
  +\frac{G^{(0)}_{|pp^2|}
 -G^{(0)}_{|p^2p^2|}}{E_{p^2}-E_p}
 + G^{(0)}_{|p|p^2|}\frac{1}{V} \sum_{n=0}^\mN  G^{(0)}_{|np|}\bigg\}.
 \end{align*}
\end{exm}
\noindent 
For the base point coming from boundary length 2, the SDE reads:
\begin{prps}\label{Prop:quart21P}
 Let $\mathcal{J}=\{J^1,..,J^{2b},J^{2b+2},..,J^{2b+b'}\}$, $J^\beta=\{p^\beta\}$ for $\beta\in\{1,..,2b\}$ and
 $J^\beta=\{q^\beta_1,q^\beta_2\}$ for $\beta\in\{2b+2,..,2b+b'\}$. Then, the $(2+1+..+1+2+..+2)$-point function
 with $2b$ boundaries of length 1 and $b'$ boundaries of length 2 satisfies
 \begin{align*}
  &G_{|q_1q_2|\mathcal{J}|}=-\frac{\lambda}{H_{q_1q_2}}\bigg\{\frac{1}{V}\sum_{m=0}^\mN \frac{G_{|q_1q_2|\mathcal{J}|}
  -G_{|mq_2|\mathcal{J}|}}{E_m-E_{q_1}}+\frac{1}{V^2}\frac{G_{|q_1|q_2|\mathcal{J}|}
  -G_{|q_2|q_2|\mathcal{J}|}}{E_{q_2}-E_{q_1}}\\
  &+\!\sum_{\beta=1}^{2b}\frac{G_{|q_2q_1p^\beta|\mathcal{J}\backslash\{J^\beta\}|}
 -G_{|q_2 p^\beta p^\beta|\mathcal{J}\backslash\{J^\beta\}|}}{E_{p^\beta}-E_{q_1}}
 +\!\sum_{i=1}^2\sum_{\beta=2b+2}^{b'+2b}\frac{G_{|q_2q_1q^{\beta}_{i+1}q^{\beta}_{i}|\mathcal{J}\backslash\{J^{\beta}\}|}
 -G_{|q_2 q^{\beta}_i q^{\beta}_{i+1}q^{\beta}_{i}|\mathcal{J}\backslash\{ J^{\beta}\}|}}{E_{q^{\beta}_i}-E_{q_1}}\\
 &+\sum_{\mathcal{I}\uplus \mathcal{I}'=\mathcal{J}}
 G_{|q_2|\mathcal{I}|}\frac{G_{|q_1|\mathcal{I}'|}-G_{|q_2|\mathcal{I}'|}}{E_{q_2}-E_{q_1}}
 +\frac{1}{V^2}\mathcal{T}_{q_1} G_{|q_1q_2|\mathcal{J}|}
 +
  \sum_{\mathcal{I}\uplus \mathcal{I}'=\mathcal{J}} G_{|q_1q_2|\mathcal{I}'|}\mathcal{T}_{q_1} G_{|\mathcal{I}|}\bigg\},
 \end{align*}
where $q^\beta_3\equiv q^\beta_1$, $G_{|\emptyset|}=F$ of \sref{Corollary}{Cor:quartF}, $\mathcal{T}_{p}G$ is defined in \sref{Definition}{Def:Dp}, and the sums over the sets are restricted to 
correlation functions, where the boundary lengths sum to an even number.
 \begin{proof}
  Let us assume $p^i$, $q^i_j$ and $q_1,q_2$ such that $E_{p^i},E_{q^i_j}$ and $E_p$ are pairwise different. Define further 
  \begin{align*}
   \hat{D}=\frac{\partial^{2b+b'-2}}{\partial J_{p^1p^1}\partial J_{p^2p^2}..
   \partial J_{p^{2b}p^{2b}} \partial J_{q^{2b+1}_1q^{2b+1}_2}\partial J_{q^{2b+2}_2p^{2b+2}_1}..
   \partial J_{q^{2b+b'}_1q^{2b+b'}_2}\partial J_{q^{2b+b'}_2p^{2b+b'}_1}}.
  \end{align*}
  By definition, the correlation function reads
  \begin{align*}
   G_{|q_1q_2|\mathcal{J}|}=&V^{2b+b'-2}\hat{D}\frac{\partial^2}{\partial J_{q_1q_2}\partial J_{q_2q_1}}
   \log\Z[J]\big\vert_{J=0}\\
   =&-\frac{V^{2b+b'-4}\lambda \hat{D}}{H_{q_1q_2}}
   \frac{\partial}{\partial J_{q_2q_1}}\sum_{n,m}\frac{\frac{\partial^3}{\partial J_{q_1n}\partial J_{nm}
   \partial J_{mq_2}}}{\Z[J]}\Z[J]\big\vert_{J=0}\\
   =&-\frac{V^{2b+b'-4}\lambda \hat{D}}{H_{q_1q_2}}\frac{\partial}{\partial J_{q_2q_1}}\frac{1}{\Z[J]}\bigg\{
   \sum_{n,m}\frac{\partial}{
   \partial J_{mq_2}}\frac{V}{E_m-E_{q_1}}\big(J_{mn}\frac{\partial}{\partial J_{q_1n}}-
 J_{nq_1}\frac{\partial}{\partial J_{nm}}\big)\Z[J]\\
 &\qquad\qquad \qquad\qquad\qquad+\frac{\partial}{\partial J_{q_1q_2}}(W^1_{q_1}[J]+W^2_{q_1}[J])\Z[J]\bigg\}\bigg\vert_{J=0},
  \end{align*}
where we have used \eqref{eq:SDeqGeneral} and
\sref{Theorem}{Thm:Raimar}. The quotient difference term is treated as usual which gives 
the quotient difference terms of the proposition, where more terms appear
compared to \sref{Proposition}{Prop:quart11P}. These terms are generated by acting with the 
derivative $\frac{\partial}{\partial J_{q_2q_1}}$ on $J_{mn}$ and $J_{nq_1}$ respectively and fixing $m=q_2$.
The last line is rewritten to
\begin{align}\nonumber
 &\frac{1}{\Z[J]}\frac{\partial}{\partial J_{q_1q_2}}(W^1_{q_1}[J]+W^2_{q_1}[J])\Z[J]\\
 =&
 \frac{\partial}{\partial J_{q_1q_2}}(W^1_{q_1}[J]+W^2_{q_1}[J])+(W^1_{q_1}[J]
 +W^2_{q_1}[J])
 \frac{\partial}{\partial J_{q_1q_2}}\log \Z[J].\label{eq:quarttemp1}
 \end{align}
Applying \sref{Lemma}{Lemma:Dp} and collecting all terms finishes the proof, where the
regularity conditions manage coinciding eigenvalues $E_i$.
 \end{proof}
\end{prps}
\begin{exm}
 Applying the genus expansion to \sref{Proposition}{Prop:quart21P}, setting $b=0$, $b'=2$ with 
 $\mathcal{J}=J^2=\{q^2_1 q^2_2\}$
 gives the linear equation for the planar $(2+2)$-point function
 \begin{align*}
  &G^{(0)}_{|q_1q_2|q^2_1 q^2_2|}=-\frac{\lambda}{H_{q_1q_2}}
  \bigg\{\frac{1}{V}\sum_{m=0}^\mN \frac{G^{(0)}_{|q_1q_2|q^2_1 q^2_2|}
  -G^{(0)}_{|mq_2|q^2_1 q^2_2|}}{E_m-E_{q_1}}
 +\sum_{i=1}^2\frac{G^{(0)}_{|q_2q_1q^2_{i+1}q^2_{i}|}
 -G^{(0)}_{|q_2 q^{2}_i q^{2}_{i+1}q^{2}_{i}|}}{E_{q^{2}_i}-E_{q_1}}\\
  &\qquad\qquad\qquad\qquad\qquad\qquad  +
   G^{(0)}_{|q_1q_2|q^2_1 q^2_2|}\frac{1}{V} \sum_{n=0}^\mN  G^{(0)}_{|nq_1|}+ G^{(0)}_{|q_1q_2|} \mathcal{T}_{q_1} G^{(0)}_{|q^2_1 q^2_2|}
\bigg\}.
 \end{align*}
\end{exm}
\noindent
We emphasise that the SDE for a correlation function with base point from a boundary of length 2 has originally a 
term of the form $\frac{1}{V}\sum_n G_{|q_1n|\mathcal{J}|}$. Since we include this term inside 
the derivative $\mathcal{T}_{q_1}
G_{|\mathcal{J}|}$, the SDE has a much simpler recursive structure. 

Performing the genus-expansion for \sref{Proposition}{Prop:quart11P} and \sref{Proposition}{Prop:quart21P} leads to linear 
recursive equations for a correlation function of Euler characteristic $\chi$ where the inhomogeneous part is some 
$g^{g,\mathcal{J}}_{inh}$
depending on 
correlation functions of Euler characteristic $\chi'>\chi$ by
\begin{align}\label{eq:quartHOSD1}
 \hat{K}^1_p G^{(g)}_{|p|\mathcal{J}|}=&g^{g,\mathcal{J}}_{inh}\\\label{eq:quartHOSD2}
 \hat{K}^2_{q_1} G^{(g)}_{|q_1q_2|\mathcal{J}|}=&g^{g,\mathcal{J}}_{inh},
 \end{align}
 where
 \begin{align}\label{eq:quartHOOP1}
 \hat{K}^1_p f(p):=&f(p)\bigg\{H_{pp}+\frac{\lambda}{V} \sum_{n=0}^\mN\bigg(\frac{1}{E_n-E_p}+G^{(0)}_{|np|}\bigg)\bigg\}-
 \frac{\lambda}{V}\sum_{n=0}^\mN\frac{f(n)}{E_n-E_p}\\\label{eq:quartHOOP2}
 \hat{K}^2_{q_1} f(q_1,q_2):=&f(q_1,q_2)\bigg\{H_{q_1q_2}+\frac{\lambda}{V} \sum_{n=0}^\mN
 \bigg(\frac{1}{E_n-E_{q_1}}+G^{(0)}_{|nq_1|}\bigg)\bigg\}-
 \frac{\lambda}{V}\sum_{n=0}^\mN\frac{f(n,q_2)}{E_n-E_{q_1}}.
\end{align}
Comparing with the cubically interacting model the equations 
\eqref{eq:quartHOSD1} and \eqref{eq:quartHOSD2} share some similar structure with \eqref{eq:cubicSD1}. 
The function $ H_{pp}+\frac{\lambda}{V}\sum_n G^{(0)}_{|np|}$ seems to take the analog r\^ole of $W^{(0)}_{|p|}$. However, 
the first equation to solve is a nonlinear 
equation for $G^{(0)}_{|pq|}$ which is a function depending on two variables instead of one. Solving this
needs a completely different strategy than for cubic interaction. The result is changed tremendously 
and has a different structure
which can also be seen in the 
perturbative expansion since 
hyperlogarithms survives at any order in $\lambda$ (see \sref{App.}{App:PertQuartic}).

\section[Solution of the Planar 2-Point Function]
{Solution of the Planar 2-Point Function}
\label{Sec.quartSolution}
In this section we are analysing the structure of the planar 2-point function. 
The main idea for that came
from an observation of the earlier known special case on the $\D=D=2$ Moyal space \cite{Panzer:2018tvy}. 
The function $I_D(w)$ defined in \sref{Definition}{def:ID} seems to have 
an ''involutive'' structure $I_D(-I_D(w))\,"=``-w$ which holds only formally since the domains have to be
specified.
Nevertheless, this formal property 
gave the right ansatz for the general solution. \sref{Sec.}{sec:tau}-\ref{sec.fm} 
is taken from our joint paper \cite{Grosse:2019jnv} with H. Grosse and R. Wulkenhaar,
where the major discovery is found by R. Wulkenhaar. The notation is adapted to the rest of the thesis, and
later results will build up from this. 

An important tool which plays an incredible r\^ole throughout this section is the Lagrange-B\"urmann inversion formula:
\begin{thrm}(\cite{Lagrange:1770??,Buermann:1799??})
\label{thm:Lagrange-inversion}
Let $\phi(w)$ be analytic at $w=0$ with $\phi(0)\neq 0$ and 
$f(w) := \frac{w}{\phi(w)}$. Then the inverse $g(z)$ of $f(w)$
with $z=f(g(z))$ is analytic at $z=0$ and given by
\begin{equation}
g(z) = \sum_{n=1}^{\infty} \frac{z^n}{n!} 
\frac{d^{n-1}}{d w^{n-1}}\Big|_{w=0} (\phi(w))^n\;.
\label{eq:Lagrange}
\end{equation}
More generally, if $H(z)$ is an arbitrary analytic function 
with $H(0)=0$, then
\begin{equation}
H(g(z)) = \sum_{n=1}^{\infty} \frac{z^n}{n!} 
\frac{d^{n-1}}{d w^{n-1}}\Big|_{w=0} \Big( H'(w) \big(\phi(w)\big)^n \Big)\;.
\label{eq:Buermann}
\end{equation}
\end{thrm}
\noindent
Taking the renormalisation for $\D<6$ for the quartic model into account (see \sref{Remark}{Rmk:quartRenorm}) gives
the nonlinear 
equation for the planar 2-point function by \sref{Proposition}{Prop:Quart2P}
\begin{align}
ZG^{(0)}_{pq}=\frac{1}{E_p+E_q}-\frac{\lambda}{V(E_p+E_q)}
\sum_{n=0}^{\mathcal{N}} 
\Big( ZG^{(0)}_{pq} \;ZG^{(0)}_{pn}
- \frac{ZG^{(0)}_{nq} -ZG^{(0)}_{pq} }{E_{n}-E_p}\Big) \;.
\label{Gab-orig}
\end{align}
Taking the large $\mN,V$-limit discussed in \sref{Sec.}{Sec:LargeLimit}, we write
\begin{align}
G_{pq}^{(0)}=: G(x,y)\Big|_{x=E_a-\mu^2_{bare}/2,\;y=E_b-\mu^2_{bare}/2}\;,
\label{Gab-continuation}
\end{align}
then $G(x,y)$ originally defined only on the (shifted) spectrum of $E$
extends to a sectionally holomorphic function which satisfies the
integral equation
\begin{align}
&(\mu_{bare}^2{+}x{+}y)ZG(x,y)
= 1-\lambda
\int_0^{\Lambda^2} \!\! dt\;\rho_0(t)
\Big( ZG(x,y) \;ZG(x,t) 
- \frac{ZG(t,y) -ZG(x,y)}{t-x}\Big) \;.
\label{eq:Gint}
\end{align}
Here we have used $\varrho_0(t)=r(t)$,
where $e(x)=x$\footnote{The linear case $e(x)=x$
can be assumed without loss of generality. For arbitrary $e(x)$ the measure has to be changed to 
$\varrho_0(t)=\frac{r(e^{-1}(t))}{e'(e^{-1}(t))}$.} 
and $r(x)$ are differentiable function defined in \sref{Sec.}{Sec:LargeLimit}.
 
We assume that the measure $\rho_0(t)$ 
is a H\"older-continuous function. The final 
result will make
perfect sense even for $\rho_0$ being a linear combination of Dirac
measures. Intermediate steps become more 
transparent if $\rho_0\in C^{0,\alpha}([0,\Lambda^2])$ is assumed. 
Using techniques for boundary values of sectionally holomorphic
functions, explained in detail in \cite{Tricomi85,Grosse:2012uv, Grosse:2014lxa,
  Panzer:2018tvy}, one finds that 
a solution for $G(a,b)$ at $0 < a,b < \Lambda^2$ should be searched in
the form
\begin{align}
ZG(a,b)
= \frac{e^{\mathcal{H}_a^{\Lambda}[\tau_b(\bullet)]} 
\sin \tau_b(a)}{\lambda \pi \rho_0(a)}
= \frac{e^{\mathcal{H}_b^{\Lambda}[\tau_a(\bullet)]} 
\sin \tau_a(b)}{\lambda \pi \rho_0(b)}\;,\quad 
\label{Gab-ansatz}
\end{align}
where the \textit{angle function} $\tau_a: (0,\Lambda^2)\to [0,\pi]$ 
for  $\lambda>0$ and $\tau_a: (0,\Lambda^2)\to [-\pi,0]$ 
for  $\lambda<0$ remains to be determined. Here, 
\begin{align*}
\mathcal{H}^{\Lambda}_a[f(\bullet)]:= \frac{1}{\pi} 
\lim_{\epsilon\to 0} \int_{[0,\Lambda^2]\setminus
  [a-\epsilon,a+\epsilon]} \frac{dt\; f(t)}{t-a}
=\lim_{\epsilon\to 0} 
\mathrm{Re}\Big(\frac{1}{\pi} 
\int_0^{\Lambda^2} \frac{dt\; f(t)}{t-(a+\mathrm{i}\epsilon)}\Big)
\end{align*}
denotes the finite Hilbert transform.
We go with the ansatz (\ref{Gab-ansatz}) into 
(\ref{eq:Gint}) at $x=a+\mathrm{i}\epsilon$ and $y=b$:
\begin{align}
&\Big(\mu_{bare}^2+a+b+\lambda \pi
  \mathcal{H}^{\Lambda}_a[\rho_0(\bullet)]
+\frac{1}{\pi} \int_0^{\Lambda^2} dt\;
e^{\mathcal{H}_t^{\Lambda}[\tau_a(\bullet)]} 
\sin \tau_a(t)
\Big) ZG(a,b)
\nonumber
\\
&= 1 + 
 \mathcal{H}^{\Lambda}_a\big[ e^{\mathcal{H}_\bullet ^{\Lambda}[\tau_b]} 
\sin \tau_b(\bullet)\big]\;.
\label{intGab-ansatz}
\end{align}
A H\"older-continuous function $\tau:(0,\Lambda^2)\to [0,\pi]$ or
$\tau:(0,\Lambda^2)\to [-\pi,0]$ satisfies
\begin{align*}
\mathcal{H}^{\Lambda}_a\big[ e^{\mathcal{H}_\bullet^{\Lambda}[\tau]} 
\sin \tau(\bullet)\big] &=  e^{\mathcal{H}_a^{\Lambda}[\tau]} 
\cos \tau(a)-1 \;,  &
\int_0^{\Lambda^2} \!\!dt\;
e^{\pm \mathcal{H}_t^{\Lambda}[\tau(\bullet)]} 
\sin \tau(t) &= \int_0^{\Lambda^2} \!\! dt\;\tau(t)\;.
\end{align*}
The first identity appeared in \cite{Tricomi85}, the second one was
proved in \cite{Panzer:2018tvy}. Inserting both identities 
into (\ref{intGab-ansatz})
gives with (\ref{Gab-ansatz}) a consistency relation for the angle function:
\begin{align}
\tau_a(p) = \arctan
\bigg(\frac{\lambda \pi
  \rho_0(p)}{\mu_{bare}^2+a+p+\lambda \pi
  \mathcal{H}^{\Lambda}_p[\rho_0(\bullet)]
+ \frac{1}{\pi} \int_0^{\Lambda^2} dt \;\tau_p(t) }\bigg)\;,
\label{cottauba}
\end{align}
where the $\arctan$-branch in $[0,\pi]$ is selected for $\lambda>0$
and the branch in $[-\pi,0]$ for $\lambda<0$.

\subsection[Solution of the Angle Function]{Solution of the Angle Function\footnote{This subsection is taken from our paper
\cite{Grosse:2019jnv}}}
\label{sec:tau}

We succeed in solving (\ref{cottauba}) for \emph{any} 
H\"older-continuous $\rho_0$ of spectral dimension $\D< 6$. The difficulty 
was to \emph{guess} 
the solution; verifying it is a straightforward exercise in complex
analysis. The main
step is to deform the measure function. 
We first introduce structures 
for a fictitious measure $\rho_c$; later $\rho_c$ will be
particularly chosen.
\begin{dfnt}
\label{def:h}
Let $\rho_c$ be a 
H\"older-continuous function on some interval $[\nu_D,\Lambda_D^2]$. 
For $\mu^2=\mu_{bare}^2$ in $D=2\lfloor\frac{ \D}{2}\rfloor =0$, and $\mu^2>\min(0,-\nu_D)$ a free 
parameter in
$D=2\lfloor\frac{ \D}{2}\rfloor \in \{2,4\}$, define functions $h_{00}\equiv h_0,
h_{02},h_{04},h_2,h_4$ on
$\mathbb{C}\setminus [\nu_D,\Lambda_D^2]$ by
\begin{align}
h_{0D}(z)&:=\int_{\nu_D}^{\Lambda_D^2} 
\frac{dt\;\rho_c(t)}{t-z}\;,\qquad 
h_{0}(z):= h_{00}(z)\;,
\label{eq:hDD}
\qquad
\\
h_2(z)&:=h_{02}(z)-h_{02}(-\mu^2)=(z+\mu^2) 
\int_{\nu_2}^{\Lambda_2^2} \frac{dt\;\rho_c(t)}{(t+\mu^2)(t-z)}\;,
\nonumber
\\
h_4(z)&:=h_{04}(z)-h_{04}({-}\mu^2)-(z{+}\mu^2)h_{04}'( {-}\mu^2)
=(z{+}\mu^2)^2 
\int_{\nu_4}^{\Lambda_4^2}\!\!\!
\frac{dt\;\rho_c(t)}{(t+\mu^2)^2(t-z)}\,.
\nonumber
\end{align}
\end{dfnt}
\begin{dfnt}
\label{def:J}
For $\lambda\in \mathbb{C}$ and $h_D$ as given in 
\sref{Definition}{def:h}, we introduce functions $R_D$ on 
$\mathbb{C}\setminus [-\mu^2-\Lambda_D^2, -\mu^2-\nu_D]$ by
\begin{align}
R_D(z):=z-\lambda h_D(-\mu^2-z)
\equiv z-\lambda (-z)^{\frac{D}{2}} 
\int_{\nu_D}^{\Lambda_D^2}
\frac{dt\;\rho_c(t)}{(t+\mu^2)^{\frac{D}{2}} (t+\mu^2+z)}
\;.
\label{eq:J}
\end{align}
\end{dfnt}
\noindent
The limits $\lim_{\Lambda_D^2\to \infty} h_D(z)$ and 
$\lim_{\Lambda_D^2\to \infty} R_D(z)$ exist for $\rho_c$
of spectral dimension $\D$ according to \sref{Definition}{Def:Spec}. 
We have
\begin{align*}
D\in \{0,2\}\quad\Rightarrow \quad  
R_D'(z)&=1+\lambda h_{0D}'(-\mu^2-z)\;,
\end{align*}
which is uniformly positive on $\mathbb{R}_+$ 
for real $\lambda >  -(h_{0D}'(-\mu^2))^{-1}$ in $D\in \{0,2\}$.
In contrast,
\begin{align*}
R_4'(z)&=1-\lambda h_{04}'(-\mu^2)+\lambda h_{04}'(-\mu^2-z)\;,
\end{align*}
which is uniformly positive in the opposite
region of real $\lambda < (h_{04}'(-\mu^2))^{-1}$. 
\begin{lemma}
\label{lem:U}
Let $|\lambda| < \big(\int_{\nu_D}^{\Lambda_D^2} dt 
\frac{\rho_c(t)}{(t+\mu^2/2)^2}
+\delta_{D,4} \int_{\nu_D}^{\Lambda_D^2} dt 
\frac{\rho_c(t)}{(t+\mu^2)^2}
\big)^{-1}$. Then:
\begin{enumerate}
\item  $R_D$ is a biholomorphic map from a right half plane
$\mathcal{R}_\mu:= \{z\in \mathbb{C}\;:~ \mathrm{Re}(z) > -\frac{\mu^2}{2}\}$
onto a domain $U_D\subset \mathbb{C}$. For $\lambda$ real, $U_D$ 
contains $[R_D(\nu_D),\infty)$. 

\item For $\lambda$ real, $\mathrm{Im}(R_D(z))$ and $\mathrm{Im}(z)$ 
have the same sign for every $z\in \mathcal{R}_\mu$.
\end{enumerate}
\begin{proof}
 1.\quad  We show that $R_D$ is injective on $\mathcal{R}_\mu$. 
Any two points $z_0\neq z_1 \in \mathcal{R}_\mu$ can be connected by a straight line 
$[0,1]\ni s \mapsto c(s)=z_0+(z_1-z_0)s \in \mathcal{R}_\mu$. 
Then for $D\in\{0,2\}$
\begin{align*}
\big|R_D(z_1)-R_D(z_0)\big| &= |z_1-z_0| \Big|1 +\lambda \int_0^1 ds 
\int_{\nu_D}^{\Lambda_D^2} 
\frac{dt\;\rho_c(t)}{(t+\mu^2+ c(s))^2}\Big|
\nonumber
\\
&
\geq 
|z_1-z_0| \Big(1 - \sup_{s\in [0,1]} |\lambda| 
\int_{\nu_D}^{\Lambda_D^2} 
\frac{dt\;\rho_c(t)}{|t+\mu^2+ c(s)|^2}\Big)>0\;.
\end{align*}
For $D=4$ we have 
\begin{align*}
&\big|R_4(z_1)-R_4(z_0)\big| 
\nonumber
\\
&= |z_1-z_0||1-\lambda h_{04}'(-\mu^2)| 
\Big|1 +\frac{\lambda}{1- \lambda h_{04}'(-\mu^2)} 
\int_0^1 ds \int_{\nu_4}^{\Lambda_4^2} 
\frac{dt\;\rho_c(t)}{(t+\mu^2+ c(s))^2}\Big|
\end{align*}
which under the adapted condition leads to the same conclusion 
$\big|R_4(z_1)-R_4(z_0)\big| > 0$. 

It follows from basic properties of holomorphic functions that
$R_D$ is, as 
holomorphic and injective function,  even a biholomorphic map 
$R_D:\mathcal{R}_\mu \to U_D:= R_D(\mathcal{R}_\mu)$.

\smallskip

2.\quad For $\lambda$ real we have
\begin{align*}
\mathrm{Im}(R_D(x+\mathrm{i}y)) 
= y\Big\{(1-\lambda h_{04}'(-\mu^2))^{\delta_{D,4}} 
+\lambda \int_{\nu_D}^{\Lambda_D^2} \frac{dt\;
\rho_c(t)}{(t+\mu^2+x)^2+y^2}\Big\}\;.
\end{align*}
The term in $\{~\}$ is strictly positive by the same reasoning as above.
\end{proof}
\end{lemma}
\noindent
We can now define the `$\lambda$-deformed' measure:
\begin{dfnt} 
Given $\lambda \in \mathbb{R}$, $\mu^2>0$ and a H\"older-continuous function 
$\rho_0:[0,\Lambda^2]\to \mathbb{R}_+$ of spectral dimension $\D$ according to 
\sref{Definition}{Def:Spec}. Then a function 
$\rho_\lambda$ on $[\nu_D,\Lambda_D^2]$ is implicitly defined 
by the equations
\begin{align}
\rho_0(t) &=:\rho_{\lambda}(R^{-1}_D(t))\quad
\Leftrightarrow\quad 
\rho_{\lambda}(x)= \rho_0(R_D(x))\;,
\label{tilderho}
\\
\Lambda_D^2&:=R_D^{-1}(\Lambda^2)\;,\qquad
\nu_D:=R_D^{-1}(0)\;,
\nonumber
\end{align}
where $h_D$ in $R_D$ is defined via (\ref{eq:J}) and (\ref{eq:hDD}) by the same 
function $\rho_c\mapsto \rho_\lambda$. 
\label{def:rho}
\end{dfnt}
\begin{rmk}
  The deformation from $\rho_0$ to $\rho_{\lambda}$ is the analogue of
  the deformation from $E$ to $E_c:=\sqrt{E^2+\frac{1}{4} c(\lambda)}$
  in the cubic model of \sref{Ch.}{chap:cubic}.  There the deformation parameter
  $c(\lambda)$ is implicitly defined in \sref{Corollary}{Coro:1P}. Neither
  that equation nor (\ref{tilderho}) in the quartic model can in
  general be solved in terms of `known' functions.
\end{rmk}

\begin{dfnt}
\label{def:ID}
Given $\lambda \in \mathbb{R}$, $\mu^2>0$ and a H\"older-continuous function 
$\rho_0:[0,\Lambda^2]\to \mathbb{R}_+$ of spectral dimension $\D$ 
according to 
\sref{Definition}{Def:Spec}. Let 
$\rho_\lambda$ be its associated deformed measure 
according to \sref{Definition}{def:rho}, and let 
$\lambda$ satisfy the 
requirements of \sref{Lemma}{lem:U} so that $R_D:\mathcal{R}_\mu \to U_D$
is biholomorphic. Then a holomorphic function
$I_D : U_D\setminus [0,\Lambda^2]\ni w \mapsto I_D(w) \in \mathbb{C}$ 
is defined by
\begin{align}
I_D(w):= -R_D(-\mu^2-R_D^{-1}(w))
=\mu^2 + R_D^{-1}(w)+ \lambda h_D(R^{-1}_D(w))\;,
\label{ID}
\end{align}
where $\mu$ in (\ref{ID}) and in \sref{Definition}{def:J} are the same 
and $R_D,h_D$ are defined with the deformed measure 
$\rho_c\mapsto \rho_\lambda$. 
\end{dfnt}

\begin{thrm}
\label{thm:tau}
Let $\rho_0:[0,\Lambda^2]\to \mathbb{R}_+$ be a H\"older-continuous
measure of spectral dimension $\D$ and
$\rho_{\lambda}$ its deformation according to 
\sref{Definition}{def:rho} for a 
real coupling constant $\lambda$  
with $|\lambda|<
\big(\int_{\nu_D}^{\Lambda_D^2} dt 
\frac{\rho_{\lambda}(t)}{(t+\mu^2/2)^2}
+\delta_{D,4} \int_{\nu_D}^{\Lambda_D^2} dt 
\frac{\rho_{\lambda}(t)}{(t+\mu^2)^2}
\big)^{-1}$. Then the consistency 
equation \eqref{cottauba} for the angle function is
solved by 
\begin{align}
\tau_a(p)= \lim_{\epsilon\to 0} 
\mathrm{Im}\big(\log (a+I_D(p+\mathrm{i}\epsilon))\big)\;,
\label{tauap-final}
\end{align}
with $I_D$ given by \sref{Definition}{def:ID},
provided that the following relations between $\mu_{bare}$ and $\mu$
are arranged:
$\mu_{bare}^2=\mu^2$ for $\D<2$ and
\begin{align}
2\leq \D&<4:\quad \mu_{bare}^2=\mu^2-2\lambda h_{02}(-\mu^2)\;,
\nonumber
\\
4\leq \D&<6:\quad \mu_{bare}^2=\mu^2\big(1-\lambda h_{04}'(-\mu^2)\big) 
-2\lambda h_{04}(-\mu^2)\;.
\label{massren}
\end{align}
\end{thrm}
\noindent\emph{Proof.}
Assume (\ref{tauap-final}). Then for the given range of $\lambda$ 
we have for $0< p< \Lambda^2$
\begin{align}
\tau_a(p)
&= \lim_{\epsilon\to 0} 
\arctan 
\Big(\frac{\mathrm{Im}(a+I_D(p+\mathrm{i}\epsilon))}{
\mathrm{Re}(a+I_D(p+\mathrm{i}\epsilon))}\Big)
= \arctan \Big(
\frac{\lim_{\epsilon\to 0} 
\lambda \,\mathrm{Im}(h_D(R_D^{-1}(p+\mathrm{i}\epsilon)))}{
\lim_{\epsilon\to 0} \mathrm{Re}(a+I_D(p+\mathrm{i}\epsilon))}\Big)
\nonumber
\\
& = \arctan \Big(
\frac{\lambda \pi  \rho_{\lambda}(R_D^{-1}(p))}{
\lim_{\epsilon\to 0} \mathrm{Re}(a+I_D(p+\mathrm{i}\epsilon))}\Big)
\nonumber
\\
&= \arctan \Big(
\frac{\lambda \pi  \rho_0(p)}{
\lim_{\epsilon\to 0} \mathrm{Re}(a+I_D(p+\mathrm{i}\epsilon))}\Big)\;,
\label{tauReIm}
\end{align}
where 2.~of \sref{Lemma}{lem:U}, the definition of $h_D$ and the 
defining relation (\ref{tilderho}) between $\rho_0$ and
$\rho_{\lambda}$ have been used. The $\arctan$ ranges in 
$[0,\pi]$ for $\lambda>0$ and in 
$[-\pi,0]$ for $\lambda<0$. Comparison with (\ref{cottauba}) shows
(after renaming variables) that we have to prove
\begin{align}
\lim_{\epsilon\to 0} \mathrm{Re}(I_D(a+\mathrm{i}\epsilon))
= \mu_{bare}^2+a+\lambda\pi \mathcal{H}^{\Lambda}_a[\rho_0(\bullet)]
+\frac{1}{\pi} \int_0^{\Lambda^2} \!\! dp\;\tau_a(p)\;.
\label{aim}
\end{align}

We evaluate the integral over $\tau_a$. For $p>\Lambda^2$ we have 
$R_D^{-1}(p)>\Lambda_D^2$ and consequently 
$\mathrm{Im}(h_D(R_D^{-1}(p+\mathrm{i}\epsilon)))=0$. This implies 
\begin{align*}
\frac{1}{\pi} \int_0^{\Lambda^2} dp\;\tau_a(p)
&=\lim_{\epsilon\to 0} \frac{1}{\pi}
\int_0^\infty dp\; \mathrm{Im}\log (a+I_D(p+\mathrm{i}\epsilon))
\nonumber
\\
&=\lim_{\epsilon\to 0}\frac{1}{2\pi \mathrm{i}} 
\int_{R_D(\gamma_\epsilon)} dw\; \log (a+I_D(w))= 
\lim_{\epsilon\to 0} T^\epsilon_D(a)\;,
\nonumber
\\
\text{where}\quad T^\epsilon_D(a)&:= \frac{1}{2\pi \mathrm{i}} 
\int_{R_D(\gamma_\epsilon)} dw\; \log \Big(
\frac{a+I_D(w)}{a+\kappa_D+(R_D^{-1}(w)+\mu^2)/c_D}\Big)\;.
\end{align*}
In the second line, the contour $\gamma_\epsilon$ encircles
$[\nu_D,\infty )$ clockwise
at distance $\epsilon$, i.e.\ it goes straight from 
$+\infty-\mathrm{i}\epsilon$ to $\nu_D-\mathrm{i}\epsilon$, in a left 
half circle to 
$\nu_D+\mathrm{i}\epsilon$ and straight again to 
$+\infty+\mathrm{i}\epsilon$. The
denominator included in $T^\epsilon_D(a)$ is holomorphic in $U_D$ and
does not contribute for $\epsilon\to 0$. The
constants $c_D$ and $\kappa_D$ are chosen as 
$c_2=c_0=1$, $\kappa_0=0$ and 
\begin{align}
c_4&=\frac{1}{1-\lambda h_{04}'(-\mu^2)}\;,\qquad 
\kappa_D=-\lambda h_{0D}(-\mu^2)\quad\text{for } D\in \{2,4\}\;.
\label{ckappa}
\end{align}
We  insert (\ref{ID}) and transform to $w=R(z)$:
\begin{align}
T^\epsilon_D(a)
&= \frac{1}{2\pi \mathrm{i}} 
\int_{\gamma_\epsilon} dz\; R_D'(z) \log \Big(
\frac{a+\mu^2+z+\lambda h_D(z)}{a+\kappa_D+(z+\mu^2)/c_D}\Big)
\nonumber
\\
&= \frac{1}{2\pi \mathrm{i}c_D} 
\int_{\gamma_\epsilon} dz\; \big(1+\lambda c_D h_{0D}'(-\mu^2-z)\big) 
\log \Big(1+  \frac{\lambda c_D h_{0D}(z)}{
c_D (a+\kappa_D) +z+\mu^2}\Big)\;.
\label{Te-1}
\end{align}

The function $h_D$ in (\ref{Te-1}) is defined with the $\lambda$-deformed 
measure $\rho_{\lambda}$. We will now 
\begin{itemize}\itemsep 0pt
\item rename $\rho_\lambda$ to $\rho_{c}$ and the given coupling constant 
$\lambda$ to $\lambda_c\in \mathbb{R}$,
\item consider a general complex $\lambda\in \mathbb{C}$ (i.e.\ 
$h_{0D}$ will be taken as in (\ref{eq:hDD}) without any relation between 
$\rho_c$ and $\lambda$),

\item take $\epsilon$ a fixed positive number. 
\end{itemize}
In this setting, $z$ in (\ref{Te-1}) keeps distance $\epsilon$ 
from $[\nu_D,\infty)$ so that (\ref{Te-1}) becomes a holomorphic function of 
$\lambda$ in a sufficiently small open ball around the
origin. We choose its radius so small that the logarithm admits 
a uniformly convergent power series expansion on $\gamma_\epsilon$. 
Hence, integral and series commute:
\begin{align}
T^\epsilon_D(a)
&= -\sum_{n=1}^\infty \frac{(-\lambda c_D)^n}{n c_D}
\frac{1}{2\pi \mathrm{i}} 
\int_{\gamma_\epsilon} dz\; \big(1+\lambda c_D h_{0D}'(-\mu^2-z)\big) 
\frac{\big(h_{0D}(z)\big)^n}{
\big(c_D (a+\kappa_D)+z+\mu^2\big)^n}\;.
\label{Ta}
\end{align}
Since $h_{0D}(z)\propto z^{-1}$ for $|z|\to \infty$, we can close
$\gamma_\epsilon$ by a large circle to a closed contour 
$\bar{\gamma}_\epsilon$ which avoids $[\nu_D,\infty)$. 

We first evaluate the part without $h_{0D}'$ (and the global 
factor $c_D^{-1}$) by the residue theorem. 
Since $h_{0D}(z)$ is holomorphic in $\mathbb{C}\setminus [\nu_D,\infty)$,
only the pole of order $n$ at $z=-c_D (a+\kappa_D)-\mu^2$ contributes:
\begin{align*} 
K_D^\epsilon(a)&:=  -\sum_{n=1}^\infty \frac{(-\lambda c_D)^n}{n}
\frac{1}{2\pi \mathrm{i}} 
\int_{\bar{\gamma}_\epsilon} dz\; 
\frac{\big(h_{0D}(z)\big)^n}{
\big(c_D (a+\kappa_D)+z+\mu^2\big)^n}
\nonumber
\\
&=  -\sum_{n=1}^\infty \frac{(-\lambda c_D)^n}{n!}
\frac{d^{n-1}}{dw^{n-1}}\Big|_{w=0}
\big(h_{0D}(w-c_D(a+\kappa_D)-\mu^2)\big)^n\;.
\end{align*}
Setting $\phi(w)=h_{0D}(w-c_D(a+\kappa_D)-\mu^2)$, the 
Lagrange inversion formula (\ref{eq:Lagrange}) 
shows that $w=-K_D^\epsilon(a)$ is the
inverse solution of the equation
$-\lambda c_D= f(-K_D^\epsilon(a))$, where $f(w)=\frac{w}{\phi(w)}$.
This means 
\begin{align}
\lambda c_D h_{0D}\big(-K_D^\epsilon(a)-c_D(a+\kappa_D)-\mu^2\big)
= K_D^\epsilon(a)\;.
\label{KDe}
\end{align}
Introducing $z(a):= K_D^\epsilon(a)+c_D(a+\kappa_D)$,
equation (\ref{KDe}) becomes 
\begin{align}
z(a)-\lambda c_D h_{0D}(-\mu^2-z(a))= 
c_D(a +\kappa_D)\;.
\label{eq:za}
\end{align}
Comparing with \sref{Definition}{def:J}, equation (\ref{eq:za})
boils down for any $D\in \{0,2,4\}$ to 
$a=R_D(z(a))$. But $a \in [0,\Lambda^2] \subset U_D$ so that we can
invert to $z(a)=R_D^{-1}(a)$. In summary, we have proved
a useful perturbative formula for $R_D^{-1}$:
\begin{lemma} For any $a\in [0,\Lambda^2]$ and $\lambda$ satisfying the assumptions of \sref{Lemma}{lem:U}, 
the inverse function of $R_D$ defined in 
(\ref{eq:J}) admits a convergent representation
\begin{align*}
R_D^{-1}(a) &= c_D (a+\kappa_D)-
\sum_{n=1}^\infty \frac{(-\lambda c_D)^n}{n!}
\frac{d^{n-1}}{dw^{n-1}}\Big|_{w=0}
\big(h_{0D}(w-c_D(a+\kappa_D)-\mu^2)\big)^n
\;.
\end{align*}
\end{lemma}

We continue with (\ref{Ta}). We insert (\ref{eq:hDD}) for $h_{0D}'$ and
change the integration order:
\begin{align*}
T^\epsilon_D(a)
&= -(a+\kappa_D) +\frac{1}{c_D} R_D^{-1}(a)
\\
&+ \lambda \int_{\nu_D}^{\Lambda_D^2} dt\;
\rho_c(t) \frac{d}{dt}
\Big(\frac{1}{2\pi\mathrm{i}} \int_{\bar{\gamma}_\epsilon} 
\frac{dz}{t+\mu^2+z}
\sum_{n=1}^\infty \frac{(-\lambda c_D)^n}{n}
\frac{\big(h_{0D}(z)\big)^n}{
\big(c_D (a+\kappa_D)+z+\mu^2\big)^n}\Big)\;.
\nonumber
\end{align*}
We first look at generic points $t\neq c_D(a+\kappa_D)$. This is 
no restriction because for H\"older-continuous $\rho_c$, 
ordinary and improper integral (the point $t=c_D(a+\kappa_D)$ removed) 
agree. The 
residue theorem picks up the simple 
pole at $z=-\mu^2-t$, for which we resum
the series to the logarithm, and the pole of order $n$ at 
$z=-c_D (a+\kappa_D)-\mu^2$:
\begin{align}
T^\epsilon_D(a)
&= -a-\kappa_D +\frac{R_D^{-1}(a)}{c_D}
- \lambda \int_{\nu_D}^{\Lambda_D^2} \!\!dt\;
\rho_c(t) \frac{d}{dt}
\log\Big(1+\frac{\lambda c_D h_{0D}(-t-\mu^2)}{
c_D (a+\kappa_D) -t}\Big)
\label{Ta-4}
\\
& + \lambda \int_{\nu_D}^{\Lambda_D^2} \!\! 
dt\;\rho_c(t) \frac{d}{dt}
\Big\{ \! \sum_{n=1}^\infty \frac{(-\lambda c_D)^n}{n!} 
\frac{d^{n-1}}{d w^{n-1}}\Big|_{w=0} \Big(
H'_{t,a}(w) 
\big(h_{0D}(w{-}\mu^2{-}c_D (a{+}\kappa_D))\big)^n\Big)\Big\},
\nonumber
\end{align}
where $H_{t,a}(w):= \log \frac{w+t-c_D(a+\kappa_D)}{t-c_D(a+\kappa_D)}$.
The original dependence on $\epsilon$ dropped out.
The B\"urmann formula (\ref{eq:Buermann}) 
identifies the term 
in $\{~\}$ of the last line of (\ref{Ta-4}) 
as $H_{t,a}(-K^\epsilon_D(a))$: 
\begin{align}
T^\epsilon_D(a)
&= -a-\kappa_D+\frac{R_D^{-1}(a)}{c_D}
+ \lambda \int_{\nu_D}^{\Lambda_D^2} \!\! 
dt\;\rho_c(t) \frac{d}{dt}
\log \Big(\frac{t-K^\epsilon_D(a)-c_D(a+\kappa_D)}{
t-c_D(a+\kappa_D)-\lambda c_D h_{0D}({-}t{-}\mu^2)}\Big)
\nonumber
\\
&= -a-\kappa_D +\frac{R_D^{-1}(a)}{c_D}
+ \lambda \int_{\nu_D}^{\Lambda_D^2} \!\! 
dt\;\rho_c(t) \frac{d}{dt}
\log \Big(\frac{t-R_D^{-1}(a)}{R_D(t)- a}\Big)\;.
\label{Ta-5}
\end{align}
We have used $K^\epsilon_D(a)+c_D(a+\kappa_D)=R_D^{-1}(a)$ and 
rearranged the denominator with (\ref{eq:J}) to $R_D(t)-a$. 

We stress that (\ref{Ta-5}) is proved for complex $\lambda$ in a 
ball about the origin of small radius determined by $\epsilon$. 
The identity theorem for holomorphic functions allows us to enlarge 
the domain of $\lambda$ on both sides back to the original domain of 
the theorem. This includes the original real value $\lambda=\lambda_c$ we 
started with, where $\rho_c=\rho_{\lambda}$ on the 
rhs and $\lim_{\epsilon\to 0} T^\epsilon_D(a)
=\frac{1}{\pi} \int_0^{\Lambda^2} dp\;\tau_a(p)$ on the lhs. Therefore, 
for the original real $\lambda$, 
\begin{align}
\frac{1}{\pi} \int_0^{\Lambda^2} dp\;\tau_a(p)
&= -a-\kappa_D +\frac{R_D^{-1}(a)}{c_D}
+ \lambda \int_{\nu_D}^{\Lambda_D^2} \!\! 
dt\;\rho_{\lambda}(t) \frac{d}{dt}
\log \Big(\frac{t-R_D^{-1}(a)}{R_D(t)- a}\Big)\;,
\label{Ta-6}
\end{align}
where also $R_D$ is built from $\rho_{\lambda}$.

The $t$-integral in (\ref{Ta-6}) does not need any exception
point. But for the next step it is useful to remove an
$\epsilon$-interval about $t=R_D^{-1}(a)$ to take the logarithms
apart. These principal value integrals can equivalently be written as 
limit of the real part when shifting $a$ to $a+\mathrm{i}\epsilon$:
\begin{align*}
\frac{1}{\pi} \int_0^{\Lambda^2} dp\;\tau_a(p)
&=\lim_{\epsilon\to 0} \mathrm{Re}\Big( 
-a-\kappa_D +\frac{R_D^{-1}(a)}{c_D}
+ \lambda \int_{\nu_D}^{\Lambda_D^2} \!\! 
\frac{dt\;\rho_{\lambda}(t)}{t-R_D^{-1}(a+\mathrm{i}\epsilon)}
\nonumber
\\
&\qquad- \lambda \int_{\nu_D}^{\Lambda_D^2} \!\! 
dt\;\rho_{\lambda}(t) \frac{d}{dt}
\log (R_D(t)- (a+\mathrm{i}\epsilon))\Big)
\nonumber
\\
&=\lim_{\epsilon\to 0} \mathrm{Re}\Big( 
-a-2\kappa_D +R_D^{-1}(a)+\mu^2 (1-c_D^{-1})
+ \lambda h_D(R_D^{-1}(a+\mathrm{i}\epsilon)) 
\nonumber
\\
&\qquad- \lambda \int_{0}^{\Lambda^2} \!\! 
\frac{dx\;\rho_\lambda(R_D^{-1}(x))}{
x- (a+\mathrm{i}\epsilon)}\Big)\;.
\end{align*}
Here we have completed the first $t$-integral 
$h_{0D}(R_D^{-1}(a+\mathrm{i}\epsilon))$ with (\ref{eq:hDD})
to 
$h_D(R_D^{-1}(a+\mathrm{i}\epsilon))$ 
and transformed in the second integral to $x=R_D(t)$.
Taking the relation (\ref{tilderho}) to the original measure into account
and recalling the definition~(\ref{ID}) of $I_D(a)$, 
we precisely confirm our aim (\ref{aim}) provided that
\begin{align*}
\mu_{bare}^2= 2\kappa_D + c_D^{-1}\mu^2\;.
\end{align*}
This finishes the proof. \hspace*{\fill} $\square$%

\begin{exm}\label{Ex:4D2}
 On the $\D=D=2$ Moyal space the eigenvalues increases linear $e(x)=x$ with the measure $r(x)=1$ 
which induces for the undeformed measure $\varrho_0(x)=r(x)=1$. 
The deformed measure $\varrho_\lambda$ coincides with the undeformed measure due to \sref{Definition}{def:rho}
$ \varrho_\lambda(x)=\varrho_0(R_2(x))=1$.
The function $R_2(x)$ was defined in \sref{Definition}{def:J} such that we have
\begin{align}
 R_2(x)=x+\lambda x\int_0^\infty \frac{dt}{(t+\mu^2) (t+\mu^2+x)}=x+\lambda \log\bigg(1+\frac{x}{\mu^2}\bigg) .
\end{align}
Setting $\mu^2=1$ which corresponds to $\mu_{bare}^2=1-2\lambda
\log(1+\Lambda^2)$ provides 
the inverses by the branches of Lambert-W 
\cite{Corless:1996??}, in particular
\begin{align}
R_2^{-1}(z)=
\lambda
W_0\Big(\frac{1}{\lambda}e^{\frac{1+z}{\lambda}}\Big)-1\;.
\end{align}
where $W_0$ is the principal branch of the Lambert function. Inserting into 
\sref{Definition}{ID} leads to
\begin{align}
\tau_a(p) &= \mathrm{Im}\log\big(a+I(p{+}\mathrm{i}\epsilon)\big)\,, &
I(z)&:= \lambda W_0\Big(\frac{1}{\lambda}e^{\frac{1+z}{\lambda}}\Big)
-\lambda\log \Big(1-
W_0\Big(\frac{1}{\lambda}e^{\frac{1+z}{\lambda}}\Big)\Big)\;,
\label{Lambert}
\end{align}
which was already found in \cite{Panzer:2018tvy}.
\end{exm}

\subsection[Solution of the 2-Point Function]{Solution of the 2-Point Function\footnote{This subsection is taken from our paper
\cite{Grosse:2019jnv}}}
With $\tau_a(p)$ determined, it remains to evaluate the Hilbert
transform in the equation (\ref{Gab-ansatz}) for 
the planar 2-point function $G(a,b)$. We first establish a general 
integral representation. In the next subsection this integral will be 
evaluated for the case of finite matrices.
\begin{thrm} 
\label{prop:HT}
The renormalised 2-point function of the  
quartic matrix field theory model with spectral dimension $\D$ and $D=2\lfloor\frac{\D}{2}\rfloor$ is given by 
\begin{align}
G(a,b)&:=
\frac{\mu^{2\delta_{D,4}} \exp (N_D(a,b))}{(\mu^2+a+b)}\;,
\label{Gab-D=4}
\end{align}
where
\begin{align}
N_D(a,b)
&=
\frac{1}{2\pi\mathrm{i}}
\int_{-\infty}^\infty dt \; 
\Big\{
\log \big(a-R_D(-\tfrac{\mu^2}{2}-\mathrm{i}t)\big)
\frac{d}{dt} \log\big(b-R_D(-\tfrac{\mu^2}{2}+\mathrm{i}t)\big)
\nonumber
\\
& \qquad\qquad
-\log \big(a-(-\tfrac{\mu^2}{2}-\mathrm{i}t)\big)
\frac{d}{dt} \log\big(b-(-\tfrac{\mu^2}{2}+\mathrm{i}t)\big)
\nonumber
\\
& \qquad\qquad
-\delta_{D,4}
\log \big(-R_D(-\tfrac{\mu^2}{2}-\mathrm{i}t)\big)
\frac{d}{dt} \log\big(-R_D(-\tfrac{\mu^2}{2}+\mathrm{i}t)\big)
\nonumber
\\
& \qquad\qquad 
+\delta_{D,4}\log \big(-(-\tfrac{\mu^2}{2}-\mathrm{i}t)\big)
\frac{d}{dt} \log\big(-(-\tfrac{\mu^2}{2}+\mathrm{i}t)\big)\Big\}
\label{ND}
\end{align}
and $R_D$ is built via (\ref{eq:J}) and (\ref{eq:hDD}) with the deformed 
measure $\rho_\lambda$ defined in (\ref{tilderho}). 
For $4\leq \D<6$, $G(a,b)$ is only determined up to a
multiplicative constant which here is normalised to 
$G(0,0)=1$ independently of $\mu$. 
For $\D<4$ there is an alternative representation
\begin{align}
G(a,b):=
\frac{\displaystyle (\mu^2{+}a{+}b)\exp \Big\{
\frac{1}{2\pi \mathrm{i}}
\int_{-\infty}^\infty \!\!\! dt \;
\log \Big(\frac{a-R_D({-}\frac{\mu^2}{2}{-}\mathrm{i}t)}{
a-({-}\frac{\mu^2}{2}{-}\mathrm{i}t)}\Big)
\frac{d}{dt}
\log \Big(\frac{b-R_D({-}\frac{\mu^2}{2}{+}\mathrm{i}t)}{
b-({-}\frac{\mu^2}{2}{+}\mathrm{i}t)}\Big)\Big\}
}{
(\mu^2+b+R^{-1}_D(a))
(\mu^2+a+R^{-1}_D(b))}.
\label{Gab-final}
\end{align}
\end{thrm}
\noindent
\emph{Proof.} We rely on structures developed during the proof 
of \sref{Theorem}{thm:tau}. 
The Hilbert transform of $\tau_a$ given by (\ref{tauap-final}) 
can be written as 
\begin{align*}
\mathcal{H}^{\Lambda}_b\big[\tau_a(\bullet)]
&= \lim_{\epsilon'\to 0 }\lim_{\epsilon\to 0 }
\mathrm{Re}\Big(\frac{1}{\pi}
\int_0^{\infty} dp \;
\frac{\mathrm{Im}\log \big(\frac{
a+I_D(p+\mathrm{i}\epsilon)}{
\mu^2+a+R_D^{-1}(p)}\big)}{p-(b+\mathrm{i}\epsilon')}\Big)
\nonumber
\\
&=\lim_{\epsilon'\to 0 }\lim_{\epsilon\to 0 }
\mathrm{Re}\Big(\frac{1}{2\pi\mathrm{i}}
\int_{R_D(\gamma_{\epsilon})} dw \;
\frac{\log \big(\frac{a+I_D(w)}{
\mu^2+a+R_D^{-1}(w)}\big)}{w-(b+\mathrm{i}\epsilon')}\Big)
\;.
\end{align*}
In the second line, $\epsilon'$ must be chosen much larger than 
$\epsilon$ so that $R_D(\gamma_{\epsilon})$ separates 
$b+\mathrm{i}\epsilon'$ from $\mathbb{R}_+$. As before, we are allowed 
to include a holomorphic denominator $\mu^2+a+R_D^{-1}(w)$. In contrast to 
the procedure in \sref{Theorem}{thm:tau} we choose it such that it 
has individually a  limit for $\Lambda\to \infty$. This leads to
the large-$w$ behaviour 
\[
\frac{a+I_D(w)}{
\mu^2+a+R_D^{-1}(w)}\propto 
\frac{1}{c_D}+ \mathcal{O}(1/R_D^{-1}(w))\;.
\]
Thus, for $\D<4$ where $c_D=1$, the integrand decays 
sufficiently fast to deform $\gamma_{\epsilon}$ near $\infty$. For
$\D\geq 4$, however, $c_4\neq 1$ prevents the deformation. This forces us to
subtract the Hilbert transform
$\mathcal{H}^{\Lambda}_r\big[\tau_r(\bullet)]$ at some reference point
$a=b=r>0$. We first move $R_D(\gamma_{\epsilon})$ past the pole 
$w=b+\mathrm{i}\epsilon'$ at expense of its residue. In the remaining 
integral (which is automatically real) we transform to $w=R_D(z)$:
\begin{subequations}
\begin{align}
&\mathcal{H}^{\Lambda}_b\big[\tau_a(\bullet)]
-\delta_{D,4}\mathcal{H}^{\Lambda}_r\big[\tau_r(\bullet)]
\nonumber
\\
&=\lim_{\epsilon'\to 0 }
\mathrm{Re}\Big(
\log \Big(\frac{a+I_D(b+\mathrm{i}\epsilon')}{
\mu^2+a+R_D^{-1}(b+\mathrm{i}\epsilon')}\Big)
-\delta_{D,4} 
\log \Big(\frac{r+I_D(r+\mathrm{i}\epsilon')}{
\mu^2+r+R_D^{-1}(r+\mathrm{i}\epsilon')}\Big)\Big)
\label{Hilbtau-12}
\\
& +
\lim_{\epsilon\to 0 }\frac{1}{2\pi\mathrm{i}}
\int_{\gamma_{\epsilon}} dz \; R_D'(z) 
\Big(
\frac{\log \big(\frac{a+\mu^2+z+\lambda h_D(z)}{
\mu^2+a+z}\big)}{R_D(z)-b}
-\delta_{D,4}
\frac{\log \big(\frac{r+\mu^2+z+\lambda h_D(z)}{
\mu^2+r+z}\big)}{R_D(z)-r}\Big)
\;.
\label{Hilbtau-13}
\end{align}
\label{Hilbtau-1all}
\end{subequations}
The line \eqref{Hilbtau-12} evaluates to 
\begin{align}
\eqref{Hilbtau-12}
= \log\Big[ \frac{\lambda\pi \rho_0(b)}{\sin \tau_a(b)}
\cdot \frac{1}{(\mu^2+a+R_D^{-1}(b))}
\Big( \frac{\lambda\pi \rho_0(r)}{\sin \tau_r(r)}
\cdot \frac{1}{(\mu^2+r+R_D^{-1}(r))}\Big)^{-\delta_{D,4}}
\Big]\;,
\label{Hilbtau-2f}
\end{align}
where real and imaginary part of 
$a+I_D(b+\mathrm{i}\epsilon')$ are rearranged to $\tau_a(b)$ as 
in (\ref{tauReIm}).  

In the last line (\ref{Hilbtau-13}), we write
\[
\frac{R_D'(z)}{R_D(z)-b}
= \frac{d}{dz} \log 
(R_D(z)-b)
= \frac{1}{z-b}
+\frac{d}{dz} \log \Big(\frac{z-b-\lambda h_D(-\mu^2-z)}{z-b}\Big)\;.
\]
Inserted back into \eqref{Hilbtau-13} we deform in the parts with
products of logarithms the contour $\gamma_\epsilon$ into the straight
line $-\frac{\mu^2}{2}+\mathrm{i}\mathbb{R}$. No poles or branch cuts are hit
during this deformation because $R_D(z)$ and $R_D(-\mu^2-z)$ are
holomorphic on the slit half plane $\{\mathrm{Re}(z)>-\frac{3}{5}\mu^2
\} \setminus [\nu_D, \infty)$. In this way we produce 
integrals which are
manifestly symmetric in both variables:
\begin{align}
\tilde{N}_D(a,b)
&= \frac{1}{2\pi \mathrm{i}}
\int_{-\infty}^\infty dt \;
\Big\{
\log \Big(1
+ \lambda \frac{h_D(-\frac{\mu^2}{2}+\mathrm{i}t)}{
\frac{\mu^2}{2} +a+\mathrm{i}t}\Big)
\frac{d}{dt}
\log \Big(1+\frac{\lambda h_D(-\frac{\mu^2}{2}-\mathrm{i}t)}{
b+\frac{\mu^2}{2}-\mathrm{i}t}\Big)
\nonumber
\\
&\quad - \delta_{D,4}
\log \Big(1
+ \lambda \frac{h_D(-\frac{\mu^2}{2}+\mathrm{i}t)}{
\frac{\mu^2}{2} +r+\mathrm{i}t}\Big)
\frac{d}{dt}
\log \Big(1+\frac{\lambda h_D(-\frac{\mu^2}{2}-\mathrm{i}t)}{
r+\frac{\mu^2}{2}-\mathrm{i}t}\Big)
\Big\}\;.
\label{tildeN}
\end{align}
The counterterm for $\D\geq 4$ is indispensable for convergence. 
Now the line (\ref{Hilbtau-13}) becomes
\begin{align}
\eqref{Hilbtau-13}
= \tilde{N}_D(a,b)
&+ \lim_{\epsilon\to 0} \frac{1}{2\pi\mathrm{i}}
\int_{\gamma_{\epsilon}} dz \; 
\Big(
\frac{\log \big(
\frac{1}{c_D}+ \frac{a-ac_D^{-1}+\kappa_D
+\lambda h_{0D}(z)
}{
\mu^2+a+z}\big)}{z-b}
\nonumber
\\
&\qquad\qquad-\delta_{D,4}
\frac{\log \big(
\frac{1}{c_D}+ \frac{r-rc_D^{-1}+\kappa_D
+\lambda h_{0D}(z)
}{
\mu^2+r+z}\big)}{z-r}
\Big)\;.
\label{Hilbtau-3}
\end{align}
For any $D\in \{0,2,4\}$ we can add the convergent integral 
$\frac{1}{2\pi\mathrm{i}}
\int_{\gamma_{\epsilon}} dz \; 
\big(
\frac{\log c_D}{z-b}
-\delta_{D,4}\frac{\log c_D}{z-r}\big)=0$ 
(in $D\in \{0,2\}$ we have $c_D=1$ whereas for $D=4$ we 
close $\gamma_\epsilon$ and use the residue theorem). 

We follow the same strategy as in \sref{Theorem}{thm:tau}: $\rho_\lambda$ is 
renamed to $\rho_c$ and held fixed, $h_D$ and $R_D$ are built with $\rho_c$  
and an independent complex $\lambda$ in a sufficiently small ball 
about the origin. Its radius is determined 
by $\epsilon$ which is also kept fixed.
Also $\Lambda_D^2$ is still finite, and 
$a-ac_D^{-1}$ and $\kappa_D$ have according to (\ref{ckappa}) a factor 
$\lambda$ in front of them. After all, 
the logarithm in (\ref{Hilbtau-3}) admits a uniformly convergent power series
expansion for any $z$ on $\gamma_\epsilon$. Every term of the
expansion decays
sufficiently fast for $z\to \infty$ to admit a closure of
$\gamma_\epsilon$ to the contour $\bar{\gamma}_\epsilon$ that avoids
$[\nu_D,\infty)$. We proceed by the residue theorem. This is simpler
than in \sref{Theorem}{thm:tau} because $\frac{1}{z-b}, \frac{1}{z-r}$ 
and $h_{0D}(z)$ are holomorphic in the interior 
of $\bar{\gamma}_\epsilon$ and on $\bar{\gamma}_\epsilon$ itself:
\begin{align*}
&\frac{1}{2\pi\mathrm{i}}
\int_{\gamma_{\epsilon}} dz \; 
\bigg(
\frac{\log \big(
\frac{1}{c_D}+ \frac{a-ac_D^{-1}+\kappa_D
+\lambda h_{0D}(z)
}{
\mu^2+a+z}\big)}{z-b}
-\delta_{D,4}
\frac{\log \big(
\frac{1}{c_D}+ \frac{r-rc_D^{-1}+\kappa_D
+\lambda h_{0D}(z)
}{
\mu^2+r+z}\big)}{z-r}\bigg)
\nonumber
\\
&=
-\sum_{n=1}^\infty \frac{(-\lambda c_D)^n}{n!}
\frac{d^{n-1}}{dw^{n-1}}\Big|_{w=0}
\Big(
H'_{a,b}(w) \big(\tfrac{a-ac_D^{-1}+\kappa_D}{\lambda}
+ h_{0D}(w-\mu^2-a)\big)^n
\nonumber
\\
&\qquad\qquad\qquad - 
\delta_{D,4} 
H'_{r,r}(w) \big(\tfrac{r-rc_D^{-1}+\kappa_D}{\lambda}
+ h_{0D}(w-\mu^2-r)\big)^n\Big)\;,
\end{align*}
where $H_{a,b}(w)=\log\big(\frac{w-\mu^2-a-b}{-\mu^2-a-b}\big)$. 
We apply the B\"urmann formula (\ref{eq:Buermann}). 
For that we need the auxiliary series
\[
-L_D(a):= \sum_{n=1}^\infty \frac{(-\lambda c_D)^n}{n!}
\frac{d^{n-1}}{dw^{n-1}}\Big|_{w=0}
\Big(\frac{a-ac_D^{-1}+\kappa_D}{\lambda}
+ h_{0D}(w-\mu^2-a)\Big)^n\;.
\]
In the same way as in the proof of (\ref{KDe}), 
the Lagrange inversion formula (\ref{eq:Lagrange}) yields 
\[
-L_D(a)=-\lambda c_D \Big(\frac{a-ac_D^{-1}+\kappa_D}{\lambda}
+ h_{0D}(-L(a)-\mu^2-a)\Big)\;,
\]
which by (\ref{eq:hDD}) and (\ref{eq:J}) rearranges into 
$a=R_D(a+L_D(a))$ for any  $D\in \{0,2,4\}$. 
We invert it to $a+L_D(a)=R_D^{-1}(a)$, but question this step for $D=4$ in 
\sref{Remark}{rem:triv}. The  
B\"urmann formula (\ref{eq:Buermann}) now gives
\begin{align}
\eqref{Hilbtau-13}
&
= \tilde{N}_D(a,b)
- H_{a,b}(-L_D(a))+\delta_{D,4}H_{r,r}(-L_D(r))
\nonumber
\\
&= \tilde{N}_D(a,b)
+ \log\Big(\frac{\mu^2+a+b}{\mu^2+b+R_D^{-1}(a)}\Big)
- \delta_{D,4}\log\Big(\frac{\mu^2+2r}{\mu^2+r+R_D^{-1}(r)}\Big)\;.
\label{Hilbtau-3f}
\end{align}
By the identity theorem for holomorphic functions, this equation holds
in the larger common $\lambda$-holomorphicity domain of both sides.  
It contains the original real coupling constant so that 
$R_D$ in (\ref{Hilbtau-3f}) extends to the situation
formulated in the proposition.

It remains to collect the pieces: We want to evaluate
(\ref{Gab-ansatz}). We set $Z=1$ in $D\in\{0,2\}$ and 
$Z=C_r e^{\mathcal{H}^{\Lambda}_r[\tau_r(\bullet)]}$ in $D=4$, where
$C_r$ is a finite number. We thus need the exponential of 
(\ref{Hilbtau-1all}), which is the exponential 
of (\ref{Hilbtau-2f}) times the exponential of 
(\ref{Hilbtau-3f}). This is to be multiplied by 
$\frac{\sin \tau_b(a)}{\lambda \pi \rho_0(a)}$ which cancels with 
the corresponding term in (\ref{Hilbtau-2f}):
\begin{align}
G(a,b):=
\frac{(\mu^2+a+b)\exp (\tilde{N}_D(a,b))}{
(\mu^2+b+R^{-1}_D(a))
(\mu^2+a+R^{-1}_D(b))}
\Big(C_r
\frac{\lambda\pi \rho_0(r)}{\sin \tau_r(r)}
\cdot \frac{\mu^2+2r}{(\mu^2+r+R_D^{-1}(r))^2}\Big)^{-\delta_{D,4}}
\;.
\label{Gab-prelim}
\end{align}
For $D\in\{0,2\}$ this 
already gives (\ref{Gab-final}) after reconstructing $R_D$ from $h_D$. 

As we will discuss in \sref{Remark}{rem:triv} after the proof, 
this equation is not appropriate for all cases of
$D=4$. We can already in \eqref{Hilbtau-13} deform the 
contour $\gamma_\epsilon$ to the straight line 
$-\frac{\mu^2}{2}+\mathrm{i}\mathbb{R}$. 
After trading $h_D$ in 
(\ref{tildeN}) for $R_D$ via (\ref{eq:J}), 
equation~(\ref{Hilbtau-3f}) can be written as
\begin{align*}
&\log\Big(\frac{\mu^2+a+b}{\mu^2+b+R_D^{-1}(a)}
\Big(\frac{\mu^2+2r}{\mu^2+r+R_D^{-1}(r)}\Big)^{-\delta_{D,4}}\Big)
\nonumber
\\
&= 
\frac{1}{2\pi\mathrm{i}}
\int_{-\infty}^\infty dt \; 
\Big\{
\log \Big(\frac{a-R_D(-\tfrac{\mu^2}{2}-\mathrm{i}t)}{
a-(-\tfrac{\mu^2}{2}-\mathrm{i}t)}\Big)
\frac{d}{dt} \log\big(b-(-\tfrac{\mu^2}{2}+\mathrm{i}t)\big)
\nonumber
\\
& \qquad\qquad-\delta_{D,4}
\log \Big(\frac{r-R_D(-\tfrac{\mu^2}{2}-\mathrm{i}t)}{
r-(-\tfrac{\mu^2}{2}-\mathrm{i}t)}\Big)
\frac{d}{dt} \log\big(r-(-\tfrac{\mu^2}{2}+\mathrm{i}t)\big)
\Big\}\;.
\end{align*}
Inserting this and its flip $a\leftrightarrow b$ back into 
(\ref{Gab-prelim}) gives rise to a representation where $R_D^{-1}$ 
is avoided completely:
\begin{align*}
G(a,b)&:=
\frac{\exp (\tilde{\tilde{N}}_D(a,b))}{(\mu^2+a+b)}
\Big(\frac{C_r}{\mu^2+2r}
\frac{\lambda\pi \rho_0(r)}{\sin \tau_r(r)}\Big)^{-\delta_{D,4}}
\nonumber
\\
\tilde{\tilde{N}}_D(a,b))
&=
\frac{1}{2\pi\mathrm{i}}
\int_{-\infty}^\infty dt \; 
\Big\{
\log \big(a-R_D(-\tfrac{\mu^2}{2}-\mathrm{i}t)\big)
\frac{d}{dt} \log\big(b-R_D(-\tfrac{\mu^2}{2}+\mathrm{i}t)\big)
\nonumber
\\
& \qquad\qquad
-\log \big(a-(-\tfrac{\mu^2}{2}-\mathrm{i}t)\big)
\frac{d}{dt} \log\big(b-(-\tfrac{\mu^2}{2}+\mathrm{i}t)\big)
\nonumber
\\
& \qquad\qquad
-\delta_{D,4}
\log \big(r-R_D(-\tfrac{\mu^2}{2}-\mathrm{i}t)\big)
\frac{d}{dt} \log\big(r-R_D(-\tfrac{\mu^2}{2}+\mathrm{i}t)\big)
\nonumber
\\
& \qquad\qquad 
+\delta_{D,4}\log \big(r-(-\tfrac{\mu^2}{2}-\mathrm{i}t)\big)
\frac{d}{dt} \log\big(r-(-\tfrac{\mu^2}{2}+\mathrm{i}t)\big)\Big\}\;.
\end{align*}
We can absorb the
$r$-dependent factors arising for $D=4$ 
by an appropriate choice of
$C_r$ and then adjust $C_r$ further to have $G(0,0)=1$. 
This amounts to replace $\tilde{\tilde{N}}_4(a,b)$ 
by $N_4(a,b):=\tilde{\tilde{N}}_4(a,b)-\tilde{\tilde{N}}_4(0,0)$.
\hfill $\square$%

\begin{rmk}
\label{rem:triv}
The representation (\ref{Gab-final}), renormalised to 
$\frac{G(a,b)}{G(0,0)}$, might fail for $4\leq \D<6$.
For finite $\Lambda$, as seen in the proof above, 
the representations 
(\ref{Gab-final}) and (\ref{Gab-D=4})+(\ref{ND}) are equivalent for $|\lambda|$ 
small enough.  But in the limit $\Lambda\to \infty$ it can happen that
$R_4(\mathbb{R}_+)$ defined by (\ref{eq:J}) develops an
upper bound for any $\lambda>0$, independently of whether $\rho_c$ is
discrete or continuous. In such a case
$R_4^{-1}(a)$ does not exist for all $a\in \mathbb{R}$ and  
(\ref{Gab-final}) becomes meaningless for $\Lambda\to \infty$,
whereas (\ref{Gab-D=4})+(\ref{ND}) do not show any problem.

In \sref{Sec.}{Sec.4dSol} we prove that for the measure function 
$\rho_0(t)=t$, of
spectral dimension exactly $\D=D=4$, there is no such problem. But other cases 
with $4<\D <6$ are very likely
affected. It is the identification $a+L_4(a)=R_4^{-1}(a)$ made before 
(\ref{Hilbtau-3f}) which might fail for $\Lambda\to \infty$.
For the same reasons, also $\tau_a(p)$ given in (\ref{tauap-final}) 
with (\ref{ID}) does not have a limit 
$\Lambda\to \infty$ for $4\leq\D<6$ and $\lambda>0$. Such problems have 
been noticed in \cite{Grosse:2014lxa}. They concern only 
auxiliary functions; the final result (\ref{Gab-D=4})+(\ref{ND}) 
is consistent for all $\lambda>0$.
\end{rmk}

\subsection[Solution for Finite Matrices $(\D=0)$]{Solution for Finite Matrices $(\D=0)$\footnote{This subsection is taken from our paper
\cite{Grosse:2019jnv}}}\label{sec.fm}
\begin{thrm}
\label{thm:main}
Consider the quartic matrix field theory model with the 
self-adjoint $\mathcal{N}\times\mathcal{N}$-matrix $E$ having distinct 
eigenvalues $0<e_1<e_2<\dots<e_{\mN'}$ of multiplicities $r_1,r_2,\dots,r_{\mN'}$. 
These data encode a meromorphic function
\begin{align*}
R(z):=z-\frac{\lambda}{V} \sum_{k=1}^{\mN'}
\frac{\varrho_k}{\varepsilon_k+z}\;,
\end{align*}
where $\{\varepsilon_k,\varrho_k\}_{k=1,\dots,{\mN'}}$ are the unique 
solutions in an open neighbourhood of $\lambda=0$ of 
\begin{align}
e_l  &= \varepsilon_l 
-\frac{\lambda}{V} \sum_{k=1}^{\mN'} \frac{\varrho_k}{
\varepsilon_k +\varepsilon_l}\;, &
1 &= \frac{r_l}{\varrho_l} 
-\frac{\lambda}{V} \sum_{k=1}^{\mN'} 
\frac{\varrho_k}{(\varepsilon_k+\varepsilon_l)^2}\;,
\quad \text{for }l=1,\dots,{\mN'}\;,
\label{er-implicit}
\end{align}
with  $\lim_{\lambda\to 0} \varepsilon_k=e_k$ and 
$\lim_{\lambda\to 0} \varrho _k=r_k$. For any $u \neq -\varepsilon_k$,
let $z\in \{u,\hat{u}^1,\dots ,\hat{u}^{\mN'}\}$ be the list of roots 
of $R(z)=R(u)$. Then the planar 2-point function
$G^{(0)}_{pq}$
satisfying the equation \eqref{eq:2Pphi4} for $g=0$ that extends into the complex plane, 
in an open neighbourhood of $\lambda=0$, 
 is solved by the \underline{rational} function 
\begin{align}
\mathcal{G}^{(0)}(z,w) &= 
\frac{\displaystyle \Bigg(
1-\frac{\lambda}{V}\sum_{k=1}^{\mN'} \frac{r_k}{(R(\varepsilon_k)-R(-w))
(R(z)-R(\varepsilon_k))} \prod_{j=1}^{\mN'} 
\frac{R(w)-R(-\widehat{\varepsilon_k}^j)}{R(w)-R(\varepsilon_j)}
\Bigg)}{
R(w)-R(-z)} 
\label{Gzw-rational}
\end{align}
with $G_{pq}^{(0)} =\mathcal{G}^{(0)}(\varepsilon_p,\varepsilon_q)$.
This function $\mathcal{G}^{(0)}(z,w)$ is symmetric in $z,w$ and
defined outside poles located at $z+w=0$, at
$z=\widehat{\varepsilon_k}^m$ and at $w=\widehat{\varepsilon_l}^n$,
for $k,l,m,n=1,\dots,\mN'$. 
\end{thrm}
\begin{rmk}
 A further proof for \sref{Theorem}{thm:main}, more intuitive, was found later and
 uses the right ansatz coming from the argumentation below \cite{Schurmann:2019mzu}.
 Here, an important tool is applied which derives the inverse of a
 Cauchy matrix, which is a matrix of the form $(\frac{1}{x_i-y_j})_{i,j}$.
 The inverse of a Cauchy matrix is given by an explicit algebraic expression, where the sum of its rows 
 and columns have a particular form \cite{MR105798}.
\end{rmk}
\begin{proof}

For the original problem of (finite) $\mathcal {N}\times 
\mathcal {N}$-matrices, the construction of the deformed measure is 
particularly transparent. It gives rise to a rational function $R$
for which the remaining integral of \sref{Theorem}{prop:HT} can be
evaluated. 

In dimension $\D=0$ the special treatment of the lowest eigenvalue $E_1
=\frac{\mu_{bare}^2}{2}=\frac{\mu^2}{2}$ is no longer necessary. 
The notation simplifies considerably when redefining 
$R(z):= \frac{\mu^2}{2}+R_0(z-\frac{\mu^2}{2})$.
Let $0<e_1<e_2<\dots<e_{\mN'}$ be the eigenvalues of $E$ and 
$r_1,\dots,r_{\mN'}$ their multiplicities, with 
$\sum_{k=1}^{\mN'} r_k=\mathcal{N}$. We shift the measure
to $\rho(t):=\rho_0(t-\frac{\mu^2}{2})$:
\begin{align*}
\rho(t)=\frac{1}{V}\sum_{k=1}^{\mN'} r_k \delta(t-e_k)\;.
\end{align*}
The deformed measure is according to (\ref{tilderho}) given by 
\begin{align*}
\rho_\lambda (x-\tfrac{\mu^2}{2})=\rho_0 (R_0(x-\tfrac{\mu^2}{2}))=
\rho (R(x))=
\frac{1}{V} \sum_{k=1}^{\mN'} \frac{r_k}{R'(R^{-1}(e_k))} \delta(x-
R^{-1}(e_k))\;,
\end{align*}
where $R_0$, and thus $R$, arises via (\ref{eq:J}) and (\ref{eq:hDD}) 
from the same measure $\rho_\lambda$:
\begin{align}
R(z)=z-\frac{\lambda}{V} \sum_{k=1}^{\mN'}
\frac{\varrho_k}{\varepsilon_k+z}\;,\qquad 
\varrho_k:= \frac{r_k}{R'(R^{-1}(e_k))}\;,\quad
\varepsilon_k := R^{-1}(e_k)\;.
\label{Jx-N}
\end{align}
This equation and its derivative evaluated at $z_l=R^{-1}(e_l)=\varepsilon_l$ 
for $l=1,\dots,{\mN'}$ provide a system of $2{\mN'}$ equations for the $2{\mN'}$ parameters 
$\{\varepsilon_k,\varrho_k\}$:
\begin{align}
e_l  &= \varepsilon_l 
-\frac{\lambda}{V} \sum_{k=1}^{\mN'}
\frac{\varrho_k}{
\varepsilon_k +\varepsilon_l}\;, &
1 &= \frac{r_l}{\varrho_l} 
-\frac{\lambda}{V} \sum_{k=1}^{\mN'} 
\frac{\varrho_k}{(\varepsilon_k+\varepsilon_l)^2}\;.
\label{rho-finite}
\end{align}
The implicit function theorem guarantees a solution in an 
open $\lambda$-interval, and one explicitly constructs a sequence 
converging to the solution $\{\varepsilon_k,\varrho_k\}$. 
Alternatively, (\ref{rho-finite}) can be interpreted as a system of 
$2{\mN'}$ polynomial equations  (${\mN'}$ of them 
of degree ${\mN'}+1$, the other ${\mN'}$ of degree $2{\mN'}+1$). Such system 
have many solutions, and they will indeed be needed in intermediate
steps. The right solution
is the one which for $\lambda\to 0$ converges to $\{e_k,r_k\}$.


Recall from (\ref{Gab-continuation}) that $G^{(0)}_{pq}=G(x,y)$ where 
$x+\frac{\mu^2}{2}=R(\varepsilon_p)$ and
$y+\frac{\mu^2}{2}=R(\varepsilon_q)$. The ansatz
(\ref{Gab-ansatz}) for $G(x,y)$ is turned with 
(\ref{Hilbtau-1all})
and (\ref{Hilbtau-2f}) into the representation
\begin{align}
G^{(0)}_{pq}
=\mathcal{G}^{(0)}(\varepsilon_p,\varepsilon_q),\quad\!
\mathcal{G}^{(0)}(u,v) = \frac{1}{R(u)+v}
\exp\Big(\lim_{\epsilon \to 0} 
\frac{1}{2\pi \mathrm{i}} \int_{\gamma'_\epsilon}\!\!
dz\; R'(z) \frac{\log (1+\frac{\lambda h_0(z-\mu^2/2)}{
z+R(u)})}{R(z)-R(v)}\Big).
\label{calGdef}
\end{align}
Here, the integration variable $z$ in (\ref{Hilbtau-13}) 
is shifted into $z+\frac{\mu^2}{2}\mapsto z$, and $\gamma_\epsilon'$ is the
shifted contour which encircles $[\nu_0+\frac{\mu^2}{2},\infty)$.
We have  $\lambda h_0(z-\mu^2/2)=-z-R(-z)$ from (\ref{eq:J}). 

For any $v,z\notin \{-\varepsilon_1,\dots,-\varepsilon_{\mN'}\}$, we can expand 
the rational function $R(z)-R(v)$ according to (\ref{Jx-N}) into 
\begin{align}
R(z)-R(v)
=(z-v) \prod_{k=1}^{\mN'} \frac{z-\hat{v}^k}{z+\varepsilon_k}\;.
\label{J-rational}
\end{align}  
Here, $\hat{v}^1,\dots,\hat{v}^{\mN'}$ are the other roots of the numerator
polynomial; they are functions of $v$ and the initial data
$E,\lambda$. For real $v$ it follows from the intermediate value theorem 
that these roots are interlaced between the poles
$\{-\varepsilon_k\}$ of $R$. In particular, for $v\geq 0$ and
$\lambda>0$ all $\hat{v}^k$ are real and located in 
$-\varepsilon_{k+1} < \hat{v}^k < -\varepsilon_{k}$ for $k=1,\dots {\mN'}-1$
and $\hat{v}^{\mN'} < -\varepsilon_{\mN'}$.

Inserting (\ref{J-rational}) into 
(\ref{calGdef}) gives
\begin{align}
\mathcal{G}^{(0)}(u,v) &= \frac{1}{R(u)+v}
\exp\Big\{\lim_{\epsilon \to 0} \frac{1}{2\pi \mathrm{i}}
\int_{\gamma'_\epsilon}
dz\; 
\Big(\frac{1}{z-v}
+\sum_{k=1}^{\mN'} \frac{1}{z-\hat{v}^k}
-\sum_{k=1}^{\mN'} \frac{1}{z+\varepsilon_k}
\Big)
\nonumber
\\
&\qquad\qquad\qquad\qquad \qquad\qquad\qquad 
\times \log \Big(1+ \frac{\lambda}{z+R(u)} \cdot \frac{1}{V}
\sum_{l=1}^{\mN'} \frac{\varrho_l}{\varepsilon_l-z}
\Big)\Big\}\;.
\label{calGdef-1}
\end{align}
\begin{lemma}
For $u,v>0$, a posteriori extented to a neighbourhood 
of $[\nu_0+\frac{\mu^2}{2},\infty)$, one has 
\begin{align}
\mathcal{G}^{(0)}
(u,v) &= \frac{1}{R(v)-R(-u)} 
\prod_{k=1}^{\mN'}
\frac{R(u)-R(-\hat{v}^k)}{R(u)-R(\varepsilon_k)}\;.
\label{calG-branch}
\end{align}
\end{lemma}
\noindent
\emph{Proof.}
As before, for finite $\epsilon$ and for $\lambda$ in a small open ball, 
the logarithm in (\ref{calGdef-1}) can be
expanded. After closing the integration contour,
the residue theorem picks up the simple poles at 
$z= \hat{v}^k$ and 
$z=- \varepsilon_k$ and the poles of $n$-th order at 
$z=-R(v)$. The other candidates $z=v$ and $z=\varepsilon_k$ are
outside the contour for real $u,v$. The poles of $n$-th order
combine (up to a global sign) to a B\"urmann formula
(\ref{eq:Buermann}) for
$H_{u,v}(w):=
\log\big(\frac{R(w-R(u))-R(v)}{
R(-R(u))-R(v)}\big)$, where $w\mapsto -L(u)$ solves the auxiliary
integral 
\begin{align*}
-L(u)&:=-\lim_{\epsilon \to 0} \frac{1}{2\pi \mathrm{i}}
\int_{\gamma'_\epsilon}
dz\; 
\log \Big(1+ \frac{\lambda}{z+R(u)} \cdot \frac{1}{V}
\sum_{l=1}^{\mN'} \frac{\varrho_l}{\varepsilon_l-z}
\Big)\Big\}
\\
&= \sum_{n=1}^\infty \frac{(-\lambda)^n}{n!} 
\frac{d^{n-1}}{dt^{n-1}}\Big|_{t=0} 
\Big(\frac{1}{V}
\sum_{l=1}^{\mN'} \frac{\varrho_l}{\varepsilon_l+R(u)-t}
\Big)^n\;.
\end{align*}
The Lagrange inversion formula (\ref{eq:Lagrange}) gives
$-\lambda \frac{1}{V}
\sum_{l=1}^{\mN'} \frac{\varrho_l}{\varepsilon_l+R(u)+L(u)}
=-L(u)$, which is solved by 
$R(u)+L(u)=u$. Putting everything together, the integral  
(\ref{calGdef-1}) evaluates to 
\begin{align*}
\mathcal{G}^{(0)}(u,v) &= \frac{1}{R(u)+v}
\cdot \frac{R(-R(u))-R(v)}{R(-u)-R(v)} 
\prod_{k=1}^{\mN'}
\frac{R(u)-R(-\hat{v}^k)}{\hat{v}^k+R(u)} 
\prod_{k=1}^{\mN'}
\frac{R(u)-\varepsilon_k}{R(u)-R(\varepsilon_k)}\;.
\end{align*}
The identity (\ref{J-rational}) applied for $z=-R(u)$ simplifies this to 
(\ref{calG-branch}).\hspace*{\fill} 
\end{proof}
\bigskip\noindent
The representation (\ref{calG-branch}) is rational in the first variable. 
There are two ways to
proceed. First, we can expand (\ref{calG-branch}) via
(\ref{J-rational}) to
\begin{align}
\mathcal{G}^{(0)}
(u,v) &= \frac{\prod_{k=1}^{\mN'} (u-\varepsilon_k)}{
(u+v)\prod_{k=1}^{\mN'} (u+\hat{v}^k)}\!
\prod_{k=1}^{\mN'} 
\frac{
(u+\hat{v}^k)\prod_{l=1}^{\mN'} 
(\hat{u}^l+\hat{v}^k)
}{
\prod_{l=1}^{\mN'} (\hat{v}^k-\varepsilon_l)}
\prod_{k=1}^{\mN'} \frac{\prod_{l=1}^{\mN'} 
(\varepsilon_k+\varepsilon_l)}{
(u-\varepsilon_k)\prod_{l=1}^{\mN'} 
(\varepsilon_k-\hat{u}^l)}
\nonumber
\\
&=\frac{1}{u+v} \prod_{k,l=1}^{\mN'} \frac{(\varepsilon_k+\varepsilon_l)
(-\hat{u}^l-\hat{v}^k)}{
(\varepsilon_k-\hat{u}^l)(\varepsilon_l-\hat{v}^k)}\;.
\label{calG-branch-1}
\end{align}
This formula is manifestly symmetric in $u,v$ --- a crucial property below. But 
it needs all roots of $R$, which exist only in a neighbourhood of
$[\nu_0+\frac{\mu^2}{2},\infty)$, not globally. 

The limit $u\to \varepsilon_p$ of (\ref{calG-branch}) 
gives with $r_p=\varrho_p R'(\varepsilon_p)$:
\begin{cor} 
\label{Cor:Gab}
For any $p=1,\dots,{\mN'}$ and $v$ in a neighbourhood of $\mathbb{R}_+$ 
one has
\begin{align}
-\frac{\lambda}{V} r_p \mathcal{G}^{(0)}(\varepsilon_p,v)
= 
\frac{\prod_{k=1}^{\mN'} (R(\varepsilon_p)-R(-\hat{v}^k))}{
\prod_{p\neq j=1}^{\mN'} (R(\varepsilon_p)-R(\varepsilon_j))}\;.
\label{corGev}
\end{align}
In particular, for any $p,q=1,\dots,{\mN'}$ one has
\begin{align}
\mathcal{G}^{(0)}(\varepsilon_p,\varepsilon_q)
= -\frac{V}{\lambda r_p}
\frac{\prod_{k=1}^{\mN'} (R(\varepsilon_p)-R(-\widehat{\varepsilon_q}^k))}{
\prod_{p\neq j=1}^{\mN'} (R(\varepsilon_p)-R(\varepsilon_j))}
= -\frac{V}{\lambda r_q}
\frac{\prod_{k=1}^{\mN'} (R(\varepsilon_q)-R(-\widehat{\varepsilon_p}^k))}{
\prod_{q\neq j=1}^{\mN'} (R(\varepsilon_q)-R(\varepsilon_j))}\;.
\label{corGee}
\end{align}
\begin{proof}
Next, we recall the basic lemma
\begin{align}
\sum_{j=0}^{{\mN'}} 
\frac{\prod_{k=1}^{\mN'}(x_j-c_k)}{
\prod_{j\neq k=0}^{{\mN'}} (x_j-x_k)}=1\;,
\label{basiclemma}
\end{align}
valid for pairwise different $x_0,\dots,x_{{\mN'}}$ and any
$c_1,\dots,c_{\mN'}$ (The rational function of $x_0$ has potential 
simple poles at $x_0=x_k$, $k=1,\dots, {\mN'}$, but all residues cancel. Hence, 
it is an entire function of $x_0$, by symmetry in all $x_k$. The
behaviour for $x_0\to \infty$ gives the assertion.).
We use (\ref{basiclemma}) for $x_0=R(u)$,
$x_k=R(\varepsilon_k)$ and 
$c_k=R(-\hat{v}^k)$ to rewrite (\ref{calG-branch}) as
\begin{align}
\mathcal{G}^{(0)}(u,v) &= \frac{1}{R(v)-R(-u)} 
\Big(1+\sum_{k=1}^{\mN'} \frac{1}{R(u)-R(\varepsilon_k)} 
\frac{\prod_{l=1}^{\mN'}( R(\varepsilon_k)-R(-\hat{v}^l))}{
\prod_{k\neq j=1}^{\mN'} (R(\varepsilon_k)-R(\varepsilon_j))}\Big)
\nonumber
\\
&= \frac{1}{R(v)-R(-u)} 
\Big(1+\frac{\lambda}{V}
\sum_{k=1}^{\mN'} \frac{r_k \mathcal{G}^{(0)}(\varepsilon_k,v)}{
R(\varepsilon_k)-R(u)} 
\Big)\;.
\label{calG-sum}
\end{align}
The second line results from (\ref{corGev}).
Using the symmetry 
$\mathcal{G}^{(0)}(\varepsilon_k,v)=\mathcal{G}^{(0)}(v,\varepsilon_k)$, 
the  previous formulae give rise to a representation of 
$\mathcal{G}^{(0)}(u,v)$ which is \emph{rational in both variables}. 
The assertion 
(\ref{Gzw-rational}) in \sref{Theorem}{thm:main} follows  
from symmetry 
$\mathcal{G}(\varepsilon_k,w)=\mathcal{G}(w,\varepsilon_k)$ and 
insertion of $\mathcal{G}(w,\varepsilon_k)$ given by (\ref{calG-branch})
into (\ref{calG-sum}). We could also insert 
the symmetrised version of (\ref{calG-sum}), 
\begin{align}
\mathcal{G}^{(0)}
(\varepsilon_k,v) &= \frac{1}{R(\varepsilon_k)-R(-v)} 
\Big(1+\frac{\lambda}{V} 
\sum_{l=1}^{\mN'} \frac{r_l \mathcal{G}^{(0)}(\varepsilon_l,\varepsilon_k)}{
R(\varepsilon_l)-R(v)} \Big)\;,
\label{calG-sumab}
\end{align}
back into (\ref{calG-sum}). 
The remaining assertion of \sref{Theorem}{thm:main} about the 
poles of $\mathcal{G}^{(0)}(z,w)$ will be established in 
\sref{Proposition}{prop:RFE} below.
\end{proof}
\end{cor}

\begin{rmk}\label{rmk:2p}
For any $z\neq \varepsilon_k$ one has 
\begin{align}
R(z) +\frac{\lambda}{V}  \sum_{k=1}^{\mN'} r_k
\mathcal{G}^{(0)}(z,\varepsilon_k)
+\frac{\lambda}{V}  \sum_{k=1}^{\mN'}
\frac{r_k}{R(\varepsilon_k)-R(z)}
=-R(-z)\;.
\label{Jzz}
\end{align}
We already know this identity. 
The original equation (\ref{Gab-orig}) 
for $G^{(0)}_{pq}
=\mathcal{G}^{(0)}(\varepsilon_p,\varepsilon_q)$ extends to 
complex variables $\varepsilon_p \mapsto z$ and 
$\varepsilon_q\mapsto w$, with $w+z\neq 0$ and 
$w,z \notin \{\varepsilon_k\}_{k=1,\dots {\mN'}} \cup 
\{\widehat{\varepsilon_k}^l\}_{k,l=1,\dots {\mN'}}$:
\begin{align}
&\Big\{R(z)+R(w) +\frac{\lambda}{V}  \sum_{k=1}^{\mN'} r_k
\mathcal{G}^{(0)}(z,\varepsilon_k)
+\frac{\lambda}{V}  \sum_{k=1}^{\mN'}
\frac{r_k}{R(\varepsilon_k)-R(z)}\Big\} \mathcal{G}^{(0)}(z,w)
\nonumber
\\[-1ex]
&= 1+ 
\frac{\lambda}{V}  \sum_{k=1}^{\mN'}
\frac{r_k\,\mathcal{G}^{(0)}(\varepsilon_k,w)}{
R(\varepsilon_k)-R(z)} \;.
\label{Gzw-orig}
\end{align}
Now (\ref{Jzz}) follows by comparison with (\ref{calG-sum}). 

Equation (\ref{Jzz}) has also been established for
H\"older-continuous measure in 
\sref{Theorem}{thm:tau}. Namely, when expressed in terms of the 
angle function and variables 
$a+\frac{\mu^2}{2}=R(\varepsilon_p)$ and
$b+\frac{\mu^2}{2}=R(\varepsilon_q)$, the terms 
$\{~\}$ in the first
line of (\ref{Gzw-orig}) 
become $\lambda\pi\rho_0(a) \cot \tau_b(a)$. In (\ref{tauReIm})
we had found 
$\lambda\pi\rho_0(a) \cot \tau_b(a)
= \lim_{\epsilon\to 0} \mathrm{Re}(b+I_0(a+\mathrm {i}\epsilon))$,
which translates into 
$\lambda\pi\rho_0(a) \cot \tau_b(a)=R(w)-R(-z)$ for $z\neq
\varepsilon_k$. From that starting point we had 
derived (\ref{calG-branch}) so that finding 
(\ref{Jzz}) from (\ref{calG-sum}) is no surprise.

But there is another line of arguments. We could have started with
(\ref{calG-branch}) as an ansatz, from which alone we arrive at
(\ref{calG-sum}). If we could also prove (\ref{Jzz}) from
(\ref{calG-branch}) alone, then (\ref{Gzw-orig}) is a consequence of
the ansatz (\ref{calG-branch}), and we have proved that
(\ref{calG-sum}) solves (\ref{Gzw-orig}). To directly verify
(\ref{Jzz}) as identity for rational functions, note that both sides
approach $z$ for $z\to \infty$. The rhs has poles only at
$z=\varepsilon_k$ with residue $\frac{\lambda}{V}
\varrho_k$. The same poles with the same residues also arise on the
lhs, taking $r_k/R'(\varepsilon_k)=\varrho_k$ into account. But the
lhs also has potential poles at $z=-\varepsilon_j$ and at all
$z=\widehat{\varepsilon_j}^n$.  We have $\Res_{z\to -\varepsilon_j}
R(z)= -\frac{\lambda}{V} \varrho_j$. Taking
(\ref{calG-branch}) for $\mathcal{G}^{(0)}(z,\varepsilon_l)$ in which
we have $\lim_{z\to -\varepsilon_j}
\frac{R(z)-R(-\widehat{\varepsilon_l}^k)}{ R(z)-R(\varepsilon_k)}=1$
for any $k,l$, one easily finds that
$\mathcal{G}^{(0)}(-\varepsilon_j,\varepsilon_l)$ is regular for
$j\neq l$ and that $\Res_{z\to -\varepsilon_j}
\frac{\lambda}{V} r_j\mathcal{G}^{(0)}(z,\varepsilon_j)=
\frac{\lambda}{V} \frac{r_j}{R'(\varepsilon_j)}$, which thus
cancels $\Res_{z\to -\varepsilon_j} R(z)=
-\frac{\lambda}{V} \varrho_j$.

Finally, from (\ref{calG-sum}) we conclude 
\[
\Res\displaylimits_{z\to \widehat{\varepsilon_j}^n}
\mathcal{G}^{(0)}(z,\varepsilon_k)=
-\frac{\lambda r_j }{V R'(\widehat{\varepsilon_j}^n)} 
\frac{\mathcal{G}^{(0)}(\varepsilon_j,\varepsilon_k)}{R(\varepsilon_k)
-R(-\widehat{\varepsilon_j}^n)}
=\! \frac{r_j}{R'(\widehat{\varepsilon_j}^n) r_k} \frac{\prod_{n \neq l=1}^{\mN'}
  (R(\varepsilon_k)-R(-\widehat{\varepsilon_j}^l))}{
\prod_{k\neq m=1}^{\mN'}   (R(\varepsilon_k)-R(\varepsilon_m))},
\]
where (\ref{corGee}) together with $\varrho_kR'(\varepsilon_k)=r_k$ 
has been used in the second equality.
The basic lemma (\ref{basiclemma}) in ${\mN'}$ variables $x_k=R(\varepsilon_k)$
gives 
$\mathrm{Res}_{z\to \widehat{\varepsilon_j}^n}
\sum_{k=1}^{\mN'} r_k \mathcal{G}^{(0)}(z,\varepsilon_k)=
\frac{r_j}{R'(\widehat{\varepsilon_j}^n)}$, 
which precisely cancels 
$\Res_{z\to \widehat{\varepsilon_j}^n}
\sum_{k=1}^{\mN'}
\frac{r_k}{R(\varepsilon_k)-R(z)}
= -\frac{r_j}{R'(\widehat{\varepsilon_j}^n)}
$. In summary, (\ref{Jzz}) is a corollary of (\ref{calG-branch}).
\hspace*{\fill} $\square$%
\end{rmk}

\begin{prps}
\label{prop:RFE}
The planar 2-point function has the (manifestly symmetric) 
rational fraction expansion 
\begin{align*}
\mathcal{G}^{(0)}(z,w)
&=\frac{1}{z+w}\bigg(1+\frac{\lambda^2}{V^2}
\sum_{k,l,m,n=1}^{\mN'}
\frac{C_{k,l}^{m,n}}{(z-\widehat{\varepsilon_k}^m)(w-\widehat{\varepsilon_l}^n)}
\bigg)\;,
\\
C_{k,l}^{m,n}&:= 
\frac{(\widehat{\varepsilon_k}^m +\widehat{\varepsilon_l}^n) 
r_k r_l \mathcal{G}^{(0)}(\varepsilon_k,\varepsilon_l)}{
R'(\widehat{\varepsilon_k}^m)R'(\widehat{\varepsilon_l}^n)
(R(\varepsilon_l)-R(-\widehat{\varepsilon_k}^m))
(R(\varepsilon_k)-R(-\widehat{\varepsilon_l}^n))
}\;.\nonumber
\end{align*}
\end{prps}
\noindent
\emph{Proof.} Expanding the first denominator in
(\ref{calG-sum}) via (\ref{J-rational}), $\mathcal{G}^{(0)}(u,v)$ has
potential poles at $u=-\hat{v}^n$ for every $n=1,\dots,{\mN'}$. 
However, for $u=-\hat{v}^n$ the sum in the first line of
(\ref{calG-sum}) becomes 
$\sum_{k=1}^{\mN'} \frac{1}{(R(-\hat{v}^n)-R(\varepsilon_k))}
\frac{\prod_{l=1}^{\mN'}( R(\varepsilon_k)-R(-\hat{v}^l))}{
\prod_{k\neq j=1}^{\mN'} (R(\varepsilon_k)-R(\varepsilon_j))}
=-1$ when using the basic lemma (\ref{basiclemma}). Consequently,
$\mathcal{G}^{(0)}(z,w)$ is regular at $z=-\hat{w}^n$ and by symmetry
at $w=-\hat{z}^n$. 

This leaves the diagonal $z+w=0$ and 
the complex lines ($z=
\widehat{\varepsilon_k}^m$, any $w$) and
($w=\widehat{\varepsilon_l}^n$, any $z$)
as the only possible poles of $\mathcal{G}^{(0)}(z,w)$.
The function $(z+w)\mathcal{G}^{(0)}(z,w)$ approaches 
$1$ for $z,w\to \infty$. Its residues at $
z=\widehat{\varepsilon_k}^m ,w=\widehat{\varepsilon_l}^n$ are
obtained from (\ref{calG-sum}):
\begin{align*}
&\Res\displaylimits_{z\to \widehat{\varepsilon_k}^m, w\to
  \widehat{\varepsilon_l}^n}
(z+w)\mathcal{G}^{(0)}(z,w)
\nonumber
\\
&= -\frac{(\widehat{\varepsilon_k}^m +\widehat{\varepsilon_l}^n)}{
(R(\varepsilon_l)-R(-\widehat{\varepsilon_k}^m))} 
\frac{\lambda r_k}{V R'(\widehat{\varepsilon_k}^m)}
\Res\displaylimits_{w\to  \widehat{\varepsilon_l}^n}
\mathcal{G}^{(0)}(\varepsilon_k,w)
\nonumber
\\
&=\Big(\frac{\lambda}{V}\Big)^2 
\frac{(\widehat{\varepsilon_k}^m +\widehat{\varepsilon_l}^n)
r_k r_l 
\mathcal{G}^{(0)}(\varepsilon_k,\varepsilon_l)
}{R'(\widehat{\varepsilon_k}^m)R'(\widehat{\varepsilon_l}^n)
(R(\varepsilon_l)-R(-\widehat{\varepsilon_k}^m))
(R(\varepsilon_k)-R(-\widehat{\varepsilon_l}^n))
} \;.
\end{align*}
The second line follows from (\ref{calG-sumab}). \hspace*{\fill}$\square$

\begin{exm}
 The extreme case of a single $r_1=\mathcal{N}$-fold degenerate
eigenvalue $E=\frac{\mu^2}{2}\cdot \mathrm{id}$ corresponds to a 
standard Hermitian 1-matrix model with action 
$S[\Phi]=V\,\mathrm{Tr}( \frac{\mu^2}{2}\Phi^2+
\frac{\lambda}{4}\Phi^4)$. This purely quartic case 
was studied in \cite{Brezin:1977sv}. Transforming 
$M \mapsto \sqrt{V} \mu \Phi$ and $g=\frac{\lambda}{4\mu^4}$ and $V=\mN$
brings \cite[eq.\ (3)]{Brezin:1977sv} into our conventions.
The equations (\ref{er-implicit}) reduce for $E_1=\frac{\mu^2}{2}$ 
and $\mN'=1$ to
\begin{align}
\frac{\mu^2}{2}  &=\varepsilon_1
- \frac{\lambda\varrho_1}{\mathcal{N}(2\varepsilon_1)}\;, 
&
1 &= \frac{\mathcal{N}}{\varrho_1} 
- \frac{\lambda\varrho_1}{\mathcal{N}(2\varepsilon_1)^2}
\label{rho-degenerate}
\end{align}
with principal solution (i.e.\ $\lim_{\lambda\to 0}
\varepsilon_1=\frac{\mu^2}{2}$)
\begin{align}
\varepsilon_1&=\frac{1}{6} \big(2\mu^2+\sqrt{\mu^4+12\lambda}\big) \;,&
\varrho_1&=\mathcal{N}\cdot  
\frac{\mu^2\sqrt{\mu^4+12\lambda}-\mu^4+ 12\lambda}{18 \lambda}\;.
\end{align}
The other root $\widehat{\varepsilon_1}^1$ with 
$R(\widehat{\varepsilon_1}^1)
=\widehat{\varepsilon_1}^1
- \frac{\lambda\varrho_1}{\mathcal{N}(\varepsilon_1
+\widehat{\varepsilon_1}^1)}=R(\varepsilon_1)=\frac{\mu^2}{2}$ 
is found to be
\begin{align}
\widehat{\varepsilon_1}^1=-\frac{1}{6}
\big(\mu^2+2\sqrt{\mu^4+12\lambda}\big)= 
\tfrac{\mu^2}{2}-2\varepsilon_1 \;.
\end{align}
The planar 2-point function $G_{11}^{(0)}$ can be evaluated via 
(\ref{corGee}) or  (\ref{calG-branch-1}) to 
\begin{align}
G_{11}^{(0)}= -\frac{1}{\lambda}\Big(\frac{\mu^2}{2}
-R(-\widehat{\varepsilon_1}^1)\Big)=
\frac{4}{3} \cdot  \frac{\mu^2+2\sqrt{\mu^4+12\lambda}}{
(\mu^2+\sqrt{\mu^4+12\lambda})^2}= 
-\frac{2\widehat{\varepsilon_1}^1}{(\varepsilon_1-
\widehat{\varepsilon_1}^1)^2}\;.
 \label{G000}
\end{align}
The result can be put into 
$G_{11}^{(0)}=\frac{1}{3\mu^2} a^2(4-a^2)$
for $a^2=\frac{2\mu^2}{\mu^2+\sqrt{\mu^4+12\lambda}}$ 
and thus agrees with the literature: 
This value for $a^2$, which
corresponds to $\frac{a^2 \lambda}{\mu^2}=\varepsilon_1-\frac{\mu^2}{2}$,
solves
\cite[eq.~(17a)]{Brezin:1977sv} for $g:=\frac{\lambda}{4\mu^4}$ 
so that (\ref{G000}) reproduces\footnote{In \cite{Brezin:1977sv} expectation values of traces 
$\langle \mathrm{Tr}(M^{2p})\rangle$ are studied, whereas we consider 
$\langle M_{11} M_{11}\rangle$, see also \sref{Remark}{rmk:hermitian}} 
\cite[eq.~(27)]{Brezin:1977sv} for $p=1$ (and the convention 
$G_{11}^{(0)}=\frac{1}{\mu^2}$ for $\lambda=0$). 

The meromorphic extension $\mathcal{G}^{(0)}(z,w)$ is most conveniently 
derived from \sref{Proposition}{prop:RFE} after 
cancelling the 
two representations 
(\ref{G000}) for
$G_{11}^{(0)}=\mathcal{G}^{(0)}(\varepsilon_1,\varepsilon_1)$:
\begin{align}
\mathcal{G}^{(0)}(z,w)=\frac{1}{z+w}\Big(1-
\frac{(\varepsilon_1+\widehat{\varepsilon_1}^1)^2}{
(z-\widehat{\varepsilon_1}^1)(w-\widehat{\varepsilon_1}^1)}\Big)
=
\frac{1}{z+w}\Big(1-\frac{\mu^4(1-a^2)^2}{
(3a^2 z+\mu^2)(3a^2 w+\mu^2)}\Big)\;.
\end{align}
We have used $R'(\widehat{\varepsilon_1}^1)=
\frac{\widehat{\varepsilon_1}^1-\varepsilon_1}{
\widehat{\varepsilon_1}^1+\varepsilon_1}$.
\end{exm}

\begin{cor}\label{cor:f0}
 Let the planar free energy be $F^{(0)}$ with $F=:\sum_{g=0}^\infty V^{-2g} F^{(g)}$ and 
 $F:=V^{-2}\log \Z[0]$. With the distinct eigenvalues $e_q$ of multiplicity $r_q$ of $E$, 
 we have
 \begin{align*}
  \frac{V}{r_q}\frac{\partial}{\partial e_q}F^{(0)}=\frac{e_q-\varepsilon_q}{\lambda}+\frac{1}{V}\sum_{k=0}^{\mN'}r_k
  \bigg(\frac{1}{e_k-e_q}-\frac{1}{R'(\varepsilon_k)(\varepsilon_k-\varepsilon_q)}\bigg).
 \end{align*}
\begin{proof}
 Combing \sref{Corollary}{Cor:quartF} after genus expansion with \eqref{Jzz} with $z=e_q$ gives the assertion.
\end{proof}
\end{cor}
\noindent Integrating the rhs of the equation of \sref{Corollary}{cor:f0} wrt 
$e_q$, multiplying with $\frac{r_q}{V}$ and summing over $q$ gives in principle the planar
free energy $F^{(0)}$. This integration is very hard to perform,
since all $\varepsilon_k$ and $R'(\varepsilon_k)$ 
depend intrinsically on $e_q$.

\subsection[Solution on the $D=4$ Moyal Space]{Solution on the $D=4$ Moyal Space
\footnote{This subsubsection is taken from our paper \cite{Grosse2020}}}\label{Sec.4dSol}
On the $\D=D=4$ Moyal space the eigenvalues increase linearly, $e(x)=x$, with the measure $r(x)=x$ 
which induces for the undeformed measure $\varrho_0(x)=r(x)=x$. 
The deformed measure $\varrho_\lambda$ is then, due to \sref{Definition}{def:rho}, given by
\begin{align*}
 \varrho_\lambda(x)=\varrho_0(R_4(x))=R_4(x).
\end{align*}
On the other hand, the function $R_4(x)$ was defined in \sref{Definition}{def:J} such that the deformed measure 
obeys a further integral equation
\begin{align}\label{Fred}
 \varrho_\lambda(x)=x-\lambda x^2\int_0^\infty \frac{dt\varrho_\lambda(t)}{(t+\mu^2)^2 (t+\mu^2+x)}.
\end{align}
This integral equation is a linear integral equation and of Fredholm type. Finding the solution of 
\eqref{Fred} solves together with \sref{Theorem}{prop:HT} the noncommutative $\Phi^4_4$-QFT model exactly.
\begin{prps}\label{Prop:Jx-final}
 Equation \eqref{Fred} is solved by 
 \begin{align}
\label{Jx-final}
 \varrho_\lambda (x)=x  \;_2F_1\Big(\genfrac{}{}{0pt}{}{
\alpha_\lambda,\;1-\alpha_\lambda}{2}\Big|-\frac{x}{\mu^2}\Big),\qquad
\text{where} \quad
\alpha_\lambda=\left\{
\begin{array}{cl}
\dfrac{\arcsin(\lambda \pi)}{\pi} & \text{for }
|\lambda| \leq \frac{1}{\pi} ,\\
\dfrac{1}{2}+\mathrm{i} \dfrac{\mathrm{arcosh}(\lambda \pi)}{\pi} & \text{for }
\lambda \geq \frac{1}{\pi}.
\end{array}\right.
\end{align}
Moreover, the particular choice
$\mu^2=\frac{\alpha_\lambda(1-\alpha_\lambda)}{\lambda}$ provides 
a natural choice for the boundary condition
differently from the condition ($\frac{\partial}{\partial a}G^{(0)}(a,0)\vert_{a=0}= -1$) 
induced by Zimmermann's forest formula.
\begin{proof}
It is convenient to symmetrise the Fredholm equation \eqref{Fred}.
Dividing by $\frac{x}{\mu^2+x}$ and defining
$\tilde{\varrho}_\lambda(x):=\frac{R_4(x)}{x(\mu^2+x)}$, we have
\begin{align}
 \tilde{\varrho}_\lambda(x)&=\frac{1}{\mu^2+x}
-\lambda \int_0^\infty \!\! dt\;\frac{\tilde{\varrho}_\lambda(t)\,tx}{(\mu^2+t)
 (\mu^2+x)(\mu^2+x+t)}
\nonumber
\\*
&= \frac{c_\lambda}{\mu^2+x}-\lambda \int_0^\infty \!\! dt\;
 \frac{\tilde{\varrho}_\lambda(t)}{\mu^2+x+t},\label{intc}
\end{align}
where $c_\lambda =1+\lambda\mu^2\int_0^\infty dt\;
\frac{\tilde{\varrho}_\lambda(t)}{\mu^2+t}=
1+\lambda\mu^2\int_0^\infty dt\;\frac{R_4(t)}{t(\mu^2+t)^2}$.
The second line results by (not so obvious) rational fraction expansion.
As proved in \sref{App.}{app},
there exists for $\lambda>-\frac{1}{\pi}$
a solution $\tilde{\varrho}_\lambda\in L^2(\mathbb{R}_+)$,
which means $\lim_{t\to \infty} t \tilde{\varrho}_\lambda (t)=0$.
Another transformation $\phi(x)=\mu^2\tilde{\varrho}_\lambda(x\mu^2)$
simplifies the problem to
\begin{align}\label{feq}
\phi(x)=\frac{c_\lambda}{1+x}
-\lambda \int_0^\infty \!\! dt\; \frac{\phi(t)}{1+t+x},\qquad
\phi(0)=1.
\end{align}
The aim is to find the differential operator $D_x$ acting on
\eqref{feq} which is reproduced under the integral on $\phi(t)$ such
that all appearing inhomogeneous parts vanish, i.e.
\begin{align*}
 D_x\phi(x)=-\lambda\int_0^\infty dt\frac{D_t \phi(t)}{1+t+x}.
\end{align*}
We compute derivatives and integrate by parts, taking the boundary 
values at $0$ and $\infty$ into account:
\begin{align}
\phi'(x)=-\frac{c_\lambda}{(1+x)^2}
+\lambda\int_0^\infty \!\! dt\;\frac{\phi'(t)}{1+t+x}+\frac{\lambda}{1+x}.
\label{dphi}
\end{align}
Also the product with $1+x$ simplifies by integration by parts:
\begin{align}
(1+x)\phi'(x)&=-\frac{c_\lambda}{(1+x)}
-\lambda\int_0^\infty \!\! dt\;\frac{t \phi'(t)}{1+t+x}
\;.
\label{xdphi}
\end{align}
We differentiate once more:
\begin{align*}
(1+x)\phi''(x)
+ \phi'(x)
&=\frac{c_\lambda}{(1+x)^2}
+\lambda\int_0^\infty \frac{dt}{(1+t+x)}\,
\frac{d}{dt}( t\phi'(t))\;,
\nonumber
\\
(1+x)\phi''(x)
&=\frac{2c_\lambda}{(1+x)^2}
+\lambda\int_0^\infty \!\! dt\;\frac{t \phi''(t)}{1+t+x}
-\frac{\lambda}{1+x}\;.
\end{align*}
We multiply by $x$ and integrate by parts:
\begin{align*}
x(1+x)\phi''(x)
&=\frac{2c_\lambda}{(1+x)}-\frac{2c_\lambda}{(1+x)^2}
-\lambda\int_0^\infty \frac{dt\;  t(1+t) \phi''(t)}{1+t+x}
+\frac{\lambda}{1+x}\;.
\end{align*}
We subtract twice \eqref{dphi} and add four times \eqref{xdphi}:
\begin{align*}
x(1+x)\phi''(x)+(2+4x) \phi'(x)
&=-\frac{2c_\lambda+\lambda}{(1+x)}
-\lambda\int_0^\infty \!\!\! dt\;\frac{t(1+t) \phi''(t)+ (2+4t)\phi'(t)
}{1+t+x}.
\end{align*}
Finally, we add $\frac{2c_\lambda+\lambda}{c_\lambda}$ times
\eqref{feq} to get $D_x=x(1+x)\frac{d^2}{dx^2}+(2+4x)\frac{d}{dx}+\frac{2c_\lambda+\lambda}{c_\lambda}$, or equivalently
\begin{align}
0&=(\mathrm{id}+\lambda \hat{A}_1) g,\qquad \text{where}
\label{idA}
\\
g(x)&=x(1+x)\phi''(x)+(2+ 4x) \phi'(x)+ \frac{2c_\lambda+\lambda}{c_\lambda}
\phi(x),
\nonumber
\end{align}
and $\hat{A}_\mu$ is the integral operator with kernel
$\hat{A}_\mu(t,u)=\frac{1}{u+t+\mu^2}$. The arguments given in
\sref{App.}{app} show that $\hat{A}_\mu$ has
spectrum $[0,\pi]$ for any $\mu\geq 0$. Therefore, equation
\eqref{idA} has for
$\lambda>-\frac{1}{\pi}$ only the trivial solution $g(x)=0$, which is
a standard hypergeometric differential equation.
The normalisation $\phi(0)=1$ uniquely fixes the solution to
\begin{align}
\phi(x)&=
{}_2F_1\Big(\genfrac{}{}{0pt}{}{1{+}\alpha_\lambda,\;2{-}\alpha_\lambda}{2}
\Big| -x\Big)
\nonumber
\\
&=\frac{1}{1+x} {}_2F_1\Big(\genfrac{}{}{0pt}{}{\alpha_\lambda,\;1{-}\alpha_\lambda}{2}
\Big| -x\Big)
\;,\quad  c_\lambda=\frac{\lambda}{\alpha_\lambda(1{-}\alpha_\lambda)}.
\label{phi-sol}
\end{align}

It remains to satisfy the boundary condition
$c_\lambda=1+\lambda \int_0^\infty dt\;\frac{\phi(t)}{1+t}$
given after (\ref{intc}).
The integral can be evaluated via the Euler integral
 \cite[\S 9.111]{gradshteyn2007},
\begin{align*}
\int_0^\infty \!\!\! dt\;\frac{\phi(t)}{1+t}
&= \frac{\Gamma(2)}{\Gamma(1-\alpha_\lambda)\Gamma(1+\alpha_\lambda)}
\int_0^\infty dt
\int_0^1 du \;\frac{u^{-\alpha_\lambda}(1-u)^{\alpha_\lambda}}{(1+ut)^{\alpha_\lambda}(1+t)^2}
\nonumber
\\
&= \frac{1}{\Gamma(1-\alpha_\lambda)\Gamma(1+\alpha_\lambda)}
\int_0^1 ds
\int_0^1 du \;\frac{u^{-\alpha_\lambda}(1-u)^{\alpha_\lambda}
(1-s)^{\alpha_\lambda}}{(1-(1-u)s)^{\alpha_\lambda}}
\nonumber
\\
&= \frac{1}{\Gamma(1-\alpha_\lambda)\Gamma(2+\alpha_\lambda)}
\int_0^1 du \;u^{-\alpha_\lambda}(1-u)^{\alpha_\lambda}\;
{}_2F_1\Big(\genfrac{}{}{0pt}{}{\alpha_\lambda,\;1}{2+\alpha_\lambda}
\Big|1-u\Big)
\nonumber
\\
&= \frac{1}{\Gamma(1-\alpha_\lambda)\Gamma(2+\alpha_\lambda)}
\int_0^1 du \;u^{\alpha_\lambda}(1-u)^{-\alpha_\lambda} \;
\Big\{
\frac{(1+\alpha_\lambda)}{\alpha_\lambda}
{}_2F_1\Big(\genfrac{}{}{0pt}{}{\alpha_\lambda,\;1}{1+\alpha_\lambda}
\Big|u\Big)
\nonumber
\\
&\qquad\qquad\qquad -\frac{1}{\alpha_\lambda}
{}_2F_1\Big(\genfrac{}{}{0pt}{}{\alpha_\lambda,\;2}{2+\alpha_\lambda}
\Big|u\Big)
\Big\}
\nonumber
\\
&= \frac{1}{\Gamma(1-\alpha_\lambda)\Gamma(2+\alpha_\lambda)}
\Big\{
\frac{(1+\alpha_\lambda)}{\alpha_\lambda}
\frac{\Gamma(1+\alpha_\lambda)\Gamma(1-\alpha_\lambda) \Gamma(1-\alpha_\lambda)}{
\Gamma(2-\alpha_\lambda)\Gamma(1)}
\nonumber
\\
&\qquad\qquad\qquad -\frac{1}{\alpha_\lambda}
\frac{\Gamma(2+\alpha_\lambda)\Gamma(1+\alpha_\lambda) \Gamma(1-\alpha_\lambda)\Gamma(1-\alpha_\lambda)}{
\Gamma(2)\Gamma(2)\Gamma(1)}
\Big\}
\nonumber
\\
&= \frac{1}{\alpha_\lambda(1-\alpha_\lambda)} -\Gamma(\alpha_\lambda)\Gamma(1-\alpha_\lambda).
\end{align*}
Here we have transformed $t=\frac{s}{1-s}$, evaluated first
the $s$-integral \cite[\S 9.111]{gradshteyn2007}
to a hypergeometric function, used its contiguous relation
\cite[\S 9.137.17]{gradshteyn2007} so that the remaining integrals are
known from \cite[\S 7.512.4]{gradshteyn2007} and
\cite[\S 7.512.3]{gradshteyn2007}.
We thus conclude
\begin{align*}
c_\lambda= 1 +
\frac{\lambda}{\alpha_\lambda(1-\alpha_\lambda)}- \frac{\lambda \pi}{\sin (\alpha_\lambda \pi)}
\stackrel{!}{=}\frac{\lambda}{\alpha_\lambda(1-\alpha_\lambda)}
\end{align*}
with solution
\begin{align}
\sin(\alpha_\lambda\pi)=\lambda\pi\;,\qquad
\alpha_\lambda=\left\{
\begin{array}{cl}
\dfrac{\arcsin(\lambda \pi)}{\pi} & \text{for }
|\lambda| \leq \frac{1}{\pi} ,\\
\dfrac{1}{2}+\mathrm{i} \dfrac{\mathrm{arcosh}(\lambda \pi)}{\pi} & \text{for }
\lambda \geq \frac{1}{\pi}.
\end{array}\right.
\label{sol-alpha}
\end{align}
The branch is uniquely selected by the requirement
$\lim_{\lambda\to 0} c_\lambda =1$. For $\lambda<-\frac{1}{\pi}$ there is
no solution for which $c_\lambda$ and $\phi$ are real.
Transforming back to $\tilde{\rho}_\lambda$ and $R_4$ gives
the result announced in \sref{Proposition}{Prop:Jx-final}, which provides the
2-point function $G(x,y)$ via
\sref{Theorem}{prop:HT}.

The choice of $\mu^2=\frac{\alpha_\lambda(1-\alpha_\lambda)}{\lambda}$ which 
is a natural choice is discussed
in great details by perturbative 
analysis in \sref{App.}{App:Solv2} by comparison to the angle function.
\end{proof}
\end{prps}
\noindent
The proof presented here was not the way how we found the solution, however it is for a reader the most 
transparent one. The perturbative analysis 
outsourced to the \sref{App.}{App:Solv2} provided the first more natural idea. 
We compared the perturbative result of \eqref{cottauba}
with the perturbative result of \eqref{Fred}.
The details came by using 
the Maple package \textsc{HyperInt} \cite{Panzer:2014caa}, which computed iterated integral in a 
symbolical way. Up to the $10^\text{th}$ order in $\lambda$, two coupled differential equations 
were found which gave a conjectural solution. This conjectural solution was then proved with the 
Meijer G-function which can be found in \sref{App.}{App:Meijer}. The easy proof above was realised 
after these discoveries.

\subsubsection*{Effective Spectral Dimension}
Let $\varrho_0(x)dx$ be the spectral measure of the
operator $E$ in the initial action of the noncommutative $\Phi^4_4$ model. The main discovery
of \sref{Sec.}{sec:tau} was that the interaction
$\frac{\lambda}{4} \mathrm{Tr}(\Phi^4)$ effectively modifies the
spectral measure to $\varrho_\lambda(x)dx$. What before, when expressed
in terms of $\varrho_0(x)dx$, was intractable became suddenly exactly
solvable in terms of the deformation $\varrho_\lambda(x)dx$.
For 4-dimensional Moyal space, one has
$\varrho_0(x)=x$ and $\varrho_\lambda(x)=R_4(x)$. The explicit solution
\eqref{Jx-final} shows that the deformation is
drastic: it changes the spectral dimension $\D$
defined in \sref{Definition}{Def:Spec} to an effective spectral dimension
$\D_\lambda:=\inf\{p\;:~ \int_0^\infty dt \;\frac{\varrho_\lambda(t)}{(1+t)^{p/2}}
<\infty\}$.

\begin{lemma}
\label{lem:bound}
For any $|\alpha_\lambda|<\frac{1}{2}$ one has
\[
\frac{1}{(1+x)^{\alpha_\lambda}} \leq
{}_2F_1\Big(\genfrac{}{}{0pt}{}{\alpha_\lambda,\;1{-}\alpha_\lambda}{2}
\Big| -x\Big)
\leq \frac{\Gamma(1-2\alpha_\lambda)}{
\Gamma(2-\alpha_\lambda)\Gamma(1-\alpha_\lambda)}
\frac{1}{(1+x)^{\alpha_\lambda}}\;.
\]
\end{lemma}
\noindent
\emph{Proof.} We transform with \cite[\S 9.131.1]{gradshteyn2007} to
\[
{}_2F_1\Big(\genfrac{}{}{0pt}{}{\alpha_\lambda,\;1{-}\alpha_\lambda}{2}
\Big| -x\Big)
=\Big(\frac{1}{1+x}\Big)^{\alpha_\lambda}\;
\frac{\displaystyle
{}_2F_1\Big(\genfrac{}{}{0pt}{}{2-\alpha_\lambda,\;1{-}\alpha_\lambda}{2}
\Big| \frac{x}{1+x}\Big) }{\displaystyle
\Big(1-\frac{x}{1+x}\Big)^{2\alpha_\lambda-1}}.
\]
By \cite[Thm. 1.10]{ponnusamy_vuorinen_1997}, the fraction on the rhs
is strictly increasing from 1 at $x=0$ to its limit
$\frac{B(2,1-2\alpha_\lambda)}{B(2-\alpha_\lambda,1-\alpha_\lambda)}
=\frac{\Gamma(1-2\alpha_\lambda)}{
\Gamma(2-\alpha_\lambda)\Gamma(1-\alpha_\lambda)}$
for $x\to \infty$. \hspace*{\fill}$\square$%

\begin{cor}
For $|\lambda|<\frac{1}{\pi}$, the deformed measure
$\varrho_\lambda=R_4$ of 4-dimensional Moyal space has spectral dimension
$\D_\lambda=4-2\frac{\arcsin(\lambda\pi)}{\pi}$.
\end{cor}
\noindent
\emph{Proof.}
Lemma~\ref{lem:bound} together with $\varrho_\lambda(x)=R_4(x)$
and \eqref{Jx-final} gives the assertion.
\hspace*{\fill}$\square$%

\bigskip\noindent
The change of spectral dimension is important. If instead of \eqref{Fred}
the function $R_4$ was given by
$\tilde{R}(x)=x-\lambda x^2 \int_0^\infty dt\; \frac{\varrho_0(t)}{
(t+\mu^2)^2(t+\mu^2+x)}$, then for $\varrho_0(x)=x$
this function $\tilde{R}$ is bounded above. Hence, $\tilde{R}^{-1}$
needed in higher topological sectors could not exist globally on
$\mathbb{R}_+$, which would render the model inconsistent for any $\lambda>0$.
The dimension drop down to $\D_\lambda=4-2\frac{\arcsin(\lambda\pi)}{\pi}$
avoids this (triviality) problem.

\subsubsection*{Perturbative Expansion of the 2-Point Function}
For a perturbative expansion, the exact solution of the planar 2-point function via \sref{Theorem}{prop:HT} and 
\sref{Proposition}{Prop:Jx-final} is not practical to expand in small $\lambda$. The expression 
\eqref{Gab-ansatz} is much more convenient for a perturbative analysis, where the measure is $\varrho_0(a)=a$ and the 
field renormalisation constant $Z$ is chosen to satisfy $G(0,0)=1$.
The angle function $\tau_b$ is derived perturbatively from \eqref{cottauba}, where the mass renormalisation is taken by 
\begin{align*}
 \mu_{bare}^2=1-\lambda \Lambda^2-\frac{1}{\pi} \int_0^{\Lambda^2} dt\, \tau_0(t).
\end{align*}
The computations are first done for finite $\Lambda^2$ where the limit is obviously convergent. 
All details about the perturbative expansion to higher order is discussed very detailed in 
\sref{App.}{App:Solv2}. We present 
here only the first two orders which are 
\begin{align*}
 p\lambda\pi\cot(\tau_a(p))&=1+a+p+\lambda\left((1+p)\log(1+p)-p\log(p)\right)\\
 &+\lambda^2\left(-p\zeta_2 +(1+p)\log(1+p)^2+(1+2p)\mathrm{Li}_2(-p)\right)+\mathcal{O}(\lambda^3)
\end{align*}
and after inserting in \eqref{Gab-ansatz}
\begin{align}\nonumber
 G(a,b)=&\frac{1}{1+a+b}-\frac{\lambda}{(1+a+b)^2}\{(1+a)\log(1+a)+(1+b)\log(1+b)\}\\\nonumber
 &+\frac{\lambda^2}{(1+a+b)^3}\{\zeta_2 ab+(1+a)(1+b)\log(1+a)\log(1+b)\\\label{eq:quartG2}
&\qquad  -a(1+b)\log(1+b)^2-b(1+a)\log(1+a)^2\\\nonumber
 &\qquad -(1+b+2a+2ab+a^2)\mathrm{Li}_2(-a)-(1+a+2b+2ab+b^2)\mathrm{Li}_2(-b)\}\\
 &+\mathcal{O}(\lambda^3).\nonumber
\end{align}
This result coincides only with the perturbative expansion through Feynman 
graphs 
for \textit{hyperlogarithms with two letters} (see \sref{App.}{App:Pert} for the definition of hyperlogarithms) 
according to the computations in \sref{App.}{App:PertQuartic}. 
Since the boundary conditions for Zimmermann's forest formula are different
both expansion cannot coincide completely.
It is hard to adjust the forest formula to obey $\frac{\partial}{\partial a}^{(0)}G(a,0)\vert_{a=0}
=-1-\lambda+\lambda^2+\mathcal{O}(\lambda^3)$ which is the right boundary condition for \eqref{eq:quartG2}.

On the other hand, it is also possible to change $\mu^2\neq \frac{\alpha_\lambda (1-\alpha_\lambda)}{\lambda}$ in 
\sref{Theorem}{prop:HT} and 
\sref{Proposition}{Prop:Jx-final} such 
that the 
condition $\frac{\partial}{\partial a}^{(0)}G(a,0)\vert_{a=0}
=-1$ is obeyed, which includes recursively some work. One would determine for arbitrary $\mu^2(\lambda)$ 
the functions $R_4$, $R_4^{-1}$ and
$I(w)$ in an expansion in $\lambda$ to get the angle function $\tau_b(a)$ as an expansion (depending on $\mu^2$).
Inserting then 
the expanded angle function in \eqref{Gab-ansatz} such that $G(0,0)=1$ and  $\frac{\partial}{\partial a}^{(0)}G(a,0)\vert_{a=0}
=-1$ holds order by order,
fixes $\mu^2$ at the first orders in $\lambda$. This is done more explicitly in \sref{App.}{App:PertQuartic}, where 
it is shown that for the same boundary conditions both approaches coincide perfectly!

\begin{rmk}\label{rmk:renorm1}
The quartic model on the $D=4$ Moyal space admits the 
renormalon problem which generates \textit{no problem} 
for the exact formula.
Determining the amplitude of the Feynman graph below according to the Feynman rules together with Zimmermann's forest formula
gives
\begin{align*}
 \frac{(-\lambda)^{3+2n}}{(1+2a)^6}\int_0^\infty y\,dy\frac{\big((1+y)\log(1+y)-y\big)^n}{(1+a+y)^{3+n}}
 \sim \,  \frac{(-\lambda)^{3+2n}}{ (1+2a)^6}\underbrace{\int_R^\infty \frac{dy}{y^2}\log(y)^n}_{\sim  n!}.
\end{align*}
\hspace*{20ex}
\def\svgwidth{0.5\textwidth}
\begingroup%
  \makeatletter%
  \providecommand\color[2][]{%
    \errmessage{(Inkscape) Color is used for the text in Inkscape, but the package 'color.sty' is not loaded}%
    \renewcommand\color[2][]{}%
  }%
  \providecommand\transparent[1]{%
    \errmessage{(Inkscape) Transparency is used (non-zero) for the text in Inkscape, but the package 'transparent.sty' is not loaded}%
    \renewcommand\transparent[1]{}%
  }%
  \providecommand\rotatebox[2]{#2}%
  \ifx\svgwidth\undefined%
    \setlength{\unitlength}{183.27583008bp}%
    \ifx\svgscale\undefined%
      \relax%
    \else%
      \setlength{\unitlength}{\unitlength * \real{\svgscale}}%
    \fi%
  \else%
    \setlength{\unitlength}{\svgwidth}%
  \fi%
  \global\let\svgwidth\undefined%
  \global\let\svgscale\undefined%
  \makeatother%
  \begin{picture}(1,0.44109824)%
    \put(0,0){\includegraphics[width=\unitlength,page=1]{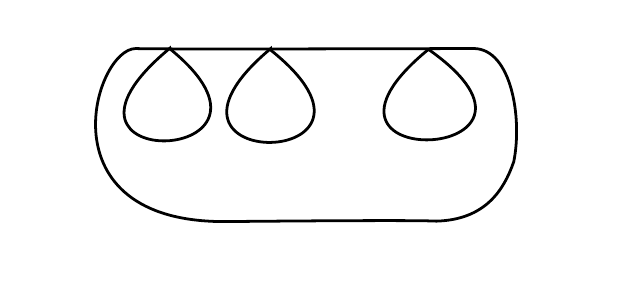}}%
    \put(0.5064805,0.29470261){\color[rgb]{0,0,0}\makebox(0,0)[lb]{\smash{...}}}%
    \put(0.2355266,0.27287758){\color[rgb]{0,0,0}\makebox(0,0)[lb]{\smash{$y_1$}}}%
    \put(0.39870971,0.26975602){\color[rgb]{0,0,0}\makebox(0,0)[lb]{\smash{$y_2$}}}%
    \put(0.64050068,0.2650791){\color[rgb]{0,0,0}\makebox(0,0)[lb]{\smash{$y_n$}}}%
    \put(0,0){\includegraphics[width=\unitlength,page=2]{Renormalon1.pdf}}%
    \put(0.41223765,0.39135252){\color[rgb]{0,0,0}\makebox(0,0)[lb]{\smash{$a$}}}%
    \put(0.43316292,0.14504152){\color[rgb]{0,0,0}\makebox(0,0)[lb]{\smash{$y$}}}%
    \put(0.18820989,0.05483218){\color[rgb]{0,0,0}\makebox(0,0)[lb]{\smash{$a$}}}%
    \put(0.41893156,0.00962305){\color[rgb]{0,0,0}\makebox(0,0)[lb]{\smash{$a$}}}%
    \put(0.88193385,0.20604836){\color[rgb]{0,0,0}\makebox(0,0)[lb]{\smash{$a$}}}%
    \put(-0.00543494,0.22467763){\color[rgb]{0,0,0}\makebox(0,0)[lb]{\smash{$a$}}}%
    \put(0,0){\includegraphics[width=\unitlength,page=3]{Renormalon1.pdf}}%
    \put(0.63060868,0.03760601){\color[rgb]{0,0,0}\makebox(0,0)[lb]{\smash{$a$}}}%
  \end{picture}%
\endgroup%
\\
\end{rmk}

\section{Higher Order Correlation Functions}\label{Sec.quartHO}
The 2-point function is the starting point in solving the entire hierarchy of all correlation functions
for the quartic model.
We have seen that the SDEs of \sref{Proposition}{Prop:quart11P} and 
\sref{Proposition}{Prop:quart21P} have when 
using $\mathcal{T}_q$ defined in \sref{Definition}{Def:Dp} a universal structure of the form
\begin{align}\label{K1}
 \hat{K}^1_p G^{(g)}_{|p|\mathcal{J}|}=&g^{g,\mathcal{J}}_{inh}\\\label{K2}
 \hat{K}^2_{q_1} G^{(g)}_{|q_1q_2|\mathcal{J}|}=&g^{g,\mathcal{J}}_{inh},
 \end{align}
 where the operators $\hat{K}^i$ are
 \begin{align*}
 \hat{K}^1_p f(p):=&f(p)\bigg\{H_{pp}+\frac{\lambda}{V} \sum_{n=0}^\mN\bigg(\frac{1}{E_n-E_p}+G^{(0)}_{|np|}\bigg)\bigg\}-
 \frac{\lambda}{V}\sum_{n=0}^\mN\frac{f(n)}{E_n-E_p}\\
 \hat{K}^2_{q_1} f(q_1,q_2):=&f(q_1,q_2)\bigg\{H_{q_1q_2}+\frac{\lambda}{V} \sum_{n=0}^\mN
 \bigg(\frac{1}{E_n-E_{q_1}}+G^{(0)}_{|nq_1|}\bigg)\bigg\}-
 \frac{\lambda}{V}\sum_{n=0}^\mN\frac{f(n,q_2)}{E_n-E_{q_1}}
\end{align*}
and $g^{g,\mathcal{J}}_{inh}$ is a inhomogeneity of less topology, i.e. a larger Euler characteristic than 
$G^{(g)}$. On the other hand, we have used the important identity \eqref{Jzz} for the proof of \sref{Theorem}{thm:main} 
mentioned in \sref{Remark}{rmk:2p}, which implies with the definition of $R(z)$ at $z=\varepsilon_p$
\begin{align*}
 E_p+\frac{\lambda}{V} \sum_{n=0}^\mN\bigg(\frac{1}{E_n-E_p}+G^{(0)}_{|np|}\bigg)=-R(-R^{-1}(E_p)),
\end{align*}
where $R^{-1}(z)$ is the principal branch with $R(\varepsilon_p)=E_p$ due to \sref{Lemma}{lem:U}. The 
expression makes sense for $n=p$ after inserting it into $\hat{K}^1$ and $\hat{K}^2$ and assuming a differentiable 
interpolation between the discrete point $E_p$. It is natural to pass from $G^{(0)}_{..}\to\G^{(0)}(..)$ by
\begin{align}
 G^{(g)}_{|p_1^1..p_{N_1}^1|..|p_1^b..p_{N_b}^b|}=:\G^{(g)}(\varepsilon_{p_1^1},
 \varepsilon_{p_2^1},..,\varepsilon_{p_{N_1}^1}|..| \varepsilon_{p_1^b},
 ..,\varepsilon_{p_{N_b}^b}),
\end{align}
where $\varepsilon_p=R^{-1}(E_p)$ due to \sref{Theorem}{thm:main}.

Assuming that $\G^{(g)}$ can be analytically continued except of for some particular points (poles)
yields for the SDEs \eqref{K1} and \eqref{K2}
\begin{align}\label{KK1}
 &(R(z)-R(-z))\G^{(g)}(z|\tilde{\mathcal{J}})-\frac{\lambda}{V}\sum_{n=0}^{\mN'}
 r_n\frac{\G^{(g)}(\varepsilon_n|\tilde{\mathcal{J}})}{R(\varepsilon_n)-R(z)}=g^{g,\tilde{\mathcal{J}}}_{inh}\\\label{KK2}
 &(R(z)-R(-w))\G^{(g)}(z,w|\tilde{\mathcal{J}})-\frac{\lambda}{V}\sum_{n=0}^{\mN'}
 r_n\frac{\G^{(g)}(\varepsilon_n,w|\tilde{\mathcal{J}})}{R(\varepsilon_n)-R(z)}=g^{g,\tilde{\mathcal{J}}}_{inh},
\end{align}
where $\tilde{\mathcal{J}}$ is a set of some complex numbers.

Comparing with \sref{Definition}{defint} and \eqref{eq:CubicSchwingerComplex} of the cubic model
implies exactly the same structure for the SDE \eqref{KK1}, 
where the base point is taken from the boundary of length one.
The SDE \eqref{KK2}, where the base point is taken from the boundary of length two, 
shows a new structure of more complexity. 

Formally, all correlation functions can be derived recursively by inverting the equations \eqref{KK1} and \eqref{KK2}. 
Take the planar $(1+1)$-point function as an example. It obeys with 
$g=0$ and $\mathcal{J}=\{w\}$ for  \eqref{KK1}
the equation 
\begin{align}
\big( R(z)-R(-z)\big) \mathcal{G}^{(0)}(z|w)
- \frac{\lambda}{V} \sum_{k=0}^{\mN'}
\frac{r_k\mathcal{G}^{(0)}(\varepsilon_k|w) }{
R(\varepsilon_k)-R(z)}
&=
\lambda \frac{\mathcal{G}^{(0)}(z,w)-\mathcal{G}^{(0)}(w,w)}{
R(z)-R(w)}\;.
\label{calG11:GW}
\end{align}
Let  $z\in \{0,\pm \alpha_0,\dots,\pm \alpha_{\mN'}\}$ 
be the solutions of $R(z)-R(-z)=0$, with all $\alpha_k>0$. When treating 
(\ref{calG11:GW}) as a Carleman-type singular integral equation as in 
 \cite{Grosse:2012uv}, it is clear that $\mathcal{G}^{(0)}(\alpha_k|w)$ 
is regular for all $w>0$. Therefore, setting $z\mapsto \alpha_k$ gives 
a system of $k$ affine equations 
\begin{align}
\frac{1}{V} \sum_{l=0}^{\mN'}
\frac{r_l\mathcal{G}^{(0)}(\varepsilon_l|w) }{
R(\alpha_k)-R(\varepsilon_l)}
&=
\frac{\mathcal{G}^{(0)}(\alpha_k,w)-\mathcal{G}^{(0)}(w,w)}{
R(\alpha_k)-R(w)}\;,\qquad k=0,\dots,{\mN'}\;.
\end{align}
They are easily solved by the inverse Cauchy matrix \cite{Schurmann:2019mzu} for 
$\frac{r_l}{V} 
\mathcal{G}^{(0)}(\varepsilon_l|w)$, $l=0,\dots,{\mN'}$, in terms of 
the planar 2-point function $\mathcal{G}^{(0)}(z,w)$, which are already known. 
Moreover, since $\mathcal{G}^{(0)}(z,w)$ depends on $z$ only via $R(\pm z)$, 
setting instead $z\mapsto -\alpha_k$ in (\ref{calG11:GW}) gives the same 
$\mathcal{G}^{(0)}(\varepsilon_l|w)$. With these 
$\mathcal{G}^{(0)}(\varepsilon_l|w)$ determined, (\ref{calG11:GW})
gives the explicit formula for $\mathcal{G}^{(0)}(z|w)$:
\begin{align}
\mathcal{G}^{(0)}(z|w)
&=\frac{\displaystyle
 \frac{\lambda}{V} \sum_{k=0}^{\mN'}
\frac{r_k\mathcal{G}^{(0)}(\varepsilon_k|w) }{
R(\varepsilon_k)-R(z)}
+
\lambda \frac{\mathcal{G}^{(0)}(z,w)-\mathcal{G}^{(0)}(w,w)}{
R(z)-R(w)}
}{R(z)-R(-z)} \;.
\label{G11-sol}
\end{align}
The global denominator 
$R(z)-R(-z)$ introduces a pole only at $z=0$, but not at $z=\pm \alpha_k$.

This procedure can be applied recursively for \eqref{KK1} and \eqref{KK2} (see \cite{Schurmann:2019mzu} for more 
information). However, the procedure does not capture the analytic structure of the correlation function. The poles 
which characterise the analytic structure are not revealed. It is known for instance that the $(1+1)$-point function is
symmetric $\mathcal{G}^{(0)}(z|w)=\mathcal{G}^{(0)}(w|z)$ which is hardly to check through the representation 
\eqref{G11-sol}. 

Comparing to the cubic model and its link to topological recursion (see \sref{Sec.}{sec.linkcub})
it is natural to conjecture also here a connection to topological recursion through the identity \eqref{Jzz} and the 
all SDEs. The spectral curve is conjecturally a rational plane algebraic curve given by
\begin{align}\label{spec:quart}
 \mathcal{E}(x(z),y(z))=0\qquad \text{with}\qquad x(z)=R(z)\,,\quad y(z)=-R(-z)
\end{align}
with the corresponding Riemann surface $\hat{\C}=\C\cup\{\infty\}$. The rational function $R(z)$ was defined in 
\sref{Theorem}{thm:main} implicitly by the eigenvalues $e_k$ of the external matrix $E$ together with
its corresponding multiplicities
$r_k$. However, it is not yet clear how the meromorphic $n$-forms $\omega_{g,n}$
coming from the spectral curve $\mathcal{E}(x(z),y(z))$, $z\in\hat{\C}$, the 1-form $\omega_{0,1}(z)=y(z)\,dx(z)$ 
and the 2-form $\omega_{0,2}(z_1,z_2)=\frac{dz_1\,dz_2}{(z_1-z_2)^2}$ through topological recursion \cite{Eynard:2007kz,
Eynard:2016yaa},
are linked to the correlation
functions of the quartic matrix field theory model. A detailed analysis of the pole structure of $\omega_{g,n}$ can bring more 
light into the dark,
but the state of not knowing the branch point $dx(z)=0$ explicitly makes the approach of topological recursion more
difficult. On the other hand, applying the procedure of inverse Cauchy matrices yields also inconvenient results.
Studying the connection between the quartic matrix field theory model and topological recursion is work in progress.

\section[Explicit Form of the Planar Recursive Equation]
{Explicit Form of the Planar Recursive Equation\footnote{This section is
taken from our paper \cite{DeJong}}}\label{Sec.quartRec}
This section studies the recursive equation of the planar $N$-point function with one boundary component
in more detail
\begin{equation}
G^{(0)}_{p_{0}\ldots p_{N-1}}=-\lambda\sum_{l=1}^{\frac{N-2}{2}}
\frac{G^{(0)}_{p_{0}\ldots p_{2l-1}}\cdot G^{(0)}_{p_{2l}\ldots p_{N-1}} 
- G^{(0)}_{p_{1}\ldots p_{2l}}\cdot G^{(0)}_{p_{0}p_{2l+1}\ldots p_{N-1}}}{
(E_{p_{0}}-E_{p_{2l}})(E_{p_{1}}-E_{p_{N-1}})}\;.\label{e:rr}
\end{equation}
which was given by
\sref{Example}{Exm:quartRec}.
We will omit the superscript of $G^{(0)}$, which indicates the genus, and let $\lambda=-1$, which is for the analysis 
of the recursion irrelevant, for the rest of the section.

Interestingly, the same relation \eqref{e:rr} appears in the planar
sector of the 2-matrix model for mixed correlation functions
\cite{Eynard:2005iq}.  The distinction between even $b_{2i}$ and odd
$b_{2i+1}$ matrix indices in \eqref{e:rr} corresponds to the different
matrices of the 2-matrix model.

The complete expression for the ($N=2k+2$)-point function 
$G_{p_{0}p_{1}\ldots p_{2k+1}}$ 
according to (\ref{e:rr}) yields $2^kc_k$ terms of the form
\begin{equation}
\frac{\pm G_{p_{i}p_{j}}\cdots G_{p_{l}p_{u}}}{
(E_{p_{n}}-E_{p_{m}})\cdots(E_{p_{v}}-E_{p_{w}})}
\label{expansion}
\end{equation}
with $i<j$, $l<u$, $n<m$ and $v<w$, where
$c_k=\frac{1}{k+1}\binom{2k}{k}$ is the $k$th Catalan number.
However, some of the terms cancel. 
We will answer the questions: \emph{Which terms 
survive the cancellations? Can they be explicitly characterised, without 
going into the recursion?} The answer will be encoded in \emph{Catalan
tables}.

First we discuss the symmetries of $G_{p_{0}\ldots p_{N-1}}$ induced by \eqref{e:rr} and not by its definition as an 
expectation value (which actually are the same). Then, we will introduce \textit{Catalan tuples} and \textit{Catalan tables}, 
certain trees and operations on them.
The Catalan numbers
$c_{k}=\frac{1}{k+1}\binom{2k}{k}$ will count various parts of our
results and will be related to the number
$d_{k}=\frac{1}{k+1}\binom{3k+1}{k}$ of 
Catalan tables of length $k+1$, see \sref{Proposition}{thrm:NrT}.

The main part is
\sref{Theorem}{thrm:RC} that Catalan tables precisely encode the 
surviving terms in the 
expansion of $G_{p_{0}\ldots p_{N-1}}$ with 
specified designated node. 

Both the Catalan tables and the $G_{p_{0}\ldots p_{N-1}}$ can be depicted conveniently 
as chord diagrams with threads, which is discussed in \sref{App.}{App:Chord}. 
Through these diagrams it will become clear that the recursion 
relation (\ref{e:rr}) is related to well-known combinatorial
problems~\cite{MR1909861,MR1603749}.

\subsection{Symmetries\label{sec:sym}}

\sref{Theorem}{prop:HT} proves the symmetry of the 2-point function, $G_{p_{i}p_{j}}=G_{p_{j}p_{i}}$. 
Because there is an even number of antisymmetric factors in 
the denominator of each term, it follows immediately that
\begin{equation}
G_{p_{0}p_{1}\ldots p_{N-1}}=G_{p_{N-1}\ldots p_{1}p_{0}}\;.\label{e:rr_s}
\end{equation}
Our aim is to prove cyclic invariance $G_{p_{0}p_{1}\ldots
  p_{N-1}}=G_{p_{1}\ldots p_{N-1}b_0}$.
We proceed by induction. Assuming that all $n$-point functions 
with $n\leq N-2$ are cyclically invariant, it is not difficult to check that
\begin{align}
&\hspace{-8mm}G_{p_{0}p_{1}\ldots p_{N-1}}
=\sum_{l=1}^{\frac{N-2}{2}}\frac{G_{p_{0}\ldots p_{2l-1}}\cdot 
G_{p_{2l}\ldots p_{N-1}} 
- G_{p_{1}\ldots p_{2l}}\cdot G_{p_{0}p_{2l+1}\ldots  p_{N-1}}}{
(E_{p_{0}}-E_{p_{2l}})(E_{p_{1}}-E_{p_{N-1}})}\nonumber
\\
&=-\sum_{l=1}^{\frac{N-2}{2}}\frac{G_{p_{0}p_{N-1}\ldots
    p_{2l+1}}\cdot 
G_{p_{2l}\ldots p_{1}}-G_{p_{N-1}\ldots p_{2l}}
\cdot G_{p_{2l-1}\ldots p_{1}p_{0}}}{
(E_{p_{0}}-E_{p_{2l}})(E_{p_{1}}-E_{p_{N-1}})}\nonumber
\\
&=\sum_{k=1}^{\frac{N-2}{2}}\frac{G_{p_{0}p_{N-1}\ldots
    p_{N-2k+1}}\cdot G_{p_{N-2k}\ldots p_{1}} 
- G_{p_{N-1}\ldots p_{N-2k}}\cdot 
G_{p_0 p_{N-2k-1}\ldots p_{1}}}{
(E_{p_{0}}-E_{p_{N-2k}})(E_{p_{N-1}}-E_{p_{1}})}\nonumber
\\
&=G_{p_{0}p_{N-1}\ldots p_{1}}
=G_{p_{1}\ldots p_{N-1}p_0}\;.\label{e:rr_rot}
\end{align}
The transformation $2l=N-2k$ and the symmetry (\ref{e:rr_s}) are 
applied here to rewrite the sum. This shows cyclic invariance.

Although the $N$-point functions are invariant under a cyclic
permutation of its indices, the preferred expansion into surving terms 
(\ref{expansion}) will depend on the choice of a designated node
$p_{0}$, the \textit{root}. Our preferred expansion will have a clear combinatorial
significance, but it cannot be unique because of 
\begin{align}
\hspace*{-1em}
\frac{1}{E_{p_i}{-}E_{p_j}}\cdot\frac{1}{E_{b_j}{-}E_{b_n}}
+\frac{1}{E_{b_n}{-}E_{b_i}}\cdot\frac{1}{E_{b_i}{-}E_{b_j}}
+\frac{1}{E_{b_j}{-}E_{b_n}}\cdot\frac{1}{E_{b_n}{-}E_{b_i}}=0\;.
\label{rfp}
\end{align}
These identities must be employed several times to establish cyclic 
invariance of our preferred expansion.

\subsection{Catalan Tuples\label{sec:CT}}

\begin{dfnt}[Catalan tuple]\label{dfnt:cattup}
  A Catalan tuple $\tilde{e}=(e_{0},\ldots,e_{k})$ of length
  $k\in\mathbb{N}_{0}$ is a tuple of integers $e_{j}\geq0$ for
  $j=0,\ldots,k$, such that
\begin{equation*}
\sum_{j=0}^{k}e_{j}=k\qquad \text{ and }\qquad 
\sum_{j=0}^{l}e_{j}>l\quad\text{for }l=0,\ldots,k-1\;.
\end{equation*}
The set of Catalan tuples of length $|\tilde{e}|:=k$ is denoted by
$\mathcal{C}_{k}$. 
\end{dfnt}
\noindent
For $\tilde{e}=(e_{0},\ldots,e_{k})$
it follows immediately that $e_{k}=0$ and $e_{0}>0$, if $k>0$. 

\begin{exm}
We have $\mathcal{C}_{0}=\{(0)\}$, $\mathcal{C}_{1}=\{(1,0)\}$ and 
$\mathcal{C}_{2}=\{(2,0,0),(1,1,0)\}$. 
  All Catalan tuples of length $3$ are given in the first column of
  \sref{Table}{f:CTk3}. 
\end{exm}

We now define two particular compositions of these objects. 
\sref{App.}{app:ex} provides a few examples.
%

\begin{dfnt}[$\circ$-composition]\label{dfnt:circ}
The composition $\circ:\mathcal{C}_{k}
\times\mathcal{C}_{l}\rightarrow \mathcal{C}_{k+l+1}$ is given by
\begin{align*}
(e_{0},\ldots,e_{k})\circ(f_{0},\ldots,f_{l})
&:=(e_{0}+1,e_{1},\ldots,e_{k-1},e_{k},f_{0},f_{1},\ldots,f_{l})\;.
\end{align*}
\end{dfnt}
\noindent
No information is lost in this composition, i.e.\ it is possible to 
uniquely retrieve both terms. In particular, $\circ$ cannot be 
associative or commutative. Consider for a Catalan tuple 
$\tilde{e}=(e_{0},\ldots,e_{k})$ partial sums 
$p_{l}:\mathcal{C}_{k}\rightarrow \{0,\ldots,k\}$ and maps 
$\sigma_{a}:\mathcal{C}_{k}\rightarrow \{0,\ldots,k\}$ defined by
\begin{align}
p_{l}(\tilde{e}) &:=-l+\sum_{j=0}^{l}e_{j}\;,\qquad\text{for } 
l=0,\ldots,k-1\;,
\\
\sigma_{a}(\tilde{e})&:=\min\{l\,|\,p_{l}(\tilde{e})=a\}\;.
\nonumber
\end{align}
Then 
\begin{equation}
\tilde{e}=(e_{0},\ldots,e_{k})
= (e_{0}-1,e_{1},\ldots,e_{\sigma_{1}(\tilde{e})})
\circ(e_{\sigma_{1}(\tilde{e})+1},\ldots,e_{k})\;.
\label{circ-factor}
\end{equation}
Because $\sigma_{1}(\tilde{e})$ exists for any
$\tilde{e}\in\mathcal{C}_{k}$ with $k\geq1$, every Catalan tuple has
unique $\circ$-factors. Only these two Catalan tuples, composed by
$\circ$, yield $(e_{0},\ldots,e_{k})$. This implies that the number
$c_k$ of Catalan tuples in $\mathcal{C}_{k}$ satisfies Segner's recurrence
relation
\begin{equation*}
c_{k}=\sum_{m=0}^{k-1}c_{m}c_{k-1-m}
\end{equation*}
together with $c_0=1$, which is solved by the Catalan numbers
$c_k=\frac{1}{k+1}\binom{2k}{k}$.

The other composition of Catalan tuples is a variant of the $\circ$-product. 
\begin{dfnt}[$\bullet$-composition]\label{dfnt:bullet}
The composition $\bullet:\mathcal{C}_{k}\times\mathcal{C}_{l}
\rightarrow \mathcal{C}_{k+l+1}$ is given by
\begin{equation*}
(e_{0},\ldots,e_{k})\bullet(f_{0},\ldots,f_{l})=(e_{0}+1,f_{0},\ldots,f_{l},e_{1},\ldots,e_{k})\quad.
\end{equation*}
\end{dfnt}
As in the case of the composition $\circ$, \sref{Definition}{dfnt:circ},
no information is lost in the product $\bullet$. It is reverted by
\begin{equation}
\tilde{e}=(e_0,\ldots,e_k)
=(e_0-1,e_{1+\sigma_{e_{0}-1}(\tilde{e})},\ldots,e_k) 
\bullet (e_1,\ldots ,e_{\sigma_{e_{0}-1}(\tilde{e})})\;.
\label{bullet-factor}
\end{equation}
Because $\sigma_{e_{0}-1}(\tilde{e})$ exists for any 
$\tilde{e}\in\mathcal{C}_{k}$ with $k\geq 1$
(also for $e_0=1$ where $\sigma_{e_0-1}(\tilde{e}) = k$),
every Catalan tuple has a unique pair of $\bullet$-factors.

Out of these Catalan tuples we will construct three sorts of trees: 
\emph{pocket tree}, \emph{direct tree}, \emph{opposite tree}. They are
all planted plane trees, which means they are embedded into the plane and 
planted into a monovalent 
phantom root which connects to a unique vertex that we consider as the 
(real) root. 
We adopt the convention that the phantom root 
is not shown; its implicit presence manifests in a different 
counting of the valencies the real root.
Pocket tree and direct tree are the same, but
their r\^ole will be different. 
Their drawing algorithms are given by the next definitions.
\begin{dfnt}[direct tree, pocket tree] \label{dfnt:rpt}
For a Catalan tuple $(e_{0},\ldots,e_{k})\in\mathcal{C}_{k}$, 
draw $k+1$ vertices on a line. Starting at the root $l=0$:
\begin{itemize}
\item unless $l=0$, connect this vertex to the last vertex ($m<l$) 
with an open half-edge;
\item if $e_{l}>0$: $e_{l}$ half-edges must be attached to vertex $l$;
\item move to the next vertex.
\end{itemize}
For direct trees, vertices will be called \emph{nodes} and edges will be 
called \emph{threads}; they are oriented from left to right.
For pocket trees, vertices are called \emph{pockets}.
\end{dfnt}

\begin{dfnt}[opposite tree]\label{dfnt:ot}
For a Catalan tuple $(e_{0},\ldots,e_{k})\in\mathcal{C}_{k}$, 
draw $k+1$ vertices on a line. Starting at the root $l=0$:
\begin{itemize}
\item if $e_{l}>0$: $e_{l}$ half-edges must be attached to vertex $l$;
\item if $e_{l}=0$:
\item[-] connect this vertex to the last vertex ($m<l$) with an open half-edge;
\item[-] if this vertex is now not connected to the 
last vertex ($n\leq m<l$) with an open half-edge, repeat this until it is;
\item move to the next vertex.
\end{itemize}
For opposite trees, vertices will be called \emph{nodes} and edges
will be called \emph{threads}; they are oriented from left to right.
\end{dfnt}

Examples of these trees can be seen in \sref{Figure}{f:PT} and
\sref{Table}{f:CTk3}. It will be explained in Sec.~\ref{sec:bijection} how these
trees relate to the recurrence relation (\ref{e:rr}) and how to 
label the nodes. 
The pocket trees will often be represented
with a top-down orientation, instead of a left-right one.

\begin{figure}[!ht]
\setlength{\unitlength}{0.75mm}
\hspace{3cm}\begin{picture}(100,42)
\put(-27,32){\mbox{DT:}}
\put(-27,5){\mbox{OT:}}
\multiput(-11,29)(10,0){15}{\textbullet}
\put(-5,30){\oval(10,5)[t]}
\put(0,30){\oval(20,10)[t]}
\put(5,30){\oval(30,15)[t]}
\put(25,30){\oval(10,5)[t]}
\put(35,30){\oval(10,5)[t]}
\put(40,30){\oval(20,10)[t]}
\put(45,30){\oval(30,15)[t]}
\put(30,30){\oval(80,20)[t]}
\put(75,30){\oval(10,5)[t]}
\put(85,30){\oval(10,5)[t]}
\put(90,30){\oval(20,10)[t]}
\put(90,30){\oval(40,15)[t]}
\put(55,30){\oval(130,25)[t]}
\put(60,30){\oval(140,30)[t]}

\multiput(-11,-1)(10,0){15}{\textbullet}
\put(-5,0){\oval(10,5)[t]}
\put(0,0){\oval(20,10)[t]}
\put(25,0){\oval(70,25)[t]}
\put(50,0){\oval(120,30)[t]}
\put(55,0){\oval(130,35)[t]}
\put(60,0){\oval(140,40)[t]}
\put(40,0){\oval(40,20)[t]}
\put(35,0){\oval(10,5)[t]}
\put(40,0){\oval(20,10)[t]}
\put(45,0){\oval(30,15)[t]}
\put(90,0){\oval(40,20)[t]}
\put(85,0){\oval(30,15)[t]}
\put(85,0){\oval(10,5)[t]}
\put(90,0){\oval(20,10)[t]}

\end{picture}
\caption{Direct tree (upper) and the opposite tree (lower) for the 
Catalan tuple $(6,0,0,1,3,0,0,0,2,2,0,0,0,0,0)
=(5,0,0,1,3,0,0,0,2,2,0,0,0,0)\circ(0)
=(5,0,1,3,0,0,0,2,2,0,0,0,0,0)\bullet (0)$.\label{f:PT}}
\end{figure}
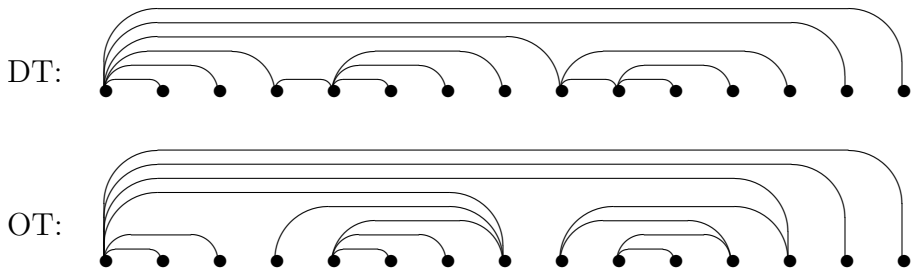

\begin{table}[!ht]
$\begin{array}{|c|c|c|c|}
\hline
{\text{Catalan tuple}} &
\text{pocket tree} &
\text{direct tree} &
\text{opposite tree} 
\\ \hline 
\raisebox{5mm}{(3,0,0,0)} 
&
\begin{picture}(20,13)
\put(10,7){\textbullet}
\put(5,2){\textbullet}
\put(10,2){\textbullet}
\put(15,2){\textbullet}
\put(11,8){\line(-1,-1){5}}
\put(11,8){\line(1,-1){5}}
\put(11,8){\line(0,-1){5}}
\end{picture}
&
\begin{picture}(27,14)
\put(5,2){\textbullet}
\put(10,2){\textbullet}
\put(15,2){\textbullet}
\put(20,2){\textbullet}
\put(8.5,3){\oval(5,5)[t]}
\put(11,3){\oval(10,10)[t]}
\put(13.5,3){\oval(15,15)[t]}
\end{picture}
&
\begin{picture}(27,14)
\put(5,2){\textbullet}
\put(10,2){\textbullet}
\put(15,2){\textbullet}
\put(20,2){\textbullet}
\put(8.5,3){\oval(5,5)[t]}
\put(11,3){\oval(10,10)[t]}
\put(13.5,3){\oval(15,15)[t]}
\end{picture}
\\ \hline 
\raisebox{6mm}{(2,1,0,0)} 
&
\begin{picture}(20,14)
\put(10,10){\textbullet}
\put(5,5){\textbullet}
\put(15,5){\textbullet}
\put(5,0){\textbullet}
\put(11,11){\line(-1,-1){5}}
\put(11,11){\line(1,-1){5}}
\put(6,6){\line(0,-1){5}}
\end{picture}
&
\begin{picture}(27,13)
\put(5,2){\textbullet}
\put(10,2){\textbullet}
\put(15,2){\textbullet}
\put(20,2){\textbullet}
\put(8.5,3){\oval(5,5)[t]}
\put(13.5,3){\oval(5,5)[t]}
\put(13.5,3){\oval(15,10)[t]}
\end{picture}
&
\begin{picture}(27,14)
\put(5,2){\textbullet}
\put(10,2){\textbullet}
\put(15,2){\textbullet}
\put(20,2){\textbullet}
\put(11,3){\oval(10,10)[t]}
\put(13.5,3){\oval(5,5)[t]}
\put(13.5,3){\oval(15,15)[t]}
\end{picture}
\\ \hline 
\raisebox{6mm}{(2,0,1,0)} 
&
\begin{picture}(20,14)
\put(10,10){\textbullet}
\put(5,5){\textbullet}
\put(15,5){\textbullet}
\put(15,0){\textbullet}
\put(11,11){\line(-1,-1){5}}
\put(11,11){\line(1,-1){5}}
\put(16,6){\line(0,-1){5}}
\end{picture}
&
\begin{picture}(27,13)
\put(5,2){\textbullet}
\put(10,2){\textbullet}
\put(15,2){\textbullet}
\put(20,2){\textbullet}
\put(8.5,3){\oval(5,5)[t]}
\put(18.5,3){\oval(5,5)[t]}
\put(11,3){\oval(10,10)[t]}
\end{picture}
&
\begin{picture}(27,13)
\put(5,2){\textbullet}
\put(10,2){\textbullet}
\put(15,2){\textbullet}
\put(20,2){\textbullet}
\put(8.5,3){\oval(5,5)[t]}
\put(18.5,3){\oval(5,5)[t]}
\put(13.5,3){\oval(15,10)[t]}
\end{picture}
\\
\hline
\raisebox{6mm}{(1,2,0,0)} 
&
\begin{picture}(20,13)
\put(10,6){\textbullet}
\put(5,1){\textbullet}
\put(10,11){\textbullet}
\put(15,1){\textbullet}
\put(11,7){\line(-1,-1){5}}
\put(10.8,7){\line(1,-1){5}}
\put(10.8,12){\line(0,-1){5}}
\end{picture}
&
\begin{picture}(27,15)
\put(5,2){\textbullet}
\put(10,2){\textbullet}
\put(15,2){\textbullet}
\put(20,2){\textbullet}
\put(8.5,3){\oval(5,5)[t]}
\put(13.5,3){\oval(5,5)[t]}
\put(16,3){\oval(10,10)[t]}
\end{picture}
&
\begin{picture}(27,13)
\put(5,2){\textbullet}
\put(10,2){\textbullet}
\put(15,2){\textbullet}
\put(20,2){\textbullet}
\put(16,3){\oval(10,10)[t]}
\put(13.5,3){\oval(5,5)[t]}
\put(13.5,3){\oval(15,15)[t]}
\end{picture}
\\ \hline 
\raisebox{8mm}{(1,1,1,0)} 
&
\begin{picture}(20,18)
\put(10,15){\textbullet}
\put(10,10){\textbullet}
\put(10,5){\textbullet}
\put(10,0){\textbullet}
\put(11,15.5){\line(0,-1){15}}
\end{picture}
&
\begin{picture}(27,18)
\put(5,5){\textbullet}
\put(10,5){\textbullet}
\put(15,5){\textbullet}
\put(20,5){\textbullet}
\put(8.5,6){\oval(5,5)[t]}
\put(13.5,6){\oval(5,5)[t]}
\put(18.5,6){\oval(5,5)[t]}
\end{picture}
&
\begin{picture}(27,18)
\put(5,5){\textbullet}
\put(10,5){\textbullet}
\put(15,5){\textbullet}
\put(20,5){\textbullet}
\put(13.5,6){\oval(15,15)[t]}
\put(16,6){\oval(10,10)[t]}
\put(18.5,6){\oval(5,5)[t]}
\end{picture}
\\ \hline 
\end{array}
$
\caption{The Catalan tuples and the corresponding planted plane trees 
for $k=3$. The phantom roots are not shown. The  real root is on top for the pocket tree and on the left for direct and opposite trees.
\label{f:CTk3}}
\end{table}

\subsection{Catalan Tables\label{sec:Ctab}}

A Catalan table is a `Catalan tuple of Catalan tuples':
\begin{dfnt}[Catalan table]\label{dfnt:cattab}
A \emph{Catalan table of length $k$} is a tuple 
$T_k=\langle\tilde{e}^{(0)},\tilde{e}^{(1)},\ldots,
\tilde{e}^{(k)}\rangle$ of Catalan tuples $\tilde{e}^{(j)}$, such 
that 
$(1+|\tilde{e}^{(0)}|,|\tilde{e}^{(1)}|,\ldots,|\tilde{e}^{(k)}|)$,
the \emph{length} tuple of $T_k$, is 
itself a Catalan tuple of length $k$. We let $\mathcal{T}_{k}$
be the set of all Catalan tables of length $k$. 
The constituent $\tilde{e}^{(j)}$ in a Catalan table is called the 
$j$-th pocket.
\end{dfnt}
\noindent 
We will show in \sref{Sec.}{sec:bijection} that a
Catalan table contains all information about individual terms in the
expansion (\ref{expansion}) of the $N$-point function $G_{p_0\ldots p_{N-1}}$.

Recall the composition $\circ$ from \sref{Definition}{dfnt:circ} and the fact
that any Catalan tuple of length $\geq 1$ has a unique pair of 
$\circ$-factors. We extend $\circ$ as follows to Catalan tables:
\begin{dfnt}[$\smalllozenge$-operation]\label{dfnt:triangle}
The operation $\smalllozenge: \mathcal{T}_{k}\times
\mathcal{T}_{l}\rightarrow\mathcal{T}_{k+l}$ is given by
\begin{equation*}
\langle \tilde{e}^{(0)},\ldots,\tilde{e}^{(k)}\rangle 
\smalllozenge\langle \tilde{f}^{(0)},\ldots,\tilde{f}^{(l)}\rangle 
:= \langle \tilde{e}^{(0)} \circ \tilde{f}^{(0)},\tilde{e}^{(1)},
\ldots,\tilde{e}^{(k)}, \tilde{f}^{(1)},\ldots,\tilde{f}^{(l)}\rangle\;.
\end{equation*}
\end{dfnt}
\noindent
Now suppose the Catalan table on the right-hand side is given. If the
$0^{\mathrm{th}}$ pocket has length $\geq 1$, then it uniquely factors
into $\tilde{e}^{(0)} \circ \tilde{f}^{(0)}$. Consider 
\begin{equation}
\hat{k}= \sigma_{1+|\tilde{f}^{(0)}|}
\big((1+|\tilde{e}^{(0)}\circ
\tilde{f}^{(0)}|,|\tilde{e}^{(1)}|,\ldots,
|\tilde{e}^{(k)}|, |\tilde{f}^{(1)}|,\ldots,|\tilde{f}^{(l)}|)\big)\;.
\label{k-lozenge}
\end{equation}
By construction, $\tilde{k}=k$ so that $\smalllozenge$ can 
be uniquely reverted. Note also that Catalan tables 
$\langle (0),\tilde{e}_1,\dots,\tilde{e}_k\rangle$ do not have a 
$\smalllozenge$-decomposition.

The composition $\bullet$ of Catalan tuples is extended as follows to 
Catalan tables:
\begin{dfnt}[$\smallblacklozenge$-operation]\label{dfnt:box}
The operation $\smallblacklozenge: 
\mathcal{T}_{k}\times\mathcal{T}_{l}\rightarrow\mathcal{T}_{k+l}$ is given by
\begin{equation*}
\langle \tilde{e}^{(0)},\ldots,\tilde{e}^{(k)}\rangle 
\smallblacklozenge\langle
\tilde{f}^{(0)},\ldots,\tilde{f}^{(l)}\rangle 
:= \langle \tilde{e}^{(0)} ,\tilde{e}^{(1)}\bullet\tilde{f}^{(0)}, 
\tilde{f}^{(1)},\ldots,\tilde{f}^{(l)}, \tilde{e}^{(2)},\ldots,
\tilde{e}^{(k)}\rangle\;.
\end{equation*}
\end{dfnt}
\noindent
If the $1^{\text{st} }$ pocket has length $\geq 1$, it uniquely
factors as $\tilde{e}^{(1)}\bullet\tilde{f}^{(0)}$, and we extract%
\begin{equation}
\hat{l}:=\sigma_{|\tilde{e}^{(0)}|+|\tilde{e}^{(1)}|+1}
\big((1+|\tilde{e}^{(0)}| ,|\tilde{e}^{(1)}\bullet\tilde{f}^{(0)}|, 
|\tilde{f}^{(1)}|,\ldots,|\tilde{f}^{(l)}|, |\tilde{e}^{(2)}|,\ldots,
|\tilde{e}^{(k)}|)\big)\;.
\label{l-blacklozenge}
\end{equation}
By construction $\hat{l}=l$, and $\smallblacklozenge$ is uniquely
reverted. 

We let $\mathcal{S}_k=\{
\langle \tilde{e}_0 , (0),\tilde{e}_2,\dots,\tilde{e}_k\rangle
\in \mathcal{T}_k\}$ be the subset of length-$k$ Catalan tables having 
$(0)$ as their first pocket. The Catalan tables $S\in 
\mathcal{S}_k$ are precisely those which do not have a 
$\smallblacklozenge$-decomposition. The distinction 
between $\mathcal{S}_l$ and ints complement in $\mathcal{T}_l$ is 
the key to determine the number of Catalan tables:
\begin{prps}\label{thrm:NrT}
  The set $\mathcal{T}_{k+1}$ of Catalan tables and its subset
  $\mathcal{S}_{k+1}$ with first pocket $(0)$ have cardinalities
 \begin{equation}
 d_{k}:=|\mathcal{T}_{k+1}|=\frac{1}{k+1}\binom{3k+1}{k}
\quad\text{and}\quad 
 h_{k}:=|\mathcal{S}_{k+1}|=\frac{1}{2k+1}\binom{3k}{k}\;.\label{e:res}
\end{equation}
\end{prps}
\begin{proof}
Let 
\[
\mathcal{D}(x):=\sum_{k=1}^\infty x^k \sum_{T\in \mathcal{T}_k} T
\qquad\text{and}\qquad
\mathcal{H}(x):=\sum_{k=1}^\infty x^k \sum_{S\in \mathcal{S}_k} S
\]
be the generating function of the set of all Catalan tables and of 
those having $(0)$ as their first pocket, respectively. Then
\begin{align}
\mathcal{D}(x)=\mathcal{D}(x) \smallblacklozenge \mathcal{D}(x)
+\mathcal{H}(x)
\label{cG}
\end{align}
because precisely the complements 
$\mathcal{T}_k\setminus \mathcal{S}_k$ have 
a unique $\smallblacklozenge$-decomposition. With the exception of 
$\langle(0),(0)\rangle \in \mathcal{S}_1=\mathcal{T}_1$, all 
$S=\langle \tilde{e}^0,(0),
\tilde{e}^2,\dots,\tilde{e}^k\rangle
\in \mathcal{S}_k$ with $k\geq 2$ have $|\tilde{e}^0|\geq 1$. 
Therefore, they have a unique 
$\smalllozenge$-decomposition, where the left factor necessarily belongs to 
$\mathcal{S}_l$  for some $l$:
\begin{align}
\mathcal{H}(x)=\mathcal{H}(x) \smalllozenge \mathcal{D}(x)
+ x \langle(0),(0)\rangle\;.
\label{cH}
\end{align}
Introducing the generating functions 
$D(x)=\sum_{k=0}^\infty x^{k+1} d_k$ and 
$H(x)=\sum_{k=0}^\infty x^{k+1} h_k$ 
of the cardinalities
$d_{k}=|\mathcal{T}_{k+1}|$ and 
$h_{k}=|\mathcal{S}_{k+1}|$, eqs.\ (\ref{cG}) and (\ref{cH}) 
project to quadratic relations
\begin{equation}
D(x)=D(x)\cdot D(x)+H(x)
\qquad\text{and}\qquad 
H(x)=H(x)\cdot D(x)+x\;.
\label{e:quadraticequation}
\end{equation}
Multiplying the first equation by $H(x)$ and
the second one by $D(x)$ gives $x\cdot D(x)=H^{2}(x)$, which
separates (\ref{e:quadraticequation}) into cubic relations
\begin{equation}
D(x)(1-D(x))^2=x\qquad\text{and} \qquad
\frac{H(x)}{\sqrt{x}} \Big(1-\Big(\frac{H(x)}{\sqrt{x}}\Big)^2\Big)
=\sqrt{x}\;.
\label{e:cubicequation}
\end{equation}
The coefficients (\ref{e:res}) can now be obtained by
the Lagrange inversion formula.
The second equation of (\ref{e:cubicequation}) results by 
taking $f(w)=w(1-w)^2$ in \sref{Theorem}{thm:Lagrange-inversion}, i.e.\
$\phi(w)=\frac{1}{(1-w)^2}$ and 
$ G(x)=\sum_{n=1}^\infty \frac{x^n}{n!}\frac{d^{n-1}}{dw^{n-1}}
\big\vert_{w=0} \frac{1}{(1-w)^{2n}}$. 
For the first equation of (\ref{e:cubicequation}), set
$\frac{H(x)}{\sqrt{x}}=w$, $z=\sqrt{x}$ and
$\phi(w)=\frac{1}{1-w^2}$
in \sref{Theorem}{thm:Lagrange-inversion}. Then
$H(x)
= \sqrt{x} \sum_{n=1}^\infty \frac{\sqrt{x}^n}{n!} 
\frac{d^{n-1}}{dw^{n-1}}\big\vert_{w=0} \frac{1}{(1-w^2)^{n}}$.

\end{proof}
\begin{rmk}
Equations (\ref{e:cubicequation}) are 
higher-order variants of the equation $C(x)(1-C(x))=x$ for the 
generating function $C(x)=\sum_{n=0}^\infty c_nx^{n+1}$ of Catalan numbers.
\end{rmk}
\begin{cor}\label{cor:dk}
The number $d_k$ of Catalan tables satisfies
\begin{equation}
d_{k}=\sum_{(e_{0},\ldots,e_{k+1})\in\mathcal{C}_{k+1}} 
c_{e_{0}-1}c_{e_{1}}\cdots c_{e_{k}}c_{e_{k+1}}\;. \label{dsumc}
\end{equation}
\end{cor}
\begin{proof}
There are $c_{|\tilde{e}_{0}|}\cdots c_{|\tilde{e}_{k+1}|}$ 
Catalan tables $\langle \tilde{e}_0,\dots,\tilde{e}_{k+1}\rangle$ 
of the same length tuple 
$(|\tilde{e}_{0}|+1,|\tilde{e}_{1}|,
\ldots,|\tilde{e}_{k+1}|)\in\mathcal{C}_{k+1}$. Set 
$e_0=|\tilde{e}_{0}|+1$ and 
$e_j=|\tilde{e}_{j}|$ for $j=1,\dots,k+1$.
\end{proof}



\subsection{The Bijection between Catalan Tables and 
Contributions to $G_{p_0\ldots p_{N-1}}$}
\label{sec:bijection}

\begin{dfnt}\label{dfnt:CTG}
To a Catalan table $T_{k+1}=\langle\tilde{e}^{(0)},\tilde{e}^{(1)},\ldots, 
\tilde{e}^{(k+1)}\rangle\in\mathcal{T}_{k+1}$ with $N/2=k+1$ we associate a monomial 
$[T]_{p_{0},\ldots,p_{N-1}}$ in $G_{p_{l}p_{m}}$ and 
$\frac{1}{E_{p_{l'}}-E_{p_{m'}}}$ as follows:
\begin{enumerate}
\item Build the pocket tree for the length tuple $(
  1+|\tilde{e}^{(0)}|,|\tilde{e}^{(1)}|,\ldots,
  |\tilde{e}^{(k+1)}|) \in \mathcal{C}_{k+1}$. It has $k+1$ edges
  and every edge has two sides. Starting from the root and turning
  counterclockwise, label the edge sides in consecutive
  order\footnote{This is the same order as in
    \cite[Fig.~5.14]{MR1676282}.\label{fn-order}}  from
  $p_{0}$ to $p_{N-1}$. An edge labelled $p_{l}p_{m}$ encodes a factor
  $G_{p_{l}p_{m}}$ in $G_{p_{0}\ldots p_{N-1}}$.

\item Label the $k+2$ vertices of the pocket tree by
  $P_{0},\ldots,P_{k+1}$ in consecutive order$^{\ref{fn-order}}$ when turning
  counterclockwise around the tree. Let $v(P_{m})$ be the valency of
  vertex $P_{m}$ (number of edges attached to $P_{m}$) and $L_{m}$ the
  distance between $P_{m}$ and the root $P_{0}$ (number of edges in
  shortest path between $P_{m}$ and $P_{0}$).

\item For every vertex $P_{m}$ that is not a leaf, read off the
  $2v(P_{m})$ side labels of edges connected to $P_{m}$. Draw two
  rows of $v(P_{m})$ nodes each. Label the nodes of the
  first row by the even edge side labels in natural order,
  i.e.\ starting at the edge closest to the root and proceed in the
  counterclockwise direction. Label the nodes of the other row by
  the odd edge side labels using the same edge order. Take the $m$-th
  Catalan tuple $\tilde{e}^{(m)}$ of the Catalan table. 
If $L_m$ is even, draw the
$\genfrac{\{}{\}}{0pt}{}{\text{direct}}{\text{opposite}}$ 
tree encoded by $\tilde{e}^{(m)}$ between the row of 
$\genfrac{\{}{\}}{0pt}{}{\text{even}}{\text{odd}}$ 
nodes.
If $L_m$ is odd, draw the
$\genfrac{\{}{\}}{0pt}{}{\text{opposite}}{\text{direct}}$ 
tree encoded by $\tilde{e}^{(m)}$ between the row of 
$\genfrac{\{}{\}}{0pt}{}{\text{even}}{\text{odd}}$ 
nodes. Encode a thread from $b_{l}$ to $b_{m}$ in the
 direct or opposite tree by a factor $\frac{1}{E_{p_{l}}-E_{p_{m}}}$.
\end{enumerate}
\end{dfnt}

\begin{rmk}\label{rmk:sidelabels}
In the proof below we sometimes have to insist that one side label of a
pocket edge is a particular $p_k$, whereas the label of the other side
does not matter. Is such a situation we will label the other side by 
$p_{\overline{k}}$. Note that if $p_k$ is 
$\genfrac{\{}{\}}{0pt}{}{\text{even}}{\text{odd}}$ then  $p_{\overline{k}}$ is 
$\genfrac{\{}{\}}{0pt}{}{\text{odd}}{\text{even}}$.
\end{rmk}

\begin{rmk}
  For the purpose of this article it is sufficient to mention that an
  explicit construction for the level function
  $L_{m}:\mathcal{C}_{k+1}\rightarrow\{0,\ldots, k\}$ exists.
\end{rmk}

\begin{exm}\label{ex:G12a}
Let 
$T=\langle(2,0,0),(1,1,0),(0),(0),(0),(1,0),(0)\rangle\in \mathcal{T}_6$.
Its length tuple is $(3,2,0,0,0,1,0)\in
\mathcal{C}_6$, which defines  the pocket tree: 
\[
\begin{picture}(75,33)
\put(25,27.5){\circle*{3}} 
\put(25,30){$P_{0}$}
\put(12.5,15){\circle*{3}}
\put(25,15){\circle*{3}}
\put(37.5,15){\circle*{3}}
\put(5,15){$P_{1}$}
\put(26,12){$P_{4}$}
\put(40,15){$P_{5}$}

\put(0,2.5){\circle*{3}}
\put(25,2.5){\circle*{3}}
\put(50,2.5){\circle*{3}}
\put(-7.5,2.5){$P_{2}$}
\put(27.5,2.5){$P_{3}$}
\put(52.5,2.5){$P_{6}$}

\put(25,27.5){\line(-1,-1){25}}
\put(14,21){\mbox{\small$p_{0}$}}
\put(2,9){\mbox{\small$p_{1}$}}
\put(6,5.5){\mbox{\small$p_{2}$}}
\put(17.5,17){\mbox{\small$p_{5}$}}
\put(25,27.5){\line(1,-1){25}}
\put(30.5,17){\mbox{\small$p_{8}$}}
\put(33,20.5){\mbox{\small$p_{11}$}}
\put(42.5,5){\mbox{\small$p_{9}$}}
\put(45,8){\mbox{\small$p_{10}$}}
\put(25,27.5){\line(0,-1){12.5}}
\put(21.5,20){\mbox{\small$p_{6}$}}
\put(25.5,20){\mbox{\small$p_{7}$}}
\put(12.5,15){\line(1,-1){12.5}}
\put(17,5.5){\mbox{\small$p_{3}$}}
\put(20,8.5){\mbox{\small$p_{4}$}}
\end{picture}
\]
The edge side labels encode 
\begin{equation*}
G_{p_{0}p_{5}}G_{p_{1}p_{2}}G_{p_{3}p_{4}}G_{p_{6}p_{7}} G_{p_{8}p_{11}}G_{p_{9}p_{10}}\;.
\end{equation*}
For vertex $P_{0}$, at even distance, we draw direct and opposite tree
encoded in $\tilde{e}^{(0)}=(2,0,0)$:
\begin{equation*}
\begin{picture}(80,8)
\put(0,2){\textbullet}
\put(10,2){\textbullet}
\put(20,2){\textbullet}
\put(-1,0){\mbox{\scriptsize$p_0$}}
\put(9,0){\mbox{\scriptsize$p_6$}}
\put(19,0){\mbox{\scriptsize$p_8$}}
\put(6,3){\oval(10,5)[t]}
\put(11,3){\oval(20,10)[t]}
\put(40,2){\textbullet}
\put(50,2){\textbullet}
\put(60,2){\textbullet}
\put(39,0){\mbox{\scriptsize$p_5$}}
\put(49,0){\mbox{\scriptsize$p_7$}}
\put(59,0){\mbox{\scriptsize$p_{11}$}}
\put(46,3){\oval(10,5)[t]}
\put(51,3){\oval(20,10)[t]}
\end{picture}
\end{equation*}
For vertex $P_{1}$, at odd distance, we draw opposite and direct tree
encoded in $\tilde{e}^{(1)}=(1,1,0)$:
\begin{equation*}
\begin{picture}(80,8)
\put(0,2){\textbullet}
\put(10,2){\textbullet}
\put(20,2){\textbullet}
\put(-1,0){\mbox{\scriptsize$p_{0}$}}
\put(9,0){\mbox{\scriptsize$p_{2}$}}
\put(19,0){\mbox{\scriptsize$p_{4}$}}
\put(16,3){\oval(10,5)[t]}
\put(11,3){\oval(20,10)[t]}
\put(40,2){\textbullet}
\put(50,2){\textbullet}
\put(60,2){\textbullet}
\put(39,0){\mbox{\scriptsize$p_{5}$}}
\put(49,0){\mbox{\scriptsize$p_{1}$}}
\put(59,0){\mbox{\scriptsize$p_{3}$}}
\put(56,3){\oval(10,5)[t]}
\put(46,3){\oval(10,5)[t]}
\end{picture}
\end{equation*}
For vertex $P_{5}$, at odd distance, we draw opposite and rooted tree
encoded in $\tilde{e}^{(5)}=(1,0)$:
\begin{equation*}
\begin{picture}(80,5)
\put(0,2){\textbullet}
\put(10,2){\textbullet}
\put(-1,0){\mbox{\scriptsize$p_{8}$}}
\put(9,0){\mbox{\scriptsize$p_{10}$}}
\put(6,3){\oval(10,5)[t]}
\put(40,2){\textbullet}
\put(50,2){\textbullet}
\put(39,0){\mbox{\scriptsize$p_{11}$}}
\put(49,0){\mbox{\scriptsize$p_{9}$}}
\put(46,3){\oval(10,5)[t]}
\end{picture}
\end{equation*}
They give rise to a factor 
\begin{align*}
&\hspace{-8mm}\frac{1}{(E_{p_{0}}-E_{p_{6}})(E_{p_{0}}-E_{p_{8}})(E_{p_{0}}-E_{p_{4}})(E_{p_{2}}-E_{p_{4}})
(E_{p_{8}}-E_{p_{10}})} \\ 
&\times\frac{1}{(E_{p_{5}}-E_{p_{7}})(E_{p_{5}}-E_{p_{11}})(E_{p_{5}}-E_{p_{1}})(E_{p_{1}}-E_{p_{3}})(E_{p_{11}}-E_{p_{9}})}\;.
\end{align*}
Later in \sref{Figure}{f:G12a} we give a diagrammatic 
representation of this Catalan table. 
\end{exm}

The following theorem shows that the Catalan tables correspond 
bijectively to the terms in the expansion of the recurrence relation 
(\ref{e:rr}).

\begin{thrm}\label{thrm:RC}
The recurrence (\ref{e:rr}) of $N$-point functions in the quartic 
matrix field theory model has the explicit solution 
\begin{equation*}
G_{p_0\ldots p_{N-1}}=\sum_{T\in \mathcal{T}_{k+1}} [T]_{p_0\ldots p_{N-1}}\;,
\end{equation*}
where the sum is over all Catalan tables of length $N/2=k+1$ 
and the monomials $[T]_{p_0\ldots p_{N-1}}$ are described in 
\sref{Definition}{dfnt:CTG}. 
\end{thrm}
\begin{proof}
  We proceed by induction in $N$. For $N=2$ the only term in the
  2-point function corresponds to the Catalan table $\langle
  (0),(0)\rangle\in\mathcal{T}_{1}$. Its associated length tuple $(1,0)$ 
encodes the pocket tree
\begin{equation*}
\begin{picture}(20,13)
\put(10,12){\circle*{2.5}}
\put(10,0){\circle*{2.5}}
\put(10,12){\line(0,-1){12}}
\put(6.5,6){\mbox{\scriptsize$p_0$}}
\put(11,6){\mbox{\scriptsize$p_1$}}
\end{picture}
\end{equation*}
whose single edge corresponds to a factor $G_{p_0p_1}$. 
The Catalan tuples of both pockets have length $0$, so that there is no 
denominator.

\smallskip

For any contribution to $G$ with $N\geq 4$, encoded by  
a length-$N/2$ Catalan table $T_{N/2}$, 
it must be shown that $T_{N/2}$ splits in one or two ways into
smaller Catalan tables whose corresponding monomials 
produce $T_{N/2}$ via  (\ref{e:rr}). There are three cases to
consider.

\bigskip

\noindent [\RNum{1}] Let $T_{k+1}= \langle
(0),\tilde{e}^{(1)},\ldots,\tilde{e}^{(k+1)}\rangle \in
\mathcal{T}_{k+1}$ with $N/2=k+1$. 
\\
It follows from \sref{Definition}{dfnt:box} that there are uniquely defined
Catalan tables $T_{l}= \langle
\tilde{f},\tilde{e}^{(2)},\ldots,\tilde{e}^{(l+1)}\rangle\in\mathcal{T}_l$ and $T_{k-l+1}=
\langle (0),\tilde{e},\tilde{e}^{(l+2)},\ldots,\tilde{e}^{(k+1)}\rangle\in\mathcal{T}_{k-l+1}$
with $\tilde{e}^{(1)}=\tilde{e}\bullet \tilde{f}$ and consequently
$T_{k-l+1}\smallblacklozenge T_{l}=T_{k+1}$. The length $l=\hat{l}$ is 
obtained via (\ref{l-blacklozenge}). Recall that $T_{k+1}$
cannot be obtained by the $\smalllozenge$-composition because 
the zeroth pocket has length $|(0)|=0$. By induction, $T_{l}$
encodes a unique contribution $[T_{l}]_{p_{1}\ldots p_{2l}}$ to
$G_{p_{1}\ldots p_{2l}}$, and $T_{k-l+1}$ encodes a unique
contribution $[T_{k-l+1}]_{p_{0}p_{2l+1}\ldots p_{N-1}}$ to
$G_{p_{0}p_{2l+1}\ldots p_{N-1}}$. We have to show that
\begin{equation*}
-\frac{[T_{l}]_{p_{1}\ldots p_{2l}}[T_{k-l+1}]_{p_{0}p_{2l+1}\ldots
    p_{N-1}}}{(E_{p_{0}}-E_{p_{2l}})(E_{p_{1}}-E_{p_{N-1}})}
\end{equation*}
agrees with $[T_{k+1}]_{p_{0}\ldots p_{N-1}}$ encoded by
$T_{k+1}$. A detail of the pocket tree of $T_{k+1}$ 
sketching $P_0,P_1$ and their attached edges is
\vspace{2ex}
\begin{align}
\parbox{80mm}{\setlength{\unitlength}{1.25mm}\begin{picture}(60,14)
\put(23,15){\circle*{2}}
\put(25,15){\mbox{\scriptsize$P_{0}$}}
\put(23,15){\line(-4,-1){16}}
\put(7,11){\circle*{2}}
\put(3,12){\mbox{\scriptsize$P_{1}$}}
\put(7,11){\line(-3,-2){10}}
\put(7,11){\line(-1,-3){3.5}}
\put(7,11){\line(1,-2){6.5}}
\put(7,11){\line(2,-1){12}}
\put(12,13.5){\mbox{\scriptsize$p_{0}$}}
\put(15,11.5){\mbox{\scriptsize$p_{N{-}1}$}}
\put(-2,7){\mbox{\scriptsize$p_{1}$}}
\put(-2,3){\mbox{\scriptsize$p_{\overline{1}}$}}
\put(12,8.5){\mbox{\scriptsize$p_{N{-}2}$}}
\put(12.8,4.3){\mbox{\scriptsize$p_{\overline{N{-}2}}$}}
\put(4,-1){\mbox{\scriptsize$p_{2l}$}}
\put(6,1.5){\mbox{\scriptsize$p_{2l{+}1}$}}
\put(13,-0.5){\mbox{\scriptsize$p_{\overline{2l{+}1}}$}}
\put(1,4){\mbox{\scriptsize$p_{\overline{2l}}$}}
\thicklines 
\qbezier[3](2,7)(2.8,6)(4.6,5.5)
\qbezier[3](10.5,5.2)(11.8,5.5)(12.8,7.3)
\end{picture}}
\end{align}
Only the gluing of the
  direct and opposite tree encoded by $\tilde{e}=(e_{0},\ldots,e_{p})$
  with the direct and opposite tree encoded by
  $\tilde{f}=(f_{0},\ldots,f_{q})$ via a thread from $p_{0}$ to
  $p_{2l}$ and a thread from $p_{N-1}$ to $p_{1}$ remains to be shown;
edge sides encoding a 2-point function and all other pockets 
are automatic. A symbolic notation is used now to
  sketch the trees. Horizontal dots are used to indicate a general
  direct tree and horizontal dots with vertical dots above them
  indicate an opposite tree. Unspecified threads are indicated by
  dotted half-edges. The four trees mentioned above are depicted as
\begin{align*}
\setlength{\unitlength}{1.25mm}
\parbox{60mm}{\begin{picture}(40,22)(-5,-1)
\put(-5,18){\mbox{$\text{OT}_{\tilde{e}}=$}}
\put(10,15){\textbullet}
\put(15,15){\textbullet}
\put(18.75,15.5){\ldots}
\put(20,16.5){\mbox{$\vdots$}}
\put(25,15){\textbullet}
\put(9,13){\mbox{\scriptsize$p_{0}$}}
\put(14,13){\mbox{\scriptsize$p_{\overline{2l+1}}$}}
\put(24,13){\mbox{\scriptsize$p_{N-2}$}}
\qbezier(11,16)(18.5,27)(26,16)
\put(-5,4){\mbox{$\text{DT}_{\tilde{e}}=$}}
\put(10,2){\textbullet}
\put(15,2){\textbullet}
\put(19.25,2.5){\ldots}
\put(25,2){\textbullet}
\put(5,0){\mbox{\scriptsize$p_{N-1}$}}
\put(14,0){\mbox{\scriptsize$p_{2l+1}$}}
\put(24,0){\mbox{\scriptsize$p_{\overline{N-2}}$}}
\qbezier(11,3)(13.5,7)(16,3)
\thicklines
\qbezier[4](11,3)(11,6)(14,7)
\qbezier[4](16,3)(16,6)(19,7)
\qbezier[4](23,7)(26,6)(26,3)
\end{picture}}
\parbox{60mm}{\setlength{\unitlength}{1.25mm}\begin{picture}(40,22)(-5,-1)
\put(-6,5){\mbox{$\text{DT}_{\tilde{f}}=$}}
\put(10,2){\textbullet}
\put(16.75,2.5){\ldots}
\put(25,2){\textbullet}
\put(9,0){\mbox{\scriptsize$p_{1}$}}
\put(24,0){\mbox{\scriptsize$p_{\overline{2l}}$}}
\thicklines
\qbezier[4](11,3)(11,6)(14,7)
\qbezier[4](23,7)(26,6)(26,3)
\thinlines 
\put(-6,17){\mbox{$\text{OT}_{\tilde{f}}=$}}
\put(10,15){\textbullet}
\put(16.75,15.5){\ldots}
\put(18.2,16.5){\mbox{$\vdots$}}
\put(25,15){\textbullet}
\put(9,13){\mbox{\scriptsize$p_{\overline{1}}$}}
\put(24,13){\mbox{\scriptsize$p_{2l}$}}
\qbezier(11,16)(18.5,27)(26,16)
\end{picture}}
\end{align*}
Here $\tilde{e}$ describes $P_{1}$, at odd distance, so that
even-labelled nodes are connected by the opposite tree. Every edge in
the pocket tree has two sides labelled $p_{r}$ and $p_{s}$, where the
convention of \sref{Remark}{rmk:sidelabels} is used when the other side
label does not matter.

The first edge in the pocket tree has side labels
$p_{0}p_{N-1}$ and descends from the root pocket. The following edge
is $p_{\overline{2l+1}}p_{2l+1}$ where $2l+2\leq \overline{2l+1} \leq
N-2$ is an even number. The
final edge is $p_{N-2}p_{\overline{N-2}}$ where $2l+1\leq
\overline{N-2} \leq N-3$ is an odd number.

Next, $\tilde{f}$ encodes $P_{0}$ in the pocket tree belonging to
$[T_{l}]_{p_{1}\ldots p_{2l}}$. It lies at even distance, but, because
the labels at $G_{p_{1}\ldots p_{2l}}$ start with an odd one,
the odd nodes of $\tilde{f}$ are connected by the direct tree and the
even nodes by the opposite tree.  Again, $2\leq \overline{1}\leq 2l $
denotes an even number and $1\leq \overline{2l}\leq 2l-1 $ an odd
number. When pasting $\tilde{f}$ into $\tilde{e}$, the first edge
remains $p_{0}p_{N-1}$, which descends from the root. Then all edges
from $\tilde{f}$ follow and, finally, the remaining edges of
$\tilde{e}$. Thus, before taking the denominators into account, the
four trees are arranged as:
\begin{align}
\setlength{\unitlength}{1.25mm}
\parbox{60mm}{\begin{picture}(40,28)(-5,-1)
\put(0,19){\mbox{$\text{OT}_{\tilde{e}} {\cup} 
\text{OT}_{\tilde{f}}$:}}
\put(20,17){\textbullet}
\put(19,15){\mbox{\scriptsize$p_{0}$}}
\put(25,17){\textbullet}
\put(31.75,17.5){\ldots}
\put(33,18.5){\mbox{$\vdots$}}
\put(40.25,17){\textbullet}
\put(45,17){\textbullet}
\put(24,15){\mbox{\scriptsize$p_{\overline{1}}$}}
\put(39,15){\mbox{\scriptsize$p_{2l}$}}
\put(45,17){\textbullet}
\put(49,17.5){\ldots}
\put(50.25,18.5){\mbox{$\vdots$}}
\put(45,15){\mbox{\scriptsize$p_{\overline{2l+1}}$}}
\put(55.25,17){\textbullet}
\put(54,15){\mbox{\scriptsize$p_{N-2}$}}
\qbezier(26,18)(33.5,27)(41,18)
\qbezier(21,18)(38.5,35)(56,18)
\put(0,5){\mbox{$\text{DT}_{\tilde{e}} {\cup} 
\text{DT}_{\tilde{f}}$:}}
\put(20,2){\textbullet}
\put(16,0){\mbox{\scriptsize$p_{N-1}$}}
\put(25,2){\textbullet}
\put(25,0){\mbox{\scriptsize$p_{1}$}}
\put(31.25,2.5){\ldots}
 \put(40,2){\textbullet}
 \put(45,2){\textbullet}
 \put(55,2){\textbullet}
 \put(44,0){\mbox{\scriptsize$p_{2l+1}$}}
 \put(39,0){\mbox{\scriptsize$p_{\overline{2l}}$}}
 \put(55,0){\mbox{\scriptsize$p_{\overline{N-2}}$}}
\qbezier(21,3)(33.5,18)(46,3)
 \thicklines
\qbezier[4](21,3)(21,6)(24,8)
\qbezier[4](26,3)(26,6)(29,7)
\qbezier[4](38,7)(41,6)(41,3)
\qbezier[4](53,7)(56,6)(56,3)
\thinlines 
\end{picture}}
\end{align}
The denominator of
$\frac{1}{(E_{p_{0}}-E_{p_{2l}})(E_{p_{N-1}}-E_{p_{1}})}$ (with
rearranged sign) corresponds to a thread between the nodes $p_{0}$ and
$p_{2l}$ and one between the nodes $p_{N-1}$ and $p_{1}$: 
\begin{align}
\setlength{\unitlength}{1.25mm}
\parbox{60mm}{\begin{picture}(40,28)(-5,-1)
\put(0,19){\mbox{$\text{OT}_{\tilde{e} \bullet \tilde{f}}$:}}
\put(20,17){\textbullet}
\put(19,15){\mbox{\scriptsize$p_{0}$}}
\put(25,17){\textbullet}
\put(31.75,17.5){\ldots}
\put(33,18.5){\mbox{$\vdots$}}
\put(40.25,17){\textbullet}
\put(45,17){\textbullet}
\put(24,15){\mbox{\scriptsize$p_{\overline{1}}$}}
\put(39,15){\mbox{\scriptsize$p_{2l}$}}
\put(45,17){\textbullet}
\put(49,17.5){\ldots}
\put(50.25,18.5){\mbox{$\vdots$}}
\put(45,15){\mbox{\scriptsize$p_{\overline{2l+1}}$}}
\put(55.25,17){\textbullet}
\put(54,15){\mbox{\scriptsize$p_{N-2}$}}
\qbezier(26,18)(33.5,27)(41,18)
\qbezier(21,18)(38.5,35)(56,18)
\qbezier(21,18)(38.5,30)(41,18)
\put(0,5){\mbox{$\text{DT}_{\tilde{e}\bullet \tilde{f}}$:}}
\put(20,2){\textbullet}
\put(16,0){\mbox{\scriptsize$p_{N-1}$}}
\put(25,2){\textbullet}
\put(25,0){\mbox{\scriptsize$p_{1}$}}
\put(31.25,2.5){\ldots}
 \put(40,2){\textbullet}
 \put(45,2){\textbullet}
 \put(55,2){\textbullet}
 \put(44,0){\mbox{\scriptsize$p_{2l+1}$}}
 \put(39,0){\mbox{\scriptsize$p_{\overline{2l}}$}}
 \put(55,0){\mbox{\scriptsize$p_{\overline{N-2}}$}}
\qbezier(21,3)(33.5,18)(46,3)
\qbezier(21,3)(23.5,5)(26,3)
 \thicklines
\qbezier[4](21,3)(21,6)(24,8)
\qbezier[4](26,3)(26,6)(29,7)
\qbezier[4](38,7)(41,6)(41,3)
\qbezier[4](53,7)(56,6)(56,3)
\thinlines 
\end{picture}}
\end{align}
The result
is precisely described by
$\tilde{e}\bullet\tilde{f}=(e_{0}+1,f_{0},\ldots,f_{q},e_{1},\ldots,e_{p})$
with \sref{Definitions}{dfnt:rpt} and \ref{dfnt:ot}. Indeed, the
increased zeroth entry corresponds to one additional half-thread
attached to the first node $p_{N-1}$ and one additional half-thread to
$p_{0}$. For the direct tree the rules imply that the next node,
$p_{1}$, is connected to $p_{N-1}$. This is the new thread from the
denominators. The next operations are done within $\tilde{f}$,
labelled $p_{1},\ldots,p_{\overline{2l}}$, without any
change. Arriving at its final node $p_{\overline{2l}}$ all
half-threads of $\tilde{f}$ are connected. The next node, labelled
$p_{2l+1}$, connects to the previous open half-thread, which is the
very first node $p_{N-1}$. These and all the following connections
arise within $\tilde{e}$ and remain unchanged. Similarly, in the
opposite tree, we first open $e_{0}+1$ half-threads at the zeroth node
$p_{0}$. Since $f_{0}>0$, we subsequently open $f_{0}$ half-threads at
the first node $p_{\overline{1}}$. The next operations remain
unchanged, until we arrive at the final node $p_{2l}$ of
$\tilde{f}$. It corresponds to $f_{q}=0$, so that we connect it to all
previous open half-threads, first within $\tilde{f}$. However, because
$e_{0}+1>0$, it is connected by an additional thread to $p_{0}$ and
encodes the denominator of $\frac{1}{E_{p_{0}}-E_{p_{2l}}}$. This
consumes the additional half-thread attached to $p_{0}$. All further
connections are the same as within $\tilde{e}$. In conclusion, we
obtain precisely the Catalan table
$T_{k+1}=\langle(0),\tilde{e}^{(1)}\ldots \tilde{e}^{(N/2)}\rangle$ we
started with.

\bigskip

\noindent [\RNum{2}] Let $T_{k+1}= \langle
\tilde{e}^{(0)},(0),\tilde{e}^{(2)},\ldots,\tilde{e}^{(k+1)}\rangle
\in \mathcal{T}_{k+1}$ and
$N/2=k+1$. \\
There are uniquely defined Catalan tables $T_{l}= \langle
\tilde{e},(0),\tilde{e}^{(2)},\ldots,\tilde{e}^{(l)}\rangle
\in \mathcal{T}_{l}$ and
$T_{k-l+1}= \langle
\tilde{f},\tilde{e}^{(l+1)},\ldots,\tilde{e}^{(k+1)}\rangle
\in \mathcal{T}_{k-l+1}$  with
$\tilde{e}^{(0)}=\tilde{e}\circ \tilde{f}$ and, consequently, $T_{l}
\smalllozenge T_{k-l+1}=T_{k+1}$. The length $l=\hat{k}$ is 
obtained via (\ref{k-lozenge}). Recall that $T_{k+1}$ cannot be
obtained by the $\smallblacklozenge$-composition, because the first
entry has length $|(0)|=0$. By the induction hypothesis, $T_{l}$
encodes a unique contribution $[T_{l}]_{p_{0}\ldots p_{2l-1}}$ to
$G_{p_{0}\ldots p_{2l-1}}$ and $T_{k-l+1}$ encodes a unique
contribution $[T_{k-l+1}]_{p_{2l}\ldots p_{N-1}}$ to
$G_{_{2l}\ldots p_{N-1}}$. It remains to be shown that
\begin{equation*}
\frac{[T_{l}]_{p_{0}\ldots p_{2l-1}}[T_{k-l+1}]_{p_{2l}\ldots
    p_{N-1}}}{(E_{p_{0}}-E_{p_{2l}})(E_{p_{1}}-E_{p_{N-1}})}
\end{equation*}
agrees with $[T_{k+1}]_{p_{0}\ldots p_{N-1}}$ encoded by $T_{k+1}$.
A detail of the pocket tree of $T_{k+1}$ 
sketching $P_0,P_1$ and their attached edges is
\begin{align}
\parbox{80mm}{\setlength{\unitlength}{1.25mm}\begin{picture}(60,12)
\put(31,12){\circle*{2}}
\put(27,13){\mbox{\scriptsize$P_{0}$}}
\put(31,12){\line(-4,-1){16}}
\put(31,12){\line(-1,-1){10}}
\put(31,12){\line(1,-6){2}}
\put(31,12){\line(4,-3){12}}
\put(31,12){\line(6,-1){20}}
\put(15,8){\circle*{2}}
\put(11,7){\mbox{\scriptsize$P_{1}$}}
\put(19,10){\mbox{\scriptsize$p_{0}$}}
\put(19,7){\mbox{\scriptsize$p_{1}$}}
\put(23,6.5){\mbox{\scriptsize$p_{2}$}}
\put(25,4){\mbox{\scriptsize$p_{\overline{2}}$}}
\put(27,0){\mbox{\scriptsize$p_{\overline{2l{-}1}}$}}
\put(33,2){\mbox{\scriptsize$p_{2l{-}1}$}}
\put(35,6){\mbox{\scriptsize$p_{2l}$}}
\put(40.5,5){\mbox{\scriptsize$p_{\overline{2l}}$}}
\put(42,8){\mbox{\scriptsize$p_{\overline{N{-}1}}$}}
\put(40,11){\mbox{\scriptsize$p_{N{-}1}$}}
\thicklines 
\qbezier[3](39,7.5)(40.2,8.2)(40.8,9.8)
\qbezier[3](28,8)(29,7)(31,7)
\end{picture}}
\end{align}
As in case [\RNum{1}] only the gluing of the direct and opposite tree
encoded by $\tilde{e}=(e_{0},\ldots,e_{p})$ with the direct and
opposite tree encoded by $\tilde{f}=(f_{0},\ldots,f_{q})$ via a thread
from $p_{0}$ to $p_{2l}$ and a thread from $p_{1}$ to $p_{N-1}$ must
be demonstrated. Everything else is automatic. These trees are
\begin{align}
\setlength{\unitlength}{1.25mm}
\parbox{60mm}{\begin{picture}(40,20)
\put(-5,5){\mbox{$\text{OT}_{\tilde{e}}=$}}
\put(10,2){\textbullet}
\put(15,2){\textbullet}
\put(19,2.5){\ldots}
\put(20.25,3.5){\mbox{$\vdots$}}
\put(25,2){\textbullet}
\put(9,0){\mbox{\scriptsize$p_{1}$}}
\put(14,0){\mbox{\scriptsize$p_{\overline{2}}$}}
\put(24,0){\mbox{\scriptsize$p_{2l-1}$}}
\qbezier(11,3)(18.5,14)(26,3)
\put(-5,17){\mbox{$\text{DT}_{\tilde{e}}=$}}
\put(10,15){\textbullet}
\put(15.25,15){\textbullet}
\put(19,15.5){\ldots}
\put(25,15){\textbullet}
\put(9,13){\mbox{\scriptsize$p_{0}$}}
\put(15,13){\mbox{\scriptsize$p_{2}$}}
\put(24,13){\mbox{\scriptsize$p_{\overline{2l-1}}$}}
\qbezier(11,16)(13.5,20)(16,16)
\thicklines
\qbezier[4](11,16)(11,19)(14,20)
\qbezier[4](16,16)(16,19)(19,20)
\qbezier[4](23,20)(26,19)(26,16)
\end{picture}}
\parbox{60mm}{\setlength{\unitlength}{1.25mm}\begin{picture}(40,25)(-5,-1)
\put(-6,17){\mbox{$\text{DT}_{\tilde{f}}=$}}
\put(10,15){\textbullet}
\put(16.75,15.5){\ldots}
\put(25,15){\textbullet}
\put(9,13){\mbox{\scriptsize$p_{2l}$}}
\put(24,13){\mbox{\scriptsize$p_{\overline{N-1}}$}}
\thicklines
\qbezier[4](11,16)(11,19)(14,20)
\qbezier[4](23,20)(26,19)(26,16)
\thinlines 
\put(-6,5){\mbox{$\text{OT}_{\tilde{f}}=$}}
\put(10,2){\textbullet}
\put(16.75,2.5){\ldots}
\put(18,3.5){\mbox{$\vdots$}}
\put(25,2){\textbullet}
\put(9,0){\mbox{\scriptsize$p_{\overline{2l}}$}}
\put(24,0){\mbox{\scriptsize$p_{N-1}$}}
\qbezier(11,3)(18.5,14)(26,3)
\end{picture}}
\label{e:case-2sep}
\end{align}
The notation is the same as in case [\RNum{1}]. The first pocket
$P_{1}$, described by the Catalan tuple $(0)$, is only $1$-valent so
that the first edge is labelled $p_{0}p_{1}$. The direct trees in
(\ref{e:case-2sep}) are put next to each other and a thread between
$p_{0}$ and $p_{2l}$ is drawn for the denominator of
$\frac{1}{E_{p_{0}}-E_{p_{2l}}}$. Similarly, the opposite trees in
(\ref{e:case-2sep}) are put next to each other and a thread between
$p_{1}$ and $p_{N-1}$ is drawn for the denominator of
$\frac{1}{E_{p_{1}}-E_{p_{N-1}}}$:
\begin{align*}
\parbox{60mm}{\setlength{\unitlength}{1.25mm}\begin{picture}(40,25)(-5,-1)
\put(-5,5){\mbox{$\text{OT}_{\tilde{e}\circ \tilde{f}}=$}}
\put(15.25,2){\textbullet}
\put(20,2){\textbullet}
\put(23.75,2.5){\ldots}
\put(25,3.5){\mbox{$\vdots$}}
\put(30,2){\textbullet}
\put(14,0){\mbox{\scriptsize$p_{1}$}}
\put(19,0){\mbox{\scriptsize$p_{\overline{2}}$}}
\put(29,0){\mbox{\scriptsize$p_{2l-1}$}}
\qbezier(16,3)(23.5,14)(31,3)
\put(40,2){\textbullet}
\put(46.75,2.5){\ldots}
\put(48,3.5){\mbox{$\vdots$}}
\put(55.5,2){\textbullet}
\put(39,0){\mbox{\scriptsize$p_{\overline{2l}}$}}
\put(54,0){\mbox{\scriptsize$p_{N-1}$}}
\qbezier(41,3)(48.5,14)(56,3)
\put(28.5,16){\oval(25,13)[t]}
\put(-5,17){\mbox{$\text{DT}_{\tilde{e}\circ\tilde{f}}=$}}
\put(15.25,15){\textbullet}
\put(20.25,15){\textbullet}
\put(23.75,15.5){\ldots}
\put(30,15){\textbullet}
\put(14,13){\mbox{\scriptsize$p_{0}$}}
\put(20,13){\mbox{\scriptsize$p_{2}$}}
\put(29,13){\mbox{\scriptsize$p_{\overline{2l-1}}$}}
\qbezier(16,16)(18.5,20)(21,16)
\thicklines
\qbezier[4](16,16)(16,19)(19,20)
\qbezier[4](21,16)(21,19)(24,20)
\qbezier[4](28,20)(31,19)(31,16)
\put(40.5,15){\textbullet}
\put(46.75,15.5){\ldots}
\put(55.5,15){\textbullet}
\put(39,13){\mbox{\scriptsize$p_{2l}$}}
\put(54,13){\mbox{\scriptsize$p_{\overline{N-1}}$}}
\thicklines
\qbezier[4](41,16)(41,19)(44,20)
\qbezier[4](53,20)(56,19)(56,16)
\thinlines 
\put(36,3){\oval(40,15)[t]}
\end{picture}}
\end{align*}
The result are precisely the direct and opposite trees of
the composition
$\tilde{e}\circ\tilde{f}=(e_{0}+1,e_{1},\ldots,e_{p},f_{0},\ldots,f_{q})$. The
increase $e_{0}\rightarrow e_{0}+1$ opens an additional half-thread at
$b_{0}$ and an additional half-thread at $p_{1}$. In the direct tree,
this new half-thread is not used by $e_{1},\ldots,e_{p}$. Only when we
are moving to $f_{0}$, labelled $p_{2l}$, we have to connect it with
the last open half-thread, i.e.\ with $p_{0}$. After that the remaining
operations are unchanged compared with $\tilde{f}$. In the opposite
tree, the additional half-thread at $p_{1}$ is not used in
$e_{1},\ldots,e_{p}$. Because $f_{0}$, labelled
$p_{\overline{2l}}$, opens enough half-threads, it is not consumed by
$f_{0},\ldots,f_{q-1}$ either. Then, the last node $f_{q}$, labelled
$p_{N-1}$, successively connects to all nodes with open half-threads,
including $p_{1}$. In conclusion, we obtain precisely the Catalan
table $T_{k+1}=\langle\tilde{e}^{(0)},(0),\tilde{e}^{(2)}\ldots
\tilde{e}^{(N/2)}\rangle$ we started with.

\bigskip

\noindent [\RNum{3}] Finally, we consider a general
$T_{k+1}=\langle\tilde{e}^{(0)},\tilde{e}^{(1)},\tilde{e}^{(2)},\ldots
\tilde{e}^{(k+1)}\rangle\in\mathcal{T}_{k+1}$ with $k+1=N/2$, $|\tilde{e}^{(0)}|\geq1 $
and $|\tilde{e}^{(1)}|\geq 1$.
There are uniquely defined Catalan tables $T_{l}= \langle
\tilde{e},\tilde{e}^{(1)},\tilde{e}^{(2)},\ldots,\tilde{e}^{(l)}\rangle\in\mathcal{T}_l$
and $T_{k-l+1}= \langle
\tilde{f},\tilde{e}^{(l+1)},\ldots,\tilde{e}^{(k+1)}\rangle\in\mathcal{T}_{k-l+1}$ with
$\tilde{e}^{(0)}=\tilde{e}\circ \tilde{f}$ and consequently $T_{l}
\smalllozenge T_{k-l+1}=T_{k+1}$. Moreover, uniquely defined Catalan
tables $T_{l'}= \langle
\tilde{f}',\tilde{e}^{(2)},\ldots,\tilde{e}^{(l'+1)}\rangle\in\mathcal{T}_{l'}$ and
$T_{k-l'+1}= \langle
\tilde{e}^{(0)},\tilde{e}',\tilde{e}^{(l'+2)},\ldots,\tilde{e}^{(k+1)}\rangle\in\mathcal{T}_{k-l'+1}$
exist, such that $\tilde{e}^{(1)}=\tilde{e}'\bullet \tilde{f}'$ and
consequently $T_{k-l'+1}\smallblacklozenge T_{l'}=T_{k+1}$. We
necessarily have $l'\leq k-1$ and $l\geq 2$, because $l'=k$
corresponds to case [\RNum{1}] and $l=1$ to case [\RNum{2}]. By the
induction hypothesis, these Catalan subtables encode unique
contributions $[T_{l}]_{p_{0}\ldots p_{2l-1}}$ to
$G_{p_{0}\ldots p_{2l-1}}$, $[T_{k-l+1}]_{p_{2l}\ldots
  p_{N-1}}$ to $G_{p_{2l}\ldots p_{N-1}}$,
$[T_{l'}]_{p_{1}\ldots p_{2l'}}$ to $G_{p_{1}\ldots p_{2l'}}$
and $[T_{k-l'+1}]_{p_{0}p_{2l'+1}\ldots p_{N-1}}$ to
$G_{p_{0}p_{2l'+1}\ldots p_{N-1}}$. We have to show that
\begin{align}
&\hspace{-8mm}\frac{[T_{l}]_{p_{0}\ldots p_{2l-1}}[T_{N/2-l}]_{p_{2l}\ldots p_{N-1}}}
{(E_{p_{0}}-E_{p_{2l}})(E_{p_{1}}-E_{p_{N-1}})}-\frac{[T_{l'}]_{p_{1}\ldots p_{2l'}}
[T_{N/2-l'}]_{p_{0}p_{2l'+1}\ldots p_{N-1}}}{(E_{p_{0}}-E_{p_{2l'}})(E_{p_{1}}-E_{p_{N-1}})}\label{e:case3-G}
\end{align}
agrees with $[T_{k+1}]_{b_{0}\ldots,b_{N-1}}$. 

\smallskip

In the pocket tree of $T_{k+1}$, there must be an edge with side labels 
$p_{0}p_{h}$, where $3\leq h\leq N-3$ and $h$ is odd. Here is a detail
of the pocket tree of $T_{k+1}$ showing $P_0,P_1$:
\vspace{2ex}
\begin{align}
\parbox{80mm}{\setlength{\unitlength}{1.25mm}\begin{picture}(60,17)
\put(31,17){\circle*{2}}
\put(27,18){\mbox{\scriptsize$P_{0}$}}
\put(31,17){\line(-4,-1){24}}
\put(31,17){\line(-1,-1){10}}
\put(31,17){\line(1,-6){2}}
\put(31,17){\line(4,-3){12}}
\put(31,17){\line(6,-1){20}}
\put(7,11){\circle*{2}}
\put(3,12){\mbox{\scriptsize$P_{1}$}}
\put(7,11){\line(-3,-2){10}}
\put(7,11){\line(-1,-3){3.5}}
\put(7,11){\line(1,-2){6.5}}
\put(7,11){\line(2,-1){12}}
\put(12,13.5){\mbox{\scriptsize$p_{0}$}}
\put(15,11.5){\mbox{\scriptsize$p_{h}$}}
\put(-2,7){\mbox{\scriptsize$p_{1}$}}
\put(-2,3){\mbox{\scriptsize$p_{\overline{1}}$}}
\put(12,8.5){\mbox{\scriptsize$p_{h{-}1}$}}
\put(13,4.5){\mbox{\scriptsize$p_{\overline{h{-}1}}$}}
\put(20,11.5){\mbox{\scriptsize$p_{h{+}1}$}}
\put(25,9){\mbox{\scriptsize$p_{\overline{h{+}1}}$}}
\put(27,5){\mbox{\scriptsize$p_{\overline{2l{-}1}}$}}
\put(33,7){\mbox{\scriptsize$p_{2l{-}1}$}}
\put(35,11){\mbox{\scriptsize$p_{2l}$}}
\put(40.5,10){\mbox{\scriptsize$p_{\overline{2l}}$}}
\put(42,13){\mbox{\scriptsize$p_{\overline{N{-}1}}$}}
\put(40,16){\mbox{\scriptsize$p_{N{-}1}$}}
\put(3.5,-1){\mbox{\scriptsize$p_{2l'}$}}
\put(5.8,1.5){\mbox{\scriptsize$p_{2l'\!{+}1}$}}
\put(13,-0.5){\mbox{\scriptsize$p_{\overline{2l'{+}1}}$}}
\put(1,4){\mbox{\scriptsize$p_{\overline{2l'}}$}}
\thicklines 
\qbezier[3](39,12.5)(40.2,13.2)(40.8,14.8)
\qbezier[3](28,13)(29,12)(31,12)
\qbezier[3](2,7)(2.8,6)(4.6,5.5)
\qbezier[3](10.5,5.2)(11.8,5.5)(12.8,7.3)
\end{picture}}
\end{align}
The direct and
opposite trees for $\tilde{e},\tilde{f}$ and $\tilde{e}^{(1)}$ can 
be sketched as
\begin{align}
\parbox{80mm}{\setlength{\unitlength}{1.25mm}\begin{picture}(60,25)(8,-1)
\put(0,5){\mbox{$\text{OT}_{\tilde{e}}{\cup}\text{OT}_{\tilde{f}}=$}}
\put(20,2){\textbullet}
\put(25,2){\textbullet}
\put(29,2.5){\ldots}
\put(30.25,3.5){\mbox{$\vdots$}}
\put(35,2){\textbullet}
\put(19,0){\mbox{\scriptsize$p_{h}$}}
\put(24,0){\mbox{\scriptsize$p_{\overline{h+1}}$}}
\put(34,0){\mbox{\scriptsize$p_{2l-1}$}}
\qbezier(21,3)(28.5,14)(36,3)
\put(45,2){\textbullet}
\put(51.75,2.5){\ldots}
\put(53.1,3.5){\mbox{$\vdots$}}
\put(60,2){\textbullet}
\put(44,0){\mbox{\scriptsize$p_{\overline{2l}}$}}
\put(59,0){\mbox{\scriptsize$p_{N-1}$}}
\qbezier(46,3)(53.5,14)(61,3)
\put(45,15){\textbullet}
\put(51.57,15.5){\ldots}
\put(60,15){\textbullet}
\put(44,13){\mbox{\scriptsize$p_{2l}$}}
\put(59,13){\mbox{\scriptsize$p_{\overline{N-1}}$}}
\thicklines
\qbezier[4](46,16)(46,19)(49,20)
\qbezier[4](58,20)(61,19)(61,16)
\thinlines 
\put(0,17){\mbox{$\text{DT}_{\tilde{e}}{\cup}\text{DT}_{\tilde{f}}=$}}
\put(20,15){\textbullet}
\put(25,15){\textbullet}
\put(29,15.5){\ldots}
\put(35,15){\textbullet}
\put(19,13){\mbox{\scriptsize$p_{0}$}}
\put(24,13){\mbox{\scriptsize$p_{h+1}$}}
\put(34,13){\mbox{\scriptsize$p_{\overline{2l-1}}$}}
\qbezier(21,16)(23.5,20)(26,16)
\thicklines
\qbezier[4](21,16)(21,19)(24,20)
\qbezier[4](26,16)(26,19)(29,20)
\qbezier[4](33,20)(36,19)(36,16)
\end{picture}}\quad
\parbox{36mm}{\setlength{\unitlength}{1.25mm}\begin{picture}(20,25)(8,-1)
\put(3,17){\mbox{$\text{OT}_{\tilde{e}^{(1)}}=$}}
\put(15,15){\textbullet}
\put(20,15){\textbullet}
\put(23.75,15.5){\ldots}
\put(25,16.5){\mbox{$\vdots$}}
\put(30,15){\textbullet}
\put(14,13){\mbox{\scriptsize$p_{0}$}}
\put(20,13){\mbox{\scriptsize$p_{\overline{1}}$}}
\put(29,13){\mbox{\scriptsize$p_{h-1}$}}
\qbezier(16,16)(23.5,27)(31,16)
\put(3,5){\mbox{$\text{DT}_{\tilde{e}^{(1)}}=$}}
\put(15,2){\textbullet}
\put(20,2){\textbullet}
\put(23.75,2.5){\ldots}
\put(30,2){\textbullet}
\put(14,0){\mbox{\scriptsize$p_{h}$}}
\put(20,0){\mbox{\scriptsize$p_{1}$}}
\put(29,0){\mbox{\scriptsize$p_{\overline{h-1}}$}}
\qbezier(16,3)(18.5,7)(21,3)
\thicklines
\qbezier[4](16,3)(16,6)(19,7)
\qbezier[4](21,3)(21,6)(24,7)
\qbezier[4](28,7)(31,6)(31,3)
\end{picture}}
\label{e:case-3sepa}
\end{align}
The denominators of
$\frac{1}{(E_{p_{0}}-E_{p_{2l}})(E_{p_{1}}-E_{p_{N-1}})}$ in
(\ref{e:case3-G}) add threads from $p_{0}$ to $p_{2l}$ and from
$p_{1}$ to $p_{N-1}$. The first one connects the direct trees for
$\tilde{e}\cup \tilde{f}$ to the direct tree encoded by
$\tilde{e}^{(0)}=\tilde{e}\circ \tilde{f}$. The second thread does \emph{not}
give a valid composition of the opposite trees for $\tilde{e}\cup
\tilde{f}$.

This is a problem. The solution is to split this
contribution. Half of the contribution is sacrificed to bring the
other half in the desired form. Afterwards, the same procedure is
repeated for the other term in (\ref{e:case3-G}) with a
minus-sign. The remainders are the same and cancel each other, whereas
the other halves add up to yield the sought for monomial.

Returning to trees, we note that in the direct tree for the pocket
$\tilde{e}^{(1)}$ there is always a thread from $p_{h}$ to $p_{1}$,
encoding a factor $\frac{1}{E_{p_{h}}-E_{p_{1}}}$. With the factor
$\frac{1}{E_{p_{1}}-E_{p_{N-1}}}$ it fulfils
\begin{equation}
\frac{1}{E_{p_{h}}{-}E_{p_{1}}}\cdot \frac{1}{E_{p_{1}}{-}E_{p_{N-1}}}
= \frac{1}{E_{p_{h}}{-}E_{p_{1}}}\cdot
\frac{1}{E_{p_{h}}{-}E_{p_{N-1}}}
+\frac{1}{E_{p_{h}}{-}E_{p_{N-1}}}\cdot
\frac{1}{E_{p_{1}}{-}E_{p_{N-1}}}\;.
\label{e:ident-E}
\end{equation}
The first term on the right-hand side of (\ref{e:ident-E}) leaves the
direct tree $\text{DT}_{\tilde{e}^{(1)}}$ as it is and connects the
parts of $\text{OT}_{\tilde{e}}{\cup}\text{OT}_{\tilde{f}}$ via the
thread from $p_{h}$ to $p_{N-1}$ to form
$\text{OT}_{\tilde{e}^{(0)}}$, where $\tilde{e}^{(0)}=\tilde{e}\circ\tilde{f}$.

\smallskip

\noindent
[*] The final term in (\ref{e:ident-E}) also unites
$\text{OT}_{\tilde{e}}{\cup}\text{OT}_{\tilde{f}}$ and forms
$\text{OT}_{\tilde{e}^{(0)}}$, but it removes in
$\text{DT}_{\tilde{e}^{(1)}}$ the thread between $p_{h}$ and
$p_{1}$. It follows from $\tilde{e}^{(1)}=\tilde{e}'\bullet
\tilde{f}'$ that this tree falls apart into the subtrees
$\text{DT}_{\tilde{e}'}$, containing $p_{h}$, and
$\text{DT}_{\tilde{f}'}$, which contains $p_{1}$. These are multiplied
by a factor $\frac{1}{E_{p_{1}}-E_{p_{N-1}}}$. The second term in
(\ref{e:case3-G}) will remove them.

Indeed, direct and opposite trees for $\tilde{e}^{(0)},\tilde{e}'$ 
and $\tilde{f}'$ can be sketched as 
\begin{align}
\parbox{50mm}{\setlength{\unitlength}{1.25mm}\begin{picture}(30,25)(7,-1)
\put(2,5){\mbox{$\text{OT}_{\tilde{e}^{(0)}}=$}}
\put(15,2){\textbullet}
\put(20,2){\textbullet}
\put(23.75,2.5){\ldots}
\put(25.2,3.5){\mbox{$\vdots$}}
\put(30,2){\textbullet}
\put(14,0){\mbox{\scriptsize$p_{h}$}}
\put(19,0){\mbox{\scriptsize$p_{\overline{h+1}}$}}
\put(29,0){\mbox{\scriptsize$p_{N-1}$}}
\qbezier(16,3)(23.5,14)(31,3)
\put(2,17){\mbox{$\text{DT}_{\tilde{e}^{(0)}}=$}}
\put(15,15){\textbullet}
\put(20,15){\textbullet}
\put(23.75,15.5){\ldots}
\put(30,15){\textbullet}
\put(15,13){\mbox{\scriptsize$p_{0}$}}
\put(19,13){\mbox{\scriptsize$p_{h+1}$}}
\put(29,13){\mbox{\scriptsize$p_{\overline{N-1}}$}}
\qbezier(16,16)(18.5,20)(21,16)
\thicklines
\qbezier[4](16,16)(16,19)(19,20)
\qbezier[4](21,16)(21,19)(24,20)
\qbezier[4](28,20)(31,19)(31,16)
\end{picture}}\quad
\parbox{60mm}{\setlength{\unitlength}{1.25mm}\begin{picture}(40,25)(14,-1)
\put(5,17){\mbox{$\text{OT}_{\tilde{e}'}
{\cup} \text{OT}_{\tilde{f}'}\,{=}$}}
\put(25,15){\textbullet}
\put(30,15){\textbullet}
\put(33.75,15.5){\ldots}
\put(35.2,16.3){\mbox{$\vdots$}}
\put(40,15){\textbullet}
\put(50,15){\textbullet}
\put(60,15){\textbullet}
\put(53.75,15.5){\ldots}
\put(55.3,16.5){\mbox{$\vdots$}}
\put(24,13){\mbox{\scriptsize$p_{0}$}}
\put(30,13){\mbox{\scriptsize$p_{\overline{1}}$}}
\put(39,13){\mbox{\scriptsize$p_{2l'}$}}
\put(47,13){\mbox{\scriptsize$p_{\overline{2l'+1}}$}}
\put(59,13){\mbox{\scriptsize$p_{h-1}$}}
\qbezier(31,16)(35,24)(41,16)
\put(43.5,16){\oval(35,12)[t]}
\put(5,5){\mbox{$\text{DT}_{\tilde{e}'}
{\cup} \text{DT}_{\tilde{f}'}\,{=}$}}
\put(25,2){\textbullet}
\put(30,2){\textbullet}
\put(33.75,2.5){\ldots}
\put(40,2){\textbullet}
\put(50,2){\textbullet}
\put(60,2){\textbullet}
\put(24,0){\mbox{\scriptsize$p_{h}$}}
\put(30,0){\mbox{\scriptsize$p_{1}$}}
\put(39,0){\mbox{\scriptsize$p_{\overline{2l}}$}}
\put(47,0){\mbox{\scriptsize$p_{2l'+1}$}}
\put(59,0){\mbox{\scriptsize$p_{\overline{h-1}}$}}
\put(38.5,3){\oval(25,12)[t]}
\thicklines
\qbezier[4](26,3)(26,6)(27,8)
\qbezier[4](31,3)(31,6)(34,7)
\qbezier[4](38,7)(41,6)(41,3)
\qbezier[4](51,3)(51,6)(54,7)
\qbezier[4](58,7)(61,6)(61,3)
\end{picture}}
\end{align}
The direct tree $\text{DT}_{\tilde{e}^{(0)}}$ remains intact and the
thread from $p_{0}$ to $p_{2l'}$ encoded in the factor
$\frac{1}{(E_{p_{0}}-E_{p_{2l'}})}$ in (\ref{e:case3-G}) connects the
opposite trees for $\tilde{e}'\cup \tilde{f}'$ to form the opposite
tree for $\tilde{e}^{(1)}=\tilde{e}'\bullet \tilde{f}'$. The direct
trees $\text{DT}_{\tilde{e}'}{\cup} \text{DT}_{\tilde{f}'}$ remain
disconnected and are multiplied by $\frac{1}{(E_{p_{1}}-E_{p_{N-1}})}$
from (\ref{e:case3-G}). With the minus-sign from (\ref{e:case3-G})
they cancel the terms described in [*]. The other trees combined yield
precisely the direct and opposite trees for both $\tilde{e}^{(0)}$ and
$\tilde{e}^{(1)}$, so that the single Catalan table we started with is
retrieved.

This completes the proof. Bijectivity between Catalan tables and
contributing terms to $(N'{<}N)$-point functions is essential: Assuming
the above construction [\RNum{1}]--[\RNum{3}] missed Catalan subtables
$T_{l}, T_{N/2-l}$, then their composition $T_{l} \smalllozenge
T_{N/2-l}$ would be a new Catalan table of length $N/2$. However, all
Catalan tables of length $N/2$ are considered. Similarly for $T_{l'}
\smallblacklozenge T_{N/2-l'}$.
\end{proof}

This theorem shows that there is a one-to-one correspondence between
Catalan tables and the diagrams/terms in
$G_{p_{0}\ldots,p_{N-1}}$ with designated node $p_{0}$. The
choice of designated node does not influence $G_{p_{0}\ldots,p_{N-1}}$, but it does
alter its expansion.

\section{Summary}
We derived all SDEs for the quartic matrix field theory model and showed that they have one of the 
following particular forms (\sref{Proposition}{Prop:quart11P} 
and \sref{Proposition}{Prop:quart21P}) 
\begin{align}\label{K1form}
 \hat{K}^1_p G^{(g)}_{|p|\mathcal{J}|}=&g^{g,\mathcal{J}}_{inh}\\\label{K2form}
 \hat{K}^2_{q_1} G^{(g)}_{|q_1q_2|\mathcal{J}|}=&g^{g,\mathcal{J}}_{inh},
 \end{align}
if the base point is taken from a boundary component of length 1 or 2.
The operators $\hat{K}^i$ are linear operators defined in \eqref{eq:quartHOOP1} and \eqref{eq:quartHOOP2}, and 
$g^{g,\mathcal{J}}_{inh}$ the inhomogeneous part of the equation depending 
on correlation functions of less topology than $G^{(g)}$.

Especially, the form \eqref{K2form} is achieved by the operator $\mathcal{T}_q$ 
defined in \sref{Definition}{Def:Dp} which rewrites $\frac{1}{V}\sum_n G^{(g)}_{|q_1n|\mathcal{J}|}$
as a derivative such that $\hat{K}^2$ gets really the form \eqref{eq:quartHOOP2}. 

If a correlation function has boundary components of lengths $\geq 3$, it can be recursively written 
by applying \sref{Proposition}{Prop:quartRec} through correlation functions of 
boundary components of length one and/or two. 
The nonlinear rhs of \sref{Proposition}{Prop:quartRec} is strictly 
different to the linear rhs of the recursive equation
\sref{Proposition}{Prop:CubisRecur}
of the cubic model.

The main achievement for the quartic model was to find the exact solution of the planar 2-point function ($g=0$) for any 
spectral dimension $\D<6$ with an appropriate renormalisation. For $4\leq\D<6$, the quartic 
matrix field theory model is just-renormalisable. 
It was known that the angle function on the $D=2$ and $D=4$ Moyal space plays an 
important r\^ole, which gave us the right ansatz for the general case with the measure $\varrho_0(x)$, and the 
functions $R_D(z)$ and $I(w)$. The angle function is then computed by \sref{Theorem}{thm:tau}, where the 
important observation was to deform the measure to $\varrho_\lambda(x)=\varrho_0(R_D(x))$. The solution
of the planar 2-point function is summarised in \sref{Theorem}{prop:HT}, which is more or less 
straightforward to compute from the solution of the angle function. The \sref{Example}{Ex:D2} on the $D=2$ Moyal space 
recovers the result of \cite{Panzer:2018tvy}.

Two more examples are of particular interest, the solution for finite matrices ($\D=0$) discussed in \sref{Sec.}{sec.fm} with 
\sref{Theorem}{thm:main} and 
the solution on the $D=4$ Moyal space discussed in \sref{Sec.}{Sec.4dSol}.

For finite matrices, the planar 2-point function is given as a rational function. Furthermore, the rational function
$R(z)$ (defined in \sref{Theorem}{thm:main}) induces a preferred variable transformation $x\mapsto z=R^{-1}(x)$ which 
brings the SDEs \eqref{K1form} and \eqref{K2form} after analytic continuation to the form
\begin{align*}
 &(R(z)-R(-z))\G^{(g)}(z|\tilde{\mathcal{J}})-\frac{\lambda}{V}\sum_{n=0}^{\mN'}
 r_n\frac{\G^{(g)}(\varepsilon_n|\tilde{\mathcal{J}})}{R(\varepsilon_n)-R(z)}=g^{g,\tilde{\mathcal{J}}}_{inh}\\
 &(R(z)-R(-w))\G^{(g)}(z,w|\tilde{\mathcal{J}})-\frac{\lambda}{V}\sum_{n=0}^{\mN'}
 r_n\frac{\G^{(g)}(\varepsilon_n,w|\tilde{\mathcal{J}})}{R(\varepsilon_n)-R(z)}=g^{g,\tilde{\mathcal{J}}}_{inh},\\
&\text{where}\qquad  \G^{(g)}(\varepsilon_{p_1^1},
 \varepsilon_{p_2^1},..,\varepsilon_{p_{N_1}^1}|..| \varepsilon_{p_1^b},..,\varepsilon_{p_{N_b}^b})=
 G^{(g)}_{|p_1^1..p_{N_1}^1|..|p_1^b..p_{N_b}^b|},
\end{align*}
with $\varepsilon_p=R^{-1}(E_p)$. The new SDEs (after variable transformation) indicate a possible relation to topological recursion
with $x(z)=R(z)$ and $y(z)=-R(-z)$
since the form of the upper
equation is identical to the cubic model with the only difference of taking $x(z)=z^2-c$.

The second example with the solution on the $D=4$ Moyal space provides even more fascinating facts.
The deformed measure $\varrho_\lambda$ coincides with the function $R_4$ via $\varrho_\lambda(x)=\varrho_0(R_4(x))=R_4(x)$
which leads to Fredholm-type integral equation \eqref{Fred} solved in \sref{Proposition}{Prop:Jx-final}. The resulting 
hypergeometric function for the deformed 
measure changes the spectral dimension $\D\to\D_\lambda= 4-2\frac{\arcsin(\lambda\pi)}
{\pi}$ which avoids the triviality problem for this particular model, since $R_4$ has 
now a global inversion on $\R_+$ for $\lambda>0$
which is necessary for higher topological sectors. The explicit result of $R_4=\varrho_\lambda$ has 
a natural choice for $\mu^2$ in \sref{Proposition}{Prop:Jx-final}, which is discussed in very detail in 
\sref{App.}{App:Solv2}. This natural choice forces $\frac{\partial}{\partial a}G(a,0)\vert_{a=0}=c_\lambda\neq-1$.

Looking at the perturbative expansion, it becomes quite fast clear that at each order in $\lambda$ hyperlogarithms survive which 
induces a much higher complexity in comparison to the cubic model (no hyperlogarithms survive). Fixing the 
boundary condition by adjusting $\mu^2$ in \sref{Proposition}{Prop:Jx-final} to obey 
$\frac{\partial}{\partial a}G(a,0)\vert_{a=0}=-1$, the perturbative expansion of the exact solution 
and the perturbative expansion through Zimmermann's forest formula coincide perfectly, which is 
shown in \sref{App.}{App:PertQuartic}. The surviving hyperlogarithms of the perturbative expansion are due to the 
integral representation of the exact solution given in \sref{Theorem}{prop:HT} after inserting $R_4$. Notice also 
that the Feynman diagrammatic expansion admits on the $D=4$ Moyal space the renormalon problem (see. \sref{Remark}{rmk:renorm1})
and the number of planar graphs with one boundary grows for $n$ vertices with $\mathcal{O}(n!)$. However,
we provide here a second example, where the expansion with Zimmermann's forest formula 
for a just-renormalisable model is resummable, and has additionally
a potential which is bounded from below.

Any planar $N$-point function is given recursively through the 2-point function via \eqref{e:rr}. Writing out the 
recursion, cancellations appear which are related to non-crossing chord diagrams 
(see \sref{App.}{App:Chord}). The combinatorial structure is captured  
by Catalan tables defined in \sref{Definition}{dfnt:cattab} which can be understood as Catalan tuples of 
Catalan tuples. The bijection between the recursive equation after cancellations and Catalan tables is given in 
\sref{Theorem}{thrm:RC}, which probably has a higher topological generalisation due 
to the general recursion \eqref{eq:quartRec}
for a correlation function of genus $g$ and $b$ boundary components.


\chapter{Conclusion and Outlook}
One of the main achievements of this thesis was the generalisation of the solution of the Kontsevich model
to the spectral dimensions $\D<8$. We found a method to derive intersection numbers 
on the moduli space $\overline{\mathcal{M}}_{g,b}$ with partial differential equations after having applied the 
theorem of Kontsevich \cite{Kontsevich:1992ti}. The definition of $c$ is changed 
by the renormalisation procedure, but the structure of all correlation functions is not.

Furthermore, the exact solution of the 2-point function
and therefore of any planar $N$-point function was computed
for the Grosse-Wulkenhaar model for $\D<6$. 
For the special case on the 4-dimensional Moyal space,
it was proved earlier that a solution exists \cite{Grosse:2012uv}.
However, \sref{Theorem}{prop:HT} gives the solution for any eigenvalue distribution of the 
external matrix $E$ of spectral dimension $\D<6$. 

The computed results coincide with perturbative calculations via Zimmermann's forest formula 
after having taken the same boundary conditions. 
To our knowledge, these two models provide the first bosonic examples for exactly solvable models 
which are just-renormalisable. 
Perturbatively, the number of graphs grows factorially and the renormalon appears (see \sref{Remark}{rmk:renorm} and 
\ref{rmk:renorm1}). Even in the cubic case, the $\beta$-function is positive (see \sref{Remark}{rmk:beta})
which is also the case in quantum electrodynamics. Nevertheless, the resummability implies that 
cancellations prevent the renormalon problem. The factorial growth of the Feynman graphs, 
the renormalon problem and a positive $\beta$-function do not imply that a just-renormalisable 
model is not resummable.
\\

What remains to investigate is whether the genus expansion
\begin{align*}
 G=\sum_{g=0}^\infty V^{-2g}G^{(g)}
\end{align*}
is Borel summable.
For a generic set of regular spectral curves, it was proved that the free energies $F^{(g)}$
obtained by topological recursion
grow at most like $\mathcal{O}((kg)!r^{-g})$ for some $r>0$ and 
$k\leq 5$ \cite{Eynard:2019mps}. But the property of Borel summability is 
still unclear, since the Borel transform needs an analytic continuation along the positive real line. 

From a quantum field theoretical perspective, the 
question of analytic continuation $V\to 0$ is of big 
interest because it recovers the commutative space. However, it should be treated with caution. 
The harmonic oscillator term breaks translational invariance for any $V<\infty$. 
Therefore, the first step would be to generalise the results 
of this thesis for QFTs on noncommutative spaces away from the 
self-dual point at $\Omega=1$. With the help of Meixner 
polynomials, it was proved that the quartic model, 
for instance, is in 4 dimensions perturbatively renormalisable 
for any $0<\Omega\leq 1$ to all orders \cite{Grosse:2004yu}.
However, the limit $\Omega\to 0$ can be performed in $D=2$, but not in $D=4$ because this
generates again the UV/IR mixing problem.

Further analysis suggests that the 4-point connected 
Schwinger function on the 4-dimensional Moyal space is not reflection positive.
Therefore, the naive construction of the Schwinger functions via equation \eqref{eq:Schwinger}
fails for a noncommutative space. 
A different way to construct the Schwinger function is proposed in \cite{Wulkenhaar2019}. 
On a noncommutative geometry, Connes' distance formula \cite{Connes94noncommutativegeometry} 
provides a metric structure
via states.
This implies that it is more natural for 
noncommutative geometry to construct 
Schwinger functions via states
(see \cite{Wulkenhaar2019} for a detailed construction).

From a more geometric perspective, it is a natural question whether the quartic matrix field 
theory model satisfies topological recursion \cite{Eynard:2007kz}. Comparing the 
conjectured spectral curve of the quartic model \eqref{spec:quart} with the one of the cubic model
\eqref{spec:cubic}
leads to an interesting observation. For finite matrices (take $V=\mN$ for simplicity), 
both spectral curves have the general description
(up to trivial factors of $2\lambda$)
\begin{align*}
 y(z)=z\mp\frac{\lambda}{\mN}\sum_{k=1}^{\mN'}\frac{r_k}{x'(\varepsilon_k)
 (\varepsilon_k-z)},\qquad e_k=x(\varepsilon_k),
\end{align*}
where $x(z)=z^2-c$ holds with the upper sign for the cubic model and $x(z)=z-
\frac{\lambda}{\mN}\sum_{k=1}^{\mN'}\frac{r_k}{x'(\varepsilon_k)(z+\varepsilon_k)}$ with the lower sign for the 
quartic model.

At this point, it is important also to recall the spectral curve of the Hermitian 1-matrix model and the
Hermitian 2-matrix model. Assuming a genus zero spectral curve for the Hermitian 2-matrix model yields the 
combined form for their meromorphic $y$-function with 
\begin{align*}
 y(z)=a z+\sum_{k=1}^{d-1}\frac{\alpha_k}{z^k},\qquad \text{where}\quad a=0\quad 
 \text{for the Hermitian 1-matrix model}
\end{align*}
and $d$ is the degree of the potential.
The meromorphic $x$-function is
\begin{align*}
 x(z)=&\gamma\bigg(z+\frac{1}{z}\bigg)+b &&\text{for the Hermitian 1-matrix model}\\
 x(z)=&\frac{a}{z}+\sum_{k=1}^{\tilde{d}-1}\beta_kz^k &&\text{for the Hermitian 2-matrix model},
\end{align*}
where $\tilde{d}$ is the degree of the second potential of the Hermitian 2-matrix model.
The constants $\gamma,a,b,\alpha_k,\beta_k$ are implicitly defined by the
coupling constants of the potential(s) \cite{Eynard:2002kg,Eynard:2016yaa}.
If both potentials of the Hermitian 2-matrix model coincide the number of free parameters is reduced with the
identities $\alpha_k=\beta_k$ and $\tilde{d}=d$. 

It should be emphasised that the meromorphic $x$-function of the cubic model and of the 
Hermitian 1-matrix model has a global symmetry given by $x(z)=x(-z)$ and $x(z)=x(\frac{1}{z})$, respectively.
By contrast, the $x$-function of the quartic model and the Hermitian 2-matrix model has no symmetry. 
However, another symmetry
can be observed by
\begin{align*}
 y(z)=&-x(-z)\qquad &&\text{for the quartic model}\\
 y(z)=&\,x\bigg(\frac{1}{z}\bigg)\qquad 
 &&\text{for the Hermitian 2-matrix model with coinciding potentials.}
\end{align*}
From this observations and the fact that the double-scaling limit of the Hermitian 1-matrix model
coincides with the large $\mN$ limit of the Kontsevich model \cite{Ambjorn:1993sj}, a relation between the 
quartic matrix field theory model and the Hermitian 2-matrix model (with coinciding potentials) can be conjectured. The 
spectral curves of these models are natural generalisations of the Hermitian 
1-matrix model and the cubic model, where the a symmetry between $x(z)$ and $y(z)$ is achieved with 
different global involutions. 

This conjecture is supported by the observation that the planar correlation
function with one mixed boundary in the Hermitian 2-matrix model \cite{Eynard:2005iq}
obeys exactly the same recursive equation as the planar $N$-point function for the quartic matrix
field theory (\sref{Example}{Exm:quartRec}). It indicates that correlation functions 
with mixed boundaries are possibly related in general to correlation functions of even length. 

In \sref{Sec.}{Sec.quartRec}, the combinatorial structure of \sref{Example}{Exm:quartRec} was analysed 
and proved to consists after cancellations 
of $d_k=\frac{1}{1+k}\binom{3k+1}{k}$ terms. This number can be understood due to 
\sref{Corollary}{cor:dk}
\begin{align*}
 d_{k}=\sum_{(e_{0},\ldots,e_{k+1})\in\mathcal{C}_{k+1}} 
c_{e_{0}-1}c_{e_{1}}\cdots c_{e_{k}}c_{e_{k+1}}
\end{align*}
as Catalan numbers of Catalan numbers. Studying the generalisation of  \sref{Example}{Exm:quartRec} given 
by equation \eqref{eq:quartRec} can give more insight into the conjecture.

Catalan numbers $c_n$ can be generalised in several forms. One of them is given by the numbers
$C_{g,b}$ which is a topological 
generalisation graded by the genus $g$ and the number of boundaries $b$. The usual Catalan number $c_n$
corresponds in this picture to the genus $g=0$ and $b=1$ case. The Laplace transform 
of the generating functions of $C_{g,b}$
was proved to satisfy topological recursion \cite{Dumitrescu:2012dka}, too.
A further question is whether the number $d_k$ has a topological generalisation as well, which should be 
encoded in the recursive equation \eqref{eq:quartRec}.

If a relation between the Hermitian 2-matrix model and the quartic model turns out to exist, 
the quartic model would not satisfy topological recursion in the sense of \cite{Eynard:2007kz},
because also the Hermitian 2-matrix model has a completely different topological recursive form 
\cite{Eynard:2007gw}. Additionally, the nonlinearity of equation \eqref{eq:quartRec} suggests 
that the quartic model has a more complex topological recursive structure. 

Characterising the poles and branch 
points for the correlation functions of Euler characteristic $\chi=0$ and 
$\chi=-1$ will be the first step. The branch points are difficult to handle
because they cannot be computed 
explicitly. 
In the large $V,\mN$-limit, the number of poles and branch points
tends to infinity which is (to our knowledge) different to 
any other model described by topological recursion, especially
the cubic model with coinciding pole and branch point at $z=0$.

\counterwithin{equation}{chapter}
\appendix 
\chapter{Moyal Base}\label{App:Moyal}
The following was developed in \cite{GraciaBondia:1987kw}. In our $D$-dimensional 
notation, it is formulated as follows:

Let the continuous variable be $x=(x_1,..,x_{D})\in \R^D$ and the discrete 
$\narrowunderline{n}=(n_1,..,n_{D/2})\in\N^{D/2}$. 
The $\star$-product is defined by 
\begin{align}\label{eq:Moyal}
 &(g\star h)(x)=\int \frac{d^D k}{(2\pi)^D}\int d^Dy\, g(x+\tfrac{1}{2}\Theta k)\,h(x+y)e^{\mathrm{i} k\cdot y},\\
 &\Theta=\mathds{1}_{D/2}\otimes \left( \begin{array}{rr}
0 & \theta   \\
 -\theta & 0 
\end{array}\right),\qquad  \theta\in \R,\qquad g,h\in\mathcal{S}(\R^D).\nonumber
\end{align}
Let $\theta>0$.
The starting point of finding the full base is the unit vector (or vacuum) given by the Gaussian
\begin{align*}
 b_0(x)=2^{D/2}e^{-\frac{1}{\theta}(x_1^2+x_2^2+..+x_D^2)}.
\end{align*}
Inserting into \eqref{eq:StarProd} leads after factorising the integrals to
\begin{align*}
 &(b_0\star b_0)(x)\\
 =&2^D\int d^Dy\int \frac{d^D k}{(2\pi)^D}e^{\I k\cdot y}
 e^{-\frac{1}{\theta}\sum_{i=1}^{D/2} (x_{2i-1}+\frac{\theta}{2}k_{2i})^2+(x_{2i}+\frac{\theta}{2}k_{2i-1})^2+
 (x_{2i-1}+y_{2i-1})^2+(x_{2i}+y_{2i})^2}\\
 =&b_0(x).
\end{align*}
The creation and annihilation operators are defined by 
\begin{align*}
 a_i(x)&=\frac{1}{\sqrt{2}}(x_{2i-1}+\I x_{2i}), &&a^{\dagger}_i(x)=\frac{1}{\sqrt{2}}(x_{2i-1}-\I x_{2i}),\\
 \frac{\partial}{\partial a_i}&=\frac{1}{\sqrt{2}}(\frac{\partial}{\partial x_{2i-1}}
 -\I \frac{\partial}{\partial x_{2i}}), &&\,\,\,\frac{\partial}{\partial a^\dagger_i}=
 \frac{1}{\sqrt{2}}(\frac{\partial}{\partial x_{2i-1}}
 +\I \frac{\partial}{\partial x_{2i}}).
\end{align*}
Direct calculation leads for $g\in\mathcal{S}(\R^D)$ to
\begin{align}\label{App:Cond11}
 (a_i\star g)(x)&=a_i(x)g(x)+\frac{\theta}{2}\frac{\partial g}{\partial a^\dagger_i}(x),&& 
 (g\star a_i)(x)=a_i(x)g(x)-\frac{\theta}{2}\frac{\partial g}{\partial a^\dagger_i}(x),\\\label{App:Cond12}
 (a^\dagger_i\star g)(x)&=a^\dagger_i(x)g(x)-\frac{\theta}{2}\frac{\partial g}{\partial a_i}(x),&& 
 (g\star a^\dagger_i)(x)=a^\dagger_i(x)g(x)+\frac{\theta}{2}\frac{\partial g}{\partial a_i}(x).
\end{align}
Since $a_i,a_j$ (and $a^\dagger_i,a^\dagger_j$) act for $i\neq j$ on different variables, they commute 
with respect to the star product
\begin{align*}
 &([a_i,a_j]_\star\star g)(x)=([a_i,a^\dagger_j]_\star\star g)(x)=([a^\dagger_i,a_j]_\star\star g)(x)
 =([a^\dagger_i,a^\dagger_j]_\star\star g)(x)\\
 =&(g\star [a_i,a_j]_\star)(x)=(g\star[a_i,a^\dagger_j]_\star)(x)=(g\star[a^\dagger_i,a_j]_\star)(x)
 =(g\star[a^\dagger_i,a^\dagger_j]_\star)(x)=0,
\end{align*}
where $[g,h]_\star=g\star h-h\star g$. 
Acting with $a_i$ from the left or with $a^\dagger_j$ from the right for some $i,j\in\{1,..,D/2\}$ 
on the unit $b_0$ with respect to the $\star$-product vanishes exactly
\begin{align*}
 (a_i\star b_0)(x)=(b_0\star a_j^\dagger)=0,
\end{align*}
whereas $a_i^\dagger$ from the left and $a_j$ from the right take the r\^ole as creation operator with
\begin{align*}
 ((a_i^\dagger)^{\star m}\star b_0)(x)=&2^m(a^\dagger_i)^m(x)b_0(x)\\
 ( b_0\star(a_j)^{\star n})(x)=&2^n(a_j)^n(x)b_0(x),
\end{align*}
where $g^{\star m}=g\star g\star..\star g$ exactly $m$ times.
One further verifies for $m,n>0$
\begin{align}\label{App:Cond1}
 (a_i\star(a_i^\dagger)^{\star m}\star b_0)(x)=&m\theta ((a_i^\dagger)^{\star (m-1)}\star b_0)(x),\\\label{App:Cond2}
 (b_0\star a_j^{\star n}\star a_j)(x)=&n\theta ( b_0\star a_j^{\star (n-1)})(x).
\end{align}
Defining now the base with $i\in\{1,..,\frac{D}{2}\}$ for the 2-dimensional space
\begin{align*}
 f_{n_im_i}(x_{2i-1},x_{2i}):=\frac{1}{\sqrt{n_i!
 m_i!\theta^{n_i+m_i}}}
 ((a_i^\dagger)^{\star n} \star 2e^{-\frac{1}{\theta}(x_i^2+x_{i+1}^2)} \star a_i^{\star m_i})(x)
\end{align*}
leads to the base of the $D$-dimensional Moyal space
with the correct normalisation factor by
\begin{align*}
 b_{\narrowunderline{n},\narrowunderline{m}}(x):&=\frac{1}{\sqrt{\narrowunderline{n}!
 \narrowunderline{m}!\theta^{|\narrowunderline{n}|+|\narrowunderline{m}|}}}
 ((a^\dagger)^{\star \narrowunderline{n}} \star b_0 \star a^{\star \narrowunderline{m}})(x)\\
 =&f_{n_1m_2}
 (x_1,x_2)f_{n_2m_2}(x_3,x_4)..f_{n_{D/2}m_{D/2}}(x_{D-1}x_D),
\end{align*}
where $\narrowunderline{n}!=n_1!n_2!..n_{D/2}!$, $|\narrowunderline{n}|=n_1+n_2+..+n_{D/2}$, 
$(a^\dagger)^{\star \narrowunderline{n}}=(a_1^\dagger)^{\star n_1}\star(a_2^\dagger)^{\star n_2}\star
..\star(a_{D/2}^\dagger)^{\star n_{D/2}}$ and $a^{\star \narrowunderline{m}}=(a_1)^{\star m_1}\star(a_2)^{\star m_2}\star
..\star(a_{D/2})^{\star m_{D/2}}$.

The matrix multiplication rule follows by \eqref{App:Cond1} and \eqref{App:Cond2}
\begin{align*}
 (b_{\narrowunderline{n}\narrowunderline{m}}\star b_{\narrowunderline{k}
 \narrowunderline{l}})(x)=\delta_{\narrowunderline{m},\narrowunderline{k}}b_{\narrowunderline{n}
 \narrowunderline{l}}(x),
\end{align*}
where $\delta_{\narrowunderline{n},\narrowunderline{m}}=\delta_{n_1,m_1}..\delta_{n_{D/2}}\delta_{m_{D/2}}$.

For $g,h\in\mathcal{S}(\R^D)$ we have
\begin{align}\label{App:MoyalTr}
 \int d^Dx\,(g\star h)(x)=\int d^Dx\,g(x)h(x).
\end{align}
This is proved by inserting the defintion \eqref{eq:Moyal} and transforming 
$x\to x'=x+\frac{1}{2}\Theta k$ to get $g$ independent of the $k$ and $y$ integral. 
The next transformation is
$y\to y'=y-\frac{1}{2}\Theta k$, where the exponential become $\I k \cdot (y'+\frac{1}{2}\Theta k)=
\I k \cdot y'$. Now, the integral over $k$ is performed and gives a $\delta$-distribution $\delta(y')$. The integral over 
$y'$ leads to $g(x)h(x)$ inside the integral over $x$.

The condition \eqref{App:MoyalTr} implies a cyclic property (as for traces)
\begin{align}\label{App:Trace}
 \int d^Dx\,(g\star h)(x)=\int d^Dx\,g(x)h(x)=\int d^Dx\,(h\star g)(x).
\end{align}

furthermore, the base $b_{\narrowunderline{n}\narrowunderline{m}}$ obeys a trace property
\begin{align*}
 &\int d^Dx\, b_{\narrowunderline{n}\narrowunderline{m}}(x)\\
 =&\frac{1}{\sqrt{\narrowunderline{n}!
 \narrowunderline{m}!\theta^{|\narrowunderline{n}|+|\narrowunderline{m}|}}}
 \int d^Dx\,((a^\dagger)^{\star \narrowunderline{n}} \star b_0 \star a^{\star \narrowunderline{m}})(x)\\
 =&\frac{1}{\sqrt{\narrowunderline{n}!
 \narrowunderline{m}!\theta^{|\narrowunderline{n}|+|\narrowunderline{m}|}}}
 \int d^Dx\,(  a^{\star \narrowunderline{m}}\star (a^\dagger)^{\star \narrowunderline{n}}\star b_0)(x)\\
 =&\delta_{\narrowunderline{n},\narrowunderline{m}}\int d^Dx\,b_0(x)
 =\delta_{\narrowunderline{n},\narrowunderline{m}}(2\pi \theta)^{D/2},
\end{align*}
where we have used the cyclic property and \eqref{App:Cond1},\eqref{App:Cond2} in the third line.

Any Schwartz function $g\in \mathcal{S}(\R^D)$ can be expanded into a convergent series
\begin{align*}
 g(x)=\sum_{\narrowunderline{n},\narrowunderline{m}}g_{\narrowunderline{n}\narrowunderline{m}}
 b_{\narrowunderline{n}\narrowunderline{m}}(x),
\end{align*}
where the sum goes over all $n_i,m_j$ from $0$ to infinity and $g_{\narrowunderline{n}\narrowunderline{m}}=
g_{n_1m_1}..g_{n_{D/2}m_{D/2}}$. Furthermore, $g$ lies in $\mathcal{S}(\R^D)$ iff
\begin{align*}
 \sum_{n_i,m_i=0}^\infty ((2n_i+1)^{2k}(2n_i+1)^{2k} |g_{n_im_i}|^2)^{1/2}<\infty
\end{align*}
for all $i,k$.

Let $\Phi\in \mathcal{S}(\R^D)$ be the scalar field. 
The kinetic part (with the harmonic oscillator) of the action for scalar fields on the Moyal space is
\begin{align}\label{App:Kinetic}
 &\int \frac{d^Dx}{(8\pi)^{D/2}} \left(\frac{1}{2}\Phi\star 
 (-\Delta +\Omega^2\|2\Theta^{-1}\cdot x\|^2+\mu_0^2)\Phi\right)(x)\\
 =&\sum_{\narrowunderline{n},\narrowunderline{m},\narrowunderline{k},\narrowunderline{l}\in\N^{D/2}}
 \Phi_{\narrowunderline{n}\narrowunderline{m}}
 G_{\narrowunderline{n}\narrowunderline{m}
  ;\narrowunderline{k}\narrowunderline{l}}\Phi_{\narrowunderline{k}\narrowunderline{l}}
\end{align}
with the regulator $\Omega\in \R$, the Laplacian $\Delta=\sum_{i=1}^{D/2}(\frac{\partial^2}{\partial x_{2i-1}^2}+
\frac{\partial^2}{\partial x_{2i}^2})=\sum_{i=1}^{D/2}
\frac{\partial^2}{\partial a_i\partial a^\dagger_i}+\frac{\partial^2}{\partial a^\dagger_i\partial a_i}$ and expansion
 $\Phi(x)=\sum_{\narrowunderline{n},\narrowunderline{m}}\Phi_{\narrowunderline{n}\narrowunderline{m}}
 b_{\narrowunderline{n}\narrowunderline{m}}(x)$.  
 The kinetic term reads in the matrix base with $\|x\|^2=\sum_{i=1}^{D/2} a_ia^\dagger_i+a^\dagger_ia_i$ 
 and \eqref{App:Cond11}, \eqref{App:Cond12}
 \begin{align}\nonumber
  &G_{\narrowunderline{n}\narrowunderline{m}
  ;\narrowunderline{k}\narrowunderline{l}}=\int \frac{d^Dx}{(8\pi)^{D/2}}
  \bigg(\frac{1}{2}b_{\narrowunderline{n}\narrowunderline{m}}\star 
 \bigg(-\sum_{i=1}^D\frac{\partial^2}{\partial x_i^2} +\frac{4\Omega^2}{\theta^2}\|x\|^2+\mu_0^2\bigg)
 b_{\narrowunderline{k}\narrowunderline{l}}\bigg)(x)\\\nonumber
 =&\int \frac{d^Dx}{(8\pi)^{D/2}}
  \bigg(\frac{1}{2}b_{\narrowunderline{n}\narrowunderline{m}}\star 
 \bigg(-\sum_{i=1}^{D/2}\bigg(\frac{\partial^2}{\partial a^\dagger_i\partial a_i}+
 \frac{\partial^2}{\partial a_i\partial a^\dagger_i}
 -\frac{4\Omega^2}{\theta^2}(a_ia_i^\dagger+a^\dagger_ia_i)\bigg)+\mu_0^2\bigg)
 b_{\narrowunderline{k}\narrowunderline{l}}\bigg)(x)\\\nonumber
 =&\int \frac{d^Dx}{(8\pi)^{D/2}}
  \bigg(\frac{1+\Omega^2}{2\theta^2}b_{\narrowunderline{n}\narrowunderline{m}}\star 
 \sum_{i=1}^{D/2} (a_i\star a_i^\dagger+a_i^\dagger\star a_i)\star b_{\narrowunderline{k}\narrowunderline{l}}\\\nonumber
 &\qquad \qquad+\frac{1+\Omega^2}{2\theta^2}b_{\narrowunderline{n}\narrowunderline{m}}\star b_{\narrowunderline{k}\narrowunderline{l}}\star 
 \sum_{i=1}^{D/2} (a_i\star a_i^\dagger+a_i^\dagger\star a_i)\\\nonumber
 &\qquad \qquad-\frac{1-\Omega^2}{\theta^2}b_{\narrowunderline{n}\narrowunderline{m}}\star 
 \sum_{i=1}^{D/2} (a^\dagger_i\star b_{\narrowunderline{k}\narrowunderline{l}} \star a_i
 +a_i\star b_{\narrowunderline{k}\narrowunderline{l}} \star a^\dagger_i)
 +\frac{\mu_0^2}{2}
  b_{\narrowunderline{n}\narrowunderline{m}}\star b_{\narrowunderline{k}\narrowunderline{l}}\bigg)(x)\\\label{App:Laplacian}
  =&\left(\frac{\theta}{4}\right)^{D/2}\bigg(\frac{\mu_0^2}{2}+\frac{1+\Omega^2}{\theta}\sum_{i=1}^{D/2}(1+n_i+m_i)\bigg)
 \delta_{\narrowunderline{m},\narrowunderline{k}}\delta_{\narrowunderline{n},\narrowunderline{l}}\\\nonumber
 &- \left(\frac{\theta}{4}\right)^{D/2}\frac{1-\Omega^2}{\theta}\sum_{i=1}^{D/2}(\sqrt{n_im_i} \delta_{n_i-1,l_i}
 \delta_{m_i-1,k_i}+\sqrt{k_il_i} \delta_{n_i+1,l_i}
 \delta_{m_i+1,k_i})\check{\delta}^i_{\narrowunderline{m},\narrowunderline{k}}\check{\delta}^i_{\narrowunderline{n},\narrowunderline{l}},
 \end{align}
where $\check{\delta}^i_{\narrowunderline{n},\narrowunderline{m}}$ is $\delta_{\narrowunderline{n},\narrowunderline{m}}$
with omitted $\delta_{n_i,m_i}$.

\chapter{Schwinger Function on the Moyal Space}\label{App:Schwinger}
We are following the derivation given in \cite{Grosse:2013iva} and 
extend it to any dimension $D$. 
The connected Schwinger function is defined as
\begin{align}\label{App:eq:Schwinger}\nonumber
	S_c(\xi_1,..,\xi_N):=&\lim_{\Lambda^2\to\infty}\lim_{\stackrel{V,\mN'\to\infty}
		{\frac{\mN'}{V^{2/D}}=\Lambda^2}}\sum_{N_1+..+N_b=N} \sum_{\narrowunderline{p}_1^1,..,\narrowunderline{p}_{N_b}^b\in \mathbb{N}^{D/2}_{\mN'}}\frac{G_{|\narrowunderline{p}_1^1..
			\narrowunderline{p}_{N_1}^1|..|\narrowunderline{p}_1^b..
			\narrowunderline{p}_{N_b}^b|}}{(8\pi)^{D/2} b!}\\
	&\times \sum_{\sigma\in S_N} \prod_{\beta=1}^{b}\frac{b_{\narrowunderline{p}_1^\beta \narrowunderline{p}_2^\beta}(\xi_{\sigma(s_\beta+1)})..b_{\narrowunderline{p}_{N_\beta}^\beta \narrowunderline{p}_1^\beta}(\xi_{\sigma(s_\beta+N_\beta)})}{VN_\beta},
\end{align}
where $s_\beta=N_1+..+N_{beta-1}$ and $S_N$ is the symmetric group.

Abbriviations are used for the correlation functions
$G_{|\narrowunderline{p}_1^1..
	\narrowunderline{p}_{N_1}^1|..|\narrowunderline{p}_1^b..
	\narrowunderline{p}_{N_b}^b|}=G_{|\narrowunderline{P}^{N_1}|..|\narrowunderline{P}^{N_b}|}$ 
	with  $\narrowunderline{P}^{N_\beta}=\narrowunderline{p}_1^\beta..\narrowunderline{p}_{N_\beta}^\beta$.
	Additionally, we define the total norm 
$\|\narrowunderline{P}^{N_\beta}\|=|\narrowunderline{p}_1^\beta|+..+|\narrowunderline{p}_{N_\beta}^\beta|$. Furthermore, 
we assume that the correlation
function $G_{|\narrowunderline{P}^{N_1}|..|\narrowunderline{P}^{N_b}|}$ has a Fourier-Laplace transform $\tilde{G}_{\mN,V}(t_1,\narrowunderline{\omega}^{N_1}|..|t_b,\narrowunderline{\omega}^{N_b})$ depending on $\mN$ and $V$, where the Laplace transform is
used to distinguish the time $t_i\geq 0$ coordinate and $\narrowunderline{\omega}^{N_\beta}=(\omega^\beta_1,..,\omega^\beta_{N_\beta-1})\in\mathbb{R}^{N_\beta-1}$. Define further the product
$\langle \narrowunderline{\omega}^{N_\beta},\narrowunderline{P}^{N_\beta}\rangle=\sum_{i=1}^{N_\beta-1}\omega_i^\beta (|\narrowunderline{p}_i^\beta|-
|\narrowunderline{p}_{i+1}^\beta|)$ such that the Fourier-Laplace transform
takes the form
\begin{align}\label{App:Laplace}\nonumber
	G_{|\narrowunderline{P}^{N_1}|..|\narrowunderline{P}^{N_b}|}=
	\int_{\mathbb{R}^b_+}d(t_1,..,t_b)&\int_{\mathbb{R}^{N-b}}d(\narrowunderline{\omega}^{N_1},..,\narrowunderline{\omega}^{N_b})\tilde{G}_{\mN,V}(t_1,\narrowunderline{\omega}^{N_1}|..|t_b,\narrowunderline{\omega}^{N_b})
	\\&\times\prod_{\beta=1}^{b}e^{-\frac{t_\beta }{V^{2/D}}\|\narrowunderline{P}^{N_\beta}\|+\frac{\mathrm{i}}{V^{2/D}}\langle \narrowunderline{\omega}^{N_\beta},\narrowunderline{P}^{N_\beta}\rangle}
\end{align}
Since the limit $\mN',V\to\infty$ is taken,
the Fourier-Laplace transform converges in this limit to 
$\tilde{G}(t_1,\narrowunderline{\omega}^{N_1}|..|t_b,\narrowunderline{\omega}^{N_b})$ 
such that the inverse Fourier-Laplace transformation gets independent of $\mN,V$. 

Inserting \eqref{App:Laplace} into \eqref{App:eq:Schwinger} the index-sum and the product can be interchanged 
\begin{align*}
	\sum_{\narrowunderline{p}_1^1,..,\narrowunderline{p}_{N_b}^b\in \mathbb{N}^{D/2}_{\mN'}} 
	 \prod_{\beta=1}^{b}= \prod_{\beta=1}^{b} \sum_{\narrowunderline{p}_1^\beta,..,\narrowunderline{p}_{N_\beta}^\beta \in \mathbb{N}^{D/2}_{\mN'}}
\end{align*}
since the Fourier-Laplace transform $\tilde{G}_{\mN,V}(t_1,\narrowunderline{\omega}^{N_1}|..|t_b,\narrowunderline{\omega}^{N_b})$ 
is independent of the index-sum and the exponentials factorise. For each $\beta\in \{1,..,b\}$ the following 
series has to be determined
\begin{align}\nonumber
	&\frac{1}{V N_\beta }\sum_{\narrowunderline{p}_1^\beta,..,\narrowunderline{p}_{N_\beta}^\beta \in \mathbb{N}^{D/2}_{\mN'}}
	b_{\narrowunderline{p}_1^\beta \narrowunderline{p}_2^\beta}(\xi_{\sigma(s_\beta+1)})..b_{\narrowunderline{p}_{N_\beta}^\beta \narrowunderline{p}_1^\beta}(\xi_{\sigma(s_\beta+N_\beta)}) 
	e^{-\frac{t_\beta }{V^{2/D}}\|\narrowunderline{P}^{N_\beta}\|+\frac{\mathrm{i}}{V^{2/D}}\langle \narrowunderline{\omega}^{N_\beta},\narrowunderline{P}^{N_\beta}\rangle}\\
	=&\frac{1}{V N_\beta }\sum_{\narrowunderline{p}_1^\beta,..,\narrowunderline{p}_{N_\beta}^\beta \in \mathbb{N}^{D/2}_{\mN'}}
	b_{\narrowunderline{p}_1^\beta \narrowunderline{p}_2^\beta}(\xi_{\sigma(s_\beta+1)})..
	b_{\narrowunderline{p}_{N_\beta}^\beta \narrowunderline{p}_1^\beta}(\xi_{\sigma(s_\beta+N_\beta)})  
	(z^\beta_1)^{|\narrowunderline{p}_1^\beta|}..(z^\beta_{N_\beta})^{|\narrowunderline{p}
	_{N_\beta}^\beta|},\label{App:bFormel}
	\\
	&\qquad \text{where}\qquad z^\beta_{i}:= \left\{
	\begin{array}{ll}
		e^{-\frac{t_\beta}{V^{2/D}}+\mathrm{i}\frac{\omega_1^\beta}{V^{2/D}}} & i=1 \\
		e^{-\frac{t_\beta}{V^{2/D}}+\mathrm{i}\frac{\omega_i^\beta-\omega_{i-1}^\beta}{V^{2/D}}} & i\in\{2,..,N_\beta-1\} . \\
		e^{-\frac{t_\beta}{V^{2/D}}-\mathrm{i}\frac{\omega_{N_\beta-1}^\beta}{V^{2/D}}} & i=N_\beta \\	
	\end{array}
	\right.\nonumber
\end{align}
\\

Recall next a lemma and its corollary proved in \cite{ Grosse:2013iva} about the 2D base $f_{nm}(\xi),\xi\in\R^2$ 
of the Moyal plane
\begin{Alemma}(\cite[Lemma 4 + Corollary 5]{Grosse:2013iva})\label{Lemma:Raimar}
 Let for  $\xi,\eta\in\R^2$ be the scalar product, the norm and 
 the determinant of the matrix $(\xi,\eta)$ given by
  $\langle \xi,\eta\rangle ,\, \|\xi\|$ and $\det(\xi,\eta)$, respetively. Then, for $\xi_i\in\R^2$, $z_i\in\C$ with
  $|z_i|<1$ and $N\in\N$ (with the cyclic structure $N+i\equiv i$)
  \begin{align*}
   &\sum_{p_1,..,p_N=0}^\infty \prod_{i=1}^Nf_{p_ip_{i+1}}(\xi_i)z_i^{p_i}\\
   =&\frac{2^N}{1-\prod_{i=1}^N(-z_i)}\exp\left(-\frac{\sum_{i=1}^N\|\xi_i\|^2}{4 V^{2/D}}\,\frac{1+\prod_{i=1}^N(-z_i)}
   {1-\prod_{i=1}^N(-z_i)}\right)\\
   &\times \exp\left(-\sum_{1\leq k<l\leq N} \bigg(\frac{\langle \xi_k,\xi_l\rangle 
   -\mathrm{i} \det(\xi_k,\xi_l)}{2V^{2/D}}\,\frac{\prod_{i=k+1}^l(-z_i)}
   {1-\prod_{i=1}^N(-z_i)}\bigg)\right)
   \\ &\times \exp\left(-\sum_{1\leq k<l\leq N} \bigg(\frac{\langle \xi_k,\xi_l\rangle 
   +\mathrm{i} \det(\xi_k,\xi_l)}{2V^{2/D}}\,\frac{\prod_{i=l+1}^{N+k}(-z_i)}
   {1-\prod_{i=1}^N(-z_i)}\bigg)\right).
  \end{align*}
\end{Alemma}
Let $\mathbb{N}^{D/2} \ni \narrowunderline{p}_i^\beta=((p_i^\beta)_1
,..,(p_i^\beta)_{D/2})$ and $\mathbb{R}^{D/2}\ni\chi_i
=((\chi_i)_1,..,(\chi_i)_{D/2})$, then \eqref{App:bFormel} can be rearranged to
\begin{align*}
 &\sum_{\narrowunderline{p}_1^\beta,..,\narrowunderline{p}_{N_\beta}^\beta \in \mathbb{N}^{D/2}}
	b_{\narrowunderline{p}_1^\beta \narrowunderline{p}_2^\beta}(\xi_{\sigma(s_\beta+1)})..
	b_{\narrowunderline{p}_{N_\beta}^\beta \narrowunderline{p}_1^\beta}(\xi_{\sigma(s_\beta+N_\beta)})  
	(z^\beta_1)^{|\narrowunderline{p}_1^\beta|}..(z^\beta_{N_\beta})^{|\narrowunderline{p}
	_{N_\beta}^\beta|}\\
	=&\sum_{\narrowunderline{p}_1^\beta,..,\narrowunderline{p}_{N_\beta}^\beta \in \mathbb{N}^{D/2}}
	\prod_{i=1}^{D/2} \prod_{j=1}^{N_\beta}
	f_{(p_j^\beta)_i(p_{j+1}^\beta)_i}((\xi_{\sigma(s_\beta+j)})_{i},
	(\xi_{\sigma(s_\beta+j)})_i)(z_j^\beta)^{(p_j^\beta)_i}\\
	=&\prod_{i=1}^{D/2}
	\sum_{(p_1^\beta)_i,..,(p_{N_\beta}^\beta)_i =0}^\infty
	\prod_{j=1}^{N_\beta}f_{(p_j^\beta)_i(p_{j+1}^\beta)_i}((\xi_{\sigma(s_\beta+j)})_{i},
	(\xi_{\sigma(s_\beta+j)})_i)(z_j^\beta)^{(p_j^\beta)_i}
\end{align*}
in the $\mN'\to \infty$ limit such that \sref{Lemma}{Lemma:Raimar} can be used. The splitting of the double scaling 
limit $V,\mN'\to\infty$ with $\frac{\mN'}{V^{2/D}}=\Lambda^2$ and interchanging it with
$\tilde{G}_{\mN',V}$ is uncritical. However, all $z_i^\beta$'s get critical $|z_i^\beta|=1$. 
In an intermediate step by deriving 
\sref{Lemma}{Lemma:Raimar}, the series $\sum_{p=0}^\infty\frac{(p+k)!}{p!k!}((-z_1)...(-z_N))^p=\frac{1}{(1-
((-z_1)...(-z_N)))^{k+1}}$ has to be taken, where a restriction to $\mN'$ summands would give an error term 
proportional to $(z_1..z_N)^{\mN'}=e^{-\Lambda^2 N t}$. This error term vanishes in the following limit
of $V\to\infty$ together with the final limit of $\Lambda^2\to\infty$. With this argumentation,
\sref{Lemma}{Lemma:Raimar} can be used. 

As $|z_i|$ converges to 1 in the $V\to\infty$ limit, the denominator 
$1-\prod_{i=1}^N(-z_i)=1-(-1)^Ne^{-\frac{Nt}{V^{D/2}}}$ converges to 2 for $N$ odd. For even $N$, the denominator 
converges at leading order in $V$ to $\frac{N t}{V^{D/2}}$. Therefore, performing the limits in the discussed way 
leads to 
\begin{align*}
 &\lim_{\Lambda^2\to\infty}\lim_{\stackrel{V,\mN'\to\infty}
		{\frac{\mN'}{V^{2/D}}=\Lambda^2}}
	\sum_{\narrowunderline{p}_1^\beta,..,\narrowunderline{p}_{N_\beta}^\beta \in \mathbb{N}^{D/2}}\frac{
	b_{\narrowunderline{p}_1^\beta \narrowunderline{p}_2^\beta}(\xi_{\sigma(s_\beta+1)})..
	b_{\narrowunderline{p}_{N_\beta}^\beta \narrowunderline{p}_1^\beta}(\xi_{\sigma(s_\beta+N_\beta)})  
	(z^\beta_1)^{|\narrowunderline{p}_1^\beta|}..(z^\beta_{N_\beta})^{|\narrowunderline{p}
	_{N_\beta}^\beta|}}{VN_\beta}\\
	= & \left\{
\begin{array}{ll}
 \frac{2^{D N_\beta/2}}{N_\beta (N_\beta t_\beta)^{D/2}} \exp\left(-\frac{\|\sum_{j=1}^{N_\beta}(-1)^{j-1}
	\xi_{\sigma(s_\beta+j)}\|^2}
	{2N_\beta t_\beta}\right) \quad\quad\quad & \text{for $N_\beta$ even} \\
0 & \text{for $N_\beta$ odd.} \\
\end{array}
\right. 	
\end{align*}

Rewrite this result with the Gaussian integral to
\begin{align*}
 &\frac{2^{D N_\beta/2}}{N_\beta (N_\beta t_\beta)^{D/2}} \exp\left(-\frac{\|\sum_{j=1}^{N_\beta}(-1)^{j-1}
	\xi_{\sigma(s_\beta+j)}\|^2}
	{2N_\beta t_\beta}\right)\\
	=&\frac{2^{D N_\beta/2}}{(2\pi)^{D/2} N_\beta}
	\int_{\R^{D/2}}dp_\beta\, e^{-\frac{\N_\beta}{2}\|p_\beta\|^2t_\beta}\,\,
	e^{\mathrm{i} \langle p_\beta,\sum_{j=1}^{N_\beta}(-1)^{j-1}
	\xi_{\sigma(s_\beta+j)}\rangle }.
\end{align*}
Inserting this back into \eqref{App:eq:Schwinger} together with \eqref{App:Laplace} give the inverse 
Laplace transform in the variable $t_\beta$. The integrals over $\narrowunderline{\omega}^{N_\beta}$
give the inverse Fourier transform if within each boundary $\beta$ the variables are taken to be the same. 
We end up in the following representation of the Schwinger function 
\begin{align*}
 S_c(\xi_1,..,\xi_N)=&\sum_{\stackrel{N_1+..+N_b=N}{N_\beta\, \text{even}}}
 \sum_{\sigma\in S_N}\left(\prod_{\beta=1}^b \frac{2^{DN_\beta/2}}{N_\beta}\int\frac{d^D p^{\beta}}{(2\pi)^{D/2}}
 e^{\mathrm{i} p^\beta\cdot (\xi_{\sigma (s_\beta+1)}-
 \xi_{\sigma (s_\beta+2)}+..-\xi_{\sigma (s_\beta+N_\beta)})}\right)\\\nonumber
 &\times\frac{1}{(8\pi)^{D/2}b!}G^0\bigg(\underbrace{\frac{\|p^1\|^2}{2},..,\frac{\|p^1\|^2}{2}}_{N_1}\bigg\vert...\bigg\vert
 \underbrace{\frac{\|p^b\|^2}{2},..,\frac{\|p^b\|^2}{2}}_{N_b}\bigg),
\end{align*}
\begin{align*}
 \text{where}\quad &
 G^g(x_1^1,..,x_{N_1}^1|..|x_1^b,..,x_{N_b}^b)=\lim_{V,\mN',\Lambda^2}
 G^{(g)}_{|\narrowunderline{p}_1^1..\narrowunderline{p}_{N_1}^1|..|\narrowunderline{p}^b_{1}
 ..\narrowunderline{p}^b_{N_b}|}\big\vert_{|\narrowunderline{p}_i^j|=x_i^jV^{2/D}}.
\end{align*}

\chapter[Relevant Lemmata for 
Theorem 3.1]{Relevant Lemmata for 
\sref{Theorem}{finaltheorem}\footnote{This is taken from the appendix of our 
paper \cite{Grosse:2019nes}}}\label{appendixC}
\begin{assumption}\label{conj1}
We assume that $\G_g(z)$ is, for $g\geq 1$, a function of $z$ and of 
$\varrho_0,\dots,\varrho_{3g-2}$ (true for $g=1$). We take eq.\
\eqref{eqfinalthm} and in particular   
$\G_g(z,J):=(2\lambda)^3 \hat{\mathrm{A}}^{\dag g}_{J,z} 
\G_g(J)$
as a {\sffamily definition} of a family of functions
$\G_g(z_{1},J)$ and derive equations for that family. 
\end{assumption}

\begin{Alemma}\label{lemma5}
Let $J=\{z_2,...,z_b\}$. Then under \sref{Assumption}{conj1}
and with \sref{Definition}{defint} of 
the operator $\hat{K}_{z_1}$ one has
\begin{align*}
\hat{K}_{z_1} \G_g(z_1,J)
&=\frac{(2\lambda)^3}{z_1^2}
\bigg(\sum_{l= 0}^{3g-3+|J|} 
(3+2l)\frac{\partial \G_g(J)}{\partial \varrho_l}
\sum_{k=0}^{l}\frac{\varrho_k}{z_1^{2+2l-2k}}
+\sum_{\zeta\in J}\frac{1}{\zeta} \frac{\partial}{\partial \zeta}
\G_g(J)\bigg).
\end{align*}
\end{Alemma}
\begin{proof}
Take \sref{Definition}{DefOp} for $\hat{\mathrm{A}}^{\dag g}_{J,z} 
\G_g(J)$ 
and apply \sref{Lemma}{lemmaopK}.
\end{proof}

\begin{Alemma}\label{lemma6} Let $J=\{z_2,...,z_b\}$.
Then under \sref{Assumption}{conj1} one has 
\begin{align*}
&\frac{(2\lambda)^3}{\zeta} \frac{\partial}{\partial \zeta}
\frac{\G_g(z_1,{J\backslash\{\zeta\}})
-\G_g(\zeta,{J\backslash\{\zeta\}})}{z_1^2-\zeta^2}
+2\lambda \G_0(z_1,\zeta )\G_{g}(z_1,{J\backslash \{\zeta \}})
\\
&=
(2\lambda)^6
\Bigg[
\sum_{l= 0}^{3g-4+|J|} 
\Bigg(
-\sum_{n=0}^1 
\frac{(3+2l)(1+2n) \varrho_{l+1}}{\varrho_0 
z_1^{4-2n}\zeta^{3+2n}}
+
\sum_{n=0}^{l+2}
\frac{(3+2l)(1+2n)}{z_1^{6+2l-2n}\zeta^{3+2n}}
\Bigg)
\frac{\partial \G_g({J\backslash\{\zeta \}})}{\partial \varrho_l}
\\
&+\sum_{\xi\in J\backslash \{\zeta \}}
\sum_{n=0}^1\frac{1+2n}{\varrho_0 \xi z_1^{4-2n}\zeta^{3+2n}}
\frac{\partial \G_g({J\backslash\{\zeta \}})}{\partial \xi}
\Bigg].
\end{align*}
\end{Alemma}
\begin{proof}
\sref{Definition}{DefOp} gives with 
$\frac{\frac{1}{z_1^{3+2j}}-\frac{1}{y^{3+2j}}}{z_1^2-y^2}=
-\sum_{l=0}^{2j+2}\frac{z_1^ly^{2j+2-l}}{z_1^{3+2j}y^{3+2j}(z+y)}$
for the first term
\begin{align*}
&\frac{\G_g(z_1,{J\backslash\{\zeta\}})
-\G_g(\zeta,{J\backslash\{\zeta\}})}{
(2\lambda)^3(z_1^2-\zeta^2)}
\\
&= \sum_{l=0}^{3g-4+|J|}
\left(
-\frac{(3+2l) \varrho_{l+1}}{\varrho_0}
\Bigg(\frac{\frac{1}{z_1^3}-\frac{1}{\zeta^3}}{z_1^2-\zeta^2}\Bigg)
+(3+2l) 
\Bigg(\frac{\frac{1}{z_1^{5+2l}}-\frac{1}{\zeta^{5+2l}}}{z_1^2-\zeta^2}\Bigg)
\right)
\frac{\partial \G_g({J\backslash\{\zeta \}})}{\partial \varrho_l}
\\
&+\sum_{\xi\in J\backslash \{\zeta \}}\frac{1}{\varrho_0 \xi} 
\Bigg(\frac{\frac{1}{z_1^3}-\frac{1}{\zeta^3}}{z_1^2-\zeta^2}\Bigg)
\frac{\partial \G_g({J\backslash\{\zeta \}})}{\partial \xi}
\\
&=\sum_{l= 0}^{3g-4+|J|} 
\Bigg(
\frac{(3+2l) \varrho_{l+1}}{\varrho_0}
\frac{\sum_{n=0}^{2}z_1^n \zeta^{2-n}}{z_1^{3}\zeta^3(z_1+\zeta)}
-(3+2l)
\frac{\sum_{n=0}^{2l+4}z_1^n \zeta^{2l+4-n}}{z_1^{5+2l}\zeta^{5+2l}
(z_1+\zeta)}
\Bigg)
\frac{\partial \G_g({J\backslash\{\zeta \}})}{\partial \varrho_l}
\\
&-\sum_{\xi\in J\backslash \{\zeta \}}\frac{1}{\varrho_0 \xi}
\frac{\sum_{n=0}^{2}z_1^n \zeta^{2-n}}{z_1^{3}\zeta^3(z_1+\zeta)}
\frac{\partial \G_g({J\backslash\{\zeta \}})}{\partial \xi}.
\end{align*}
The second term reads
\begin{align*}
&\frac{1}{(2\lambda)^3} 
\G_0(z_1,\zeta )\G_{g}(z_1,{J\backslash \{\zeta \}})
\\
&=-\frac{4\lambda^2}{z_1\zeta} \frac{\partial}{\partial \zeta} 
\frac{1}{(z_1+\zeta)}\Bigg[
\sum_{l=0}^{3g-4+|J|} 
\left(\frac{-(3+2l) \varrho_{l+1}}{\varrho_0 z_1^3}
+\frac{(3+2l)}{z_1^{5+2l}} \right)
\frac{\partial \G_g({J\backslash\{\zeta \}})}{\partial \varrho_l}
\\
& +\sum_{\xi\in J\backslash \{\zeta \}}\frac{1}{\varrho_0\xi z_1^3}
\frac{\partial \G_g({J\backslash\{\zeta \}})}{\partial \xi}
\Bigg].
\end{align*}
The denominator $(z_1+\zeta)$ cancels in the combination 
of interest:
\begin{align*}
&\frac{8\lambda^3}{\zeta} \frac{\partial}{\partial \zeta}
\frac{\G_g(z_1,{J\backslash\{\zeta\}})
-\G_g(\zeta,{J\backslash\{\zeta\}})}{z_1^2-\zeta^2}
+2\lambda \G_0(z_1,\zeta )\G_{g}(z_1,{J\backslash \{\zeta \}})
\\
&= 
\frac{(2\lambda)^6}{\zeta} \frac{\partial}{\partial \zeta}
\Bigg[
\sum_{l= 0}^{3g-4+|J|} 
\Bigg(
\frac{(3+2l) \varrho_{l+1}}{\varrho_0}
\frac{z_1^2+\zeta^2}{z_1^4\zeta^3}
-(3+2l)
\frac{\sum_{n=0}^{l+2}z_1^{2n} \zeta^{2l+4-2n}}{z_1^{6+2l}\zeta^{5+2l}}
\Bigg)
\frac{\partial \G_g({J\backslash\{\zeta \}})}{\partial \varrho_l}
\\
&-\sum_{\xi\in J\backslash \{\zeta \}}\frac{1}{\varrho_0 \xi}
\frac{z_1^2+\zeta^2}{z_1^4\zeta^3}
\frac{\partial \G_g({J\backslash\{\zeta \}})}{\partial \xi}
\Bigg].
\end{align*}
The remaining $\zeta$-derivative confirms the assertion.
\end{proof}

\begin{Alemma}\label{lemma7}
Let $J=\{z_2,...,z_b\}$.
Then under \sref{Assumption}{conj1} one has 
\begin{align*}
-(2\lambda)^{3b-3}\hat{\mathrm{A}}^{\dag g}_{z_1,\dots,z_B}..
\hat{\mathrm{A}}^{\dag g}_{z_1,z_2}\hat{K}_{z_1} \G_g(z_1)=\!\lambda\!\sum_{h=1}^{g-1} \sum_{I\subset J}
\G_h(z_1,I)\G_{g-h}(z_1,{J\backslash I})
+\lambda \G_{g-1}(z_1,z_1,J).
\end{align*}
\end{Alemma}
\begin{proof}
Equation \eqref{eq:CubicSchwingerComplex} can be rewritten for $b=1$ as 
\[
-\hat{K}_{z_1} \G_g(z_1)
= \lambda \sum_{h=1}^{g-1} 
\G_h(z_1)\G_{g-h}(z_1)
+\lambda \G_{g-1}(z_1,z_1).
\]
Operating with $-(2\lambda)^{3B-3}\hat{\mathrm{A}}^{\dag g}_{z_1,\dots,z_B}...
\hat{\mathrm{A}}^{\dag g}_{z_1,z_2}$ and taking the Leibniz rule into account, the assertion follows. 
\end{proof}

\begin{Alemma}\label{lemma8}
Let $J=\{z_2,...,z_b\}$. Then under \sref{Assumption}{conj1} one has 
\begin{align*}
&(2\lambda)^3[\hat{K}_ {z_1},\hat{\mathrm{A}}^{\dag g}_{z_1,\dots,z_B}] 
\G_g(z_1,{J\backslash \{z_b\}})
\\
&=(2\lambda)^6\Bigg[\sum_{l=0}^{3g-4+|J|}\frac{3+2l}{z_1^2z_b^3}
\Bigg(\frac{\varrho_{l+1}}{\varrho_0 z_1^2}
+\frac{3 \varrho_{l+1}}{\varrho_0 z_b^2}
- \frac{1}{z_1^{4+2l}}
-\frac{(5+2l)}{z_b^{4+2l}}\Bigg)\frac{\partial}{\partial \varrho_{l}}
\\
&\qquad
-\sum_{l=0}^{3g-4+|J|}\sum_{k=0}^l\frac{(3+2l)(3+2k)}{z_1^{4+2l-2k}
z_b^{5+2k}}\frac{\partial}{\partial \varrho_{l}}
-\sum_{\xi\in J\backslash \{z_b \}}
\frac{1}{\varrho_0 z_1^2\xi z_b^3}
\Big(\frac{1}{z_1^2}+\frac{3}{z_b^2}\Big)
\frac{\partial}{\partial \xi}
	\Bigg]\G_g({J\backslash \{z_b\}}).
	\end{align*}
\end{Alemma}
\begin{proof}
The first term of the lhs, 
$\hat{K}_{z_1}\hat{\mathrm{A}}^{\dag g}_{z_1,\dots,z_b} 
\G_g(z_1,{J\backslash \{z_b\}})$, is given by \sref{Lemma}{lemma5} 
and $\G_g(J)=(2\lambda)^3\hat{\mathrm{A}}^{\dag g}_{z_2,\dots,z_b}
\G_g({J\backslash \{z_b\}})$ to
\begin{align*}
&\hat{K}_{z_1} (\G_g(z_1,J))
\\
&=\frac{(2\lambda)^6}{z_1^2}\Bigg[
\sum_{l=0}^{3g-3+|J|} (3+2l)\sum_{k=0}^{l}\frac{\varrho_k}{z_1^{2+2l-2k}}
\\
&\quad\times \frac{\partial}{\partial \varrho_{l}}\Bigg(
\sum_{l'=0}^{3g-4+|J|} \!\!
\Big({-}\frac{(3{+}2l') \varrho_{l'+1}}{\varrho_0 z_b^3}
+\frac{3{+}2l'}{z_b^{5+2l'}}\Big)
\frac{\partial}{\partial \varrho_{l'}}
+\!\!\! \sum_{\xi\in J\backslash \{z_b \}} \frac{1}{\varrho_0 \xi z_b^3}
\frac{\partial}{\partial \xi}\Bigg)
\big(\G_g({J\backslash \{z_b\}})\big)
\\
&+\sum_{\zeta\in J}\frac{1}{\zeta}\frac{\partial}{\partial \zeta}    
\Bigg(
\sum_{l'=0}^{3g-4+|J|} \!\!
\Big({-}\frac{(3{+}2l') \varrho_{l'+1}}{\varrho_0 z_b^3}
+\frac{3{+}2l'}{z_b^{5+2l'}}\Big)
\frac{\partial}{\partial \varrho_{l'}}
+ \!\!\! \sum_{\xi\in J\backslash \{z_b \}} \frac{1}{\varrho_0 \xi z_b^3}
\frac{\partial}{\partial \xi}\Bigg)
\big(\G_g({J\backslash \{z_b\}})\big)
\Bigg]
\\
&=\frac{(2\lambda)^6}{z_1^2}\Bigg[
\sum_{l=0}^{3g-4+|J|} \sum_{k=0}^{l+1}
\frac{-(5+2l)(3+2l)\varrho_k}{\varrho_0 z_1^{4+2l-2k}z_b^3}
\frac{\partial}{\partial \varrho_{l}}
+\sum_{l'= 0}^{3g-4+|J|}
\frac{3(3+2l')\varrho_0\varrho_{l'+1}}{\varrho_0^2 
z_1^2 z_b^3}\frac{\partial}{\partial \varrho_{l'}}
\\
&-\!\!\! \sum\limits_{\xi\in J\backslash \{z_b \}}
\frac{3\varrho_0}{\varrho_0^2z_1^2 \xi z_b^3}\frac{\partial}{\partial \xi}
+\sum_{l,l'=0}^{3g-4+|J|}
\sum\limits_{k=0}^{l}\frac{(3{+}2l) \varrho_k}{z_1^{2+2l-2k}}
\Big({-}\frac{(3{+}2l') \varrho_{l'+1}}{\varrho_0 z_b^3}
+\frac{3{+}2l'}{z_b^{5+2l'}}\Big)
\frac{\partial^2}{\partial\varrho_l\partial\varrho_{l'}}
\\
&+\sum_{l=0}^{3g-4+|J|}\sum_{k=0}^l\sum_{\xi\in J\backslash \{z_b \}}
\frac{(3+2l) \varrho_k}{\varrho_0 z_1^{2+2l-2k} \xi z_b^3 }
\frac{\partial^2}{\partial \varrho_l\partial \xi}
\\
&+\sum_{\zeta \in J\backslash \{z_b \}}\Bigg(
\sum_{l'=0}^{3g-4+|J|}\Big(-\frac{(3+2l') \varrho_{l'+1}}{\varrho_0 
\zeta z_b^3}+\frac{3+2l'}{\zeta z_b^{5+2l'}}\Big)
\frac{\partial^2}{\partial \zeta\partial \varrho_{l'}}
+
\sum_{\xi \in J\backslash \{z_b \}} \frac{1}{\varrho_0 \zeta z_b^3}
\frac{\partial}{\partial \zeta} \frac{1}{\xi} \frac{\partial}{\partial \xi}
\Bigg)
\\
&-\!\!\!\sum_{l'=0}^{3g-4+|J|}
\Big(-\frac{3(3+2l') \varrho_{l'+1}}{\varrho_0 z_b^5}
+\frac{(3+2l')(5+2l')}{z_b^{7+2l'}}\Big)
\frac{\partial}{\partial \varrho_{l'}}
-\!\!\! \sum_{\xi\in J\backslash \{z_b \}} \frac{3}{\varrho_0 \xi z_b^5}
\frac{\partial}{\partial \xi}\Bigg]\G_g({J\backslash \{z_b\}}).
\end{align*}
We have used that $\G_g({J\backslash \{z_b\}})$ can 
only depend on $\varrho_l$ for $l\leq 3g-4+|J|$.
For the second term of the lhs, 
$\hat{\mathrm{A}}^{\dag g}_{z_1,\dots,z_b}\hat{K}_{z_1}\G_g(z_1,{J\backslash z_b})$, 
\sref{Lemma}{lemma5} can also be used with $b-1$ instead of $b$:
\begin{align*}
&(2\lambda)^3\hat{\mathrm{A}}^{\dag g}_{z_1,\dots,z_b}\hat{K}_{z_1}\G_g(z_1,{J\backslash \{z_b\}})
\\
&=(2\lambda)^6 
\Bigg(\sum_{l'=0}^{3g-3+|J|}\!\!\! \Big(
{-}\frac{(3+2l') \varrho_{l'+1}}{\varrho_0 z_b^3}+\frac{3+2l'}{z_b^{5+2l'}}
\Big)\frac{\partial}{\partial \varrho_{l'}}
+\sum_{\xi\in J\backslash \{z_b \}} \frac{1}{\varrho_0 \xi z_b^3}
\frac{\partial}{\partial \xi}
+\frac{1}{\varrho_0 z_1 z_b^3}\frac{\partial}{\partial z_1}\Bigg)
\\
&\times \frac{1}{z_1^2}\Bigg[
\sum_{l= 0}^{3g-4-|J|} (3+2l)\sum_{k=0}^{l}\frac{\varrho_k}{z_1^{2+2l-2k}}
\frac{\partial}{\partial \varrho_l}
+\sum_{\zeta\in J\backslash \{z_b \}}\frac {1  }{\zeta} 
\frac{\partial}{\partial \zeta}\Bigg]
\G_g({J\backslash \{z_b\}})
\\
&=(2\lambda)^6\Bigg[
\sum_{l=0}^{3g-4+|J|}\sum_{k=0}^l\frac{(3+2l)}{z_1^{4+2l-2k}}
\Big( {-}\frac{(3+2k) \varrho_{k+1}}{\varrho_0 z_b^3}
+\frac{3+2k}{z_b^{5+2k}}\Big)\frac{\partial}{\partial \varrho_l}
\\
&+\sum_{l,l'= 0}^{3g-4-|J|}
\sum_{k=0}^l\frac{(3+2l)\varrho_k}{z_1^{4+2l-2k}}
\Big({-}\frac{(3+2l') \varrho_{l'+1}}{\varrho_0 z_b^3}
+\frac{3+2l'}{z_b^{5+2l'}}\Big)
\frac{\partial^2}{\partial\varrho_l\partial\varrho_{l'}}
\\
&+\sum_{\xi\in J\backslash \{z_b \}}
\bigg(\sum_{l=0}^{3g-4+|J|}\sum_{k=0}^{l} 
\frac{(3+2l)\varrho_k}{
\varrho_0 z_1^{4+2l-2k} \xi z_b^3}
\frac{\partial^2}{\partial \varrho_l\partial \xi}
+\sum_{\zeta\in J\backslash \{z_b \}}
\frac{1}{\varrho_0 z_1^2 \xi z_b^3}
\frac{\partial}{\partial \xi} \frac{1}{\zeta} 
\frac{\partial}{\partial \zeta} \bigg)
\\
&+\sum_{l'=0}^{3g-4+|J|}\sum_{\zeta\in J\backslash \{z_b \}}
\Big( {-}\frac{(3+2l') \varrho_{l'+1}}{\varrho_0 z_1^2 \zeta z_b^3}
+\frac{3+2l'}{z_1^2 \zeta z_b^{5+2l'}}\Big)
\frac{\partial^2}{\partial \varrho_{l'}\partial \zeta}
\\
&-\frac{2}{\varrho_0z_1^4 z_b^3}\bigg(\sum_{l=0}^{3g-4+|J|}\sum_{k=0}^{l} 
\frac{(3+2l)\varrho_k}{z_1^{2+2l-2k}}
\frac{\partial}{\partial \varrho_l}+ \sum_{\zeta\in J\backslash \{z_b \}}
\frac{1}{\zeta} \frac{\partial}{\partial \zeta}\bigg)
\\
&-\frac{1}{\varrho_0z_1^2 z_b^3}\sum_{l=0}^{3g-4+|J|}\sum_{k=0}^{l} 
\frac{(3+2l)(2+2l-2k)\varrho_k}{z_1^{4+2l-2k}}
\frac{\partial}{\partial \varrho_l}
\Bigg]\G_g({J\backslash \{z_b\}}).
\end{align*}
Subtracting the second from the first expression proves the Lemma.
\end{proof}

\begin{Alemma}\label{lemma9}
Let $J=\{z_2,...,z_b\}$.The linear integral equation \eqref{eq:CubicSchwingerComplex} is
under \sref{Assumption}{conj1} and with \sref{Definition}{DefOp} 
equivalent to the expression
\begin{align*}
0=&(2\lambda)^{3}[\hat{K}_{z_1},\hat{\mathrm{A}}^{\dag g}_{z_1,\dots,z_b}] 
\G_g(z_1,{J\backslash \{z_b\}})
+2\lambda \G_0(z_1,z_b)
\G_g(z_1,{J\backslash\{z_b\}})
\\
&+(2\lambda)^3\frac{1}{z_b}\frac{\partial}{\partial z_b}
\frac{\G_g(z_1,{J\backslash\{z_b\}})
-\G_g(z_b,{J\backslash\{z_b\}})}{z_1^2-z_b^2}.
\end{align*}
\end{Alemma}
\begin{proof}
With \sref{Lemma}{lemma7} we can rewrite the linear 
integral equation \eqref{eq:CubicSchwingerComplex} in the form
\begin{align}\label{genuseqop}\nonumber
0=&(2\lambda)^{3b-3}\hat{K}_{z_1} \hat{\mathrm{A}}^{\dag g}_{z_1,\dots,z_b} \dots 
\hat{\mathrm{A}}^{\dag g}_{z_1,z_2} \G_g(z_1)
-(2\lambda)^{3b-3} \hat{\mathrm{A}}^{\dag g}_{z_1,\dots,z_b} \dots 
\hat{\mathrm{A}}^{\dag g}_{z_1,z_2} \hat{K}_{z_1} \G_g(z_1)
\\
&+2\lambda\G_g(z_1) \G_0(z_1,J)
+2\lambda \!\!\! \sum\limits_{\substack{I\subset J\\1\leq|I|<|J|}} \!\!\! 
\G_0(z_1,{I})\G_g(z_1,{J\backslash I})
\nonumber
\\
&+(2\lambda)^3
\sum\limits_{\zeta\in J} 
\frac{1}{\zeta} 
\frac{\partial}{\partial \zeta}
\frac{\G_g(z_1,{J\backslash\{\zeta\}})
-\G_g(\zeta,{J\backslash\{\zeta\}})}{z_1^2-\zeta^2}.
\end{align}
By using this formula for $
\hat{\mathrm{A}}^{\dag g}_{z_1,\dots,z_{b-1}} \dots \hat{\mathrm{A}}^{\dag g}_{z_1,z_2} \hat{K}_{z_1} \G_g(z_1)$ 
and inserting it back into \eqref{genuseqop} gives
\begin{align*}
0=&
(2\lambda)^{3b-3}[\hat{K}_{z_1}, \hat{\mathrm{A}}^{\dag g}_{z_1,\dots,z_b}] 
\hat{\mathrm{A}}^{\dag g}_{z_1,\dots,z_{b-1}} \dots 
\hat{\mathrm{A}}^{\dag g}_{z_1,z_2} \G_g(z_1)
\\
&+2\lambda\G_g(z_1)\G_0(z_1,J)-(2\lambda)^3\hat{\mathrm{A}}^{\dag g}_{z_1,\dots,z_b}
(\G_g(z_1)\G_0(z_1,{J\backslash z_b}))
\\
&+2\lambda \!\!\! \sum\limits_{\substack{I\subset J\\1\leq|I|<|J|}} \!\!\!
\G_0(z_1,{I})\G_g(z_1,{J\backslash I})
-(2\lambda)^4\hat{\mathrm{A}}^{\dag g}_{z_1,\dots,z_b} \!\!\! 
\sum\limits_{\substack{I\subset J\backslash\{z_b\}\\1\leq|I|<|J|-1}}
\!\!\! \G_0(z_1,{I})\G_g(z_1,{J\backslash \{I,z_b\}})
\\
&+(2\lambda)^3
\sum\limits_{\zeta \in J} 
\frac{1}{\zeta} 
\frac{\partial}{\partial \zeta}
\frac{\G_g(z_1,{J\backslash\{\zeta\}})
-\G_g(\zeta,{J\backslash\{\zeta\}})}{z_1^2-\zeta^2}
\\
&-(2\lambda)^6\hat{\mathrm{A}}^{\dag g}_{z_1,\dots,z_b}
\sum\limits_{\zeta \in J\setminus\{z_b\}} 
\frac{1}{\zeta} 
\frac{\partial}{\partial \zeta}
\frac{\G_g(z_1,{J\backslash\{\zeta,z_b\}})
-\G_g(\zeta,{J\backslash\{\zeta,z_b\}})}{z_1^2-\zeta^2}.
\end{align*}
The second and third line break down to
$2\lambda \G_0(z_1,z_b)\G_g(z_1,{J\backslash\{z_b\}}).$
Therefore, the assertion follows if we can show that, 
in the fourth line, the part of the sum which excludes 
$\zeta=z_b$ cancels with the fifth line.
This is true because of  
\[
\Big[\hat{\mathrm{A}}^{\dag g}_{z_1,\dots,z_b},\frac{1}{\zeta}\frac{\partial}{\partial \zeta}
\Big]=0 \quad\text{and} \quad 
\hat{\mathrm{A}}^{\dag g}_{z_1,\dots,z_b} \frac{1}{z_1^2-\zeta^2}=0.
\]
Consequently, the linear integral equation can be written 
by operators of the form given in this Lemma.
\end{proof}

\chapter{Perturbative Computations on the Moyal Space}\label{App:Pert}
The following calculations confirm the results derived in the thesis. We will show the first non-trivial examples in 
different dimensions which are already complicated enough. For correlation functions with more boundary components 
or higher genus, the number of Feynman graphs is to large for an appropriate example. 
For a more convenient representation, the ribbon graphs will be drawn with lines instead of ribbons.

\section{Cubic Interaction}\label{App:PertCubic}
We will focus on the planar 1-point function
on the Moyal space in different dimensions. The lowest order Feyman graphs with only 3-valent vertices 
$\Gamma_1,\Gamma_2,\Gamma_3,\Gamma_4,\Gamma_5\in\mathfrak{G}_x^{(0,1)}$ are:

\vspace{3ex}
\def\svgwidth{0.9\textwidth}
\begingroup%
  \makeatletter%
  \providecommand\color[2][]{%
    \errmessage{(Inkscape) Color is used for the text in Inkscape, but the package 'color.sty' is not loaded}%
    \renewcommand\color[2][]{}%
  }%
  \providecommand\transparent[1]{%
    \errmessage{(Inkscape) Transparency is used (non-zero) for the text in Inkscape, but the package 'transparent.sty' is not loaded}%
    \renewcommand\transparent[1]{}%
  }%
  \providecommand\rotatebox[2]{#2}%
  \ifx\svgwidth\undefined%
    \setlength{\unitlength}{472.19380798bp}%
    \ifx\svgscale\undefined%
      \relax%
    \else%
      \setlength{\unitlength}{\unitlength * \real{\svgscale}}%
    \fi%
  \else%
    \setlength{\unitlength}{\svgwidth}%
  \fi%
  \global\let\svgwidth\undefined%
  \global\let\svgscale\undefined%
  \makeatother%
  \begin{picture}(1,0.57267756)%
    \put(0,0){\includegraphics[width=\unitlength,page=1]{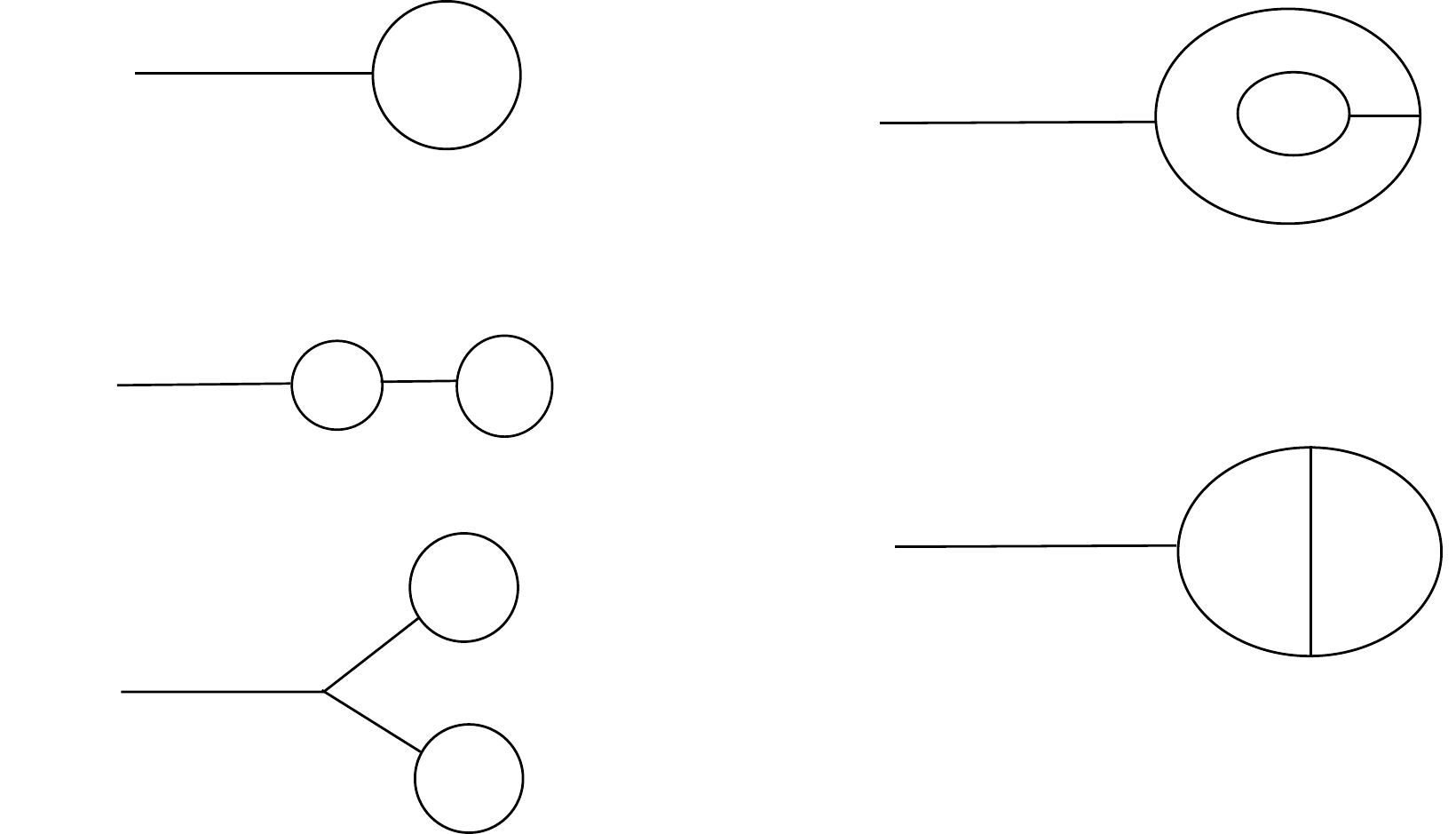}}%
    \put(-0.01410574,0.52921777){\color[rgb]{0,0,0}\makebox(0,0)[lb]{\smash{}}}%
    \put(0.01156556,0.5155264){\color[rgb]{0,0,0}\makebox(0,0)[lb]{\smash{$\Gamma_1$=}}}%
    \put(-0.0021095,0.30093152){\color[rgb]{0,0,0}\makebox(0,0)[lb]{\smash{$\Gamma_2$=}}}%
    \put(0.00131334,0.09213826){\color[rgb]{0,0,0}\makebox(0,0)[lb]{\smash{$\Gamma_3$=}}}%
    \put(0.52158513,0.48148632){\color[rgb]{0,0,0}\makebox(0,0)[lb]{\smash{$\Gamma_4$=}}}%
    \put(0.53442072,0.19225626){\color[rgb]{0,0,0}\makebox(0,0)[lb]{\smash{$\Gamma_5$=}}}%
    \put(0.15705259,0.54224176){\color[rgb]{0,0,0}\makebox(0,0)[lb]{\smash{$x$}}}%
    \put(0.13204959,0.32889643){\color[rgb]{0,0,0}\makebox(0,0)[lb]{\smash{$x$}}}%
    \put(0.15045582,0.11668035){\color[rgb]{0,0,0}\makebox(0,0)[lb]{\smash{$x$}}}%
    \put(0.7066674,0.2236441){\color[rgb]{0,0,0}\makebox(0,0)[lb]{\smash{$x$}}}%
    \put(0.68270749,0.52057553){\color[rgb]{0,0,0}\makebox(0,0)[lb]{\smash{$x$}}}%
    \put(0.28309089,0.51800838){\color[rgb]{0,0,0}\makebox(0,0)[lb]{\smash{$y$}}}%
    \put(0.21377834,0.30400069){\color[rgb]{0,0,0}\makebox(0,0)[lb]{\smash{$y_1$}}}%
    \put(0.3028754,0.16658502){\color[rgb]{0,0,0}\makebox(0,0)[lb]{\smash{$y_1$}}}%
    \put(0.84026138,0.19482344){\color[rgb]{0,0,0}\makebox(0,0)[lb]{\smash{$y_1$}}}%
    \put(0.81116731,0.47292923){\color[rgb]{0,0,0}\makebox(0,0)[lb]{\smash{$y_1$}}}%
    \put(0.33196955,0.30349864){\color[rgb]{0,0,0}\makebox(0,0)[lb]{\smash{$y_2$}}}%
    \put(0.30373112,0.03566131){\color[rgb]{0,0,0}\makebox(0,0)[lb]{\smash{$y_2$}}}%
    \put(0.92583239,0.19396775){\color[rgb]{0,0,0}\makebox(0,0)[lb]{\smash{$y_2$}}}%
    \put(0.86935544,0.49089918){\color[rgb]{0,0,0}\makebox(0,0)[lb]{\smash{$y_2$}}}%
  \end{picture}%
\endgroup%

\noindent
where the $\Gamma_i$'s should be understood as ribbon graphs.
Applying the Feynman rules $\mathfrak{h}$ of \sref{Sec.}{Sec.Feynman}, 
performing the $\mN',V$-limit of \sref{Sec.}{Sec:LargeLimit}
with the ratio $\frac{\mN'}{V^{2/D}}=\Lambda^2$ and taking
the eigenvalue distribution induced by the $D$-dimensional Moyal space of \sref{Sec.}{Sec.Moyal} leads for the first diagram 
$\Gamma_1$ to
\begin{align*}
 \lim_{\mN',V\to\infty}\mathfrak{h}(\Gamma_1)=\frac{-\lambda}{1+2x}
 \int_0^{\Lambda^2}\frac{dy\,r(y)}{1+x+y},
 \qquad \text{where}\quad r(x)=
 \frac{x^{D/2-1}}{(\frac{D}{2}-1)!}.
\end{align*}
To renormalise this kind of graphs we apply the procedure described in \sref{Sec.}{Sec.Zimmer} 
with \sref{Theorem}{Thm:Zimmermann}.
The superficial degree of divergence of the unique subgraph $\gamma\in\Gamma_1$ is by formula \eqref{eq:superficial}
\begin{align*}
 \omega(\gamma)=\frac{D}{2}(2-0-1-1)+\bigg(\frac{D}{2}-1\bigg) \frac{3\cdot 1-1}{2}=\frac{D}{2}-1.
\end{align*}
The graph $\Gamma_1$ has two forests $\U_{\Gamma_1}=\emptyset$ 
and $\U_{\Gamma_1}=\gamma$, where $\gamma$ consists of the face labelled by $y$, the attached edge and the vertex.
Then \sref{Theorem}{Thm:Zimmermann} leads to the renormalised diagram by 
\begin{align*}
 &-\lambda
 \int_0^{\Lambda^2}dy\,r(y) 
 \bigg\{\underbrace{\frac{1}{(1+x+y)(1+2x)}}_{U_{\Gamma_1}=\emptyset}
 -\underbrace{\underbrace{\frac{1}{1+2x}}_{I_{\Gamma_1\backslash\U_{\Gamma_1}}=I_{
\Gamma_1\backslash\gamma }}\Big(T^{\frac{D}{2}-1}_x\Big)\bigg(\frac{1}{1+x+y}\bigg)}_{U_{\Gamma_1}=\gamma}\bigg\}\\
=&\frac{-\lambda}{1+2x}\int_0^{\Lambda^2}dy\,r(y)
\bigg\{\frac{1}{1+x+y}-\Big(T^{\frac{D}{2}-1}_x\Big)\bigg(\frac{1}{1+x+y}\bigg)\bigg\}
\end{align*}
where $\Lambda^2\to\infty$ is now finite. 
This formula matches with our renormalisation conditions \eqref{eq:cond1}, \eqref{eq:cond2} and \eqref{eq:cond3}.
Taking $D=2,4,6$ gives specially
\begin{align*}
 &D=2:&& \frac{-\lambda}{1+2x}
 \int_0^{\infty}dy\,\bigg(\frac{1}{1+x+y}-\frac{1}{1+y}\bigg)= \lambda\frac{\log(1+x)}{1+2x}\\
 &D=4:&& \frac{-\lambda}{1+2x}
 \int_0^{\infty}dy\,y\bigg(\frac{1}{1+x+y}-\frac{1}{1+y}+\frac{x}{(1+y)^2}\bigg)= 
 \lambda\frac{x-(1+x)\log(1+x)}{1+2x}\\
 &D=6:&& \frac{-\lambda}{1+2x}
 \int_0^{\infty}dy\,\frac{y^2}{2}\bigg(\frac{1}{1+x+y}-\frac{1}{1+y}+\frac{x}{(1+y)^2}-\frac{x^2}{(1+y)^3}\bigg)\\
 & &&= \lambda \frac{2(1+x)^2\log(1+x)-x(2+3x)}{4(1+2x)}
\end{align*}
which coincides for the first order with \sref{Example}{Ex:D2}, \sref{Example}{Ex:D4} and \sref{Example}{Ex:D6}.

For the over-subtracted form in the $D=2$ with the renormalisation of $D=4$ and $D=6$ respectively, we have
\begin{align*}
 &\frac{-\lambda}{1+2x}
 \int_0^{\infty}dy\bigg(\frac{1}{1+x+y}-\frac{1}{1+y}+\frac{x}{(1+y)^2}\bigg)= 
 \lambda\frac{\log(1+x)-x}{1+2x}\\
 &\frac{-\lambda}{1+2x}
 \int_0^{\infty}dy\bigg(\frac{1}{1+x+y}-\frac{1}{1+y}+\frac{x}{(1+y)^2}-\frac{x^2}{(1+y)^3}\bigg)\\
 &= \lambda \frac{2\log(1+x)-x(2-x)}{2(1+2x)}
\end{align*}
which coincides for the first order with \sref{Example}{Ex:D24} and \sref{Example}{Ex:D26}.

\subsubsection*{Iterated Integrals}
The integrals appearing the first time for $\Gamma_2,\Gamma_3,\Gamma_4,\Gamma_5$ are of new complexity.
The integration theory of the appearing integrals is completely understood
in form of iterated integrals \cite{Brown:2009qja}. They form a
shuffle algebra, which is symbolically implemented in the Maple
package \textsc{HyperInt} \cite{Panzer:2014caa}.

The hyperlogarithms Hlog appearing in \textsc{HyperInt}
are defined by the iterated integrals via
\begin{align*}
 \mathrm{Hlog}(a,[k_1,...,k_n]):=&\int_0^a \frac{dx_1}{x_1-k_1}\int_0^{x_1}\frac{dx_2}{x_2-k_2}...
\int_0^{x_{n-1}}\frac{dx_n}{x_n-k_n},
\end{align*}
where the $k_i$ are called \textit{letters}. An alternative notation is
$\mathrm{Hlog}(a,[k_1,...,k_n])=L_{k_1,...,k_n}(a)$. Important special cases are
$\mathrm{Hlog}(a,[\underbrace{-k,...,-k}_n])=
\frac{\log(1+\tfrac{a}{k})^n}{n!}$ for $k\in \mathbb{N}^\times $,
$\mathrm{Hlog}(a,[\underbrace{0,...,0}_n]):=\frac{\log(a)^n}{n!}$ and
$\mathrm{Hlog}(a,[\underbrace{0,...,0}_n,-1])=-\mathrm{Li}_{1+n}(-a)$.

\subsubsection*{Renormalised $\Gamma_2$}
The graph $\Gamma_2$ has four forests which are $\U_{\Gamma_2}=\emptyset$, $\U_{\Gamma_2}=\{\gamma_1\}$,
 $\U_{\Gamma_2}=\{\gamma_2\}$ and $ \U_{\Gamma_2}=\{\gamma_1,\gamma_2\}$, where $\gamma_i$ is the ribbon subgraph
 with the face labelled by $y_i$ and its attached edges and vertices. Note there is no subgraph which includes both faces since
 this is not a 1PI graph. We have $g_{\gamma_1}=g_{\gamma_2}=0$, 
 $b_{\gamma_1}=b_{\gamma_2}=1$, $v_{\gamma_1}=2v_{\gamma_2}=2$, $N_{\gamma_1}=2N_{\gamma_2}=2$, and $k_3=2$ vertices for 
 $\gamma_1$ and $k_3=1$ vertex for $\gamma_2$. The superficial degree of divergence is therefore
 \begin{align*}
  \omega(\gamma_1)=\frac{D}{2}-2,\qquad 
  \omega(\gamma_2)=\frac{D}{2}-1.
 \end{align*}
The set of variables for the Taylor-subtraction is $f(\gamma_1)=f(\gamma_2)=\{x\}$ and the rational functions
$r_{\gamma_1}=\frac{1}{(1+x+y_1)^2}$ and $r_{\gamma_2}=\frac{1}{1+x+y_2}$.
Inserting in \sref{Theorem}{Thm:Zimmermann}, writing out the sum over the forests and simplifying yields
\begin{align*}
 \frac{(-\lambda)^3}{(1+2x)^2}\int_0^\infty dy_1\, r(y_1) 
 &\bigg\{\frac{1}{(1+x+y_1)^2}-\Big(T^{\frac{D}{2}-2}_x\Big)\Big(\frac{1}{(1+x+y_1)^2}\Big)\bigg\}\\
 \times\int_0^\infty dy_2\, r(y_2)&\bigg\{\frac{1}{1+x+y_2}-\Big(T^{\frac{D}{2}-1}_x\Big)\Big(\frac{1}{1+x+y_2}\Big)\bigg\}.
\end{align*}
This formula matches with our renormalisation conditions \eqref{eq:cond1}, \eqref{eq:cond2} and \eqref{eq:cond3}. 
Plugging in the different dimensions gives
\begin{align*}
 &D=2:&& \lambda^3\frac{\log(1+x)}{(1+2x)^2(1+x)}\\
 &D=4:&& \lambda^3\frac{\log(1+x)((1+x)\log(1+x)-x)}{(1+2x)^2} \\
 &D=6:&& \lambda^3 \frac{2(1+x)^3\log(1+x)^2-x(1+x)(5x+4)\log(1+x)+x^2(3x+2)}{4(1+2x)^2}
\end{align*}
and for the over-subtracted from in $D=2$ with
\begin{align*}
 &\text{renormalisation of }D=4:&& \lambda^3\frac{x(x-\log(1+x))}{(1+2x)^2(1+x)} \\
 &\text{renormalisation of }D=6:&& \lambda^3\frac{x^2(2\log(1+x)+x(x-2))}{2(1+x)(1+2x)^2}.
\end{align*}

\subsubsection*{Renormalised $\Gamma_3$}
The graph $\Gamma_3$ is renormalised in the same way as $\Gamma_1$ for each subgraph, respectively.
The general result is
\begin{align*}
 \frac{(-\lambda)^3}{(1+2x)^3}\int_0^\infty dy_1\, r(y_1) 
 &\bigg\{\frac{1}{1+x+y_1}-\Big(T^{\frac{D}{2}-1}_x\Big)\Big(\frac{1}{1+x+y_1}\Big)\bigg\}\\
 \times\int_0^\infty dy_2\, r(y_2)&\bigg\{\frac{1}{1+x+y_2}-\Big(T^{\frac{D}{2}-1}_x\Big)\Big(\frac{1}{1+x+y_2}\Big)\bigg\}.
\end{align*}
This formula matches with our renormalisation conditions \eqref{eq:cond1}, \eqref{eq:cond2} and \eqref{eq:cond3}.
Plugging the different dimensions in gives
\begin{align*}
 &D=2:&& -\lambda^3\frac{\log(1+x)^2}{(1+2x)^3}\\
 &D=4:&& -\lambda^3\frac{(x-(1+x)\log(1+x))^2}{(1+2x)^3}\\
 &D=6:&& -\lambda^3 \frac{(2(1+x)^2\log(1+x)-x(2+3x))^2}{16(1+2x)^3}
\end{align*}
and for the over-subtracted from in $D=2$ with
\begin{align*}
 &\text{renormalisation of }D=4:&& -\lambda^3\frac{(\log(1+x)-x)^2}{(1+2x)^3}\\
 &\text{renormalisation of }D=6:&& -\lambda^3 \frac{(2\log(1+x)-x(2-x))^2}{4(1+2x)^3}.
\end{align*}

\subsubsection*{Renormalised $\Gamma_4$}
The graph $\Gamma_4$ has six forests which are $\U_{\Gamma_4}=\emptyset$, $\U_{\Gamma_4}=\{\gamma_1\}$,
 $\U_{\Gamma_4}=\{\gamma_2\}$, $ \U_{\Gamma_4}=\{\gamma_{12}\}$,
 $ \U_{\Gamma_4}=\{\gamma_{12},\gamma_1\}$ and $ \U_{\Gamma_4}=\{\gamma_{12},\gamma_2\}$, 
 where $\gamma_{12}$ is the ribbon subgraph
 with the face labelled by $y_1$ and $y_2$ together with its attached edges and vertices. 
 The superficial degrees of divergence is therefore
 \begin{align*}
  \omega(\gamma_1)=\frac{D}{2}-4,\qquad 
  \omega(\gamma_2)=\frac{D}{2}-1,\quad 
  \omega(\gamma_{12})=D-4.
 \end{align*}
The set of variables for the Taylor-subtraction is $f(\gamma_1)=f(\gamma_{12})=\{x\}$ and 
$ f(\gamma_2)=\{y_1\}$. The rational functions are
$r_{\gamma_1}=r_{\gamma_{12}}=\frac{1}{(1+x+y_1)^2(1+2y_1)(1+y_1+y_2)}$ and $r_{\gamma_2}=\frac{1}{1+y_1+y_2}$.
Note that $\gamma_{12}$ has the subgraph $\gamma_2=o(\gamma_{12})$ or $\gamma_1=o(\gamma_{12})$, but then
$\gamma_{12}\backslash\gamma_1$ is only a face which does not count.
We will neglect the forest $\U_{\Gamma_4}=\{\gamma_1\}$ since the degree of divergence matters for dimensions $D\geq 8$.
Inserting in \sref{Theorem}{Thm:Zimmermann}, writing out the sum over the forests and simplifying yields
\begin{align*}
 \frac{(-\lambda)^3}{1+2x}\int_0^\infty dy_1\, r(y_1) 
 &\bigg\{\frac{1}{(1+x+y_1)^2}-\Big(T^{D-4}_x\Big)\Big(\frac{1}{(1+x+y_1)^2}\Big)\bigg\}\\
 \times\frac{1}{1+2y_1}\int_0^\infty dy_2\, r(y_2)&\bigg\{\frac{1}{1+y_1+y_2}
 -\Big(T^{\frac{D}{2}-1}_{y_1}\Big)\Big(\frac{1}{1+y_1+y_2}\Big)\bigg\}.
\end{align*}
Since the normalisation conditions \eqref{eq:cond1}, \eqref{eq:cond2} and \eqref{eq:cond3}
imply a different subtraction, we see here the first time an adaption of the forest formula in the first line 
by changing the degree of divergence
\begin{align}\nonumber
 \frac{(-\lambda)^3}{1+2x}\int_0^\infty dy_1\, r(y_1) 
 &\bigg\{\frac{1}{(1+x+y_1)^2}-\Big(T^{\frac{D}{2}-1}_x\Big)\Big(\frac{1}{(1+x+y_1)^2}\Big)\bigg\}\\\label{eq:app1}
 \times\frac{1}{1+2y_1}\int_0^\infty dy_2\, r(y_2)&\bigg\{\frac{1}{1+y_1+y_2}
 -\Big(T^{\frac{D}{2}-1}_{y_1}\Big)\Big(\frac{1}{1+y_1+y_2}\Big)\bigg\}
\end{align}
which corresponds to an appropriate subtraction of the $\gamma_{12}$ graph. Note that the $D=6$ case is unchanged.

Plugging in the different dimensions \eqref{eq:app1} and using \textsc{HyperInt} gives
\begin{align*}
 &D=2:&& \lambda^3\bigg(\frac{2\mathrm{Hlog}(x,[0,-1])+4x(1+x)(\log(2)^2-\zeta_2+1)+1}{(1+2x)^3}-
 \frac{2\log(1+x)}{x(1+2x)^2}\bigg)\\
 &D=4:&& \lambda^3\bigg(\frac{(1+2x+2x^2)\mathrm{Hlog}(x,[0,-1])-x^2(4x+3)(\log(2)^2-2\log2-\zeta_2+3)}{(1+2x)^3}\\
 & && \quad +\frac{(1+x)\log(1+x)-x(x+2)}{(1+2x)^3}\bigg)\\
 &D=6:&& \lambda^3\bigg(\frac{-2x(1+x)(1+3x+3x^2)\mathrm{Hlog}(x,[0,-1])+2x^4-15x^3-16x^2-4x}{4(1+2x)^3}\\
 & && \quad +\frac{(1+x)^2(2+7x+7x^2)
 \log(1+x)+2x^3(3x+2)(\log(2)^2-\log2-\zeta_2)}{8(1+2x)^3}\bigg)\\
\end{align*}
and for the over-subtracted from in $D=2$ with
\begin{align*}
 &\text{renormalisation of }D=4:&& \lambda^3\bigg(\frac{2\mathrm{Hlog}(x,[0,-1])}{(1+2x)^3}-\frac{(1+x)\log(1+x)}{x(1+2x)^3}\\
 & &&\!\!\!\!\!\!\!\!\!\!\!\!\!\!\!\!\!\!\!\! -\frac{8x^2(4x+3)(\log(2)^2-\log2-\zeta_2)+60x^3+44x^2+3x-2}{2(1+2x)^3}\bigg)\\
 &\text{renormalisation of }D=6:&& \lambda^3\bigg(\frac{2\mathrm{Hlog}(x,[0,-1])}{(1+2x)^3}-\frac{(1+x+x^2+x^3)
 \log(1+x)}{x(1+2x)^3}\\
 & &&\!\!\!\!\!\!\!\!\!\!\!\!\!\!\!\!\!\!\!\!\!\!\!\!\!\!\!\!\!\!\!\!\!\!\!\!\!\!\!\!\!\!\!\!\!\!\!\!\!\!\!\!\!\!\!\!
 +\frac{24x^3(3x+2)(4\log(2)^2-5\log2-4\zeta_2)+584x^4+392x^3+8x^2-9x+6}{6(1+2x)^3}\bigg).
\end{align*}

\subsubsection*{Renormalised $\Gamma_5$}
The graph $\Gamma_5$ has six forests which are $\U_{\Gamma_5}=\emptyset$, $\U_{\Gamma_5}=\{\gamma_1\}$,
 $\U_{\Gamma_5}=\{\gamma_2\}$, $ \U_{\Gamma_5}=\{\gamma_{12}\}$, $ \U_{\Gamma_5}=\{\gamma_{12},\gamma_1\}$ and 
 $ \U_{\Gamma_5}=\{\gamma_{12},\gamma_2\}$
  defined as above. 
  For this graph the overlapping divergences need the forest formula in its full generality.
 The superficial degrees of divergence is therefore
 \begin{align*}
  \omega(\gamma_1)=\frac{D}{2}-3,\qquad 
  \omega(\gamma_2)=\frac{D}{2}-2,\quad 
  \omega(\gamma_{12})=D-4.
 \end{align*}
The set of variables for the Taylor-subtraction is $f(\gamma_1)=\{x,y_2\}$, $f(\gamma_{2})=\{x,y_1\}$ and 
$ f(\gamma_{12})=\{x\}$. The rational functions are
$r_{\gamma_1}=\frac{1}{(1+x+y_1)^2(1+y_1+y_2)}$, $r_{\gamma_{2}}=\frac{1}{(1+y_1+y_2)(1+x+y_2)}$ 
and $r_{\gamma_{12}}=\frac{1}{(1+x+y_1)^2(1+y_1+y_2)(1+x+y_2)}$.
Inserting in \sref{Theorem}{Thm:Zimmermann}, writing out the sum over the forests and simplifying yields
\begin{align*}
 &\frac{(-\lambda)^3}{1+2x}\int_0^\infty dy_1\, r(y_1) \int_0^\infty dy_2\, r(y_2)\\
 \times&\bigg[\bigg\{\frac{1}{(1+y_1+y_2)}
 \bigg(1-\Big(T^{D-4}_x\Big)\bigg)\frac{1}{(1+x+y_1)^2(1+x+y_2)}\bigg\}\\
 +&\bigg\{\bigg(1
 -\Big(T^{D-4}_{x}\Big)\bigg)\frac{1}{(1+x+y_1)^2}
 \Big(-T^{\frac{D}{2}-2}_{x,y_1}\Big)\Big(\frac{1}{(1+y_1+y_2)(1+x+y_2)}\Big)\bigg\}\\
 +&\bigg\{\bigg(1
 -\Big(T^{D-4}_{x}\Big)\bigg)\frac{1}{1+x+y_2}
 \Big(-T^{\frac{D}{2}-3}_{x,y_2}\Big)\Big(\frac{1}{(1+y_1+y_2)(1+x+y_1)^2}\Big)\bigg\}\bigg].
\end{align*}
Notice that in the last and the second last line $T^{D-4}_{x}$ acts also on the other Taylor polynomial recursively, 
which depends
on $x$.
Since the normalisation conditions \eqref{eq:cond1}, \eqref{eq:cond2} and \eqref{eq:cond3}
implies a stronger subtraction, we adapt the forest formula
by changing the degree of divergence again for  the $\gamma_{12}$ graph to
\begin{align*}
 &\frac{(-\lambda)^3}{1+2x}\int_0^\infty dy_1\, r(y_1) \int_0^\infty dy_2\, r(y_2)\\
 \times&\bigg[\bigg\{\frac{1}{(1+y_1+y_2)}
 \bigg(1-\Big(T^{\frac{D}{2}-1}_x\Big)\bigg)\frac{1}{(1+x+y_1)^2(1+x+y_2)}\bigg\}\\
 +&\bigg\{\bigg(1
 -\Big(T^{\frac{D}{2}-1}_{x}\Big)\bigg)\frac{1}{(1+x+y_1)^2}
 \Big(-T^{\frac{D}{2}-2}_{x,y_1}\Big)\Big(\frac{1}{(1+y_1+y_2)(1+x+y_2)}\Big)\bigg\}\\
 +&\bigg\{\bigg(1
 -\Big(T^{\frac{D}{2}-1}_{x}\Big)\bigg)\frac{1}{1+x+y_2}
 \Big(-T^{\frac{D}{2}-3}_{x,y_2}\Big)\Big(\frac{1}{(1+y_1+y_2)(1+x+y_1)^2}\Big)\bigg\}\bigg].
\end{align*}
Plugging in the different dimensions and using \textsc{HyperInt} gives
\begin{align*}
 &D=2:&& \lambda^3\bigg(\frac{\log(1+x)^2-2\mathrm{Hlog}(x,[0,-1])+4x(1+x)(\zeta_2-1)-1}{(1+2x)^3}\\
 & &&\qquad +\frac{\log(1+x)}{x(1+x)(1+2x)^2}\bigg)\\
 &D=4:&& \lambda^3\bigg(\frac{-x(1+x)\log(1+x)^2-(1+2x+2x^2)\mathrm{Hlog}(x,[0,-1])}{(1+2x)^3}\\
 & &&\qquad  + \frac{-x^2(4x+3)\zeta_2+8x^3+8x^2+2x}{(1+2x)^3}-\frac{\log(1+x)}{(1+2x)^2}\bigg)\\
 &D=6:&& \lambda^3\bigg(\frac{-(1+3x)(1+x)^3\log(1+x)^2
 +x2(1+x)(1+3x+3x^2)\mathrm{Hlog}(x,[0,-1])}{4(1+2x)^3}\\
 & &&\qquad  + \frac{x^3(3x+2)\zeta_2-4x^4+4x^3+7x^2+2x}{4(1+2x)^3}-\frac{(1+x)(3x+2)\log(1+x)}{4(1+2x)^2}\bigg)\\
\end{align*}
and for the over-subtracted from in $D=2$ with
\begin{align*}
 &\text{renormalisation of }D=4:&& \lambda^3
 \bigg(\frac{\log(1+x)^2-2\mathrm{Hlog}(x,[0,-1])}{(1+2x)^3}+\frac{\log(1+x)}{x(1+2x)^2(1+x)}\\
 & &&  \!\!\!\!\!\!\!\!\!\!\!\!\!\!\!\!\!\!\!\!\!\!\!\!\!\!\!\!\!\!\!\!\!\!\!\!
 - \frac{8x^2\zeta_2(4x+3)(1+x)+52x^4+88x^3+41x^2+x-2}{2(1+2x)^3(1+x)}
 \bigg)\\
 &\text{renormalisation of }D=6:&& \lambda^3
 \bigg(\frac{\log(1+x)^2-2\mathrm{Hlog}(x,[0,-1])}{(1+2x)^3}+\frac{\log(1+x)}{x(1+2x)^2(1+x)}\\
 & &&  \!\!\!\!\!\!\!\!\!\!\!\!\!\!\!\!\!\!\!\!\!\!\!\!\!\!\!\!\!\!\!\!\!\!\!\!
 + \frac{96x^3\zeta_2(3x+2)(1+x)-476x^5-784x^4-319x^3+7x^2+3x-6}{6(1+2x)^3(1+x)}
 \bigg).
\end{align*}
    
\subsubsection*{Sum of $\Gamma_2,\Gamma_3,\Gamma_4,\Gamma_5$}
    
Taking the sum of the results of all diagrams of order $\lambda^3$ breaks down to
\begin{align*}
 &D=2:&& \lambda^3 \frac{(2\log2)^2(1+x)x}{(1+2x)^3}\\
 &D=4:&& -\lambda^3 \frac{(1-\log2)^2(4x+3)x^2}{(1+2x)^3}\\
 &D=6:&& \lambda^3 \frac{x^3(2+3x)(2\log2-1)^2}{16(1+2x)^3}\\
\end{align*}
and for the over-subtracted from in $D=2$ with
\begin{align*}
 &\text{renormalisation of }D=4:&& -\lambda^3 \frac{x^2(3+4x)(2\log2-1)^2}{(1+2x)^3}\\
 &\text{renormalisation of }D=6:&& \lambda^3 \frac{x^3(2+3x)(8\log2-5)^2}{4(1+2x)^3}
\end{align*}
which confirms \sref{Example}{Ex:D2}, \sref{Example}{Ex:D4}, \sref{Example}{Ex:D6}, \sref{Example}{Ex:D24} and 
\sref{Example}{Ex:D26}.
\\
We emphasise that for the final results all hyperlogarithms cancels perfectly.

\section{Quartic Interaction}\label{App:PertQuartic}
For the $D=4$ Moyal space, the 2-point function possesses at order $\lambda$ two graphs and 
at order $\lambda^2$ nine graphs. The exact result for the 2-point function given in \sref{Sec.}{Sec.4dSol} 
has a natural choice for the 
boundary conditions $G^{(0)}(0,0)$ and $\frac{\partial }{\partial a}G^{(0)}(a,0)\vert_{a=0}$ which does not coincide with 
the natural choice of Zimmermann's forest formula. We know that $G^{(0)}(0,0)=1$ also holds for the exact solution. From
\eqref{eq:quartG2}, the first orders for the derivative are $\frac{\partial}{\partial a}G^{(0)}(a,0)\vert_{a=0}
=-1-\lambda+\lambda^2+\mathcal{O}(\lambda^3)$. 
On the other hand, Zimmermann's forest formula demands
$\frac{\partial}{\partial a}G^{(0)}(a,0)\vert_{a=0}=-1$, which cannot coincide with the exact solution 
without further effort. 
Assuming a general $\mu^2$ in \sref{Theorem}{prop:HT} and 
\sref{Proposition}{Prop:Jx-final} and computing recursively order by order $R_4$, $R_4^{-1}$ and $I(w)$,
this is done in \sref{App.}{Sec2}.
The angle function $\tau_b(a)$ has now 
a general boundary condition. This boundary condition can then be fixed 
by inserting the angle function into \eqref{Gab-ansatz}, 
where $Z$ and $\mu^2$ are chosen to satisfy $G(0,0)=1$ and $\frac{\partial }{\partial a}G^{(0)}(a,0)\vert_{a=0}=-1$.

The first order is the same as usual 
\begin{align*}
 \frac{1}{1+a+b}.
\end{align*}
With $\mu^2=1+\lambda \mu_1^2+\lambda\mu^2_2+..$ and $Z= C_\lambda\cdot e^{\mathcal{H}_0[\tau_0]}$ with 
$C_\lambda=1+\lambda k_1+\lambda^2k_2+..<\infty$, it follows for the next order
\begin{align*}
 \lambda\bigg(\frac{k_1}{1+a+b}-\frac{\mu_1^2+1}{(1+a+b)^2}-\frac{(1+a)\log(1+a)+(1+b)\log(1+b)}{(1+a+b)^2}\bigg).
\end{align*}
The boundary conditions are achieved with $k_1=1$ and $\mu^2_1=0$, which leads to 
\begin{align}\label{od1}
 \lambda\bigg(\frac{a+b}{(1+a+b)^2}-\frac{(1+a)\log(1+a)+(1+b)\log(1+b)}{(1+a+b)^2}\bigg).
\end{align}
The second order is also straightforward to compute, which is with $k_1=1$ and $\mu^2_1=0$
\begin{align}\nonumber
 &\frac{\lambda^2}{(1+a+b)^2}\big[\zeta_2(1+a+b+ab)-\!\mu^2_2(1+a+b)+\!(a+b)^2\!-\!(1+a+b)+k_2(1+a+b)^2\\\nonumber
 &+(1+a)(1+b)\log(1+a)\log(1+b)-a(1+b)\log(1+b)^2-b(1+a)\log(1+a)^2\\\label{graph1}
 &-(1+b+2a+2ab+a^2)\mathrm{Li}_2(-a)-(1+a+2b+2ab+b^2)\mathrm{Li}_2(-b)\\\nonumber
 &+((1+a-b)-(1+a)(1+a+b))\log(1+a)\\\nonumber
 &+((1+b-a)-(1+b)(1+a+b))\log(1+b)\big].
\end{align}
The boundary conditions are fulfilled with $\mu^2_2=\zeta_2-2$ and $k_2=-1$. We see that it now coincides with the graph expansion 
below \eqref{graph}.

\subsubsection*{Graph Expansion}
The graphs up to second order are the following:
\vspace{2ex}

\def\svgwidth{0.9\textwidth}
\begingroup%
  \makeatletter%
  \providecommand\color[2][]{%
    \errmessage{(Inkscape) Color is used for the text in Inkscape, but the package 'color.sty' is not loaded}%
    \renewcommand\color[2][]{}%
  }%
  \providecommand\transparent[1]{%
    \errmessage{(Inkscape) Transparency is used (non-zero) for the text in Inkscape, but the package 'transparent.sty' is not loaded}%
    \renewcommand\transparent[1]{}%
  }%
  \providecommand\rotatebox[2]{#2}%
  \ifx\svgwidth\undefined%
    \setlength{\unitlength}{357.14335235bp}%
    \ifx\svgscale\undefined%
      \relax%
    \else%
      \setlength{\unitlength}{\unitlength * \real{\svgscale}}%
    \fi%
  \else%
    \setlength{\unitlength}{\svgwidth}%
  \fi%
  \global\let\svgwidth\undefined%
  \global\let\svgscale\undefined%
  \makeatother%
  \begin{picture}(1,0.6432017)%
    \put(0.0131113,0.58208716){\color[rgb]{0,0,0}\makebox(0,0)[lb]{\smash{$a$}}}%
    \put(0.30991091,0.58688714){\color[rgb]{0,0,0}\makebox(0,0)[lb]{\smash{$a$}}}%
    \put(0.60991046,0.59808714){\color[rgb]{0,0,0}\makebox(0,0)[lb]{\smash{$a$}}}%
    \put(0.02671128,0.48128727){\color[rgb]{0,0,0}\makebox(0,0)[lb]{\smash{$a$}}}%
    \put(0.32751089,0.48288729){\color[rgb]{0,0,0}\makebox(0,0)[lb]{\smash{$a$}}}%
    \put(0.62991041,0.44208733){\color[rgb]{0,0,0}\makebox(0,0)[lb]{\smash{$a$}}}%
    \put(0.01915035,0.43408734){\color[rgb]{0,0,0}\makebox(0,0)[lb]{\smash{$b$}}}%
    \put(0.03035035,0.53728721){\color[rgb]{0,0,0}\makebox(0,0)[lb]{\smash{$b$}}}%
    \put(0.35914989,0.53968721){\color[rgb]{0,0,0}\makebox(0,0)[lb]{\smash{$b$}}}%
    \put(0.29995,0.43568734){\color[rgb]{0,0,0}\makebox(0,0)[lb]{\smash{$b$}}}%
    \put(0.60954959,0.55648718){\color[rgb]{0,0,0}\makebox(0,0)[lb]{\smash{$b$}}}%
    \put(0.59994956,0.39728741){\color[rgb]{0,0,0}\makebox(0,0)[lb]{\smash{$b$}}}%
    \put(0.01378262,0.21676723){\color[rgb]{0,0,0}\makebox(0,0)[lb]{\smash{$a$}}}%
    \put(0.05102163,0.17116728){\color[rgb]{0,0,0}\makebox(0,0)[lb]{\smash{$b$}}}%
    \put(0.05627219,0.23288761){\color[rgb]{0,0,0}\makebox(0,0)[lb]{\smash{$y_1$}}}%
    \put(0.05317533,0.28998053){\color[rgb]{0,0,0}\makebox(0,0)[lb]{\smash{$y_2$}}}%
    \put(0.01342058,0.08012161){\color[rgb]{0,0,0}\makebox(0,0)[lb]{\smash{$b$}}}%
    \put(0.03071129,0.12448781){\color[rgb]{0,0,0}\makebox(0,0)[lb]{\smash{$a$}}}%
    \put(0.0615753,0.08038076){\color[rgb]{0,0,0}\makebox(0,0)[lb]{\smash{$y_1$}}}%
    \put(0.05597531,0.02198095){\color[rgb]{0,0,0}\makebox(0,0)[lb]{\smash{$y_2$}}}%
    \put(0.33938217,0.24316721){\color[rgb]{0,0,0}\makebox(0,0)[lb]{\smash{$a$}}}%
    \put(0.35662121,0.19836723){\color[rgb]{0,0,0}\makebox(0,0)[lb]{\smash{$b$}}}%
    \put(0.36898218,0.1165718){\color[rgb]{0,0,0}\makebox(0,0)[lb]{\smash{$a$}}}%
    \put(0.33982127,0.0725718){\color[rgb]{0,0,0}\makebox(0,0)[lb]{\smash{$b$}}}%
    \put(0.39677484,0.0747809){\color[rgb]{0,0,0}\makebox(0,0)[lb]{\smash{$y_1$}}}%
    \put(0.39197486,0.24838059){\color[rgb]{0,0,0}\makebox(0,0)[lb]{\smash{$y_1$}}}%
    \put(0.39637484,0.28798054){\color[rgb]{0,0,0}\makebox(0,0)[lb]{\smash{$y_2$}}}%
    \put(0.38637486,0.03238085){\color[rgb]{0,0,0}\makebox(0,0)[lb]{\smash{$y_2$}}}%
    \put(0.70591034,0.19168768){\color[rgb]{0,0,0}\makebox(0,0)[lb]{\smash{$a$}}}%
    \put(0.70074942,0.13648776){\color[rgb]{0,0,0}\makebox(0,0)[lb]{\smash{$b$}}}%
    \put(0.82437417,0.19838067){\color[rgb]{0,0,0}\makebox(0,0)[lb]{\smash{$y_1$}}}%
    \put(0.82637422,0.13758073){\color[rgb]{0,0,0}\makebox(0,0)[lb]{\smash{$y_2$}}}%
    \put(0.06787216,0.43031273){\color[rgb]{0,0,0}\makebox(0,0)[lb]{\smash{$y$}}}%
    \put(0.0600426,0.5929076){\color[rgb]{0,0,0}\makebox(0,0)[lb]{\smash{$y$}}}%
    \put(0.33917493,0.43398031){\color[rgb]{0,0,0}\makebox(0,0)[lb]{\smash{$y_1$}}}%
    \put(0.3423749,0.5979801){\color[rgb]{0,0,0}\makebox(0,0)[lb]{\smash{$y_1$}}}%
    \put(0.6431745,0.39398039){\color[rgb]{0,0,0}\makebox(0,0)[lb]{\smash{$y_1$}}}%
    \put(0.65517449,0.61358007){\color[rgb]{0,0,0}\makebox(0,0)[lb]{\smash{$y_1$}}}%
    \put(0.41317481,0.59878009){\color[rgb]{0,0,0}\makebox(0,0)[lb]{\smash{$y_2$}}}%
    \put(0.43677479,0.43318035){\color[rgb]{0,0,0}\makebox(0,0)[lb]{\smash{$y_2$}}}%
    \put(0.74997425,0.45958033){\color[rgb]{0,0,0}\makebox(0,0)[lb]{\smash{$y_2$}}}%
    \put(0.73557444,0.54518016){\color[rgb]{0,0,0}\makebox(0,0)[lb]{\smash{$y_2$}}}%
    \put(0,0){\includegraphics[width=\unitlength,page=1]{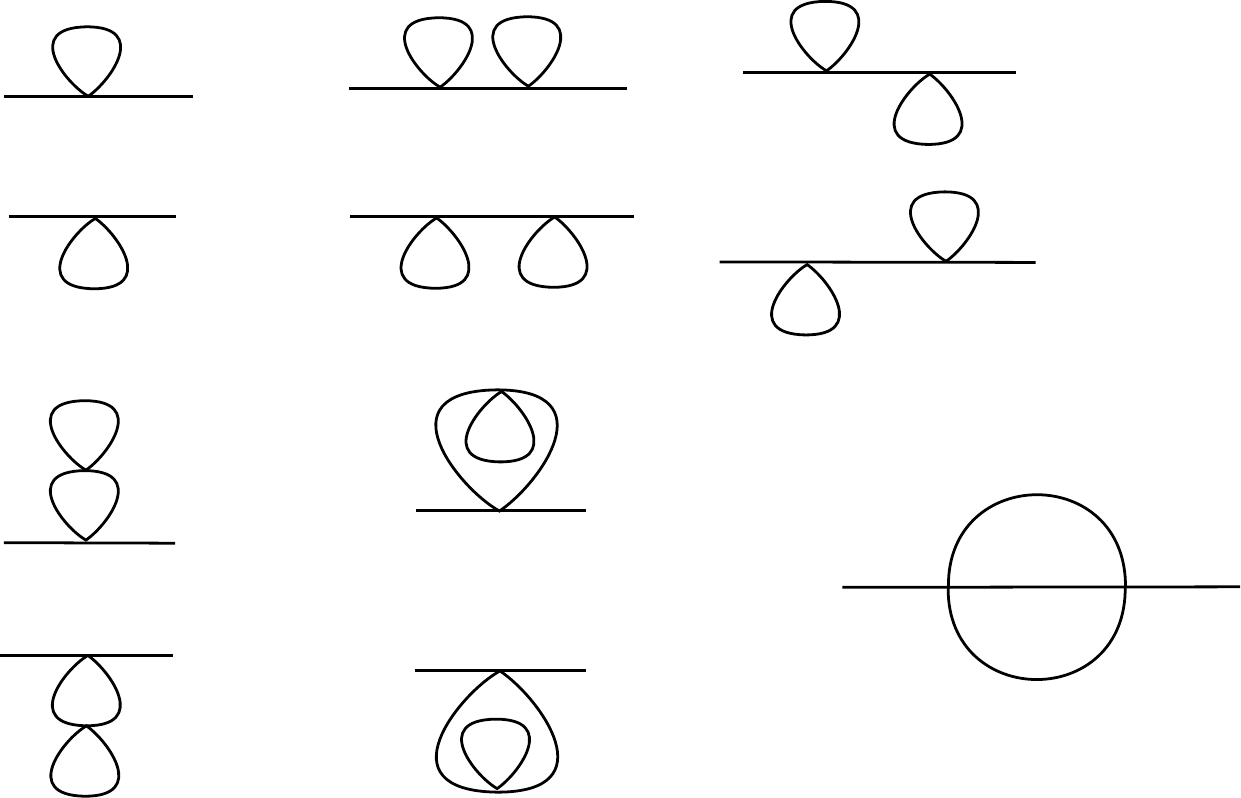}}%
  \end{picture}%
\endgroup%

\noindent
We will determine the expressions for the graphs analogously to previous section.
Only the last graph, the sunrise diagram, needs more discussions through Zimmermann's forest formula. 

For the two graphs at order $\lambda$, the forest formula gives
\begin{align*}
 &-\frac{\lambda}{(1+a+b)^2}\int_0^\infty\!\!\! y\,dy\bigg(\frac{1}{1+a+y}+
 \frac{1}{1+b+y}
 -\frac{2}{1+y}+\frac{a}{(1+y)^2}+\frac{b}{(1+y)^2}\bigg)\\
 &=-\lambda\frac{(1+a)\log(1+a)+(1+b)\log(1+b)-(a+b)}{(1+a+b)^2}
\end{align*}
which coincides with \eqref{od1}.

\subsubsection*{Graphs of the Second Order in $\lambda$}
The first four graphs are computed as the graphs of order $\lambda^1$ which leads to the four results
\begin{align}\label{DD1}
 &\lambda^2\frac{((1+a)\log(1+a)-a)^2}{(1+a+b)^3}\\\label{DD2}
 &\lambda^2\frac{((1+b)\log(1+b)-b)^2}{(1+a+b)^3}\\\label{D3}
 &\lambda^2\frac{((1+b)\log(1+b)-b)((1+a)\log(1+a)-a)}{(1+a+b)^3}\\\label{DD4}
 &\lambda^2\frac{((1+a)\log(1+a)-a)((1+b)\log(1+b)-b)}{(1+a+b)^3}.
\end{align}
The two graphs which have the second loop on top
have the forests $\emptyset,\{\gamma_1\},\{\gamma_2\}$ and $\{\gamma_1,\gamma_2\}$ according to the notation of the 
previous section. The subgraph $\{\gamma_{12}\}$ is not possible since it becomes disjoint after removing the vertex.
The result of these two graphs is therefore easily computed to
\begin{align}\label{DD5}
  \lambda^2\frac{a\log(1+a)+b\log(1+b)-(1+a)\log(1+a)^2-(1+b)\log(1+b)^2}{(1+a+b)^2}.
\end{align}
For the two graphs where the second loop is inside the first, Zimmermann's forest formula 
leads analogously to the discussions of the previous section of the graph $\Gamma_4$ to
\begin{align}\label{DD6}
 -\lambda^2\frac{(1+2a)\mathrm{Li}_2(-a)+(1+2b)\mathrm{Li}_2(-b)+(1+a)\log(1+a)+(1+b)\log(1+b)}{(1+a+b)^2}.
\end{align}
The sunrise graph needs Zimmermann's forest formula in its full beauty. Using the six forests 
$\emptyset$, $\{\gamma_1\}$,
 $\{\gamma_2\}$, $ \{\gamma_{12}\}$, $ \{\gamma_{12},\gamma_1\}$ and 
 $ \{\gamma_{12},\gamma_2\}$, inserting it in the forest formula and counting each degree of divergence leads after 
 simplifying to 
 \begin{align}\nonumber
  &\frac{\lambda^2}{(1+a+b)^2}\int_0^\infty \int_0^\infty y_1\,dy_1\,y_2\,dy_2\\\nonumber
  \times&\bigg[\bigg\{\frac{1}{(1+y_1+y_2)}
 \bigg(1-\Big(T^{1}_{a,b}\Big)\bigg)\frac{1}{(1+a+y_1)(1+b+y_2)}\bigg\}\\\nonumber
 +&\bigg\{\bigg(1
 -\Big(T^{1}_{a,b}\Big)\bigg)\frac{1}{1+b+y_2}
 \Big(-T^{0}_{a,y_2}\Big)\Big(\frac{1}{(1+y_1+y_2)(1+a+y_1)}\Big)\bigg\}\\\nonumber
 +&\bigg\{\bigg(1
 -\Big(T^{1}_{a,b}\Big)\bigg)\frac{1}{1+a+y_1}
 \Big(-T^{0}_{b,y_1}\Big)\Big(\frac{1}{(1+y_1+y_2)(1+b+y_2)}\Big)\bigg\}\bigg]\\\label{DD}
 =&\frac{\lambda^2}{(1+a+b)^3}\big(b(1+b)\mathrm{Li}_2(-b)+a(1+a)\mathrm{Li}_2(-a)-(1+a)(1+b)\log(1+a)\log(1+b)\\\nonumber
 &+(1+a)(1+a+b)\log(1+a)+(1+b)(1+a+b)\log(1+b)\\
 &+ab\zeta_2-a-b-a^2-b^2-2ab\big).\nonumber
 \end{align}
Summing up the result of all graphs at order $\lambda^2$ \eqref{DD1}-\eqref{DD} gives
\begin{align}\nonumber
 &\frac{\lambda^2}{(1+a+b)^3}\{\zeta_2 ab-a-b+(1+a)(1+b)\log(1+a)\log(1+b)\\\label{graph}
&\qquad  -a(1+b)\log(1+b)^2-b(1+a)\log(1+a)^2\\\nonumber
 &\qquad -(1+b+2a+2ab+a^2)\mathrm{Li}_2(-a)-(1+a+2b+2ab+b^2)\mathrm{Li}_2(-b)\\
 &\qquad -(a+2b+a^2+ab)\log(1+a)-(b+2a+b^2+ab)\log(1+b)\},
\end{align} 
Adjusting the boundary conditions for the exact solution in the right way, leads to the same 
results compared to \eqref{graph1}. The natural choice for the $\mu^2$ in \eqref{eq:quartG2}
coincides with graphs expansions for the the hyperlogarithms with two letters. The 
hyperlogarithms with less letters, e.g. $\log(1+a)$ and $\log(1+b)$, are affected by the boundary conditions.

\chapter[Analysis of the Fredholm Equation]{Analysis of the Fredholm 
Equation and a Second Proof of \\ \sref{Proposition}{Prop:Jx-final}\footnote{This is taken from our
paper \cite{Grosse2020}}}\label{App:Solv}
\section[Perturbative Analysis]{Perturbative Analysis}\label{App:Solv2}
\subsection{Direct Expansion}\label{Sec1}

Expanding equation \eqref{cottauba} with the renormalisation motivated by Taylor-subtraction 
\begin{align*}
 \mu_{bare}^2=1-\lambda \Lambda^2-\frac{1}{\pi} \int_0^{\Lambda^2} dt\, \tau_0(t).
\end{align*}
and finite cut-off gives
\begin{align}\label{winkeleq3}
 p\lambda\pi \cot(\tau_a(p))=1+a+p+\lambda p\log\bigg(\frac{\Lambda^2-p}{p}\bigg)
 +\frac{1}{\pi}
 \int_0^{\Lambda^2} dt\left(\tau_p(t)-\tau_0(t)\right).
\end{align}
The first order is read out directly
\begin{align*}
 p\lambda\pi\cot(\tau_a(p))=1+a+p+\mathcal{O}(\lambda^1)\quad \Rightarrow\quad
 \tau_a(p)=\frac{p\lambda\pi}{1+a+p}+\mathcal{O}(\lambda^2),
\end{align*}
which gives after inserting back at the next order
\begin{align*}
 p\lambda\pi\cot(\tau_a(p))=&1+a+p+\lambda\bigg((1+p)\log(1+p)-p\log(p)\\
 &+p\log\left(\frac{\Lambda^2-p}{1+p+\Lambda^2}\right)
 +\log\left(\frac{1+\Lambda^2}{1+p+\Lambda^2}\right)\bigg)+\mathcal{O}(\lambda^2).
\end{align*}
The limit $\Lambda^2\to \infty$ gives finite results for
$\cot(\tau_a(p))$ as well as for $\tau_a(p)$ order by order, however
the limit has to be taken with caution. Integral and limit do
\textit{not} commute.  Namely, for and expansion
$\tau_a(p)=\sum_{n=1}^\infty \lambda^n\tau_a^{(n)}(p)$ we have
\begin{align*}
 \lim_{\Lambda^2\to\infty}\int_0^{\Lambda^2} dt\left(\tau_p^{(n)}(t)-\tau_0^{(n)}(t)\right)\neq
 \int_0^{\infty} dt\lim_{\Lambda^2\to\infty}\left(\tau_p^{(n)}(t)-\tau_0^{(n)}(t)\right), \qquad n>1.
\end{align*}
As an example we will look at the next order of both integrals. They give
\begin{align*}
  &\lim_{\Lambda^2\to\infty}\frac{1}{\pi}\int_0^{\Lambda^2} dt\left(\tau_p^{(2)}(t)-\tau_0^{(2)}(t)\right)\\
 &=(1+p)\log(1+p)^2+(1+2p)\mathrm{Li}_2(-p)-p\zeta_2,\\[10pt]
 &\frac{1}{\pi}\int_0^{\infty} dt\lim_{\Lambda^2\to\infty}\left(\tau_p^{(2)}(t)-\tau_0^{(2)}(t)\right)\\
 &=\int_0^{\infty} dt\,t\left(\frac{t\log(t)-(1+t)\log(1+t)}{(1+t+p)^2}-
 \frac{t\log(t)-(1+t)\log(1+t)}{(1+t)^2}\right)\\
 &=(1+p)\log(1+p)^2+(1+2p)\mathrm{Li}_2(-p)+2p\zeta_2,
\end{align*}
respectively, where $\mathrm{Li}_n(x)$ is the $n^{\text{th}}$
polylogarithm and $\zeta_n\equiv\zeta(n)$ is the Riemann zeta value at
integer $n$. The last term makes the difference.  Taking the ''wrong''
second result and plugging it back into \eqref{winkeleq3} would lead
to divergences at the next order. Consequently, we have to treat the
perturbative expansion of \eqref{winkeleq3} with a finite cut-off
$\Lambda^2$ at all orders, where each order has a finite limit.


We computed the first 6 orders via \textsc{HyperInt} described in \sref{App.}{App:Pert} for finite
$\Lambda^2$. Sending $\Lambda^2\to\infty$ is well-defined
at any order as expected. The first orders read
explicitly
\begin{align}\nonumber
 \lim_{\Lambda^2\to\infty}p\lambda\pi\cot(\tau_a(p))&=1+a+p+\lambda\left((1+p)\log(1+p)-p\log(p)\right)\\
 &+\lambda^2\left(-p\zeta_2 +(1+p)\log(1+p)^2+(1+2p)\mathrm{Li}_2(-p)\right)\nonumber\\
 &+\lambda^3\big(\zeta_2 \log(1+p)- \frac{1+p}{2 p}\log(1+p)^2+(1+p)\log(1+p)^3\nonumber\\
 &\qquad +2p\zeta_3-2p\mathrm{Li}_3(-p)-(1+2p)\mathrm{Hlog}(p,[-1,0,-1])\nonumber\\
 &\qquad-2(2+3p)\mathrm{Hlog}(p,[0,-1,-1])\big)\label{perturbative}
 +\mathcal{O}(\lambda^4).
\end{align}

The defintion of the hyperlogarithms Hlog is given by iterated integrals and can be found in \sref{App.}{App:Pert}. 

The perturbative expansion shows that the branch point
at $p=-1$ plays an important role. Its boundary value is found to be
$
 \lim_{\substack{\Lambda^2\to\infty\\\varepsilon \searrow0}}
\cot(\tau_0(-1+\mathrm{i}\varepsilon))=
-\mathrm{i}+ \mathcal{O}(\lambda^7)$. It is natural to conjecture that it holds
at any order,
\begin{align}\label{conj}
 \lim_{\substack{\Lambda^2\to\infty\\ \varepsilon \searrow 0}}
\cot(\tau_0(-1+\mathrm{i}\varepsilon))=-\mathrm{i}.
\end{align}

The perturbative expansion with a finite cut-off $\Lambda^2$
is quite inefficient. The boundary value \eqref{conj} admits a more
efficient strategy.
We take the derivative of \eqref{winkeleq3} with respect to $p$:
\begin{align*}
 1+\lambda \log\Big(\frac{\Lambda^2{-}p}{p}\Big)
-\lambda\frac{\Lambda^2}{\Lambda^2{-}p}+\frac{1}{\pi}
 \int_0^{\Lambda^2}\!\!\! dt\;\frac{d\tau_p(t)}{dp}=
\lambda \pi \cot(\tau_a(p))+p \lambda\pi \frac{\partial}{\partial p}\cot(\tau_a(p)).
\end{align*}
Multiplying this equation by $p$ and subtracting it from \eqref{winkeleq3} again leads to
\begin{align}\label{AblFormel}
 -p^2\lambda\pi \frac{\partial}{\partial p}\cot(\tau_a(p))=
 1+a+\lambda\frac{p\Lambda^2}{\Lambda^2{-}p}
 +\frac{1}{\pi }\int_0^{\Lambda^2} \!\!\! dt
\left(\tau_p(t)-\tau_0(t)-p\frac{d\tau_p(t)}{dp}\right),
\end{align}
where the limit $\Lambda^2\to\infty$ is now safe from the
beginning and commutes with the integral.
We divide \eqref{AblFormel} by $-p^2$ and integrate it for
all orders higher than $\lambda^1$ over $p$ from
$-1$ (here  \eqref{conj} is assumed) up to some $q$ to get
$\lim_{\Lambda^2\to\infty} \lambda\pi \cot(\tau_a(q))$ on the lhs. On the rhs
the order of integrals $\int_{-1}^q dp\int_0^{\infty} dt$
can be exchanged. The integral over $p$ is
\begin{align}\label{intangl}
 \int_{-1}^q dp \frac{1}{p^2}\left(\tau_p(t)-\tau_0(t)
-p\frac{d\tau_p(t)}{dp}\right),
\end{align}
assuming H\"older continuity of $\tau_p(t)$ so that the integral splits after
taking principal values. The last term is
computed for small $\epsilon$ and all $\mathcal{O}(\lambda^{>1})$-contributions  via integration by parts
\begin{align}\nonumber
 \int_{[-1,q]\backslash [-\epsilon,\epsilon]}dp\frac{\frac{d\tau_p(t)}{dp}}{p}=&
 \frac{\tau_p(t)}{p}\bigg\vert_{p=\epsilon}^q+\frac{\tau_p(t)}{p}\bigg\vert_{p=-1}^{-\epsilon}
 +\int_{[-1,q]\backslash [-\epsilon,\epsilon]}dp\frac{\tau_p(t)}{p^2}\\
 =&\frac{\tau_q(t)}{q}+\tau_{-1}(t)+\int_{[-1,q]\backslash [-\epsilon,\epsilon]}dp\frac{\tau_p(t)}{p^2}
 -\frac{\tau_{-\epsilon}(t)+\tau_{\epsilon}(t)}{\epsilon}.\label{sitestep}
\end{align}
The first term in \eqref{intangl} cancels. The second term in
\eqref{intangl} integrates to a boundary term
$+2\frac{\tau_0(t)}{\epsilon}$, which is also cancelled by the last term
of \eqref{sitestep}.
Multiplying by $q$ and including the special
$\mathcal{O}(\lambda)$-contribution we arrive in the limit
$\Lambda^2\to \infty$ where \eqref{conj} is (conjecturally) available at
\begin{align}\label{winkeleq4}
 q\lambda\pi \cot(\tau_a(q))=1+a+q-\lambda q \log(q)
 +\frac{1}{\pi}
 \int_0^{\infty} dt\left(\tau_q(t)-(1+q)\tau_0(t)+q \tau_{-1}(t)\right).
\end{align}
This equation is much more appropriate for the perturbation theory
because the number of terms is reduced tremendously order by
order. Obviously, the first six order coincide with the earlier but
much harder perturbative expansion of \eqref{winkeleq3}.

Using \eqref{winkeleq4} the perturbative expansion is increased up to
$\lambda^9$ with \textsc{HyperInt}.  As consistency check of
assumption \eqref{conj} we inserted
the next orders
$\tau_a^{(n)}(p)$ into \eqref{Gab-ansatz} to get the expansion
$G(a,b)=\sum_{n=0}^\infty \lambda^n G^{(n)}(a,b)$. This confirmed
the symmetry $G^{(n)}(a,b)=G^{(n)}(b,a)$ which would easily be lost by
wrong assumptions. We are thus convinced to have
the correct expressions for $\tau_a^{(n)}(p)$ for $6<n<10$.

\subsection{Expansion of the Fredholm Equation}
\label{Sec2}

To access the angle function $\tau_a(p)$
we first have to determine the expansion of the
deformed measure ${\varrho}_\lambda(x)=R_4(x)$ through the Fredholm
equation \eqref{Fred}. The constant $\mu^2(\lambda)$ is not yet fixed and needs
a further expansion
\begin{align*}
 \mu^2=\sum_{n=0}^\infty \lambda^n\mu^2_n.
\end{align*}
First orders of the deformed measure are given iteratively through
\eqref{Fred}
\begin{align*}
 {\varrho}_\lambda(x)=&x-\lambda ((x+\mu^2_0)\Hlog(x,[-\mu^2_0])-x)\\
 &-\frac{\lambda^2}{\mu^2_0}(-\mu^2_0x\Hlog(x,[0,-\mu^2_0])+\mu^2_0(\mu^2_1+\mu^2_0+x)\Hlog(x,[-\mu^2_0])-x(\mu^2_1+\mu^2_0))\\
 &+\mathcal{O}(\lambda^3).
\end{align*}
Recall that the inverse of $R_4(x)={\varrho}_\lambda(x)=p$ exists for
all $p\in\mathbb{R}_+$ in case $\lambda< \left(\int_0^\infty
\frac{dt \,{\varrho}_\lambda (t)}{(t+\mu^2)^2 }\right)^{-1}$.
If $\varrho_\lambda(x)$ had the same asymptotics as $\varrho_0(x)=x$
then $R_4^{-1}$ could not be defined globally for $\lambda>0$.
We proved in \sref{Sec.}{Sec.4dSol}
that the asymptotics of $\varrho_\lambda(x)$ is
altered in such a way that $R_4^{-1}$ is defined.
Anyway, in each order of perturbative expansion the inverse $R_4^{-1}$
is globally defined on $\mathbb{R}_+$.
At this point it suffices to assume that $R_4^{-1}(p)$ is a formal power
series in $\lambda$, which is achieved by \eqref{Fred}
\begin{align*}
 R_4^{-1}(p)=p+\lambda (R_4^{-1}(p))^2\int_0^\infty\frac{dt\,{\varrho}_\lambda(t)}{(t+\mu^2)^2(t+\mu^2+R_4^{-1}(p))}.
\end{align*}
Expanding ${\varrho}_\lambda(t)$ and $\mu^2$, the first orders are
\begin{align*}
 R_4^{-1}(p)=&p-\lambda(p-(\mu^2_0+p)\Hlog(p,[-\mu^2_0]))\\
 &-\frac{\lambda^2}{\mu^2_0}(p\mu^2_0\Hlog(p,[0,-\mu^2_0])-2\mu^2_0(p+\mu^2_0)\Hlog(p,[-\mu^2_0,-\mu^2_0])\\
 &\qquad -\mu^2_0(\mu^2_1+\mu^2_0)\Hlog(p,[-\mu^2_0])+p(\mu^2_1+\mu^2_0))+\mathcal{O}(\lambda^3).
\end{align*}
The last step is to determine $\lim_{\varepsilon\to0}\mathrm{Re}I(p+i\varepsilon)=p\lambda\pi\cot(\tau_0(p))$ for $\Lambda^2\to \infty$
via
\begin{align*}
 I(z)=\mu^2+R_4^{-1}(z)+\lambda(\mu^2+R_4^{-1}(z))^2\int_0^\infty \frac{dt\,{\varrho}_\lambda(t)}{(t+\mu^2)^2(t-R_4^{-1}(z))},
\end{align*}
as a formal series.
The first few orders are
\begin{align*}
& \lim_{\varepsilon\to 0} I(p+i\varepsilon)
\\
&=\mu_0^2+p+\lambda(\mathrm{i}\pi p+\mu^2_0+\mu^2_1+(\mu^2_0+p)\Hlog(p,[-\mu^2_0])+p\log(\mu^2_0)-p\log(p)))\\
 &+\lambda^2\big(\mu_0^2(1-\zeta_2)+\mu_1^2+\mu_2^2-p\zeta_2
 +(\mu_0^2+\mu_1^2)\Hlog(p,[-\mu^2_0])\\
 &\quad +2(\mu_0^2+p)\Hlog(p,[-\mu^2_0,-\mu^2_0])-(\mu^2_0+2p)
 \Hlog(p,[0,-\mu^2_0])\big)+\mathcal{O}(\lambda^3).
\end{align*}
Comparing it with \eqref{perturbative} through equation \eqref{tauap-final}
fixes every $\mu^2_i$ uniquely and confirms
\begin{align*}
 \lim_{\varepsilon\to 0}I(p+\mathrm{i}\varepsilon)=\lambda\pi p\cot(\tau_0(p))+\mathrm{i}\lambda\pi p.
\end{align*}
Furthermore,
the first 10 orders are identical with the
expansion of \eqref{winkeleq4}, provided that the
$\mu^2_i$'s are fixed to
\begin{align}\label{mu}
 \mu^2=&1-\lambda+\frac{1}{6}(\pi\lambda)^2-\lambda \frac{1}{3}(\pi\lambda)^2
 +\frac{3}{40}(\pi\lambda)^4-\lambda\frac{8}{45}(\pi\lambda)^4
 +\frac{5}{112}(\pi\lambda)^6-\lambda\frac{4}{35}(\pi\lambda)^6\nonumber \\*
 &+\frac{35}{1152}(\pi\lambda)^8
 -\lambda\frac{128}{1575}(\pi\lambda)^8+\frac{63}{2816}(\pi\lambda)^{10}+\mathcal{O}(\lambda^{11}).
\end{align}
The conjectured behavior of $\cot(\tau_0(p))$ at $p=-1+\mathrm{i}\varepsilon$ in the previous subsection
\eqref{conj} is now equivalent to
\begin{align*}
 \lim_{\varepsilon\to0} I(-1+\mathrm{i}\varepsilon)=0\quad \Rightarrow\quad R_4^{-1}(-1)=-\mu^2.
\end{align*}
We find that
the expansion \eqref{mu} of $\mu^2$ obeys an unexpected boundary condition
\begin{align}\label{condmu}
 \int_0^\infty\frac{dt\,{\varrho}_\lambda(t)}{(\mu^2+t)^3}=\frac{1}{2}+\mathcal{O}(\lambda^{10}).
\end{align}

For further study we pass as in \sref{Sec.}{Sec.4dSol} to the rescaled measure
$\phi(x)=
\mu^2 \tilde{\varrho}_\lambda(\mu^2 x):=\frac{\varrho_\lambda(\mu^2 x)}{
\mu^2 x(1+ x)}$.
The pattern of coefficients of the $\mu^2$-expansion in \eqref{mu}
suggests to distinguish between even an odd powers in
$\lambda$. The even powers $\lambda^{2n}$ are given by the formula
\begin{align*}
 \frac{(2n-1)!!}{(2n)!! (2n+1)}=\frac{(2n)!}{4^nn!^2(2n+1)},
\end{align*}
and the odd powers $\lambda^{2n+1}$ by
\begin{align*}
 2\frac{(2n)!!}{(2n+1)!!(2n+2)}=2\frac{4^nn!^2}{(2n+2)!}.
\end{align*}
Both series are convergent for $|\lambda|<\frac{1}{\pi}$ with the result (up to order $\lambda^{10}$)
\begin{align*}
 \mu^2=\frac{\arcsin(\lambda\pi)}{\lambda\pi}-\lambda \left(\frac{\arcsin(\lambda\pi)}{\lambda\pi}\right)^2.
\end{align*}
This result suggests that $\frac{\arcsin(\lambda\pi)}{\pi}$ is a
better expansion parameter than $\lambda$ itself. The factors
$\pi^{2n}$ are produced by $\zeta_{2n}$ in the  iterated integrals.
We thus reorganise the perturbative solution of \eqref{feq}
into a series in $\frac{\arcsin(\lambda\pi)}{\pi}$.
The power of $\frac{\arcsin(\lambda\pi)}{\lambda\pi}$
depends on the
number of letters of the hyperlogarithm, which alternate between
 $-1$ and $0$.
The expansion which holds up to order $\lambda^{10}$ is given by
\begin{align}
\phi(x)=&c_\lambda \frac{\arcsin(\lambda\pi)}{\lambda\pi(1+x)}\sum_{n=0}^\infty \mathrm{Hlog}(x,[\underbrace{0,-1,...,0,-1}_{n}])
 \left(\frac{\arcsin(\lambda\pi)}{\pi}\right)^{2n}\\
 &-\lambda c_\lambda\frac{\arcsin(\lambda\pi)^2}{x(\lambda\pi)^2}\sum_{n=0}^\infty \mathrm{Hlog}(x,[-1,\underbrace{0,-1,...,0,-1}_{n}])
 \left(\frac{\arcsin(\lambda\pi)}{\pi}\right)^{2n},\nonumber
\end{align}
where the underbrace with $n$ means that we have $n$ times the letters $0$ and $-1$ in an alternating way.

In the limit $x\to0$ only the terms with $n=0$ in both sums survive,
\begin{align*}
1\equiv \phi(0)&=c_\lambda\frac{\arcsin(\lambda\pi)}{\lambda\pi}
\lim_{x\to0}\frac{\mathrm{Hlog}(x,[\,])}{1+x}-
 \lambda c_\lambda\frac{\arcsin(\lambda\pi)^2}{(\lambda\pi)^2}\lim_{x\to0}\frac{\mathrm{Hlog}(x,[-1])}{x}\\*
 &= \frac{c_\lambda}{\lambda}
\frac{\arcsin(\lambda\pi)}{\pi}
\Big(1-\frac{\arcsin(\lambda\pi)}{\lambda\pi}\Big).
\end{align*}
This value was found in \sref{Sec.}{Sec.4dSol} by another method.
We also remark that $c_\lambda=\frac{1}{\mu^2}$ for the special renormalisation.

Next define the functions
\begin{align*}
 f(x)&:=\sum_{n=0}^\infty \mathrm{Hlog}(x,[\underbrace{0,-1,...,0,-1}_{n}])
 \,\alpha_\lambda^{2n}\\
 g(x)&:=\sum_{n=0}^\infty\mathrm{Hlog}(x,[-1,\underbrace{0,-1,...,0,-1}_{n}])
 \,\alpha_\lambda^{2n},
\end{align*}
where $\alpha_\lambda=\frac{\arcsin(\lambda\pi)}{\pi}$.
Both together obey the differential equations
\begin{align*}
 f'(x)=\frac{\alpha_\lambda^2}{x}g(x)\qquad g'(x)=\frac{1}{1+x}f(x),
\end{align*}
or equivalently
\begin{align*}
 f''(x)+\frac{f'(x)}{x}-\alpha_\lambda^2 \frac{f(x)}{(1+x)x}=0,\quad
g''(x)+\frac{g'(x)}{1+x}-\alpha_\lambda^2 \frac{g(x)}{(1+x)x}=0,
\end{align*}
with the boundary conditions $f(0)=1$, $f'(0)=\alpha_\lambda^2$,
$g(0)=0$ and $g'(0)=1$.
The solution is given by hypergeometric functions $_2F_1$
\begin{align*}
 f(x)={}_2F_1\Big(\genfrac{}{}{0pt}{}{\alpha_\lambda,\;{-}\alpha_\lambda}{1}
\Big| -x\Big) \qquad
g(x)=\frac{x}{\alpha_\lambda^2}f'(x)=x{}_2F_1
\Big(\genfrac{}{}{0pt}{}{1{+}\alpha_\lambda,\;1{-}\alpha_\lambda}{2}
\Big| -x\Big).
\end{align*}
In summary. the solution of equation \eqref{feq} is conjectured to be
\begin{align}
\phi(x)&=\frac{\alpha_\lambda c_\lambda}{\lambda(1+x)}
{}_2F_1\Big(\genfrac{}{}{0pt}{}{\alpha_\lambda,\;{-}\alpha_\lambda}{1}
\Big| -x\Big)
 -\frac{\alpha_\lambda^2c_\lambda}{\lambda}
{}_2F_1
\Big(\genfrac{}{}{0pt}{}{1{+}\alpha_\lambda,\;1{-}\alpha_\lambda}{2}
\Big| -x\Big)
\nonumber
\\
&={}_2F_1
\Big(\genfrac{}{}{0pt}{}{1{+}\alpha_\lambda,\;2{-}\alpha_\lambda}{2}
\Big| -x\Big)
\end{align}
or equivalently for \eqref{Fred}
\begin{align}
 R_4(x)=\varrho_\lambda(x)=&\frac{x}{\mu^2}\bigg(1+\frac{x}{\mu^2}\bigg)f\bigg(\frac{x}{\mu^2}\bigg)
=x  \,{}_2F_1
\Big(\genfrac{}{}{0pt}{}{\alpha_\lambda,\;1{-}\alpha_\lambda}{2}
\Big| -\frac{x}{\mu^2}\Big),
\label{Jx}
\end{align}
where we have used
the Gauss recursion formula \cite[$\S$ 9.137.7]{gradshteyn2007}
for hypergeometric functions.
Finally, we note that
\[
\int_0^\infty \frac{dt\; \varrho_\lambda(t)}{(t+\mu^2)^3}
=\lim_{x\to 0} \frac{x-\varrho_\lambda(x)}{\lambda x^2}
= \frac{\alpha_\lambda(1-\alpha_\lambda)}{2\lambda\mu^2}=\frac{1}{2c_\lambda\mu^2}\;.
\]
Thus choosing $\mu^2=\frac{\alpha_\lambda(1-\alpha_\lambda)}{\lambda}$
we confirm (\ref{condmu}) exactly.

\section[Proof with Meijer G-Function]
{Proof with Meijer G-Function}\label{App:Meijer}
We find it interesting to directly check that the hypergeometric function
$\tilde{\varrho}_\lambda(x)=\frac{1}{\mu^2}\phi(\frac{x}{\mu^2})$,
see (\ref{phi-sol}), solves the integral equation (\ref{intc}).
The hypergeometric function can be expressed through the more general
Meijer-G function. A Meijer G-function is defined by
\begin{align}\label{AppDef}
 G^{m,n}_{p,q}\Big(z\Big|
\genfrac{}{}{0pt}{}{a_1,...,a_p}{b_1,...,b_q}\Big)=\frac{1}{2\pi \I}\int_L
\frac{\prod_{j=1}^m \Gamma(b_j-s)\prod_{j=1}^n \Gamma(1-a_j+s)}
{\prod_{j=m+1}^q \Gamma(1-b_j+s)\prod_{j=n+1}^p \Gamma(a_j-s)}z^sds,
\end{align}
with $m,n,p,q\in\N$, with $m\leq q$ and $n\leq p$, and
poles of
$\Gamma(b_j-s)$ different from poles of $\Gamma(1-a_j+s)$.
The infinite contour $L$ separates between the poles of
$\Gamma(b_j-s)$ and $\Gamma(1-a_j+s)$, and its behavior to
infinity depends on $m,n,p,q$ (see \cite[\S 9.3]{gradshteyn2007}).

The Meijer G-function has by definition the property
\begin{align}\label{AppInv}
 G^{m,n}_{p,q}\Big(z\Big|
\genfrac{}{}{0pt}{}{a_1,...,a_p}{b_1,...,b_q}\Big)=\frac{1}{z}G^{n,m}_{q,p}\Big(\frac{1}{z}\Big|
\genfrac{}{}{0pt}{}{-b_1,...,-b_q}{-a_1,...,-a_p}\Big).
\end{align}
It obeys the convolution formula \cite[$\S$ 7.811.1]{gradshteyn2007}
\begin{align}\nonumber
 &\int_0^\infty dx
 G^{m,n}_{p,q}\Big(\alpha x\Big|
\genfrac{}{}{0pt}{}{a_1,...,a_p}{b_1,...,b_q}\Big)
G^{m',n'}_{p',q'}\Big(\beta x\Big|
\genfrac{}{}{0pt}{}{a'_1,...,a'_{p'}}{b'_1,...,b'_{q'}}\Big)\\
&\qquad =\frac{1}{\alpha}G^{n+m' ,m+n'}_{q+p',p+q'}\Big(\frac{\beta}{\alpha}\Big|
\genfrac{}{}{0pt}{}{-b_1,..,-b_m,a'_1,..,a'_{p'},-b_{m+1},..,-b_q}
{-a_1,..,-a_n,b'_1,..,b'_{q'},-a_{n+1},..,-a_p}\Big),\label{AppConv}
\end{align}
which is the source of numerous impressive integrals over $\mathbb{R}_+$
of products of special functions.
If no two $b_j$ differ by an integer, either $p<q$ or $p=q$ with $|z|<1$,
then a Meijer G-function can
be expressed by hypergeometric functions
\begin{align}\label{AppHyp}
G^{m,n}_{p,q}\Big(z\Big|
\genfrac{}{}{0pt}{}{a_1,\dots,a_p}{b_1,\dots,b_q}\Big)
&=\sum_{k=1}^m\frac{\prod_{j=1}^{'m}\Gamma(b_j-b_k)\prod_{j=1}^n\Gamma(1+b_k-a_j)}
{\prod_{j=m+1}^q\Gamma(1+b_k-b_j)\prod_{j=n+1}^p\Gamma(a_j-b_k)}z^{b_k}
\\
&\times \,_pF_{q-1}
\Big(\genfrac{}{}{0pt}{}{1+b_k-a_1,\dots,1+b_k-a_p}{
1+b_k-b_1,..,\star,..,1+b_k-b_q}\Big|(-1)^{p-n-m}z\Big)\nonumber,
\end{align}
where primed sum and the $\star$ means that the term with $j=k$ is omitted.

We need another identity which is derived directly from the definition
\begin{align}\nonumber
 G^{3,2}_{3,3}\Big( z\Big|
\genfrac{}{}{0pt}{}{0,0,1}{b_1,b_2,0}\Big)=&\frac{1}{2\pi \I}\int_L
\frac{\Gamma(b_1-s)\Gamma(b_2-s)\Gamma(-s)\Gamma(1+s)^2}
{\Gamma(1-s)}z^sds
\\\nonumber
=&-\frac{1}{2\pi \I}\int_L
 \Gamma(b_1-s)\Gamma(b_2-s)\Gamma(s)\Gamma(1+s)
z^sds
\\\nonumber
=&\Gamma(b_1)\Gamma(b_2)-\frac{1}{2\pi \I}\int_{L'}
 \Gamma(b_1-s)\Gamma(b_2-s)\Gamma(s)\Gamma(1+s)
z^sds\\
=&\Gamma(b_1)\Gamma(b_2)-G^{2,2}_{2,2}\Big( z\Big|
\genfrac{}{}{0pt}{}{0,1}{b_1,b_2}\Big)\label{AppSpec},
\end{align}
where the contour is changed $L\to L'$ such that it is moved
through $s=0$ and picked up the residue.
The contour $L'$ fulfils the definition \eqref{AppDef} for $G^{2,2}_{2,2}\Big( z\Big|
\genfrac{}{}{0pt}{}{0,1}{b_1,b_2}\Big)$.

From \eqref{AppHyp} one can establish
\begin{align*}
\tilde{\varrho}_\lambda(t)&=
\frac{1}{\mu^2}
\frac{1}{\Gamma(2-\alpha_\lambda)\Gamma(1+\alpha_\lambda)}
G^{1,2}_{2,2}\Big(\frac{t}{\mu^2}\Big|
\genfrac{}{}{0pt}{}{\alpha_\lambda-1,-\alpha_\lambda}{0,-1}\Big)\;,
\end{align*}
and $\frac{1}{x+t+\mu^2}= \frac{1}{x+\mu^2}
{}_1F_0\big(\genfrac{}{}{0pt}{}{1}{-}\big|
{-}\frac{t}{x+\mu^2}\big)=
\frac{1}{x+\mu^2}
G^{1,1}_{1,1}\big(\frac{t}{x+\mu^2}\big|
\genfrac{}{}{0pt}{}{0}{0}\big)$.
The convolution theorem \eqref{AppConv} of
Meijer G-functions thus allows to evaluate the integral
\begin{align}
&\lambda
\int_0^\infty \frac{dt\;\tilde{\varrho}_\lambda(t)}{x+t+\mu^2}
\nonumber
\\
&\stackrel{\text{\eqref{AppConv}}}{=} \frac{\lambda}{\mu^2\Gamma(2-\alpha_\lambda)\Gamma(1+\alpha_\lambda)}
G^{2,3}_{3,3}\Big(\frac{x+\mu^2}{\mu^2}\Big|
\genfrac{}{}{0pt}{}{\alpha_\lambda-1,-\alpha_\lambda,0}{0,0,-1}\Big)
\nonumber
\\
&\stackrel{\text{\eqref{AppInv}}}{=}
\frac{\lambda}{(x+\mu^2)\Gamma(2-\alpha_\lambda)\Gamma(1+\alpha_\lambda)}
G^{3,2}_{3,3}\Big(\frac{\mu^2}{x+\mu^2}\Big|
\genfrac{}{}{0pt}{}{0,0,1}{1-\alpha_\lambda,\alpha_\lambda,0}\Big)
\nonumber
\\
&\stackrel{\text{\eqref{AppSpec}}}{=}
\frac{\lambda}{(x+\mu^2)\Gamma(2-\alpha_\lambda)\Gamma(1+\alpha_\lambda)}\left(\Gamma(1-\alpha_\lambda)
\Gamma(\alpha_\lambda)-
G^{2,2}_{2,2}\Big(\frac{\mu^2}{x+\mu^2}\Big|
\genfrac{}{}{0pt}{}{0,1}{1-\alpha_\lambda,\alpha_\lambda}\Big)\right)
\nonumber
\\
&\stackrel{\text{\eqref{AppHyp}}}{=}
\frac{\lambda}{(x+\mu^2)}
\Big\{\frac{1}{\alpha_\lambda(1-\alpha_\lambda)}
\nonumber
\\
&-
\frac{\Gamma(2\alpha_\lambda-1)
\Gamma(1-\alpha_\lambda)}{\Gamma(1+\alpha_\lambda)}
\Big(\frac{\mu^2}{x+\mu^2}\Big)^{1-\alpha_\lambda}
\;{}_2F_1
\Big(\genfrac{}{}{0pt}{}{2{-}\alpha_\lambda,\;1{-}\alpha_\lambda}{2{-}2\alpha_\lambda}
\Big| \frac{\mu^2}{x+\mu^2}\Big)
\nonumber
\\
& -\frac{\Gamma(1-2\alpha_\lambda)
\Gamma(\alpha_\lambda) }{\Gamma(2-\alpha_\lambda)}
\Big(\frac{\mu^2}{x+\mu^2}\Big)^{\alpha_\lambda}
\;{}_2F_1
\Big(\genfrac{}{}{0pt}{}{1{+}\alpha_\lambda,\;\alpha_\lambda}{2\alpha_\lambda}
\Big| \frac{\mu^2}{x+\mu^2}\Big)
\Big\}
\nonumber
\\
&=
\frac{1}{(x+\mu^2)}
\frac{\lambda}{\alpha_\lambda(1-\alpha_\lambda)}
- \frac{\lambda\pi}{\sin (\alpha_\lambda \pi)} \tilde{\varrho}_\lambda(x)\;.
\end{align}
We have used the expansion of a Meijer G-function into hypergeometric
functions and applied in the last step  \cite[\S 9.132.1]{gradshteyn2007}.
The result is precisely (\ref{intc})
provided that
$c_\lambda=\frac{\lambda}{\alpha_\lambda(1-\alpha_\lambda)}$ (see \eqref{phi-sol}) and
$\sin( \alpha_\lambda\pi)=\lambda\pi$ (see \eqref{sol-alpha}).

\section[On the Spectrum of the Fredholm Integral Operator]
{On the Spectrum of the Fredholm Integral Operator
\footnote{Contributed by Robert Seiringer to \cite{Grosse2020}}}\label{app}
Abstractly, the integral equation (\ref{intc}) is of the form
\[
\psi = f_\mu - \lambda A_\mu \psi,
\]
where $\psi(t)=\tilde{\varrho}_\lambda(t)$,
$f_\mu(t)= (t+\mu^2)^{-1}$ and $A_\mu$ is the operator with integral kernel
\begin{align}
A_\mu(t,u)=  \frac {u t}{(u+\mu^2)(u+t+\mu^2)(t+\mu^2)}.
\end{align}
Note that $A_\mu$ is symmetric and positive. The equation can thus be
solved for $\psi$ if $\lambda > \lambda_c = -\|A_\mu\|^{-1}$.

By scaling, the spectrum of $A_\mu$ is independent of $\mu$ for $\mu>0$.
We claim that
\begin{align}
\|A_\mu\| = \|A_0\| = \pi.
\end{align}
In particular, $\lambda_c = -1/\pi$.

Since $A_\mu$ has a positive kernel which is monotone in $\mu$, one
readily obtains $\|A_\mu\|\leq \|A_0\|$. On the other hand, $A_0$ is the
weak limit of $A_\mu$ as $\mu\to 0$, hence $\|A_0\|\leq \liminf_{\mu\to
  0}\|A_\mu\|$, which proves that $\|A_\mu\|=\|A_0\|$. Rational fraction expansion of $A_\mu$ gives $f_\mu$ 
  an additional factor and changes the integral kernel to $A_0\to
(u+t)^{-1}$. Introducing logarithmic coordinates, we have
\begin{align}
\int_0^\infty\int_0^\infty \frac{ \phi(u)^* \phi(t)}{u+t} du dt
&=
\int_{\mathbb{R}} \int_{\mathbb{R}} \frac{\phi^*(e^v) \phi(e^s)}{e^v + e^s}
e^{v+s} dv ds
\nonumber
\\
&=  \int_{\mathbb{R}} \int_{\mathbb{R}}
\frac{\phi^*(e^v)e^{v/2} \phi(e^s)e^{s/2}}{2 \cosh( \tfrac 12(v-s) )} dv ds
\end{align}
which can be diagonalised via Fourier transforms. Since
\[
\int_{\mathbb{R}} \frac 1 {2 \cosh(v/2) } dv = \pi,
\]
this shows that the spectrum of $A_0$ equals $[0,\pi]$, and
indeed $\|A_0\|=\|A_\mu\|=\pi$.

\chapter[Catalan Tuples and Catalan Tables]
{Examples for Catalan Tuples, Catalan Tables and its Diagrammatic Representation\footnote{This 
is taken from the appendix of our paper \cite{DeJong}}}\label{App:Expl}
\section{Examples for Catalan Tuples and Catalan Tables}
\label{app:ex}
Catalan tuples are introduced in \sref{Definition}{dfnt:cattup} and Catalan tables in \sref{Definition}{dfnt:cattab}.
\begin{exm}
We have $(1,0)=(0)\circ(0)$,  $(2,0,0)=(1,0)\circ(0)$,  
$(1,1,0)=(0)\circ(1,0)$ and 
$(3,1,0,0,2,0,0)=(2,1,0,0)\circ(2,0,0)$.
\end{exm}

\begin{exm}
We have $(1,0)=(0)\bullet(0)$,  $(2,0,0)=(1,0)\bullet (0)$,  
$(1,1,0)=(0)\bullet (1,0)$ and 
$(3,1,0,0,2,0,0)=(2,0,2,0,0)\bullet(1,0)$.
\end{exm}

\begin{exm} We have 
\begin{align*}
\hspace*{-1.4cm}
\mathcal{T}_1 &=\{\langle (0),(0)\rangle \}\;, 
\\
\hspace*{-1.4cm}\mathcal{T}_2 &=\{\langle (1,0),(0),(0)\rangle ,~ \langle (0),(1,0),(0)\rangle
\}
\\
\hspace*{-1.4cm}\mathcal{T}_3 &=\{\langle (2,0,0),(0),(0),(0)\rangle ,~
\langle (1,1,0),(0),(0),(0)\rangle ,~
\langle (1,0),(1,0),(0),(0)\rangle ,
\\
&\qquad \langle (1,0),(0),(1,0),(0)\rangle ,~
\langle (0),(2,0,0),(0),(0)\rangle ,~
\langle (0),(1,1,0),(0),(0)\rangle , 
\\
&\qquad \langle (0),(1,0),(1,0),(0)\rangle \}\;.
\end{align*}
Later in \sref{Figure}{f:G4} and \ref{f:G6} we give a diagrammatic representation
of the Catalan tables in $\mathcal{T}_2$ and $\mathcal{T}_3$,
respectively. 
\end{exm} 

\begin{exm}
We have $\langle(2,0,0),(0),(0),(0)\rangle=
\langle(1,0),(0),(0)\rangle \smalllozenge \langle(0),(0)\rangle$ and
$\langle(1,1,0),(0),(0),(0)\rangle=\langle(0),(0)\rangle\smalllozenge 
\langle(1,0),(0),(0)\rangle$. In \sref{Example}{ex:G12a} and 
\sref{Figure}{f:G12a} 
we considered the Catalan table
$\langle (2,0,0),(1,1,0),(0),(0),(0),(1,0),(0)\rangle
= \langle (1,0),(1,1,0),(0),(0),(0)\rangle 
\smalllozenge \langle (0),(1,0),(0)\rangle$.
Another example will be given in \sref{Example}{exm:triangle}.
\end{exm}

\begin{exm}
We have $\langle(0),(2,0,0),(0),(0)\rangle=
\langle(0),(1,0),(0)\rangle \smallblacklozenge \langle(0),(0)\rangle$ and
$\langle(0),(1,1,0),(0),(0)\rangle=\langle(0),(0)\rangle\smallblacklozenge 
\langle(1,0),(0),(0)\rangle$. In \sref{Example}{ex:G12a} and 
\sref{Figure}{f:G12a} 
we considered the Catalan table
$\langle (2,0,0),(1,1,0),(0),(0),(0),(1,0),(0)\rangle
= \langle (2,0,0),(0), (0), (1,0),(0)\rangle 
\smallblacklozenge \langle (1,0),(0),(0)\rangle$.
Another example will be given in \sref{Example}{exm:box}.
\end{exm}

\section{Chord Diagrams with Threads}\label{App:Chord}
For uncovering the combinatorial structure 
of (\ref{e:rr}), it was extremely helpful for us to have a 
graphical presentation as diagrams of chords and threads. 
To every term 
of the expansion (\ref{expansion}) of an $N$-point function we 
associate a diagram as follows: 
\begin{dfnt}[diagrammatic presentation]
Draw $N$ nodes on a circle, label them
from $p_0$ to $p_{N-1}$. Draw a green chord 
between $p_r,p_s$ for every
factor $G_{p_rp_s}$ in (\ref{expansion}) and a (orange
for $t,u$ even, blue for $t,u$ odd) 
thread between 
$p_t,p_u$ for every factor $\frac{1}{E_{p_t}-E_{p_u}}$. 
The convention $t<u$ is chosen so that the diagrams come with a sign.
\end{dfnt}
\noindent 
It was already known in \cite{Grosse:2012uv} that the chords do not cross each
other (using cyclic invariance (\ref{e:rr_rot})) and that the
threads do not cross the chords (using  (\ref{rfp})). But the
combinatorial structure was not understood in \cite{Grosse:2012uv} and no
algorithm for a canonical set of chord diagrams could be given. This work
repairs this omission.

The $N/2=k+1$ chords in such a diagram divide the circle into $k+2$
pockets. The pocket which contain the arc segment between the
designated nodes $p_{0}$ and $p_{N-1}$ is by definition the root 
pocket $P_{0}$. Moving in the
counterclockwise direction, every time a new pocket is entered it is
given the next number as index, as in \sref{Definition}{dfnt:CTG}. The
tree of these $k+2$ pockets, connecting vertices if the pockets border
each other, is the pocket tree. A pocket is called even (resp.\ odd)
if its index is even (resp.\ odd).

Inside  every even pocket, the orange threads (between even nodes) 
form the direct tree, the blue threads (between odd nodes) 
form the opposite tree. 
Inside every odd pocket, the orange threads (between even nodes) 
form the opposite tree, the blue threads (between odd nodes) 
form the direct tree. 

The sign $\tau$ of the diagram is given by 
\begin{equation}
\tau(T)=(-1)^{\sum_{j=1}^{k+1}e_{0}^{(j)}}\;,\label{e:sign}
\end{equation}
where $e^{(j)}_{0}$ is the first entry of the Catalan tuple 
corresponding to a pocket $P_{j}$. Indeed, 
for every pocket that is not a leaf or the root pocket, the chain of 
odd nodes starts with the highest index, which implies that every 
thread emanating from this node contributes a factor $(-1)$ to the
monomial (\ref{expansion}) compared with the lexicographic order
chosen there. In words: count for all pockets
other than the root pocket the total number $K$ of threads which go from
the smallest node into the pocket. The sign is even (odd) if $K$ is
even (odd). 

\sref{Figure}{f:G4} and \ref{f:G6} show Catalan tables and chord diagrams of the
$4$-point function and $6$-point function, respectively. 
\sref{Figure}{f:G12a} shows the chord diagram discussed in 
\sref{Example}{ex:G12a}.
\begin{figure}[!hpt]
\begin{picture}(120,30)\setlength{\unitlength}{1mm}
\put(25,3){\includegraphics[width=3cm]{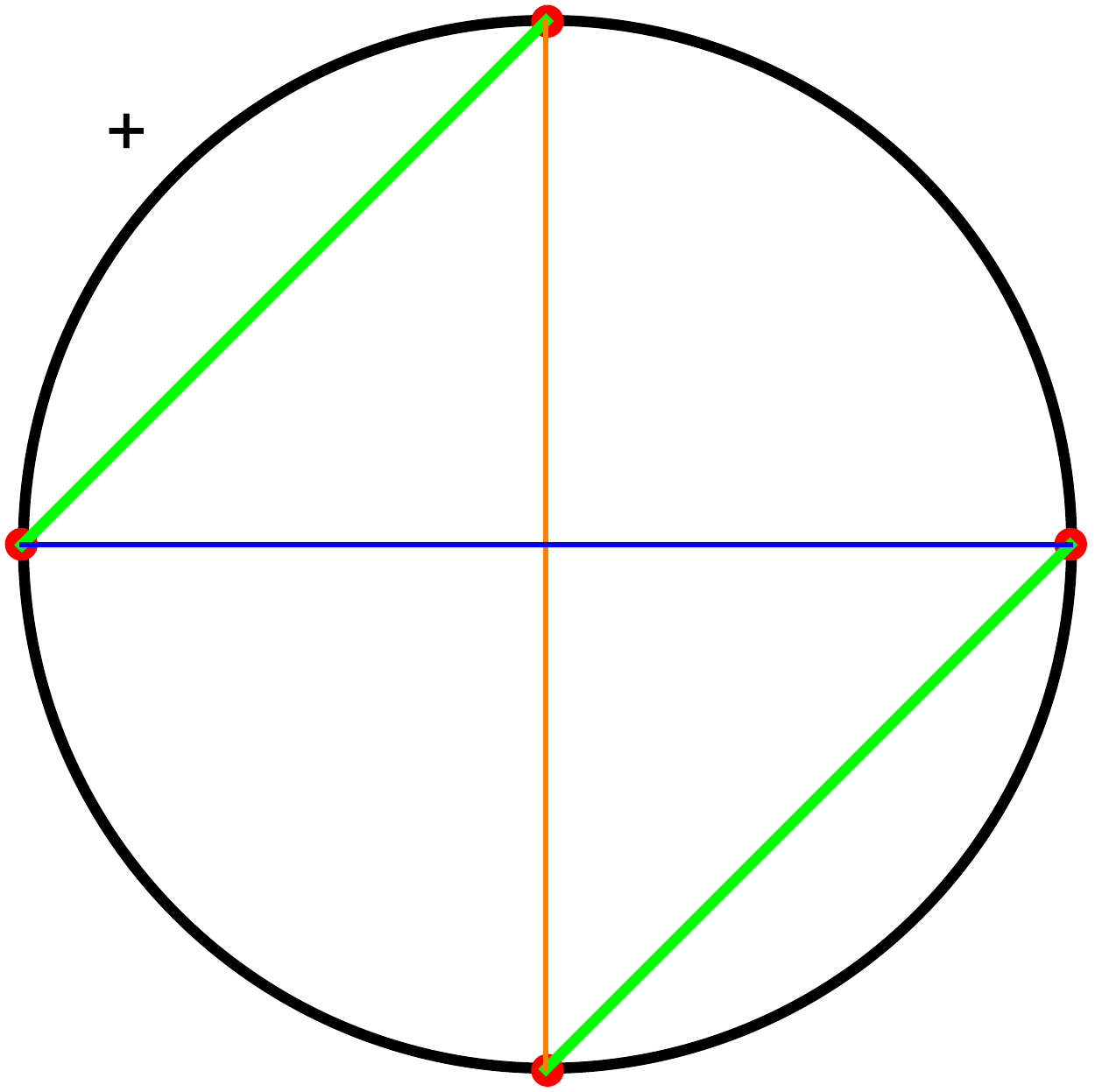}}
\put(65,3){\includegraphics[width=3cm]{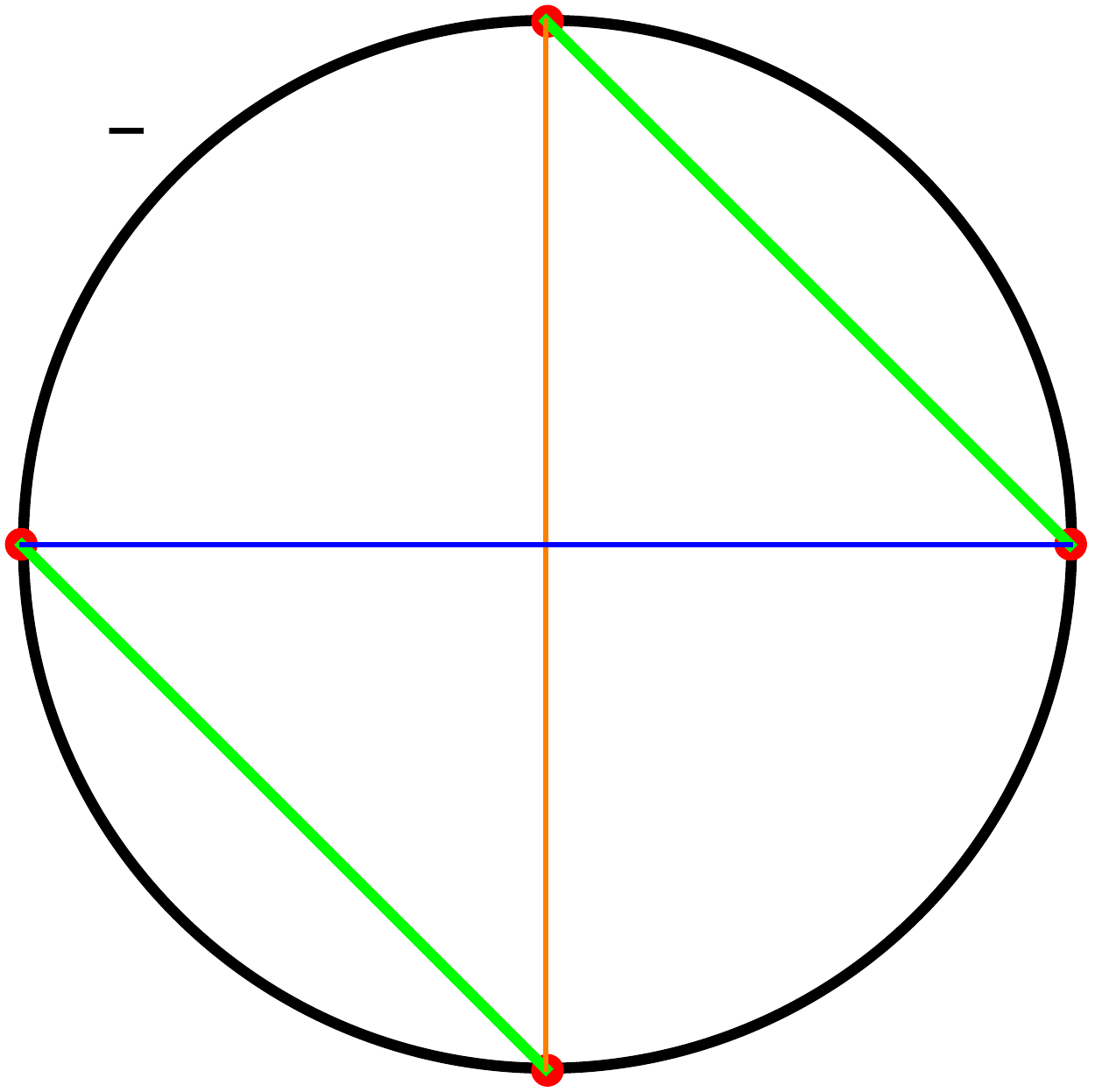}}
\put(30,0){\mbox{\scriptsize$\langle (1,0),(0),(0)\rangle$}}
\put(70,0){\mbox{\scriptsize$\langle (0),(1,0),(0)\rangle$}}
\end{picture}
\caption{\noindent The two chord diagrams and Catalan tables 
of $G^{(4)}_{p_{0}p_{1}p_{2}p_{3}}$.\label{f:G4}}
\end{figure}
\begin{figure}[!hpt]
\begin{picture}(120,73)\setlength{\unitlength}{1mm}
\put(0,41){\includegraphics[width=3cm]{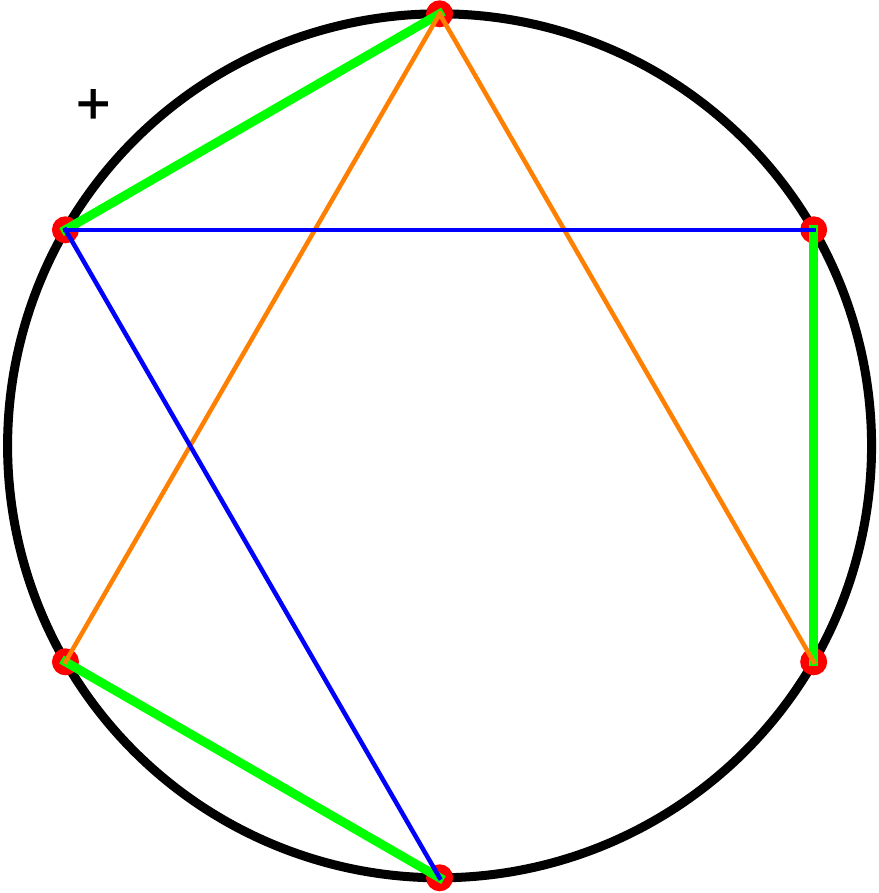}}
\put(35,41){\includegraphics[width=3cm]{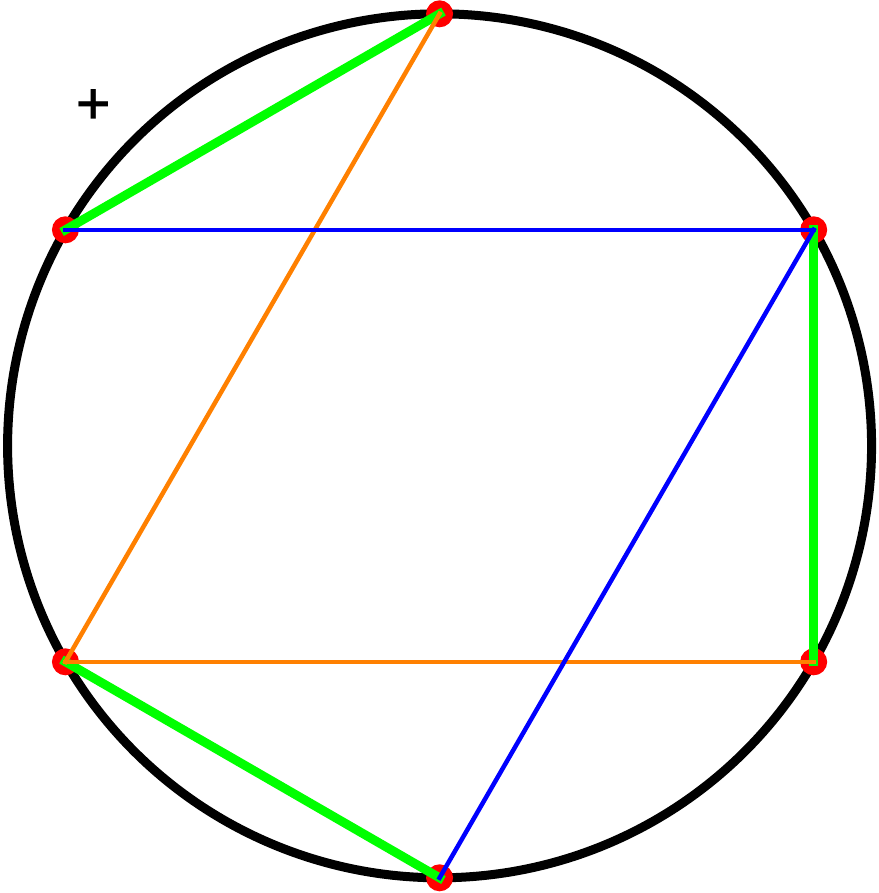}}
\put(70,41){\includegraphics[width=3cm]{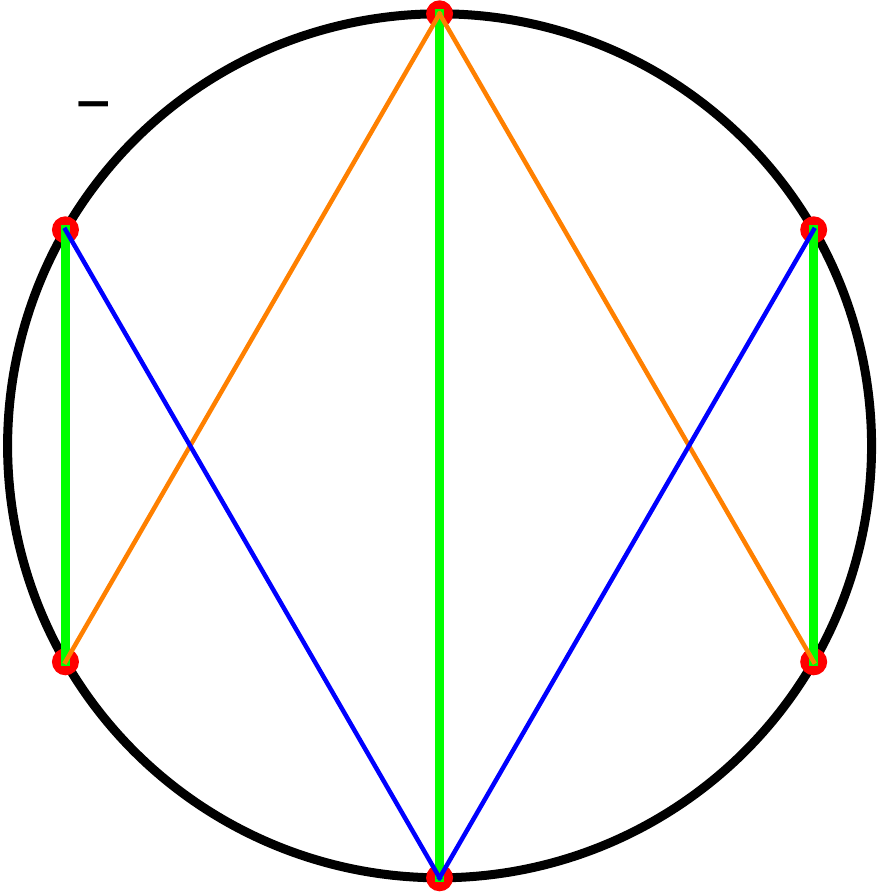}}
\put(105,41){\includegraphics[width=3cm]{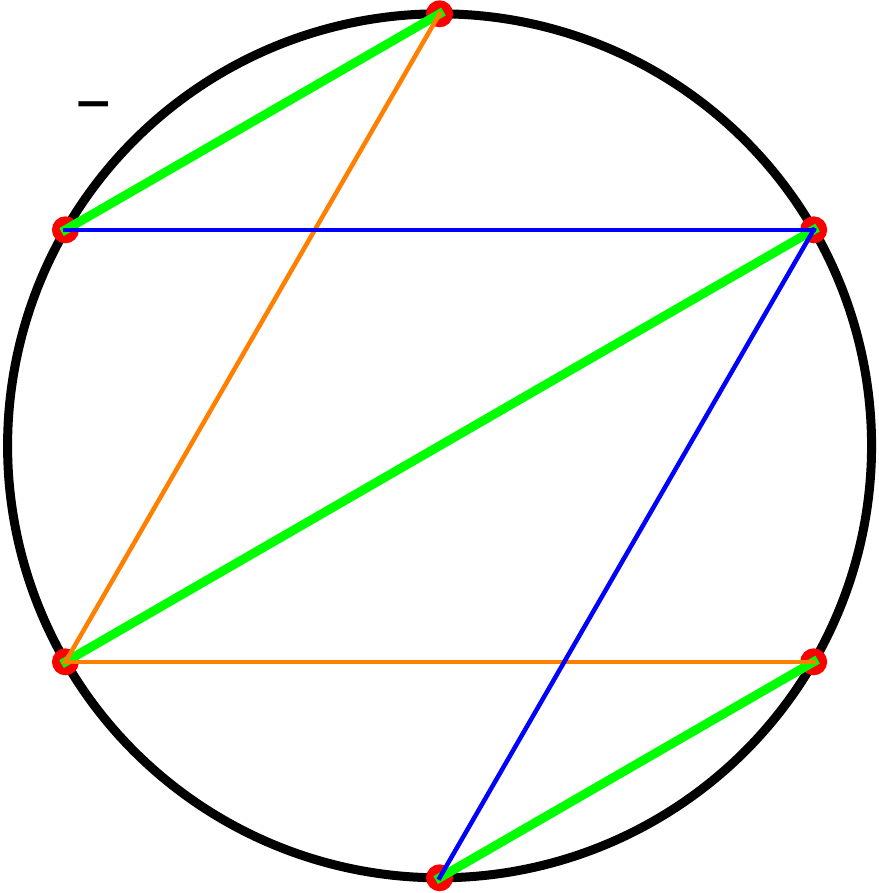}}
\put(0,38){\mbox{\scriptsize$\langle (2,0,0),(0),(0),(0)\rangle$}}
\put(35,38){\mbox{\scriptsize$\langle (1,1,0),(0),(0),(0)\rangle$}}
\put(70,38){\mbox{\scriptsize$\langle (1,0),(1,0),(0),(0)\rangle$}}
\put(105,38){\mbox{\scriptsize$\langle (1,0),(0),(1,0),(0)\rangle$}}
\put(17.5,3){\includegraphics[width=3cm]{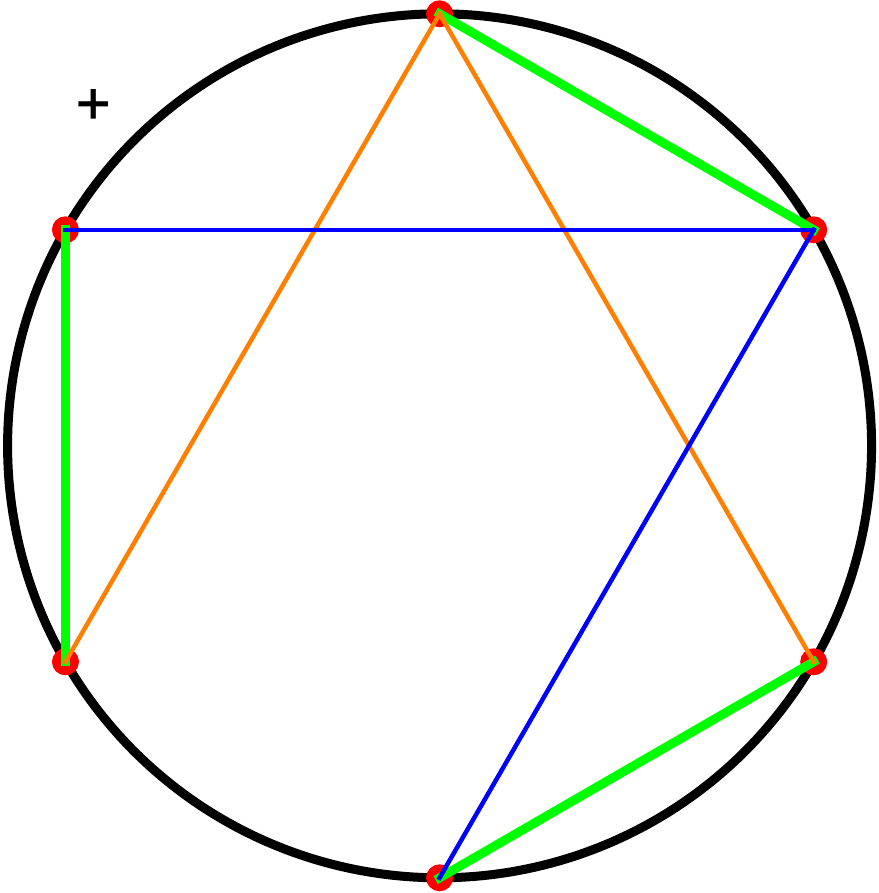}}
\put(52.5,3){\includegraphics[width=3cm]{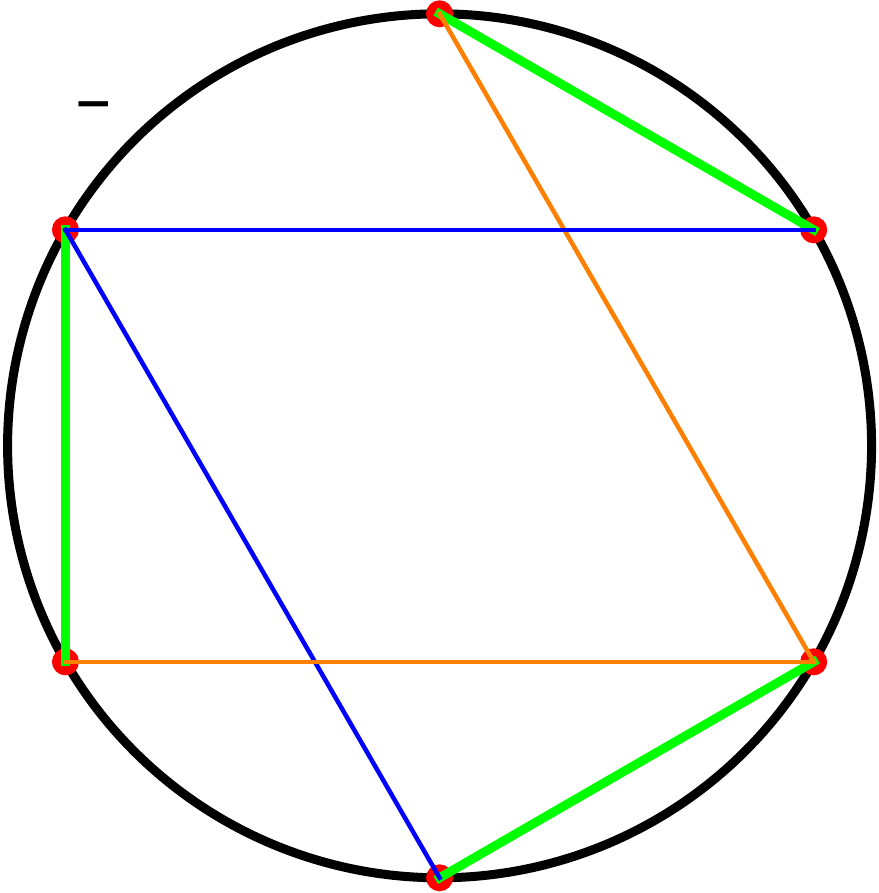}}
\put(87.5,3){\includegraphics[width=3cm]{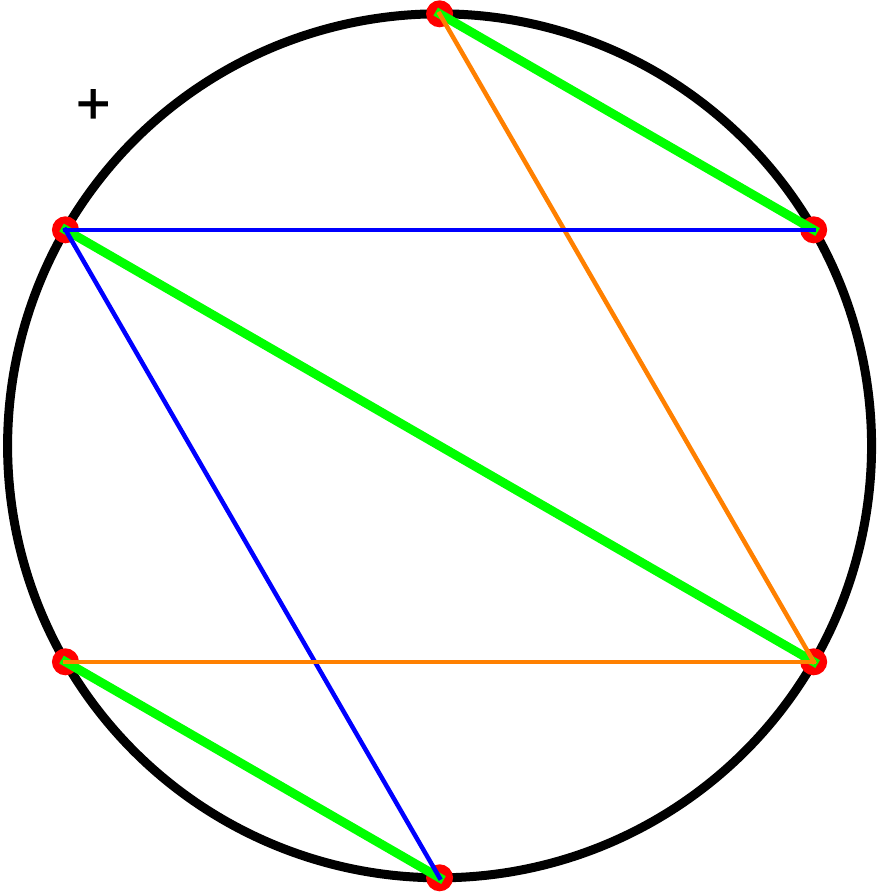}}
\put(17.5,0){\mbox{\scriptsize$\langle (0),(2,0,0),(0),(0)\rangle$}}
\put(52.5,0){\mbox{\scriptsize$\langle (0),(1,1,0),(0),(0)\rangle$}}
\put(87.5,0){\mbox{\scriptsize$\langle (0),(1,0),(1,0),(0)\rangle$}}
\end{picture}
\caption{The seven chord diagrams and Catalan tables of 
$G_{p_{0}p_{1}p_{2}p_{3}p_{4}p_{5}}$.\label{f:G6}}
\end{figure}
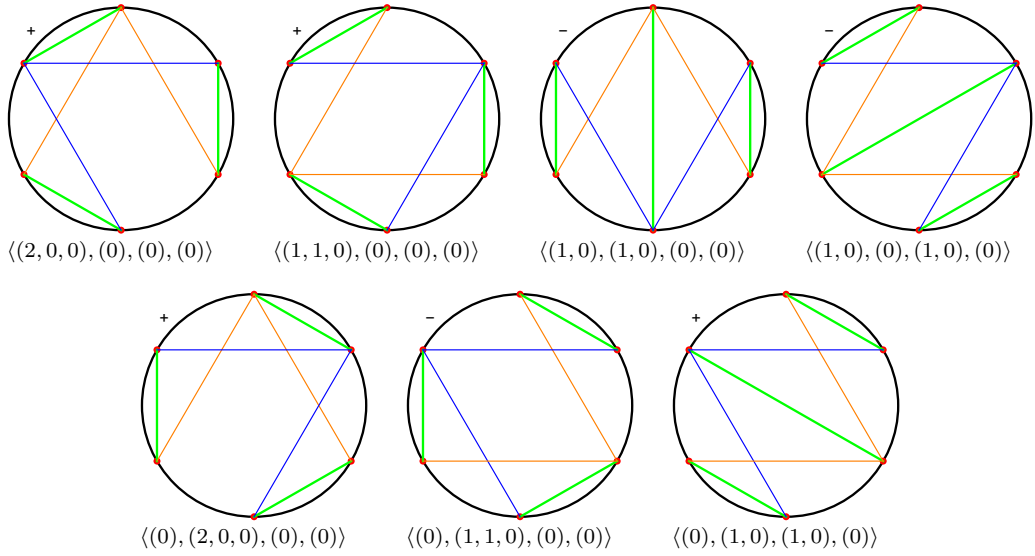
\begin{figure}[!hpt]
\begin{picture}(70,50)
\put(50,15){\includegraphics[width=5cm]{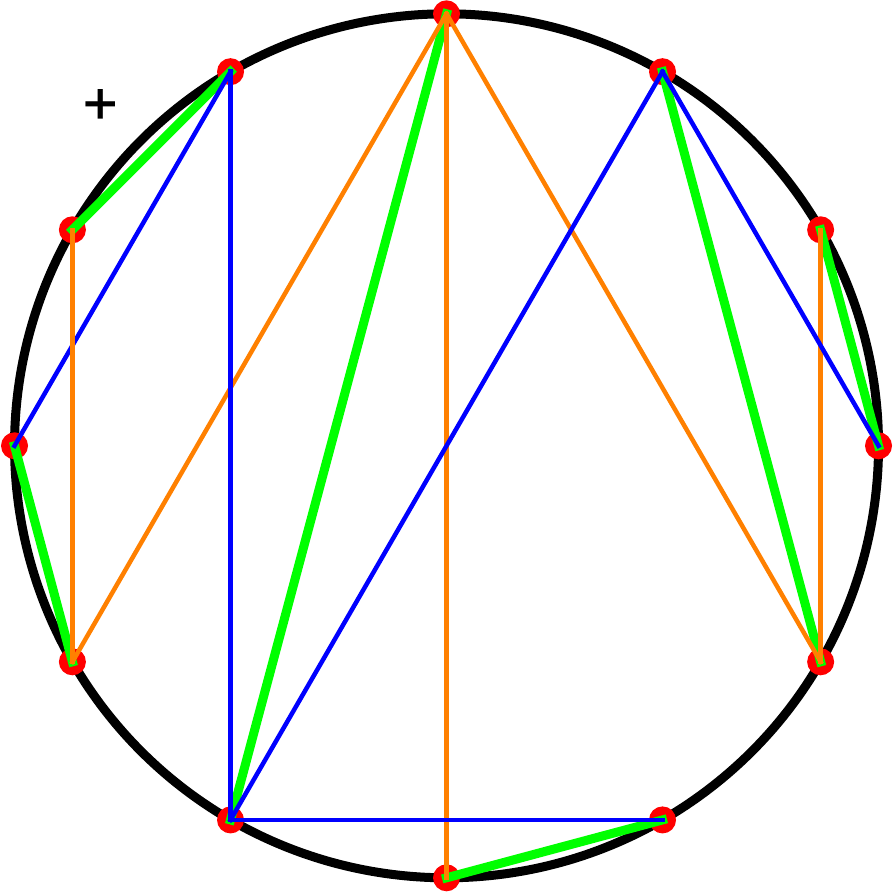}}
\put(45,5){\mbox{\footnotesize
$\langle(2,0,0),(1,1,0),(0),(0),(0),(1,0),(0)\rangle$}}
\end{picture}
\caption{A chord diagram and Catalan table contributing to 
$G_{p_0 \ldots p_{11}}$. Pocket tree and all non-trivial direct and opposite trees
have been given in \sref{Example}{ex:G12a}.\label{f:G12a}}
\end{figure}

Now that a visual way to study the recursion relation (\ref{e:rr}) has
been introduced, it is much easier to demonstrate the concepts
introduced in \sref{Secs.}{sec:CT} and \ref{sec:Ctab}.
\begin{exm}\label{exm:triangle}
  The operation $\smalllozenge\,$ is best demonstrated by an
  example:
\begin{equation*}
\langle(1,0),(0),(0)\rangle \smalllozenge 
\langle(0),(1,0),(0)\rangle=\langle(2,0,0),(0),(0),(1,0),(0)\rangle\;.
\end{equation*}
The corresponding chord diagrams are 
\[
\begin{picture}(120,32)
\put(-3,0){
\includegraphics[width=3.8cm]{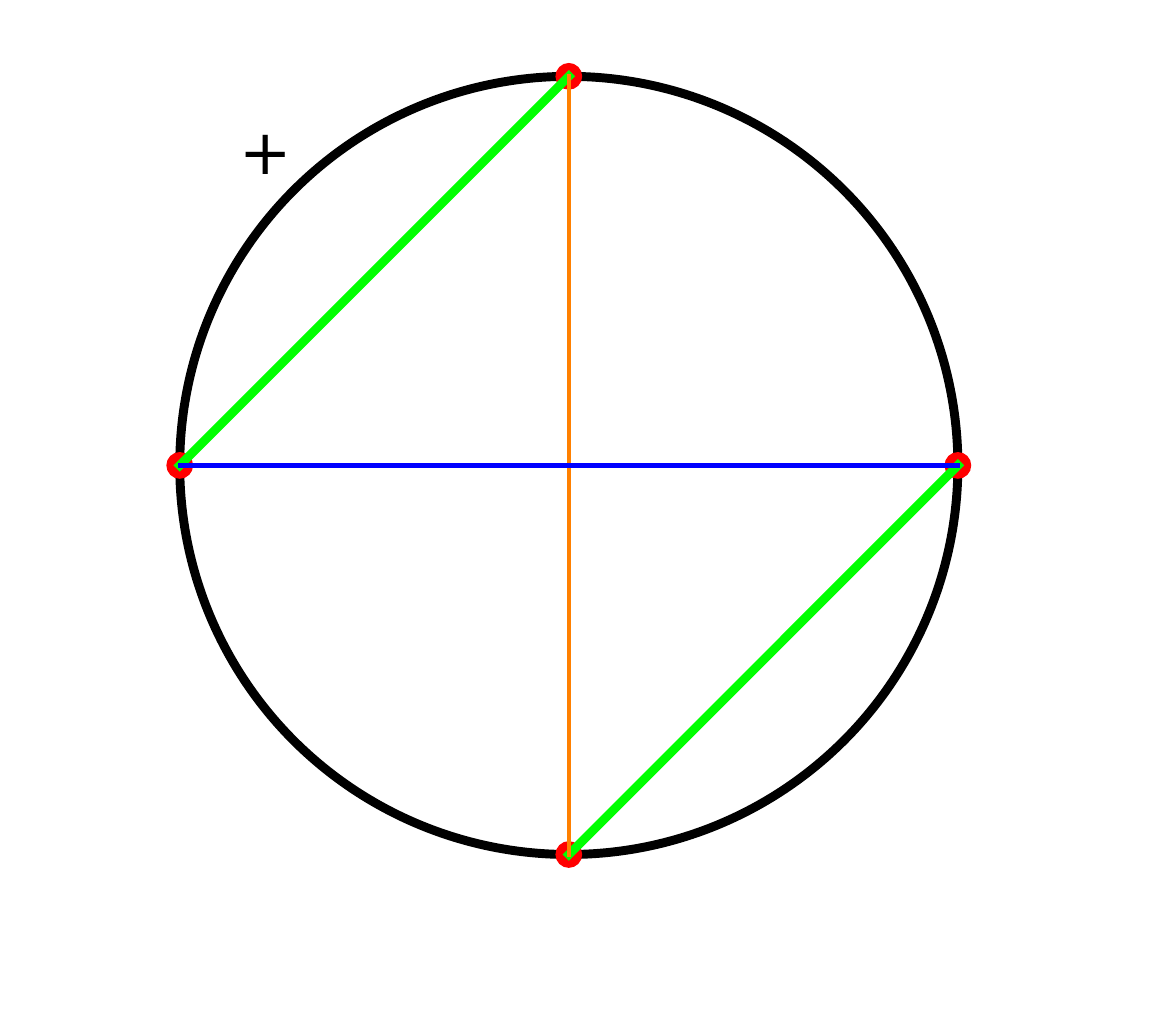}}
\put(34,16){\mbox{$\lozenge$}}
\put(40,3){\includegraphics[width=3.3cm]{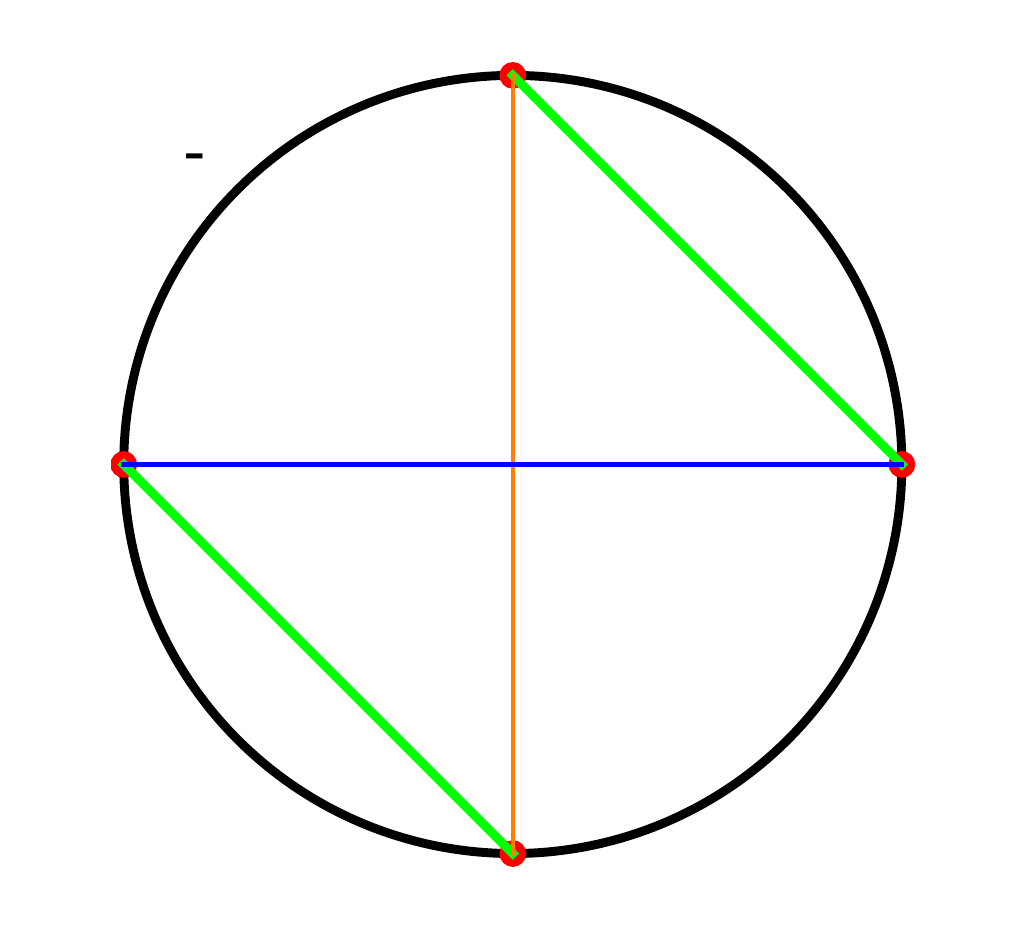}}
\put(78,16){\mbox{$=$}}
\put(87,3){\includegraphics[width=3.5cm]{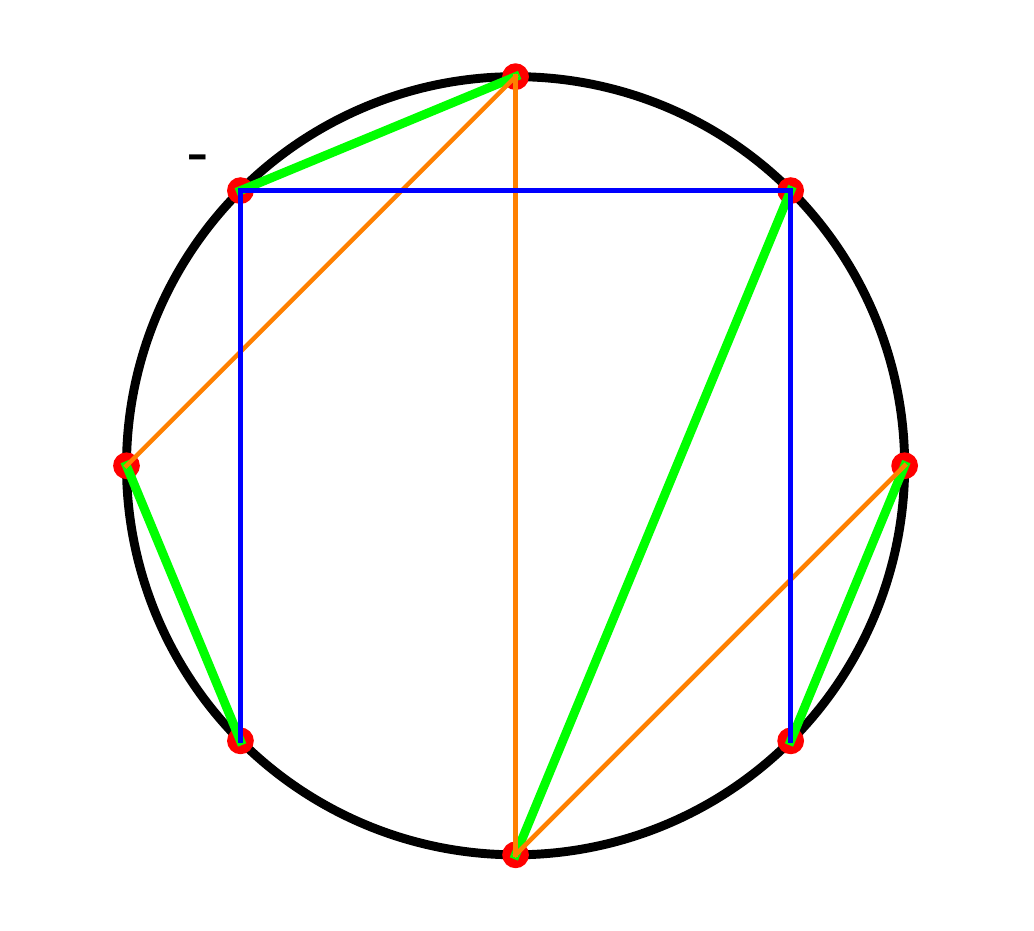}}
\put(3,0){\mbox{\scriptsize$\langle (1,0),(0),(0)\rangle$}}
\put(43,0){\mbox{\scriptsize$\langle (0),(1,0),(0)\rangle$}}
\put(85,0){\mbox{\scriptsize$\langle (2,0,0),(0),(0),(1,0),(0)\rangle$}}
\end{picture}
\]
The diagrammatic recipe is to cut both diagrams on the right side of
the designated node and paste the second into the first, where the
counterclockwise order of the nodes must be preserved. Then both
designated nodes (here $p_0,p_4$) are connected by a orange 
thread and nodes $p_{1}$
and $p_7=p_{N-1}$ by a blue thread.

To $\smalllozenge$-decompose the Catalan table $\langle
(2,0,0),(0),(0),(1,0),(0)\rangle$, we first $\circ$-factorise 
the zeroth pocket $(2,0,0)$ via (\ref{circ-factor}).
Here $\sigma_{1}\big((2,0,0)\big)=1$ and, hence, $(2,0,0)=(1,0)\circ
(0)$. Next, we evaluate the number $\hat{k}$ defined in 
(\ref{k-lozenge}). We have 
$1+|\tilde{f}^{(0)}|=1$ and
$\sigma_{1}\big((3,0,0,1,0)\big)=2$. Consequently, 
we get from \sref{Definition}{dfnt:triangle}
\begin{equation*}
\langle(2,0,0),(0),(0),(1,0),(0)\rangle=\langle (1,0),(0),(0)\rangle \smalllozenge \langle(0),(1,0),(0)\rangle\;.
\end{equation*}
\end{exm}

\begin{exm}\label{exm:box}
We employ the same example (with diagrams switched) to demonstrate 
the operation $\smallblacklozenge\,$. 
In terms of Catalan tables this becomes
\begin{equation*}
\langle(0),(1,0),(0)\rangle \smallblacklozenge 
\langle(1,0),(0),(0)\rangle=\langle(0),(2,1,0,0),(0),(0),(0)\rangle\;,
\end{equation*}
for which the chord diagrams are
\[
\begin{picture}(120,32)
\put(0,3){\includegraphics[width=3.3cm]{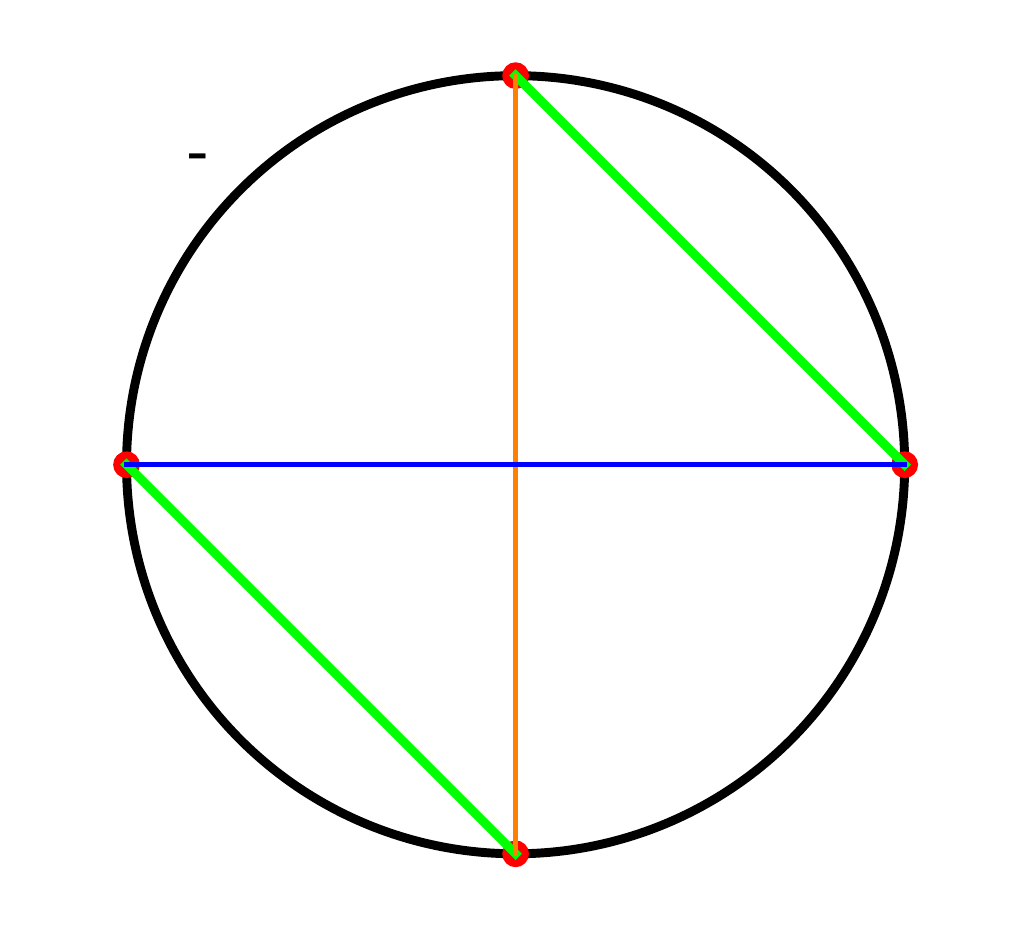}}
\put(34,16){\mbox{$\blacklozenge$}}
\put(40,3){\includegraphics[width=3.3cm]{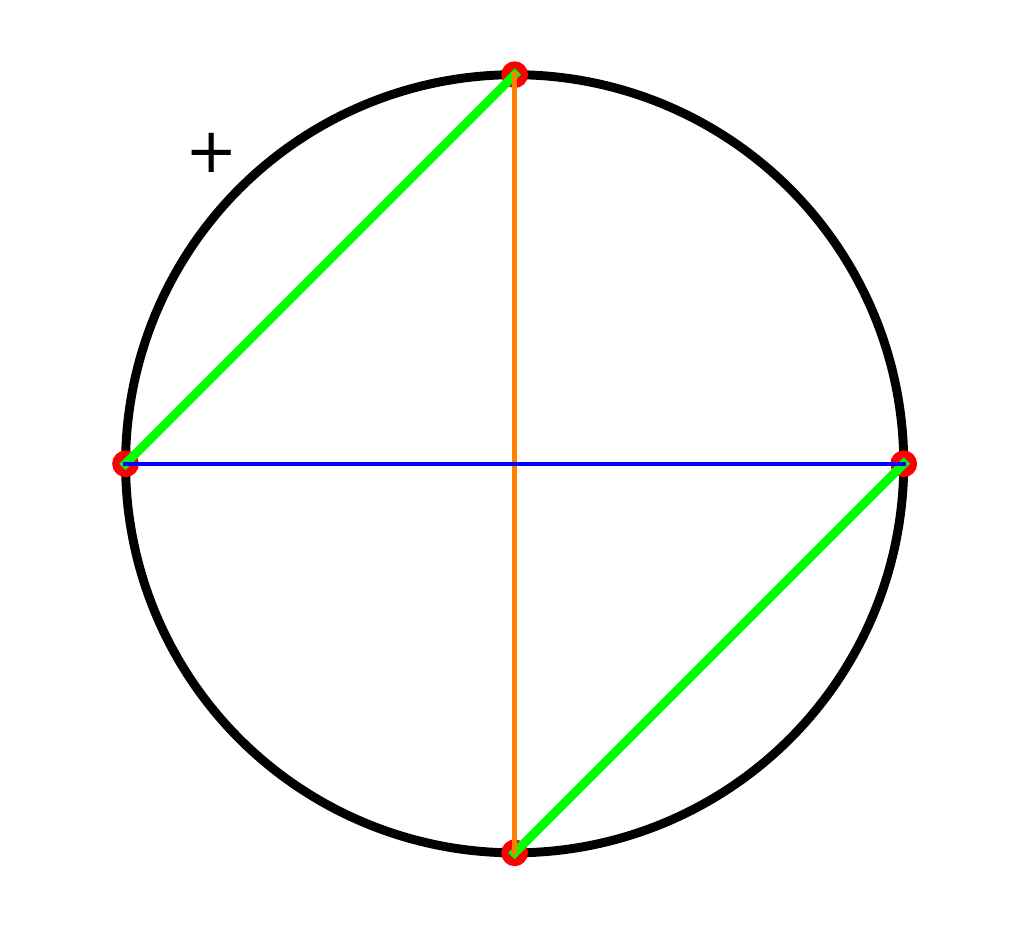}}
\put(78,16){\mbox{$=$}}
\put(87,3){\includegraphics[width=3.5cm]{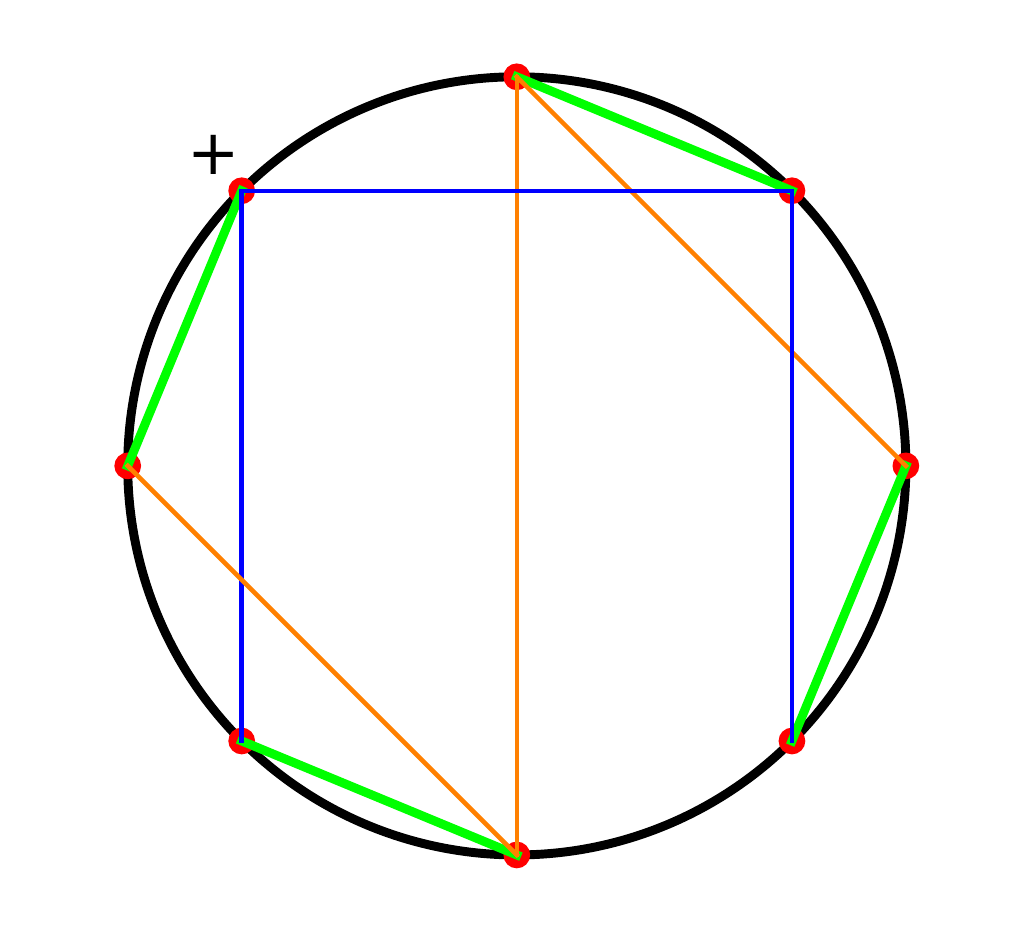}}
\put(3,0){\mbox{\scriptsize$\langle (0),(1,0),(0)\rangle$}}
\put(43,0){\mbox{\scriptsize$\langle (1,0),(0),(0)\rangle$}}
\put(85,0){\mbox{\scriptsize$\langle (0),(2,1,0,0),(0),(0),(0)\rangle$}}
\end{picture}
\]
The diagrammatic recipe is to cut the first diagram on the
left side of the designated node and the second diagram on the right
side. Then paste the second into the first,
where the counterclockwise order of the nodes must be preserved. The
threads in the second diagram switch blue and orange by doing so. Then, the
designated node of the first diagram is connected to the last node of
the second by a orange thread, the designated node of the second
diagram is connected to the last node of the first diagram by a blue
thread. 

Conversely, to $\smallblacklozenge$-decompose the Catalan table $\langle
(0),(2,1,0,0),(0),(0),(0)\rangle$, 
we first $\bullet$-factorise 
the first pocket $e^{(1)}=(2,1,0,0)$ via (\ref{bullet-factor}).
We have $e_{0}^{(1)}-1=1$, hence consider 
$\sigma_{1}\big((2,1,0,0)\big)=2$ and conclude 
$(2,1,0,0)=(1,0)\bullet (1,0)$. 
Next, we evaluate the number $\hat{l}$ in
(\ref{l-blacklozenge}). With 
$|\tilde{e}^{(0)}|+|\tilde{e}^{(1)}|+1=0+1+1=2$
the decomposition follows from $\sigma_{2}\big((1,3,0,0,0)\big)=2$ 
and yields 
\begin{equation*}
\langle (0),(2,1,0,0),(0),(0),(0)\rangle 
=\langle (0),(1,0),(0)\rangle \smallblacklozenge 
\langle(1,0),(0),(0)\rangle \;.
\end{equation*}
\end{exm}

\chapter[3-Coloured Model with Cubic Interaction]
{3-Coloured Model with Cubic Interaction\footnote{This is a summary of our paper 
\cite{Hock:2018wup}}}\label{App:3C}
This appendix will consider a different type of matrix field theory which can be seen as multi-matrix field theory,
studied in our paper \cite{Hock:2018wup}.
We assume three different matrices $\Phi^a$ with the ''colours'' $a\in\{1,2,3\}$. The interaction is a cubic interaction such that the 
1-point function does not exist. Accordingly, the action is given by
\begin{align}\label{actionmatrix}
 S[\Phi]&=V\left(\sum_{a=1}^3\sum_{n,m=0}^\mathcal{N}\frac{H_{nm}}{2}\Phi^a_{nm}\Phi^a_{mn}+\frac{\lambda}{3}
		\sum_{a,b,c=1}^3\sum_{n,m,l=0}^\mathcal{N}\sigma_{abc}\Phi^a_{nm}
		\Phi^b_{ml}\Phi^c_{ln}\right)\\
		H_{nm}&:=E_n+E_m,
\end{align}
where $(\Phi^a_{nm})$ are Hermitian matrices and $\sigma_{abc}=1$ 
for $a\neq b\neq c\neq a$ and $\sigma_{abc}=0$ else.

We demonstrate the techniques to determine
correlation functions from the partition function for a coloured model.  The partition
function $\mathcal{Z}[J]$ of the 3-colour model with
external Hermitian matrices $\left(J^a_{nm}\right)$ and
$a\in\{1,2,3\}$ is formally defined by
\begin{align}\label{Zustandssumme}
	\mathcal{Z}[J]:=&\int \left(\prod_{a=1}^{3}D\Phi^a\right)
	\exp\left(-S[\Phi]+V\sum_{a=1}^{3}\sum_{n,m=0}^\mathcal{N}J^a_{nm}\Phi^a_{mn}\right)\\
	=&K\,\exp\left(-\frac{\lambda}{3V^2}\sum_{a,b,c=1}^3\sum_{n,m,l=0}^\mathcal{N}\sigma_{abc}\frac{\partial^3}{\partial J^a_{nm}\partial J^b_{ml}\partial J^c_{ln}}
	\right)\mathcal{Z}_{free}[J],\nonumber\\
	 \mathcal{Z}_{free}[J]:=&\exp\left(\sum_{a=1}^3\sum_{n,m=0}^\mathcal{N}\frac{V}{2H_{nm}}J^a_{nm}J^a_{mn}\right),\\
	K:=&\int \left(\prod_{a=1}^{3} D\Phi^a\right)\exp\left(-\sum_{a=1}^3\sum_{n,m=0}^\mathcal{N}\frac{VH_{nm}}{2}\Phi^a_{nm}\Phi^a_{mn}\right).\nonumber
\end{align}
The logarithm of $\mathcal{Z}[J]$ will be expanded into a series of
moments with different number $b$ of boundary components. The sources are cyclic
within every boundary $\beta\in \{1,...,b\}$. For simplification we use
the notation
$\mathbb{J}_{p^\beta_1...p^\beta_{N_\beta}}^{a^\beta_1...a^\beta_{N_\beta}}:=\prod_{i=1}^{N_\beta}
J_{p_i^\beta p_{i+1}^\beta}^{a_i^\beta}$ with $N_\beta+1\equiv 1$. The
correlation functions are then defined by
\begin{align}\label{Entwicklungskoeffizienten}
	\log\frac{\mathcal{Z}[J]}{\mathcal{Z}[0]}=:
	\sum_{b=1}^\infty \sum_{ N_1,..., N_b=1}^\infty 
	\sum_{p_1^1,...,p_{N_b}^b=0}^\mathcal{N}
	\sum_{a_1^1,...,a_{N_b}^b=1}^{3}\!\!\!\!\!\!V^{2-b}\frac{G^{a_1^1...a^1_{N_1}|...
	|a_1^b...a^B_{N_b}}_{|p_1^1...p^1_{N_1}|...|p_1^b...p^b_{N_b}|}}{b!}
	\prod_{\beta=1}^b\frac{\mathbb{J}_{p^\beta_1...p^\beta_{N_\beta}}^{a^\beta_1...a^\beta_{N_\beta}}}{N_\beta}.
\end{align}

Due to the vanishing 1-point function for the 3-colour model, the partition function can be expanded 
with (\ref{Entwicklungskoeffizienten}) to 
\begin{align}\label{partitionexpansion}
	 &\frac{\mathcal{Z}[J]}{\mathcal{Z}[0]}=1+
	\sum_{a,b=1}^3\sum_{n,m=0}^\mathcal{N}\left(\frac{V}{2}G^{ab}_{|nm|}\mathbb{J}^{ab}_{nm}+
	\frac{1}{2}G^{a|b}_{|n|m|}\mathbb{J}^a_{n}\mathbb{J}^b_{m}\right)\\
	 &+\sum_{a,b,c=1}^3\sum_{n,m,l=0}^\mathcal{N}\Bigg(\frac{V}{3}G^{abc}_{|nml|}\mathbb{J}^{abc}_{nml}
	 +\frac{1}{2}G^{a|bc}_{|n|ml|}\mathbb{J}^a_{n}\mathbb{J}^{bc}_{ml}+\frac{1}{6V}G^{a|b|c}_{|n|m|l|}\mathbb{J}^a_{n}
	 \mathbb{J}^b_{m}\mathbb{J}^c_{l}\Bigg)\nonumber\\
	 &+\sum_{a,b,c,d=1}^3\sum_{n,m,l,p=0}^\mathcal{N}\Bigg(\frac{V}{4}G^{abcd}_{|nmlp|}\mathbb{J}^{abcd}_{nmlp}
	 +\frac{1}{3}G^{a|bcd}_{|n|mlp|}\mathbb{J}^a_{n}\mathbb{J}^{bcd}_{mlp}\nonumber\\
	 &+\left(\frac{1}{8}G^{ab|cd}_{|nm|lp|}+\frac{V^2}{8}G^{ab}_{|nm|}G^{cd}_{|lp|}\right)\mathbb{J}^{ab}_{nm}\mathbb{J}^{cd}_{lp}
	 +\left(\frac{1}{4V}G^{a|b|cd}_{|n|m|lp|}+\frac{V}{4}G^{a|b}_{|n|m|}G^{cd}_{|lp|}\right)\mathbb{J}^a_{n}
	 \mathbb{J}^b_{m}\mathbb{J}^{cd}_{lp}\nonumber\\
	 &+\left(\frac{1}{24V^2}G^{a|b|c|d}_{|n|m|l|p|}+\frac{1}{8}G^{a|b}_{|n|m|}G^{c|d}_{|l|p|}\right)\mathbb{J}^a_{n}
	 \mathbb{J}^b_{m}\mathbb{J}^c_{l}\mathbb{J}^d_{p}\Bigg)+\dots\nonumber\quad .
\end{align}

The calculation rule for later purpose is
\begin{align*}
	\frac{\partial}{\partial J^a_{p_1p_2}}J^b_{p_3p_4}=\delta_{ab}\delta_{p_1p_3}\delta_{p_2p_4}+J^b_{p_3p_4}\frac{\partial}{\partial J^a_{p_1p_2}}.
\end{align*}

\section{Ward-Takahashi Identity}
The Ward-Takahashi identity is obtained by the requirement of invariance of $\mathcal{Z}[J]$ under inner automorphisms
(see \sref{Proposition}{Prop:WardId}).
For a colour model, we choose a transformation as follows:
$\phi^a \mapsto (\phi^a)'=U^\dagger \phi^a U$ for $U\in \mathrm{U}(\mathcal{N})$ for \textit{one} colour $a\in \{1,2,3\}$.
The Ward-Takahashi identity following from this transformation for $E_{p_1}\neq E_{p_2}$ is given by
\begin{align}\label{Ward1}
	\sum_{m=0}^\mathcal{N}&\frac{\partial^2}{\partial J^a_{p_1m}\partial J^a_{mp_2}}\mathcal{Z}[J]+\frac{V}{E_{p_1}-E_{p_2}}
	\sum_{m=0}^\mathcal{N}\left(J^a_{p_2m}\frac{\partial}{\partial J^a_{p_1m}}-J^a_{mp_1}\frac{\partial}{\partial J^a_{mp_2}}\right)\mathcal{Z}[J]\\
	=&\frac{\lambda}{V(E_{p_1}-E_{p_2})}\sum_{m,n=0}^\mathcal{N}\sum_{b,c=1}^3\sigma_{abc}\left(\frac{\partial^3}{\partial J^a_{p_1m}\partial J^b_{mn}\partial J^c_{np_2}}
	-\frac{\partial^3}{\partial J^b_{p_1m}\partial J^c_{mn}\partial J^a_{np_2}}
	\right)\mathcal{Z}[J].\nonumber
\end{align}
The interaction terms are not invariant under the 
transformation of only one colour. However, the sum over all colours in (\ref{Ward1}) gives
\begin{align}\label{Ward2}
	\sum_{a=1}^3\sum_{m=0}^\mathcal{N}\frac{\partial^2}{\partial J^a_{p_1m}\partial J^a_{mp_2}}\mathcal{Z}[J]=
	\frac{V}{(E_{p_1}-E_{p_2})}\sum_{a=1}^3\sum_{m=0}^\mathcal{N}\left(J^a_{mp_1}\frac{\partial}{\partial J^a_{mp_2}}-J^a_{p_2m}\frac{\partial}{\partial J^a_{p_1m}}\right)\mathcal{Z}[J],
\end{align}
which has the usual form of a Ward-Takahashi identity (see \sref{Proposition}{Prop:WardId}). 
Equation (\ref{Ward2}) shows that the interaction term is invariant under the simultaneous 
transformation of all three colours. 

A more general identity plays the crucial r$\hat{\text{o}}$le (analog to \sref{Proposition}{Prop:GenWardId}):

\begin{prps}\label{WardPorp2}
	Let $E_{p_1}\neq E_{p_2}$. The generalised Ward-Takahashi identity for the 3-colour matrix model with an external field
  $E$ is  
	\begin{align*}
		\sum_{m=0}^\mathcal{N}&\frac{\partial^2}{\partial J^a_{p_1m}\partial J^b_{mp_2}}\mathcal{Z}[J]+\frac{V}{E_{p_1}-E_{p_2}}
		\sum_{m=0}^\mathcal{N}\left(J^b_{p_2m}\frac{\partial}{\partial J^a_{p_1m}}-J^a_{mp_1}\frac{\partial}{\partial J^b_{mp_2}}\right)\mathcal{Z}[J]\\
		=&\frac{\lambda}{V(E_{p_1}-E_{p_2})}\sum_{m,n=0}^\mathcal{N}\sum_{c,d=1}^3\left(\sigma_{bcd}\frac{\partial^3}{\partial J^a_{p_1m}\partial J^c_{mn}\partial J^d_{np_2}}
		-\sigma_{acd}\frac{\partial^3}{\partial J^c_{p_1m}\partial J^d_{mn}\partial J^b_{np_2}}
		\right)\mathcal{Z}[J].
	\end{align*}
	\begin{proof}
		Let $S_{int}[\Phi]= V\frac{\lambda}{3}
		\sum_{a,b,c=1}^3\sum_{n,m,l=0}^\mathcal{N}\sigma_{abc}\Phi^a_{nm}
		\Phi^b_{ml}\Phi^c_{ln}$ be the interaction term of the action.
		Direct computation gives then
		\begin{align*}
			&\frac{E_{p_1}-E_{p_2}}{V}\sum_{m=0}^\mathcal{N}\frac{\partial^2}{\partial J^a_{p_1m}\partial J^b_{mp_2}}\mathcal{Z}[J]\\
			=&\frac{1}{V}\sum_{m=0}^\mathcal{N}\frac{\partial^2}{\partial J^a_{p_1m}\partial J^b_{mp_2}}\left((E_{p_1}+E_m)-(E_m+E_{p_2})\right)\mathcal{Z}[J]\\
			=&K\sum_{m=0}^\mathcal{N}\Bigg\{\frac{\partial}{\partial J^b_{mp_2}}\exp\left(-S_{int}\left[\frac{1}{V}\frac{\partial}{\partial J}\right]\right)J^a_{mp_1}\\
			&\qquad \qquad-
			\frac{\partial}{\partial J^a_{p_1m}}\exp\left(-S_{int}\left[\frac{1}{V}\frac{\partial}{\partial J}\right]\right)J^b_{p_2m}\Bigg\}\mathcal{Z}_{free}[J]\\
			=&\sum_{m=0}^\mathcal{N}\left(J^a_{mp_1}\frac{\partial}{\partial J^b_{mp_2}}-J^b_{p_2m}\frac{\partial}{\partial J^a_{p_1m}}\right)\mathcal{Z}[J]\\
			&-\frac{\lambda}{V^2}\sum_{m,n=0}^\mathcal{N}\sum_{c,d=1}^3\left(\sigma_{acd}\frac{\partial^3}{\partial J^c_{p_1n}\partial J^d_{nm}\partial J^b_{mp_2}}-
			\sigma_{bcd}\frac{\partial^3}{\partial J^a_{p_1m}\partial J^c_{mn}\partial J^d_{np_2}}\right)\mathcal{Z}[J].
\end{align*}
We have used the second form of $\mathcal{Z}[J]$ in
(\ref{Zustandssumme}) and the Leibniz rule in the last
step. Technically, one expands the exponential function and resums after
using the Leibniz rule. Since $E_{p_1}\neq E_{p_2}$ the proof is
finished.
	\end{proof}
\end{prps}
\noindent
Equation (\ref{Ward1}) is a special case of
\sref{Proposition}{WardPorp2} by setting $b=a$.  The derivation of both identities
is completely different. \sref{Proposition}{WardPorp2} cannot be obtained 
by a symmetry transformation of only one colour
due to the discrete mixing of the colours if $a\neq b$.
Applying the procedure of the proof of \sref{Proposition}{WardPorp2}, it
is also possible to derive the usual Ward-Takahashi identity even in
other models.

For later purpose, we combine two identities to get a more useful expression:
\begin{lemma}\label{WarsLemma}
	Let $a$ be fixed and $E_{p_1}\neq E_{p_2}$, then it follows
	\begin{align*}
		&\sum_{b,c=1}^{3}\sum_{m=0}^\mathcal{N}\sigma_{abc}\frac{\partial^2}{\partial J^b_{p_1m}\partial J^c_{mp_2}}\mathcal{Z}[J]
		\\
		=&\frac{V}{E_{p_1}-E_{p_2}}\Bigg[\sum_{b,c=1}^{3}\sigma_{abc}
		\sum_{m=0}^\mathcal{N}\left(J^b_{mp_1}\frac{\partial}{\partial J^c_{mp_2}}-J^c_{p_2m}\frac{\partial}{\partial J^b_{p_1m}}\right)\\
		&\qquad\qquad+\frac{\lambda}{V^2}\sum_{b=1}^3\Bigg\{\sum_{m=0}^\mathcal{N}\left(\frac{\partial^3}{\partial J^b_{p_1m}\partial J^b_{mp_1}\partial J^a_{p_1p_2}}-\frac{\partial^3}{\partial J^a_{p_1p_2}\partial J^b_{p_2m}\partial J^b_{mp_2}}\right)\\
		&\qquad\qquad+\sum_{\substack{m,n=0 \\ n\neq p_1}}^{\mathcal{N}}\frac{V}{E_{p_1}-E_n}\frac{\partial}{\partial J^a_{np_2}}\left(J^b_{mp_1}\frac{\partial}{\partial J^b_{mn}}-J^b_{nm}\frac{\partial}{\partial J^b_{p_1m}}\right)\\
		&\qquad\qquad-\sum_{\substack{m,n=0 \\ n\neq p_2}}^{\mathcal{N}}\frac{V}{E_{p_2}-E_{n}}\frac{\partial}{\partial J^a_{p_1n}}\left(J^b_{p_2m}\frac{\partial}{\partial J^b_{nm}}-J^b_{mn}\frac{\partial}{\partial J^b_{mp_2}}\right)\Bigg\}\Bigg]\mathcal{Z}[J].
	\end{align*}
	\begin{proof}
Inserting \sref{Proposition}{WardPorp2} for the lhs yields 
		\begin{align}\label{GLLemmaWard}
			&\sum_{b,c=1}^{3}\sum_{m=0}^\mathcal{N}\sigma_{abc}\frac{\partial^2}{\partial J^b_{p_1m}\partial J^c_{mp_2}}\mathcal{Z}[J]\nonumber\\
			=&\frac{V}{E_{p_1}-E_{p_2}}\sum_{b,c=1}^{3}\sigma_{abc}
			\sum_{m=0}^\mathcal{N}\left(J^b_{mp_1}\frac{\partial}{\partial J^c_{mp_2}}-J^c_{p_2m}\frac{\partial}{\partial J^b_{p_1m}}\right)\mathcal{Z}[J]\\
			+&\frac{\lambda}{V(E_{p_1}-E_{p_2})}\!\!\sum_{m,n=0}^\mathcal{N}\sum_{b,c,d,e=1}^3\!\!\!\!\!\sigma_{abc}
			\!\left(\! \sigma_{cde}\frac{\partial^3}{\partial J^b_{p_1m}\partial J^d_{mn}\partial J^e_{np_2}}-\sigma_{bde}\frac{\partial^3}{\partial J^d_{p_1m}\partial J^e_{mn}\partial J^c_{np_2}}\nonumber
			\right)\!\mathcal{Z}[J].
\end{align} 
By the sum over the colours $b,c,d,e$, we obtain for the
multiplication of two $\sigma$'s with one common index
\begin{align*}
			\sigma_{abc}\sigma_{cde}=&\sigma_{abc}(\delta_{ad}\delta_{be}+\delta_{ae}\delta_{bd})\\
			\sigma_{abc}\sigma_{bde}=&\sigma_{abc}(\delta_{ad}\delta_{ce}+\delta_{ae}\delta_{cd}).
		\end{align*}
Therefore, the last line in (\ref{GLLemmaWard}) gives 
\begin{align}\label{proofGL1}
			&\frac{\lambda}{V(E_{p_1}-E_{p_2})}\sum_{m,n=0}^\mathcal{N}\sum_{b,c=1}^3\sigma_{abc}\Bigg(\frac{\partial^3}{\partial J^b_{p_1m}\partial J^a_{mn}\partial J^b_{np_2}}+\frac{\partial^3}{\partial J^b_{p_1m}\partial J^b_{mn}\partial J^a_{np_2}}\\
			&\qquad\qquad\qquad\qquad\qquad\qquad\qquad\quad-\frac{\partial^3}{\partial J^a_{p_1m}\partial J^c_{mn}\partial J^c_{np_2}}-\frac{\partial^3}{\partial J^c_{p_1m}\partial J^a_{mn}\partial J^c_{np_2}}\Bigg)\mathcal{Z}[J].\nonumber
\end{align}
The first and the last term in parentheses vanish because of the
total symmetry of $\sigma_{abc}$. Adding
$0=\left(\frac{\partial^3}{\partial J^a_{p_1m}\partial
    J^a_{mn}\partial J^a_{np_2}}-\frac{\partial^3}{\partial
    J^a_{p_1m}\partial J^a_{mn}\partial
    J^a_{np_2}}\right)\mathcal{Z}[J]$ and renaming the indices,
(\ref{proofGL1}) can be rewritten to
\begin{align*}
			\frac{\lambda}{V(E_{p_1}-E_{p_2})}\sum_{m,n=0}^\mathcal{N}\sum_{b=1}^3\left(\frac{\partial^3}{\partial J^b_{p_1m}\partial J^b_{mn}\partial J^a_{np_2}}-\frac{\partial^3}{\partial J^a_{p_1m}\partial J^b_{mn}\partial J^b_{np_2}}\right)\mathcal{Z}[J].
		\end{align*}
		Inserting (\ref{Ward2}) for $E_n\neq E_{p_1}$ 
		in the first and $E_m\neq E_{p_2}$ in the second term finally gives after renaming indices
		\begin{align}\label{GLLemmaWard2}
			\frac{\lambda}{V(E_{p_1}-E_{p_2})}\sum_{b=1}^3&\Bigg\{\sum_{m=0}^\mathcal{N}\left(\frac{\partial^3}{\partial J^b_{p_1m}\partial J^b_{mp_1}\partial J^a_{p_1p_2}}-\frac{\partial^3}{\partial J^a_{p_1p_2}\partial J^b_{p_2m}\partial J^b_{mp_2}}\right)\nonumber\\*
			&+\sum_{\substack{m,n=0 \\ n\neq p_1}}^{\mathcal{N}}\frac{V}{E_{p_1}-E_n}\frac{\partial}{\partial J^a_{np_2}}\left(J^b_{mp_1}\frac{\partial}{\partial J^b_{mn}}-J^b_{nm}\frac{\partial}{\partial J^b_{p_1m}}\right)\\*
			&-\sum_{\substack{m,n=0 \\ n\neq p_2}}^{\mathcal{N}}\frac{V}{E_{p_2}-E_{n}}\frac{\partial}{\partial J^a_{p_1n}}\left(J^b_{p_2m}\frac{\partial}{\partial J^b_{nm}}-J^b_{mn}\frac{\partial}{\partial J^b_{mp_2}}\right)\Bigg\}\mathcal{Z}[J]\nonumber.
		\end{align}
		The identity follows by combining (\ref{GLLemmaWard}) and (\ref{GLLemmaWard2}).
	\end{proof}
\end{lemma}

\section{Schwinger-Dyson Equations for $b=1$}
In this section we derive the SDEs with the help
of Ward-Takahashi identity.
\begin{prps}\label{Prop2Punkt}
 The SDE for the 2-point function in the 3-colour matrix model with an external field
 $E$ is for $E_{p_1}\neq E_{p_2}$ given by
 \begin{align*}
  &G^{aa}_{|p_1p_2|}=\frac{1}{H_{p_1p_2}}+\frac{\lambda^2}{(E^2_{p_1}-E^2_{p_2})V}\\
  &\times\Bigg[\sum_{m=0}^\mathcal{N}\sum_{b=1}^3\left(G^{aa}_{|p_1p_2|}\left(
 G^{bb}_{|p_2m|}
-G^{bb}_{|p_1m|}\right)+\frac{1}{V}\left(G^{aabb}_{|p_2p_1p_2m|}-G^{aabb}_{|p_1p_2p_1m|}\right)\right)\\
&+\sum_{b=1}^3\frac{1}{V^2}\bigg(\sum_{m=0}^\mathcal{N}\left(G^{aa|bb}_{|p_2p_1|p_2m|}-G^{aa|bb}_{|p_1p_2|p_1m|}\right)
+\left(G^{b|baa}_{|p_2|p_2p_2p_1|}-G^{b|baa}_{|p_1|p_1p_1p_2|}\right)\bigg)\\
&+\sum_{b=1}^3\left(\frac{1}{V^3}\left(G^{b|b|aa}_{|p_2|p_2|p_2p_1|}-G^{b|b|aa}_{|p_1|p_1|p_1p_2|}\right)
+\frac{1}{V}G^{aa}_{|p_1p_2|}\left(G^{b|b}_{|p_2|p_2|}-G^{b|b}_{|p_1|p_1|}\right)\right)\\
&+\sum_{\substack{m=0\\ m\neq p_2}}^\mathcal{N}\frac{G^{aa}_{|p_1m|}-G^{aa}_{|p_1p_2|}}{E_{p_2}-E_m}-
\sum_{\substack{m=0\\ m\neq p_1}}^\mathcal{N}
\frac{G^{aa}_{|p_1p_2|}-G^{aa}_{|p_2m|}}{E_{m}-E_{p_1}}+\frac{1}{V}\frac{G^{a|a}_{|p_1|p_1|}-2G^{a|a}_{|p_1|p_2|}+G^{a|a}_{|p_2|p_2|}}{E_{p_2}-E_{p_1}}\Bigg].
 \end{align*}
\begin{proof}
 Assuming $E_{p_1}\neq E_{p_2}$ the 2-point function is given via 
definition (\ref{Entwicklungskoeffizienten}) and expansion
 (\ref{partitionexpansion}). Using (\ref{Zustandssumme}) leads to
 \begin{align*}
  G^{aa}_{|p_1p_2|}
  =&\frac{1}{V}\frac{\partial^2}{\partial J^a_{p_1p_2}\partial J^a_{p_2p_1}}\mathrm{log}\mathcal{Z}[J]\Big\vert_{J=0}
  =\frac{1}{V\mathcal{Z}[0]}\frac{\partial^2}{\partial J^{a}_{p_1p_2}\partial J^{a}_{p_2p_1}}\mathcal{Z}[J]\Big|_{J=0}\\
  =&\frac{K}{H_{p_1p_2}\mathcal{Z}[0]}\frac{\partial}{\partial J^{a}_{p_2p_1}}\exp\left(-S_{int}\left[\frac{1}{V}\frac{\partial}{\partial J}\right]\right)
  J^{a}_{p_2p_1}\mathcal{Z}_{free}[J]\Big|_{J=0}\\
  =&\frac{1}{H_{p_1p_2}}-\frac{\lambda }{H_{p_1p_2}\mathcal{Z}[0]V^2}\frac{\partial}{\partial J^{a}_{p_2p_1}}
 \sum_{b,c=1}^3\sum_{m=0}^\mathcal{N}\sigma_{abc}\frac{\partial^2}{\partial J^b_{p_1m}\partial J^c_{mp_2}}\mathcal{Z}[J]\Big|_{J=0}.
 \end{align*}
 Inserting the expansion of (\ref{partitionexpansion}) would give the
 SDE between the 2-point and 3-point function.
 At first sight, the application of 
 \sref{Lemma}{WarsLemma} seems to make the equation more complicated.  However,
 it yields a better behaviour in the $V$-expansion. The
 first term on the rhs of the equation of \sref{Lemma}{WarsLemma}
 vanishes by setting $J$ to zero.  Therefore, we obtain
 \begin{align*}
  =&\frac{1}{H_{p_1p_2}}-\frac{\lambda^2 }{(E^2_{p_1}-E^2_{p_2})\mathcal{Z}[0]V^3}\\
  &\times\Bigg\{
  \sum_{b=1}^3\sum_{m=0}^\mathcal{N}\left(\frac{\partial^4}{\partial J^b_{p_1m}\partial J^b_{mp_1}\partial J^a_{p_1p_2}\partial J^a_{p_2p_1}}-
  \frac{\partial^4}{\partial J^a_{p_2p_1}\partial J^a_{p_1p_2}\partial J^b_{p_2m}\partial J^b_{mp_2}}\right)\mathcal{Z}[J]\Big|_{J=0}\\
  &\qquad+\sum_{\substack{m,n=0 \\ n\neq p_1}}^{\mathcal{N}}\frac{V}{E_{p_1}-E_n}\frac{\partial^2}{\partial J^a_{p_2p_1}\partial J^a_{np_2}}\left(J^a_{mp_1}\frac{\partial}{\partial J^a_{mn}}-
  J^a_{nm}\frac{\partial}{\partial J^a_{p_1m}}\right)\mathcal{Z}[J]\Big|_{J=0}\\
  &\qquad-\sum_{\substack{m,n=0 \\ n\neq p_2}}^{\mathcal{N}}\frac{V}{E_{p_2}-E_{n}}\frac{\partial^2}{\partial J^a_{p_2p_1}\partial J^a_{p_1n}}\left(J^a_{p_2m}\frac{\partial}{\partial J^a_{nm}}
  -J^a_{mn}\frac{\partial}{\partial J^a_{mp_2}}\right)
  \mathcal{Z}[J]\Big|_{J=0}\Bigg\},
 \end{align*}
  where $H_{p_1p_2}(E_{p_1}-E_{p_2})=(E^2_{p_1}-E^2_{p_2})$ has been used and 
  the fact that in the last two lines only colour $a$ survives. 
  By taking $E_{p_1}\neq E_{p_2}$ into account and $J=0$ gives with the Leibniz rule
  \begin{align*}
    =&\frac{1}{H_{p_1p_2}}-\frac{\lambda^2 }{(E^2_{p_1}-E^2_{p_2})\mathcal{Z}[0]V^3}\\
  &\times\Bigg\{
  \sum_{b=1}^3\sum_{m=0}^\mathcal{N}\left(\frac{\partial^4}{\partial J^b_{p_1m}\partial J^b_{mp_1}\partial J^a_{p_1p_2}\partial J^a_{p_2p_1}}-
  \frac{\partial^4}{\partial J^a_{p_2p_1}\partial J^a_{p_1p_2}\partial J^b_{p_2m}\partial J^b_{mp_2}}\right)\mathcal{Z}[J]\Big|_{J=0}\\
  &\qquad+\sum_{\substack{m=0 \\ m\neq p_1}}^{\mathcal{N}}\frac{V}{E_{p_1}-E_m}
  \left(\frac{\partial^2}{\partial J^a_{mp_2}\partial J^a_{p_2m}}-
  \frac{\partial^2}{\partial J^a_{p_1p_2}\partial J^a_{p_2p_1}}\right)\mathcal{Z}[J]\Big|_{J=0}\\
  &\qquad-\sum_{\substack{m=0 \\ m\neq p_2}}^{\mathcal{N}}\frac{V}{E_{p_2}-E_{m}}
  \left(\frac{\partial^2}{\partial J^a_{mp_1}\partial J^a_{p_1m}}-
  \frac{\partial^2}{\partial J^a_{p_1p_2}\partial J^a_{p_2p_1}}\right)\mathcal{Z}[J]\Big|_{J=0}\\
  &\qquad+ \frac{V}{E_{p_2}-E_{p_1}}\left( \frac{\partial^2}{\partial J^a_{p_2p_2}\partial J^a_{p_1p_1}}-
  \frac{\partial^2}{\partial J^a_{p_2p_2}\partial J^a_{p_1p_1}}\right)\mathcal{Z}[J]\Big|_{J=0}\Bigg\}.
  \end{align*}
  The first line generates for $E_m\neq E_{p_1}$ and $E_m\neq E_{p_2}$ either a 
  4-point function with one boundary or two 2-point functions with one boundary, respectively.
  Functions with higher boundaries $b\geq 2$ appear in case of $E_m=E_{p_1}$ or $E_m=E_{p_2}$ . All terms are found by comparing with the expansion
  (\ref{partitionexpansion}).
\end{proof}
\end{prps}
\noindent
We remind that in \sref{Proposition}{Prop2Punkt} correlation functions
of genus $g\geq1$ are also included. 
The SDE of the 2-point 
function depends on $\lambda^2$, since graphs exist only with an even number of vertices.

\begin{prps}\label{PropNPunkt}
 Let $N\geq 3$. The SDE for the $N$-point function in the 3-colour matrix model with an external field
 $E$ is for pairwise different $E_{p_i}, E_{p_j}$ given by 
 \begin{align*}
  &G^{a_1...a_N}_{|p_1...p_N|}
  =-\frac{\lambda }{(E^2_{p_1}-E^2_{p_2})}\sum_{b=1}^3\left(\sigma_{a_1a_Nb}G^{a_2...a_{N-1}b}_{|p_2...p_{N-1}p_N|}
 -\sigma_{a_1a_2b}G^{ba_3a_4...a_{N}}_{|p_1p_3p_4...p_N|}\right)\nonumber\\
 &-\frac{\lambda^2 }{V^2(E^2_{p_1}-E^2_{p_2})}\\
 &\times\Bigg\{V\Bigg(\sum_{\substack{m=0\\m\neq p_1}}^\mathcal{N}\frac{G^{a_1a_2...a_N}_{|mp_2...p_N|}-G^{a_1a_2...a_N}_{|p_1p_2...p_N|}}{E_{p_1}-E_m}
 -\sum_{\substack{m=0\\m\neq p_2}}^\mathcal{N}\frac{G^{a_1a_2a_3...a_N}_{|p_1mp_3...p_N|}-G^{a_1a_2...a_N}_{|p_1p_2...p_N|}}{E_{p_2}-E_m}\Bigg)\nonumber\\
&+\sum_{k=2}^{N}\Bigg(\frac{G^{a_1a_2...a_{k-1}|a_ka_{k+1}...a_N}_{|p_kp_2...p_{k-1}|p_kp_{k+1}...p_N|}-G^{a_1a_2...a_{k-1}|a_ka_{k+1}...a_N}_{|p_kp_2...p_{k-1}|p_1p_{k+1}...p_N|}}{E_{p_1}-E_{p_k}}\\
&\qquad\qquad\qquad-\frac{G^{a_2a_3...a_{k}|a_1a_{k+1}...a_N}_{|p_{k+1}p_3...p_{k}|p_1p_{k+1}...p_N|}-
G^{a_2...a_k|a_1a_{k+1}...a_N}_{|p_2...p_k|p_1p_{k+1}...p_N|}}{E_{p_2}-E_{p_{k+1}}}\Bigg)\nonumber\\
&+\sum_{k=3}^{N-1}V^2\Bigg(G^{a_1a_2...a_{k-1}}_{|p_kp_2...p_{k-1}|}\frac{G^{a_k...a_N}_{|p_k...p_N|}
-G^{a_ka_{k+1}...a_N}_{|p_1p_{k+1}...p_N|}}{E_{p_1}-E_{p_k}}\\
&\qquad\qquad\qquad-G^{a_1a_{k+1}...a_N}_{|p_1p_{k+1}...p_N|}\frac{G^{a_2a_3...a_{k}}_{|p_{k+1}p_3...p_{k}|}
-G^{a_2...a_k}_{|p_2...p_k|}}{E_{p_2}-E_{p_{k+1}}}\Bigg)\\
&+\sum_{b=1}^3\sum_{m=0}^\mathcal{N}\bigg(G^{bba_1...a_N}_{|p_1mp_1...p_N|}-G^{a_1bba_2...a_N}_{|p_1p_2mp_2...p_N|}+\frac{1}{V}
\left(G^{bb|a_1...a_N}_{|p_1m|p_1...p_N|}-G^{bb|a_1...a_N}_{|p_2m|p_1...p_N|}\right)\nonumber\\
&\qquad\qquad\qquad+VG^{a_1...a_N}_{|p_1...p_N|}\left(G^{bb}_{|p_1m|}-G^{bb}_{|p_2m|}\right)\bigg)\nonumber\\
&+\sum_{b=1}^3\sum_{k=2}^N\bigg(\frac{1}{V}\left(G^{ba_1...a_{k-1}|ba_k...a_N}_{|p_kp_1...p_{k-1}|p_1p_k...p_N|}-G^{ba_2...a_{k}|ba_{k+1}...a_Na_1}_{|p_{k+1}p_2...p_{k}|p_2p_{k+1}...p_Np_1|}\right)\\
&\qquad\qquad\qquad+V\left(G^{ba_1...a_{k-1}}_{|p_kp_1...p_{k-1}|}G^{ba_k...a_N}_{|p_1p_k...p_N|}-G^{ba_2...a_{k}}_{|p_{k+1}p_2...p_{k}|}G^{ba_{k+1}...a_Na_1}_{|p_2p_{k+1}...p_Np_1|}\right)\bigg)\\
&+\sum_{b=1}^3\bigg(\frac{1}{V^2}\left(G^{b|b|a_1...a_N}_{|p_1|p_1|p_1...p_N|}-G^{b|b|a_1...a_N}_{|p_2|p_2|p_1...p_N|}\right)+
\frac{1}{V}\left(G^{b|ba_1...a_N}_{|p_1|p_1p_1...p_N}-G^{b|ba_2...a_Na_1}_{|p_2|p_2p_2...p_Np_1}\right)\\*
&\qquad\qquad\qquad+G^{a_1...a_N}_{|p_1...p_N|}\left(G^{b|b}_{|p_1|p_1|}-G^{b|b}_{|p_2|p_2|}\right) \bigg)  \Bigg\},
 \end{align*}
 where $p_{N+1}\equiv p_1$.
\begin{proof}
 We use the definition of the $N$-point function for pairwise different $E_{p_i},E_{p_j}$. 
 With the expression of the partition function
 (\ref{Zustandssumme}), we obtain
 \begin{align*}
  G^{a_1...a_N}_{|p_1...p_N|}=&\frac{1}{V}\frac{\partial^N}{\partial J^{a_1}_{p_1p_2}...J^{a_N}_{p_Np_1}}\frac{\mathcal{Z}[J]}{\mathcal{Z}[0]}\bigg|_{J=0}\nonumber\\
=&-\frac{\lambda }{H_{p_1p_2}V^2\mathcal{Z}[0]}\frac{\partial^{N-1}}{\partial J^{a_2}_{p_2p_3}...J^{a_N}_{p_Np_1}}
\sum_{b,c=1}^3\sigma_{a_1bc}\sum_{n=0}^\mathcal{N}\frac{\partial^{2}}{\partial J^{b}_{p_1n}J^{c}_{np_2}}\mathcal{Z}[J]\bigg|_{J=0}.
 \end{align*}
Here the first derivative $\frac{\partial}{\partial J^{a_1}_{p_1p_2}}$
  applied to $\mathcal{Z}_{free}[J]$ yields $\frac{V}{H_{p_1p_2}}
J^{a_1}_{p_2p_1}$, which can only be differentiated by the interaction
in $\mathcal{Z}[Z]$ because of $p_3\neq p_1$ and $p_2\neq p_4$.
 Applying Lemma \ref{WarsLemma} yields
 \begin{subequations}
 \begin{align}
  =&-\frac{\lambda}{(E^2_{p_1}-E^2_{p_2})V\mathcal{Z}[0]}\frac{\partial^{N-1}}{\partial J^{a_2}_{p_2p_3}...J^{a_N}_{p_Np_1}}\nonumber\\
  &\times \Bigg[\sum_{b,c=1}^{3}\sigma_{a_1bc}
		\sum_{m=0}^\mathcal{N}\left(J^b_{mp_1}\frac{\partial}{\partial J^c_{mp_2}}-J^c_{p_2m}\frac{\partial}{\partial J^b_{p_1m}}\right)\label{I1}\\
		&\qquad\qquad+\frac{\lambda}{V^2}\sum_{b=1}^3\Bigg\{\sum_{m=0}^\mathcal{N}\left(\frac{\partial^3}{\partial J^b_{p_1m}\partial J^b_{mp_1}\partial J^{a_1}_{p_1p_2}}
		-\frac{\partial^3}{\partial J^{a_1}_{p_1p_2}\partial J^b_{p_2m}\partial J^b_{mp_2}}\right)\label{I2}\\
		&\qquad\qquad+\sum_{\substack{m,n=0 \\ n\neq p_1}}^{\mathcal{N}}\frac{V}{E_{p_1}-E_n}\frac{\partial}
		{\partial J^{a_1}_{np_2}}\left(J^b_{mp_1}\frac{\partial}{\partial J^b_{mn}}-J^b_{nm}\frac{\partial}{\partial J^b_{p_1m}}\right)\label{I3}\\
		&\qquad\qquad-\sum_{\substack{m,n=0 \\ n\neq p_2}}^{\mathcal{N}}\frac{V}{E_{p_2}-E_{n}}\frac{\partial}{\partial J^{a_1}_{p_1n}}
		\left(J^b_{p_2m}\frac{\partial}{\partial J^b_{nm}}-J^b_{mn}\frac{\partial}{\partial J^b_{mp_2}}\right)\Bigg\}\Bigg]\mathcal{Z}[J]\bigg|_{J=0}\label{I4}.
 \end{align}
 \end{subequations}
The first term of (\ref{I1}) contributes only for $b=a_N$ and $E_m=E_{p_N}$ and the second 
term only for $c=a_2$ and $E_m=E_{p_3}$. This generates
the term proportional to $\lambda$. Line (\ref{I2}) produces three 
different types of terms for arbitrary $E_m$, the $(2+N)$-point functions with one boundary, the multiplication of 2-point 
with $N$-point functions, and $(2+N)$-point functions with two boundaries. 
If in (\ref{I2}) $E_m=E_{p_k}$ for the first term with $2\leq k\leq N$ (for the second term 
with $3\leq k\leq N$ or $k=1$), additionally 
$(k+(N+2-k))$-point 
functions with two boundaries and the multiplication of $k$-point with $(N+2-k)$-point 
functions with one boundary are generated.
In case of $E_m=E_{p_1}$ for the left term ($E_m=E_{p_2}$ for the right term)  (\ref{I2}) 
produces either $(1+1+N)$-point functions with three boundaries, $(1+(1+N))$-point functions with two boundaries
or the multiplication of $(1+1)$-point with $N$-point functions. 

Finally, we look at (\ref{I3}) and (\ref{I4}) together. The first terms
again contribute only for $b=a_N$ and $E_m=E_{p_N}$ in (\ref{I3}) or for $b=a_2$ and $E_m=E_{p_3}$ in
(\ref{I4}). Since the sum over $n$ survives, 
$N$-point functions arise. If $E_n=E_{p_k}$ for $k\neq 1$ in (\ref{I3}) and for $k\neq 2$ in (\ref{I4}) 
one gets either $(k+(N-k))$-point functions or 
the multiplication of $k$-point functions with $(N-k)$-point functions with one boundary. For 
the second term in (\ref{I3}) and (\ref{I4}), each derivative have to be taken into account.
If the derivative in front of the brackets in (\ref{I3}) and (\ref{I4}) acts on $J^b_{nm}$ or 
$J^b_{mn}$, the sum over $n$ survives again
and has a prefactor depending on $E_n$, but no $n$ appears in the $N$-point function. 
If any other derivative $\frac{\partial}{\partial J^{a_{k+1}}_{p_{k+1} p_{k+2}}}$, for some $k\geq 1$, 
acts on the second term, $n,m,b$ will be
fixed and it will produces $N$-point functions, $(k+(N-k))$-point functions with two boundaries
and the multiplication of $k$-point with
$(N-k)$-point functions. Collecting all and making use of (\ref{Entwicklungskoeffizienten}) to get the correct prefactor in
$V$, one finds all the terms appearing in \sref{Proposition}{PropNPunkt}.
\end{proof}
\end{prps}\noindent
The first term shows that a $(N-1)$-point function only contributes for different adjacent colours, because of $\sigma_{a_1a_Nb}$ and 
$\sigma_{a_1a_2b}$. This fact is in perfect accordance with a loop expansion. Furthermore, 
the 2-point function is assigned with a special 
r$\hat{\text{o}}$le, since the sum over $m$ only appears for the $N$-point and 2-point 
function even in the large $\mathcal{N},V$-limit.

It should be emphasised that not all combinations of the colours for the correlation functions are possible. 
The 2-point function is of the form $G^{aa}_{|p_1p_2|}$ and
the 3-point function $\sigma_{abc}G^{abc}_{|p_1p_2p_3|}$. 
There exists no 4-point function equipped with all three colours simultaneously, and so on. These properties which are first 
recognized by loop expansion are intrinsically presented in the SDEs. 

Correlation functions with more boundary components satisfy also SDE which can be computed. However, to determine them 
an anlog of \sref{Theorem}{Thm:Raimar} for the coloured model is necessary.

\section{Link to the Cubic and Quartic Model}
The 3-coloured model with cubic interaction is of particular interest because its Feynman graphs are a
subset of the Feynman graphs of the cubic model (see \sref{Ch.}{chap:cubic}). However, the
graphs are additionally decorated by colours which prohibits for instance the tadpole graph and therefore induces 
a vanishing 1-point function. On the other hand,
some graphs need to be counted several times since a graph can have with different colourings which gives a symmetry factor.
This symmetry is due to the external matrix $E$ which is taken equally for each colour. 

A more detailed perturbative analysis is performed 
in our paper \cite{Hock:2018wup}. The perturbative calculation of the 2-point function through 
the SDE is compared to the Feynman graph calculation 
up to three loops on the two dimensional Moyal space, and ofcourse both results coincide perfectly. 
No renormalisation was necessary since each graph is UV finite for spectral dimension $\D<4$. 

The SDE of the 2-point function (\sref{Proposition}{Prop2Punkt}) takes a much easier form after $V$-expansion
\begin{align}
 &G^{aa}_{|p_1p_2|}=\frac{1}{H_{p_1p_2}}+\frac{\lambda^2}{E^2_{p_1}-E^2_{p_2}}
  \Bigg[G^{aa}_{|p_1p_2|}\frac{1}{V}\sum_{m=0}^\mathcal{N}\sum_{b=1}^3\left(
 G^{bb}_{|p_2m|}
-G^{bb}_{|p_1m|}\right)\\
&\qquad \qquad\qquad\qquad +\frac{1}{V}\sum_{\substack{m=0\\ m\neq p_2}}^\mathcal{N}\frac{G^{aa}_{|p_1m|}-G^{aa}_{|p_1p_2|}}{E_{p_2}-E_m}-
\frac{1}{V}
\sum_{\substack{m=0\\ m\neq p_1}}^\mathcal{N}
\frac{G^{aa}_{|p_1p_2|}-G^{aa}_{|p_2m|}}{E_{m}-E_{p_1}}\Bigg]+\mathcal{O}(V^{-1}).\nonumber
\end{align}
This equation is manifestly symmetric in $p_1,p_2$. Symmetrising the SDE (\sref{Proposition}{Prop:Quart2P}) of the 
2-point function for the quartic model after $V$-expansion yields
\begin{align}
  G_{|pq|}&=
\frac{1}{E_p+E_q}-\frac{\lambda}{2(E_p+E_q)} 
\bigg[  G_{|pq|}\frac{1}{V}\sum_{n=0}^\mN \left(G_{|pn|}+G_{|qn|}\right)\\
&\qquad \qquad\qquad\qquad+\frac{1}{V}\sum_{n=0}^\mN \frac{G_{|pq|}- G_{|nq|}}{E_n-E_q}
+\frac{1}{V}\sum_{n=0}^\mN \frac{G_{|pq|}- G_{|nq|}}{E_n-E_p} 
\bigg]+\mathcal{O}(V^{-1}).\nonumber
 \end{align}
 The equations show an incredible similarity. Looking at the graphs
of the 3-coloured model, all propagators of a chosen colour can be contracted.
It means that all vertices which were connected by this colour are now coincident. 
The resulting graph has 
vertices of valence 4 each weighted with a factor $\lambda^2$. Furthermore, the vertices carry
the dynamics of the contracted 
propagator. However, topologically the same graphs appears in the perturbative expansion as in the quartic model, but with 
some additional decoration and constraints.
\begin{figure}[h]
    \centering
    \def\svgwidth{0.9\textwidth}
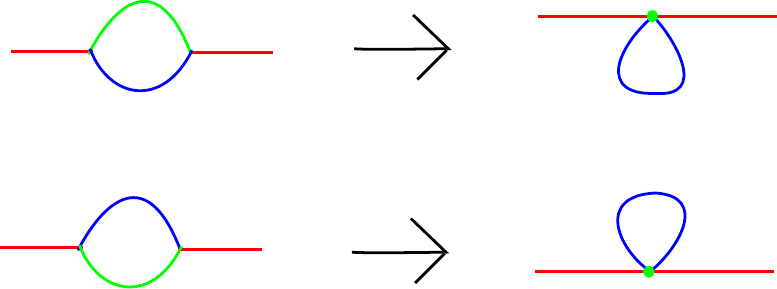
    \caption{On the lhs, two Feynman graphs are at order $\lambda^2$ of the 3-coloured model. After contracting the 
    green coloured propagator (rhs), two Feynman graphs appears with an quartic interaction.}
    \label{Fig:3c}
\end{figure}

This contraction of propagators of one chosen colour is understood as integrating out 
the corresponding field for the partition function.
It is a straightforward calculation to integrate out for instance the field $\Phi^3$ to get  
\begin{align*}
 \Z[0]&=\int \prod_{a=1}^3D\Phi^a \exp\bigg\{-\Tr\bigg(\sum_{a=1}^3 E\Phi^a\Phi^a+\lambda
 \Phi^1\Phi^2\Phi^3+\lambda
 \Phi^1\Phi^3\Phi^2\bigg)\bigg\}\\
 =&C(E)\int D\Phi^1D\Phi^2\exp\bigg\{-\Tr\big(E\Phi^1\Phi^1+
 E\Phi^1\Phi^1\big)+\frac{\lambda^2}{2}\sum_{n,m=0}^\mN \frac{(\Phi^1\Phi^2)_{nm}
 (\Phi^2\Phi^1)_{mn}}{E_n+E_m}\bigg\},
\end{align*}
where $C(E)$ is a constant depending on $E$. Now, we have a quartic interaction with two propagators 
of each colour attached to 
the vertex. The vertex has the weight $\lambda^2$ and an additional dynamics through the denominator of 
$\frac{(\Phi^1\Phi^2)_{nm}
 (\Phi^2\Phi^1)_{mn}}{E_n+E_m}$, which fits perfectly with the considerations at the Feynman graph level above.

\phantomsection
\addcontentsline{toc}{chapter}{Bibliography}
\bibliographystyle{halpha-abbrv}
\bibliography{./Bibo}


\end{document}